\documentclass[a4paper,english,10pt]{article}
\usepackage[left=.9in, right=.9in, bottom=1.2in, top=1in,dvips]{geometry}
\makeatletter

\AtBeginDocument{\setlength\abovedisplayskip{5pt plus 2pt minus 1pt}
\setlength\belowdisplayskip{5pt plus 2pt minus 1pt}}
\usepackage[utf8]{inputenc}
\usepackage{lmodern}
\usepackage{indentfirst}
\usepackage[affil-it]{authblk}
\usepackage{rotating}
\usepackage{tocloft}
\usepackage{titlesec}
 \titleformat{\subparagraph}[hang]{\normalfont}{\thesubparagraph}{0pt}{\underline}
 \titleformat{\paragraph}[hang]{\normalfont}{\theparagraph}{0pt}{\myuline}
\usepackage{titletoc}

\usepackage[textsize=scriptsize]{todonotes} 
\makeatletter
\renewcommand{\todo}[2][]{\tikzexternaldisable\@todo[#1]{#2}\tikzexternalenable}
\makeatother

\usepackage{standalone}
\usepackage{subfiles}

\usepackage[hidelinks]{hyperref}
\usepackage{enumitem}
\setlist[enumerate]{itemsep=2.0pt plus 1.0 pt minus 0.5pt, topsep=4.0pt plus 2.0 pt minus 1.0pt}
\setlist[itemize]{itemsep=2.0pt plus 1.0 pt minus 0.5pt, topsep=4.0pt plus 2.0 pt minus 1.0pt}
\usepackage{array}
\newcolumntype{M}[1]{>{\centering\arraybackslash}m{#1}}
\usepackage{multirow}
\usepackage{tikz}
\usetikzlibrary{external}
\usepackage{calc}
\usepackage{graphicx}
\usepackage{subcaption}
\usepackage[labelfont=bf]{caption}
\usepackage{epsfig}
\usepackage{color, soul}

\usepackage{contour}
\usepackage[normalem]{ulem}

\contourlength{1pt}

\newcommand{\myuline}[1]{%
  \uline{\phantom{#1}}%
  \llap{\contour{white}{#1}}%
}

\usepackage{marginnote}

\usepackage{color,soulutf8}

\usepackage{amsmath, amsthm, slashed}
\usepackage[capitalize]{cleveref}
\usepackage{amssymb}
\usepackage{bbm}
\usepackage{bm}
\usepackage{marvosym}
\usepackage{mathtools}
\usepackage{mathrsfs}
\usepackage{upgreek}
\usepackage{accents}
\usepackage{textcomp}
\usepackage{slashed} 
\usepackage{centernot}
\usepackage{cancel}
\usepackage{harmony} 
\usepackage{wasysym} 
\wasyfamily
\DeclareFontShape{U}{wasy}{b}{n}{ <-10> ssub * wasy/m/n
 <10> <10.95> <12> <14.4> <17.28> <20.74> <24.88>wasyb10 }{}
\DeclareFontShape{U}{wasy}{b}{it}{ <-10> ssub * wasy/m/n
 <10> <10.95> <12> <14.4> <17.28> <20.74> <24.88>wasyb10 }{}
\DeclareMathAlphabet\mathbfcal{OMS}{cmsy}{b}{n}
\newcommand\numberthis{\addtocounter{equation}{1}\tag{\theequation}}
\renewcommand{\Re}{\operatorname{Re}}
\renewcommand{\Im}{\operatorname{Im}}
\DeclareMathOperator{\sign}{sign}
\allowdisplaybreaks

\numberwithin{equation}{section}

\theoremstyle{plain}

\newtheorem{alphtheorem}{Theorem}

\newtheorem{arabictheorem}{Theorem}
\newtheorem{theorem}{Theorem}[section]
\newtheorem*{theorem*}{Theorem}
 
\newtheorem*{conjecture*}{Conjecture}

\newtheorem*{corollary*}{Corollary}
\newtheorem{proposition}{Proposition}[subsection]
\newtheorem{lemma}[proposition]{Lemma}

\theoremstyle{definition}
\newtheorem{definition}{Definition}[subsection]
\newtheorem{remark}{Remark}[subsection]
\newtheorem*{remark*}{Remark}
\newtheorem{bigquestion}{Question}

\usepackage{csquotes}

\newcommand{\mb}[1]{\mathbb{#1}}

\newcommand{\mc}[1]{\mathcal{#1}}
\newcommand{\mr}[1]{\mathrm{#1}}

\newcommand{\p}{\partial}
\newcommand{\lp}{\left}
\newcommand{\rp}{\right}
\newcommand\supp{\mathrm{supp}}

\newlength{\dhatheight}

\newcommand{\uL}{\underline{L}}
\newcommand{\swei}[2]{#1^{[{#2}]}}

\newcommand{\sml}[2]{#1_{ml}^{[{#2}],\,a\omega}}
\newcommand{\smlambda}[2]{#1_{m\Lambda}^{[{#2}],\,a\omega}}

\begin{document}
 \title{\LARGE \textbf{Boundedness and decay for the Teukolsky equation \\ on Kerr in the full subextremal range $|a|<M$:\\ frequency space analysis}}

\author[1]{{\Large Yakov Shlapentokh-Rothman}}
\author[2]{{\Large  Rita \mbox{Teixeira da Costa}\vspace{0.4cm}}}

\affil[1]{\small  Princeton University, Department of Mathematics, Washington Road, Princeton, NJ 08544,
United States \vspace{0.2cm} \ }

\affil[2]{\small 
University of Cambridge, Department of Pure Mathematics and Mathematical Statistics, Wilberforce Road, Cambridge CB3 0WA, United Kingdom}

\date{\today}

\maketitle

\begin{abstract}
This paper is the first of a series regarding the Teukolsky equation of spin $\pm 1$ and spin $\pm 2$ on Kerr backgrounds in the full subextremal range of parameters $|a|<M$. In the present paper, we study fixed frequency solutions of the transformed system of equations introduced by Dafermos, Holzegel and Rodnianski, obtaining estimates which are uniform in the separation parameters. A corollary of our result, to be laid out in the second paper of the series, is that solutions of the Teukolsky equation on subextremal Kerr arising from regular initial data remain bounded and decay in time. This is a key step in establishing the full linear stability of Kerr under electromagnetic and gravitational perturbations. Our estimates can also be applied to understanding more delicate features of the Teukolsky equation, such as their scattering properties.
\end{abstract}

\bigskip\bigskip

\tableofcontents

\section{Introduction}

The Einstein vacuum equations 
\begin{equation}
\mr{Ric}(g)=0\,, \label{eq:einstein-vacuum}
\end{equation}
to be satisfied by a 4-dimensional Lorentzian manifold $(\mc{M},g)$ representing spacetime, are the system of geometric partial differential equations which lie at the core of the theory of General Relativity. 

Aside from the flat Minkowski space, several other explicit solutions of \eqref{eq:einstein-vacuum} are known, most notably the Kerr family of black holes \cite{Kerr1963}. Kerr black holes are axisymmetric, asymptotically flat spacetimes parametrized by a mass, $M>0$, and specific angular momentum, $a$, satisfying the bound $|a|\leq M$. We say these are subextremal if $|a|<M$, with the spherically symmetric case $a=0$ being known as the Schwarzschild solution \cite{Schwarzschild1916}, and extremal if $|a|=M$.  

Subextremal Kerr black holes have been shown to be isolated and, under stronger assumptions, unique in the class of regular stationary, asymptotically flat black hole solutions to \eqref{eq:einstein-vacuum} (see \cite{Alexakis2010}, the review \cite{Chrusciel2012} and  references therein); thus, they play a privileged role in our understanding of vacuum General Relativity. For this reason, a great deal of literature on General Relativity is centered around the question

\noindent
\begin{center}
\noindent \textit{What are the dynamics of perturbations of Kerr under the Einstein vacuum equations \eqref{eq:einstein-vacuum}?}
\end{center}

Indeed, the Einstein vacuum equations \eqref{eq:einstein-vacuum} can be seen as evolution equations  for given data: they are well-posed for sufficiently regular data given on a Cauchy hypersurface (see \cite{Foures-Bruhat1952,Choquet-Bruhat1969} and the recent \cite{Sbierski2016}) or on a characteristic hypersurface (see \cite{Rendall1990,Luk2012}). One can view Kerr black hole spacetimes as arising from some explicit initial data on a suitable Cauchy hypersurface or alternatively from some characteristic data for \eqref{eq:einstein-vacuum} in the ancient past or far future, i.e.\ an asymptotic past or future state, respectively.  Thus, we can realize the previous question in the more concrete problem:

\begin{bigquestion} \label{qn:Kerr-perturbations-nonlinear-2}
In evolution under the Einstein vacuum equations \eqref{eq:einstein-vacuum}, 
\begin{enumerate}[label=\bfseries\Alph*.,ref=\Alph*]
\item \textit{Stability for Cauchy data}. Do initial data sets close to that corresponding to a Kerr black hole give rise to a spacetime whose exterior remains close to said member of the Kerr family, and does the spacetime asymptotically approach a Kerr solution with nearby mass and specific angular momentum? \label{it:Kerr-perturbations-nonlinear-2-Cauchy}
\item \textit{Stability for scattering data}.  Is an asymptotic past (resp.\ future) state close to that corresponding to a Kerr black hole mapped to an asymptotic future (resp.\ past) state which is close to that of some member of the Kerr family?\label{it:Kerr-perturbations-nonlinear-2-scattering}
\end{enumerate}
\end{bigquestion}

The answer to these questions is thought to be positive, at least in the subextremal $|a|<M$ case. To date, stability for general perturbations of Cauchy data has only been shown to hold for the Minkowski solution, first by the monumental work of Christodoulou and Klainerman \cite{Christodoulou1993} and then by an alternative method in \cite{Lindblad2010}; see also \cite{Bieri2009,Huneau2018b,Hintz2020,Keir2018} for further developments along the lines of these two approaches. Question~\ref{qn:Kerr-perturbations-nonlinear-2}\ref{it:Kerr-perturbations-nonlinear-2-Cauchy} has inspired a lot of work which we will address in more detail shortly. We highlight already two particular results. In \cite{Dafermos2016a}, Dafermos, Holzegel and Rodnianski showed the linear stability of the Schwarzschild ($a=0$) subfamily. In \cite{Klainerman2017},  Klainerman and Szeftel then showed that the Schwarzschild subfamily is nonlinearly stable under the special class of polarized axisymmetric perturbations of Cauchy data; see also the forthcoming \cite{Taylor2019}. These results pertain to the black hole exterior; the question of stability in the \textit{interior} of the Kerr black hole is of an entirely different nature, see \cite{Dafermos2017a}.

On the other hand, questions relating to stability for scattering data are not yet fully understood even for Minkowski spacetime in $3+1$ dimensions (though note \cite{Wang2013} for $4+1$ and higher dimensions). We also point out work \cite{Dafermos2013} related to Question~\ref{qn:Kerr-perturbations-nonlinear-2}\ref{it:Kerr-perturbations-nonlinear-2-scattering}, where the authors construct a nontrivial family of spacetimes solving \eqref{eq:einstein-vacuum} whose exterior geometry settles down to that of a fixed Kerr black hole.

Attempts have also been made to answer analogues of Question~\ref{qn:Kerr-perturbations-nonlinear-2} for solutions to the Einstein equations with a cosmological constant, $\Lambda$, i.e.\ 
\begin{align}
\mr{Ric}(g)-\frac12\Lambda g=0\,. \label{eq:einstein-vacuum-Lambda}
\end{align}
In the case $\Lambda>0$, stability of the trivial solution of \eqref{eq:einstein-vacuum-Lambda}, known as de Sitter space, under perturbations of Cauchy data was shown in \cite{Friedrich1986}. The problem is considerably easier than the stability of Minkowski settled in \cite{Christodoulou1993}, as $\Lambda>0$ ensures faster decay rates than for $\Lambda=0$. Indeed, more progress has been made in the analogue of Question~\ref{qn:Kerr-perturbations-nonlinear-2}\ref{it:Kerr-perturbations-nonlinear-2-Cauchy} in this setting: Kerr--de Sitter, the $\Lambda>0$ cousin of the Kerr black hole, has been shown to be nonlinearly stable for very small angular momentum in \cite{Hintz2016}.  

In the case $\Lambda<0$, the initial value problem for the Einstein equations \eqref{eq:einstein-vacuum-Lambda} requires an additional boundary condition at infinity. Assuming reflecting boundary conditions at infinity, one expects instability rather than stability for the trivial solution, known as Anti-de Sitter space \cite{Dafermos2006}. This has indeed been proven for \eqref{eq:einstein-vacuum-Lambda} coupled to several spherically symmetric matter models by Moschidis \cite{Moschidis2017a,Moschidis2018}.


In the study of questions of nonlinear stability such as Question~\ref{qn:Kerr-perturbations-nonlinear-2} and their $\Lambda\neq 0$ analogues, one needs to first define what one means by \textit{perturbations} of a given spacetime. One approach, known in the physics literature as the Newman--Penrose formalism \cite{Newman1962}, is to consider perturbations of the Ricci curvature and connection components expressed in terms of a null frame; this leads to \eqref{eq:einstein-vacuum} being expressed as a nonlinear system of equations called the \textit{null structure equations}. Naturally, the system only becomes tractable once the diffeomorphism invariance of \eqref{eq:einstein-vacuum} has been eliminated by a particular choice of gauge which leads to a well-posed system of PDEs.

In order to understand nonlinear stability, one must first understand linear stability. For the Kerr black hole exterior, there is a particular choice of frame in the Newman--Penrose formalism, called the \textit{algebraically special frame}, for which a miracle occurs: the equations for two gauge-invariant components of the Ricci curvature tensor, known as \textit{extremal} components, fully decouple from the remaining ones at the linear level. Indeed, they satisfy the Teukolsky equation \cite{Teukolsky1973}, which assumes the form
\begin{align}
\begin{split}
\bigg[\Box_g &+\frac{2s}{\rho^2}(r-M)\p_r +\frac{2s}{\rho^2}\lp(\frac{a(r-M)}{\Delta}+i\frac{\cos\theta}{\sin^2\theta}\rp)\p_\phi\\
&+\frac{2s}{\rho^2}\lp(\frac{M(r^2-a^2)}{\Delta}-r -ia\cos\theta\rp)\p_t +\frac{1}{\rho^2}\lp(s-s^2\cot^2\theta\rp)\bigg]\upalpha^{[s]}  =0\,,
\end{split}\label{eq:Teukolsky-equation-intro}
\end{align}
with $s=\pm 2$. Here, $\Box_g$ is the covariant scalar wave operator, $\swei{\upalpha}{s}$ is an $s$-spin weighted function (see already Definition~\ref{def:smooth-spin-weighted}) and we have used Boyer--Lindquist coordinates $(t,r,\theta,\phi)\in\mathbb{R}\times(r_+,\infty)\times\mathbb{S}^2$ and 
\begin{align*}
\Delta:=(r-r_+)(r-r_-)\,, \quad r_\pm :=M\pm \sqrt{M^2-a^2}\,,\qquad \rho^2=r^2+a^2\cos^2\theta\,.
\end{align*}
In order to analyze the system of linearized null structure equations,  $\swei{\upalpha}{\pm 2}$ are a natural starting point: their vanishing completely characterizes pure gauge solutions and, indeed, one can in principle hope to control the remaining quantities in the system of linearized null structure equations by $\swei{\upalpha}{\pm 2}$ once a suitable gauge has been fixed (see for instance \cite{Dafermos2016a}).  

The Teukolsky variable $\swei{\upalpha}{s}$ can be understood in a suitable space of functions for more general spin $s\in\frac12\mathbb{Z}$, and indeed physical meaning can be ascribed to the Teukolsky equation for spins other than $\pm 2$ (see for instance \cite{Chandrasekhar}): for $s=0$, it reduces to the wave equation; for $s=\pm 1/2$, it describes neutrino propagation and can be related to the Dirac equation; for $s=\pm 3/2$, it models propagation of spin-3/2 fermions and can be related to the Rarita--Schwinger equation. Arguably, it is the case $s=\pm 1$ that most resembles $s=\pm 2$. Similarly to what occurs for linear gravitational perturbations, in the Newman--Penrose formalism, the linearized Maxwell equations form a system of coupled equations where those for the extremal components of the electromagnetic tensor decouple; the Teukolsky equation for $s=\pm 1$, which describes the dynamics of these gauge-invariant components, is also a natural starting point for any analysis of electromagnetic perturbations. 

This brings us to the main question we address in the present series of papers:
\begin{bigquestion} \label{qn:Teukolsky-physical-space}
Describe the behavior of general solutions to the Teukolsky equation \eqref{eq:Teukolsky-equation-intro} for all Kerr parameters $|a|\leq M$; specifically,
\begin{enumerate}[label=\bfseries\Alph*., ref=\Alph*]
\item \textit{Solutions arising from Cauchy data.} In the black hole exterior, do solutions to \eqref{eq:Teukolsky-equation-intro} arising from  regular Cauchy data remain bounded, and do they also decay at a sufficiently fast rate, in time?\label{it:Teukolsky-physical-space-bddness-decay}
\item \textit{Solutions arising from scattering data.}  Is there a ``natural'' energy space such that a scattering theory for \eqref{eq:Teukolsky-equation-intro} on the black hole exterior is well defined in the sense of existence, uniqueness and asymptotic completeness of asymptotic past and future states? In particular, can we control the strength of superradiant reflection? \label{it:Teukolsky-physical-space-scattering}
\end{enumerate}
\end{bigquestion}

The answer is thought to be affirmative at least for subextremal parameters $|a|<M$, similarly to Question~\ref{qn:Kerr-perturbations-nonlinear-2}. Question~\ref{qn:Teukolsky-physical-space} has been answered in the full subextremal range $|a|<M$ for the wave equation, $s=0$, by Dafermos, Rodnianski and the first author in \cite{Dafermos2016b}, for Cauchy data (Question~\ref{qn:Teukolsky-physical-space}\ref{it:Teukolsky-physical-space-bddness-decay}), and in \cite{Dafermos2014}, for scattering data (Question~\ref{qn:Teukolsky-physical-space}\ref{it:Teukolsky-physical-space-scattering}). On the other hand, in the cases $s=\pm 1,\pm 2$, progress on Question~\ref{qn:Teukolsky-physical-space} has only been achieved under restrictions on the rotation parameter. Regarding Question~\ref{qn:Teukolsky-physical-space}\ref{it:Teukolsky-physical-space-bddness-decay} for Schwarzschild, $a=0$, the case $s=\pm 2$ was first dealt with by Dafermos, Holzegel and Rodnianski in \cite{Dafermos2016b}, followed by a similar approach to $s=\pm 1$ by Pasqualotto in \cite{Pasqualotto2016} (see also \cite{Blue2008}); both works then go on to establish full linear stability under the gravitational or electromagnetic, respectively, perturbations to Cauchy data. Generalizations of \cite{Dafermos2016b} to very slowly rotating Kerr, $|a|\ll M$, were obtained independently by Dafermos, Holzegel and Rodnianski in \cite{Dafermos2017} for $s=\pm 2$ and Ma \cite{Ma2017,Ma2017a} for $s=\pm 1, \pm 2$, respectively. Finally, Question~\ref{qn:Teukolsky-physical-space}\ref{it:Teukolsky-physical-space-scattering} for $s=\pm 2$ and $a=0$ is considered by Masaood in \cite{Masaood2020}. See already Section~\ref{sec:intro-previous-work} for further reading on the topic.

We remark that, in the extremal case $|a|=M$, it is yet unclear what to expect the answer to Question~\ref{qn:Kerr-perturbations-nonlinear-2} to be (see already Section~\ref{sec:intro-previous-work-extremal}). Already for axisymmetric solutions to \eqref{eq:Teukolsky-equation-intro}, an instability mechanism, known as the Aretakis instability, for higher order derivatives transverse to the event horizon was discovered in \cite{Aretakis2012} for $s=0$ and generalized to $s\in\mathbb{Z}$ in \cite{Lucietti2012}. Nevertheless, the second author has shown that general solutions to the Teukolsky equation~\eqref{eq:Teukolsky-equation-intro} exhibit a weak stability property known as mode stability \cite{TeixeiradaCosta2019} to which we will return shortly. 

From the previous discussion, we see that for electromagnetic and gravitational perturbations alike there is good reason to consider Question~\ref{qn:Teukolsky-physical-space}, together with the issue of gauge, to be the key missing piece to  address Question~\ref{qn:Kerr-perturbations-nonlinear-2}. In this paper and in our upcoming \cite{SRTdC2022}, we will settle Question~\ref{qn:Teukolsky-physical-space}\ref{it:Teukolsky-physical-space-bddness-decay} in the full subextremal range $|a|<M$:

\begin{alphtheorem} \label{thm:bddness-decay-big} Fix $s\in\{0,\pm 1,\pm 2\}$, $M>0$ and $|a|<M$. On the exterior of a subextremal Kerr black hole spacetime with parameters $(a,M)$, general solutions to the Teukolsky equation \eqref{eq:Teukolsky-equation-intro} arising from sufficiently regular initial data on a Cauchy surface 
\begin{itemize}[noitemsep]
\item have uniformly bounded energy fluxes through a suitable spacelike foliation of the black hole exterior, through the future event horizon, $\mc{H}^+$, and through future null infinity, $\mc{I}^+$, in terms of an energy flux of initial data at the same level of regularity; 
\item satisfy a suitable version of ``integrated local energy decay'' with loss of derivatives at trapping;
\item and satisfy similar statements for higher-order energies.
\end{itemize}
Moreover, solutions to \eqref{eq:Teukolsky-equation-intro} remain uniformly bounded pointwise and, in fact, decay at a suitable inverse polynomial rate, with estimates which depend only on $M$, $|a|$, $s$ and a suitable initial data quantity.
\end{alphtheorem}

Though understanding scattering is not the goal of this series of works, our methods also allow us to infer what can be seen as a partial answer to Question~\ref{qn:Teukolsky-physical-space}\ref{it:Teukolsky-physical-space-scattering} in the subextremal $|a|<M$ case; see Theorem~\ref{thm:frequency-estimates-big}B.

Ours are the first quantitative results regarding general electromagnetic and gravitational perturbations valid for the entire subextremal range $|a|<M$ of the Kerr family, and pave the way for a positive resolution of Question~\ref{qn:Kerr-perturbations-nonlinear-2}\ref{it:Kerr-perturbations-nonlinear-2-Cauchy}.  We highlight the fact that the energy boundedness statement we obtain does not lose derivatives, making Theorem~\ref{thm:bddness-decay-big} a true orbital stability result, in addition to a statement of asymptotic stability. Finally, we recall the case $s=0$ had already been settled by the first author together with Dafermos and Rodnianski in \cite{Dafermos2016b, Dafermos2014}.


In our analysis, instead of \eqref{eq:Teukolsky-equation-intro}, we often consider a transformed system (see already Section~\ref{sec:intro-transformed}): for $s\in\mathbb{Z}$ and $0\leq k\leq |s|$, let $\swei{\uppsi}{s}_{(k)}$ be obtained by taking $k$ appropriately $r$-weighted null derivatives of $\swei{\upalpha}{s}$, i.e.\
\begin{align}
\uppsi_{(0)}^{[s]}=j_{s,|s|}^{-1}(r)\swei{\upalpha}{s}\,;\qquad \uppsi_{(k+1)} = j_{s,k}^{-1}(r)\mc{L} \uppsi_{(k)}^{[s]}\,,\,\, k=0,...,|s|-1\,, \label{eq:def-upppsi-k-intro-basic}
\end{align}
where $\mc{L}$ denotes null vector fields and $j_{s,k}$ are appropriate functions of $r$. Then $\swei{\Psi}{s}=\swei{\uppsi}{s}_{(|s|)}$ satisfies
\begin{equation}
\swei{\mathfrak{R}}{s} \swei{\Psi}{s} = a \swei{\mathfrak{J}}{s}\hspace{-2pt}\lp(\uppsi_{(k)}, \p_\phi\uppsi_{(k)} \rp)\,, \quad  k=0,...,|s|-1\,,\label{eq:transformed-equation-intro-basic}
\end{equation}
for a linear operator $\swei{\mathfrak{J}}{s}$. For $a=0$, the differential operator $\swei{\mathfrak{R}}{s}$ is called the Regge--Wheeler operator when $s=\pm 2$ and the Fakerell--Ipser operator when $s=\pm 1$. For any $|a|<M$, it behaves very much like $\Box_g$; indeed, it can be seen as a wave operator for spin-weighted, rather than scalar, functions as it does not contain first order terms, unlike the Teukolsky equation~\eqref{eq:Teukolsky-equation-intro}. In our work, we often begin by proving estimates for the system formed by \eqref{eq:transformed-equation-intro-basic} and the PDEs for each $\swei{\uppsi}{s}_{(k)}$ so that estimates on the Teukolsky equation~\eqref{eq:Teukolsky-equation-intro} follow by integrating the transport equations \eqref{eq:def-upppsi-k-intro-basic} that define the $k$th variable in terms of the $(k+1)$th variable. 

We note that the transformed system we will consider is a simple extension to $s\in\mathbb{Z}$ of that given in the analogous work of Dafermos, Holzegel and Rodnianski \cite{Dafermos2017} for $|a|\ll M$ and $s=\pm 2$, which can be taken as a prequel to our series of papers. In fact, the system considered in \cite{Dafermos2017} and here is a generalization to $a\neq 0$ of earlier work of the same authors for $a=0$ and $s=\pm 2$ in the aforementioned \cite{Dafermos2016a} and of the transformations for $a=0$ and $s=\pm 1$ considered in \cite{Pasqualotto2016}. The independent work \cite{Ma2017,Ma2017a} in the very slowly rotating, $|a|\ll M$ setting is based on the same strategy, but relies on a different generalization of the system in \cite{Dafermos2016a}.


In our analysis, we will apply a mix of physical and frequency space methods. As understood by Carter \cite{Carter1968} for $s=0$ and generalized by Teukolsky \cite{Teukolsky1973} for $s\in\frac12\mathbb{Z}$, the Teukolsky equation is formally separable and, hence, so is the aforementioned transformed system. Our frequency space techniques make use of this separability property, which in principle reduces the study of the PDEs to the study of ODEs for the angular and radial components of $\swei{\upalpha}{s}$, $\swei{\uppsi}{s}_{(k)}$ and $\swei{\Psi}{s}$. For instance, by abuse of notation, writing  $\swei{\upalpha}{s}$ and $\swei{\Psi}{s}$ for their radial component, the radial ODE corresponding to \eqref{eq:Teukolsky-equation-intro} is
\begin{equation}
\frac{\Delta}{r^2+a^2}\frac{d}{dr}\lp[\frac{\Delta}{r^2+a^2}\frac{d}{dr}\lp((r^2+a^2)^{\frac12}\Delta^{\frac{s}{2}}\swei{\alpha}{s}\rp)\rp]+\lp(\omega^2-\smlambda{{V}}{s}(r)\rp)(r^2+a^2)^{\frac12}\Delta^{\frac{s}{2}}\swei{\upalpha}{s} = 0\,,\label{eq:radial-ODE-alpha-intro-basic}
\end{equation}
for an explicit potential $\smlambda{{V}}{s}(r)$, and the radial ODE corresponding to \eqref{eq:transformed-equation-intro-basic} is
\begin{equation}
\frac{\Delta}{r^2+a^2}\frac{d}{dr}\lp[\frac{\Delta}{r^2+a^2}\frac{d}{dr}\lp(\swei{\Psi}{s}\rp)\rp]+\lp(\omega^2-\smlambda{\mc{V}}{s}(r)\rp)\swei{\Psi}{s} = a \swei{\mathfrak{J}}{s}\hspace{-2pt}\lp(\uppsi_{(k)}, im \uppsi_{(k)} \rp)\,, \quad  k=0,...,|s|-1\,,\label{eq:radial-ODE-Psi-intro-basic}
\end{equation}
for an explicit  potential $\smlambda{\mc{V}}{s}(r)$, where $\omega$ is the time frequency associated with $\p_t$, $m$ is the azimuthal frequency associated with $\p_\phi$ and $\Lambda$ is a separation constant (see already Section~\ref{sec:intro-separability}).

The radial ODE \eqref{eq:radial-ODE-alpha-intro-basic} arising from the Teukolsky equation (and thus also those of the transformed system, such as \eqref{eq:radial-ODE-Psi-intro-basic}) already captures the rich structure of waves on rotating Kerr black holes, be they scalar, electromagnetic or gravitational waves. For real $\omega$,  as long as the asymptotics of $\swei{\upalpha}{s}$, $\swei{\uppsi}{s}_{(k)}$ and $\swei{\Psi}{s}$ are such that Fourier transform in the variable $t$ is well defined, the relation between the PDEs and corresponding ODEs follows by Fourier inversion. Thus, any $r$-weighted $L^2$ estimates for the radial ODEs \eqref{eq:radial-ODE-alpha-intro-basic} and \eqref{eq:radial-ODE-Psi-intro-basic} which are uniform in the \textit{real} frequency parameters $(\omega,m,\Lambda)$ immediately imply analogous estimates in physical space by Plancherel's theorem\footnote{Note that, though the separable anzatz allows us to take $\omega\in\mathbb{C}$,  it is not clear that solutions to the resulting angular ODE always form a complete basis for the space of suitably regular spin-weighted functions.}. In order to justify the Fourier transform, we will use a continuity argument in the Kerr parameter $a$ in the style of that developed by Dafermos, Rodnianski and the first author in \cite{Dafermos2016b} for the $s=0$ case. We flesh out this, and other by now standard physical space arguments, in our upcoming \cite{SRTdC2022}. In all the steps mentioned, it will be technically convenient to apply cutoffs to our PDEs, thus it is useful to also consider inhomogenous versions of \eqref{eq:Teukolsky-equation-intro} and \eqref{eq:transformed-equation-intro-basic} (c.f.\ \eqref{eq:teukolsky-alpha} and Proposition~\ref{prop:transformed-system}), as well as \eqref{eq:radial-ODE-alpha-intro-basic} and \eqref{eq:radial-ODE-Psi-intro-basic}, among others. For the remainder of this section, we will focus on obtaining estimates on these, possibly inhomogeneous, ODEs which are uniform in the \textit{real} frequency parameters $(\omega,m,\Lambda)$.


The first step is to establish \textit{quantitative mode stability} (see already Section~\ref{sec:intro-separability}).  Let us begin with a \textit{qualitative} version. A mode solution is a separable solution of the homogeneous Teukolsky~\eqref{eq:Teukolsky-equation-intro} with finite initial energy on a suitable spacelike hypersurface which is regular at the event horizon. In the physics literature, mode stability is often taken to mean that there are no mode solutions which are exponentially growing in time, i.e.\ which have $\Im\omega>0$.  In the context of our analysis, we are most interested in the \textit{real} $\omega$ case: we say that \eqref{eq:Teukolsky-equation-intro} is \textit{real} modally stable if there are no mode solutions which are bounded but not decaying in time, i.e.\ with $\omega\in\mathbb{R}\backslash\{0\}$. The conservation law for \eqref{eq:radial-ODE-alpha-intro-basic} (see already Section~\ref{sec:intro-energy-identity}) can be used to infer such a non-existence result in a subset of the $\omega$ upper half-plane and real axis, but not its entirety unless $a= 0$. For instance, when $\omega$ is real and $|s|\in\lp\{0, 1, 2\rp\}$, \textit{superradiance} of the frequency parameters, i.e.\
\begin{align*}
\omega\lp(\omega-\frac{am}{2Mr_+}\rp)<0\,,
\end{align*}
prevents us from inferring mode stability from the conservation law. Nevertheless, the Teukolsky equation \eqref{eq:Teukolsky-equation-intro} is indeed both real and complex modally stable  for subextremal $|a|<M$ \cite{Whiting1989,Shlapentokh-Rothman2015,Andersson2017,TeixeiradaCosta2019} and even for extremal $|a|=M$ \cite{TeixeiradaCosta2019} Kerr! We remark that this property seems very specific to the algebraic structure of the Kerr metric (see \cite{Shlapentokh-Rothman2014,Moschidis2017b}) and similar metrics (see, e.g.\ \cite{Cvetic2020}).

In order to state a \textit{quantitative} version of mode stability, let us introduce some notation. Denote by $\swei{\alpha}{s}_{\mc{H}^\mp }$ (resp.\ $\swei{\alpha}{s}_{\mc{I}^\mp }$) the radial components of separable solutions of the homogeneous Teukolsky equation~\eqref{eq:Teukolsky-equation-intro} which are characterized by some real $(\omega,m,\Lambda)$ and which are regular, with respect to the algebraically special frame in which the Teukolsky variable is defined, at the past or future event horizon  (resp.\ past or future null infinity). We define the Wronskian
$$\swei{\mathfrak{W}}{s}(\omega,m,\Lambda):=\Delta^{1+s}\lp[\frac{d}{dr}\lp(\swei{\alpha}{s}_{\mc{I}^+}\rp)\cdot \swei{\alpha}{s}_{\mc{H}^+}-\frac{d}{dr}\lp(\swei{\alpha}{s}_{\mc{H}^+}\rp)\cdot \swei{\alpha}{s}_{\mc{I}^+}\rp]\,.$$ 
\noindent As a consequence of the mode stability property of \eqref{eq:Teukolsky-equation-intro}, proved in \cite{Whiting1989,Shlapentokh-Rothman2015,Andersson2017,TeixeiradaCosta2019}, one has that $\swei{\mathfrak{W}}{s}\neq 0$; more concretely, one in fact has
\begin{arabictheorem}[Quantitative real mode stability] Fix $s\in\{0,\pm 1,\pm 2\}$, $M>0$ and $|a|<M$. For each set of real frequency parameters $\mc{A}$ such that $C_{\mc{A}}:=\sup_{\mathcal{A}} \lp(|\omega|+|\omega|^{-1}+|m|+|\Lambda|\rp)<\infty$,
there is a constant $B(a,M,|s|,C_{\mc{A}})>0$ such that 
$$\lp|\swei{\mathfrak{W}}{s}\rp|^{-1}\leq B(a,M,|s|,C_{\mc{A}})<\infty \text{~~for~}(\omega,m,\Lambda)\in\mc{A}\,.$$ 
The constant $B(a,M,|s|,C_{\mc{A}})$ can be {\normalfont explicitly} computed in terms of $a$, $M$, $s$ and $C_{\mc{A}}$ \cite{Shlapentokh-Rothman2015,TeixeiradaCosta2019}. The result also holds for $|a|=M$ except for $\omega\approx \frac{am}{2Mr_+}$. 
\label{thm:mode-stability-intro-AB}
\end{arabictheorem}


While Theorem~\ref{thm:mode-stability-intro-AB} may be shown to imply, using standard ODE theory, estimates for homogeneous and inhomogeneous versions of the radial ODEs~\eqref{eq:radial-ODE-alpha-intro-basic} and \eqref{eq:radial-ODE-Psi-intro-basic} in any set of real frequencies $\mc{A}$ as in its statement (see e.g.\ \cite{Shlapentokh-Rothman2015,TeixeiradaCosta2019}), it cannot be used to probe high frequencies, $|\omega|+|m|+|\Lambda|\to \infty$, nor the zero-frequency limit, $\omega\to 0$. This is precisely the  goal of the current paper:

\begin{arabictheorem} \label{thm:frequency-estimates-big} Fix $s\in\{0,\pm 1,\pm 2\}$, $M>0$ and $|a|<M$. Consider a separable solution, $\swei{\Psi}{s}$, to \eqref{eq:transformed-equation-intro-basic} which arises from a separable solution of the homogeneous or inhomogeneous Teukolsky equation~\eqref{eq:Teukolsky-equation-intro} and is characterized by real $(\omega,m,\Lambda)$ with $\omega\neq 0$. By abuse of notation, write $\swei{\Psi}{s}$ for its radial component, which then satisfies the  ODE \eqref{eq:radial-ODE-Psi-intro-basic}.
\begin{enumerate}[font=\normalfont,label=\bfseries\Alph*., ref=\Alph*]
\item {\normalfont \textbf{Integrated estimates}}.  One may fix a sufficiently large $R$ and sufficiently small $\epsilon$ so that the following holds. There is a value $r_{\rm trap}(\omega,m,\Lambda)$ which is either 0 or satisfies the uniform bounds $|r_{\rm trap}|+|r_{\rm trap}-r_+|^{-1}\leq B(a,M,s)$ such that for all $(\omega,m,\Lambda)$, any smooth $\swei{\Psi}{s}(r)$ arising from an \myuline{inhomogeneous} Teukolsky equation as described and which has appropriate boundary conditions as $r\to r_+,r\to \infty$ satisfies
\begin{gather}
\begin{gathered}
\lp(\omega-\frac{am}{2Mr_+}\rp)^2\lp|\swei{\Psi}{s}\rp|^2_{r=r_+}+\omega^2\lp|\swei{\Psi}{s}\rp|^2_{r=\infty}\\
+\int_{r_+}^\infty \frac{\Delta}{r^4}\lp\{\lp|\frac{\Delta}{r^2}\p_r\swei{\Psi}{s}\rp|^2+\lp[\lp(\omega^2+r^{-1}\Lambda\rp)\lp(1-\frac{r_{\rm trap}(\omega,m,\Lambda)}{r}\rp)+r^{-1}(s^2+r^{-1})\rp]\lp|\swei{\Psi}{s}\rp|^2\rp\}dr\\
\leq B(a,M,s) \sum_{k=0}^{|s|} \lp\{\int_{r_+}^\infty\mathfrak{G}_{(k)}\cdot \lp(\uppsi_{(k)},\uppsi_{(k)}'\rp) dr+\int_{r_++\epsilon}^{R}\lp(|\uppsi_{(k)}|^2+|\uppsi_{(k)}|^2\rp)dr^*\rp\}\,. 
\end{gathered}\label{eq:intro-main-thm-A}
\end{gather}
Here, $\mathfrak{G}_{(k)}$ is the inhomogeneity in the equation for the $k$th transformed variable and $\mathfrak{G}_{(k)}\cdot \lp(\uppsi_{(k)},\uppsi_{(k)}'\rp)$ is a nontrivial function of the inhomogeneity and the transformed variable. 
 
\item {\normalfont \textbf{Scattering estimates}}. For any real frequency triple $(\omega,m,\Lambda)$ with $\omega\neq 0,\frac{am}{2Mr_+}$, any smooth $\swei{\Psi}{s}(r)$ arising from a \myuline{homogeneous} Teukolsky equation  admits a decomposition
\begin{align*}
\swei{\Psi}{s}= \swei{A}{s}_{\mc{H}^+}\swei{\Psi}{s}_{\mc{H}^+}+\swei{A}{s}_{\mc{H}^-}\swei{\Psi}{s}_{\mc{H}^-}=\swei{A}{s}_{\mc{I}^+}\swei{\Psi}{s}_{\mc{I}^+}+\swei{A}{s}_{\mc{I}^-}\swei{\Psi}{s}_{\mc{I}^-}
\end{align*}
for coefficients $\swei{A}{s}_{\mc{H}^\mp}(\omega,m,\Lambda),\swei{A}{s}_{\mc{I}^\mp}(\omega,m,\Lambda)\in\mathbb{C}$, and where $\swei{\Psi}{s}_{\mc{H}^\pm}$ and $\swei{\Psi}{s}_{\mc{I}^\pm}$ denote normalized transformed variables arising from $\swei{\upalpha}{s}_{\mc{H}^\pm}$ and $\swei{\upalpha}{s}_{\mc{I}^\pm}$, respectively, and satisfies
\begin{align}
\begin{split}
&\omega^2\frac{\mathfrak{C}_s}{\mathfrak{D}_s^{\mc I}}\lp|\swei{A}{s}_{\mc{I}^+}\rp|^2+\lp(\omega-\frac{am}{2Mr_+}\rp)^2\lp|\swei{A}{s}_{\mc{H}^+}\rp|^2\\
&\quad\leq B(a,M,s)\lp[\omega^2\lp|\swei{A}{s}_{\mc{I}^-}\rp|^2+\lp(\omega-\frac{am}{2Mr_+}\rp)^2\frac{\mathfrak{C}_s}{\mathfrak{D}_s^{\mc H}}\lp|\swei{A}{s}_{\mc{H}^-}\rp|^2\rp]\,,\,\, s\leq 0,\\
&\omega^2\lp|\swei{A}{s}_{\mc{I}^+}\rp|^2+\lp(\omega-\frac{am}{2Mr_+}\rp)^2\frac{\mathfrak{C}_s}{\mathfrak{D}_s^{\mc H}}\lp|\swei{A}{s}_{\mc{H}^+}\rp|^2\\
&\quad\leq B(a,M,s)\lp[\omega^2\frac{\mathfrak{C}_s}{\mathfrak{D}_s^{\mc I}}\lp|\swei{A}{s}_{\mc{I}^-}\rp|^2+\lp(\omega-\frac{am}{2Mr_+}\rp)^2\lp|\swei{A}{s}_{\mc{H}^-}\rp|^2\rp]\,, \,\, s>0,
\end{split}\label{eq:intro-main-thm-B}
\end{align}
where the Teukolsky--Starobinsky constants $\mathfrak{C}_s$, and the $\mathfrak{D}_s^{\mc I}$ and $\mathfrak{D}_s^{\mc H}$ are explicit, real, non-negative constants which depend only on the black hole parameters $(a,M)$, the spin $|s|$, and the frequency triple $(\omega,m,\Lambda)$ and which may only vanish on a discrete set of real frequencies.
\end{enumerate}
\end{arabictheorem}

\begin{remark*} Some remarks regarding the main theorem of this paper, Theorem~\ref{thm:frequency-estimates-big}, are in order.
\begin{enumerate}[font=\normalfont,label=\bfseries\Alph*.,ref=\Alph*]
\item \label{it:remark-intro-A} Theorem~\ref{thm:frequency-estimates-big}A, along with Theorem~\ref{thm:mode-stability-intro-AB}, is the main ingredient in the proof of Theorem~\ref{thm:bddness-decay-big}, see already Section~\ref{sec:intro-to-freq-space}.

\item \label{it:remark-intro-B} While~\eqref{eq:intro-main-thm-B} is unconditional, one can only establish direct control over the boundary terms with the usual frequency weights if the Teukolsky--Starobinsky and boundary term constants satisfy
\begin{equation}
\mathfrak{C}_s\sim \mathfrak{D}_s^{\mc I}\sim \mathfrak{D}_s^{\mc H}\sim 1\,. \label{eq:constants-are-comparable}
\end{equation}
If $a=0$, or $s=0$, by the explicit formulas for these constants, one easily has \eqref{eq:constants-are-comparable}. For general subextremal  $|a|<M$ parameters and $s\in\{\pm 1,\pm 2\}$, \eqref{eq:constants-are-comparable} may fail to hold in a certain set of real frequencies which depends on the Kerr parameters $(a,M)$ and spin $s$, and which may be nontrivial even in the very slowly rotating $|a|\ll M$ case. We note, however, that \eqref{eq:constants-are-comparable} holds in a neighborhood of the zero time frequency $|\omega|=0$. Thus, in the zero frequency limit $\omega\to 0$, \eqref{eq:intro-main-thm-B} indeed implies effective control over the weighted boundary terms of $\swei{\Psi}{s}$.
\end{enumerate}

We further note that, while the boundary terms in the estimates \eqref{eq:intro-main-thm-A} and \eqref{eq:intro-main-thm-B} are not well defined for $\omega=0$ nor, in the latter case, $\omega=\frac{am}{2Mr_+}$, for applications to the Teukolsky equation \eqref{eq:Teukolsky-equation-intro} and the transformed system, such as equation \eqref{eq:transformed-equation-intro-basic}, in physical space, by Plancherel, it suffices to obtain estimates for $\omega\in\mathbb{R}\backslash\lp\{0, \frac{am}{2Mr_+}\rp\}$ which are \textit{uniform} in the limit $\omega\to 0$ and, if $|a|<M$,  $\omega\to \frac{am}{2Mr_+}$. Estimates \eqref{eq:intro-main-thm-A} and \eqref{eq:intro-main-thm-B} indeed have this property. Finally, we remark that estimate \eqref{eq:intro-main-thm-A} is uniform in the extremal Kerr limit $|a|\to M$ as long as the frequency parameters are taken to be outside a neighborhood of $\omega=\frac{am}{2Mr_+}$.
\end{remark*}

The proof of Theorem~\ref{thm:frequency-estimates-big} not only requires the $s=0$ scalar wave equation techniques developed previously, but also essential features of the $s\neq 0$ Teukolsky equation which have not yet been explored in this setting. Let us give a telegraphic account of the main obstacles to the proof of our Theorem~\ref{thm:frequency-estimates-big} when compared to the $a=0$ or $s=0$ cases (see already Section~\ref{sec:intro-proof-freq-estimates}).

The main challenge we face is that, when $a\neq 0$ and $s\neq 0$, the equation \eqref{eq:transformed-equation-intro-basic} for $\swei{\Psi}{s}$ is coupled to those of the $k<|s|$ transformed variables through $\swei{\mathfrak{J}}{s}$: one needs to obtain a hierarchy of estimates for the transformed variables. While this idea is not new, and has been implemented successfully in the very slowly rotating case, the hierarchy must be much more refined in order to be of any use in the general subextremal $|a|<M$ Kerr case. In fact, in the very slowly rotating case $|a|\ll M$, since $\swei{\mathfrak{J}}{s}$ comes multiplied by a small parameter, it was shown by Dafermos, Holzegel and Rodnianski in \cite{Dafermos2017} for $s=\pm 2$ that such a hierarchy can be constructed using only the transport equations connecting the $k$th to the $(k+1)$th transformed variable. Hence, the analysis of \eqref{eq:transformed-equation-intro-basic} for $|a|\ll M$ is essentially reduced to that of studying the $s=0$ scalar wave equation in this setting.  

For general subextremal Kerr black holes, however, this is no longer possible: \textbf{the coupling terms are often not lower order perturbations} in the radial ODE~\eqref{eq:radial-ODE-Psi-intro-basic}. The consequences can be felt both at the level of bulk estimates such as in \eqref{eq:intro-main-thm-A} and at the level of boundary estimates such as in \eqref{eq:intro-main-thm-A} and \eqref{eq:intro-main-thm-B}. 

For instance, while in obtaining bulk estimates such as that in \eqref{eq:intro-main-thm-A}, we rely on frequency-localized virial currents which are identical to those considered for $s=0$, the analogy with the scalar wave equation ends there. Indeed, for any frequency triple $(\omega,m,\Lambda)$ for which the coupling $\swei{\mathfrak{J}}{s}$ is not obviously lower order, we need precise estimates on each of the transformed variables $\swei{\uppsi}{s}_{(k)}$. To obtain them, we rely on an analysis, through frequency-localized virial currents, of the entire transformed system at once, including the radial ODEs for $\swei{\uppsi}{s}_{(k)}$ (see already \eqref{eq:radial-ODE-uppsi-intro}). Note that these retain the bad properties of Teukolsky radial ODE~\eqref{eq:radial-ODE-alpha-intro-basic}, compared to the transformed radial ODE~\eqref{eq:radial-ODE-Psi-intro-basic}, to which we have alluded earlier. It is by exploiting the underlying structure of the system that we are finally able to establish the bulk estimate of Theorem~\ref{thm:frequency-estimates-big}A.

Let us now turn to the issue of controlling boundary terms. Any Killing energy current, i.e.\ any energy current associated to Killing vector fields, necessarily fails to be conserved by \eqref{eq:transformed-equation-intro-basic} when $a,s\neq 0$,  due to the coupling terms $\swei{\mathfrak{J}}{s}$. Moreover, it cannot even be seen as an ``approximate'' conservation law unless $|a|\ll M$. We can build an energy current which is blind to $\swei{\mathfrak{J}}{s}$ from  the natural conservation law for the Teukolsky radial ODE~\eqref{eq:radial-ODE-alpha-intro-basic}, as it implies an identity for the radial component of $\swei{\Psi}{s}$: e.g., if $s\leq 0$,
\begin{equation}
\omega^2\frac{\mathfrak{C}_s}{\mathfrak{D}_s^{\mc I}}\lp|\swei{A}{s}_{\mc{I}^+}\rp|^2+\omega\lp(\omega-\frac{am}{2Mr_+}\rp)\lp|\swei{A}{s}_{\mc{H}^+}\rp|^2=\omega^2\lp|\swei{A}{s}_{\mc{I}^-}\rp|^2+\omega\lp(\omega-\frac{am}{2Mr_+}\rp)\frac{\mathfrak{C}_s}{\mathfrak{D}_s^{\mc H}}\lp|\swei{A}{s}_{\mc{H}^-}\rp|^2\,, \label{eq:Wronskian-conservation-Psi-intro}
\end{equation}
in the notation of Theorem~\ref{thm:frequency-estimates-big}. We refer to  \eqref{eq:Wronskian-conservation-Psi-intro} as the Teukolsky--Starobinsky energy identity, because it makes use of the so-called Teukolsky--Starobinsky identities \cite{Teukolsky1974,Starobinsky1973,Kalnins1989} (see already Section~\ref{sec:energy-identity-transformed}). This explains the presence of the  the Teukolsky--Starobinsky constant $\mathfrak{C}_s$; the constants $\mathfrak{D}_s^{\mc I}$ and $\mathfrak{D}_s^{\mc H}$ simply relate the boundary terms of the Teukolsky and the transformed variables. As we have mentioned in Remark \ref{it:remark-intro-B}, for general $s\in\{0,\pm 1, \pm 2\}$ and $|a|<M$, there is a set of frequencies $(\omega, m,\Lambda)$ for which \eqref{eq:constants-are-comparable} may not hold, hence for which the Teukolsky--Starobinsky energy identity  does not ensure control over one of the past and one of the future boundary terms: e.g.\ we may lose control over the terms $\omega^2\lp|\swei{A}{s}_{\mc{I}^+}\rp|^2$ and $\omega\lp(\omega-\frac{am}{2Mr_+}\rp)\lp|\swei{A}{s}_{\mc{H}^-}\rp|^2$ in \eqref{eq:Wronskian-conservation-Psi-intro} for certain real $(\omega,m,\Lambda)$. However, one can attempt to recover the ``lost'' boundary term using the Killing energy associated to $\p_t$: though it is not conserved by the Teukolsky equation~\eqref{eq:Teukolsky-equation-intro} nor by the transformed equation~\eqref{eq:transformed-equation-intro-basic} when $s\neq 0$, for some frequencies $(\omega, m,\Lambda)$ we can appeal to the hierarchy of estimates we develop in this paper to control the errors which arise. It is by this procedure that in \eqref{eq:intro-main-thm-A} we avoid the dependence on the constants $\mathfrak{C}_s$, $\mathfrak{D}_s^{\mc I}$ and $\mathfrak{D}_s^{\mc H}$ altogether.


In the remainder of this introduction, we will try to unpack many of the ideas to which we have alluded above, as well as explain our strategy to address the main new difficulties. We begin, in Section~\ref{sec:intro-previous-work}, by reviewing the previous work concerning the Teukolsky equation~\eqref{eq:Teukolsky-equation-intro} with a focus on Question~\ref{qn:Teukolsky-physical-space}. To tackle Question~\ref{qn:Teukolsky-physical-space} in its full generality, as we have said, we will actually consider the transformed system of PDEs, including \eqref{eq:transformed-equation-intro-basic}, discussed above; this will be discussed more carefully in Section~\ref{sec:intro-transformed}.  In Section~\ref{sec:intro-separability}, we discuss the formal separability properties of the Teukolsky equation and of the transformed system: we review  Theorem~\ref{thm:mode-stability-intro-AB} and other aspects of the Teukolsky and transformed system which separability allows us to easily probe, and explain how it can be used to reduce Theorem~\ref{thm:bddness-decay-big} to Theorem~\ref{thm:frequency-estimates-big}A. Finally, Section~\ref{sec:intro-proof-freq-estimates} is meant to guide the reader through the key points in our proof of Theorem~\ref{thm:frequency-estimates-big}: we first give an overview of the ideas used to analyze the transformed operator $\swei{\mathfrak{R}}{s}$, then we focus on the novel ideas required to deal with the coupling terms $\swei{\mathfrak{J}}{s}$ on the right hand side of \eqref{eq:transformed-equation-intro-basic}.

\subsection{Previous work on the Teukolsky equation}
\label{sec:intro-previous-work}

Our work comes after more than a decade of intense research on stability of Schwarzschild and Kerr, of which we now give a brief account, using the Teukolsky equation~\eqref{eq:Teukolsky-equation-intro} as a reference.

\subsubsection{Schwarzschild \texorpdfstring{($a=0$)}{} and subextremal \texorpdfstring{($|a|<M$)}{} Kerr}

The wave equation case, $s=0$, which can already be seen as the ``poor man's stability problem,'' was settled in the full subextremal range $|a|<M$ by Dafermos, Rodnianski and the first author in \cite{Dafermos2016b}, where boundedness and decay of \eqref{eq:Teukolsky-equation-intro} were established, and \cite{Dafermos2014}, where a scattering theory was developed (see also \cite{Georgescu2017} for an earlier, fixed $m$, scattering theory for very slowly rotating Kerr-dS). The first result has also been generalized to the subextremal Kerr--Newman family of charged, rotating black holes in \cite{Civin2014}.  We also note the earlier results on the Schwarzschild family which are summarized in the lecture notes \cite{Dafermos2008},  and their extensions to very slowly rotating $|a|\ll M$ Kerr in the independent works \cite{Dafermos2010},  \cite{Tataru2011} and \cite{Andersson2015}. For applications to nonlinear wave equations, see for instance \cite{Luk2013,Ionescu2015,Lindblad2018,Lindblad2020,Oliver2020}.

For electromagnetic perturbations, $s=\pm 1$, the work of \cite{Pasqualotto2016} to which we have alluded was preceded by a proof, using different methods, of linear stability to perturbations of Cauchy data by Blue \cite{Blue2008} (see also \cite{Sterbenz2014}). We also note the extensions to very slowly rotating Kerr \cite{Andersson2015a,Ma2017}. Pasqualotto's work on the $s=\pm 1$ the linear setting was the basis for his resolution of a nonlinear problem, the analogue of Question~\ref{qn:Kerr-perturbations-nonlinear-2}\ref{it:Kerr-perturbations-nonlinear-2-Cauchy} for the Maxwell--Born--Infeld model of nonlinear electromagnetism \cite{Pasqualotto2017}. Moreover, together with \cite{Dafermos2016a}, to which we return below, it was an important ingredient in the proof of linear stability for the charged version of Schwarzschild ($a=0$), known as Reissner--Nordstr\"{o}m spacetime, in \cite{Giorgi2019a,Giorgi2020}.

Gravitational perturbations, $s=\pm 2$ have received the most attention in recent years. As we mentioned, boundedness and decay of \eqref{eq:Teukolsky-equation-intro} was first established in \cite{Dafermos2016a}, who also used this fact to derive linear stability for the entire system of perturbations in a well-posed double null gauge. Later, Johnson
established the linear stability of Schwarzschild in a generalized wave gauge \cite{Johnson2018} (see also the more recent \cite{Hung2018a,Hung2019}); approaches based on metric perturbations were developed in  \cite{Hung2017}. More recently, boundedness and decay of $\swei{\upalpha}{\pm 2}$ was generalized to very slowly rotating Kerr by Dafermos, Holzegel and Rodnianski \cite{Dafermos2017} and Ma \cite{Ma2017a} independently.  \cite{Andersson2019} have constructed a gauge in which these results could be exploited to obtain estimates for the rest of the metric components in the linearized setting. Using different methods, \cite{Hafner2019} have also shown decay for linearized metric perturbations in a generalized wave gauge for very slowly rotating, $|a|\ll M$, Kerr spacetimes. We note also the papers \cite{Finster2009,Finster2017}.

For results concerning the half-integer spin cases $s=\pm 1/2,\pm 3/2$, see \cite{Hafner2004,Mason1999} and references therein, respectively.

\subsubsection{A remark on extremal \texorpdfstring{($|a|=M$)}{} Kerr}
\label{sec:intro-previous-work-extremal}

Conspicuously absent from most of the previous discussion is the extremal Kerr case, $|a|=M$. In fact, as we have said, it is not clear whether in this limiting case Question~\ref{qn:Teukolsky-physical-space}, let alone Question~\ref{qn:Kerr-perturbations-nonlinear-2}, should have a positive or negative answer: already in the linearized setting, a mixture of stability and instability results have been obtained. As we mentioned, in \cite{Aretakis2012}, Aretakis showed that higher order transversal derivatives of axisymmetric solutions to the wave equation along the event horizon blow up, a phenomenon that is known as Aretakis instability; the result was extended to the Teukolsky equation \eqref{eq:Teukolsky-equation-intro} for general integer spin in \cite{Lucietti2012}.    In recent work, Casals, Gralla and Zimmerman \cite{Casals2016} suggest that outside of axisymmetry the Aretakis instability might set in even at a lower level of differentiability and, for solutions arising from smooth initial data supported away from the horizon, compute the expected decay rate along and away from the event horizon (the latter originally obtained in \cite{Andersson2000}), suggesting it is significantly slower than in the subextremal, $|a|<M$, case. However, axisymmetric solutions to the wave equation, $s=0$, are bounded and do decay, albeit at a slower rate than in the subextremal, $|a|<M$, case \cite{Aretakis2012a}. Moreover, mode stability has been shown by the second author \cite{TeixeiradaCosta2019} to hold for general solutions to \eqref{eq:Teukolsky-equation-intro} of any spin.

Recent work has largely focused on a spherically symmetric toy model for extremal Kerr, extremal Reissner--Nordstr\"{o}m, in the case of the wave equation, $s=0$. Aretakis in \cite{Aretakis2011,Aretakis2011b} and Aretakis, Angelopoulos and Gajic in \cite{Angelopoulos2019a} have answered an analogue of Question~\ref{qn:Teukolsky-physical-space} for $s=0$ in both cases A and B, respectively. As expected for Kerr, the sharp decay rates for the wave equation, obtained by the aforementioned authors in \cite{Angelopoulos2018a}, are significantly slower than for subextremal Reissner--Nordstr\"{o}m. It is unclear what the implications of the slow decay and Aretakis instability in the linear setting could be for the full nonlinear problem as described in Question~\ref{qn:Kerr-perturbations-nonlinear-2};  see however  \cite{Angelopoulos2019} for stability and instability results concerning nonlinear wave equations satisfying the null condition on the extremal Reissner--Nordstr\"{o}m black hole background and the numerics \cite{Murata2013} concerning nonlinear stability of the extremal Reissner--Nordstr\"{o}m exterior in the Einstein-scalar field model.

\subsection{From the Teukolsky equation to a transformed system of PDEs}
\label{sec:intro-transformed}

 As the Teukolsky equation~\eqref{eq:Teukolsky-equation-intro} is a wave-type equation, one is tempted to apply the techniques that have been developed to study the $s=0$ case, $\Box_g\swei{\alpha}{0}=0$, on black hole spacetimes in the past decade. However, there are some fundamental differences. For instance, as \eqref{eq:Teukolsky-equation-intro} does not arise from a Lagrangian formulation, there is no clear physical space conserved energy  when $s\neq 0$ (though there is one in the separated picture; see already Section~\ref{sec:intro-energy-identity}). Moreover, the first order terms present in \eqref{eq:Teukolsky-equation-intro} have poor decay in $r$ as $r\to\infty$ and $r\to r_+$. These novel features for $s\neq 0$ represent substantial difficulties in order for us to apply to \eqref{eq:Teukolsky-equation-intro} the methods which have been developed over the past two decades to analyze wave equations on black hole spacetimes. A natural question that arises, then, is whether one can substitute the $s\neq 0$ Teukolsky equation by another PDE which behaves more like the $s=0$ case.

\subsubsection{A transformed system in Schwarzschild \texorpdfstring{($a=0$)}{}}

Let us begin by reviewing how the difficulties discussed in the previous paragraphs are overcome for $a=0$. In \cite{Dafermos2016a}, the authors substitute the $s=\pm 2$ Teukolsky equation by the equation satisfied by a variable $\swei{\Psi}{\pm 2}$ which is obtained from $\swei{\upalpha}{\pm 2}$ by taking two appropriately $r$-weighted derivatives in a suitable null direction. As Chandrasekhar \cite{Chandrasekhar1975a} had observed in frequency space (\cite{Dafermos2016a} generalized his findings to physical space), $\swei{\Psi}{\pm 2}$ satisfies a decoupled equation in Schwarzschild ($a=0$), known as the Regge--Wheeler equation \cite{Regge1957}:
\begin{equation}
\lp[\frac12 \lp(L\underline{L}+\underline{L}L\rp) + \frac{\Delta}{r^4}\mathring{\slashed{\triangle}}^{[\pm 2]}+\frac{2\Delta}{r^4}\lp(2 -\frac{ 3M}{r}\rp)\rp]\swei{\Psi}{\pm 2}=0 \label{eq:RW-Schw-2}
\end{equation}
where $L$ and $\uL$ are linearly independent null vectors and $\mathring{\slashed{\triangle}}^{[\pm 2]}$ is the $(\pm 2)$-spin-weighted laplacian (see already Section~\ref{sec:prelims-vector-fields-operators}). Equation~\eqref{eq:RW-Schw-2} does not involve any first order terms in $\p_t$ or $\p_\phi$, c.f.\ \eqref{eq:transformed-equation-intro-basic}; futhermore, the zeroth order terms are of a repulsive nature, which means the Regge--Wheeler equation behaves even better than the wave equation. A similar situation occurs in the $s=\pm 1$ case:  $\swei{\phi}{\pm 1}$, obtained from $\swei{\upalpha}{\pm 1}$ by taking one appropriately $r$-weighted derivative in a suitable null direction satisfies a good decoupled wave equation, known as the Fackerell--Ipser \cite{Fackerell1972} equation (see also \cite{Pasqualotto2016}). 

In both \cite{Pasqualotto2016,Dafermos2016a}, where we recall Question~\ref{qn:Kerr-perturbations-nonlinear-2}\ref{it:Teukolsky-physical-space-bddness-decay} is settled for $|s|=1,2$, respectively, in the Schwarzschild $a=0$ case, estimates are first obtained for $\swei{\Psi}{s}$ in an analogous way to the wave equation case; then, they are integrated along null directions to provide estimates on $\swei{\upalpha}{s}$, $|s|=1,2$.

\subsubsection{A transformed system for general subextremal Kerr \texorpdfstring{($|a|< M$)}{}}

In the general subextremal $|a|< M$ Kerr case, it would be convenient to obtain a transformation, similarly to the $a=0$ case, which eliminates the first order terms in \eqref{eq:Teukolsky-equation-intro}. If we assume that a new variable $\swei{\Psi}{s}$ is again obtained by taking $|s|$ $r$-weighted derivatives along a suitable null direction, estimates for $\swei{\upalpha}{s}$ can be obtained by integrating the transport equations relating $\swei{\Psi}{s}$ to the Teukolsky variable, as in the Schwarzschild case. Thus, recall, from the preamble to this introduction, \eqref{eq:def-upppsi-k-intro-basic} i.e\
$$\swei{\uppsi}{s}_{(0)}=j_{s,0}^{-1}\swei{\upalpha}{s}\,,\qquad \swei{\uppsi}{s}_{(k+1)}=j_{s,k}^{-1}\mc{L}\swei{\uppsi}{s}_{(k)}\,, \quad 0\leq k \leq |s|-1\,, \qquad \swei{\Psi}{s}:=\swei{\uppsi}{s}_{|s|}\,,$$ 
for $\mc{L}$ representing a null direction and $j_{s,k}=j_{s,k}(r)$ be some $r$-weights for $k\in\{0,...,|s|\}$.  One can easily check that there is no choice of $r$-weights $j_{s,k}$ which would lead to a full decoupling of the equation for $\swei{\Psi}{s}$ from the equation for $\swei{\upalpha}{s}$ unless $a=0$. As we must accept some degree of coupling between the equation for $\swei{\Psi}{s}$ and those for $\swei{\uppsi}{s}_{(k)}$, $0\leq k <|s|$, the goal now is to choose our weights $j_{s,k}(r)$ in a way which, first and foremost, produces an equation for $\swei{\Psi}{s}$ without the troublesome $\p_t$ and $\p_\phi$ first order derivatives we encounter in the Teukolsky equation~\eqref{eq:Teukolsky-equation-intro} and, secondly, ``minimizes'' the degree of coupling.

In the present paper, we demonstrate that the choice
\begin{align}
\uppsi_{(0)}^{[s]}=(r^2+a^2)^{-|s|+1/2}\Delta^{|s|/2(1+\sign s)}\swei{\upalpha}{s}\,;\qquad\frac{\Delta}{(r^2+a^2)^2}\uppsi_{(k+1)} = \mc{L} \uppsi_{(k)}^{[s]}\,,\,\, k=0,...,|s|\,, \label{eq:def-upppsi-k-intro}
\end{align}
where $\mc{L}=L$ when $s<0$ and $\mc{L}=\uL$ when $s>0$ is sufficient to satisfy the first requirement, and to have the coupling to the variables indexed by $k=0,...,|s|$ occur only through dependence on those variables and their azimuthal, $\p_\phi$, derivative. Indeed, we show that if, for any $s\in\mathbb{Z}$, $\swei{\upalpha}{s}$ solves the homogeneneous Teukolsky equation~\eqref{eq:Teukolsky-equation-intro}, then the partial differential equation satisfied by $\swei{\Psi}{s}$ is (c.f.\ the compressed preliminary version given in \eqref{eq:transformed-equation-intro-basic})
\begin{equation}
\swei{\mathfrak{R}}{s}\swei{\Psi}{s}=a\swei{\mathfrak{J}}{s}\,,\qquad \swei{\mathfrak{J}}{s}=\sum_{k=0}^{|s|-1} \frac{\Delta}{(r^2+a^2)^2} \lp[ c^{\Phi}_{s,\,|s|,\,k}(r) \Phi+c^\mr{id}_{s,\,|s|,\,k}(r)\rp]\swei{\uppsi}{s}_{(k)}\,, \label{eq:transformed-equation-intro}
\end{equation}
where $c_{s,\,|s|,\,k}^\Psi$ and $c_{s,\,|s|,\,k}^{\rm id}$ are bounded functions in $[r_+,\infty)$ that can be explicitly computed by a recursive formula. The operator $\swei{\mathfrak{R}}{s}$, given by
\begin{align}
\begin{split}
\swei{\mathfrak{R}}{s}&:=\frac12 \lp(L\underline{L}+\underline{L}L\rp) +\frac{\Delta}{(r^2+a^2)^2}\lp(\mathring{\slashed{\triangle}}^{[s]}-a^2\sin^2\theta\p_t^2-2a\p_t\p_\phi+2ias\cos\theta\p_t+ s^2\rp) \\
&\qquad+\frac{\Delta\lp[(1-2s^2)a^2\Delta + 2(1-s^2)Mr(r^2-a^2)\rp]}{(r^2+a^2)^4}\,, \end{split}\label{eq:transformed-R-intro}
\end{align}
where $\mathring{\slashed{\triangle}}^{[s]}$ is the $s$-spin-weighted laplacian (see already Section~\ref{sec:prelims-vector-fields-operators}), can truly be seen as an analogue of $\Box_g$ (acting on the radiation field) for spin-weighted functions: it reduces to the wave operator acting on the radiation field when $s=0$, it has no first order terms, and its zeroth order terms, which are of a repulsive nature, are in fact better behaved than for $s=0$. The coupling terms are multiplied by $a$, so we recover the decoupled $a=0$ transformed equation in \eqref{eq:RW-Schw-2} and, moreover, they have ``good'' $r$-weights, that decay with $\Delta$ as $r\to r_+$ and with $r^{-2}$ as $r\to \infty$. Our framework is in fact a generalization to $s\in\mathbb{Z}$ of that developed in  \cite{Dafermos2017} for $s=\pm 2$.

We note the following trade-off when $a\neq 0$: if we opt for physical space transformations of the Teukolsky equation \eqref{eq:Teukolsky-equation-intro} similar to those for $a=0$, we have to accept in \eqref{eq:transformed-equation-intro} some degree of coupling, absent if $a=0$, to the Teukolsky equation. Conversely, if our main goal is to obtain a fully decoupled new equation, we may have to consider transformations with no clear physical space representation. This path was taken in much of the physics literature on rotating $a\neq 0$ Kerr: for $s=\pm 2$, \cite{Chandrasekhar1976,Chandrasekhar1976a,Sasaki1982} have indeed succeeded in finding transformed quantities that satisfy decoupled equations by considering transformations that can only be interpreted in the separated picture (see also \cite{Glampedakis2017} for an overview of previously considered transformations). We do not consider such transformations since, as they cannot be easily translated into the physical space picture, they pose a greater obstacle to our methods than the mild coupling in \eqref{eq:transformed-equation-intro}.

\subsection{The role of separation of variables}
\label{sec:intro-separability}

Let $\swei{\upalpha}{s}$ be a solution to the Teukolsky equation \eqref{eq:Teukolsky-equation-intro}. By taking a  Fourier--Laplace transform with respect to $t$ and noting \eqref{eq:Teukolsky-equation-intro} commutes with $\p_\phi$, we can formally decompose $\swei{\upalpha}{s}$ into functions
\begin{align}
\upalpha^{[s],\,(a\omega)}_{m}(t,r,\theta,\phi) = e^{-i\omega t} e^{im\phi}\tilde{\upalpha}^{[s],\,(a\omega)}_m(r,\theta)\,, \label{eq:unseparated-mode}
\end{align}
for $\omega\in\mathbb{C}$ and $m$ an integer or half-integer depending on $s$.  Plugging in the ansatz \eqref{eq:unseparated-mode} into the Teukolsky equation \eqref{eq:Teukolsky-equation-intro}, we obtain a PDE for $\tilde{\upalpha}^{[s],\,(a\omega)}_m$ in the $r$ and $\theta$ variables. This PDE is clearly separable if $a=0$, as \eqref{eq:Teukolsky-equation-intro} commutes with the generators of the spherical symmetries of the spacetime. 

For rotating Kerr black holes, $a\neq 0$, it would seem that there are not enough symmetries for the same result to hold; yet, as we will see in Section~\ref{sec:intro-carter-constant}, it does. In our series of papers, we view separation of variables as a concrete, convenient and explicit approach to frequency analysis, intimately tied to the Kerr geometry, that allows one to address several aspects of the Teukolsky equation simultaneously under the same formalism, such as: the existence of a conservation law for Teukolsky with no clear physical space analogue (Section~\ref{sec:intro-energy-identity}), the possibility of superradiance and a proof of the so-called mode stability of the Teukolsky equation (Section~\ref{sec:intro-mode-stability}), the phenomenon of trapping for the transformed system (ignoring coupling terms) as well as its quantitative non-superradiance (for $|a|<M$, see Section~\ref{sec:intro-trapping}), etc. All these aspects are reflected in Theorem~\ref{thm:mode-stability-intro-AB} and in Theorem~\ref{thm:frequency-estimates-big} (see already Section~\ref{sec:intro-proof-freq-estimates}). In Section~\ref{sec:intro-to-freq-space}, we explain how Theorem~\ref{thm:bddness-decay-big} can be reduced to obtaining the uniform in frequency estimates of Theorem~\ref{thm:frequency-estimates-big}A.

\subsubsection{Carter's constant, formal separation and the Teukolsky radial ODE}
\label{sec:intro-carter-constant}

Besides the conserved quantities induced by the stationary, $\p_t$, and axisymmetric, $\p_\phi$, Killing fields on Kerr, it turns out that there is a hidden conserved quantity for geodesics. The conservation law was first found by Carter \cite{Carter1968}, with the quantity being subsequently known as \textit{Carter's constant}, and was later shown to be associated with a Killing tensor \cite{Walker1970}. Using this additional symmetry, Carter was able to show that, if $s=0$, the PDE for $\tilde{\alpha}^{[s],\,(a\omega)}_m(r,\theta)$ admits separable solutions in the entire $|a|\leq M$ range; his result was then extended by Teukolsky \cite{Teukolsky1973} for $s\neq 0$ by analogy.

In what follows, a separable solution characterized by the frequency triple $(\omega,m,\Lambda)$, where $m$ is an integer or half-integer depending on $s$ and  $\omega,\Lambda\in\mathbb{C}$, is written as 
\begin{align}
\upalpha^{[s],\,(a\omega)}_{m\Lambda}(t,r,\theta,\phi) = e^{-i\omega t}e^{im\phi}S^{[s],\,(a\omega)}_{m\Lambda}(\theta)(r^2+a^2)^{-1/2}\Delta^{-s/2}u^{[s],\,(a\omega)}_{m\Lambda}(r)\,. \label{eq:separable-intro}
\end{align}
Here, $S^{[s],\,(a\omega)}_{m\Lambda}$ satisfies an angular ODE together with boundary conditions that constrain the values of the separation constant $\Lambda$; e.g.\ (see for instance \cite[Proposition 2.1]{TeixeiradaCosta2019})
\begin{align}
\omega\in\mathbb{R}\Rightarrow \Lambda\in\mathbb{R}\,;\qquad \Im\omega>0\Rightarrow\Im\lp[\lp(\Lambda-a^2\omega^2\rp)\overline\omega\rp]<0\,. \label{eq:intro-angular-constraints}
\end{align}
The function $u^{[s],\,(a\omega)}_{m\Lambda}$ satisfies the radial ODE given before as \eqref{eq:radial-ODE-alpha-intro-basic}, which is written in full as
\begin{equation}\label{eq:radial-ODE-u-intro}
\begin{gathered}
\lp(\swei{u}{s}\rp)''+\lp(\omega^2-\swei{V}{s}\rp)\swei{u}{s}=0\,, \\
\begin{aligned}
\swei{V}{s} &= \frac{\Delta\Lambda+4Mram\omega -a^2m^2}{(r^2+a^2)} + s^2\frac{(r-M)^2}{(r^2+a^2)^2}+\frac{\Delta\lp(a^2\Delta+2Mr(r^2-a^2)\rp)}{(r^2+a^2)^4}\\
&\qquad-2is\frac{\omega \lp(r(r^2+a^2)+M(a^2-3r^2)\rp)+am(r-M)}{(r^2+a^2)^2}\,,
\end{aligned}
\end{gathered}
\end{equation}
with $'$ denoting a derivative with respect to a modified radial coordinate $r^*$ that maps $r=r_+$ to $r^*=-\infty$ leaving $r=\infty$ unchanged. We note the last line of the potential, which decays at a significantly slower rate as $r\to \infty$ than the first two lines, is only present when $s\neq 0$.

\subsubsection{The Teukolsky--Starobinsky conservation law for the Teukolsky equation}
\label{sec:intro-energy-identity}

As was mentioned already in the preamble to this introduction, the Killing energy associated  to the vector field $\p_t$ is not conserved for the Teukolsky equation~\eqref{eq:Teukolsky-equation-intro} due to the equation's first order terms. Nonetheless, a conservation law for the Teukolsky equation exists in the separated picture\footnote{It may be possible to find additional conservation laws, associated to the Teukolsky equation itself or arising from the entire system of linearized gravitational or electromagnetic perturbations,  which hold in physical space, see \cite{Holzegel2016} for $a=0$.}. It follows from the fact that,  since $\swei{V}{s}=\overline{\swei{V}{-s}}$ in \eqref{eq:radial-ODE-u-intro}, the Wronskian between $\swei{u}{+s}$ and $\overline{\swei{u}{-s}}$ is conserved for the homogeneous Teukolsky equation. When $(\omega,m,\Lambda)$ are real, using the Teukolsky--Starobinsky identities \cite{Teukolsky1974,Starobinsky1973,Kalnins1989} to map solutions of the radial ODE of spin $+s$ to spin $-s$ or vice-versa, we can upgrade the conservation of the Wronskian to a conservation law involving only one sign of spin; we refer to the conserved quantity as the Teukolsky--Starobinsky energy. For $s\neq 0$, the Teukolsky--Starobinsky energy assumes the role that the Killing energy has for $s=0$.

In order to state the Teukolsky--Starobinsky conservation law for real frequencies, let us define $\swei{u}{s}_{\mc{I}^\mp}$ to be a solution to \eqref{eq:radial-ODE-u-intro} such that the function $r^{\frac{s}{2}(1\mp 1)}\swei{u}{s}_{\mc{I}^\mp}$ is regular at past or future null infinity, respectively, and $\swei{u}{s}_{\mc{H}^\mp}$ to be a solution to \eqref{eq:radial-ODE-u-intro} such that the function $\Delta^{\frac{s}{2}(1\mp 1)}\swei{u}{s}_{\mc{H}^\mp}$ is regular at the past or future event horizon. Then, any solution to the homogeneous radial ODE~\eqref{eq:radial-ODE-u-intro} can be written as, dropping most sub and superscripts,
\begin{align*}
\swei{u}{s}=\swei{a}{s}_{\mc{I}^+}\swei{u}{s}_{\mc{I}^+}+\swei{a}{s}_{\mc{I}^-}\swei{u}{s}_{\mc{I}^-} =\swei{a}{s}_{\mc{H}^+}\swei{u}{s}_{\mc{H}^+}+\swei{a}{s}_{\mc{H}^-}\swei{u}{s}_{\mc{H}^-}
\end{align*}
for some complex coefficients $\swei{a}{s}_{\mc{I}^\pm}$ and $\swei{a}{s}_{\mc{H}^\pm}$.  Without loss of generality, let $s\geq 0$. For integer $s$ and real $(\omega,m,\Lambda)$, the Teukolsky--Starobinsky energy conservation in the separated picture can be written in terms of these coefficients as
\begin{gather}
\begin{gathered}
(2\omega)^{2s}\lp|\swei{a}{\pm s}_{\mc{I}^\pm }\rp|^2 +\frac{\omega-m\upomega_+}{\omega}
\mathfrak{C}_s\lp\{\begin{array}{lr} 
\frac{\prod_{j=1}^{|s|}\lp\{\lp[4Mr_+(\omega-m\upomega_+)\rp]^2+(s-j)^2(r_+-r_-)^2\rp\}}{(r_+-r_-)^{\frac{s}{2}(1\pm 1)}},\,&|a|<M\\
\lp[4M^2(\omega-m\upomega_+)\rp]^{2s},\,&|a|=M
\end{array}\rp\}^{-1}
\lp|\swei{a}{\pm s}_{\mc{H}^\pm}\rp|^2\\
=\frac{\mathfrak{C}_s}{(2\omega)^{2s}}\lp|\swei{a}{\pm s}_{\mc{I}^\mp}\rp|^2+\frac{\omega-m\upomega_+}{\omega}
\lp\{\begin{array}{lr} 
\frac{\prod_{j=0}^{|s|-1}\lp\{\lp[4Mr_+(\omega-m\upomega_+)\rp]^2+(s-j)^2(r_+-r_-)^2\rp\}}{(r_+-r_-)^{\frac{s}{2}(1\mp 1)}},\,&|a|<M\\
\lp[4M^2(\omega-m\upomega_+)\rp]^{2s},\,&|a|=M
\end{array}\rp\}
\lp|\swei{a}{\pm s}_{\mc{H}^\mp}\rp|^2
\end{gathered} \label{eq:Wronskian-conservation-u-intro}
\end{gather}
where we used the notation 
$$\upomega_+:=\frac{a}{2Mr_+}\,,$$
and $\mathfrak{C}_s$ are the so-called Teukolsky--Starobinsky constants, which depend only on $a$, $M$, $|s|$ and the frequency triple $(\omega,m,\Lambda)$. For $|s|\leq 2$, these constants satisfy
\begin{align*}
\mathfrak{C}_0=1\,,\,\,\mathfrak{C}_s\geq 0\,,\text{~for~} |a|\leq M\,; \qquad \mathfrak{C_s}\geq (2|s|)^{2|s|}>0\,,\text{~for~} a=0\,.
\end{align*}
Note that, for $s=0$, \eqref{eq:Wronskian-conservation-u-intro} reduces to the Killing energy identity for the vector field $\p_t$. We also note that, for spin $+s\geq 0$, the identity \eqref{eq:Wronskian-conservation-u-intro} had already appeared at the level of the transformed variable $\swei{\Psi}{s}$ in the preamble to this introduction in \eqref{eq:Wronskian-conservation-Psi-intro}.

\subsubsection{Superradiance and mode stability for the Teukolsky equation}
\label{sec:intro-mode-stability}

In virtue of the separability property of the Teukolsky equation, which admits solutions \eqref{eq:separable-intro} for $\omega,\Lambda\in\mathbb{C}$ and $m\in\frac12\mathbb{Z}$ integer or half-integer depending on $s$, one can immediately ask whether \textit{mode solutions} exist with $\Im\omega>0$ or $\omega\in\mathbb{R}$; if not, as we have seen, we say \textit{mode stability}  holds for \eqref{eq:Teukolsky-equation-intro}. Mode solutions are solutions to \eqref{eq:Teukolsky-equation-intro} taking the separable form \eqref{eq:separable-intro} with the boundary conditions $\swei{a}{s}_{\mc{H}^-}=\swei{a}{s}_{\mc{I}^-}=0$ (in the notation of Section~\ref{sec:intro-energy-identity}); note such solutions, if they exist, are exponentially growing in time, or bounded but non-decaying, respectively. We have seen already that the answer is no, i.e.\ mode stability holds (see Theorem~\ref{thm:mode-stability-intro-AB}); in the remainder of this section, we recall to the reader why this result is not immediate and give context to the versions of the result we have already cited in the preamble to the introduction.

For $a=0$, mode stability is in fact immediate from the Teukolsky--Starobinsky energy conservation, which for integer $s$ takes the form \eqref{eq:Wronskian-conservation-u-intro}, and the fact that $\mathfrak{C}_s>0$. For $s=0$, where it is the Killing field $\p_t$ which gives rise to the conservation law, we can interpret this result as a consequence of the fact that $\p_t$ is timelike everywhere in the black hole exterior, except at the horizon. 

On the other hand, if $a\neq 0$, there is a region close to the event horizon, called the ergoregion, where $\p_t$ becomes spacelike. Thus waves coming in from null infinity may acquire more energy when passing through the ergoregion, which could, of course, be worrisome from the point of view of the stability statements we are after. As an example, let $\omega\in\mathbb{R}\backslash\{0\}$ and $s=0$; conservation of the Killing energy associated to $\p_t$ (equivalently, specializing \eqref{eq:Wronskian-conservation-u-intro} to $s=0$) gives the identity
\begin{align}
\omega(\omega-m\upomega_+)\lp|\swei{a}{0}_{\mc{H}^+}\rp|^2+\omega^2\lp|\swei{a}{0}_{\mc{I}^+}\rp|^2=0\,,  \qquad \upomega_+:=\frac{a}{2Mr_+}\,.\label{eq:energy-identity-wave-intro}
\end{align}
Clearly, if $\omega(\omega-m\upomega_+)\geq 0$ (note this is always true if $a=0$), we conclude that $\swei{a}{0}_{\mc{I}^+}=0$, which, since also $\swei{a}{0}_{\mc{I}^-}=0$ by assumption, implies $u\equiv 0$ by standard ODE theory. In this case, the wave does not gain any extra energy, and we see that the wave equation is modally stable. On the other hand, \eqref{eq:energy-identity-wave-intro} fails to provide an upper bound on the values $u$ can take if the frequency parameters are  \textit{superradiant}, i.e.\ if
\begin{align}
\omega(\omega-m\upomega_+)< 0 \,. \label{eq:superradiance-intro}
\end{align}
Similar statements holds true for real $\omega$ and spin $s=\pm 1,\pm 2$ when one considers the frequency space Teukolsky--Starobinsky energy conservation  \eqref{eq:Wronskian-conservation-u-intro}. For $\Im\omega>0$, though the identity following from the Teukolsky--Starobinsky energy conservation does not take the form \eqref{eq:Wronskian-conservation-u-intro}, one can again see (at least in the $s=0$ case) that there are frequencies for which mode stability cannot be inferred from it. In conclusion, the conservation law for the Teukolsky equation of spin $s\in\{0,\pm 1, \pm 2\}$ does not guarantee its mode stability for rotating, $a\neq 0$, Kerr black holes.

Yet, remarkably, as we have mentioned, mode stability does hold for the Teukolsky equation \eqref{eq:Teukolsky-equation-intro}, as summarized in Theorem~\ref{thm:mode-stability-intro-AB}. The pioneering result is due to Whiting \cite{Whiting1989}, who showed that, for subextremal $|a|<M$ parameters, there are no nontrivial mode solutions with $\Im\omega>0$ for any $s\in\frac12 \mathbb{Z}$. Whiting's proof is based on an injective transformation of the $s\leq 0$ mode solutions into solutions of the scalar wave equation on a spacetime with no ergoregion: the Killing energy with respect to $\p_t$ is conserved for the new solutions, and $\p_t$ is never spacelike in the exterior in the new spacetime. Specifically, he considers an \textit{integral} transformation of the radial component,
\begin{align}
\begin{split}
\tilde{u}(x) &:=(x^2+a^2)^{1/2}(x-r_-)^{-s}(x-r_+)^{-2iM\omega}e^{-i\omega x}  \times\\
&\qquad\times  \int_{r_+}^\infty\lp\{ e^{\frac{2i\omega}{r_+-r_-}(x-r_-)(r-r_-)}(r-r_-)^{i\frac{2Mr_-\omega-am}{r_+-r_-}}(r-r_+)^{-i\frac{2Mr_+\omega-am}{r_+-r_-}} e^{i\omega r}\times\rp.\\
&\qquad\qquad\qquad\qquad\lp.\times\Delta^{-s/2}(r^2+a^2)^{-1/2}\swei{u}{s}(r)\rp\} dr\,,
\end{split} \label{eq:intro-subextremal-transformation}
\end{align}
and a \textit{differential} transformation of the angular component which are well-defined for $s\leq 0$; to conclude mode stability for $s>0$ from the $s\leq 0$ result, he them appeals to the Teukolsky--Starobinsky identities we have mentioned in the previous section. More recently, the first author understood how to extend Whiting's radial transformation \eqref{eq:intro-subextremal-transformation} to  $\omega\in\mathbb{R}\backslash\{0\}$, thus establishing a real mode stability result for the wave equation, $s=0$, on subextremal, $|a|<M$, Kerr in \cite{Shlapentokh-Rothman2015}; this was later generalized to $s\in\frac12\mathbb{Z}$ \cite{Andersson2017} and to the wave equation $s=0$ in the Kerr--Newman (Kerr's charged cousin) subextremal black hole \cite{Civin2014a}. Whiting's transformation \eqref{eq:intro-subextremal-transformation} fails in the extremal limit $|a|\to M$; however, the second author has shown that another transformation of the mode solutions, specifically tailored to a degenerate horizon, can be employed to show mode stability for extremal $|a|=M$ Kerr as well \cite{TeixeiradaCosta2019}, both in the upper half-plane and on the real axis.  We emphasize that the known proofs of mode stability on the real axis, which rely on transformations like \eqref{eq:intro-subextremal-transformation} of separable mode solutions  and, for $s\neq 0$, on the Teukolsky--Starobinsky identities, crucially use the separability property of the Teukolsky equation~\eqref{eq:Teukolsky-equation-intro}.

To conclude the section, let us recall to the reader that, as we have mentioned in the preamble, we will be interested in the decompositions into separable solutions of the form \eqref{eq:separable-intro} only for \textbf{real frequencies} $(\omega, m,\Lambda)$. In this case, one can show separable solutions form a complete basis of a space of ``sufficiently integrable in time'' solutions to the Teukolsky equation~\eqref{eq:Teukolsky-equation-intro} (see already Section~\ref{sec:intro-to-freq-space}). Thus, from this point onwards, $(\omega, m,\Lambda)$ will always denote a set of real frequency parameters.

\subsubsection{Trapping in the separated transformed system}
\label{sec:intro-trapping}

Another important feature of Kerr black holes is the presence of \textit{trapped} null geodesics. These are null geodesics which never intersect the event horizon nor escape to future null infinity, $\mc{I}^+$. An analysis of the equations of null geodesic flow easily yields that the future trapped null geodesics on the Schwarzschild black hole asymptote to a hypersurface of codimension one, characterized by $r=3M$ in Boyer--Lindquist coordinates, usually called the photon sphere. By contrast, for rotating Kerr $a\neq 0$, they do not asymptote to a single hypersurface and there is a range of $r$ values for which trapping can occur \cite[Section 62]{Chandrasekhar}. Yet in the separated picture, trapping for waves on $a\neq 0$ Kerr is immediately seen to be similar to that on Schwarzschild  ($a=0$). 

To understand the previous statement, let us consider our transformed system from Section~\ref{sec:intro-transformed}, which naturally inherits the separability property of the Teukolsky equation: if $\swei{\uppsi}{s}_{(k)}$ be obtained from a solution to \eqref{eq:Teukolsky-equation-intro} of the form \eqref{eq:separable-intro} by the procedure in Section~\ref{sec:intro-transformed} and, by abuse of notation, we identify  the variables with their radial component, we can write the radial ODE corresponding to \eqref{eq:transformed-equation-intro} as
\begin{equation}\label{eq:radial-ODE-Psi-intro}
\begin{gathered}
\lp(\swei{\Psi}{s}\rp)''+\lp(\omega^2-\swei{\mc{V}}{s}\rp)\swei{\Psi}{s}=a\sum_{k=0}^{|s|-1} \frac{\Delta}{(r^2+a^2)^2} \lp[ im c^{\Phi}_{s,\,|s|,\,k}(r) +c^\mr{id}_{s,\,|s|,\,k}(r)\rp]\swei{\uppsi}{s}_{(k)}\,, \\
\begin{aligned}
\swei{\mc{V}}{s} &= \frac{\Delta\Lambda+4Mram\omega -a^2m^2}{(r^2+a^2)} + s^2\frac{\Delta}{(r^2+a^2)^2}\\
&\quad+\frac{\Delta}{(r^2+a^2)^4}\lp[(1-2s^2)a^2\Delta+2(1-s^2)Mr(r^2-a^2)\rp]\,.
\end{aligned}
\end{gathered}
\end{equation}
Note that this radial ODE has a real potential (because left the hand side of \eqref{eq:transformed-equation-intro} has no first order derivatives) and that it behaves even better when $s\neq 0$  than in the wave equation case $s=0$ due to an additional, repulsive, term in the potential in the former case. 
 
Let us focus on the radial ODE~\eqref{eq:radial-ODE-Psi-intro} either for $s=0$, or for any $s\in\mathbb{Z}$ if we also disregard all coupling terms. It is easy to see that, for each frequency triple $(\omega,m,\Lambda)$,  the potential admits a unique maximum at $r=r_{\max}(\omega,m,\Lambda)$. As the frequency triple $(\omega,m,\Lambda)$ is trapped if $\mc{V}(r_{\rm max})=\omega^2$, we find that the set where trapping can occur corresponds to a single $r$-value as in the Schwarzschild ($a=0$) case---hence it is an \textit{unstable} trapping---but where now this value of $r$ depends on the frequency triple. 

Returning to the physical space picture, one notices that, unless $|a|\ll M$ is sufficiently small, some of the $r$ values for which trapping are inside the ergoregion. Thus one could be concerned with the existence of waves that are trapped inside a certain $r$-range and continuously interact with the ergoregion to increasingly gain energy; such waves would likely not be bounded or decay in time. In the separated picture however, at least for subextremal, $|a|<M$, black holes, one can immediately see that this does not happen. There (see \cite{Dafermos2016b} for the $s=0$ case), one can show that solutions of the form \eqref{eq:separable-intro} which are superradiant, i.e.\ for which \eqref{eq:superradiance-intro} holds, and high energy ($|\omega|\gg 1$) are quantitatively non-trapped in the sense that there is an $\epsilon>0$ such that
\begin{align}
\omega(\omega-m\upomega_+)\leq 0\,, \,\, |\omega|\gg 1 \implies V(r_{\rm max})-\omega^2\geq \epsilon \omega^2\gg 1\,. \label{eq:superrad-trapping-qtitative-intro}
\end{align}
We note that $\epsilon$ can be made independent of $\omega$ in the entire subextremal Kerr range $|a|<M$; however, in the limit $|a|\to M$, we have that $\epsilon\to 0$ as $\omega\to m\upomega_+$, so that the quantitative disjointess of trapping and superradiance does not always hold.

We emphasize that \eqref{eq:superrad-trapping-qtitative-intro} is one of the main observations that allowed for a proof of boundedness and decay for the wave equation, $s=0$, in the full subextremal $|a|<M$ range \cite{Dafermos2016b}. Its failure in the extremal $|a|=M$ case is one of the main reasons why general, non-axisymmetric solutions to the wave equation on extremal Kerr black holes are still poorly understood.

\subsubsection{The reduction of Theorem~\ref{thm:bddness-decay-big} \texorpdfstring{in \cite{SRTdC2022} }{}to Theorem~\ref{thm:frequency-estimates-big}A in this paper}
\label{sec:intro-to-freq-space}

Having discussed separation of variables, let us briefly return to Theorem~\ref{thm:bddness-decay-big}. One reason to consider some kind of frequency space analysis in proving Theorem~\ref{thm:bddness-decay-big} comes from the observation, due to Alinhac \cite{Alinhac2009}, that unless $a=0$, already for $s=0$, there are no classical physical space vector field multipliers which can capture trapping  (see Section~\ref{sec:intro-trapping}). In view of Section~\ref{sec:intro-trapping}, this is a difficulty that can easily be addressed with the help of separation of variables. In fact,  separation of variables is a convenient technique for frequency analysis as it allows us to address all the issues discussed in the present section at once: specifically, we can
\begin{enumerate}[noitemsep,label=(\roman*)]
\item Have a very explicit characterization of trapping: it may occur only at a single $r$-value which depends on the frequency triple $(\omega, m,\Lambda)$ and, moreover, if it does then $(\omega, m,\Lambda)$ are quantitatively non-superradiant in the subextremal $|a|<M$ case; see Section~\ref{sec:intro-trapping}.
\item For $s\neq 0$, make use of the Teukolsky--Starobinsky energy conservation law (see Section~\ref{sec:intro-energy-identity}) and an improved estimate to handle the coupling terms $\swei{\mathfrak{J}}{s}$ at trapping (see already Section~\ref{sec:outline-proof-estimate-unbdd}) which are both well understood in the separated picture only.
\item Maintain the same formalism when invoking mode stability: mode stability, which is essential to any proof of Theorem~\ref{thm:bddness-decay-big} outside the realm of very slowly rotating $|a|\ll M$ parameters, is proven using a transformation and, if $s\neq 0$, Teukolsky--Starobinsky identities, both with no clear analogue outside the separated picture; see Section~\ref{sec:intro-mode-stability}.
\end{enumerate}
The frequency localized estimates exploiting the features (i) and (ii) constitute Theorem~\ref{thm:frequency-estimates-big}A; (iii) is embodied by Theorem~\ref{thm:mode-stability-intro-AB}.

Let us now give a preview of how Theorem~\ref{thm:bddness-decay-big} is actually proven from Theorems~\ref{thm:mode-stability-intro-AB} and ~\ref{thm:frequency-estimates-big}A in \cite{SRTdC2022} (see also Remark~\ref{it:remark-intro-A}).  As we have mentioned in Section~\ref{sec:intro-carter-constant}, the Teukolsky equation \eqref{eq:Teukolsky-equation-intro} is formally separable, in that it admits separable solutions of the form \eqref{eq:separable-intro} for, if $s\in\mathbb{Z}$, $(\omega,m,\Lambda)\in\mathbb{C}\times\mathbb{Z}\times\mathbb{C}$. Moreover, for $\omega,\Lambda\in \mathbb{R}$ ($\Lambda\in\mathbb{R}$ follows from $\omega\in\mathbb{R}$ by  \eqref{eq:intro-angular-constraints}), modes of form the \eqref{eq:separable-intro} form a complete basis for ``sufficiently integrable in time'', in the sense that they admit a Fourier transform in $t$, solutions of \eqref{eq:Teukolsky-equation-intro}. An equivalent statement holds for the transformed system described in Section~\ref{sec:intro-transformed}, in particular for the PDE \eqref{eq:transformed-equation-intro}: $L^2$ estimates for ``sufficiently integrable in time'' solutions of this system follow from $L^2$ estimates on the corresponding system of radial ODEs, such as \eqref{eq:radial-ODE-Psi-intro}, which are uniform in the frequency parameters $(\omega,m,\Lambda)$, by Plancherel's theorem. Hence Theorem~\ref{thm:frequency-estimates-big}A, combined with Theorem~\ref{thm:mode-stability-intro-AB}, can be shown to imply an estimate for the PDEs \eqref{eq:Teukolsky-equation-intro} and \eqref{eq:transformed-equation-intro} (see our upcoming \cite{SRTdC2022} for details).  

We highlight again that the assumption that $(\omega,m,\Lambda)$ are real frequency parameters is key; for  $\omega\in\mathbb{C}$, it is not clear that the separable ansatz provides a basis for sufficiently integrable in time solutions. 

Given that our analysis is based entirely on the real frequencies, an important technical difficulty is obtaining \textit{a priori} that, if $a\neq 0$, general solutions to \eqref{eq:Teukolsky-equation-intro} and to the transformed system, e.g.\ \eqref{eq:transformed-equation-intro}, are sufficiently integrable in time. As we have mentioned in the preamble, this requires a continuity argument in $a$, of the style obtained by the first author together with Dafermos and Rodnianski to address the case $s=0$ in \cite{Dafermos2016b}, which we flesh out in our upcoming \cite{SRTdC2022}.

\subsubsection{Alternatives to separation of variables for \texorpdfstring{$a=0$}{the scalar wave equation} \texorpdfstring{and $|a|\ll M$}{on very slowly rotating Kerr}}

To conclude the section, we contrast our methods with others which can be used in the very slowly rotating case $|a|\ll M$. Only in this perturbative setting, where one need not appeal to the mode stability of Section~\ref{sec:intro-mode-stability}, can one hope to be able to avoid separation of variables altogether. Indeed, several techniques bypassing explicit separation of variables have been developed for $s=0$ and $|a|\ll M$. 

Uniform boundedness of energy for the scalar wave equation $s=0$ in the very slowly rotating case $|a|\ll M$  was first proven in \cite{Dafermos2011}. The argument had the interesting feature that it was able to avoid any detailed analysis of trapping, and in particular the use of separation of variables, by only relying on a decomposition of the solution into a superradiant and non-superradiant piece.  

In order to show decay, rather than just boundedness, one must understand trapping. Microlocal analysis based on a local reduction to Euclidean space provides an alternative framework to connect the unstable trapping of the geodesic flow with so-called positive commutator estimates; this approach has been used successfully, for $s = 0$, to establish integrated energy decay statements in the very slowly rotating $|a| \ll M$ Kerr case \cite{Tataru2011} (also later in the work \cite{Hafner2019}). See also \cite{Wunsch2011}, \cite{Dyatlov2015}, \cite{Hintz2018}, and the references therein for results which use microlocal analysis to relate the high frequency behavior of waves with the underlying normal hyperbolic trapping. If one imposes control on higher norms in the initial data, thus leaving the scope of Alinhac's observation, an analogue of Theorem~\ref{thm:bddness-decay-big} for $s = 0$ and $|a|\ll M$ can also be established using a physical space multiplier, based on the Killing tensor representing Carter's hidden conservation law, to capture trapping \cite{Andersson2015}.

\subsection{The proof of Theorem~\ref{thm:frequency-estimates-big}}
\label{sec:intro-proof-freq-estimates}

For $(\omega,m,\Lambda)\not\in\mc{A}$ where $\mc{A}$ is a set of frequencies where $|\omega|+|\omega|^{-1}+|m|+|\Lambda|<\infty$, Theorem~\ref{thm:frequency-estimates-big} is proved by applying multiplier currents to the separated version of the transformed system from Section~\ref{sec:intro-transformed}, in particular to the radial ODE~\eqref{eq:radial-ODE-Psi-intro}. These currents, which are of virial and energy type, are tailored to the frequency triple $(\omega,m,\Lambda)$, though they result in uniform estimates. Roughly speaking, we employ virial currents to produce control over a bulk term with good properties, such as that in \eqref{eq:intro-main-thm-A}; energy currents are used to control boundary terms as in \eqref{eq:intro-main-thm-A} and \eqref{eq:intro-main-thm-B}. In this section, we give an overview of our strategy to build such currents and, thus, prove Theorem~\ref{thm:frequency-estimates-big}. The discussion here is quite detailed, so we encourage the reader to return to it upon reading the proof.

We begin by reviewing, in Section~\ref{sec:intro-estimates-Box}, the techniques which are useful in studying the transformed operator $\swei{\mathfrak{R}}{s}$, when we disregard the coupling terms $\swei{\mathfrak{J}}{s}$; the discussion is guided by the analogy between the tranformed operator $\swei{\mathfrak{R}}{s}$ and the wave operator $\Box_g$ and the analysis of the later in \cite{Dafermos2016b} (see the section summary in Table~\ref{tab:current-wave-RW-operator}). The coupling terms $\swei{\mathfrak{J}}{s}$ and the additional difficulties they pose are then discussed in Section~\ref{sec:intro-overview-coupling-difficulties} (see the section summary in Table~\ref{tab:difficulties-coupling}).

The remainder of the section discusses the novel ideas on which we rely in the present paper. One such idea, presented in Section~\ref{sec:energy-identity-transformed}, is the adaptation of the Teukolsky--Starobinsky energy conservation law that exists for the general $s\in\frac12\mathbb{Z}$ Teukolsky equation in the separated picture (see  Section~\ref{sec:intro-energy-identity}) to an energy current which is suitable for application to the transformed system. The remaining techniques, which are more specifically tailored to some frequency ranges, are introduced in Section~\ref{sec:outline-proof-estimate-bdd}, if they are best suited to bounded frequency regimes, or in Section \ref{sec:outline-proof-estimate-unbdd}, if they are best suited to  unbounded regimes. Table~\ref{tab:currents-small-vs-large-a} summarizes the key differences in the frequency space analysis we conduct in this series of papers, when compared to the approach to the very slowly rotating $|a|\ll M$ Kerr case given by Dafermos, Holzegel and Rodnianski in \cite{Dafermos2017}.

\subsubsection{Frequency space analysis of the transformed operator\texorpdfstring{ $\swei{\mathfrak{R}}{s}$}{}}
\label{sec:intro-estimates-Box}

Since $\swei{\mathfrak{R}}{s}$ can be seen as a spin-weighted analogue of the $s=0$ case, where  $\swei{\mathfrak{R}}{0}$ is the wave operator $\Box_g$ acting on the radiation field, the techniques we rely on for its analysis (while ignoring the coupling terms of \eqref{eq:transformed-equation-intro}) are precisely those one would consider when studying the scalar wave equation on a Kerr background. Thus, while the contents of this section build heavily on all of the previous literature, they are based especially on the results for the wave equation on subextremal Kerr \cite{Dafermos2008,Dafermos2016b,Dafermos2014} and those for the Teukolsky equation of spin $|s|=1,2$, respectively \cite{Pasqualotto2016,Dafermos2016a}, on Schwarzschild.

Recall the transformed radial ODE \eqref{eq:radial-ODE-Psi-intro} with any $s\in\mathbb{Z}$, focusing only on the left hand side (i.e.\ ignoring the coupling) or setting $s=0$. In order to understand the behavior of the ODE, the main phenomena to keep in mind are superradiance, mode stability and trapping, which have been discussed in the previous two sections. Concretely, we partition the space of admissible frequency parameters as follows (see Table~\ref{tab:current-wave-RW-operator} for a summary).
\begin{itemize}[noitemsep]
\item \textit{Bounded frequency parameters}. Either (i) or (ii), as below, hold.
\begin{enumerate}[label=(\roman*),noitemsep]
\item $1\lesssim |\omega|\lesssim 1$, i.e.\ $(\omega,m,\Lambda)\in\mc{A}$ where $\mc{A}$ is the set which is excluded from Theorem~\ref{thm:frequency-estimates-big}A. Nevertheless, let us remark how one can obtain estimates in this case. The very slowly rotating case $|a|\ll M$, one can use the smallness of this parameter to construct suitable multiplier currents (see \cite{Dafermos2010}). In the general subextremal case $|a|<M$ considered in our series of works, one can appeal to Theorem~\ref{thm:mode-stability-intro-AB} to obtain a scattering statement such as Theorem~\ref{thm:frequency-estimates-big}B and to complement Theorem \ref{thm:frequency-estimates-big}A (see the Remark in the preamble and our upcoming \cite{SRTdC2022}).

\item  $|\omega|\ll 1$. In this case, it is the frequency-independent part of the potential that drives the radial ODE; one can construct virial multiplier currents that take advantage of its properties.
\end{enumerate}
\item \textit{Unbounded frequency parameters}. When one or more of the frequency parameters is large, we again want to build targeted virial multiplier currents for the radial ODE under consideration. The frequency-independent part of the potential becomes a lower order correction to the behavior of solutions of the radial ODE; to determine what the virial currents should look like, we are mainly interested to know whether it is $\omega^2$ or $\Lambda$ (a proxy for the size of the potential) which is driving the ODE. If $\omega^2\gg \Lambda$ or if $\Lambda\gg \omega^2$ are large, we can exploit the largeness of the difference between $\omega^2$ and the potential or vice-versa and there is no trapping. Only when $\omega^2\sim \Lambda$ do we not know whether it is $\omega^2$ or the potential which is dominating and, then, trapping can occur; implying a loss, at the trapping $r$ value, in our estimates.
\end{itemize}

For $(\omega,m,\Lambda)\not\in\mc{A}$, where $\mc{A}$ is as in Theorem~\ref{thm:frequency-estimates-big}A, the previous discussion is our guide in the construction of virial multiplier currents for the radial ODE \eqref{eq:radial-ODE-Psi-intro} which, coupling $\swei{\mathfrak{J}}{s}$ aside, are effective in capturing the properties of $\swei{\mathfrak{R}}{s}$ in a good bulk estimate. Indeed, the virial currents we apply to the radial ODE \eqref{eq:radial-ODE-Psi-intro} are exactly those developed for the $s=0$ case in \cite{Dafermos2016b}, aside from some minor technical improvements (see already Section~\ref{sec:virial-current-templates} for the virial current templates). Their interaction with the coupling $\swei{\mathfrak{J}}{s}$, however, means that for general $|a|<M$, we cannot rely on scalar wave techniques alone, as we will see in the remainder of the introduction.

To conclude, note that even when $\swei{\mathfrak{J}}{s}$ is absent, one produces boundary terms which must be controlled by applications of some energy current.  If superradiance \eqref{eq:superradiance-intro} occurs, then one of the boundary terms in our energy current will have bad sign (recall e.g.\  \eqref{eq:energy-identity-wave-intro} for $s=0$), hence the current must be applied in a localized fashion, thus producing localization errors. It is fortunate that, on subextremal, $|a|<M$, Kerr black holes, if trapping occurs and brings about the expected loss in our estimates, the frequency parameters cannot be superradiant, hence, one does not need to deal with such localization errors. Nevertheless, the necessity of appealing to energy currents leads to the inevitable question: which energy? Should we invoke the Killing energy associated to, for instance, $\p_t$ which past work on the case $s=0$  on \cite{Dafermos2016b} and even, if $|a|\ll M$, $s=\pm 1,\pm 2$ \cite{Dafermos2017,Ma2017,Ma2017a} has relied on; or, departing from scalar wave equation techniques, should we invoke  the Teukolsky--Starobinsky energy from Section~\ref{sec:intro-energy-identity} with no physical space analogue? On this question, we refer the reader to Section~\ref{sec:energy-identity-transformed} already.

\begin{table}[!htb]
\begin{center}
\caption{The frequency space analysis of the operator $\swei{\mathfrak{R}}{s}$ (ignoring coupling) in the full subextremal range $|a|<M$ and their treatment. The analysis is based on the $\Box_{g}$ case addressed in \cite{Dafermos2016b}. We denote by $\mc{V}$ the potential in the radial ODE associated with these operators.}
\begin{small}
\begin{tabular}{|c|M{1.4cm}|M{1.2cm}|M{11.6cm}|} \hline
\multicolumn{3}{|c|}{$(\omega,m,\Lambda)$} & $|a|<M$\\ \hline
\multicolumn{3}{|c|}{all} & Virial currents produce boundary terms which must be controlled by energy currents. If superradiance occurs, these must be applied in a localized fashion, producing bulk localization errors.\\ \hline
\multirow{2}{*}[-2.2em]{\rotatebox{90}{bounded}}& \multicolumn{2}{c|}{$|\omega|\gtrsim 1$  } &  Use quantitative mode stability(*) or, if $|a|\ll M$, the smallness of superradiance.\\ \cline{2-4}
&\multicolumn{2}{c|}{$|\omega|\ll 1$}& Superradiance may occur; localization errors are absorbed by a virial current which takes advantage of  positivity of $\mc{V}-\omega^2$ in large region. 
This virial current is glued in to another which uses $\mc{V}'<0$ as $r\to \infty$ and, as $r\to r_+$,
\begin{itemize}[noitemsep,nolistsep]
\item if $|\omega-m\upomega_+|\ll 1$, to a virial current using $\mc{V}'>0$ as $r\to r_+$;
\item if $|\omega-m\upomega_+|\gtrsim 1$, to a global virial current using the non-smallness of this parameter.
\end{itemize} \vspace{-\baselineskip}~\\ \hline

\multirow{5}{*}[-9.2em]{\rotatebox{90}{unbounded}}& \multicolumn{2}{c|}{all} & The potential has at most one maximum.\\  \cline{2-4}
& $\omega^2\gg \Lambda$, $\omega^2\gg 1$& \multirow{3}{*}[-6em]{\shortstack[c]{non-\\ trapped}} &No superradiance or trapping can occur. It is easy to build a global virial current which takes advantage of the positivity and largeness of $\mc{V}-\omega^2$.\\ \cline{2-2}\cline{4-4}
&$\Lambda\gg \omega^2$, $\Lambda\gg 1$ &  &Trapping cannot occur; frequencies may be superradiant. Due to $r$-decay, $\mc{V}-\omega^2$ is not globally positive. However, $\mc{V}'$ has a single zero at a finite $r$ value, which can be captured with a global virial current. $\mc{V}$ attains a maximum at that point of size $\Lambda\gg \omega^2$, leading to a large gain in the bulk term which is enough to absorb energy localization errors. \\ \cline{2-2}\cline{4-4}
&\multirow{2}{*}[-3em]{\shortstack[c]{$\Lambda\sim \omega^2$,\\ $\Lambda,\omega^2\gg 1$}} & & 

\begin{itemize}[leftmargin=*,topsep=0pt]
\item If superradiant, then cannot be trapped(**): the potential behaves as in the $\Lambda\gg \omega^2$, $\Lambda\gg 1$ case and can use the same strategy.
\item If quantitatively non-superradiant,  $\mc{V}-\omega^2$ is quantitatively negative near $r=r_+$. 
\end{itemize}
\begin{itemize}[noitemsep,nolistsep,leftmargin=3.5em]
\item[--] Either $\mc{V}-\omega^2$ globally quantitatively negative, in which case trapping cannot occur. We can apply the strategy for $\omega^2\gg \Lambda$, $\omega^2\gg 1$.
\end{itemize}\vspace{-\baselineskip}~\\\cline{3-4}
&& trapped &
\begin{itemize}[noitemsep,topsep=0pt, leftmargin=3.5em]
\item[--] Or else $\mc{V}$ surely has a maximum away from $r=r_+$, which is trapped. Then, use a virial current capturing sign and degeneracy of $\mc{V}'$, coupled to one exploit negativity of $\mc{V}-\omega^2$ near $r=r_+$. The bulk term controlled exhibits a loss at the maximum of $\mc{V}$ which cannot be eliminated; this would make absorption of energy localization errors  difficult if it were necessary.
\end{itemize}\vspace{-\baselineskip}~
\\ \hline
\end{tabular}
\label{tab:current-wave-RW-operator}
\end{small}
\end{center}

\begin{small}
\noindent(*) The quantitative control of mode stability only degenerates in the double limit $|a|\to M$, $|\omega-m\upomega_+|\to 0$ \cite{TeixeiradaCosta2019}.

\smallskip
\noindent (**) If $|a|<M$, the two alternatives given for $\omega^2\sim \Lambda\gg 1$ exhaust the space of possible frequencies (see Section~\ref{sec:intro-trapping}), because superradiant frequencies are non-trapped by a quantifiable amount. As $|a|/M$ increases, this amount becomes smaller; when $|a|=M$, the conclusion holds only if frequencies approaching the superradiant threshold are excluded.
\end{small}
\end{table}

\subsubsection{Frequency space analysis of the coupling\texorpdfstring{ $\swei{\mathfrak{J}}{s}$}{}: main difficulties}
\label{sec:intro-overview-coupling-difficulties}

In the previous section, we reviewed the methods for frequency-localized virial currents developed in the last decade to analyze the operator $\Box_g$ in the subextremal range $|a|<M$ and which apply to our transformed operator $\swei{\mathfrak{R}}{s}$ in \eqref{eq:transformed-equation-intro}, provided that the coupling to $\uppsi_{(k)}^{[s]}$ for $k<|s|$ could be ignored. This section focuses precisely on these coupling terms in $\swei{\mathfrak{J}}{s}$.

It is convenient to already recall what knowledge we have of the lower level transformed variables. For each $k<|s|$, $\uppsi_{(k)}^{[s]}$ is related to $\swei{\Psi}{s}$ by a sequence of $|s|-k$ transport equations, as given in \eqref{eq:def-upppsi-k-intro}. By abuse of notation, as usual let us identify a fixed frequency $\uppsi_{(k)}^{[s]}$ with its radial component; then \eqref{eq:def-upppsi-k-intro} read in the frequency space picture as
\begin{align}
\frac{\Delta}{(r^2+a^2)^2}\uppsi_{(k+1)}^{[s]}=-\sign s \lp(\uppsi_{(k)}^{[s]}\rp)'-i\lp(\omega-\frac{am}{r^2+a^2}\rp)\uppsi_{(k)}^{[s]}\label{eq:def-upppsi-k-freq-intro}\,.
\end{align}
On the other hand, by virtue of their dependence on the Teukolsky variable, $\uppsi_{(k)}^{[s]}$ themselves satisfy wave-type equations, which translate into the radial ODEs, c.f.\ \eqref{eq:radial-ODE-u-intro} and \eqref{eq:radial-ODE-Psi-intro},
\begin{align}
\begin{split}
&\lp(\uppsi_{(k)}^{[s]}\rp)''+(|s|-k)\frac{2(r^3-3Mr^2+a^2r+a^2M)}{(r^2+a^2)^2}\lp(\uppsi_{(k)}^{[s]}\rp)'+\lp(\omega^2-\mc{V}_{(k)}^{[s]}\rp)\uppsi_{(k)}^{[s]} \\
&\quad= a\sum_{i=0}^{k-1}\frac{\Delta}{(r^2+a^2)^2} \lp[ c^{\Phi}_{s,\,k,\,i}(r) im+c^\mr{id}_{s,\,k,\,i}(r)\rp]\swei{\uppsi}{s}_{(i)}\,,
\end{split}\label{eq:radial-ODE-uppsi-intro}
\end{align}
for some $c^{\Phi}_{s,\,k,\,i}(r)$ and $c^\mr{id}_{s,\,k,\,i}(r)$ which are bounded for $r\in[r_+,\infty)$. The potential $\mc{V}_{(k)}^{[s]}$ has a nontrivial imaginary component with poor decay as $r^*\to \pm \infty$  only if $k\neq |s|$. From now on, we may drop the superscript from $\swei{\uppsi}{s}_{(k)}$ and $\mc{V}_{(k)}^{[s]}$ or $\swei{\Psi}{s}$ and $\mc{V}^{[s]}$ if this does not lead to confusion.

It should be clear that, in order to control the coupling terms in \eqref{eq:radial-ODE-Psi-intro}, we must use either \eqref{eq:def-upppsi-k-freq-intro}, or \eqref{eq:radial-ODE-uppsi-intro}, or both. Given the fact that \eqref{eq:radial-ODE-uppsi-intro} for $k<|s|$ retains more of the Teukolsky equation's properties (an imaginary component to the potential which is hard to treat) than of the linear wave equation, it is tempting to start by ignoring these problematic-looking ODEs and focusing on the transport equations \eqref{eq:def-upppsi-k-freq-intro} only. 

An immediate estimate one can obtain from the transport equation \eqref{eq:def-upppsi-k-freq-intro} is, for instance, 
\begin{align}
\int_a^B c'(r)\lp|\uppsi_{(k)}\rp|^2dr^* = 2\lp[c(r)\lp|\uppsi_{(k)}\rp|^2\rp]_{r^*=A}^{r^*=B} +4\int_A^B \frac{\Delta^2}{(r^2+a^2)^4}\frac{c^2}{c'}\lp|\uppsi_{(k+1)}\rp|^2dr^*\,. \label{eq:basic-estimate-1-intro}
\end{align}
where $c(r)$ is a continuous, piecewise $C^1$, $r$-weight with $c'>0$, and $A,B\in\mathbb{R}\cup \{\pm \infty\}$. Estimate \eqref{eq:basic-estimate-1-intro} provides an easy way of climbing the hierarchy, i.e.\ relating integrated estimates for $\uppsi_{(k)}$ to integrated estimates for $\uppsi_{(k+1)}$ and, eventually, to integrated estimates for $\Psi$. However, it turns out \eqref{eq:basic-estimate-1-intro} is not enough: Table~\ref{tab:difficulties-coupling} summarizes the difficulties one encounters when trying to apply the methods of Section~\ref{sec:intro-estimates-Box}  and the rudimentary estimate \eqref{eq:basic-estimate-1-intro} to address \eqref{eq:transformed-equation-intro}. 

\begin{table}[!htbp]
\begin{center}
\caption{Difficulties associated with the frequency space analysis of the coupling $\swei{\mathfrak{J}}{s}$ in \eqref{eq:transformed-equation-intro}, in the case $|a|\ll M$ and in the full subextremal range $|a|<M$. We assume here that the left hand side of \eqref{eq:transformed-equation-intro} has been treated by the methods contained in Table~\ref{tab:current-wave-RW-operator} and examine coupling errors arising from those methods in view of the rudimentary estimate \eqref{eq:basic-estimate-1-intro} that is easily derived from the transport equations relating the transformed variables.}
\begin{small}
\begin{tabular}{|c|M{1.4cm}|M{1.2cm}|M{4.3cm}|M{7.3cm}|} \hline
\multicolumn{3}{|c|}{$(\omega,m,\Lambda)$} & $|a|\ll M$ \cite{Dafermos2017} &$|a|<M$\\ \hline

\multirow{2}{*}[-.5em]{\rotatebox{90}{bounded}}& \multicolumn{2}{c|}{$|\omega|\gtrsim 1$  } & \multirow{5}{*}[-5em]{\parbox{4cm}{\centering control coupling errors in applications of virial and classical energy currents by  \eqref{eq:basic-estimate-1-intro} and smallness of $|a|$}}  & need quantitative mode stability(*)\\ \cline{2-3}\cline{5-5}
&\multicolumn{2}{c|}{$|\omega|\ll 1$}& & the frequency independent part of the coupling terms can always be large; \eqref{eq:basic-estimate-1-intro} fails to give sufficient control over coupling errors in general \\ \cline{1-3}\cline{5-5}

\multirow{5}{*}[-6em]{\rotatebox{90}{unbounded}}&  $\omega^2\gg \Lambda$, $\omega^2\gg 1$& \multirow{3}{*}[-3em]{\shortstack[c]{non-\\ trapped}} & & $m^2\ll \omega^2$ so coupling terms are small and errors due to both virial and classical energy currents can be controlled using \eqref{eq:basic-estimate-1-intro}\\ \cline{2-2}\cline{5-5}
&$\Lambda\gg \omega^2$, $\Lambda\gg 1$ &  &  &$m^2\sim \Lambda$ so coupling errors due to virial currents or classical energy current involving $m$ weights are not small; only errors due to classical energy currents with $\omega$ multipliers can be absorbed using \eqref{eq:basic-estimate-1-intro} \\ \cline{2-2}\cline{5-5}
&\multirow{2}{*}[-4.5em]{\shortstack[c]{$\Lambda\sim \omega^2$,\\ $\Lambda,\omega^2\gg 1$}} & & ~\newline & \multirow{2}{*}[-.7em]{\parbox{7.3cm}{\centering $m^2\sim\Lambda$ so coupling errors due to both virial or classical energy currents could be large; in the trapped case especially, a simple aplication of \eqref{eq:basic-estimate-1-intro} gives quadratic frequency weights on $|\Psi|^2$ which are not controlled near the maximum of the potential and cannot be absorbed elsewhere as they have no smallness parameter}}\\\cline{3-4}
&& trapped non-superra-diant(**) & the global energy current has errors weighted by $m\omega$, so a simple application of \eqref{eq:basic-estimate-1-intro} gives quadratic frequency weights on $|\Psi|^2$ which are not controlled near the maximum of the potential &\\ \hline 
\end{tabular}
\label{tab:difficulties-coupling}
\end{small}
\end{center}
\begin{small}
(*), (**) See the analogous notes in Table~\ref{tab:current-wave-RW-operator}.
\end{small}
\end{table}

In the next sections, we break down the issues arising in each of the frequency ranges identified in Table~\ref{tab:difficulties-coupling} and explain the novel ideas required to deal with them.

\subsubsection{Energy currents for the transformed system: the Killing and the Teukolsky--Starobinsky energy currents}
\label{sec:energy-identity-transformed}

In our proof of Theorem~\ref{thm:frequency-estimates-big}, an important role will be played by two distinct types of energy currents: Teukolsky--Starobinsky energy currents and Killing energy currents.

The Teukolsky--Starobinsky energy is a conserved quantity for the Teukolsky radial ODE~\eqref{eq:radial-ODE-u-intro} without a clear physical space analogue (see Section~\ref{sec:intro-energy-identity}). Transformed variables arising from the Teukolsky equation as described in Section~\ref{sec:intro-transformed} inherit this separated picture conservation law. Indeed, focusing on separable solutions now, let (the radial component of) $\swei{\Psi}{s}_{\mc{H}^\pm}$ (resp.\ $\swei{\Psi}{s}_{\mc{I}^\pm}$) be obtained from $\swei{u}{s}_{\mc{H}^\pm}$ (resp.\ $\swei{u}{s}_{\mc{I}^\pm}$), solutions of \eqref{eq:radial-ODE-u-intro} which are regular at the past or future event horizon (resp.\ past or future null infinity), as described in Section~\ref{sec:intro-transformed}. Then, a solution to \eqref{eq:radial-ODE-Psi-intro} arising from a solution of the Teukolsky radial ODE \eqref{eq:radial-ODE-u-intro} can be written as
\begin{align*}
\swei{\Psi}{s}=\swei{A}{s}_{\mc{I}^+}\swei{\Psi}{s}_{\mc{I}^+}+\swei{A}{s}_{\mc{I}^-}\swei{\Psi}{s}_{\mc{I}^-} =\swei{A}{s}_{\mc{H}^+}\swei{\Psi}{s}_{\mc{H}^+}+\swei{A}{s}_{\mc{H}^-}\swei{\Psi}{s}_{\mc{H}^-}\,,
\end{align*}
for some coefficients $\swei{A}{s}_{\mc{H}^\pm}$ and $\swei{A}{s}_{\mc{I}^\pm}$. Without loss of generality, restricting to $s\geq 0$, the Teukolsky--Starobinsky energy conservation \eqref{eq:Wronskian-conservation-u-intro}  becomes (c.f.\ \eqref{eq:Wronskian-conservation-Psi-intro} in the preamble)
\begin{gather}
\omega^2\lp[\lp|\swei{A}{\pm s}_{\mc{I}^\pm }\rp|^2-\frac{\mathfrak{C}_s}{\mathfrak{D}_s^\mc{I}}\lp|\swei{A}{\pm s}_{\mc{I}^\mp}\rp|^2\rp] +\omega(\omega-m\upomega_+)
\lp[\frac{\mathfrak{C}_s}{\mathfrak{D}_s^{\mc H}}\lp|\swei{A}{\pm s}_{\mc{H}^\pm}\rp|^2-\lp|\swei{A}{\pm s}_{\mc{H}^\mp}\rp|^2\rp]=0\,.\label{eq:Wronskian-conservation-Psi-intro-2}
\end{gather}
Another version of the Teukolsky--Starobinsky energy conservation \eqref{eq:Wronskian-conservation-Psi-intro-2} is obtained trading a factor of $\omega$ by a factor $\omega-m\upomega_+$ in all terms. Here, $\mathfrak{D}_s^{\mc H}$ and $\mathfrak{D}_s^{\mc H}$  are constants which, like the Teukolsky--Starobinsky constant, can be explicitly given in terms of $a$, $M$, $s$ and $(\omega, m,\Lambda)$ and that, for $s=0,\pm 1, \pm 2$, satisfy
\begin{align*}
\mathfrak{C}_0=\mathfrak{D}_0^{\mc I}=\mathfrak{D}_0^{\mc H}=1\,,\,\,\,\mathfrak{C}_s,\mathfrak{D}_s^{\mc I},\mathfrak{D}_s^{\mc H}\geq 0\,,\text{~for~} |a|\leq M\,; \qquad \mathfrak{C_s},\mathfrak{D}_s^{\mc I},\mathfrak{D}_s^{\mc H}\geq (2|s|)^{2|s|}>0\,,\text{~for~} a=0\,.
\end{align*}

The Teukolsky--Starobinsky conservation \eqref{eq:Wronskian-conservation-Psi-intro-2} should be compared with the identity
\begin{align}
\begin{split}
&\omega^2\lp[\lp|\swei{A}{\pm s}_{\mc{I}^+}\rp|^2-\lp|\swei{A}{\pm s}_{\mc{I}^-}\rp|^2\rp] +\omega(\omega-m\upomega_+)
\lp[\lp|\swei{A}{\pm s}_{\mc{H}^+}\rp|^2-\lp|\swei{A}{\pm s}_{\mc{H}^-}\rp|^2\rp]=-\int_{-\infty}^\infty a\omega\Im\lp[\overline{\Psi}\cdot\mathfrak{J}\rp] dr^* \,,
\end{split}\label{eq:classical-energy-Psi-intro}
\end{align} 
arising from an application of the Killing energy current associated to the stationary Killing field $\p_t$ to the radial ODE~\eqref{eq:radial-ODE-Psi-intro}. (We can also compare \eqref{eq:Wronskian-conservation-Psi-intro-2} where one $\omega$ weight is replaced by $\omega-m\upomega_+$ with a version of \eqref{eq:classical-energy-Psi-intro} where the same procedure has occurred, in which case the latter identity is the Killing energy associated to the null generator of the event horizon.) In contrast with Teukolsky--Starobinsky energy current, the Killing energy current has a physical space analogue. We note however that the right hand side of \eqref{eq:classical-energy-Psi-intro} is only trivial if $a=0$ or if $s=0$, in which cases the Teukolsky--Starobinsky and Killing energies coincide, i.e.\ the Killing energy is, in fact, a conserved quantity for the top equation in the transformed system, \eqref{eq:transformed-equation-intro}.

In our analysis for $s=\pm 1,\pm 2$ and general subextremal $|a|<M$ Kerr parameters, we will apply Teukolsky--Starobinsky energy currents by default, as they are the natural, conserved, energy for the transformed system and for the radial ODE~\eqref{eq:radial-ODE-Psi-intro} in particular. However, for rotating $a\neq 0$ Kerr black holes there may be admissible frequency triples $(\omega,m,\Lambda)$ for which the ratios $\mathfrak{C}_s/\mathfrak{D}_s^{\mc I}$ or $\mathfrak{C}_s/\mathfrak{D}_s^{\mc H}$ are not be upper or lower bounded (see Remark \ref{it:remark-intro-B} in the preamble to the introduction), in which case the Teukolsky--Starobinsky energy identities, such as \eqref{eq:Wronskian-conservation-Psi-intro-2}, do not guarantee good control over the boundary terms themselves, weighted by $\omega^2$ or $(\omega-m\upomega_+)^2$ for terms at $\mc{I}^\pm$ or $\mc{H}^\pm$, respectively. 

The degeneration of terms in \eqref{eq:Wronskian-conservation-Psi-intro-2} is problematic if the end goal is to control boundary terms as in  \eqref{eq:intro-main-thm-A}. It is also an issue if the estimate we seek requires the use of virial currents or of the rudimentary estimate \eqref{eq:basic-estimate-1-intro} that produce contributions to the boundary term ``lost'' by the Teukolsky--Starobinsky energy identity. Note that we employ virial currents in both cases of Theorem~\ref{thm:frequency-estimates-big}: for statement A, we do so to control a good bulk integral; for statement B, they will be necessary to absorb localization errors whenever $(\omega,m,\Lambda)\not\in\mc{A}$ are superradiant. 

By contrast, though the Killing energies are manifestly not conserved if $a,s\neq 0$ and introduce errors due to coupling, the boundary terms produced, e.g.\ in \eqref{eq:classical-energy-Psi-intro} are not weighted by any constants $\mathfrak{C}_s$, $\mathfrak{D}_s^{\mc H}$ or $\mathfrak{D}_s^{\mc{I}}$, and hence their $\omega^2$ or $(\omega-m\upomega_+)^2$ (for terms at $\mc{I}^\pm$ or $\mc{H}^\pm$) weighted versions do not degenerate for any $(\omega,m,\Lambda)$.  Thus, for any frequency triple $(\omega,m,\Lambda)$ for which $\mathfrak{C}_s/\mathfrak{D}_s^{\mc I}$ or $\mathfrak{C}_s/\mathfrak{D}_s^{\mc H}$ may degenerate, we must add a multiple of a Killing energy current, such as \eqref{eq:classical-energy-Psi-intro}, to our multiplier estimates and find a way of treating the coupling error terms introduced by this energy current. Moreover, the multiple cannot be small if the boundary term which identity \eqref{eq:Wronskian-conservation-Psi-intro-2} ``loses'' is one to which virial currents contribute. 

Over the next sections, we explain how to resolve these issues for each of the frequency regimes in the partition introduced in Section~\ref{sec:intro-estimates-Box}.

\subsubsection{Frequency space analysis of the coupling\texorpdfstring{ $\swei{\mathfrak{J}}{s}$}{}: bounded frequency parameters}
\label{sec:outline-proof-estimate-bdd}

When all frequency parameters are bounded, $|\omega|+|m|+|\Lambda|\lesssim 1$, if $|a|\ll M$, as in  \cite{Dafermos2017} (or \cite{Ma2017,Ma2017a}, obtained independently), we see that in \eqref{eq:transformed-equation-intro} the coupling terms appear to be error terms which do not affect the overall behavior of the equation. Indeed, the transformed equation \eqref{eq:radial-ODE-Psi-intro} can be treated analogously to the wave equation for $|a|\ll M$ in \cite{Dafermos2010}, and coupling terms can be absorbed by the resulting estimate simply by climbing the hierarchy with the rudimentary estimate~\eqref{eq:basic-estimate-1-intro}. 

For general subextremal $|a|< M$, it is already clear from \eqref{eq:transformed-equation-intro} that there should be no hope that the coupling terms are always neglectable. From the relation \eqref{eq:def-upppsi-k-freq-intro} or, more concretely, from the estimate \eqref{eq:basic-estimate-1-intro}, we could hope for a gain in $r$-weights that would make the coupling terms less significant as $r\to \infty$ and/or $r\to r_+$. However, the gain is hard to establish  simultaneously at both ends due to the fact that \eqref{eq:def-upppsi-k-freq-intro} is a \textit{non-local} relation. Moreover, it is unclear how to use \eqref{eq:basic-estimate-1-intro} to further provide smallness in the coupling errors supported for a range of bounded $r$-values which are bounded away from $r_+$: there, the coupling terms compete with the frequency-independent part of the potential $\swei{\mc{V}}{s}$, even when the frequency parameters are at their smallest.

In what follows, we describe how these difficulties are overcome for all bounded frequencies. 

\paragraph{Bounded frequency parameters with $1\lesssim |\omega|\lesssim 1$}

This is the set $\mc{A}$, excluded from Theorem~\ref{thm:frequency-estimates-big}A, but where estimates for Theorem~\ref{thm:frequency-estimates-big}B are recovered invoking Theorem~\ref{thm:mode-stability-intro-AB} (see our upcoming \cite{SRTdC2022}).

\paragraph{Bounded frequency parameters with $|\omega|\ll 1$}

In the small $|\omega|$ regime, in applying the methodology from Section~\ref{sec:intro-estimates-Box}, we produce several error terms due to coupling. The coupling errors are not multiplied by any small constant in the general subextremal $|a|< M$ case, hence they cannot be seen as lower order and, as we have seen, the estimate \eqref{eq:basic-estimate-1-intro} is not enough. Thus, if $|a|$ is not small, we need to use the radial ODEs~\eqref{eq:radial-ODE-uppsi-intro} for all $k=0,...,|s|$.

To do so, we resort to two observations:
\begin{itemize}[noitemsep]
\item The real component of the potential in \eqref{eq:radial-ODE-uppsi-intro}, $\Re\mc{V}_{(k)}^{[s]}$ enjoys similar properties, for all $k<|s|$, as the potential at the top level $k=|s|$. Thus, if we could ignore the imaginary component of the potential, the dependence on the first derivatice of the $k$th transformed variable and the coupling terms in \eqref{eq:radial-ODE-uppsi-intro}, all of the multiplier estimate techniques employed in for the wave equation, $s=0$, in \cite{Dafermos2016a}, could be applied for each level $k\leq |s|$.
\item The imaginary component of the potential and the dependence on the first derivative of the $k$th transformed variable in \eqref{eq:radial-ODE-uppsi-intro} can be combined to produce a ``coupling term'' to the $(k+1)$th equation, if $k<|s|$, plus an imaginary component to the potential which has good decay properties.
\end{itemize}
When we apply virial multipliers current to \eqref{eq:radial-ODE-uppsi-intro}  with the previous two remarks in mind, we obtain very similar estimates for each $k$: boundary terms aside, a bulk term of $\uppsi_{(k)}$ is controlled in terms of
\begin{itemize}[noitemsep]
\item if $k> 0$, bulk terms of $\uppsi_{(i)}$ for all $i<k$;
\item if $k< |s|$, a bulk term of $\uppsi_{(k+1)}$.
\end{itemize}
By judicious choices of currents, we can obtain multiplier estimates for each $k$ where at least one of the previous two kinds of terms comes with a small parameter. Then, while estimates for $k$ cannot be closed at level $k$ alone, they can be closed by borrowing from the multiplier estimates at level $k+1$ if $k<|s|$. Iterating along $k=0,\dots, |s|$, we can obtain a good estimate on the $\swei{\Psi}{s}$ bulk term. 

The simplified description of our strategy does not do justice to the multitude of boundary terms that we generate by it. These should be controlled by applications of energy currents. As the coupling terms generated by the Killing energy currents, such as \eqref{eq:classical-energy-Psi-intro}, come with too strong $r$-weights to be absorbed, we cannot tolerate but a small portion of these currents unless $|a|\ll M$ is a small parameter. Hence, for general subextremal $|a|< M$, the boundary terms we generate through virial currents must be absorbed using the Teukolsky--Starobinsky energy current \eqref{eq:Wronskian-conservation-Psi-intro-2} adapted to $s\neq 0$, which is immune to coupling terms. This strategy relies crucially on the fact that in the $\omega\to 0$ limit, one can establish suitable upper and lower bounds on the constants $\mathfrak{C}_s$, $\mathfrak{D}_s^{\mc I}$ and $\mathfrak{D}_s^{\mc H}$ (see Remark \ref{it:remark-intro-B} in the preamble and Section~\ref{sec:energy-identity-transformed}).

We also refer the reader to the later Section~\ref{sec:bounded-smallness-overview} for a more detailed discussion.

\subsubsection{Frequency space analysis of the coupling\texorpdfstring{ $\swei{\mathfrak{J}}{s}$}{}: unbounded frequency parameters}
\label{sec:outline-proof-estimate-unbdd}

In the case of unbounded frequency parameters, let us begin by identifying the coupling errors we need to contend with, arising from interaction of $\swei{\mathfrak{J}}{s}$ with virial or Killing energy currents applied to the radial ODE \eqref{eq:radial-ODE-Psi-intro}. By a simple application of Cauchy--Schwarz, we can show that the coupling errors due to the virial currents are no worse than bulk error terms of the form $a^2(m^2+1)|\uppsi_{(k)}|^2$, with suitable $r$-weights; by the rudimentary estimate \eqref{eq:basic-estimate-1-intro}: the errors $a^2(m^2+1)|\uppsi_{(k)}|^2$ should be controlled by $r$-weighted $a^2(m^2+1)|\Psi|^2$ bulk terms. On the other hand, Killing energy currents, which are not conserved for the radial ODE~\eqref{eq:radial-ODE-Psi-intro} when $a,s\neq 0$, also produce these errors in addition to a small multiple of $(\omega^2+m^2)|\Psi|^2$, with suitable $r$-weights.

We must now distinguish between non-trapped frequencies, for which $\omega^2-\mc{V}(r_{\rm max})$ has a clear sign, and trapped frequencies, for which $\omega^2\sim \mc{V}(r_{\rm max})$. Outside of trapping, controlling $(\omega^2+m^2)|\Psi|^2$ for all $r$ is precisely what the wave equation techniques from the Section~\ref{sec:intro-estimates-Box} guarantee; thus, if $|a|\ll M$, all coupling terms can be seen as small errors for any frequency triple which is not trapped. Absorbing the coupling errors in the general subextremal case $|a|<M$ will, of course, require extra knowledge even outside of trapping (see also Table~\ref{tab:difficulties-coupling}). 

Let us begin, however, by discussing the more involved case of trapped frequencies.

\paragraph{Trapped frequencies}

In the case of trapping, by a naive application of Cauchy--Schwarz as described above, we produce bulk terms $a^2(m^2+1)|\Psi|^2$ and, if a classical energy current is used,  $a^2(\omega^2+m^2)|\Psi|^2$, with suitable $r$-weights, which cannot be absorbed at the trapped $r$-value even in the very slowly rotating Kerr case.

However, there is still information to be extracted from the transport equation \eqref{eq:def-upppsi-k-freq-intro}. To start with, we note the following identities, originally obtained in \cite[Propositions 5.3.1 and 8.4.1]{Dafermos2017} but which hold for any $|a|\in[0,M]$: for $k=|s|-1$, 
\begin{align}
\begin{split}
- c\Im\lp[im\omega \overline{\Psi}\uppsi_{(k)}\rp]\sign s &= -\lp[cm\lp(\Im\lp[\overline{\Psi}\uppsi_{(k)}\rp]+\frac{am^2(r^2+a^2)}{2\Delta}\lp|\uppsi_{(k)}\rp|^2\rp)\rp]'\\
&\qquad+mc\Im\lp[\overline{\Psi}'\uppsi_{(k)}\rp]+\frac12\lp(\frac{(r^2+a^2)c}{\Delta}\rp)'am^2\lp|\uppsi_{(k)}\rp|^2 \,;
\end{split}\label{eq:basic-identity-3-intro}
\end{align}
and a similar identity can be obtained for $k<|s|-1$. Identity \eqref{eq:basic-identity-3-intro} and its $k<|s|-1$ analogue are especially useful in dealing with the coupling errors arising from applications of the Killing energy currents of waves: in the right hand side, applications of Cauchy--Schwarz allow us to shift all of the frequency weights onto the lower level transformed variables. Hence, we can treat error terms arising from the energy currents more delicately to conclude that, when the Killing energy current applied comes with multiplier $\omega$ as is the case at trapping (see e.g.\ Table~\ref{tab:current-wave-RW-operator}), the coupling errors reduce to $a^2(m^2+1)|\uppsi_{(k)}|^2$ bulk terms, $k<|s|$.

Absorbing $a^2(m^2+1)|\uppsi_{(k)}|^2$ bulk terms is still challenging for trapped frequencies: using estimate \eqref{eq:basic-estimate-1-intro} means that we must be able to control the coupling errors by bulk terms in $\Psi$ with the same frequency weights; however, at the trapping $r$-value we lose control over $(m^2+1)|\Psi|^2$. The estimate we seek, then, is one which controls $m^2|\uppsi_{(k)}|^2$ by $|\Psi'|^2+|\Psi|^2$ at least near the trapping $r$-value.

In the very slowly rotating $|a|\ll M$ case, the necessary estimate comes once more from the transport relation \eqref{eq:def-upppsi-k-intro}. With a bit more work, one can show \cite[Proof of Proposition 8.3.1]{Dafermos2017} 
\begin{align}
\begin{split}
\int_A^B c'\lp(\omega-\frac{am}{r^2+a^2}\rp)^2\lp|\uppsi_{(k)}\rp|^2dr^*&\leq \lp[c(r)\lp|\uppsi_{(k)}'\rp|^2\rp]_{r^*=A}^{r^*=B}+\int_{A}^B \frac{2\Delta^2}{(r^2+a^2)^4}\lp(\frac{2c^2}{c'}+c'\rp)\lp|\uppsi_{(k+1)}'\rp|^2dr^*\\
&\qquad+\int_A^B \frac{16 r^2\Delta^2}{(r^2+a^2)^6}\frac{c^2}{c'}a^2m^2|\uppsi_{(k)}|^2dr^*\,,
\end{split}\label{eq:basic-estimate-1'-intro}
\end{align}
where $c(r)$ is a continuous, piecewise $C^1$, $r$-weight with $c'>0$, and $A,B\in\mathbb{R}\cup \{\pm \infty\}$. In the trapping regime where $m^2\lesssim \Lambda\sim \omega^2$, if $|a|\ll M$ is sufficiently small, we have
\begin{align}
\lp(\omega-\frac{am}{r^2+a^2}\rp)^2-a^2m^2 \gtrsim \omega^2\gtrsim m^2\,, \label{eq:trapping-slowly-rotating-hack}
\end{align}
thus by \eqref{eq:basic-estimate-1'-intro} the coupling error terms can be absorbed at trapping for $|a|\ll M$. 

Turning to the general subextremal case $|a|<M$, we see that \eqref{eq:trapping-slowly-rotating-hack} cannot remain valid, hence \eqref{eq:basic-estimate-1'-intro} is of no use. Once again we look to the seemingly bad radial ODEs \eqref{eq:radial-ODE-uppsi-intro} to circumvent the issue. By a very simple procedure, we derive from \eqref{eq:radial-ODE-uppsi-intro} an ``improved estimate'' relating  $\uppsi_{(k)}$ and $\uppsi_{(k+1)}$,
\begin{align}
\int_{-\infty}^\infty w\Lambda|\uppsi_{(k)}|^2dr^*\leq B(M,|s|) \int_{-\infty}^\infty w|\uppsi_{(k+1)}|^2\,, \label{eq:basic-estimate-2-intro}
\end{align}
that \textbf{does not lose derivatives}, allowing us to bound bulk terms $a^2(m^2+1)|\uppsi_{(k)}|^2$ by bulk terms in $|\Psi|^2$ with suitable $r$-weights. Note that \eqref{eq:basic-estimate-2-intro} in combination with \eqref{eq:basic-identity-3-intro} and its $k<|s|-1$ analogue allows us to control any coupling error terms arising from application of a Killing energy current, such as \eqref{eq:classical-energy-Psi-intro}. 

Indeed, estimate \eqref{eq:basic-estimate-2-intro} is one of the key points in our proof. Though \eqref{eq:basic-estimate-2-intro} cannot hold in every unbounded frequency regime, we show that, remarkably, it holds for the entire range of trapped, nonsuperradiant frequencies. Likewise, the estimate $\mathfrak{C}_s\sim\mathfrak{D}_s^{\mc H}\sim \mathfrak{D}_s^{\mc{I}}$ holds for trapped, nonsuperradiant $(\omega,m,\Lambda)$ so that, in this frequency regime, the boundary terms produced by virial currents can be controlled by the Teukolsky--Starobinsky energy current \eqref{eq:Wronskian-conservation-Psi-intro-2} or by the Killing energy current \eqref{eq:classical-energy-Psi-intro}.

Since, for subextremal Kerr black holes, $|a|<M$, all superradiant frequencies are quantitatively nontrapped (see Section~\ref{sec:intro-trapping}), this concludes the analysis for trapped frequency parameters.

\paragraph{Non-trapped frequencies}

As we have mentioned, if $|a|\ll M$ is very small, then $a^2(m^2+1)|\uppsi_{(k)}|^2$ error terms can easily be absorbed outside of trapping; for general subextremal $|a|<M$, smallness must come from elsewhere. 
\begin{itemize}[itemsep=0pt]
\item \textit{The cases $\Lambda\gg \max\{\omega^2,1\}$, and $\Lambda\sim\omega^2\gg 1$ with $(\omega,m,\Lambda)$ superradiant non-trapped.} In these frequency regimes, our improved estimate \eqref{eq:basic-estimate-2-intro} provides the necessary smallness. Indeed, $a^2(m^2+1)|\uppsi_{(k)}|^2$ bulk terms for $k<|s|$ are controlled by $|\Psi|^2$ bulk terms by \eqref{eq:basic-estimate-1-intro}; these can easily be absorbed into the bulk term we control by the methods of Section \ref{sec:intro-estimates-Box}, which is an appropriately $r$-weighted version of $\Lambda|\Psi|^2\gg |\Psi|^2$. The  energy currents we consider can either be of Teukolsky--Starobinsky  energies or Killing energies: the latter's coupling error terms are controlled by application of \eqref{eq:basic-identity-3-intro} and \eqref{eq:basic-estimate-2-intro}; the former, as we show $\mathfrak{C}_s\sim\mathfrak{D}_s^{\mc H}\sim \mathfrak{D}_s^{\mc{I}}$ for all $(\omega,m,\Lambda)$ in these frequency ranges, gives direct control over the boundary terms with natural ($\omega^2$ or $(\omega-m\upomega_+)^2$) frequency weights.
\item \textit{The cases $\omega^2\gg \max\{\Lambda,1\}$, and $\Lambda\sim\omega^2\gg 1$ with $(\omega,m,\Lambda)$ non-superradiant and non-trapped.} In both these frequency regimes, one can show $(\omega,m,\Lambda)$ are not superradiant, hence no virial currents are needed to establish estimate \eqref{eq:intro-main-thm-B} in the scattering setting: an application of the Teukolsky--Starobinsky energy identity \eqref{eq:Wronskian-conservation-Psi-intro-2} suffices. Therefore, we focus on the setting of Theorem~\ref{thm:frequency-estimates-big}A, where both virial and energy estimates are required. As we cannot ensure that  \eqref{eq:basic-estimate-2-intro} holds, a different strategy is needed to treat the coupling terms, which are roughly represented by a bulk term in $a^2(m^2+1)|\uppsi_{(k)}|^2$. Moreover, we do not have appropriate control over $\mathfrak{C}_s$, $\mathfrak{D}_s^{\mc H}$ and $\mathfrak{D}_s^{\mc{I}}$, so Killing energy currents such as \eqref{eq:classical-energy-Psi-intro} must be invoked to deal with any boundary terms arising from virial currents or from applications of the rudimentary estimate \eqref{eq:basic-estimate-1-intro}.
\begin{itemize}[itemsep=0pt]
\item \textit{The case $\omega^2\gg \max\{\Lambda,1\}$.} Since $m^2,\Lambda\ll \omega^2$, the smallness we require to absorb coupling errors due to the Killing energy current or to virial currents is built-in to the frequency range. In exploiting this smallness, we appeal to the rudimentary estimate \eqref{eq:basic-estimate-1-intro}. The weight $c$ we choose either produces two boundary terms for $\uppsi_{(k)}$ at $r=r_+$, related to $A_{\mc{H}^\pm}$, or two boundary terms at $r=\infty$, related to $A_{\mc{I}^\pm}$; but in the relation between $\uppsi_{(k)}$ and $\Psi$ boundary terms, the constant $\mathfrak{D}_s^{\mc H}$ appears multiplying one of $A_{\mc{H}^\pm}$ and the constant $\mathfrak{D}_s^{\mc I}$ appears multiplying one of $A_{\mc{I}^\pm}$. In the setting of Theorem~\ref{thm:frequency-estimates-big}A, where only  the future boundary terms of $\Psi$ are nonzero, we can choose $c$ to avoid these problematic constants. Finally, we are ready to control the boundary terms generated: the Killing energy current \eqref{eq:classical-energy-Psi-intro} is certainly strong enough for the task and its coupling terms are controlled; the Teukolsky--Starobinsky energy identity \eqref{eq:Wronskian-conservation-Psi-intro-2} can also be relied on to control one (but possibly not both, see above) of the future boundary terms.

\item \textit{ The case $\Lambda\sim\omega^2\gg 1$ with $(\omega,m,\Lambda)$ non-superradiant and non-trapped.} Unlike in the previous case, there is no smallness built into the frequency range: by the rudimentary estimate \eqref{eq:basic-estimate-1-intro}, we can only hope to convert bulk terms like $a^2(m^2+1)|\uppsi_{(k)}|^2$ into bulk terms of size $\Lambda|\Psi|^2$, i.e.\ bulk terms comparable with what we control by the methods of Section~\ref{sec:intro-estimates-Box}. We take two routes to circumvent this difficulty:
\begin{itemize}[itemsep=0pt]
\item  If errors arise due to application of the virial currents from Section~\ref{sec:intro-estimates-Box}, we adjust the currents' growth so that a small parameter falls on $a^2(m^2+1)|\uppsi_{(k)}|^2$ during application of Cauchy--Schwarz.
\item If errors arise due to application of the Killing energy current \eqref{eq:classical-energy-Psi-intro}, we multiply the current by a small parameter. Because the Teukolsky--Starobinsky energy identity \eqref{eq:Wronskian-conservation-Psi-intro-2} can only relied on to control one of the future boundary terms, this strategy can only hold if the virial current is constructed so as not to provide a contribution with a bad sign to the boundary term ``lost'' by the Teukolsky--Starobinsky energy identity and ``recovered'' by the Killing energy current.
\end{itemize}
\end{itemize}
\end{itemize}
\medskip

Table~\ref{tab:currents-small-vs-large-a} summarizes the previous exposition concerning our methods to address the difficulties in the frequency space analysis in the current paper, as raised in Table~\ref{tab:difficulties-coupling}. We also refer the reader to the later Section~\ref{sec:unbounded-overview} for a more detailed discussion.

\begin{table}[!htbp]
\centering
\caption{The frequency space analysis of the coupling $\swei{\mathfrak{J}}{s}$ in \eqref{eq:transformed-equation-intro}: treatment of the difficulties, identified in Table~\ref{tab:difficulties-coupling}, associated with the frequency space analysis of the coupling $\swei{\mathfrak{J}}{s}$ in \eqref{eq:transformed-equation-intro}, in the case $|a|\ll M$ and in the full subextremal range $|a|<M$. Note that the case $|a|\ll M$ can be treated ignoring all but one of the equations in the transformed system, whereas in the general $a$ case, we occasionally appeal to the radial ODE for $\uppsi_{(k)}$.} 
\begin{small}
\begin{tabular}{|c|M{1.2cm}|M{4.0cm}|M{9.0cm}|} \hline
\multicolumn{2}{|c|}{$(\omega,m,\Lambda)$} &$|a|\ll M$ \cite{Dafermos2017} &  $|a|<M$\\ \hline
\multirow{2}{*}[-1.5em]{\rotatebox{90}{bounded}}& $|\omega|\gtrsim 1$  & \multirow{2}{4.0cm}[-1em]{\centering coupling terms are  small, treat as wave equation for $|a|\ll M$} & appeal to quantitative mode stability \cite{Shlapentokh-Rothman2015,TeixeiradaCosta2019}\\ \cline{2-2} \cline{4-4}
&$|\omega|\ll 1$& & apply multiplier estimates to all $k$, controlling $k$th bulk in terms of $i<k$ bulk and $(k+1)$ bulk, one of these with smallness; iterate along hierarchy to obtain good bulk estimate for $k=|s|$ and finalize with Teukolsky--Starobinsky energy currents, e.g.\ \eqref{eq:Wronskian-conservation-Psi-intro-2} \\ \hline
\multirow{4}{*}[-7em]{\rotatebox{90}{unbounded}}& \multirow{2}{*}[-1em]{all} & \multicolumn{2}{M{13.0cm}|}{use transport equation estimates to show worse coupling error terms are bulk terms $a^2(m^2+1)\Delta r^{-4}|\uppsi_{(k)}|^2$}\\ \cline{3-4}
& & -- & use equation for $\uppsi_{(k)}$ to show that, when $\Lambda-2am\omega\gtrsim \Lambda\gg 1$, have improved estimate: the bulk term $a^2(1+m^2)\Delta r^{-4}|\uppsi_{(k)}|^2$ is controlled by $\Delta r^{-4}|\Psi|^2$\\ \cline{2-4}
&non-trapped & use smallness of $a$ & have smallness from the frequency range OR use improved elliptic estimate for $\uppsi_{(k)}$ if available  OR create smallness by: (1) for errors due multiplier current, manipulating the ratio current-to-derivative; (2) for errors due to the energy current, adding only a small multiple of it and supplementing by Teukolsky--Starobinsky energy currents, e.g.\ \eqref{eq:Wronskian-conservation-Psi-intro-2} \\ \cline{2-4}
&trapped & use transport equqation alone to prove, using smallness of $|a|$, $a^2(1+m^2)\Delta r^{-4}|\uppsi_{(k)}|^2$ bulk is controlled by trapped bulk term $\Delta r^{-4}|\Psi|^2$ & use potential to show that improved estimate must always be available when trapped \\ \hline
\end{tabular}
\end{small}
\label{tab:currents-small-vs-large-a}
\end{table}

\subsection{Outline and acknowledgments} 

This paper is organized as follows.
\begin{itemize}
\item \textbf{Section 2.} We recall to the reader the Kerr background geometry and fix our notation for relevant vector fields,  operators and other shorthand notation. We also introduce the class of smooth $s$-spin-weighted functions, allowing us to make sense of \eqref{eq:Teukolsky-equation-intro}.
\item \textbf{Section 3.} We introduce the main PDEs we will study in our series: the Teukolsky equation and the transformed system. We justify our choices in constructing the transformed system and establish its relevant properties for our analysis. These were already discussed in Section~\ref{sec:intro-transformed} of this introduction.
\item \textbf{Section 4.} We derive the angular and radial ODEs associated to the differential equations of the previous section under the separability assumption discussed in Sections~\ref{sec:intro-separability} and \ref{sec:intro-to-freq-space} of this introduction. We perform a crude analysis of the angular ODE which allows us to obtain rudimentary bounds for the separation constant $\Lambda$ in terms of the remaining frequency parameters, $\omega$ and $m$. For the radial ODEs, we perform a finer asymptotic analysis and discuss the boundary conditions for the solutions we seek to consider. Finally, we recall the precise statement of the Teukolsky--Starobinsky identities verified by the Teukolsky angular and radial ODE.
\item \textbf{Section 5.} We define the multiplier currents we will consider in our analysis of the radial ODEs introduced in the previous section. These involve virial estimates, the classical energy currents from wave analysis, and the Teukolsky--Starobinsky energy current discussed in Section~\ref{sec:energy-identity-transformed} of this introduction. 

\item \textbf{Section 6.} This section contains the proof of our main result in the present paper, Theorem~\ref{thm:frequency-estimates-big}. We begin, in Section~\ref{sec:frequency-partitions} by partitioning  frequency space into the frequency regimes which we will address individually. Then, in Section~\ref{sec:basic-estimates-transformed} we prove estimates for the lower level transformed variables based only on the transport equations. Finally, we give proofs of our main result in the bounded frequency regimes in Sections~\ref{sec:bounded-smallness} and \ref{sec:intermediate} and in the unbounded frequency regimes in Section~\ref{sec:unbounded}; in each of these sections we give an overview of our constructions, reinforcing the relevant points made in Section~\ref{sec:intro-proof-freq-estimates} and introducing the extra estimates on the lower level transformed variables we will need. The estimates for high and low frequency are combined in Section~\ref{sec:combining-ode-estimates} to yield Theorem~\ref{thm:frequency-estimates-big}A; and, by combining our estimates with those arising from real mode stability (Theorem~\ref{thm:mode-stability-intro-AB}), we also show Theorem~\ref{thm:frequency-estimates-big}B.
\end{itemize}

\vspace{\baselineskip}

\noindent\textbf{Acknowledgements.} During the period this research was carried out, YS was supported by  NSF grant DMS-1900288.  RTdC acknowledges support through EPSRC grant EP/L016516/1 as well as the hospitality of Princeton University. The authors thank Mihalis Dafermos and Igor Rodnianski for discussions and comments on the paper.

\section{Preliminaries}

\subsection{The Kerr exterior spacetime and metric}

In this section, we introduce the Kerr exterior. We refer the reader to \cite[Section 2]{Dafermos2010} for more detail.

For Kerr black hole parameters $(a,M)$ satisfying $M>0$ and $|a|\leq M$ and define
\begin{equation}
r_\pm:=M\pm\sqrt{M^2-a^2}\,. \label{eq:rpm}
\end{equation} 
The  Kerr exterior spacetime, $\mc{R}$, is a manifold-with-boundary which is covered by Kerr-star coordinates  $(t^*,r,\theta^*,\phi^*)\in \mathbb{R}\times [r_+,\infty)\times \mathbb{S}^2$ globally, apart from the usual degeneration of spherical coordinates. The future event horizon is defined to be $\mc{H^+}:=\p \mc{R}=\{r=r_+\}$. 

It will also be convenient to define a new radial coordinate $r^*\colon(r_+,\infty)\to (-\infty,\infty)$ which is the unique function satisfying $r^*(3M)=0$ and
\begin{gather}
\frac{dr^*}{dr}=\frac{r^2+a^2}{\Delta} \text{,~~with~~} \Delta:=r^2-2Mr+a^2=(r-r_+)(r-r_-)\,. \label{eq:r-star}
\end{gather}

When we are interested in $\mr{int}(\mc{R})$ only, it will be convenient to consider  Boyer--Lindquist coordinates $(t,r,\theta,\phi)\in\mathbb{R}\times (r_+,\infty)\times \mathbb{S}^2$. These are obtained from Kerr-star coordinates by the relations
\begin{gather}
t(t^*,r):=t^*-\overline{t}(r)\,,  \quad \quad \theta:=\theta^*\, , \quad \quad \phi(\phi^*,r):=\phi^* -\overline{\phi}(r) \mod 2\pi\,, \label{eq:kerr-star}
\end{gather}
where the functions $\overline{\phi}$ and $\overline{t}(r)$ can be defined so as to vanish for $r\geq 9/4M$ (see \cite[Section 2.4]{Dafermos2010}). With respect to Boyer--Lindquist coordinates, the Kerr metric becomes
\begin{align}
\begin{split}
g_{a,M} &= -\frac{\Delta -a^2\sin^2\theta}{\rho^2}dt^2- \frac{4Mar\sin^2\theta}{\rho^2} dtd\phi  \\ 
&\qquad+\left(\frac{(r^2+a^2)^2-\Delta a^2 \sin^2\theta}{\rho^2}\right)\sin^2\theta d\phi^2 +\frac{\rho^2}{\Delta}dr^2+\rho^2 d\theta^2\,,
\end{split} \label{eq:kerr-metric}
\end{align}
where 
\begin{gather}
\rho^2 := r^2+a^2\cos^2\theta, \quad\quad \Delta:=r^2-2Mr+a^2=(r-r_+)(r-r_-) \,. \label{eq:rho-delta}
\end{gather}

One can further extend $\mc{R}$ as follows. Starting from Boyer--Lindquist coordinates, we define coordinates $(^*t,r,^*\theta,^*\phi)\in\mathbb{R}\times[r_+,\infty)\times\mathbb{S}^2$ by 
\begin{gather}
^*t(t,r):=t-\overline{t}(r)\,,  \quad \quad ^*\theta:=\theta\,, \quad \quad ^*\phi(\phi,r):=\phi -\overline{\phi}(r) \mod 2\pi\,. \label{eq:kerr-star-2}
\end{gather}
In such coordinates, the metric extends smoothly to $\mc{H}^-:=\{r=r_+\}$ (by $r$ we mean the variable in the $(^*t,r,^*\theta,{}^*\phi)$ coordinates), so $\mc{R}$ can be extended to a larger manifold with boundary $\mc{R}\cup\mc{H}^-$.

\subsection{Smooth spin-weighted functions}

In this sequence of papers, we will consider partial differential equations in the space of spin-weighted function. The present section will contain the relevant regularity spaces for this type of function.

Letting $(\theta,\phi^*)$ denote standard spherical coordinates in $\mathbb{S}^2$, consider the vector fields
\begin{gather}
\begin{gathered} \label{eq:Zi}
\tilde{Z}_1 = -\sin \phi^* \partial_\theta + \cos \phi^*  \left( -is \csc \theta - \cot \theta \partial_{\phi^*}\right) \,, \\
\tilde{Z}_2 = - \cos \phi^* \partial_\theta - \sin \phi^*  \left( -is \csc \theta - \cot \theta \partial_{\phi^*}\right) \,, \quad 
\tilde{Z}_3 = \partial_{\phi^*} \, .
\end{gathered}
\end{gather}

We recall the following notion of smooth $s$-spin-weighted function (see \cite[Section 2.2.1]{Dafermos2017}):
\begin{definition}[Smooth spin-weighted functions] \label{def:smooth-spin-weighted}
Fix some $s\in\frac12\mathbb{Z}$.

\begin{enumerate}[label=(\roman*)]
\item Let $f$ be a complex-valued function of $(\theta,\phi^*)\in\mathbb{S}^2$. We say $f$ is a {\normalfont smooth $s$-spin-weighted function on $S^2$}, i.e.\ $f\in \swei{\mathscr{S}}{s}_\infty$,  if for any $k_1,k_2,k_3 \in \mathbb{N}_0$, 
$$(\tilde{Z}_1)^{k_1} (\tilde{Z}_2)^{k_2} (\tilde{Z}_3)^{k_3} f$$ 
is a function of $(\theta,\phi)$ which is smooth for $\theta\neq 0,\pi$ and  such that
$$e^{is\phi}(\tilde{Z}_1)^{k_1} (\tilde{Z}_2)^{k_2} (\tilde{Z}_3)^{k_3} f \text{~and~}e^{-is\phi}(\tilde{Z}_1)^{k_1} (\tilde{Z}_2)^{k_2} (\tilde{Z}_3)^{k_3} f$$ 
extend continuously to, respectively, $\theta=0$ and $\theta=\pi$.  \label{it:smooth-spin-weighted-S2}

\item Now, fix Kerr parameters $M>0$ and $|a|\leq M$. Consider a function $f$ of the Kerr-star variables $(t^*,r,\theta,\phi^*)\in\mathbb{R}\times [r_+,\infty)\times\mathbb{S}^2$. We say $f$ is a {\normalfont smooth $s$-spin-weighted function on $\mc{R}$}, i.e.\ $f\in\mathscr{S}_\infty^{[s]}(\mathcal{R})$, if for any $k_4,k_5 \in \mathbb{N}_0$, the function
$$(\p_{t^*})^{k_4} (\p_{r})^{k_5} f\Big|_{(t^*,r)=(T^*,R)}$$
is a smooth $s$-spin-weighted function on $\mathbb{S}^2$ for any $T^*\in\mathbb{R}$ and $R\in[r_+,\infty)$. \label{it:smooth-spin-weighted-R}
\end{enumerate}
\end{definition}

\subsection{Relevant vector fields and differential operators}
\label{sec:prelims-vector-fields-operators}

We will often use the shorthand notation $T:=\p_{t^*}$ and $\Phi:=\p_{\phi^*}$, where $t^*$ and $\phi^*$ are part of the Kerr-star system of coordinates.  Moreover, letting
\begin{align*}
\upomega_+:=\frac{a}{2Mr_+}\,,
\end{align*}
we let $K=T+\upomega_+\Phi$ be the null generator of the future event horizon, $\mc{H}^+$.

Recall the definition of $r^*$ in \eqref{eq:r-star}. In $(t,r^*,\theta,\phi)$ coordinates, let
\begin{align*}
L:=\p_{r^*}+T+\frac{a}{r^2+a^2}\Phi\,, \qquad \uL:=-\p_{r^*}+T+\frac{a}{r^2+a^2}\Phi\,,
\end{align*}
define two principal null directions on Kerr, with the normalization condition
\begin{align*}
g(L,\uL)=-2\rho^2\frac{\Delta}{(r^2+a^2)^2}\,.
\end{align*}

Now, consider functions in $\mathscr{S}_\infty^{[s]}$. For any $s\in\frac{1}{2}\mb{Z}$, the spin-weighted laplacian 
\begin{align} \label{eq:spin-weighted-laplacian}
\mathring{\slashed\triangle}^{[s]}&=  -\frac{1}{\sin \theta} \frac{\partial}{\partial \theta} \left(\sin \theta \frac{\partial}{\partial \theta}\right) - \frac{1}{\sin^2 \theta} \partial_\phi^2 - 2s i\frac{ \cos \theta}{\sin^2 \theta} \partial_\phi + s^2 \cot^2 \theta - s \\
&=-\tilde{Z}_1^2-\tilde{Z}_2^2-\tilde{Z}_3^2-s-s^2 \nonumber
\end{align}
where $\tilde Z_1$, $\tilde Z_2$ and $\tilde Z_3$ are defined in \eqref{eq:Zi} and the spinorial gradient 
\begin{align} \label{eq:spin-gradient}
\mathring{\slashed{\nabla}}^{[s]}=( \partial_\theta\,,~ 
\partial_\phi + is \cos \theta)\,,
\end{align}
are smooth differential operators on $\mathscr{S}_\infty^{[s]}$. For functions in $\mathscr{S}_\infty^{[s]}(\mc{R})$, the following notation  will be convenient:
\begin{align*}
\mathring{\slashed\triangle}^{[s]}_T:=\mathring{\slashed\triangle}^{[s]}-a^2\sin^2\theta TT-2aT\Phi+2ias\cos\theta T\,.
\end{align*}

\subsection{Parameters and conventions}

Throughout the paper, we rely on the notation
\begin{align}
w:=\frac{\Delta}{(r^2+a^2)^2}\,,
\end{align}
for an $r$-weight that will be heavily used; recall $\Delta$ is given in \eqref{eq:rho-delta}.

In our estimates, we use $B$ to denote possibly large positive constants and $b$ to denote possibly small positive constants depending only on $M>0$. Whenever the constant depends, additionally, on another parameter that has not yet been fixed, say $x$, we write $B(x)$ or $b(x)$; we revert to $B$ and $b$ once it has been fixed. We also note the algebra of constants
$$B+B=BB=B\,, \qquad b+b=bb=b\,, B+b=B\,, \qquad Bb =B\,, \qquad b^{-1}=B\,,\text{~etc}.$$


\section{The Teukolsky equation and the transformed system}

\subsection{The Teukolsky equation}

Consider the differential operator, $\mathfrak{T}^{[s]}$, which is given in Boyer--Lindquist coordinates by
\begin{equation} 
\begin{split}
\rho^2 \mathfrak{T}^{[s]}
&=\Delta^{-s}\p_r(\Delta^{s+1}\p_r)-\frac{(r^2+a^2)^2}{\Delta}\lp(\p_t+\frac{a}{r^2+a^2}\p_\phi\rp)^2+s\frac{w'}{w^2}\lp(\p_t+\frac{a}{r^2+a^2}\p_\phi\rp)+\frac{4sar}{r^2+a^2}\p_\phi\\
&\qquad -\mathring{\slashed{\triangle}}^{[s]} +(2a\p_t\p_\phi+a^2\sin^2\theta \p_t^2-2ias\cos\theta\p_t)\,.
\\
\end{split}\label{eq:teukolsky-operator}
\end{equation}

We say $\swei{\upalpha}{s}\in\mathscr{S}_\infty^{[s]}(\mc{R}\backslash \mc{H}^+)$ satisfies the \textit{inhomogeneous} Teukolsky equation if
\begin{equation}
\rho^{2}\mathfrak{T}^{[s]}\upalpha^{[s]}=\frac{\Delta^{|s|/2(1-\sign s)}}{(r^2+a^2)^{|s|}}\swei{F}{s}. \label{eq:teukolsky-alpha}
\end{equation}
for some $\swei{F}{s}\in\mathscr{S}_\infty^{[s]}(\mc{R}\backslash \mc{H}^+)$ and the homogeneous Teukolsky equation if $\swei{F}{s}$ vanishes.

\subsection{The transformed system}
\label{sec:transformed-system}

The Teukolsky equation~\eqref{eq:teukolsky-alpha} differs from the wave equation by a collection of lower order derivatives, the most concerning being those taken in $\p_t$ and $\p_\phi$ direction. In the separated picture (see Section~\ref{sec:radial-ODE-Teukolsky-transformed}), they translate into an imaginary, slowly decaying potential for the radial equation, even for Schwarzschild, for which one cannot expect the strategy employed for the wave equation in \cite{Dafermos2016b} to be successful.

To circumvent issues associated with the imaginary component of the potential, several authors have introduced transformations that, applied to the Teukolsky equation, could alter it. The idea goes back to the fixed-frequency work of Chandrasekhar for Schwarzschild \cite{Chandrasekhar1975a}, who constructed differential transformations connecting the Bardeen--Press \cite{Bardeen1973} equation (Teukolsky equation for $a=0$) and the Regge--Wheeler equation \cite{Regge1957} or the Zerilli equation \cite{Zerilli1970a}. In the separated picture, the latter have real, short-range potentials, similar to that of the wave equation, allowing for the same type of estimates as for the wave equation. For this reason, a generalization of Chandrasekhar's work to physical space was crucially used in the proof of linear stability for Schwarzschild \cite{Dafermos2016a}.

Several attempts have been made to generalize Chandrasekhar's transformations for Kerr spacetimes, such as \cite{Chandrasekhar1976,Sasaki1982,Hughes2000} and, more recently, \cite{Dafermos2017}. However, when $a\neq 0$, the type of differential transformations considered cannot, as this section will demonstrate, yield an equation with a wave-type potential in frequency space which decouples from the Teukolsky equation.

In this section, based on the framework of \cite{Hughes2000}, we introduce a system of equations which will be the starting point for our estimates, where the latter coupling for $a\neq 0$ is at least \textit{mild} from the point of view of the estimates to follow. The transformations and the equations derived here reduce to the transformations already introduced for $s=\pm 1$ in \cite{Pasqualotto2016}
 and $s=\pm 2$ in \cite{Dafermos2016a,Dafermos2017}.

\subsubsection{Constructing the relevant transformations}

We begin with a sufficiently general transformation, for integer spins, $s\in\mb{Z}$. Let $\swei{\upalpha}{s}\in\swei{\mathscr{S}}{s}_\infty(\mc{R}\backslash \mc{H}^+)$ be a solution to the inhomogeneous Teukolsky equation (\ref{eq:teukolsky-alpha}). For some $j_{s,\,k}=j_{s,\,k}(r)$ with $k=0,...,s$, consider the rescaling of the Teukolsky variable
\begin{align}
\swei{\uppsi}{s}_{(0)}:=j_{s,\,|s|}^{-1}\swei{\upalpha}{s}\,, \label{eq:def-psi0-general}
\end{align}
and the transformed quantities
\begin{align} \label{eq:Psi-general-transformation}
\swei{\uppsi}{s}_{(k)}:=j_{s,\,|s|-k}^{-1}\mathcal{L}\swei{\uppsi}{s}_{(k-1)}\,,\qquad   k=1,...,|s|\,, 
\end{align}
with $\mathcal{L}=L$ if $s<0$, $\mathcal{L}=\underline{L}$ if $s>0$ and $\mathcal{L}$ being the identity if $s=0$.

To obtain a wave-type equation for $\swei{\Psi}{s}:=\swei{\uppsi}{s}_{(|s|)}$, we iterate the following procedure. At each step $k=0,...,|s|$, divide by $j_{s,\,|s|-k}$ and apply $\mc{L}$. To obtain the $k$th level equation, simplify each term containing two derivatives, except the term where both derivatives are along the null directions $L$ and $\uL$, by using the definitions of $\swei{\uppsi}{s}_{(i)}$ for $i=0,...,k$; then, rewrite $\mc{L}\underline{\mc{L}}$ using the relation
$$\underline{L}L=\frac{1}{2}\lp(L\underline{L}+\underline{L}L\rp)+\frac{2raw}{r^2+a^2}\Phi\,.$$
Otherwise, write $\mc{L}\underline{\mc{L}}=\underline{\mc{L}}\mc{L}+[\mc{L},\underline{\mc{L}}]$, simplify further introducing $\swei{\uppsi}{s}_{k+1}$, and repeat the procedure. For each $k=0,...,|s|$, we find that 
\begin{gather*}
\mc{L}\underline{\mc{L}}\swei{\uppsi}{s}_{(k)}+\mr{sgn}(s)\lp(\sum_{i=1}^k\frac{j_{s,\,|s|-i}'}{j_{s,\,|s|-i}}-\swei{h}{s}\rp)\mc{L}\swei{\uppsi}{s}_{(k)} + \sum_{x\in\mathfrak{X}} \sum_{i=0}^k c_{s,\,k,\,i}^X \lp(X\swei{\uppsi}{s}_{(k)}\rp)\\= \lp[\prod_{i=0}^{k-1} j_{s,\,|s|-k+i}^{-1}\mathcal{L}\rp]\lp(\frac{\Delta^{|s|/2(1-\sign s)}}{(r^2+a)^{|s|}j_{s,\,|s|}}wF \rp)\,, \numberthis \label{eq:pde-Psi-i}
\end{gather*}
where the differential operator on the inhomogeneity is absent if $k=0$ and where $\mathfrak{X}=\{T,\,\Phi,\,\mathring{\slashed{\triangle}}^{[s]}_T,\, \mr{id}\}$. Here,  $c_{s,\,k,\,i}^X$ are functions of $r$ which can be obtained from the recursive relation
\begin{align}
\begin{dcases}
c_{s,\,k,\,0}^X&=-\mr{sgn}(s)\lp(\frac{c_{s,\,k-1,\,0}^X}{j_{s,|s|-k}}\rp)' \\
c_{s,\,k,\,i}^V&=-\mr{sgn}(s)\lp(\frac{c_{s,\,k-1,\,i}^X}{j_{s,|s|-k}}\rp)'+\frac{j_{s,\,|s|-k+1}}{j_{s,\,|s|-k}}c_{s,\,k-1,\,i-1}^X \quad \text{~ for ~} i=0,...,k-1 
\end{dcases}\,, \label{eq:recursive-relation-coupling-coefs}
\end{align} 
with initial values given at any $k=i=0,...,|s|$ by
\begin{equation}
\begin{split}
c_{s,\,k,\,k}^{\mathring{\slashed{\triangle}}}&= w\,,\\
c_{s,\,k,\,k}^T&=2s\frac{M(r^2-a^2)-r\Delta}{(r^2+a^2)^2}-\mr{sgn}(s)\swei{h}{s}(r)\,,\\ 
c_{s,\,k,\,k}^\Phi &= 2s\frac{a(r-M)}{(r^2+a^2)^2}-\mr{sgn}(s)\frac{a}{r^2+a^2}\swei{h}{s}(r)+\mr{sgn}(s)\frac{(4k-2)arw}{r^2+a^2}\,, \\
c_{s,\,k,\,k}^{\mr{id}} &= -\Delta^{-s}\p_r\lp(\Delta^{s+1}\p_r j_{s,\,s}(r)\rp) \frac{w}{j_{s,\,s}(r)}+\lp(k\swei{h}{s}-\sum_{i=1}^k(k-i+1)\frac{j_{s,\,|s|-i}'}{j_{s,\,|s|-i}}\rp)'\,, \\
\swei{h}{s}(r)&= -\frac{w}{j_{s,\,s}(r)}\frac{r^2+a^2}{\Delta}\lp(\Delta\p_r j_{s\,,s}(r) +\Delta^s\p_r\lp(\Delta^{s+1}j_{s,\,s}(r)\rp)\rp)-w\Delta\p_r\lp(\frac{r^2+a^2}{\Delta}\rp)\,.
\end{split}\label{eq:recursive-relation-coupling-coefs-i=k}
\end{equation}

For $\swei{\Psi}{s}$ to satisfy a PDE which is decoupled from the remaining $\swei{\uppsi}{s}_{(k)}$, the weights $c^X_{s,|s|,i}$ must vanish for all $i=0,...,|s|-1$ and all $X\in\mathfrak{X}$; the requirement can be cast as a system of $4|s|$ coupled ODEs ($3|s|$ if $a=0$, because then $c^\Phi_{s,|s|,i}=0$ automatically) for $|s|+1$ unknowns -- $j_{s,\,0}(r)$ through $j_{s,\,|s|}(r)$. Though, for $a=0$, this overdetermined system has more than one solution, it is easy to check that there are no solutions even for $|s|=1$ when $a\neq 0$:

\begin{lemma} \label{lemma:no-decoupling-general-Kerr} Consider the recursive relation \eqref{eq:recursive-relation-coupling-coefs} initialized by \eqref{eq:recursive-relation-coupling-coefs-i=k}. For $|s|=1,2$, the system
$$c^X_{s,|s|,i}=0\quad \forall i=0,...,|s|-1\,, \,\,\forall X \in\mathfrak{X}:=\{T,\,\Phi,\,\mathring{\slashed{\triangle}}^{[s]}_T,\, \mr{id}\},$$ 
has no solutions for $|s|=1,2$ unless $a=0$.
\end{lemma}
\begin{proof}
Consider first $|s|=1$. We have to solve four ODEs for the two unknowns $j_{\pm 1, 1}(r)$ and $j_{\pm 1,0}(r)$. We choose the two ODEs
\begin{align*}
&\lp(\frac{c_{\pm 1,\,0,\,0}^{\mathring{\slashed{\triangle}}}}{j_{\pm 1,0}}\rp)'=0  \Leftrightarrow j_{\pm 1,0}=Aw\,,\\
&\lp(\frac{c_{\pm 1,\,0,\,0}^T}{j_{\pm 1,0}}\rp)'=0 \Leftrightarrow
 j_{\pm 1,1} =\Delta^{\mp (1\pm 1)/2}(r^2+a^2)^{1/2}\times 
 \begin{dcases}
 \exp \left(\pm \frac{B}{2a} \arctan  \left(\frac{r}{a}\right)\right) &\text{~if~} a \neq 0\\
 \exp \left(\mp \frac{B}{2r} \right) &\text{~if~} a = 0
   \end{dcases}\,,
\end{align*}
and conclude
\begin{gather*}
\lp(\frac{c^\Phi_{\pm 1,\,0,\,0}}{j_{\pm 1,0}}\rp)' = \frac{2 a \lp[Br\pm(r^2-a^2)\rp]}{A(r^2+a^2)}\,,\qquad
\lp(\frac{c^{\rm id}_{\pm 1,\,0,\,0}}{j_{\pm 1,0}}\rp)'=\frac{\left(4 a^2+B^2\right) \left(a^2 (M+r)+r^2 (r-3 M)\right)}{2 A\left(a^2+r^2\right)\Delta}\,,
\end{gather*}
hence, a full decoupling is only possible if $B=0$ and, in addition, $a=0$. The case $|s|=2$ is similar.
\end{proof}

Our only hope, therefore, is to choose weights which allow the coupling to be as mild as possible.

\subsubsection{The DHR transformed system}

In this section, we will introduce our choice of transformed quantities, which are a  generalization of \cite{Dafermos2017}. Our goal is to choose weights $j_{s,\,k}(r)$, $k=0,...,|s|$ such that our transformations yield a more tractable, from the point of view of the estimates to follow, system than the Teukolsky equation.

We would like the coupling of any $\swei{\uppsi}{s}_{(k)}$ to the $\swei{\uppsi}{s}_{(i)}$, $i=0,...,k-1$, to involve of at most first derivatives of the latter quantities. Such a requirement is equivalent to asking that the coefficients of $\mathring{\slashed{\triangle}}_T^{[s]}$ vanish, for which a sufficient condition is having $j_{s,\,k}=w$ for any $k=0,...,|s|-1$.  Under this choice, the $\mc{L}$ derivative of $\swei{\Psi}{s}$ is absent in the ODE for this quantity (i.e.\ the equation will be in canonical form). 

On the other hand, to avoid a complex potential in the next section, when we perform separation of variables, we would like the equation for $\swei{\Psi}{s}$ not to involve time or azimuthal derivatives of this quantity. Luckily, these conditions are compatible and they yield 
\begin{equation}
j_{s,\,|s|}(r)= w^{|s|/2}(r^2+a^2)^{-1/2}\Delta^{-s/2}\,.
\end{equation}

\begin{proposition}\label{prop:transformed-system} Fix $s\in\mathbb{Z}$. Let $\swei{\upalpha}{s}\in\swei{\mathscr{S}}{s}_\infty(\mc{R}\backslash\mc{H}^+)$ be a solution to the Teukolsky equation (\ref{eq:teukolsky-alpha}). Let $\swei{\uppsi}{s}_{(0)}\in\swei{\mathscr{S}}{s}_\infty(\mc{R}\backslash\mc{H}^+)$ be the rescaling of the Teukolsky variable given by
\begin{equation}\label{eq:def-psi0}
\swei{\uppsi}{s}_{(0)}:=(r^2+a^2)^{-|s|+1/2}\Delta^{|s|/2(1+\mr{sign}\,s)}\swei{\alpha}{s}\,.
\end{equation}
Define $\swei{\uppsi}{s}_{(k)}\in\swei{\mathscr{S}}{s}_\infty(\mc{R}\backslash\mc{H}^+)$ by the system of transport equations
\begin{equation}
\swei{\uppsi}{s}_{(k)}:=\frac{1}{w}\mc{L}\swei{\uppsi}{s}_{(k-1)}\,,\qquad  k=1,...,|s|\,, \label{eq:transformed-transport}
\end{equation}
with $\mathcal{L}=L$ if $s<0$, $\mathcal{L}=\underline{L}$ if $s>0$ and $\mathcal{L}$ being the identity if $s=0$. Then, for $k=0,...,|s|$, $\swei{\uppsi}{s}_{(k)}$ satisfy
\begin{equation}
\begin{gathered}
\frac12 \lp(L\underline{L}+\underline{L}L\rp)\swei{\uppsi}{s}_{(k)}+w\lp[\mathring{\slashed{\triangle}}^{[s]}_T+|s|+k(2|s|-k-1)\rp]\swei{\uppsi}{s}_{(k)}+\mr{sign}(s)(k-|s|)\frac{w'}{w}\mc{L}\swei{\uppsi}{s}_{(k)}\\
-\frac{2arw}{r^2+a^2}\sign{s}\lp(2|s|-2k+1\rp)\Phi\swei{\uppsi}{s}_{(k)}
+ \frac{a^2\Delta^2}{(r^2+a^2)^4}\lp[1-2|s|-2k(2|s|-k-1)\rp]\swei{\uppsi}{s}_{(k)}\\
+\frac{2Mr(r^2-a^2)\Delta}{(r^2+a^2)^4}\lp[1-3|s|+2s^2-3k(2|s|-k-1)\rp]\swei{\uppsi}{s}_{(k)}\\
=\lp(\frac{\mc{L}}{w}\rp)^{k}\lp(\frac{\Delta^{1+|s|}}{(r^2+a^2)^{2|s|+3/2}}\swei{F}{s}\rp)+aw\sum_{i=0}^{k-1}\lp(c_{s,\,k,\,i}^{\Phi}(r)\Phi+ c_{s,\,k,\,i}^{\mr{id}}(r)\rp)\swei{\uppsi}{s}_{(i)}\,,
\end{gathered} \label{eq:transformed-k}
\end{equation}
where $a c_{s,\,k,\,i}^\mr{id}$ should be replaced by $c_{s,\,k,\,i}^\mr{id}$ if and only if $|s|\neq 1$, $k=1$ and $i=0$. Here, $c_{s,\,k,\,i}^\Phi$ and $c_{s,\,k,\,i}^\mr{id}$, $i=0,...,|s|$ are functions of $r$ which can be explicitly computed for each $s,k,i$ through a recursive relation \eqref{eq:recursive-relation-coupling-coefs-final} initialized by \eqref{eq:recursive-relation-coupling-coefs-i=k-final}; the functions are, at most, of $O(1)$ as $r^*\to\pm\infty$, their derivatives with respect to $r^*$ are, at most, of $O(w)$ as $r^*\to\pm\infty$, and $\lp(c_{s,\,1,\,0}^\mr{id}\rp)'=a\times O(w)$ as $r^*\to\pm\infty$.

In particular, $\swei{\Psi}{s}:=\swei{\uppsi}{s}_{(|s|)}$ solves the transformed equation
\begin{equation}
\swei{\mathfrak{R}}{s}\swei{\Psi}{s}=\swei{\mathfrak{J}}{s}+\swei{\mathfrak{G}}{s}\,, \label{eq:transformed-equation}
\end{equation}
where
\begin{align}
\swei{\mathfrak{R}}{s}&:=\frac12 \lp(L\underline{L}+\underline{L}L\rp) + w\lp(\mathring{\slashed{\triangle}}^{[s]}_T+ s^2\rp) +\frac{\Delta\lp[(1-2s^2)a^2\Delta + 2(1-s^2)Mr(r^2-a^2)\rp]}{(r^2+a^2)^4}\,, \label{eq:transformed-R}\\
\swei{\mathfrak{J}}{s}&:= \sum_{k=0}^{|s|-1} aw \lp[ c^{\Phi}_{s,\,|s|,\,k}(r) \Phi+c^\mr{id}_{s,\,|s|,\,k}(r)\rp]\swei{\uppsi}{s}_{(k)}\,,  \label{eq:transformed-J}\\
\swei{\mathfrak{G}}{s}&:=\lp(\frac{\mc{L}}{w}\rp)^{|s|}\lp(\frac{\Delta^{1+|s|}}{(r^2+a^2)^{2|s|+3/2}}\swei{F}{s}\rp) \label{eq:transformed-G}\,.
\end{align}
\end{proposition}
\begin{proof} With our choice of weights, the PDEs (\ref{eq:pde-Psi-i}) for the transformed quantities take the form 
 \begin{gather*}
\mc{L}\underline{\mc{L}}\swei{\uppsi}{s}_{(k)} +\mr{sgn}(s)(k-|s|)\frac{w'}{w}\mc{L}\swei{\uppsi}{s}_{(k)}+\sum_{X\in\mathfrak{X}}  \sum_{i=0}^k c_{s,\,k,\,i}^X \swei{\uppsi}{s}_{(i)}=\lp(\frac{\mc{L}}{w}\rp)^{k}\lp(\frac{\Delta^{1+|s|}}{(r^2+a^2)^{2|s|+3/2}}\swei{F}{s}\rp)\,, 
\end{gather*}
where $\mathfrak{X}=\{T,\,\Phi,\,\mathring{\slashed{\triangle}}^{[s]}_T,\, \mr{id}\}$ and coefficients $c^X_{s,\,k,\,i}$ satisfy the recursive relation \eqref{eq:recursive-relation-coupling-coefs}
\begin{align}
\begin{dcases}
c_{s,\,k,\,0}^X&=-\mr{sgn}(s)\lp(\frac{c_{s,\,k-1,\,0}^X}{w}\rp)' \\
c_{s,\,k,\,i}^X&=-\mr{sgn}(s)\lp(\frac{c_{s,\,k-1,\,i}^X}{w}\rp)'+c_{s,\,k-1,\,i-1}^X \quad \text{~ for ~} i=1,...,k-1 
\end{dcases}\,, \label{eq:recursive-relation-coupling-coefs-final}
\end{align} 
initialized by the relations, for $i=k=0,...,|s|$, (see \eqref{eq:recursive-relation-coupling-coefs-i=k})
\begin{align*}
\begin{split}
c_{s,\,k,\,k}^T&=0\,, \qquad c_{s,\,k,\,k}^{\mathring{\slashed{\triangle}}}= w\,,\qquad
c_{s,\,k,\,k}^\Phi = \frac{[-4s+(4k-2)\mr{sgn}(s)]arw}{r^2+a^2}\,, \\
c_{s,\,k,\,k}^{\mr{id}} &= w\lp(|s|+\frac{a^2\Delta (1-2|s|) +2Mr(r^2-a^2)(1-3|s|+2s^3)}{(r^2+a^2)^2}\rp)+\frac{k(2|s|-k-1)}{2}\lp(\frac{w'}{w}\rp)'\\
&=w\lp[|s|+k(2|s|-k-1)\rp]+\frac{a^2\Delta w}{(r^2+a^2)^2} [1-2|s|-2k(2|s|-k-1)]\\
&\qquad +\frac{2Mr(r^2-a^2)w}{(r^2+a^2)^2}[1-3|s|+2s^2-3k(2|s|-k-1)]\,. 
\end{split} \label{eq:recursive-relation-coupling-coefs-i=k-final}
\end{align*}

The recursive relations for each $c_{s,\,k,\,i}^X$ can be used to reduce these coefficients, for any $i=0,...,k-1$, to a linear combination of derivatives, weighted by $w^{-1}$, of $c_{s,\,i,\,i}^X$  along $r^*$, for $i=0,...,k-1$. We would like to show that, for $i\leq k-1$,
\begin{align*}
c_{s,\,k,\,i}^X= a w \tilde{c}_{s,\,k,\,i}^X 
\end{align*}
suppressing the $a$ when $(k,i)=(1,0)$ and $|s|\neq 1$, where $\tilde{c}_{s,\,k,\,i}^X$ is at most of $O(1)$  and $(d/dr^*)^n\tilde{c}_{s,\,k,\,i}^X=O(w)$ as $r^*\to\pm \infty$. For $X=T,\swei{\mathring{\slashed{\triangle}}}{s}_T$, this is clear, since
\begin{gather*}
\lp(\frac{c_{s,\,k,\,k}^T}{w}\rp)'=\lp(\frac{c_{s,\,k,\,k}^{\mathring{\slashed{\triangle}}}}{w}\rp)'=0\,.
\end{gather*}
For $X=\Phi$, the result follows from the simple computation
\begin{gather*}
\underbrace{\lp(\frac{1}{w}\cdots\lp(\frac{c_{s,\,k,\,k}^\Phi}{w}\rp)'\cdots\rp)'}_\text{$n$ weighted derivatives} = -(-4)^{\lfloor n/2\rfloor}[-4s+(4k-2)\mr{sgn}(s)]a^nw\begin{dcases}
 \frac{r^2-a^2}{r^2+a^2} &\text{~ if $n$ odd}\\
 \frac{a r }{r^2+a^2} &\text{~ if $n$ even}
\end{dcases}\,.
\end{gather*}
For $X=\mr{id}$, we compute
\begin{align*}
\frac{1}{w}\underbrace{\lp(\frac{1}{w}\cdots\lp(\frac{w'}{w}\rp)'\cdots\rp)'}_\text{$n$ weighted derivatives}&= O(1) \text{~as~}r^*\to\pm \infty\,, \\
\frac{1}{w}\underbrace{\lp(\frac{1}{w}\cdots\lp(\frac{r(r^2-a^2)}{(r^2+a^2)^2}\rp)'\cdots\rp)'}_\text{$n$ weighted derivatives} &= \frac{1}{(r^2+a^2)^2}
\begin{dcases}
(16a^2)^{(n-1)/2}(r^4-6a^2r^2+a^4) &\text{~if $n$ odd} \\
(16a^2)^{n/2}r(r^2-a^2) &\text{~if $n$ even}
\end{dcases}\,.
\end{align*}
By the computations above and the recursive relations, we see that the only possible obstruction to the claim is for $c_{s,\,1\, 0}^\mr{id}$, when $n=1$ and so $(16a^2)^{(n-1)/2} =1$. For $s=\pm 1$, this is not an issue, however:
$$c_{\pm 1,\,1,\, 0}^\mr{id}=\lp(\frac{c_{\pm 1,\, 0,\, 0}^\mr{id}}{w}\rp)'=-\frac{a^2 w'}{w} \,,$$
since  $[1-3|s|+2s^2-3k(2|s|-k-1)]=0$ for $(s,k)=(\pm 1,0)$.
\end{proof}

\begin{remark}[Computation of the auxiliary functions]
In Proposition~\ref{prop:transformed-system}, we have written the equations for the transformed system \eqref{eq:transformed-k} and \eqref{eq:transformed-equation} in terms of some functions 
$c_{s,\,k,\,i}^\Phi$ and $c_{s,\,k,\,i}^\mr{id}$,  $i<k=0,...,|s|$, that can be explicitly computed: for $i=k$, these are given by \eqref{eq:recursive-relation-coupling-coefs-i=k-final}, whereas for $i=0,...,k-1$ they are obtained from the recursive relation \eqref{eq:recursive-relation-coupling-coefs-final}. We compute the resulting functions in the physically interesting cases $|s|=1,2$. For $|s|=1$,
\begin{align*}
c_{\pm 1,\,1,\,0}^\mr{id}=\frac{2a(r^3-3Mr^2+a^2r+a^2M)}{(r^2+a^2)^2}\,, \qquad
c_{\pm 1,\,1,\,0}^\Phi= \sign s\frac{r^2-a^2}{r^2+a^2}\,.
\end{align*}
For $|s|=2$ and $k=2$,
\begin{alignat*}{3}
&c_{\pm 2,\,2,\,1}^\mr{id}=\frac{20a\sign s(r^3-3Mr^2+a^2r+a^2M)}{(r^2+a^2)^2}\,, \qquad
&&c_{\pm 2,\,2,\,1}^\Phi= 8\frac{r^2-a^2}{r^2+a^2}\,,\\
&c_{\pm 2,\,2,\,0}^\mr{id}=-\frac{6a(r^4+10Mr^3-6a^2Mr-a^4)}{(r^2+a^2)^2}\,, \qquad
&&c_{\pm 2,\,2,\,0}^\Phi= -24\sign s\frac{a^2r}{r^2+a^2}\,;
\end{alignat*}
finally, for $k=1$,
\begin{align*}
c_{\pm 2,\,1,\,0}^\mr{id}=\frac{6r\sign s (Mr^3-a^2r^2-3Mr a^2-a^4)}{(r^2+a^2)^2}\,, \qquad
c_{\pm 2,\,1,\,0}^\Phi=-10\frac{r^2-a^2}{r^2+a^2}\,,
\end{align*}
where we recall that, in \eqref{eq:transformed-k} with $(s,k,i)=(\pm 2,1,0)$, $a c_{\pm 2,\,1,\,0}^\mr{id}$ should be replaced by $c_{\pm 2,\,1,\,0}^\mr{id}$.
\end{remark}

It is clear that decoupling into a wave-type PDE with real potential (in the separated picture) is ensured through this strategy if the $\Phi$ derivative and the zeroth order term for levels 0 to $|s|-1$ of $\Psi$ are made to vanish under these simplifications, at least. By explicity computing these weights, we find that this is the case if and only if $a=0$:

\begin{lemma} \label{prop:decoupling-a} Fix $s\in\mathbb{Z}$. In Proposition~\ref{prop:transformed-system}, we have $\swei{\mathfrak{J}}{s}=0$  if and only if $a=0$. 
\end{lemma}
\begin{proof}
From the form of the coupling in \eqref{eq:transformed-J}, it is clear that $a=0$ provides a decoupling; we now want to show that there is no possible decoupling if $a\neq 0$ via this strategy. Note that, in the particular cases $|s|=1,2$, the result already follows from Lemma~\ref{lemma:no-decoupling-general-Kerr}.

Consider the coupling coefficient of the second highest order $\Psi$. From \eqref{eq:recursive-relation-coupling-coefs-final} and \eqref{eq:recursive-relation-coupling-coefs-i=k-final}, we have
\begin{align*}
c^\mr{id}_{s,\,|s|,\,|s|-1} = -\mr{sgn}(s) \sum_{n=0}^{|s|-1}\lp(\frac{c^\mr{id}_{s,\,n,\,n}}{w}\rp)' = -\frac{2 a^2 \left(4 s^2-1\right) s  \left(a^2 (M+r)+r^2 (r-3 M)\right)}{3 \left(a^2+r^2\right)^2}\,,
\end{align*}
since $\sum_{n=0}^{|s|-1} [1-3|s|-3k(2|s|-k-1)]=0$. This term is obviously not identically zero in $(r_+,\infty)$ unless $a=0$, which gives the negative result for rotating Kerr black holes.
\end{proof}


\section{The Teukolsky and transformed system ODEs}

In this section, we introduce the ordinary differential equations arising from a formal separation of variables of the Teukolsky equation~\eqref{eq:teukolsky-alpha} and of the transformed system introduced in Proposition~\ref{prop:transformed-system}.

\subsection{Admissible frequencies}
\label{sec:admissible-freq}

For $\omega\in\mathbb{R}$ and $m\in\frac12\mathbb{Z}$, it will be convenient to define:
\begin{equation}
\xi := -i\frac{2M r_+}{r_+-r_-}(\omega-m\upomega_+) \,,\quad \beta:=2iM^2(\omega-m\upomega_+)\,, \quad \upomega_+:=\frac{a}{2Mr_+}\,. \label{eq:xi-upomega+}
\end{equation}

For the remainder of this paper, we will be interested in the following parameters:
\begin{definition}[Admissible frequencies] \label{def:admissible-freqs} Fix $s\in\frac12\mathbb{Z}$ and $M>0$. 
\begin{enumerate}[noitemsep]
\item We say the frequency parameter $m$ is admissible with respect to $s$ when, if $s$ is an integer, $m$ is also an integer and when, if $s$ is a half-integer, so is $m$.
\item We say the frequency pair $(m,l)$ is admissible with respect to $s$ when $m$ is admissible with respect to $s$, $l$ is an integer or half-integer if $s$ is an integer or half-integer, respectively, and $l\geq \max\{|m|,|s|\}$. 
\item We say the frequency triple $(\omega,m,l)$ is admissible with respect to $s$ when the pair $(m,l)$ is admissible with respect to $s$ and $\omega\in\mathbb{R}$.
\item We say the frequency triple $(\omega,m,\Lambda)$ is admissible with respect to $s$ when $m$ is admissible with respect to $s$ and $\Lambda\in\mathbb{R}$ satisfies
\begin{enumerate}[label = (\alph*), ref=\theenumi{}(\alph*),itemsep=0pt]
\item if $s=0$, $\Lambda \geq 2|am\omega|$ and  $\Lambda \geq m^2$; \label{it:admissible-freqs-triple-wave}
\item $\Lambda \geq \max\{|m|,|s|\}(\max\{|m|,|s|\}+1)-s^2-2|s||a\omega|$; \label{it:admissible-freqs-triple-Lambda-lower-bound}
\item $|\Lambda| \leq \Lambda +s^2+a^2\omega^2+4|s||a\omega|$; \label{it:admissible-freqs-triple-Lambda-upper-bound}
\item assuming $q:=a\omega/m$, if $q<0$ then, for sufficiently large $|m|$, 
$\Lambda -2m\nu \geq \frac12 |q|m^2$;
if $b<q<1$ then, for sufficiently large $|m|$ depending on $b$, $\Lambda-2m\nu \geq  (q-1)m^2$. \label{it:admissible-freqs-triple-Lambda-q-bound}
\item $\mathfrak{C}_s\geq 0$, for $|s|\leq 2$, for $\mathfrak{C}_s$ as given in \eqref{eq:TS-radial-constants}. \label{it:admissible-freqs-triple-Cs-positivity}
\end{enumerate}
\item We say the frequency triples $(\omega,m,l)$ and $(\omega,m,\Lambda)$ are admissible with respect to $s$ and $a$ when they are admissible with respect to $s$ and, moreover, $\omega$ satisfies $\omega\in\mathbb{R}\backslash\{0\}$ and, if $|a|=M$, $\omega\neq m\upomega_+$. In this case, we write   $(\omega,m,\Lambda)\in\mc{F}_{\mathrm{ admiss}(a,M,s)}$.
\end{enumerate}
\end{definition}

\subsection{The angular ODE}
\label{sec:angular-ODE}

In this section, we introduce the ODE satisfied by  the angular part of the Teukolsky variable under separation of variables and derive some of the properties of the separation constant.

\begin{proposition}[Smooth spin-weighted spheroidal harmonics] \label{prop:angular-ode}
Fix $s\in\frac12\mathbb{Z}$, let $m$ be admissible with respect to $s$, and assume $\nu\in\mathbb{R}$. Consider the angular ODE
\begin{gather}
\begin{gathered}
-\frac{d}{d\theta}\lp(\sin\theta\frac{d}{d\theta}\rp)S_{m,\bm\uplambda}^{[s],\,(\nu)}(\theta)
+ \lp(\frac{(m+s\cos\theta)^2}{\sin^2\theta}-\nu^2\cos^2\theta+2\nu s \cos\theta\rp)S_{m,\bm\uplambda}^{[s],\,(\nu)}(\theta)\\ =\bm\uplambda_{m}^{[s],\,(\nu)} S_{m,\bm\uplambda}^{[s],\,(\nu)}(\theta) \,,
\end{gathered} \label{eq:angular-ode}
\end{gather}
with the boundary condition that $e^{im\phi}S_{m,\bm\uplambda}^{[s],\,(\nu)}$ is a non-trivial smooth $s$-spin-weighted function. 

For each $\nu\in\mathbb{R}$, there are countably many solutions to the problem; using $l$ as an index, we write such solutions, also called $s$-spin-weighted spheroidal harmonics with spheroidal parameter $\nu$, as $e^{im\phi}S_{ml}^{[s],\,(\nu)}$ and denote the corresponding eigenvalues by $\bm\uplambda_{ml}^{[s],\,(\nu)}$. The parameter $l$ is chosen to be admissible with respect to $s$ and such that $\bm\uplambda_{ml}^{[s],\,(0)}=l(l+1)-s^2$ for $\nu =0$ and $\bm\uplambda_{ml}^{[s],\,(\nu)}$ varies smoothly with $\nu$. The eigenvalues satisfy
\begin{align}
\bm\uplambda^{[s],\,(\nu)}_{ml}=\bm\uplambda^{[-s],\,(\nu)}_{ml}=\bm\uplambda^{[-s],\,(-\nu)}_{-m,\,l}\,. \label{eq:lambda-symmetries}
\end{align}

The $s$-spin weighted spheroidal harmonics $\lp\{e^{im\phi}S_{ml}^{[s],\,(\nu)}\rp\}_{ml}$ form a complete orthonormal basis on the space of smooth $s$-spin-weighted spheroidal functions (see Definition \ref{def:smooth-spin-weighted}).
\end{proposition}

\begin{proof}  For $\nu\in\mathbb{R}$, the operator 
$$\mathring{\slashed{\triangle}}^{[s]}_m=-\frac{d}{d\theta}\lp(\sin\theta\frac{d}{d\theta}\rp)
+ \lp(\frac{(m+s\cos\theta)^2}{\sin^2\theta}-\nu^2\cos^2\theta+2\nu s \cos\theta\rp)$$
is self-adjoint and it follows from Sturm-Liouville theory that it has a countable set of eigenfunctions, which form a complete basis os the space of $s$-spin-weighted spheroidal functions, and countable set of corresponding eigenvalues (see, for instance, \cite{Dafermos2017}). For each $s$, $m$ and $\nu$, these can be indexed by $l\in\frac12\mathbb{Z}$ satisfying the constraints in the statement, so that, in particular, $\lambda_{ml}^{[s],\,(\nu)}$ is smooth in $\nu$ (see \cite[Section 3.22, Proposition 1]{Meixner1954}). 
\end{proof}

\begin{lemma}[Properties of the angular eigenvalues] \label{lemma:angular-eigenvalues-properties}
Let $\bm\Lambda_{ml}^{[s],\,(\nu)}$ be defined by 
\begin{align}
\bm\Lambda^{[s],\,(\nu)}_{ml}:=\bm\uplambda^{[s],\,(\nu)}_{ml}+\nu^2 \label{eq:def-Lambda}
\end{align}
where $\bm\uplambda^{[s],\,(\nu)}_{ml}$ is an angular eigenvalue identified in Proposition~\ref{prop:angular-ode}. 
\begin{enumerate}[noitemsep]
\item $\bm\Lambda^{[s],\,(\nu)}_{ml}$ depends only on $m$, $\nu$ and $|s|$.
\item Basic bounds for $s=0$: $\bm\Lambda^{[0],\,(\nu)}_{ml} \geq 2|m\nu|$ and $\bm\Lambda^{[0],\,(\nu)}_{ml} \geq m^2$.
\item Basic bounds for general $|s|$:$\bm\Lambda^{[s],\,(\nu)}_{ml} \geq l(l+1)-s^2-2|s||\nu|$ and $\lp|\bm\Lambda^{[s],\,(\nu)}_{ml}\rp| \leq \bm\Lambda^{[s],\,(\nu)}_{ml}+s^2+a^2\omega^2+4|s||\nu|$.

\item Set $q=\nu/m$. If $q<0$ then, for sufficiently large $|m|$, $\bm\Lambda^{[s],\,(\nu)}_{m l} -2m\nu \geq \frac12 |q|m^2$; if $0\leq q<1-b$ then, for sufficiently large $|m|$ depending on $b$, $\bm\Lambda^{[s],\,(\nu)}_{m l} -2m\nu \geq (q-1)m^2$.
\end{enumerate}
\end{lemma}

\begin{remark} With Lemma~\ref{lemma:angular-eigenvalues-properties}, setting $\nu=a\omega$, we can now justify statement 4 in Definition~\ref{def:admissible-freqs}, which is the concretization of statements 2 and 3 for such a choice of $\nu$. 
\end{remark}

\begin{proof}[Proof of Lemma~\ref{lemma:angular-eigenvalues-properties}]
First note that statement 1 follows from \eqref{eq:lambda-symmetries}.

For statements 2 and 3, let $\swei{\Xi}{s}$ be an $s$-spin-weighted spheroidal harmonic normalized to have unit $L^2$ norm (see Proposition~\ref{prop:angular-ode}). Then, by construction, $\bm\Lambda^{[s],\,(\nu)}_{ml}$, defined in \eqref{eq:def-Lambda}, satisfies
\begin{gather}
\begin{gathered}
\bm\Lambda^{[s],\,(\nu)}_{ml}=\int_0^\pi\int_0^{2\pi} \lp(\lp|\p_\theta\swei{\Xi}{s}\rp|^2+V_{\mr{ang}}(\theta)\lp|\swei{\Xi}{s}\rp|^2\rp)\sin\theta d\theta d\phi\,,
\end{gathered}\label{eq:angular-variational-problem}
\end{gather}
where, setting $q=\nu/m$
\begin{equation}
\begin{split}
V_{\mr{ang}}(\theta)&=\frac{(m+s\cos\theta)^2}{\sin^2\theta} +\nu^2\sin^2\theta +2\nu s \cos\theta \\
&= \frac{m^2}{\sin^2\theta}(1-q\sin^2\theta)^2+\frac{2sm\cos\theta}{\sin^2\theta}(1+q\sin^2\theta) +\frac{s^2\cos^2\theta}{\sin^2\theta}+2m\nu\,.
\end{split}\label{eq:angular-potential}
\end{equation}

In the case $\nu=0$ (see Proposition~\ref{prop:angular-ode}), 
\begin{align}\label{eq:eigenvalues-schwarzschild}
\bm\Lambda^{[s],\,(\nu)}_{ml}=l(l+1)-s^2\,.
\end{align}
Hence, as $V_{\mr{ang}}=V_{\mr{ang}}|_{\nu=0}+ \nu^2\sin^2\theta+2s\nu \cos\theta$, we can immediately obtain the first bound in statement 3 as well as the second, since
\begin{align*}
\lp|\bm\Lambda^{[s],\,(\nu)}_{ml}\rp|\leq l(l+1)+\nu^2+2|s||\nu|\leq \bm\Lambda^{[s],\,(\nu)}_{ml}+s^2+\nu^2+4|s||\nu|,. 
\end{align*}

From the first bound in statement 3, the second bound from statement 2 follows. The first bound in statement 2, follows from the fact that, if $s=0$, Young's inequality gives
\begin{equation*}
V_{\mr{ang}}(\theta)=\frac{m^2}{\sin^2\theta} +\nu^2\sin^2\theta \geq 2m\nu \,.
\end{equation*} 
Finally, we turn to statement 4. If $q<0$, then
\begin{align*}
V_{\mr{ang}}(\theta)-2m\nu \geq \frac{1}{\sin^2	\theta}(1+|q|\sin^2\theta)\max\{1,|q|\sin^2\theta\}(m^2-2|s||m|)\geq \frac12 |q|m^2\,,
\end{align*}
assuming $|m|\geq 4|s|$. Hence, assume $q>0$ and, without loss of generality, $m>0$. In case $q<1$, we also easily have 
\begin{align*}
V_{\mr{ang}}(\theta)-2m\nu \geq \frac{m^2}{\sin^2\theta}+2q(1-q)\frac{\cos^2\theta}{\sin^2\theta}m^2-\frac{4|s||m|\cos\theta}{\sin^2\theta} \geq m^2(q-1)^2\,,
\end{align*}
as long  $q-1>b$ and $m$ is sufficiently large depending on $b$. 
\end{proof}

\subsection{The radial ODEs}
\label{sec:radial-ODEs}

In this section, we introduce the radial ODEs whose analysis is the main goal of the present paper. The ODEs are given in Section \ref{sec:radial-ODE-Teukolsky-transformed}, paving the way to examining their asymptotics in Section~\ref{sec:radial-ODE-asymptotics} and establishing suitable boundary conditions and notation for our applications in Sections \ref{sec:radial-ODE-bdry-conditions-Teukoslsky} and \ref{sec:radial-ODE-bdry-conditions-transformed}. We conclude by recalling the Teukolsky--Starobinsky identities in the separated picture in Section~\ref{sec:Teukolsky--Starobinsky}.

\subsubsection{The Teukolsky radial ODE and the transformed radial ODEs}
\label{sec:radial-ODE-Teukolsky-transformed}

For some spin $s\in\mathbb{Z}$ and $a\in[-M,M]$, let $(\omega,m,\Lambda)$ be an admissible frequency triple with respect to $a$ and $s$. We are interested in solutions $\smlambda{\upalpha}{s}(r)$ of the inhomogeneous Teukolsky radial ODE
\begin{align*}
&\lp[\Delta^{-s}\frac{d}{dr}\lp(\Delta^{s+1}\frac{d}{dr}\rp)  +\frac{[\omega(r^2+a^2)-am]^2-2is(r-M)[\omega(r^2+a^2)-am]}{\Delta}\rp] \smlambda{\upalpha}{s}(r)\\
&\qquad\quad+\lp(4is\omega r -\Lambda-s+2am\omega\rp)\smlambda{\upalpha}{s}(r)=\frac{\Delta^{|s|/2(1-\sign s)}}{(r^2+a^2)^{|s|}}\smlambda{F}{s}(r)\,,\numberthis \label{eq:radial-ODE-alpha}
\end{align*}
where $\smlambda{F}{s}(r)$ is bounded for $r\in [r_+,\infty)$. Defining 
\begin{equation}
\smlambda{u}{s}:=\Delta^{s/2}(r^2+a^2)^{1/2}\smlambda{\upalpha}{s}\,, \label{eq:def-u}
\end{equation}
 the Teukolsky radial ODE \eqref{eq:radial-ODE-alpha} can also be rewritten as
\begin{gather} \label{eq:radial-ODE-u}
\lp(\smlambda{u}{s}\rp)''  + \lp(\omega^2- V^{[s]}\rp) \smlambda{u}{s} =\smlambda{H}{s}\,,\qquad  \sml{H}{s}:=\frac{\Delta^{|s|/2+1}}{(r^2+a^2)^{|s|+3/2}} \smlambda{F}{s}\,,
\end{gather}
letting $'$ denote $r^*$ derivatives. Here, $\swei{V}{s}=\swei{V}{s}_0+V_1+i\swei{V}{s}_2$, where
\begin{equation}
\begin{split}
\swei{V}{s}_0&=\frac{4Mra m\omega -a^2m^2+\Delta\Lambda_{ml}^{[s]}}{(r^2+a^2)^2}\,, \\
\swei{V}{s}_1&=s^2\frac{(r-M)^2}{(r^2+a^2)^2}+\frac{\Delta}{(r^2+a^2)^4}\lp(a^2\Delta+2Mr(r^2-a^2)\rp) \geq 0 \,,  \\
\swei{V}{s}_2&=-\frac{2s\omega \lp(r(r^2+a^2)+M(a^2-3r^2)\rp)+2sam(r-M)}{(r^2+a^2)^2}\,.
\end{split} \label{eq:radial-ODE-u-potentials} 
\end{equation}

Moreover, define
\begin{align}
\smlambda{\lp(\uppsi_{(0)}\rp)}{s}&:=(r^2+a^2)^{-|s|+1/2}\Delta^{|s|/2(1+\sign s)}\smlambda{\alpha}{s}=w^{-|s|/2}\sml{u}{s}\,,\label{eq:def-psi0-separated}\\
\smlambda{\lp(\uppsi_{(k)}\rp)}{s}&:=\frac{1}{w}\mc{L}\smlambda{\lp(\uppsi_{(k-1)}\rp)}{s}\,,\qquad  k=1,...,|s|\,, \qquad \smlambda{\Psi}{s}:=\smlambda{\lp(\uppsi_{(k)}\rp)}{s}\,, \label{eq:transformed-transport-separated}
\end{align}
with $\mathcal{L}=L$ if $s<0$, $\mathcal{L}=\underline{L}$ if $s>0$ and $\mathcal{L}$ being the identity if $s=0$; here and throughout the section we are considering, by abuse of notation, $L$ and $\uL$ to be given by their separated picture analogues,
\begin{align*}
L=\frac{d}{dr^*}-i\omega+\frac{iam}{r^2+a^2}\,,\quad \uL=-\frac{d}{dr^*}-i\omega+\frac{iam}{r^2+a^2}\,.
\end{align*}

From the Teukolsky radial ODE \eqref{eq:radial-ODE-alpha}, we derive the radial ODEs for  $\smlambda{\lp(\uppsi_{(k)}\rp)}{s}$ when $k=0,...,|s|$:
\begin{equation}
\begin{gathered}
\lp(\smlambda{\lp(\uppsi_{(k)}\rp)}{s}\rp)''-(|s|-k)\lp(\frac{w'}{w}\rp)\lp(\smlambda{\lp(\uppsi_{(k)}\rp)}{s}\rp)' + \lp(\omega^2-\smlambda{\lp({\mc{V}}_{(k)}\rp)}{s}\rp)\smlambda{\lp(\uppsi_{(k)}\rp)}{s}\\
= \smlambda{\lp(\mathfrak{G}_{(k)}\rp)}{s}+aw\sum_{i=0}^{k-1}\lp(im c_{s,\,k,\,i}^{\Phi}+ c_{s,\,k,\,i}^{\mr{id}}\rp)\smlambda{\lp(\uppsi_{(i)}\rp)}{s}\,,
\end{gathered} \label{eq:transformed-k-separated}
\end{equation}
where  $c_{s,\,k,\,i}^{\Phi},c_{s,\,k,\,i}^{\mr{id}}=O(1)$ as $r^*\to \pm \infty$ are the functions introduced in Proposition~\ref{prop:transformed-system} (recall $ac_{s,\,1,\,0}^{\mr{id}}$ in \eqref{eq:transformed-k-separated} should be replaced by $c_{s,\,1,\,0}^{\mr{id}}$ when $|s|\neq 1$) and
\begin{equation}
\smlambda{\lp(\mathfrak{G}_{(k)}\rp)}{s}=\lp(\frac{\mc{L}}{w}\rp)^{k}\lp(\frac{\Delta^{1+|s|}}{(r^2+a^2)^{2|s|+3/2}}\smlambda{F}{s} \rp)\,.
\end{equation}
The potential $\smlambda{\lp({\mc{V}}_{(k)}\rp)}{s}$ is a complex-valued function of $r$ involving the frequency parameters, where
\begin{align*}
\Re{\mc{V}}_{(k)}&=\frac{\Delta\Lambda+4Mram\omega-a^2m^2}{(r^2+a^2)^2}+w\lp[|s|+k(2|s|-k-1)\rp]+\frac{a^2\Delta w}{(r^2+a^2)^2} [1-2|s|-2k(2|s|-k-1)]\\
&\qquad +\frac{2Mr(r^2-a^2)w}{(r^2+a^2)^2}[1-3|s|+2s^2-3k(2|s|-k-1)]\,,\numberthis \label{eq:radial-ODE-potential-k-tilde}\\
\Im{\mc{V}}_{(k)}&=\sign s (|s|-k)\lp[\frac{w'}{w}\lp(\omega-\frac{am}{r^2+a^2}\rp)-\frac{4amrw}{(r^2+a^2)}\rp]\,. 
\end{align*}

In particular, $\smlambda{\Psi}{s}:=\smlambda{\lp(\uppsi_{(|s|)}\rp)}{s}$ solves the transformed radial ODE
\begin{align}
\lp(\smlambda{\Psi}{s}\rp)''+\lp(\omega^2-\smlambda{\mc{V}}{s}\rp)\smlambda{\Psi}{s}=\smlambda{\mathfrak{G}}{s}+aw\sum_{k=0}^{|s|-1}\lp(im c_{s,\,|s|,\,k}^{\Phi}+ c_{s,\,|s|,\,k}^{\mr{id}}\rp)\smlambda{\lp(\uppsi_{(k)}\rp)}{s}\,, \label{eq:radial-ODE-Psi}
\end{align}
where $\smlambda{\mc{V}}{s}=\mc{V}^{s}_{0}+\swei{\mc{V}}{s}_{1}$ is real and given by
\begin{equation}
\begin{split}
\swei{\mc{V}}{s}_{0}&:=\frac{4Mram\omega-a^2m^2+\Delta\Lambda}{(r^2+a^2)^2} \\
\swei{\mc{V}}{s}_{1}&:=s^2\frac{\Delta}{(r^2+a^2)^2}+\frac{\Delta}{(r^2+a^2)^4}\lp[(1-2s^2)a^2\Delta + 2(1-s^2)Mr(r^2-a^2)\rp] \,.
\end{split} \label{eq:radial-ODE-Psi-potentials}
\end{equation}

\subsubsection{Asymptotic analysis}
\label{sec:radial-ODE-asymptotics}

Consider a second order ODE. Around any point where it is defined, we can span the space of solutions by two linearly independent asymptotic (to arbitrary order) solutions; these asymptotic solutions, whether around a regular or singular point of the ODE, are standard (we refer the reader to \cite[Chapters 5 and 7]{Olver1973} or \cite{Erdelyi1956,Ince1956} for more detail). In this section, we apply this theory to derive information regarding the asymptotic behavior of general solutions of  the Teukolsky and transformed radial ODEs, respectively \eqref{eq:radial-ODE-alpha} and \eqref{eq:transformed-k-separated}, as $r\to r_+,\infty$. First, some notation:
\begin{lemma} \label{lemma:asymp-f-g}
Fix real $(a,M)$ such that $|a|\leq M$, and some real $s$, $\omega$, $m$, $\Lambda$. Define $f,g\colon (r_+,\infty)\to \mathbb{C}$ by  
\begin{align*}
f(r)&=\frac{2(s+1)(r-M)}{\Delta}\,, \\
 g(r)&= \frac{[\omega(r^2+a^2)-am]^2-2is(r-M)[\omega(r^2+a^2)-am]}{\Delta^2}+\frac{4is\omega r -\Lambda-s+2am\omega}{\Delta}\,.
\end{align*}
The functions $f$ and $g$ admit asymptotic expansions as $r\to r_+$ and $r\to \infty$. Indeed, for some complex $f_k^{\infty}$ and $g_k^{\infty}$, one has 
\begin{align}
f(r)&=\sum_{k=0}^{2|s|}f_k^\infty r^{-k}+O(r^{-2|s|-1})\,,  \quad g(r)=\sum_{k=0}^{2|s|}g_k^\infty r^{-k}+O(r^{-2|s|-1})\,,  \label{eq:asymp-f-g-infinity}
\end{align}
as $r\to \infty$ and, as $r\to r_+$, either
\begin{align}
\begin{split}
f(r)(r-M)^2&=\sum_{k=0}^{2|s|}f_k^{M}(r-M)^{k}+O\lp((r-M)^{2|s|+1}\rp)\,,   \text{~~as~$r\to M$ if $|a|=M$}\\
 g(r)(r-M)^4&=\sum_{k=0}^{2|s|}g_k^{M} (r-M)^{k}+O\lp((r-M)^{2|s|+1}\rp)\,,   \text{~~as~$r\to M$ if $|a|=M$}
 \end{split}\label{eq:asymp-f-g-M}
\end{align}
for some $f_k^{M},g_k^{M}\in\mathbb{C}$, if $|a|=M$ or, if $|a|<M$,
\begin{align}
\begin{split}\label{eq:asymp-f-g-r+}
f(r)(r-r_+)&=\sum_{k=0}^{2|s|}f_k^{r_+}(r-r_+)^{k}+O\lp((r-r_+)^{2|s|+1}\rp)\,,  \text{~~as~$r\to r_+$ if $|a|<M$}\\
 g(r)(r-r_+)^2&=\sum_{k=0}^{2|s|}g_k^{r_+} (r-r_+)^{k}+O\lp((r-r_+)^{2|s|+1}\rp)\,, \text{~~as~$r\to r_+$ if $|a|<M$}\,,
 \end{split}
\end{align}
for some $f_k^{r_+},g_k^{r_+}\in\mathbb{C}$. The coefficients in the expansions can be that can be explicity computed in terms of $s$, $a$, $M$, $\omega$, $m$ and $\Lambda$: for instance, one has for $f$
\begin{alignat*}{4}
f_0^{\infty}&=0\,, \quad ~~~~~f_1^\infty &&=2(1+s)\,, \quad f_2^\infty&&=2M(1+s)\,,\\
f_0^{M}&=0\,, \quad ~~~~~f_1^M &&=-2s\,, \quad ~~~~~f_2^M&&=0\,,\\
f_0^{r_+}&=1+s\,, \quad f_1^{r_+} &&=\frac{1+s}{r_+-r_-}\,, \quad f_2^{r_+}&&=-\frac{1+s}{(r_+-r_-)^2}\,,
\end{alignat*}
and, for $g$,
\begin{align*}
&\begin{alignedat}{4}
g_0^{\infty}&=-(i\omega)^2\,, \quad g_1^\infty &&=2i\omega(s-2iM\omega)\,, \quad ~~~~~~g_2^\infty&&=-\Lambda+s+2iM\omega(s-6iM\omega)\,,\\
g_0^{M}&=-\beta^2\,, \quad ~~~g_1^M &&=-2\beta(s+2iM\omega)\,, \quad ~~~~~g_2^M&&=-\Lambda+s+8M^2\omega^2\,,
\end{alignedat}\\
&g_0^{r_+}=\xi(s-\xi)\,, \quad g_1^{r_+} =\frac{\Lambda-s-2am\omega-2\xi^2-2i\omega r_+(s+\xi)}{r_+-r_-}\,, \\
&g_2^{r_+}=\frac{\Lambda-s-2am\omega -\xi(s+3\xi)+2i\omega[\sqrt{M^2-a^2}(s-2\xi)-2M(s+2\xi)]+4r_+^2\omega^2}{(r_+-r_-)^2}\,.
\end{align*}
\end{lemma}

Having fixed the relevant notation, we apply standard ODE theory to the Teukolsky radial ODE \eqref{eq:radial-ODE-alpha}:

\begin{lemma}[Asymptotic analysis of the Teukolsky radial ODE] \label{lemma:aymptotic-analysis-radialODE-Teukolsky}
Let $s\in\mathbb{Z}$, $a\in[-M,M]$, and $(\omega,m,\Lambda)$ be an admissible frequency triple with respect to $a$ and $\Lambda$ such that $\omega\in\mathbb{R}\backslash\{0,m\upomega_+\}$. For $k=0,...,|s|$, let $f_k^\infty$, $f_k^M$, $f_k^{r_+}$ and $g_k^\infty$, $g_k^M$, $g_k^{r_+}$ be determined by \eqref{eq:asymp-f-g-infinity}, \eqref{eq:asymp-f-g-M} and \eqref{eq:asymp-f-g-r+}. Assuming $\smlambda{F}{s}(r)$ is bounded for $r\in[r_+,\infty)$, a solution, $\smlambda{\upalpha}{s}(r)$, of the radial ODE~\eqref{eq:radial-ODE-alpha} admits an asymptotic expansion which is given by linear combinations of the two following solutions indexed by $\pm$:
\begin{itemize}
\item as $r\to \infty$, 
\begin{align*}
e^{\sigma_\pm  r}r^{\mu_+}\lp[\sum_{k=0}^{|s|}c_k^{[s],\infty,\pm} r^{-k}+O(r^{-2|s|-1})\rp]\,, 
\end{align*}
where $c_k^{[s],\infty,\pm}$ are fully determined by $c_k^{[s],\infty,\pm}$ and by the recursive relation 
\begin{align}
(f_0^\infty+2\sigma)k c_k^{[s],\infty,\pm} &=(k-\mu_\pm)(k-1-\mu_\pm)c_{k-1}^{[s],\infty,\pm} \nonumber\\
&\qquad+\sum_{j=1}^k\lp(\sigma_\pm f_{j+1}^\infty+g_{j+1}^\infty-(k-j-\mu_\pm)f_j^\infty\rp)c_{k-j}^{[s],\infty,\pm}\,, \label{eq:recursion-irregular-singularity}
\end{align}
and where $\sigma_\pm =\pm i\omega$, $\mu_+=2iM\omega -1-2s$ and $\mu_-=-2iM\omega -1$;
\item as $r\to M$, if $|a|=M$,
\begin{align*}
&e^{\sigma_\pm(r-M)^{-1}}(r-M)^{-\mu_\pm}\lp[\sum_{k=0}^{|s|}c_k^{[s],M,\pm} (r-M)^{k}+O\lp((r-M)^{2|s|+1}\rp)\rp]\,, 
\end{align*}
where $c_k^{[s],M,\pm}$ are fully determined by $c_0^{[s],M,\pm}$ and by the recursive relation  \eqref{eq:recursion-irregular-singularity} with superscript $M$ instead of $\infty$ and with $\sigma_\pm=\pm 2iM^2(\omega-m\upomega_+)$, $\mu_+=2iM\omega+2s$ and $\mu_-=-2iM\omega$;
\item as $r\to r_+$,
\begin{align*}
(r-r_+)^{\sigma_\pm-\mu_\pm}\lp[\sum_{k=0}^{|s|}c_k^{[s],r_+,\pm} (r-r_+)^{k}+O\lp((r-r_+)^{|s|+1}\rp)\rp]\,,
\end{align*}
where $\sigma_\pm=\pm\xi$, $\mu_\pm = \pm s$ and $c_k^{[s],r_+,\pm}$ are fully determined by $c_0^{[s],r_+,\pm}$ and by the recursive relation 
\begin{align}
\lp[(\sigma_\pm+k)(\sigma_\pm +k-1)+f_0^{r_+}\sigma_\pm\rp]c_k^{[s],r_+,\pm}=-\sum_{j=0}^{k-1}\lp[(\sigma_\pm+j)f_{k-j}^{r_+}+g_{k-j}^{r_+}\rp]c_j^{[s],r_+,\pm}\, \label{eq:recursion-regular-singularity}
\end{align}
\end{itemize}

If $\smlambda{F}{s}\equiv 0$, the expansions above are valid up to arbitrary order; in particular to order $2|s|$. 
\end{lemma}
\begin{proof}
In the homogeneous case of \eqref{eq:radial-ODE-alpha}, the ODE has a singularity at $r=r_\pm$, which is regular if $|a|<M$ but irregular of rank 1 if $|a|=M$, and an irregular singularity of rank 1 at $r=\infty$; the result follows by application of the theory laid out in the references at the beginning of the section. If the equation is inhomogeneous but with $\smlambda{F}{s}$ bounded for $r\in[r_+,\infty)$, then it does not affect the expansions near the singularities until order $|s|+1$ and higher.
\end{proof}

For $0\leq k\leq |s|$, the transformed variables are defined by \eqref{eq:def-psi0-separated} and \eqref{eq:transformed-transport-separated}, i.e.\ by differentiating $r$-weighted solutions of \eqref{eq:radial-ODE-alpha}. As a corollary of Lemma~\ref{lemma:aymptotic-analysis-radialODE-Teukolsky},  we obtain their asymptotic series expansions:

\begin{lemma}[Asymptotic analysis of the transformed radial ODEs] \label{lemma:aymptotic-analysis-radialODE-transformed}
Let $s\in\mathbb{Z}$, $k\in\{0,..,|s|\}$, $a\in[-M,M]$, and $(\omega,m,\Lambda)$ be an admissible frequency triple with respect to $a$ and $\Lambda$ such that $\omega\in\mathbb{R}\backslash\{0,m\upomega_+\}$. Assuming $\smlambda{F}{s}(r)$ is bounded for $r\in[r_+,\infty)$, a function, $\smlambda{\lp(\uppsi_{(k)}\rp)}{s}(r)$, defined inductively by the relation~\eqref{eq:transformed-transport-separated} from a function $\smlambda{\upalpha}{s}(r)$ which solves the  radial ODE~\eqref{eq:radial-ODE-alpha}, admits an asymptotic expansion which is given by linear combinations of the two following solutions indexed by $\pm$:
\begin{itemize}
\item as $r\to \infty$, 
\begin{align*}
e^{\sigma_\pm  r}r^{\mu_{(k),\pm}}\lp[\sum_{j=0}^{|s|-k}b_{(k),j}^{[s],\infty,\pm} r^{-j}+O(r^{-|s|+k-1})\rp]\,, 
\end{align*}
where $\sigma_{\pm}=\pm i\omega$, $\mu_{(k),\pm}=\pm 2iM\omega-(|s|-k)(1\pm \sign s)$ and $b_{(k),j}^{[s],\infty,\pm}$ are fully determined by $a$, $s$, $k$, $\omega$, $m$ and the coefficients of the asymptotic expansion for $\smlambda{\upalpha}{s}$ as $r\to \infty$; 
\item as $r\to M$, if $|a|=M$,
\begin{align*}
&e^{\sigma_\pm(r-M)^{-1}}(r-M)^{-\mu_{(k),\pm}}\lp[\sum_{j=0}^{|s|-k}b_{(k),j}^{[s],M,\pm} (r-M)^{j}+O\lp((r-M)^{|s|-k+1}\rp)\rp]\,, 
\end{align*}
where  $\sigma_\pm=\pm 2iM^2(\omega-m\upomega_+)$, $\mu_{(k),\pm}=2iM\omega+(|s|-k)(1\mp \sign s)$ and $b_{(k),j}^{[s],M,\pm}$ are fully determined by $a$, $s$, $k$, $\omega$, $m$ and the coefficients of the asymptotic expansion for $\smlambda{\upalpha}{s}$ as $r\to M$ and $|a|=M$;
\item as $r\to r_+$,
\begin{align*}
(r-r_+)^{\sigma_\pm - \mu_{(k),\pm}}\lp[\sum_{j=0}^{|s|-k}b_{(k),j}^{[s],r_+,\pm} (r-r_+)^{j}+O\lp((r-r_+)^{|s|-k+1}\rp)\rp]\,,
\end{align*}
where $\sigma_\pm=\pm$, $\mu_{(k),\pm} = \frac12(|s|-k)(1\pm \sign s)$ and $b_{(k),j}^{[s],r_+,\pm}$ are fully determined by $a$, $s$, $k$, $\omega$, $m$ and the coefficients of the asymptotic expansion for $\smlambda{\upalpha}{s}$ as $r\to r_+$ and $|a|<M$.
\end{itemize}

In the previous asymptotic expansions, the coefficients can be related to those in Lemma~\ref{lemma:aymptotic-analysis-radialODE-Teukolsky}. For instance, at infinity,
\begin{alignat*}{4}
\lp|b_{(k),0}^{[s],\infty, -\sign s}\rp|^2&= \frac{\mathfrak{D}_{s,k}^\mc{I}}{(2\omega)^{2k}}\lp|c_0^{[s],\infty, -\sign s}\rp|^2\,, &&\qquad \lp|b_{(k),0}^{[s],\infty, \sign s}\rp|^2&&= (2\omega)^{2k}\lp|c_0^{[s],\infty, \sign s}\rp|^2\,; 
\end{alignat*}
at the horizon, if $|a|=M$,
\begin{align*}
\lp|b_{(k),0}^{[s],M, \sign s}\rp|^2&= (2M^2)^{1+2(k-|s|)}\frac{\mathfrak{D}_{s,k}^{\mc{H}}}{\lp[4M^2(\omega-m\upomega_+)\rp]^{2k}}\lp|c_0^{[s],M, \sign s}\rp|^2\,, \\
\lp|b_{(k),0}^{[s],M, -\sign s}\rp|^2&=(2M^2)\lp[4M^2(\omega-m\upomega_+)\rp]^{2k} \lp|c_0^{[s],M, -\sign s}\rp|^2\,;
\end{align*}
and finally, if $|a|<M$, letting $s\geq 0$ without loss of generality, at the horizon,
\begin{align*}
\lp|b_{(k),0}^{[s], r_+, +}\rp|^2&= (2Mr_+)^{1+2(k-|s|)} \frac{\mathfrak{D}_{s,k}^{\mc{H}}(r_+-r_-)^{2|s|}}{\prod_{j=1}^{k}\lp\{[4Mr_+(\omega-m\upomega_+)]^2+(s-j)^2(r_+-r_-)^2\rp\}}\lp|c_0^{[s],r_+, +}\rp|^2\,, \\
\lp|b_{(k),0}^{[-s], r_+, -}\rp|^2&= (2Mr_+)^{1+2(k-|s|)} \frac{\mathfrak{D}_{s,k}^{\mc{H}}}{\prod_{j=1}^{k}\lp\{[4Mr_+(\omega-m\upomega_+)]^2+(s-j)^2(r_+-r_-)^2\rp\}}\lp|c_0^{[-s],r_+, -}\rp|^2\,, \\
\lp|b_{(k),0}^{[\pm s],r_+, \mp}\rp|^2&= (2Mr_+)\prod_{j=0}^{k-1}\lp[\lp(\frac{4Mr_+}{r_+-r_-}\rp)^2(\omega-m\upomega_+)^2+(s-j)^2\rp]\lp|c_0^{[\pm s],r_+, \mp}\rp|^2\,,
\end{align*}
(the product denoted by $\prod$ should be replaced by the identity if $k=0$ or $s=0$), 
for some $\mathfrak{D}_{s,k}^\mc{I},\mathfrak{D}_{s,k}^\mc{H}\geq 0$ depending only on the frequency parameters. While $\mathfrak{D}_{s,0}^\mc{I}=\mathfrak{D}_{s,0}^\mc{H}=1$, the remaining coefficients are generally non-trivial: for $|s|\in\{0,1,2\}$, these are
\begin{align} 
\begin{split}
\mathfrak{D}_{s,1}^{\mc{I}}&=(\Lambda-2am\omega+|s|)^2\,,\\
\mathfrak{D}_{s,2}^{\mc{I}}&=\lp[(\Lambda-2am\omega+2)(\Lambda-2am\omega+3|s|-2)+4(2|s|-1)am\omega\rp]^2+16(|s|-1)^2(2|s|-1)^2M^2\omega^2\,;
\end{split}\label{eq:Ds-infinity}\\
\begin{split}
\mathfrak{D}_{s,1}^{\mc{H}}&=\lp(\Lambda-2am\omega+|s|+(|s|-1)(2|s|-1)\frac{r_+-M}{M}\rp)^2+\frac{a^2m^2}{M^2}(2|s|-1)^2\,,\\
\mathfrak{D}_{s,2}^{\mc{H}}&=\lp[(\Lambda-2am\omega+|s|)(\Lambda-2am\omega+3|s|-2)+\frac{4(r_+-M)}{M}(\Lambda-2am\omega+|s|)(|s|-2)(|s|-1)\rp.\\
&\qquad\qquad \lp. +4am\omega\frac{r_+-M}{M}(2|s|-1)-m^2(2|s|-1)\lp(\frac{a^2}{M^2}(2|s|-3)+\frac{2r_-(r_+-M)}{M^2}\rp)\rp.\\
&\qquad\qquad \lp. +\frac{(r_+-M)^2}{M^2}(|s|-2)(|s|-1)(2|s|-1)\lp(2|s|-1-\frac{2M}{r_+}\rp)\rp]^2\\
&\qquad+\lp[\frac{2\omega}{M}(2|s|-1)\lp(2a^2(|s|-1)-M(r_+-M)-(r_+-M)^2(2|s|-3)\rp)\rp.\\
&\qquad\qquad\lp.+\frac{am(r_+-M)}{Mr_+}(2|s|-3)(2|s|-1)\lp(2|s|-3+\frac{r_+-M}{M}(2|s|-1)\rp)\rp.\\
&\qquad\qquad\lp.+\frac{4am}{M}(|s|-1)(\Lambda-2am\omega+|s|)\rp]^2\,.
\end{split}\label{eq:Ds-horizon}
\end{align}
For convenience, when $k=|s|$, we write $\mathfrak{D}_{s}^{\mc H}=\mathfrak{D}_{s,k}^{\mc H}$, $\mathfrak{D}_{s}^{\mc I}=\mathfrak{D}_{s,k}^{\mc I}$, $\swei{A}{s}_{\mc{I}^\pm}=\swei{A}{s}_{|s|,\,\mc{I}^\pm}$ and $\swei{A}{s}_{\mc{H}^\pm}=\swei{A}{s}_{|s|,\,\mc{H}^\pm}$.
\end{lemma}

\begin{proof}
The entire statement follows by considering the action of $\lp(\frac{\mc{L}}{w}\rp)^{k}$ on the asymptotic series identified in Lemma~\ref{lemma:aymptotic-analysis-radialODE-Teukolsky}, after multiplication by a suitable weight. 

As $r\to \infty$, recalling 
$$\mc{L} = -\sign s\lp(1-\frac{2Mr}{r^2+a^2}\rp)\frac{d}{dr}-i\omega+\frac{iam}{r^2+a^2}\,,$$
we find that there is a cancellation between $-i\omega$ in the operator if the $r$ derivative hits $e^{-\sign s i\omega r}r^{-2iM \sign s}$:
\begin{align*}
\lp(-\sign s\frac{\Delta}{r^2+a^2}\frac{d}{dr}-i\omega\rp)\lp(e^{-\sign s \, i\omega r}r^{-2iM \sign s}\rp)&=\lp(i\omega+\frac{2iM\omega}{r}\rp)\lp(1-\frac{2M}{r}\rp)-i\omega+O(r^{-2})=O(r^{-2})\,.
\end{align*}
In fact, as can be seen above, upon application of $\mc{L}$, the zeroth and first order terms in the series cancel out in these cases, which is compensated by multiplication by $w$. Thus, the leading order coefficient in the resulting series, $b_{(k),0}^{[s],\infty,-\sign s}$, is a nontrivial combination of $\omega$, $a$, $m$ and the $k$ first coefficients of the original asymptotic expansion for solutions of \eqref{eq:radial-ODE-alpha}; as those coefficients can be computed explicitly by the formula \eqref{eq:recursion-irregular-singularity} and the expansions in Lemma~\ref{lemma:asymp-f-g}, one can in principle obtain $b_{(k),0}^{[s],\infty,-\sign s}$ explicitly in terms of $c_{0}^{[s],\infty,\sign s }$. For simplicity, we give the result only for the cases $k=0,1,2$.

On the other hand, if the exponential dependence is $e^{\sign s\, i\omega}r^{2iM\omega \sign s}$, application of $\frac{\mc{L}}{w}$ does not create any cancellation: the series is essentially multiplied by $-2i\sign s \omega r^2$, so the new leading order coefficient, $b_{(k),0}^{[s],\infty,+\sign s}$,  depends only on $\omega$ and on the previous leading order coefficient.

For $|a|=M$, as $r\to M$, we use the identity
$$\mc{L} = \frac{\Delta}{r^2+M^2}\lp(-\sign s\frac{d}{dr}-\frac{\beta}{(r-M)^2}-\frac{i\omega(r+M)}{r-M}\rp)\,,$$
and for $|a|<M$, as $r\to r_+$, we use
$$\mc{L} = \frac{\Delta}{r^2+a^2}\lp(-\sign s\frac{d}{dr}+\frac{\xi}{r-r_+}-\frac{\xi +i\omega(r+r_+)}{r-r_-}\rp)\,.$$
With these, we can note that by application of $\frac{\mc{L}}{w}$, there can either be cancellation between the terms involving $\beta$ or $\xi$, respectively, or the opposite. The leading order coefficient is very simple to obtain in the latter case. In the former, we will need to appeal to \eqref{eq:recursion-irregular-singularity}  or \eqref{eq:recursion-regular-singularity} and the expansions in Lemma~\ref{lemma:asymp-f-g} to compute $b_{(k),0}^{[s],M,\pm}$ or $b_{[s],(k),0}^{r_+,\pm}$, respectively.

To be more concrete, let us consider the case $a<M$. We can explicitly compute, for $s>0$,
\begin{align*}
&\lp(\frac{\uL}{w}\rp)^k \lp(\sum_{j=0}^{|s|}c_{j}^{[s],\infty,+} (r-r_+)^{\xi-s+j}\Delta^{s}(r^2+a^2)^{1/2-s}\rp)\\
&\quad =\lp(\frac{\uL}{w}\rp)^{k-1}\sum_{j=0}^{|s|} c_{j}^{[s],\infty,+}\lp(\frac{j}{r-r_+}+\frac{(2s-1)r}{r^2+a^2}-\frac{\xi+s +i\omega(r+r_+)}{r-r_-}\rp)(r-r_+)^{\xi+j} (r-r_-)^s(r^2+a^2)^{3/2-s}\\
&\quad = \dots
\end{align*}
For instance, if $k=1$, we obtain
\begin{align*}
b_{(1),0}^{[+|s|],r_+,+} &= (2Mr_+)^{3/2-|s|}(r_+-r_-)^{|s|-1}\lp[c_{1}^{[+|s|],r_+,+}(r_+-r_-)\rp.\\
&\qquad\qquad\qquad\qquad\qquad\qquad\qquad\qquad\lp.+\lp(\frac{(2|s|-1)r_+}{2Mr_+} -\xi-|s| -i\omega(r+r_+)\rp) c_{0}^{[+|s|],r_+,+}\rp]\\
&=c_{0}^{[+|s|],r_+,+}(2Mr_+)^{3/2-|s|}(r_+-r_-)^{|s|-1}\frac{\mathfrak{D}_{s,1,+}^\mc{H}}{\lp(-2\xi +|s|-1\rp)}\,,
\end{align*}
where $\mathfrak{D}_{s,1,+}^\mc{H}$ is some complex-valued constant depending on the angular parameters and $s$; by a similar procedure, we can obtain
\begin{align*}
b_{(k),0}^{[+|s|],r_+,+}&=c_{0}^{[+|s|],r_+,+}(2Mr_+)^{1/2+k-|s|}(r_+-r_-)^{|s|-k}\frac{\mathfrak{D}_{s,k,+}^\mc{H}}{\prod_{j=1}^{k}\lp(-2\xi +|s|-j\rp)}\,,\\
b_{(k),0}^{[-|s|],r_+,-}&=c_{0}^{[-|s|],r_+,-}(2Mr_+)^{1/2+k-|s|}(r_+-r_-)^{-k}\frac{\mathfrak{D}_{s,k,-}^\mc{H}}{\prod_{j=1}^{k}\lp(-2\xi -|s|+j\rp)}\,,
\end{align*}
where $\lp|\mathfrak{D}_{s,k,+}^\mc{H}\rp|^2=\lp|\mathfrak{D}_{s,k,-}^\mc{H}\rp|^2$ is identified as $\mathfrak{D}_{s,k}^\mc{H}$.

Our computations hold in the inhomogeneous case as well because,  by the properties of the inhomogeneity
\begin{align*}
\lp(\frac{\mc{L}}{w}\rp)^{k}\lp(\frac{\Delta^{1+|s|}}{(r^2+a^2)^{2|s|+3/2}}\smlambda{F}{s} \rp)&= O\lp(r^{-2(|s|-k)-1}\rp)\text{~as~} r\to \infty\,,\\
\lp(\frac{\mc{L}}{w}\rp)^{k}\lp(\frac{\Delta^{1+|s|}}{(r^2+a^2)^{2|s|+3/2}}\smlambda{F}{s} \rp) &= O\lp(\Delta^{1+|s|-k}\rp) \text{~as~} r\to r_+\,,
\end{align*}
 hence it does not affect the leading asymptotics of the solution.
\end{proof}

\subsubsection{Boundary conditions for the Teukolsky radial ODE}
\label{sec:radial-ODE-bdry-conditions-Teukoslsky}

To discuss boundary conditions for \eqref{eq:radial-ODE-alpha}, it is convenient to label the  functions in Lemma~\ref{lemma:aymptotic-analysis-radialODE-Teukolsky}:
\begin{definition} Fix $M>0$, $|a|\leq M$, $s\in\frac12\mathbb{Z}$ and an admissible frequency triple $(\omega,m,\Lambda)$ with respect to $s$ such that $\omega\in\mathbb{R}\backslash\{0,m\upomega_+\}$.
\begin{enumerate}
\item Define $\upalpha^{[s],\, (a\omega)}_{ml,\,\mc{H}^+}$ and $\upalpha^{[s],\, (a\omega)}_{ml,\,\mc{H}^-}$ to be the unique classical solutions to the {\normalfont homogeneous} radial ODE~(\ref{eq:radial-ODE-alpha}) with boundary conditions
\begin{enumerate}
\item if $|a|<M$,
  \begin{enumerate}
\item $\upalpha^{[s],\, (a\omega)}_{ml,\,\mc{H}^\pm}(r)(r-r_+)^{\mp  \xi + s(1\pm 1)/2}$ are smooth at $r=r_+\,,$
\item $\lp|\lp((r-r_+)^{\mp  \xi + s(1\pm 1)/2}\upalpha^{[s],\, (a\omega)}_{ml,\,\mc{H}^\pm}\rp)\big|_{r=r_+}\rp|^2=1\,;$
  \end{enumerate}
\item if $|a|=M$,
  \begin{enumerate}
\item $(r-M)^{\pm 2iM\omega+s(1\pm 1)}e^{\mp \beta(r-M)^{-1}}\upalpha^{[s],\, (a\omega)}_{ml,\,\mc{H}^\pm }(r)$ are smooth at $r=M\,,$
\item $\lp|\lp((r-M)^{\pm 2iM\omega+s(1\pm 1)}e^{\mp \beta(r-M)^{-1}}\upalpha^{[s],\, (a\omega)}_{ml,\,\mc{H}^\pm }\rp)\big|_{r=M}\rp|^2=1\,.$
  \end{enumerate}
\end{enumerate}
\item Define $\upalpha^{[s],\, (a\omega)}_{ml,\,\mc{I}^+}$ and $\upalpha^{[s],\, (a\omega)}_{ml,\,\mc{I}^-}$ to be the unique classical solution to the {\normalfont homogeneous} radial ODE~(\ref{eq:radial-ODE-alpha}) and  boundary conditions
\begin{enumerate}
\item $\upalpha^{[s],\, (a\omega)}_{ml,\,\mc{I}^\pm }\sim e^{\pm i\omega r}r^{\pm 2Mi\omega-s(1\pm 1)-1}$ asymptotically\footnote{This notation means that there are constants $\{c_k\}_{k=0}^\infty$ such that for every $N\geq 1$, as $r\to\infty$ (see Lemma~\ref{lemma:aymptotic-analysis-radialODE-Teukolsky}), $$\upalpha^{[s],\, (a\omega)}_{ml,\,\mc{I}^+}(r) = e^{i\omega r+ 2iM\omega \log r}\sum_{k=0}^{N}c_kr^{-2s-k-1}+O(r^{-2s-N-2})\,.$$} as $r\to \infty\,.$
\item $\lp|\lp(e^{\mp i\omega r}r^{\mp 2Mi\omega+s(1\pm 1)+1}\upalpha^{[s],\, (a\omega)}_{ml,\,\mc{I}^\pm }\rp)\big|_{r=\infty}\rp|^2=1\,.$
\end{enumerate}
\end{enumerate}
\label{def:uhor-uout}

Finally, we define $u^{[s],\, (a\omega)}_{ml,\,\mc{I}^\pm}$ and $u^{[s],\, (a\omega)}_{ml,\,\mc{H}^\pm}$ by the following rescaling of the preveiously introduced functions:
\begin{align*}
u^{[s],\, (a\omega)}_{ml,\,\mc{I}^\pm}=(r^2+a^2)^{1/2}\upalpha^{[s],\, (a\omega)}_{ml,\,\mc{I}^\pm}\,, \qquad u^{[s],\, (a\omega)}_{ml,\,\mc{H}^\pm}=(r^2+a^2)^{1/2}\upalpha^{[s],\, (a\omega)}_{ml,\,\mc{H}^\pm}\,.
\end{align*}
\end{definition}

By Lemma~\ref{lemma:aymptotic-analysis-radialODE-Teukolsky}, it is clear that any solution to the radial ODE \eqref{eq:radial-ODE-alpha} has a representation  in terms of those defined in Definition~\ref{def:uhor-uout}:

\begin{lemma} \label{lemma:apha-general-asymptotics} Fix $M>0$, $|a|\leq M$, $s\in\frac12\mathbb{Z}$ and an admissible frequency triple $(\omega,m,\Lambda)$ with respect to $a$ and $s$ such that $\omega\in\mathbb{R}\backslash\{0,m\upomega_+\}$. If $\smlambda{F}{s}$ is bounded for $r\in[r_+,\infty)$, there are $\swei{a}{s}_{\mc{H}^\pm},\swei{a}{s}_{\mc{I}^\pm}\in\mathbb{C}$, depending only on $(\omega,m,\Lambda)$ such that a solution $\upalpha_{m\Lambda}^{[s],\,(a\omega)}$ to the radial ODE~\eqref{eq:radial-ODE-alpha} can be written as,
\begin{equation}
\begin{split}
\smlambda{\upalpha}{s}&= \swei{a}{s}_{\mc{H}^+}\cdot\upalpha^{[s],\, (a\omega)}_{ml,\,\mc{H}^+}+\swei{a}{s}_{\mc{H}^-}\cdot\upalpha^{[s],\, (a\omega)}_{ml,\,\mc{H}^-} + O(\Delta^{|s|-k+1})\,,\\ 
\smlambda{\upalpha}{s}&= \swei{a}{s}_{\mc{I}^+}\cdot\upalpha^{[s],\, (a\omega)}_{ml,\,\mc{I}^+}+\swei{a}{s}_{\mc{I}^-}\cdot\upalpha^{[s],\, (a\omega)}_{ml,\,\mc{I}^-} +O(r^{-2(|s|-k)-2})\,,
\end{split}\label{eq:alpha-general-asymptotics}
\end{equation}
respectively as $r\to r_+$ and $r\to \infty$.
\end{lemma}

Often, we want to consider a more restrictive set of solutions to \eqref{eq:radial-ODE-alpha}, which are said to have \textit{outgoing boundary conditions}.

\begin{definition}[Outgoing boundary conditions] \label{def:outgoing-bdry-teukolsky} We say $\smlambda{\upalpha}{s}$ is a solution of the radial ODE \eqref{eq:radial-ODE-alpha} with {\normalfont outgoing boundary conditions} if 
\begin{equation}\label{eq:outgoing-bdry-teukolsky} 
\begin{split}
\smlambda{\upalpha}{s}&= \swei{a}{s}_{\mc{H}^+}\cdot\upalpha^{[s],\, (a\omega)}_{ml,\,\mc{H}^+}+ O(\Delta)\,,\\
\smlambda{\upalpha}{s}&= \swei{a}{s}_{\mc{I}^+}\cdot\upalpha^{[s],\, (a\omega)}_{ml,\,\mc{I}^+} +O(r^{-2})\,,
\end{split}
\end{equation}
respectively as $r\to r_+$ and $r\to \infty$. Note that, unlike in the previous Lemma~\ref{lemma:apha-general-asymptotics}, the last term in the identities needn't be lower order.
\end{definition}

\begin{remark} Definitions \ref{def:uhor-uout} and \ref{def:outgoing-bdry-teukolsky} as well as Lemma~\ref{lemma:apha-general-asymptotics} can be alternatively written in terms of $\smlambda{u}{s}=(r^2+a^2)^{1/2}\Delta^{s/2}\smlambda{\upalpha}{s}$: we simply let
\begin{align*}
u^{[s],\, (a\omega)}_{m\Lambda,\,\mc{I}^\pm}:=(r^2+a^2)^{1/2}\Delta^{s/2}\upalpha^{[s],\, (a\omega)}_{m\Lambda,\,\mc{I}^\pm}\,, \quad u^{[s],\, (a\omega)}_{ml,\,\mc{H}^\pm}:=(r^2+a^2)^{1/2}\Delta^{s/2}\upalpha^{[s],\, (a\omega)}_{ml,\,\mc{H}^\pm}\,.
\end{align*}
\end{remark}

\subsubsection{Boundary conditions for the transformed radial ODEs}
\label{sec:radial-ODE-bdry-conditions-transformed}

To discuss boundary conditions for \eqref{eq:transformed-k-separated}, it is again convenient to label the  functions in Lemma~\ref{lemma:aymptotic-analysis-radialODE-transformed}:
\begin{definition} Fix $M>0$, $|a|\leq M$, $s\in\frac12\mathbb{Z}$ and an admissible frequency triple $(\omega,m,\Lambda)$ with respect to $s$ such that $\omega\in\mathbb{R}\backslash\{0,m\upomega_+\}$.
\begin{enumerate}
\item Define ${\lp(\uppsi_{(k)}\rp)}^{[s],\, (a\omega)}_{m\Lambda,\,\mc{H}^+}$ and ${\lp(\uppsi_{(k)}\rp)}^{[s],\, (a\omega)}_{m\Lambda,\,\mc{H}^-}$ as the unique solutions to the {\normalfont homogeneous} radial ODE~(\ref{eq:transformed-k-separated}) with boundary conditions
\begin{enumerate}
\item if $|a|<M$,
  \begin{enumerate}
\item ${\lp(\uppsi_{(k)}\rp)}^{[s],\, (a\omega)}_{m\Lambda,\,\mc{H}^\pm }(r)(r-r_+)^{\mp \xi - \frac12(|s|-k)(1\mp\sign s)}$ are smooth at $r=r_+\,,$ and
\item $\lp|\lp(w^{-\frac{|s|-k}{2}}\Delta^{\pm \frac{|s|-k}{2}\sign s}(r-r_+)^{\mp \xi}{\lp(\uppsi_{(k)}\rp)}^{[s],\, (a\omega)}_{ml,\,\mc{H}^\pm}\rp)\Big|_{r=r_+}\rp|^2=1\,;$
  \end{enumerate}
\item if $|a|=M$,
  \begin{enumerate}
\item $(r-M)^{\pm 2iM\omega-(|s|-k)(1\mp \sign s)}{\lp(\uppsi_{(k)}\rp)}^{[s],\, (a\omega)}_{ml,\,\mc{H}^\pm}(r)e^{\mp \beta(r-M)^{-1}}$ are  smooth at $r=M\,,$ and 
\item $\lp|\lp(w^{-\frac{|s|-k}{2}}(r-M)^{\pm 2iM\omega \pm (|s|-k)\sign s}e^{\mp \beta(r-M)^{-1}}{\lp(\uppsi_{(k)}\rp)}^{[s],\, (a\omega)}_{ml,\,\mc{H}^\pm}\rp)\Big|_{r=M}\rp|^2=1\,.$
  \end{enumerate}
\end{enumerate}
\item Define ${\lp(\uppsi_{(k)}\rp)}^{[s],\, (a\omega)}_{ml,\,\mc{I}^+}$ and ${\lp(\uppsi_{(k)}\rp)}^{[s],\, (a\omega)}_{ml,\,\mc{I}^-}$ as the unique classical solutions to the {\normalfont homogeneous} radial ODE~(\ref{eq:transformed-k-separated}) with  boundary conditions
\begin{enumerate}
\item ${\lp(\uppsi_{(k)}\rp)}^{[s],\, (a\omega)}_{ml,\,\mc{I}^\pm}\sim e^{\pm i\omega r}r^{\pm 2Mi\omega-(|s|-k)(1\pm \sign s)}$ asymptotically  as $r\to \infty\,,$ and
\item $\lp|\lp(e^{\mp i\omega r}r^{\mp 2iM\omega}w^{-\frac{|s|-k}{2}}\Delta^{\pm \frac{|s|-k}{2}\sign s}{\lp(\uppsi_{(k)}\rp)}^{[s],\, (a\omega)}_{ml,\,\mc{I}^\pm}\rp)\big|_{r=\infty}\rp|^2=1\,.$ 
\end{enumerate}
\end{enumerate}
\label{def:psihor-psiout}
\end{definition}

By Lemma~\ref{lemma:aymptotic-analysis-radialODE-transformed}, a solution to the radial ODE \eqref{eq:transformed-k-separated}  arising from a solution to the radial ODE \eqref{eq:radial-ODE-alpha} has a representation in terms of the functions defined in Definition~\ref{def:psihor-psiout}:

\begin{lemma} \label{lemma:uppsi-general-asymptotics} Fix $M>0$, $|a|\leq M$, $s\in\frac12\mathbb{Z}$, $k\in\{0,...,|s|\}$ and an admissible frequency triple $(\omega,m,\Lambda)$ with respect to $a$ and $s$ such that $\omega\in\mathbb{R}\backslash\{0,m\upomega_+\}$. If $\smlambda{F}{s}$ is bounded for $r\in[r_+,\infty)$, there are $\swei{A}{s}_{k,\mc{H}^\pm},\swei{A}{s}_{k,\mc{I}^\pm}\in\mathbb{C}$, depending only on $(\omega,m,\Lambda)$ such that a function $\lp(\uppsi_{(k)}\rp)_{m\Lambda}^{[s],\,(a\omega)}$ defined inductively, via \eqref{eq:transformed-transport-separated}, from a solution, $\smlambda{\upalpha}{s}$, of the radial ODE~\eqref{eq:radial-ODE-alpha} can be written as,
\begin{equation}
\begin{split}
\smlambda{\lp(\uppsi_{(k)}\rp)}{s}&= \swei{A}{s}_{k,\mc{H}^+}\lp(\uppsi_{(k)}\rp)^{[s],\, (a\omega)}_{m\Lambda,\,\mc{H}^+}+\swei{A}{s}_{k,\mc{H}^-}\lp(\uppsi_{(k)}\rp)^{[s],\, (a\omega)}_{m\Lambda,\,\mc{H}^-}\,,\qquad \\
\smlambda{\lp(\uppsi_{(k)}\rp)}{s}&= \swei{A}{s}_{k,\mc{I}^+}\lp(\uppsi_{(k)}\rp)^{[s],\, (a\omega)}_{m\Lambda,\,\mc{I}^+}+\swei{A}{s}_{k,\mc{I}^-}\lp(\uppsi_{(k)}\rp)^{[s],\, (a\omega)}_{m\Lambda,\,\mc{I}^-}\,.
\end{split}\label{eq:uppsi-general-asymptotics}
\end{equation}
The coefficients $\swei{A}{s}_{k,\mc{H}^\pm},\swei{A}{s}_{k,\mc{I}^\pm}$ can be related to the coefficients $\swei{a}{s}_{\mc{H}^\pm},\swei{a}{s}_{\mc{I}^\pm}$ from Lemma~\ref{lemma:apha-general-asymptotics}:
\begin{align}\label{eq:a-to-A-infinity}
\begin{split}
\lp|\swei{A}{\pm s}_{k,\,\mc{I}^\pm}\rp|^2&=(2\omega)^{2k}\lp|\swei{a}{\pm s}_{\mc{I}^\pm}\rp|^2\,,\qquad
\lp|\swei{A}{\mp s}_{k,\,\mc{I}^\pm}\rp|^2=\frac{\mathfrak{D}_{s,k}^{\mc{I}}}{(2\omega)^{2k}}\lp|\swei{a}{\mp s}_{\mc{I}^\pm}\rp|^2\,;
\end{split}
\end{align}
if $|a|=M$,
\begin{equation}\label{eq:a-to-A-horizon-extremal}
\begin{split}
\lp|\swei{A}{\mp s}_{k,\,\mc{H}^\pm}\rp|^2&=2M^2[4M^2(\omega-m\upomega_+)]^{2k}\lp|\swei{a}{\mp s}_{\mc{H}^\pm}\rp|^2\,, \\
\lp|\swei{A}{\pm s}_{k,\,\mc{H}^\pm}\rp|^2&=2M^2\frac{\mathfrak{D}_{s,k}^{\mc{H}}(2M^2)^{2(k-|s|)}}{[4M^2(\omega-m\upomega_+)]^{2k}}\lp|\swei{a}{\pm s}_{\mc{H}^\pm}\rp|^2\,,
\end{split}
\end{equation}
and if $|a|<M$,
\begin{align}\label{eq:a-to-A-horizon-sub}
\begin{split}
\lp|\swei{A}{\mp s}_{k,\,\mc{H}^\pm}\rp|^2&=2Mr_+\prod_{j=0}^{k-1}\lp[\lp(\frac{4Mr_+}{r_+-r_-}\rp)^2(\omega-m\upomega_+)^2+(s-j)^2\rp]\lp|\swei{a}{\mp s}_{\mc{H}^\pm}\rp|^2\,, \\
\lp|\swei{A}{\mp s}_{k,\,\mc{H}^\mp}\rp|^2&=2Mr_+\frac{\mathfrak{D}_{s,k}^{\mc{H}}(2Mr_+)^{2(k-|s|)}(r_+-r_-)^{|s|(1\mp 1)}}{\prod_{j=1}^{k}\lp\{\lp[4Mr_+(\omega-m\upomega_+)\rp]^2+(s-j)^2(r_+-r_-)^2\rp\}}\lp|\swei{a}{\pm s}_{\mc{H}^\pm}\rp|^2\,.
\end{split}
\end{align}
Here $\mathfrak{D}_{s,k}^{\mc{I}}$ and $\mathfrak{D}_{s,k}^{\mc{H}}$ coincide with those in Lemma~\ref{lemma:aymptotic-analysis-radialODE-transformed}.

When $k=|s|$, we write $\swei{A}{s}_{\mc{I}^\pm}=\swei{A}{s}_{|s|,\mc{I}^\pm}$ and $\swei{A}{s}_{\mc{I}^\pm}=\swei{A}{s}_{|s|,\mc{I}^\pm}$. 
\end{lemma}

\begin{definition}[Outgoing boundary conditions] \label{def:outgoing-bdry-uppsi} We say $\lp(\uppsi_{(k)}\rp)_{m\Lambda}^{[s],\,(a\omega)}$ is a solution of the radial ODE \eqref{eq:transformed-k-separated} with {\normalfont outgoing boundary conditions} if 
\begin{equation}\label{eq:outgoing-bdry-uppsi} 
\begin{split}
\lp(\uppsi_{(k)}\rp)_{m\Lambda}^{[s],\,(a\omega)}&= \swei{A}{s}_{k,\mc{H}^+}\cdot(\uppsi_{(k)})^{[s],\, (a\omega)}_{ml,\,\mc{H}^+}+ O(\Delta)\,,\\
\lp(\uppsi_{(k)}\rp)_{m\Lambda}^{[s],\,(a\omega)}&= \swei{A}{s}_{k,\mc{I}^+}\cdot(\uppsi_{(k)})^{[s],\, (a\omega)}_{ml,\,\mc{I}^+} +O(r^{-2})\,,
\end{split}
\end{equation}
respectively as $r\to r_+$ and $r\to \infty$. Then, $\swei{A}{s}_{k,\mc{H}^+}$ can be related to $\swei{a}{s}_{\mc{H}^+}$ from Definition~\ref{def:outgoing-bdry-teukolsky} by the formulas of Lemma~\ref{lemma:uppsi-general-asymptotics} if $s<0$; $\swei{A}{s}_{k,\mc{I}^+}$ can be related to $\swei{a}{s}_{\mc{I}^+}$ from Definition~\ref{def:outgoing-bdry-teukolsky} by the formulas of Lemma~\ref{lemma:uppsi-general-asymptotics} if $s>0$.
\end{definition}

\subsubsection{The Teukolsky--Starobinsky identities}
\label{sec:Teukolsky--Starobinsky}

In this section, we introduce the Teukolsky--Starobinsky identities, originally from \cite{Teukolsky1974,Starobinsky1974} for $|s|=1,2$ and later generalized to all spins in \cite{Kalnins1989}. To do so, we begin by defining the differential operators
\begin{align}
\hat{\mc{D}}^{\pm}_n &= \frac{d}{dr}\pm i\lp(\frac{\omega(r^2+a^2)}{\Delta} -\frac{am}{\Delta}\rp)+\frac{2n(r-M)}{\Delta}\,, \qquad
\hat{\mc{L}}^{\pm}_n &= \frac{d}{d\theta}\pm\lp(\frac{m}{\sin\theta}-a\omega\cos\theta\rp)+n\cot\theta\,. \label{eq:def-D-L-pm}
\end{align}

Next, we introduce the Teukolsky--Starobinsky constants:

\begin{proposition}\label{prop:TS-angular-radial-constant} Fix $s\in\{0,\frac12,1,\frac32,2\}$ and $|a|\leq M$.
\begin{enumerate}[label=(\roman*)]
\item  Fix a frequency triple $(\omega,m,l)$ admissible with respect to $s$. Solutions of the angular ODE~\eqref{eq:angular-ode} with spin $\pm s$ are eigenfunctions of the operator 
\begin{align*}
\lp(\prod_{j=0}^{2s-1}\hat{\mc{L}}_{s-j}^\mp\rp)\lp(\prod_{k=0}^{2s-1}\hat{\mc{L}}_{s-k}^\pm\rp) \equiv \lp(\sin\theta\rp)^{2s}\lp(\frac{\hat{\mc{L}}_{s}^\mp}{\sin\theta}\rp)^{2s}\lp(\sin\theta\rp)^{2s}\lp(\frac{\hat{\mc{L}}_{s}^\pm}{\sin\theta}\rp)^{2s}\,, 
\end{align*}
with indices $j,k$ increasing from right to left on the product, and the latter being replaced by the identity if $s=0$, for the same eigenvalue. This eigenvalue, $\mathfrak B_s=\mathfrak B_s(|s|,\omega,m,l)$, called the {\normalfont angular Teukolsky--Starobinsky constant}, satisfies $(-1)^{2s}\mathfrak B_s\geq 0$.

\item Fix a frequency triple $(\omega,m,\Lambda)$ admissible with respect to $s$. Dropping most subscripts, for $\swei{\upalpha}{\pm s}$ solutions of the homogenenous radial ODE~\eqref{eq:radial-ODE-alpha} of spin $\pm s$, set
\begin{align*}
P^{[+s]}:=\Delta^s \swei{\upalpha}{+s}\,,\qquad P^{[-s]}:=\swei{\upalpha}{-s}\,,
\end{align*}
then $P^{[\pm s]}$ is an eigenfunction of the operator
\begin{align*}
\Delta^s\lp(\hat{\mc{D}}^{\mp}_0\rp)^{2s}\lp[\Delta^s\lp(\hat{\mc{D}}^{\pm}_0\rp)^{2s}\rp] \equiv \prod_{j=0}^{2s-1}\lp(\Delta^{1/2}\hat{\mc{D}}^{\mp}_{j/2}\rp)\prod_{k=0}^{2s-1}\lp(\Delta^{1/2}\hat{\mc{D}}^{\pm}_{k/2}\rp)\,, 
\end{align*}
with indices $j,k$ increasing from right to left in the product, and the latter being replaced by the identity if $s=0$, for the same eigenvalue. This eigenvalue, $\mathfrak C_s=\mathfrak C_s(\omega,m,\Lambda)\in\mathbb{R}$ is called the {\normalfont radial Teukolsky--Starobinsky constant}; if $\mathfrak C_s=0$, we say $(\omega,m,\Lambda)$ is an {\normalfont algebraically special frequency triple}.
\end{enumerate}
The Teukolsky--Starobinksy constants $\mathfrak B_s$ and $\mathfrak C_s$ can be computed explicity; for instance
\begin{align}\label{eq:TS-radial-constants}
\begin{split} 
\mathfrak C_{1}&= (\Lambda-2am\omega+1)^2+4am\omega-4a^2\omega^2\,,\\
\mathfrak C_2&= \lp[(\Lambda-2am\omega+2)(\Lambda-2am\omega+4)\rp]^2+40a\omega(\Lambda-2am\omega+2)^2(m-a\omega)\\
&\qquad+48a\omega(\Lambda-2am\omega+2)(m+a\omega)+144a^2\omega^2(m-a\omega)^2+144M^2\omega^2\,,
\end{split}\\
 \label{eq:TS-angular-constants}
\begin{split}
\mathfrak B_{1/2}&= -\mathfrak{C}_{1/2}\,, \quad \mathfrak B_{1}=\mathfrak{C}_1\,,\quad \mathfrak B_{3/2}=  -\mathfrak{C}_{3/2}\,,\quad\mathfrak B_2= \mathfrak{C_2}-144M^2\omega^2\,.
\end{split}
\end{align}
\end{proposition}
\begin{proof} Existence of the Teukolsky--Starobinsky constants, as well as their formulas in terms of the frequency parameters, follows from applying the operators in the products above to angular or radial functions one by one and using the angular or radial ODEs, respectively, to trade second derivatives of those functions by first and zeroth order terms (see for instance \cite[Sections 70 and 81]{Chandrasekhar} for $s=\pm 1, \pm 2$). Alternatively, they can be obtained by computations involving asymptotic expansions, as in the proof of Lemma~\ref{lemma:aymptotic-analysis-radialODE-transformed}. 
\end{proof}

\begin{remark} \label{rmk:TS-constant-higher-spin} Though we have stated Proposition~\ref{prop:TS-angular-radial-constant} only for $|s|\leq 2$, we have verified that the Teukolsky--Starobinsky constants exist for spins $|s|\leq 9/2$ (see also \cite{Kalnins1989} for computations up to $|s|= 4$). We expect that, indeed, the result can be shown for general $s$ by tapping into to an inductive structure for the computation of $\mathfrak B_s$ and $\mathfrak C_s$ that is still unknown at present.

Of particular importance in this work will be the radial Teukolsky--Starobinsky constant. In  statement 2 of Proposition~\ref{prop:TS-angular-radial-constant}, $\mathfrak{C}_s$ is defined by making no assumption on the class of solutions of the radial ODE~\eqref{eq:radial-ODE-alpha}. However, if one were to assume outgoing boundary conditions, $\mathfrak{C}_s$ could alternatively be defined by the radial Teukolsky--Starobinsky identities in Proposition~\ref{prop:TS-radial}. For our main results, which are restricted to the physically relevant integer spins  $|s|=0,1,2$, however, the distinction between the two ways of defining $\mathfrak C_s$ is not crucial; hence, for simplicity, we always take $\mathfrak C_s$ to be defined by Proposition~\ref{prop:TS-angular-radial-constant}.
\end{remark} 

\begin{remark}\label{rmk:TS-constant-sign}
Let $\Lambda=\smlambda{\bm\Lambda}{s}$. For any $s\in\frac12\mathbb{Z}$ for which an angular Teukolsky--Starobinsky constant can be defined (see Remark~\ref{rmk:TS-constant-higher-spin}), $(-1)^{2s}\mathfrak B_s\geq 0$ (see \cite[Section 2.3]{TeixeiradaCosta2019}). For $|s|\leq 2$, explicit computations of the latter (see \cite{Kalnins1989}, for instance) show
\begin{align} \label{eq:Cs-Bs-relation}
\mathfrak{C}_s= (-1)^{2s}\mathfrak{B}_s\geq 0\,, \quad |s|\leq 2\,,\quad\quad
\mathfrak{C}_2= (-1)^{2s}\mathfrak{B}_s+144M^2\omega^2\geq 0\,,
\end{align}
although $\mathfrak{C}_s-(-1)^{2s}\mathfrak{B}_s\geq 0$ does not necessarily hold for $|s|>2$ and any $|a|\leq M$. For $|s|\leq 2$, the non-negativity in \eqref{eq:Cs-Bs-relation} motivates the admissibility condition 4(e) in Definition~\ref{def:admissible-freqs}. 

If one can upgrade \eqref{eq:Cs-Bs-relation} for a given spin into a strict positivity statement, then there are no real algebraically special frequencies for that spin.
It is well-known that this is the case for $a\omega=0$, as
$$\mathfrak{C}_s\geq (-1)^{2s}\mathfrak{B}_s=\lp[l(l+1)-s^2+|s|\rp]^{2|s|}\geq \lp[2|s|\rp]^{2|s|} \,, \quad 0<|s|\leq 2\,, \quad a\omega=0\,.$$
The same strict inequality is not as clear for $a\omega\neq 0$ (see, for instance, \cite{Wald1973}). However, we note that in the case $s=2$, as $144M^2\omega^2>0$ whenever $\omega\neq 0$, one obtains from \eqref{eq:Cs-Bs-relation} that
$\mathfrak{C}_s> 0$
unconditionally, so we conclude there are no real algebraically special frequencies for this spin. To our knowledge, this interesting feature of gravitational perturbations has not been observed before.
\end{remark}

The radial Teukolsky--Starobinsky identities are differential identities which relate solutions of the homogeneous radial ODE~\eqref{eq:radial-ODE-alpha} with spin $+s$ and spin $-s$. 
\begin{proposition}[Radial Teukolsky--Starobinsky identities] \label{prop:TS-radial} Fix $M>0$, $|a|\leq M$ and $s\in\{0,\frac12,1,\frac32,2\}$. Let $(\omega,m,l)$ be an admissible frequency triple with respect to $s$ and $a$. Dropping most subscripts, let $\swei{\upalpha}{s}$ be solutions to the homogeneous radial ODE~\eqref{eq:radial-ODE-alpha} of spin $\pm s$. If $\omega\in\mathbb{R}\backslash\{0,m\upomega_+\}$, recall from \eqref{eq:alpha-general-asymptotics}
\begin{equation*}
\begin{split}
\swei{\upalpha}{\pm s}&= \swei{a}{\pm s}_{\mc{H}^+}\cdot\swei{\upalpha}{\pm s}_{\mc{H}^+}+\swei{a}{\pm s}_{\mc{H}^-}\cdot\swei{\upalpha}{\pm s}_{\mc{H}^-}= \swei{a}{\pm s}_{\mc{I}^+}\cdot\swei{\upalpha}{\pm s}_{\mc{I}^+}+\swei{a}{\pm s}_{\mc{I}^-}\cdot\swei{\upalpha}{\pm s}_{\mc{I}^-}\,,
\end{split}
\end{equation*}
for some complex $\swei{a}{\pm s}_{\mc{H}^{+}}$, $\swei{a}{\pm s}_{\mc{H}^{-}}$, $\swei{a}{\pm s}_{\mc{I}^{+}}$ and $\swei{a}{\pm s}_{\mc{H}^{-}}$. For such $\swei{\upalpha}{\pm s}$,
\begin{align}
\begin{split}
\Delta^s \lp(\hat{\mc{D}}_0^{+}\rp)^{2s}\lp(\Delta^s \swei{\upalpha}{+ s}\rp)&=\swei{a}{+s}_{\mc{I}^{+}}\mathfrak{C}_s^{(1)}\swei{\upalpha}{- s}_{\mc{I}^+}+\swei{a}{+s}_{\mc{I}^{-}}\mathfrak{C}_s^{(7)}\swei{\upalpha}{- s}_{\mc{I}^-}=\swei{a}{+s}_{\mc{H}^{+}}\mathfrak{C}_s^{(4)}\swei{\upalpha}{- s}_{\mc{H}^+}+\swei{a}{+s}_{\mc{H}^{-}}\mathfrak{C}_s^{(6)}\swei{\upalpha}{- s}_{\mc{H}^-}\,, \\
 \lp(\hat{\mc{D}}_0^{-}\rp)^{2s}\swei{\upalpha}{- s}&=\swei{a}{-s}_{\mc{I}^{+}}\mathfrak{C}_s^{(3)}\swei{\upalpha}{+ s}_{\mc{I}^+}+\swei{a}{+s}_{\mc{I}^{-}}\mathfrak{C}_s^{(5)}\swei{\upalpha}{+ s}_{\mc{I}^-}=\swei{a}{-s}_{\mc{H}^{+}}\mathfrak{C}_s^{(2)}\swei{\upalpha}{+ s}_{\mc{H}^+}+\swei{\upalpha}{+ s}_{\mc{H}^{-}}\mathfrak{C}_s^{(8)}\swei{\upalpha}{+ s}_{\mc{H}^-}\,,
\end{split}\label{eq:TS-radial-general}
\end{align}
where, recalling the radial Teukolsky--Starobinsky constant $\mathfrak{C}_s$ from Proposition~\ref{prop:TS-angular-radial-constant}, $\mathfrak{C}_{s}^{(i)}$ are given by:
\begin{gather*}
\mathfrak{C}_s^{(1)}=(2i\omega)^{2s}\,, \quad
\mathfrak{C}_s^{(2)}=\begin{dcases}
(-2\beta)^{2s} &\text{~if~} |a|=M\\
\prod_{j=0}^{2s-1}\lp(2\xi+s-j\rp)&\text{~if~} |a|<M
\end{dcases}\,, \qquad
\mathfrak{C}_s^{(3)}=\frac{\mathfrak{C}_s}{\mathfrak{C}_s^{(1)}}\,, \quad
\mathfrak{C}_s^{(4)}=\frac{\mathfrak{C}_s}{\mathfrak{C}_s^{(2)}}\,, \\
\mathfrak{C}_s^{(5)}=(-2i\omega)^{2s}\,,\quad
\mathfrak{C}_s^{(6)}=\begin{dcases}
(2\beta)^{2s} &\text{~if~} |a|=M\\
(r_+-r_-)^{2|s|}\prod_{j=0}^{2s-1}\lp(-2\xi+s-j\rp)&\text{~if~} |a|<M
\end{dcases}\,, \quad
\mathfrak{C}_s^{(7)}= \frac{\mathfrak{C}_s}{\mathfrak{C}_s^{(5)}}\,,\quad
 \mathfrak{C}_s^{(8)}=\frac{\mathfrak{C}_s}{\mathfrak{C}_s^{(6)}}\,.
\end{gather*}
\end{proposition}
\begin{proof}
The proof follows as that of Lemma~\ref{lemma:aymptotic-analysis-radialODE-transformed}, see \cite[Proposition 2.14]{TeixeiradaCosta2019} for details.
\end{proof}


\section{The separated current templates}
\label{sec:current-templates}

In this section, we introduce the current templates which will be used, in Section~\ref{sec:freq-estimates-Psi} to obtain estimates on solutions of the radial ODEs in Sections~\ref{sec:radial-ODE-Teukolsky-transformed}. For the latter, we rewrite the ODEs in the general form 
\begin{align}
U''+W U'+(\omega^2-V)U=H\,, \label{eq:schrodinger-type}
\end{align}
where $\omega$ is real and $V,W,H$ are known functions of $r^*$ with $W$ being real-valued. 

\subsection{The virial current templates}
\label{sec:virial-current-templates}

Let $U(r^*)$ be  sufficiently regular. For $h=h(r^*)$ a $C^1$ and piecewise $C^2$, $y=y(r^*)$ a $C^0$ and piecewise $C^1$ and $f=f(r^*)$ a $C^2$ and piecewise $C^3$ functions, we define the frequency-localized virial currents
\begin{equation}
\begin{split} \label{eq:virial-current-definitions}
Q^h[U]&=h\Re(U\overline{U}')-\frac12 h'|U|^2+\frac12hW|U|^2\,, \\
Q^y[U]&=y|U'|^2+y(\omega^2-\Re V)|U|^2\,, \\
Q^f[U]&=f|U'|^2+\lp(f(\omega^2-\Re V) -\frac12 f'' +\frac12 f'W\rp)|U|^2+f'\Re(U\overline{U}')\,,
\end{split}
\end{equation}
which, if $U$ satisfies the model ODE \eqref{eq:schrodinger-type}, have derivatives
\begin{equation}
\begin{split} \label{eq:virial-current-identities}
(Q^h[U])'&=h|U'|^2+h\lp(\Re V+\frac12 W'-\omega^2\rp)|U|^2-\frac12 \lp(h''-W h'\rp)|U|^2+h\Re(H\overline{U})\,, \\
\lp(Q^y\rp)'[U]&=(y'-2yW)|U'|^2+y'\omega^2|U|^2-(y\Re V)'|U|^2-2y(\Im V)\Im(U\overline{U}')+2y\Re(H\overline{U}')\,, \\
(Q^f[U])'&=2(f'+fW)|U'|^2+\lp(-f\Re V'-\frac12 f'''+\frac12 Wf''\rp)|U|^2-2f(\Im V) \Im (U\overline{U}')\\ &\quad+f'\Re(H\overline{U})+2f\Re(H\overline{U}')\,.
\end{split}
\end{equation}
These identities are obtained multiplying \eqref{eq:schrodinger-type} by $h\overline{U}$, $2y\overline{U}'$ and $f'\overline{U}+2f\overline{U}'$, respectively, and taking the real part. The expressions above already hint at the applications of each type of current:
\begin{itemize}[itemsep=0pt]
\item We usually consider $h$ currents when $\lp(\Re V+\frac12 W'-\omega^2\rp)$ is quantitatively positive or, if we have a way of controlling the bulk error due to $h|U'|^2$, when $\lp(\Re V+\frac12 W'-\omega^2\rp)$ is quantitatively negative in a given region of $r^*$; for this reason $h$ is typically taken to be a compactly supported function. Since the term $h''$ usually appears with the wrong sign, one must be take care to introduce $h$ only when that error can be controlled by other multiplier currents. We also note the $h$ current is the unique among the virial currents above that does not involve $\Im V$.
\item We usually consider $f$ currents when we want to exploit information about the critical point structure of $\Re V'$, even if we do not have information on the sign of  $\Re V-\omega^2$. Once again, one has to ensure that errors introduced by the term $f'''$  are controlled.
\item The $y$ current is usually reserved best at the $|r^*|\gg 1$ regions, but, as we do not require it to be more regular than $C^0$, it can be very versatile: sometimes it is used to exploit the positivity of $\omega^2-\Re V$; sometimes to exploit the sign of $\Re V'$. The term involving $\Im V$ can be quite worrisome unless either $\Im V$ is sufficiently small in the support of $y$ or one can exploit a special structure in the equation (for the latter, see Lemma~\ref{lemma:h-y-identity-k} as an example).
\end{itemize}

For our system of transformed equations, the multiplier identities for $h$ and $y$ can be rewritten:

\begin{lemma}  \label{lemma:h-y-identity-k} Fix $s\in\mathbb{Z}$ and $0\leq k\leq |s|$. Define $\uppsi_{(k)}$ by \eqref{eq:transformed-transport-separated} and recall the definitions of virial current templates in \eqref{eq:virial-current-definitions}.  Replacing $U$ by $\uppsi_{(k)}$, we have from \eqref{eq:virial-current-identities}
\begin{align} \label{eq:h-identity-k}
\begin{split}
(Q^h)'[\uppsi_{(k)}]&=h|\uppsi_{(k)}'|^2+h\lp[\Re\mc{V}_{(k)}-\frac12(|s|-k)\lp(\frac{w'}{w}\rp)'-\omega^2\rp]|\uppsi_{(k)}|^2+h\Re\lp(\mathfrak{G}_{(k)}\overline{\uppsi_{(k)}}\rp)\\
&\qquad-\frac12 \lp(h''+(|s|-k)\frac{w'}{w} h'\rp)|\uppsi_{(k)}|^2+ awh\sum_{i=0}^{k-1}h\Re\lp[\uppsi_{(i)}\lp(c_{s,k,i}^{\rm id}+imc_{s,k,i}^{\Phi}\rp)\overline{\uppsi_{(k)}}\rp]\,, 
\end{split}\\ 
\label{eq:y-identity-k}
\begin{split}
\lp(Q^y\rp)'[\uppsi_{(k)}]&=y'|\uppsi_{(k)}'|^2+y'\omega^2|\uppsi_{(k)}|^2-(y\Re \tilde{\mc{V}}_{(k)})'|\uppsi_{(k)}|^2 +2y\Re\lp(\mathfrak{G}_{(k)}\overline{\uppsi_{(k)}}'\rp)\\
&\qquad+2y(|s|-k)w'\lp\{w|\uppsi_{(k+1)}|^2-\lp(\omega-\frac{am}{r^2+a^2}\rp)\Re\lp[\uppsi_{(k+1)}\overline{\uppsi_{(k)}}\rp]\rp\} \\
&\qquad+ 2y(|s|-k)\frac{4amrw}{(r^2+a^2)}\lp\{ w\Im \lp[\uppsi_{(k+1)}\overline{\uppsi}_{(k)}\rp]+ \lp(\omega-\frac{am}{r^2+a^2}\rp)|\uppsi_{(k)}|^2\rp\}\,.
\end{split}
\end{align}
\end{lemma}

\begin{proof}
Note that, for $U=\uppsi_{(k)}$, we have
\begin{align*}
W=-(|s|-k)\frac{w'}{w}\,.
\end{align*} 
For the $h$ current, the result follows directly from \eqref{eq:virial-current-identities}. For the $y$ current, given \eqref{eq:transformed-k-separated}, we can easily derive the identities
\begin{align*}
|\uppsi_{(k)}'|^2 &=  w\Re\lp[\uppsi_{(k+1)}\lp(-\sign s\overline{\uppsi_{(k)}}'\rp)\rp]+\sign s \lp(\omega-\frac{am}{r^2+a^2}\rp)\Im\lp[\uppsi_{(k)}\overline{\uppsi_{(k)}}'\rp]\\
&=w^2|\uppsi_{(k+1)}|^2-w\lp(\omega-\frac{am}{r^2+a^2}\rp)\Im\lp[\uppsi_{(k+1)}\overline{\uppsi_{(k)}}\rp]+\sign s \lp(\omega-\frac{am}{r^2+a^2}\rp)\Im\lp[\uppsi_{(k)}\overline{\uppsi_{(k)}}'\rp]\,, \\
\Im\lp[\uppsi_{(k)}\overline{\uppsi_{(k)}}'\rp]&= \sign s\lp[w\Im \lp[\uppsi_{(k+1)}\overline{\uppsi}_{(k)}\rp]+ \lp(\omega-\frac{am}{r^2+a^2}\rp)|\uppsi_{(k)}|^2\rp]\,.
\end{align*}
Hence, adding $-W|U'|$ and $-\Im V\Im(U\overline{U}')$ and using \eqref{eq:radial-ODE-potential-k-tilde}, we obtain
\begin{align*}
&(|s|-k)\frac{w'}{w}|\uppsi_{(k)}'|^2-\Im \mc{V}_{(k)}\Im\lp[\uppsi_{(k)}\overline{\uppsi_{(k)}}'\rp]\\
&\quad=(|s|-k)w'\lp\{w|\uppsi_{(k+1)}|^2-\lp(\omega-\frac{am}{r^2+a^2}\rp)\Re\lp[\uppsi_{(k+1)}\overline{\uppsi_{(k)}}\rp]\rp\} \\
&\quad\qquad+ (|s|-k)\frac{4amrw}{(r^2+a^2)}\sign s\Im\lp[\uppsi_{(k)}\overline{\uppsi_{(k)}}'\rp]\\
&\quad=(|s|-k)w'\lp\{w|\uppsi_{(k+1)}|^2-\lp(\omega-\frac{am}{r^2+a^2}\rp)\Re\lp[\uppsi_{(k+1)}\overline{\uppsi_{(k)}}\rp]\rp\} \\
&\quad\qquad+ (|s|-k)\frac{4amrw}{(r^2+a^2)}\lp\{ w\Im \lp[\uppsi_{(k+1)}\overline{\uppsi}_{(k)}\rp]+ \lp(\omega-\frac{am}{r^2+a^2}\rp)|\uppsi_{(k)}|^2\rp\}\,,
\end{align*}
which concludes the proof.
\end{proof}

We now describe some of the ubiquitous constructions for model virial currents in this paper.

\paragraph{A model $y$ current at large $r$}

For all frequency ranges, we will have a large $r$ current given, for some $\delta\in(0,1]$ and $R^*\geq \delta^{-1}$, by 
\begin{align}
y_\delta&=\lp(1-\frac{(R^*)^{\delta}}{(r^*)^{\delta}}\rp)\mathbbm{1}_{[R^*,\infty)}\Rightarrow y_\delta'=\delta\frac{(R^*)^\delta}{(r*)^{1+\delta}}\mathbbm{1}_{[R^*,\infty)}\,. \label{eq:standard-current-y-delta}
\end{align}
This current satisfies
\begin{align*}
\lp|\frac{wy_\delta}{y_\delta'}\rp|+\lp|\frac{w}{y_\delta'}\rp|&\leq 2\lp(\frac{r^*}{r}\rp)^{1+\delta}\lp(\frac{R}{R^*}\rp)^\delta\frac{1}{\delta R}\mathbbm{1}_{(R^*,\infty)} 
\leq B(M)\frac{1}{\sqrt{R^*}}\mathbbm{1}_{[R^*,\infty)}\,. \numberthis \label{eq:standard-current-y-delta-properties}
\end{align*}

\paragraph{Exponential-type $y$ currents}

Two $y$ currents we will require often are the $\hat{y}_0$ and $y_0$. For some $\delta\in(0,1]$, $R^*\geq \delta^{-1}$ and $C\geq 1$ to be fixed, these are given by
\begin{align*}
y_0(r^*)&:=\exp\lp(-C\int_{r(r^*)}^\infty \frac{dr}{\Delta} \rp)\mathbbm{1}_{[-R^*,R^*]}+\frac{y_0(-R^*)(R^*)^{1/2}}{(-r^*)^{1/2}}\mathbbm{1}_{(-\infty,-R^*)}+y_0(R^*)\lp(2-\frac{(R^*)^{\delta}}{(r^*)^{\delta}}\rp)\mathbbm{1}_{(R^*,\infty)}\,,\\
y'_0(r^*)&=\frac{y_0}{2(-r^*)}\mathbbm{1}_{(-\infty,-R^*)}+C\frac{y_0}{r^2+a^2}\mathbbm{1}_{[-R^*,R^*)}+\delta\frac{y(R^*)(R^*)^{\delta}}{(r^*)^{1+\delta}}>0\,, \numberthis \label{eq:standard-current-y0}
\end{align*}
and
\begin{align*}
\hat y_0(r^*)&:=-\exp\lp(-C(r-r_+)\rp)\mathbbm{1}_{[-R^*,R^*]}+\hat y(-R^*)\lp(2-\frac{(R^*)^{1/2}}{(-r^*)^{1/2}}\rp)\mathbbm{1}_{(-\infty,-R^*)}+\frac{\hat y_0(R^*)(R^*)^{\delta}}{(r^*)^{\delta}}\mathbbm{1}_{(R^*,\infty)}\,,\\
\hat y'_0(r^*)&=\frac{\hat y_0(-R^*)(R^*)^{1/2}}{2(-r^*)^{3/2}}\mathbbm{1}_{(-\infty,-R^*)}+C\frac{\hat y_0}{r^2+a^2}\mathbbm{1}_{[-R^*,R^*)}+\delta\frac{\hat{y}_0}{r^*}\mathbbm{1}_{(R^*,\infty)}>0\,, \numberthis \label{eq:standard-current-y0-hat}
\end{align*}
so that $y_0(\infty)=2y(R^*)$, $y_0(-\infty)=0$, $\hat y_0(\infty)=0$ and $\hat y_0(-\infty)=2\hat y_0(-R^*)$. We note that, though defined on the entire real line, these currents are clearly stronger at one of the ends ($r^*=\infty$ for $y$ and $r^*=-\infty$ for $\hat{y}$) and very weak at the other. Moreover, these satisfy
\begin{align*}
\lp|\frac{wy_0}{y_0'}\rp|&\leq 2\lp(\frac{r^*}{r}\rp)^{1+\delta}\lp(\frac{R}{R^*}\rp)^\delta\frac{1}{\delta R}\mathbbm{1}_{(R^*,\infty)} +\lp|\frac{wy_0}{y_0'}\rp|\mathbbm{1}_{(-\infty,R^*]}\\
&\leq B(M)\lp(\frac{\mathbbm{1}_{[-R^*,R^*]^c}}{\sqrt{R^*}}+\frac{\mathbbm{1}_{[-R^*,R^*]}}{C}\rp)\leq B(M) \max\{R^{-1/2},C^{-1}\}\,, \numberthis \label{eq:standard-current-y0-properties}\\
\lp|\frac{w}{y_0'}\rp|
&\leq B\lp(w+\frac{1}{R^*|y_0(-R^*)|}\rp)\mathbbm{1}_{[-R^*,R^*]^c}+B \frac{|y_0(-R^*)|^{-1}}{C}\mathbbm{1}_{[-R^*,R^*]}\leq B \frac{|y_0(-R^*)|^{-1}}{C}\leq B \exp(BCR^*)\,,
\end{align*}
and, similarly,
\begin{align*}
\lp|\frac{w\hat y_0}{\hat y_0'}\rp|\leq B(M) \max\{R^{-1/2},C^{-1}\}\,, \qquad \lp|\frac{w}{\hat y_0'}\rp|\leq  B \frac{|\hat y_0(R^*)|^{-1}}{C}\leq B \exp(BCR^*)\,, \numberthis\label{eq:standard-current-y0-hat-properties}
\end{align*}
as long as $R$ is large enough that $\sqrt{R}\geq \delta^{-1}$ and $y_0(R^*),|\hat{y}_0(R^*)|\geq 1/2$.

\paragraph{A model $f$ current near the maximum of the real potential}

Whenever we are looking to counteract the change of sign of a potential $\Re V$ at its maximum point, localized at, say $r=r_{\rm max}$, it will be convenient to consider an $f$-type current with function
\begin{align}
f_0(r)=\arctan(r-r_{\rm max})\,. \label{eq:standard-current-f0}
\end{align}
Differentiating,
\begin{align*}
f_0'=\frac{1}{(r-r_{\rm max})^2+1}\frac{\Delta}{r^2+a^2}\,, \qquad |f_0'''|\leq B\frac{w(r-M)^2}{r^4} \,,
\end{align*}
so that we have $f_0>0$ for $r\in[r_{\rm max},\infty)$, $f_0<0$ for $r\in[r_+,r_{\rm max}]$ and $f_0'> 0$ for $r^*\in(-\infty,\infty)$. Moreover,
\begin{align}
wf_0'\leq B w\,, \quad \frac{w^2f_0^2}{f_0'} \leq B w\,. \label{eq:standard-current-f0-properties}
\end{align}

\subsection{The Killing energy currents}

Let $U(r^*)$ be a sufficiently regular function. We define the frequency-localized energy currents
\begin{align}
Q^T[U]&=-\omega\Im(U'\overline{U})\,, \label{eq:micro-T-current}\\
Q^K[U]&=-(\omega-m\upomega_+)\Im(U'\overline{U})\label{eq:micro-K-current}\,,
\end{align}
which, if $U$ satisfies the ODE \eqref{eq:schrodinger-type}, have derivatives
\begin{align*}
\begin{split}
(Q^T)'[U] &=-\omega \Im(H\overline{U})+\omega \Im V |U|^2-\omega  W\Im(\overline{U}U')\,, \\
(Q^K)'[U] &=-(\omega-m\upomega_+) \Im(H\overline{U})+(\omega-m\upomega_+) \Im V |U|^2-(\omega-m\upomega_+)  W\Im(\overline{U}U')\,.
\end{split} \numberthis\label{eq:Killing-current-identities}
\end{align*}
These identities are obtained multiplying \eqref{eq:schrodinger-type} by $i\omega\overline{U}$ and $i(\omega-m\upomega_+)\overline{U}$, respectively, which are the frequency space analogues of $T\overline{U}$ and $K\overline{U}$ (see definitions in Section \ref{sec:prelims-vector-fields-operators}), and taking the real part.

For our system of transformed equations, the above identities can be rewritten:
\begin{lemma}  \label{lemma:Killing-identity-k} Fix $s\in\mathbb{Z}$ and $0\leq k\leq |s|$. Define $\uppsi_{(k)}$ by \eqref{eq:transformed-transport-separated} and recall the definitions of virial current templates in \eqref{eq:virial-current-definitions}.  Replacing $U$ by $\uppsi_{(k)}$, we have from \eqref{eq:Killing-current-identities}
\begin{align*}
(Q^T)'[\uppsi_{(k)}] &=-\omega \Im[\mathfrak{G}_{(k)}\overline{\uppsi_{(k)}}]-\omega\sum_{j=0}^{k-1}\Im[(ac_{s,k,j}^{\rm id}+imc_{s,k,j}^\Phi)\uppsi_{(j)}\overline{\uppsi_{(k)}}]\\
&\qquad -\sign s (|s|-k)\omega\lp[w'\Im[\uppsi_{(k+1)}\overline{\uppsi_{(k)}}]+\frac{4amrw}{r^2+a^2}|\uppsi_{(k)}|^2\rp]\,, \\
(Q^K)'[\uppsi_{(k)}] &=-(\omega-m\upomega_+) \Im[\mathfrak{G}_{(k)}\overline{\uppsi_{(k)}}]-(\omega-m\upomega_+)\sum_{j=0}^{k-1}\Im[(ac_{s,k,j}^{\rm id}+imc_{s,k,j}^\Phi)\uppsi_{(j)}\overline{\uppsi_{(k)}}]\\
&\qquad-\sign s (|s|-k)(\omega-m\upomega_+) \lp[w'\Im[\uppsi_{(k+1)}\overline{\uppsi_{(k)}}]+\frac{4amrw}{r^2+a^2}|\uppsi_{(k)}|^2\rp]\,. 
\end{align*}
\end{lemma}

\subsection{The Teukolsky--Starobinsky energy current}
\label{sec:wronskian-current}

In this section, we will introduce alternative energy currents which are better suited for the $s\neq 0$ case of the radial ODE~\eqref{eq:radial-ODE-u}, or equivalently for the transformed system \eqref{eq:transformed-k-separated}. For simplicity we drop all sub and superscripts other than the spin in what follows. Without loss of generality, let  $s\geq 0$ and set
\begin{align}
\begin{split}
q^W\lp[\swei{\uppsi}{\pm s}_{(0)}\rp]&:=\lp(w^{-s/2}\swei{\uppsi}{\pm s}_{(0)}\rp)'\cdot\overline{(r^2+a^2)^{s+1/2}(\mc{D}_0^\pm)^{2s}\lp((r^2+a^2)^{s-1/2}\swei{\uppsi}{\pm s}_{(0)}\rp)}\\
&\qquad \lp.-w^{-s/2}\swei{\uppsi}{\pm s}_{(0)}\cdot\lp[w^{s}(r^2+a^2)^{s+1/2}\overline{(\mc{D}_0^\pm)^{2s}\lp((r^2+a^2)^{s-1/2}\swei{\uppsi}{\pm s}_{(0)}\rp)}\rp]'\rp\} \,, 
\end{split}\label{eq:q}
\end{align} 
or, equivalently,
\begin{align*}
q^W\lp[\swei{u}{\pm s}\rp]&:=\lp(\swei{u}{\pm s}\rp)'\cdot\overline{\Delta^{s/2}(r^2+a^2)^{1/2}(\mc{D}_0^\pm)^{2s}\lp(\Delta^{s/2}(r^2+a^2)^{-1/2}\swei{u}{\pm s}\rp)}\\
&\qquad -\swei{u}{\pm s}\cdot\lp[\overline{\Delta^{s/2}(r^2+a^2)^{1/2}(\mc{D}_0^\pm)^{2s}\lp(\Delta^{s/2}(r^2+a^2)^{-1/2}\swei{u}{\pm s}\rp)}\rp]' \,.
\end{align*} 
Then, we have
\begin{align*}
\lp(q^W\lp[\swei{\uppsi}{\pm s}_{(0)}\rp]\rp)'&=\frac{\swei{\mathfrak{G}}{\pm s}_{(0)}}{w}w(r^2+a^2)^{s+1/2}\overline{(\mc{D}_0^\pm)^{2s}\lp[(r^2+a^2)^{s-1/2}\swei{\uppsi}{\pm s}_{(0)}\rp]}\\
&\qquad- \swei{\uppsi}{\pm s}_{(0)}w(r^2+a^2)^{s+1/2}\overline{(\mc{D}_0^\pm)^{2s}\lp[(r^2+a^2)^{s-1/2}\frac{\swei{\mathfrak{G}}{\pm s}_{(0)}}{w}\rp]} \,.
\end{align*}

\begin{remark} These definitions are motivated by the fact that, for any two solutions of the homogeneous radial ODE \eqref{eq:radial-ODE-u} of spin $\pm s$ and $\omega\in\mathbb{R}$,  $\swei{u}{+s}$ and $\swei{u}{-s}$, respectively, the Wronskian 
\begin{equation}\label{eq:wronskian-current}
W\lp(\swei{u}{+s},\overline{\swei{u}{-s}}\rp)=\lp(\swei{u}{+s}\rp)'\cdot\overline{\swei{u}{-s}}-\swei{u}{+s}\cdot\lp(\overline{\swei{u}{-s}}\rp)'
\end{equation}
is conserved in $r^*$, since $\swei{V}{s}=\overline{\swei{V}{+s}}$. As the Teukolsky--Starobinsky identities provide a way of generating solutions of spin $+s$ from solutions of spin $-s$ and vice-versa, to obtain an identity for one of the sign of spin, we fix the second solution of the pair to be obtained via these identities. E.g.\, we may impose
\begin{align*}
\swei{R}{-s}:=\Delta^s\lp(\mc{D}_0^+\rp)^{2s}\lp(\Delta^s\swei{R}{+s}\rp)\,, \qquad \swei{H}{-s}:=\Delta^{-s/2}(r^2+a^2)^{1/2}\lp(\mc{D}_0^+\rp)^{2s}\lp(\Delta^{s/2}(r^2+a^2)^{-1/2} \swei{H}{+s}\rp)\,,
\end{align*}
to  obtain the identity for spin $+ s$, and, to obtain the identity for spin $- s$, we impose
\begin{align*}
\swei{R}{+s}:=\lp(\mc{D}_0^-\rp)^{2s}\swei{R}{-s}\,, \qquad \swei{H}{+s}:=\Delta^{-s/2}(r^2+a^2)^{1/2}\lp(\mc{D}_0^-\rp)^{2s}\lp(\Delta^{s/2}(r^2+a^2)^{-1/2} \swei{H}{-s}\rp)\,.
\end{align*}
The relations $\uppsi_{(0)}=w^{|s|/2}u$ and $\mathfrak{G}_{(0)}=w^{|s|/2}H$  follow from the definitions of $\uppsi_{(0)}$, $u$, $\mathfrak{G}_{(0)}$ and $H$.
\end{remark}

The following lemma is crucial to the application of the Teukolsky--Starobinsky energy current.

\begin{lemma}[Teukolsky--Starobinsky energy currents]\label{lemma:wronskian-energy-currents} Fix $M>0$, $|a|\leq M$ and $s\in\mathbb{Z}_{\geq 0}$. Let $(\omega,m,\Lambda)$ be an admissible frequency triple with respect to $s$ and $a$ with $\omega\in\mathbb{R}\backslash\{0,m\upomega_+\}$ and define $\mathfrak C_s=\mathfrak C_s(\omega,m,\Lambda)$ by Proposition~\ref{prop:TS-angular-radial-constant}. Dropping most sub and superscripts, let $\swei{\upalpha}{\pm s}$ be solutions to \eqref{eq:radial-ODE-alpha} given by \eqref{eq:alpha-general-asymptotics} for some complex $\swei{a}{\pm s}_{\mc{H}^\pm}$, $\swei{a}{\pm s}_{\mc{I}^\pm}$. Define $\swei{\uppsi}{\pm s}_{(0)}$ via \eqref{eq:def-psi0-separated} and $\swei{A}{\pm s}_{\mc{H}^\pm}$, $\swei{A}{\pm s}_{\mc{I}^\pm}$ via \eqref{eq:a-to-A-infinity}, and either \eqref{eq:a-to-A-horizon-extremal}, if $|a|=M$, or \eqref{eq:a-to-A-horizon-sub} otherwise. Define the $T$-type and $K$-type Teukolsky--Starobinsky energy currents, respectively, by 
\begin{align*}
Q^{W,\,T}_\pm\lp[\swei{\uppsi}{\pm s}_{(0)}\rp]&:=\mp \frac12 \omega\Im\lp(i^{2s}q^W_\pm\lp[\swei{\uppsi}{\pm s}_{(0)}\rp]\rp)\,, \quad
Q^{W,\,K}_\pm\lp[\swei{\uppsi}{\pm s}_{(0)}\rp]&:=\mp \frac12(\omega-m\upomega_+)\Im\lp( i^{2s}q^W_\pm\lp[\swei{\uppsi}{\pm s}_{(0)}\rp]\rp)\,. 
\end{align*}

The boundary terms for these currents are as follows.
\begin{itemize}
\item For any $|a|\leq M$, we have
\begin{align}\label{eq:bdry-wronskian-T-K-infinity}
\begin{split}
Q^{W,\,T}_\pm (+\infty)&=\mp \omega^2\lp[(2\omega)^{2s}\lp|\swei{a}{\pm s}_{\mc{I}^\pm }\rp|^2-\frac{\mathfrak{C}_s}{(2\omega)^{2s}}\lp|\swei{a}{\pm s}_{\mc{I}^\mp}\rp|^2\rp]=\mp\omega^2\lp[\lp|\swei{A}{\pm s}_{\mc{I}^\pm }\rp|^2-\frac{\mathfrak{C}_s}{\mathfrak{D}_{s}^{\mc I}}\lp|\swei{A}{\pm s}_{\mc{I}^\mp}\rp|^2\rp]\,,\\
Q^{W,\,K}_\pm (+\infty)&=\mp\omega(\omega-m\upomega_+)\lp[(2\omega)^{2s}\lp|\swei{a}{\pm s}_{\mc{I}^\pm }\rp|^2-\frac{\mathfrak{C}_s}{(2\omega)^{2s}}\lp|\swei{a}{\pm s}_{\mc{I}^\mp}\rp|^2\rp]\\
&=\mp\omega(\omega-m\upomega_+)\lp[\lp|\swei{A}{\pm s}_{\mc{I}^\pm }\rp|^2-\frac{\mathfrak{C}_s}{\mathfrak{D}_{s}^{\mc I}}\lp|\swei{A}{\pm s}_{\mc{I}^\mp}\rp|^2\rp]\,.
\end{split}
\end{align}
\item If $|a|=M$ or if $s=0$, we have
\begin{align}\label{eq:bdry-wronskian-T-K-horizon-extremal-integer}
\begin{split}
-Q^{W,\,T}_\pm(-\infty)&=\mp\omega(\omega-m\upomega_+)2M^2\lp[\frac{\mathfrak C_s}{[4M^2(\omega-m\upomega_+)]^{2s}}\lp|\swei{a}{\pm s}_{\mc{H}^\pm}\rp|^2-[4M^2(\omega-m\upomega_+)]^{2s}\lp|\swei{a}{\pm s}_{\mc{H}^\mp}\rp|^2\rp]\\
&=\mp \omega(\omega-m\upomega_+)\lp[\frac{\mathfrak C_s}{\mathfrak{D}_s^{\mc{H}}}\lp|\swei{A}{\pm s}_{\mc{H}^\pm}\rp|^2-\lp|\swei{A}{\pm s}_{\mc{H}^\mp}\rp|^2\rp]\,,\\
-Q^{W,\,K}_\pm(-\infty)&=\mp(\omega-m\upomega_+)^22M^2\lp[\frac{\mathfrak C_s}{[4M^2(\omega-m\upomega_+)]^{2s}}\lp|\swei{a}{\pm s}_{\mc{H}^\pm}\rp|^2-[4M^2(\omega-m\upomega_+)]^{2s}\lp|\swei{a}{\pm s}_{\mc{H}^\mp}\rp|^2\rp]\\
&=\mp(\omega-m\upomega_+)^2\lp[\frac{\mathfrak C_s}{\mathfrak{D}_s^{\mc{H}}}\lp|\swei{A}{\pm s}_{\mc{H}^\pm}\rp|^2-\lp|\swei{A}{\pm s}_{\mc{H}^\mp}\rp|^2\rp]\,.
\end{split}
\end{align}

\item If $|a|<M$, then we have
\begin{align*}
-Q^{W,\,T}_\pm(-\infty)&=\mp\omega(\omega-m\upomega_+)2Mr_+\lp[\frac{\mathfrak{C}_s}{\mathfrak{C}_s^{(9)}}(r_+-r_-)^{s(1\pm 1)}\lp|\swei{a}{\pm s}_{\mc{H}^\pm}\rp|^2-\mathfrak{C}_s^{(10)}(r_+-r_-)^{-s(1\mp 1)}\lp|\swei{a}{\pm s}_{\mc{H}^\mp}\rp|^2\rp]\\
&=\mp \omega(\omega-m\upomega_+)\lp[\frac{\mathfrak C_s}{\mathfrak{D}_s^{\mc{H}}}\lp|\swei{A}{\pm s}_{\mc{H}^\pm}\rp|^2-\lp|\swei{A}{\pm s}_{\mc{H}^\mp}\rp|^2\rp]\,,\numberthis \label{eq:bdry-wronskian-T-K-horizon-sub-integer}\\
Q^{W,\,K}_\pm(-\infty)&=\mp(\omega-m\upomega_+)^22Mr_+\lp[\frac{\mathfrak{C}_s}{\mathfrak{C}_s^{(9)}}(r_+-r_-)^{s(1\pm 1)}\lp|\swei{a}{\pm s}_{\mc{H}^\pm}\rp|^2-\mathfrak{C}_s^{(10)}(r_+-r_-)^{-s(1\mp 1)}\lp|\swei{a}{\pm s}_{\mc{H}^\mp}\rp|^2\rp] \\
&=-(\omega-m\upomega_+)^2\lp[\frac{\mathfrak C_s}{\mathfrak{D}_s^{\mc{H}}}\lp|\swei{A}{\pm s}_{\mc{H}^\pm}\rp|^2-\lp|\swei{A}{\pm s}_{\mc{H}^\mp}\rp|^2\rp]\,,
\end{align*}
where we have
\begin{align*}
\mathfrak{C}_s^{(9)}&=\prod_{j=1}^{|s|}\lp\{\lp[4Mr_+(\omega-m\upomega_+)\rp]^2+(s-j)^2(r_+-r_-)^2\rp\}\,, \\ \mathfrak{C}_s^{(10)}&=\prod_{j=0}^{|s|-1}\lp\{\lp[4Mr_+(\omega-m\upomega_+)\rp]^2+(s-j)^2(r_+-r_-)^2\rp\}\,.
\end{align*}
\end{itemize}

Moreover, for a smooth $\chi(r^*)$ of compact support in $r^*$ in a region such that $r\geq L\gg 1$ and with size comparable to $L$, one has
\begin{align}
&\int_{-\infty}^\infty \chi' Q^{W,\, \{T,K\}}_{\pm}-\int_{-\infty}^\infty \chi' \{\omega,\omega-m\upomega_+\} \Im\lp[\Psi'\overline{\Psi}\rp] \label{eq:local-errors-Wronskian} \\
&\quad \leq \frac{B\{|\omega|,|\omega-m\upomega_+|\}}{L}\int_{\supp(\chi')} \lp(1+\frac{|am|}{L}\rp)\sum_{k=0}^{|s|}\sum_{j=0}^{|s|-1}L^{(|s|-k)+(|s|-j)}\lp[\lp|\uppsi_{(k)}'\rp|\lp|\uppsi_{(j)}\rp|+L^{-1}\lp|\uppsi_{(k)}\rp|\lp|\uppsi_{(j)}\rp|\rp], \nonumber
\end{align}
where the alternatives in braces correspond to the $T$-type and $K$-type Teukolsky--Starobinsky currents as defined above.
\end{lemma}

\begin{proof}
We begin by computing the boundary values of the Teukolsky--Starobinsky currents. Consider the representation for $\swei{\upalpha}{\pm s}$ in \eqref{eq:alpha-general-asymptotics}, and let $\swei{u}{\pm s}=(r^2+a^2)^{1/2}\Delta^{\pm s/2}\swei{\upalpha}{\pm s}$ as usual. Since
\begin{align*}
W\lp(\swei{u}{+s},\overline{\swei{u}{-s}}\rp) &= W\lp((r^2+a^2)^{1/2}\Delta^{s/2}\swei{\upalpha}{+s},(r^2+a^2)^{1/2}\Delta^{-s/2}\swei{\upalpha}{-s}\rp)
\end{align*}
we obtain, as $r\to r_+$,
\begin{align*}
W\lp(\swei{u}{+s},\overline{\swei{u}{-s}}\rp)&= \swei{a}{+s}_{\mc{H}^+}\overline{\swei{a}{-s}_{\mc{H}^+}}\cdot W\lp((r^2+a^2)^{1/2}\Delta^{-s/2}\Delta^s\swei{\upalpha}{+s}_{\mc{H}^+},(r^2+a^2)^{1/2}\Delta^{-s/2}\overline{\swei{\upalpha}{-s}_{\mc{H}^+}}\rp)\\
&\qquad +\swei{a}{+s}_{\mc{H}^-}\overline{\swei{a}{-s}_{\mc{H}^-}}\cdot W\lp((r^2+a^2)^{1/2}\Delta^{s/2}\swei{\upalpha}{+s}_{\mc{H}^-},(r^2+a^2)^{1/2}\Delta^{s/2}\Delta^{-s}\overline{\swei{\upalpha}{-s}_{\mc{H}^-}}\rp)\\
&= \swei{a}{+s}_{\mc{H}^+}\overline{\swei{a}{-s}_{\mc{H}^+}}\lp[\Delta\frac{d}{dr}(\Delta^s\swei{\upalpha}{+s}_{\mc{H}^+})\Delta^{-s}\overline{\swei{\upalpha}{-s}_{\mc{H}^+}}-\Delta\frac{d}{dr}\lp(\overline{\swei{\upalpha}{-s}_{\mc{H}^+}}\rp)\swei{\upalpha}{+s}_{\mc{H}^+}\rp]\\
&\qquad+\swei{a}{+s}_{\mc{H}^-}\overline{\swei{a}{-s}_{\mc{H}^-}}\lp[\Delta\frac{d}{dr}(\swei{\upalpha}{+s}_{\mc{H}^-})\overline{\swei{\upalpha}{-s}_{\mc{H}^-}}-\Delta\frac{d}{dr}\lp(\Delta^{-s}\overline{\swei{\upalpha}{-s}_{\mc{H}^-}}\rp)\Delta^s\swei{\upalpha}{+s}_{\mc{H}^-}\rp]\\
&=\lp\{\begin{array}{lr}
-(-2\xi+s)(r_+-r_-) &\text{~if~} |a|<M\\
-2\beta &\text{~if~} |a|=M
\end{array}\rp\}\lp(\swei{a}{+s}_{\mc{H}^+}\overline{\swei{a}{-s}_{\mc{H}^+}}-\swei{a}{+s}_{\mc{H}^-}\overline{\swei{a}{-s}_{\mc{H}^-}}\rp)+O(r-r_+)\,,
\end{align*}
and, as $r\to \infty$,
\begin{align*}
W\lp(\swei{u}{+s},\overline{\swei{u}{-s}}\rp)&=2i\omega\lp(\swei{a}{+s}_{\mc{I}^+}\overline{\swei{a}{-s}_{\mc{I}^+}}-\swei{a}{+s}_{\mc{I}^-}\overline{\swei{a}{-s}_{\mc{I}^-}}\rp)+O(r^{-1})\,,
\end{align*}
where we have used the asymptotic expansions for $\swei{\upalpha}{+s }_{\mc{H}^\pm}$, $\swei{\upalpha}{+s }_{\mc{I}^\pm}$ and for their spin $-s$ counterparts from Lemma~\ref{lemma:aymptotic-analysis-radialODE-Teukolsky}. 

To transform the expressions above into expression in terms of one sign of spin only, we apply the correspondence, given by  the radial Teukolsky--Starobinsky identities \eqref{eq:TS-radial-general} in Proposition~\ref{prop:TS-radial}, between $\swei{a}{+s}_{\mc{H}^\pm}$ and $\swei{\upalpha}{+s }_{\mc{I}^\pm}$ and their spin $-s$ counterparts.   For instance, if we are looking to obtain the expression of the boundary term at $r=\infty$ in terms of spin $-s$, we assume $\swei{\upalpha}{+s}$ is obtained by \eqref{eq:TS-radial-general} (see also \cite[Proposition 2.21]{TeixeiradaCosta2019}). For integer spin, we use \eqref{eq:a-to-A-infinity}, and either \eqref{eq:a-to-A-horizon-extremal} if $|a|=M$, or \eqref{eq:a-to-A-horizon-sub} otherwise, introduced in Lemma~\ref{lemma:uppsi-general-asymptotics}, to simplify. Thus, as $r\to \infty$,
\begin{align*}
W(-\infty)&= 2i\omega\lp(\swei{a}{+s}_{\mc{I}^+}\overline{\swei{a}{-s}_{\mc{I}^+}}-\swei{a}{+s}_{\mc{I}^-}\overline{\swei{a}{-s}_{\mc{I}^-}}\rp)
= \pm 2i\omega\lp[\frac{\mathfrak{C}_s}{(2i\omega)^{2s}}\lp|\swei{a}{\mp s}_{\mc{I}^\pm }\rp|^2-(-2i\omega)^{2s}\lp|\swei{a}{\mp s}_{\mc{I}^\mp }\rp|^2\rp]\\
&= \pm (2i\omega)^{1-2s}\lp[\mathfrak{C}_s\lp|\swei{a}{\mp s}_{\mc{I}^\pm }\rp|^2-(2\omega)^{4s}\lp|\swei{a}{\mp s}_{\mc{I}^\mp }\rp|^2\rp]= \pm 2i\omega\lp[\frac{\mathfrak{C}_s}{\mathfrak{D}_s^{\mc{I}}}\lp|\swei{A}{\mp s}_{\mc{I}^\pm }\rp|^2-(2\omega)^{4s}\lp|\swei{A}{\mp s}_{\mc{I}^\mp }\rp|^2\rp]\,.
\end{align*}
If $|a|=M$, then, as $r\to r_+$, 
\begin{align*}
-W(-\infty)&= 2\beta\lp(\swei{a}{+s}_{\mc{H}^+}\overline{\swei{a}{-s}_{\mc{H}^+}}-\swei{a}{+s}_{\mc{H}^-}\overline{\swei{a}{-s}_{\mc{H}^-}}\rp)
= \pm 2\beta\lp[\frac{\mathfrak{C}_s}{(2\beta)^{2s}}\lp|\swei{a}{\mp s}_{\mc{H}^\pm }\rp|^2-(-2\beta)^{2s}\lp|\swei{a}{\mp s}_{\mc{H}^\mp }\rp|^2\rp]\\
&= \pm \lp[2i(\omega-m\upomega_+)\rp]^{1-2s}(2M^2)^{1-2s}\lp[\mathfrak{C}_s\lp|\swei{a}{\mp s}_{\mc{H}^\pm }\rp|^2-\lp[4M^2(\omega-m\upomega_+)\rp]^{4s}\lp|\swei{a}{\mp s}_{\mc{H}^\mp }\rp|^2\rp]\\
&= \pm 2i(\omega-m\upomega_+)\lp[\frac{\mathfrak{C}_s}{\mathfrak{D}_s^{\mc{H}}}\lp|\swei{A}{\mp s}_{\mc{H}^\pm }\rp|^2-(2\omega)^{4s}\lp|\swei{A}{\mp s}_{\mc{H}^\mp }\rp|^2\rp]\,.
\end{align*}
Finally, if $|a|<M$, since
\begin{align*}
\mathfrak{C}_s^{(2)}:=\prod_{j=0}^{2s-1}(2\xi+s-j)=\begin{dcases}
-2\xi(2\xi+s)(-1)^{(s-1)}\prod_{j=1}^{s-1}\lp[4|\xi|^2+(s-j)^2\rp] &\text{~if $s$ integer}\\
(2\xi+s)(-1)^{s-1/2}\prod_{j=1}^{s-1/2}\lp[4|\xi|^2+(s-j)^2\rp] &\text{~if $s$ half-integer}
\end{dcases}\,,
\end{align*}
we conclude, for $s$ integer,
\begin{align*}
-W(-\infty)&= \mp 2\xi(r_+-r_-)\frac{-2\xi+s}{-2\xi}\lp[\mathfrak{C}_s^{(2)}(r_+-r_-)^{2s-s(1\pm 1)}\lp|\swei{a}{\mp s}_{\mc{H}^\mp}\rp|^2-\frac{\mathfrak{C}_s(r_+-r_-)^{s(1\mp 1)}}{\overline{\mathfrak{C}_s^{(2)}}(r_+-r_-)^{2s}}\lp|\swei{a}{\mp s}_{\mc{H}^\pm}\rp|^2\rp]\\
&=\pm 2i(\omega-m\upomega_+)i^{-2s}(2Mr_+)\\
&\qquad\times\lp\{(2\xi+s)(-2\xi+s)\prod_{j=1}^{s-1}\lp[|4\xi|^2+(s-j)^2\rp](r_+-r_-)^{2s-s(1\pm 1)}\lp|\swei{a}{\mp s}_{\mc{H}^\mp}\rp|^2\rp.\\
&\qquad\qquad \lp. -\frac{\mathfrak{C}_s(r_+-r_-)^{s(1\mp 1)}}{(2\xi)(-2\xi)\prod_{j=1}^{s-1}\lp[|4\xi|^2+(s-j)^2\rp]}(r_+-r_-)^{2s}\lp|\swei{a}{\mp s}_{\mc{H}^\pm}\rp|^2\rp\}\\
&=\pm 2i(\omega-m\upomega_+)i^{-2s}(2Mr_+)\lp\{\mathfrak{C}_s^{(10)}(r_+-r_-)^{-s(1\pm 1)}\lp|\swei{a}{\mp s}_{\mc{H}^\mp}\rp|^2-\frac{\mathfrak{C}_s(r_+-r_-)^{s(1\mp 1)}}{\mathfrak{C}_s^{(9)}}\lp|\swei{a}{\mp s}_{\mc{H}^\pm}\rp|^2\rp\}\,;
\end{align*}
for $s$ half-integer
\begin{align*}
-W(-\infty)&= \pm i^{1-2s} \lp\{\mathfrak{C}_s^{(10)}(r_+-r_-)^{-s(1\pm 1)-1}\lp|\swei{a}{\mp s}_{\mc{H}^\mp}\rp|^2-\frac{\mathfrak{C}_s(r_+-r_-)^{s(1\mp 1)-1}}{\mathfrak{C}_s^{(9)}}\lp|\swei{a}{\mp s}_{\mc{H}^\pm}\rp|^2\rp\}\,.
\end{align*}
Here, we have $\mathfrak{C}_s^{(9)}=\mathfrak{C}_s^{(10)}=1$ when $s=0$, $\mathfrak{C}_s^{(9)}=1$ and $\mathfrak{C}_s^{(10)}=\lp[4Mr_+(\omega-m\upomega_+)\rp]^2+(r_+-r_-)^2/4$ if $s=\pm 1/2$, and
\begin{align*}
\mathfrak{C}_s^{(9)}&=
\begin{dcases}
\displaystyle\prod_{j=1}^{|s|} \lp\{\lp[4Mr_+(\omega-m\upomega_+)\rp]^2+(s-j)^2(r_+-r_-)^2\rp\} &\text{if~} s\in\mathbb{Z}\\
\displaystyle\prod_{j=1}^{|s|-1/2} \lp\{\lp[4Mr_+(\omega-m\upomega_+)\rp]^2+(s-j)^2(r_+-r_-)^2\rp\} &\text{if~} s\in\lp(\frac12\mathbb{Z}\rp)\backslash\mathbb{Z}
\end{dcases}\,,\\
\mathfrak{C}_s^{(10)}&=
\begin{dcases}
\displaystyle\prod_{j=0}^{|s|-1} \lp\{\lp[4Mr_+(\omega-m\upomega_+)\rp]^2+(s-j)^2(r_+-r_-)^2\rp\} &\text{if~} s\in\mathbb{Z}\\
\displaystyle\prod_{j=0}^{|s|-1/2} \lp\{\lp[4Mr_+(\omega-m\upomega_+)\rp]^2+(s-j)^2(r_+-r_-)^2\rp\} &\text{if~} s\in\lp(\frac12\mathbb{Z}\rp)\backslash\mathbb{Z}
\end{dcases}\,.
\end{align*}

We now turn to the localization errors. The goal is to integrate by parts, allowing the $\hat{\mc{D}}_0^\pm$ operators in $q_W$ to be equally distributed between $\uppsi_{(0)}$ and its image under the Teukolsky--Starobinsky maps, and using the relation \eqref{eq:transformed-k-separated} to make $\uppsi_{(k)}$ appear. Indeed, taking $s>0$ without loss of generality, we have
\begin{align*}
\chi' q_W &=  \lp[2w^{s/2}\lp(w^{-s/2}\swei{\uppsi}{\pm s}_{(0)}\rp)'\chi' +\chi''\swei{\uppsi}{\pm s}_{(0)}\rp](r^2+a^2)^{s+1/2}\overline{\lp((r^2+a^2)^{-1}\frac{{\mc{L}}}{w}\rp)^{2s}\lp[(r^2+a^2)^{s-1/2}\swei{\uppsi}{\pm s}_{(0)}\rp]}\\
&= (-1)^s w(r^2+a^2)\lp[(r^2+a^2)^{-1}\frac{\mc{L}}{w}\rp]^s\lp\{\lp[2w^{s/2}\lp(w^{-s/2}\swei{\uppsi}{\pm s}_{(0)}\rp)'\chi' +\chi''\swei{\uppsi}{\pm s}_{(0)}\rp](r^2+a^2)^{s-1/2}w^{-1}\rp\}\\
&\qquad\times \overline{\lp((r^2+a^2)^{-1}\frac{{\mc{L}}}{w}\rp)^{s}\lp[(r^2+a^2)^{s-1/2}\swei{\uppsi}{\pm s}_{(0)}\rp]}\,, \numberthis \label{eq:wronskian-intermediate}
\end{align*}
where we have ignored full derivatives due tot the compact support of $\chi'$ (upon integration, they generate boundary terms that vanish). By this procedure, one obtains, for $s=1$,
\begin{align*}
&\Im\lp(\chi'q_W\rp) \\
&\quad= -2\chi'\Im\lp[\Psi'\overline{\Psi}\rp] +2r\chi'\Im\lp[\Psi' \overline{\uppsi_{(0)}}\rp]-2r\lp(\frac{\chi''}{rw}+\chi'\rp)\Im\lp[\Psi \overline{\uppsi_{(0)}}'\rp]\\
&\quad\qquad-\Im\lp\{\lp[\frac{\chi'''}{w}-\frac{2\chi'' w'}{w^2}-\lp(\frac{w'}{w^2}\rp)'\chi'-4iam\frac{r}{r^2+a^2}\chi'\rp]\lp[\Psi \overline{\uppsi_{(0)}}\rp]\rp\} -2r^2\lp(\frac{\chi''}{rw}+\chi'\rp)\Im\lp[\uppsi_{(0)}'\overline{\uppsi_{(0)}}\rp]\,,
\end{align*}
again ignoring full derivatives. More generally, if $\chi'$ is supported at large $r$, we have 
\begin{align*}
&-\Im\lp(\chi'q_W\rp) \leq 2\chi'\Im\lp[\overline{\Psi} \Psi'\rp] + \frac{B}{L}\sum_{k=0}^{|s|}\sum_{j=0}^{|s|-1}(1+|am|)L^{(|s|-k)+(|s|-j)}\lp[\lp|\uppsi_{(k)}'\rp|\lp|\uppsi_{(j)}\rp|+L^{-1}\lp|\uppsi_{(k)}\rp|\lp|\uppsi_{(j)}\rp|\rp]\,,
\end{align*}
by tracking the $r$ weights in \eqref{eq:wronskian-intermediate}.
\end{proof}


\section{Precise statement and proof of Theorem~\ref{thm:frequency-estimates-big}}
\label{sec:freq-estimates-Psi}

The goal of this section will be to prove the main result of this paper, Theorem~\ref{thm:frequency-estimates-big}. 

\subsection{Precise statements of Theorem~\ref{thm:frequency-estimates-big} and outline of the section}

In this section, we prove Theorem~\ref{thm:frequency-estimates-big}A, which can be stated more precisely as follows.
\begin{theorem}\label{thm:ode-estimates-Psi-A} Fix $s\in\{0,\pm 1,\pm 2\}$, $M>0$ and $a_0\in[0,M)$. There are parameters $R^*>0$, $\beta_1(a_0)>0$, $\beta_2(a_0)>0$, $\beta_4>0$, $\omega_{\rm low}>0$, $\omega_{\rm high}>0$, $\varepsilon_{\rm width}>0$, $E>0$ and $E_W>0$ such that the following results hold. Consider Kerr black hole parameters $(a,M)$ satisfying $|a|\leq M$ and let $(\omega,m,\Lambda)$ be an admissible frequency triple with respect to $s$ and $a$. Suppose $\sml{\upalpha}{s}$ is a smooth solution to the \myuline{inhomogeneous} radial ODE~\eqref{eq:radial-ODE-alpha} and, for $k\in\{0,...,|s\}$, define $\smlambda{\lp(\uppsi_{(k)}\rp)}{s}$  via \eqref{eq:def-psi0-separated} and \eqref{eq:transformed-transport-separated}; $\smlambda{\lp(\uppsi_{(k)}\rp)}{s}$ itself satisfies the radial ODE~\eqref{eq:transformed-k-separated}. Assume $\sml{\upalpha}{s}$, and hence $\smlambda{\lp(\uppsi_{(k)}\rp)}{s}$, have outgoing boundary conditions (see Definitions~\ref{def:outgoing-bdry-teukolsky} and \ref{def:outgoing-bdry-uppsi}).

Then, for any admissible $(\omega,m,\Lambda)$ which, if $a\in(a_0,M]$, is \myuline{not} contained in the set 
\begin{align*}
&\lp\{(\omega,m,\Lambda)\text{~admissible~}\colon~ \lp(|\omega|\geq \omega_{\rm high}\,,\,\,\varepsilon_{\rm width}\omega^2\leq \Lambda\leq \varepsilon_{\rm width}^{-1}\omega^2\rp) \text{~or~} \lp(\Lambda\geq \varepsilon_{\rm width}^{-1}\max\{\omega^2,\omega_{\rm high}^2\}\rp)\,,\,\,\rp.\\
&\hphantom{\lp\{(\omega,m,\Lambda)\text{~admissible~}\colon~\rp\}~~}\lp.  m^2\upomega_+-\beta_2\Lambda\leq m\omega \leq m^2\upomega_++\beta_1\Lambda\,, \,\, m\omega>0\rp\}\,,
\end{align*}
for all $\delta\in(0,1]$, and for any $\smlambda{\Psi}{s}:=\smlambda{\lp(\uppsi_{(|s|)}\rp)}{s}$ as described, there is a parameter $r_{\rm trap}$ and a choice of functions $y_{(k)}$, $\hat{y}_{(k)}$, $\tilde{y}$, $f$,  $h_{(k)}$, $\chi_1$ and $\chi_2$, for each $0\leq k\leq |s|$ (we drop the subscripts when $k=|s|$), depending on the frequency triple but satisfying the uniform bounds
\begin{gather*}
|y_{(k)}|+|\hat{y}_{(k)}|+|\tilde{y}|+|f|+|f'|+|h_{(k)}|+|\chi_1|+|\chi_2|+|\chi_3|+|\chi_4|+|r_{\rm trap}|+|r_{\rm trap}-r_+|^{-1}\leq B\,,\\
 (|y'_{(k)}|+|\hat{y}_{(k)}'|+|\tilde y'|+|f'|+|h_{(k)}|)(r^{1+\delta}+|r^*|^{3/2})\leq B\,,
\end{gather*}
such that, using the short-hand notation
\begin{align*}
&\sum_{k=0}^{|s|}\mathfrak{G}_k\cdot\lp(f,y,h,\chi\rp)\cdot\lp(\uppsi_{(k)},\uppsi_{(k)}'\rp) \numberthis \label{eq:ODE-estimates-inhomogeneity-terms}\\
&\quad = 2(y+\hat{y}+\tilde{y}+f)\Re\lp[\mathfrak{G}\overline{\Psi}'\rp]+(h+f')\Re\lp[\mathfrak{G}\overline{\Psi}\rp]-E\lp(\chi_2\omega+\chi_1(\omega-m\upomega_+)\rp)\Im\lp[\mathfrak{G}\overline{\Psi}\rp]+\\
&\quad\qquad+E_W\lp[\chi_3(Q^{W,K})'+\chi_4(Q^{W,T})'\rp]+|s|\mathbbm{1}_{\mc F_{\rm unbdd}}\sum_{k=0}^{|s|-1}\Re\lp[\mathfrak{G}_{(k)}\overline{\uppsi_{(k)}}\rp]\\
&\quad\qquad+|s|\mathbbm{1}_{\mc F_{\rm bdd}}\sum_{k=0}^{|s|-1}\lp\{ (y_{(k)}+\hat{y}_{(k)})\Re\lp[\mathfrak{G}_{(k)}\overline{\uppsi_{(k)}}'\rp]+h_{(k)}\Re\lp[\mathfrak{G}_{(k)}\overline{\uppsi_{(k)}}\rp]\rp\}\\
&\quad\qquad +|s|\mathbbm{1}_{\mc F_{\rm bdd}}\sum_{k=0}^{|s|-1}E\lp(\chi_2\omega+\chi_1(\omega-m\upomega_+)\rp)\Im\lp[\mathfrak{G}_{(k)}\overline{\uppsi}_{(k)}\rp]\,,
\end{align*}
we have that $\smlambda{\Psi}{s}$ satisfies, dropping sub and superscripts, 
\begin{gather} \label{eq:ode-estimates-outgoing}
\begin{gathered}
b(a_0)\lp[\lp(|\Psi'|^2+\omega^2\lp|\Psi\rp|^2\rp)_{r=\infty}+\lp(|\Psi'|^2+(\omega-m\upomega_+)^2\lp|\Psi\rp|^2\rp)_{r=r_+}\rp] \\
+ b(a_0,\delta)\int_{-\infty}^\infty \frac{\Delta}{r^2+a^2}\lp\{\frac{1}{r^{1+\delta}}|\Psi'|^2+\lp[(1-r_{\rm trap}r^{-1})\lp(\frac{\omega^2}{r^{1+\delta}}+\frac{\Lambda}{r^{3}}\rp)+\frac{1}{r^{3+\delta}}\rp]|\Psi|^2\rp\}dr^*\\
\leq \int_{-\infty}^\infty \sum_{k=0}^{|s|}\mathfrak{G}_k\cdot\lp(f,y,h,\chi\rp)\cdot\lp(\uppsi_{(k)},\uppsi_{(k)}'\rp)dr^* + \mathbbm{1}_{\mc F_{\rm int,1}}\sum_{k=0}^{|s|}\int_{-R^*}^{R^*}\lp(|\uppsi_{(k)}'|^2+|\uppsi_{(k)}|^2\rp)dr^*\,,
\end{gathered}
\end{gather}
if all the precise boundary term relations of Lemma~\ref{lemma:uppsi-general-asymptotics} hold; or if not then we take $E_W\equiv 0$ and add, to the right hand side of \eqref{eq:ode-estimates-outgoing}
\begin{align*}
\mathbbm{1}_{\mc F_{\rm bdd}}\lp(|\uppsi_{(0)}|^2_{r=r_+} + \sum_{k=0}^{|s|-1}\lp|\frac{\mathfrak G_{(k)}}{w}\rp|^2_{r=r_+}\rp)\,.\numberthis\label{eq:ode-estimates-outgoing-addition}
\end{align*}
\end{theorem}

\begin{remark} In proving Theorem~\ref{thm:ode-estimates-Psi-A} we will, of course, not only control bulk terms of $\swei{\Psi}{s}$ but, by the very nature of the proof, we will in fact control bulk terms of $\uppsi_{(k)}$ for every $k=0,\dots,|s|$. For $k<|s|$, these do not experience any generation at trapping, and we can furthermore control bulk terms in $\uppsi_{(k)}'$ at the cost of replacing $\mathfrak{G}_{(k)}$ by $\mathfrak{G}_{(k)}'$ in the above. For brevity, we have opted not to include this more general statement here, but it will be useful to keep it in mind in our \cite{SRTdC2022}.
\end{remark}

Analogously to the comparable result for $s=0$ in \cite{Dafermos2016b}, we may also consider a version of Theorem~\ref{thm:ode-estimates-Psi-A} which holds for fixed $m$: 
\begin{theorem}  \label{thm:ode-estimates-Psi-fixed-m}
Consider the setup of Theorem~\ref{thm:ode-estimates-Psi-A} once again. For $a\in[0,a_0]$ and for any admissible $m$ with respect to $s$, there is a choice of functions $y_{(k)}$, $\hat{y}_{(k)}$, $\tilde{y}$, $f$,  $h_{(k)}$, $\chi_1$ and $\chi_2$ for each $0\leq k\leq |s|$ and a choice of $\omega_{\rm high}=\omega_{\rm high}(m)$ such that \eqref{eq:ode-estimates-outgoing} holds with constant $b$ depending on $m$ and with either 
\begin{align*}
r_{\rm trap}=0 \text{~~or~~} (1+\sqrt{2})M<r_{\rm trap}\leq B(M,|s|)\,.
\end{align*} 
\end{theorem}
\begin{remark} Theorem~\ref{thm:ode-estimates-Psi-fixed-m} reflects the fact that, on any \textbf{subextremal} Kerr black hole, trapping occurs outside the ergoregion for fixed $m$. Thus, for solutions of the radial ODE with fixed azimuthal number, we can obtain an estimate whose bulk is nondegenerate in the ergorregion.
\end{remark}

In this section, we also prove Theorem~\ref{thm:frequency-estimates-big}B, which can be stated more precisely as follows.
\begin{theorem}\label{thm:ode-estimates-Psi-B} Fix $s\in\{0,\pm 1,\pm 2\}$, $M>0$ and $a_0\in[0,M)$. Consider Kerr black hole parameters $(a,M)$ satisfying $|a|\leq M$ and let $(\omega,m,\Lambda)$ be an admissible frequency triple with respect to $s$ and $a$. Suppose $\sml{\upalpha}{s}$ is a smooth solution to the \myuline{homogeneous} radial ODE~\eqref{eq:radial-ODE-alpha} and, for $k\in\{0,...,|s\}$, define $\smlambda{\lp(\uppsi_{(k)}\rp)}{s}$ via \eqref{eq:def-psi0-separated} and \eqref{eq:transformed-transport-separated}; $\smlambda{\lp(\uppsi_{(k)}\rp)}{s}$ itself satisfies the radial ODE~\eqref{eq:transformed-k-separated}. Assume $\sml{\upalpha}{s}$, and hence $\smlambda{\lp(\uppsi_{(k)}\rp)}{s}$, have the boundary conditions in \eqref{eq:alpha-general-asymptotics} and \eqref{eq:uppsi-general-asymptotics}.

For any admissible $(\omega, m,\Lambda)$ and for any $\smlambda{\Psi}{s}:=\smlambda{\lp(\uppsi_{(|s|)}\rp)}{s}$ as described, we have the uniform bounds, dropping sub and superscripts,
\begin{equation}\label{eq:scattering-boundary-terms} 
\begin{split}
&\omega^2
\lp\{\begin{array}{lr} \displaystyle
\frac{\mathfrak{C}_s}{\mathfrak{D}_s^{\mc I}}\,, &s\leq 0\\
1\,, &s>0
\end{array}\rp\}
\lp|\swei{A}{s}_{\mc{I}^+}\rp|^2+(\omega-m\upomega_+)^2
\lp\{\begin{array}{lr}
1\,, &s\leq 0\\
\displaystyle\frac{\mathfrak{C}_s}{\mathfrak{D}_s^{\mc H}}\,, &s>0
\end{array}\rp\}
\lp|\swei{A}{s}_{\mc{I}^+}\rp|^2
\\
&\quad\leq B(a_0,M,s)\lp[\omega^2
\lp\{\begin{array}{lr}
1\,, &s\leq 0\\
\displaystyle\frac{\mathfrak{C}_s}{\mathfrak{D}_s^{\mc I}}\,, &s>0
\end{array}\rp\}
\lp|\swei{A}{s}_{\mc{I}^-}\rp|^2+(\omega-m\upomega_+)^2
\lp\{\begin{array}{lr}\displaystyle
\frac{\mathfrak{C}_s}{\mathfrak{D}_s^{\mc H}}\,, &s\leq 0\\
1\,, &s>0
\end{array}\rp\}
\lp|\swei{A}{s}_{\mc{I}^-}\rp|^2\rp]\,.
\end{split}
\end{equation}

If $a=0$, or $s=0$, or $(\omega,m,\Lambda)$ is \myuline{not} contained in the set 
\begin{align*}
\mc{B}:=\lp\{(\omega,m,\Lambda)\text{~admissible~}\colon~ \rp. &\lp(\omega_{\rm low}\leq |\omega|\leq \omega_{\rm high}\,,\,\, \Lambda\leq \varepsilon_{\rm width}^{-1}\omega_{\rm high}^2\,,\,\, m^2\leq 2\varepsilon_{\rm width}^{-1}\omega_{\rm high}^2\rp) \text{~or~} \\
& \lp(|\omega|\geq \omega_{\rm high}\,, \,\, \omega^2\geq \varepsilon_{\rm width}^{-1}\Lambda\rp) \text{~or~} \\
&\lp.\lp(|\omega|\geq \omega_{\rm high}\,,\,\, \varepsilon_{\rm width}\omega^2\leq \Lambda\leq \varepsilon_{\rm width}^{-1}\omega\,,\,\,\Lambda-2am\omega<\beta_4\Lambda\rp)\rp\}\,,
\end{align*}
for some parameters $\omega_{\rm low}>0$, $\omega_{\rm high}>0$, $\varepsilon_{\rm width}>0$, and $\beta_4>0$, we moreover have the uniform bounds
\begin{equation} \label{eq:main-thm-scattering-bdd-coefficients}
\lp|\frac{\mathfrak{D}_s^\mc{H}}{\mathfrak{D}_s^\mc{I}}\rp|+\lp|\frac{\mathfrak{D}_s^\mc{H}}{\mathfrak{D}_s^\mc{I}}\rp|^{-1}+\lp|\frac{\mathfrak{C}_s}{\mathfrak{D}_s^\mc{I}}\rp|+\lp|\frac{\mathfrak{C}_s}{\mathfrak{D}_s^\mc{I}}\rp|^{-1}+\lp|\frac{\mathfrak{C}_s}{\mathfrak{D}_s^\mc{H}}\rp|+\lp|\frac{\mathfrak{C}_s}{\mathfrak{D}_s^\mc{H}}\rp|^{-1}\leq B(M,s)\,.
\end{equation}  
\end{theorem}


\begin{remark}[The case of integer spin parameter $|s|\geq 2$]\label{rmk:spin-restriction-ode-estimates} 
We will prove Theorems~\ref{thm:ode-estimates-Psi-A}, \ref{thm:ode-estimates-Psi-fixed-m} and \ref{thm:ode-estimates-Psi-B} by applying a combination of the virial and energy currents introduced in Section~\ref{sec:current-templates} to the radial \eqref{eq:transformed-k-separated}. These equations are defined for any $s\in\mathbb{Z}$ and, indeed, the properties of the radial ODEs \eqref{eq:transformed-k-separated} that we appeal to in the construction of the virial currents and classical energy currents are valid for any integer $s$. However, in controlling the boundary terms in estimates \eqref{eq:ode-estimates-outgoing} and \eqref{eq:scattering-boundary-terms}, we often rely on specific properties of the constants, $\mathfrak{D}_{s,k}^{\mc{I}}$ and $\mathfrak{D}_{s,k}^{\mc{H}}$ (Lemma~\ref{lemma:uppsi-general-asymptotics}), relating boundary terms of each $\uppsi_{(k)}$ and, when we add Wronskian-type currents, of the Teukolsky--Starobinsky constants, $\mathfrak{C}_s$ (Proposition~\ref{prop:TS-angular-radial-constant}). An example of such a property is the non-negativity of the Teukolsky--Starobinsky constant (otherwise one of the boundary terms for the current would not have a good sign). We only verify that the properties being used hold for $|s|\leq 2$ (see Remark~\ref{rmk:TS-constant-sign}).

Were the conditions we appeal to regarding $\mathfrak{C}_s$, $\mathfrak{D}_{s,k}^\mc{I}$ and $\mathfrak{D}_{s,k}^\mc{H}$ to hold for integer $|s|> 2$, our Theorems~\ref{thm:ode-estimates-Psi-A}, \ref{thm:ode-estimates-Psi-fixed-m} and \ref{thm:ode-estimates-Psi-B} would naturally extend to such spins. The case of half-integer spin is not within the scope of our approach to the Teukolsky equation via the transformed system introduced in Section~\ref{sec:transformed-system}.
\end{remark}

The section is organized as follows. We begin by introducing the partition of admissible frequencies that we consider in Section~\ref{sec:frequency-partitions}. The coupling of the equation for $\Psi$ to $\uppsi_{(k)}$ requires us to derive mechanisms for estimating the latter by the former; some general methods are already introduced in Section~\ref{sec:basic-estimates-transformed}. In the next two sections, Section \ref{sec:bounded-smallness} and Section~\ref{sec:intermediate}, we prove versions of Theorems~\ref{thm:ode-estimates-Psi-A} and \ref{thm:ode-estimates-Psi-B} for bounded frequencies; in the following Section~\ref{sec:unbounded}, we prove the results for unbounded frequencies. All these estimates are put together in Section~\ref{sec:combining-ode-estimates} to yield the full Theorems~\ref{thm:ode-estimates-Psi-A} and \ref{thm:ode-estimates-Psi-fixed-m} as well as Theorem~\ref{thm:ode-estimates-Psi-B}.

\subsection{Partitioning the admissible frequencies}
\label{sec:frequency-partitions}

For the remainder of the present section, for simplicity, we will consider only Kerr black holes with positive specific angular momentum, $a\geq 0$. This restriction, which enables us to state and prove each of the steps leading up to Theorems~\ref{thm:ode-estimates-Psi-A}, \ref{thm:ode-estimates-Psi-fixed-m} and \ref{thm:ode-estimates-Psi-B} more easily, is without loss of generality:
\begin{lemma}[Reduction on sign of $a$] It is enough to prove Theorems~\ref{thm:ode-estimates-Psi-A}, \ref{thm:ode-estimates-Psi-fixed-m} and \ref{thm:ode-estimates-Psi-B} under the assumption $a\geq 0$. \label{lemma:reduction-a-positive}
\end{lemma}
\begin{proof}
Recall that $(\omega,m,\Lambda)$ is an admissible frequency triple with respect to $s$ and $a$ if it satisfies the requirements in Definition~\ref{def:admissible-freqs}. In particular, the constraints on $\Lambda$ do not depend on the sign of either $a$ or $m$, but only possibly on the sign of $am$. Thus, for $k=0,...,|s|$, all terms in the radial ODEs~\ref{eq:transformed-k-separated} are invariant under the $a\mapsto -a$, $m\mapsto -m$. We conclude that, it is enough to establish Theorems~\ref{thm:ode-estimates-Psi-A}, \ref{thm:ode-estimates-Psi-fixed-m} and \ref{thm:ode-estimates-Psi-B} for $a>0$ and any $m\in\mathbb{R}$ admissible.
\end{proof}

Fix $s\in\mathbb{Z}$ and $M>0$. For some parameters $a_0,\tilde{a}_0,\tilde{a}_1\in[0,M)$,  $\varepsilon_{\rm width},\omega_{\rm high}>0$, $r_0'\in[r_+,\infty)$ and $\beta_1,\beta_2,\beta_3,\beta_4>0$. Consider the following frequency ranges (see Figure \ref{fig:freqs}):

\begin{itemize}
\leftmarginii=.33\leftmargini\relax
\leftmarginiii=.75\leftmarginii\relax

\item Bounded frequencies 
$$\mc{F}_{\rm bdd}:=\lp\{a\in[0,M]\,, (\omega,m,\Lambda) \text{~admissible~}\colon \Lambda\leq \varepsilon_{\rm width}^{-1}\omega_{\rm high}^2\,,\,\, |\omega|< \omega_{\rm high}\rp\}$$
\begin{itemize}
\item Bounded frequencies with smallness: $\mc{F}_{\rm low}:=\mc{F}_{\rm low,1}\cup \mc{F}_{\rm low,2}\cup\mc{F}_{\rm low,3}$, where we have
\begin{itemize}
\item Low frequencies and axisymmetry or slowly rotating background $\mc{F}_{\rm low,1}=\mc{F}_{\rm low,1a}\cup\mc{F}_{\rm low,1b}\cup\mc{F}_{\rm low,1c}$, with
\begin{align*}
\mc{F}_{\rm low,1a}&:=\lp\{a\in[0,\tilde{a}_0]\,,\,\, (\omega,m,\Lambda) \text{~admissible~}\colon \Lambda\leq \varepsilon_{\rm width}^{-1}\omega_{\rm high}^2\,,\,\, |\omega|\leq \omega_{\rm low}\rp\},\\
\mc{F}_{\rm low,1b}&:=\lp\{a\in[0,M-\tilde{a}_1)\,,\,\, (\omega,m,\Lambda) \text{~admissible~}\colon \Lambda\leq \varepsilon_{\rm width}^{-1}\omega_{\rm high}^2\,,\,\, |\omega|\leq \omega_{\rm low}\,,\,\, m=0\rp\},\\
\mc{F}_{\rm low,1c}&:=\lp\{a\in[M-\tilde{a}_1,M]\,,\,\, (\omega,m,\Lambda) \text{~admissible~}\colon \Lambda\leq \varepsilon_{\rm width}^{-1}\omega_{\rm high}^2\,,\,\, |\omega|\leq \omega_{\rm low}\,,\,\, m=0\rp\},
\end{align*}
\item Low frequencies outside axisymmetry for fast rotating background
$$\mc{F}_{\rm low,2}:=\lp\{a\in[0,M]\,,\,\,(\omega,m,\Lambda) \text{~admissible~}\colon \Lambda\leq \varepsilon_{\rm width}^{-1}\omega_{\rm high}^2\,,\,\, |\omega|\leq \omega_{\rm low}\,,\,\, am>\tilde{a}_0\rp\},$$
\item Frequencies near the superradiant threshold for nearly extremal background
\begin{align*}
\mc{F}_{\rm low,3}&:=\lp\{a\in[M-\tilde{a}_1,M]\,,\,\, (\omega,m,\Lambda) \text{~admissible}\colon \Lambda\leq \varepsilon_{\rm width}^{-1}\omega_{\rm high}^2\,,\,\, |\omega-m\upomega_+|\leq \omega_{\rm low}\,,\,\,  m\neq0\rp\},
\end{align*}
\end{itemize}

\item Intermediate bounded frequencies, where we have
\begin{itemize}
\item Bounded frequencies with nonzero time frequency
$$\mc{F}_{\rm int,1}:=\lp\{a\in[0,a_0)\,, (\omega,m,\Lambda) \text{~admissible~}\colon \Lambda\leq \varepsilon_{\rm width}^{-1}\omega_{\rm high}^2\,,\,\, \omega_{\rm low}<|\omega|< \omega_{\rm high}\rp\}$$
\item Bounded frequencies without smallness
$\mc{F}_{\rm int,2}:=\mc{F}_{\rm bdd}\backslash\lp(\mc{F}_{\rm low,1}\cup\mc{F}_{\rm low,2}\cup\mc{F}_{\rm low,3}\rp);$
\end{itemize}
\end{itemize}

\item Unbounded frequencies $\mc{F}_{\rm unbdd}:=\lp\{a\in[0,M]\,,\,\, (\omega,m,\Lambda) \text{~admissible}\rp\}\backslash\mc{F}_{\rm bdd}$
\begin{itemize}[noitemsep]

\item Time dominated regime:
$$\mc{F}_{\text{\ClockLogo}}:=\lp\{a\in[0,M]\,,\,\,(\omega,m,\Lambda) \text{~admissible}\colon \Lambda< \varepsilon_{\rm width}\omega^2\,,\,\, |\omega|\geq \omega_{\rm high}\rp\};$$

\item Angular dominated regime: $\mc{F}_{\measuredangle}=\mc{F}_{\measuredangle,1}\cup \mc{F}_{\measuredangle,2}\cup \mc{F}_{\measuredangle,3}$, with
\begin{align*}
\mc{F}_{\measuredangle,1}&:=\lp\{a\in[0,a_0)\,,(\omega,m,\Lambda) \text{~admissible~}\colon \Lambda\geq \varepsilon_{\rm width}^{-1}\omega^2\,,\,\, \Lambda\geq \varepsilon_{\rm width}^{-1}\omega_{\rm high}^2\,,\rp.\\
&\qquad\qquad\qquad\qquad\qquad\qquad\qquad\qquad\qquad \lp. m\omega\leq m^2\upomega_++\beta_1\Lambda\rp\},\\
\mc{F}_{\measuredangle,2}&:=\lp\{a\in(a_0,M]\,,(\omega,m,\Lambda) \text{~admissible~}\colon \Lambda\geq \varepsilon_{\rm width}^{-1}\omega^2\,,\,\, \Lambda\geq \varepsilon_{\rm width}^{-1}\omega_{\rm high}^2\,,\rp.\\
&\qquad\qquad\qquad\qquad\qquad\qquad\qquad\qquad\qquad \lp. m\omega\leq m^2\upomega_+-\beta_2\Lambda\rp\},\\
\mc{F}_{\measuredangle,3}&:=\lp\{a\in(a_0,M]\,,(\omega,m,\Lambda) \text{~admissible~}\colon \Lambda\geq \varepsilon_{\rm width}^{-1}\omega^2\,,\,\, \Lambda\geq \varepsilon_{\rm width}^{-1}\omega_{\rm high}^2\,,\rp.\\
&\qquad\qquad\qquad\qquad\qquad\qquad\qquad\qquad\qquad \lp. m^2\upomega_+-\beta_2\Lambda < m\omega< m^2\upomega_++\beta_1\Lambda\rp\},
\end{align*}

\item Large, comparable regime: 
$\mc{F}_{\rm comp}:=\lp\{(\omega,m,\Lambda) \text{~admissible}\colon \varepsilon_{\rm width}^{-1}\omega^2\leq\Lambda\leq \varepsilon_{\rm width}^{-1}\omega^2\,,\,\, |\omega|\geq \omega_{\rm high}\rp\},$
comprising
\begin{itemize}
\item Large, comparable, non-trapped and non-superradiant
\begin{align*}
\mc{F}_{\rm comp,2}&:=\{a\in[0,M]\,,\,\, (\omega,m,\Lambda) \text{~admissible}\colon \varepsilon_{\rm width}\omega^2\leq \Lambda\leq \varepsilon_{\rm width}^{-1}\omega^2\,,\,\, |\omega|\geq \omega_{\rm high}\,,\,\, r_0'=\infty\,,\\
&\qquad\qquad m\omega\notin(0, m^2\upomega_++\beta_3\Lambda\},
\end{align*}
with $\mc{F}_{\rm comp,2a}=\mc{F}_{\rm comp,2}\cap \lp\{\Lambda-2am\omega\geq \beta_4\Lambda\rp\}$ and with $\mc{F}_{\rm comp,2b}=\mc{F}_{\rm comp,2}\backslash \mc{F}_{\rm comp,2a}$, 
\item Large, comparable, non-trapped superradiant regime, $\mc{F}_{\sun}=\mc{F}_{\sun,1}\cup\mc{F}_{\sun,2}$,
where we have
\begin{itemize}
\item Large, comparable, non-trapped superradiant regime 1
\begin{align*}
\mc{F}_{\sun,1}&:=\lp\{a\in[0,a_0]\,,(\omega,m,\Lambda) \text{~admissible~}\colon \varepsilon_{\rm width}\omega^2\leq \Lambda\leq \varepsilon_{\rm width}^{-1}\omega^2\,,\,\, |\omega|\geq \omega_{\rm high} \,,\,\,\rp.\\
&\qquad\qquad \lp. m\omega\in (0,m^2\upomega_++\beta_1 \Lambda)\rp\},
\end{align*}
\item Large, comparable, non-trapped superradiant regime 2
\begin{align*}
\mc{F}_{\sun,2}&:=\lp\{a\in(a_0,M]\,,(\omega,m,\Lambda) \text{~admissible~}\colon \varepsilon_{\rm width}\omega^2\leq \Lambda\leq \varepsilon_{\rm width}^{-1}\omega^2\,,\,\, |\omega|\geq \omega_{\rm high} \,,\,\,\rp. \\
&\qquad\qquad\lp. m\omega\in (0,m^2\upomega_+-\beta_2 \Lambda)\rp\};
\end{align*}
\end{itemize}
\item Large, comparable, trapped, non-superradiant regime
\begin{align*}
\mc{F}_{\rm comp,1}&:=\{a\in[0,M]\,,\,\,(\omega,m,\Lambda) \text{~admissible}\colon \varepsilon_{\rm width}\omega^2\leq \Lambda\leq \varepsilon_{\rm width}^{-1}\omega^2\,,\,\, |\omega|\geq \omega_{\rm high}\,,\,\, r_0'<\infty\,,\\
&\qquad\qquad m\omega\notin(0, m^2\upomega_++\beta_3\Lambda)\},
\end{align*}
\item Large, comparable, trapped and marginally superradiant regime
\begin{align*}
\mc{F}_{\rm comp,3}&:=\{a\in(a_0,M]\,,\,\,(\omega,m,\Lambda) \text{~admissible}\colon \varepsilon_{\rm width}\omega^2\leq \Lambda\leq \varepsilon_{\rm width}^{-1}\omega^2\,,\,\, |\omega|\geq \omega_{\rm high}\,,\\
&\qquad\qquad m^2\upomega_+-\beta_2\Lambda< m\omega<m^2\upomega_++\beta_3\Lambda \}.
\end{align*}

\end{itemize}

\end{itemize}
\end{itemize}

\begin{figure}[htbp]
\centering
\includegraphics[scale=1]{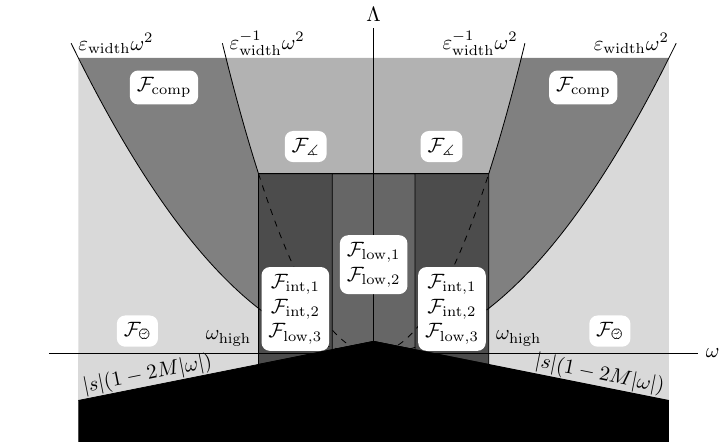}
\caption{Frequency ranges in the proof of Theorems~\ref{thm:ode-estimates-Psi-A}, \ref{thm:ode-estimates-Psi-fixed-m} and \ref{thm:ode-estimates-Psi-B}.}
\label{fig:freqs}
\end{figure}

Using the previous partition we have, for the entire Kerr black hole range $|a|\leq M$,
\begin{lemma}[Kerr black hole frequency regimes] \label{lemma:black-hole-partition}  Fix $s\in\mathbb{Z}$ and $M>0$. For some parameters $a_0,\tilde{a}_0,\tilde{a}_1\in[0,M)$,  $r_0'\in(r_+,\infty]$, $\varepsilon_{\rm width},\omega_{\rm high},\omega_{\rm low}>0$, and $\beta_1,\beta_2,\beta_3>0$ such that 
\begin{itemize}[noitemsep]
\item $\beta_3\leq \beta_1$,
\item $\omega_{\rm high}$ is sufficiently large,
\item $\varepsilon_{\rm width}^{-1}$ is sufficiently small depending on $\beta_1$,
\end{itemize}
the frequency ranges defined in the present section, i.e.\
\begin{gather*}\mc{F}_{\rm low, 1a}\,, \,\,\mc{F}_{\rm low, 1b}\,, \,\,\mc{F}_{\rm low, 1c}\,, \,\,\mc{F}_{\rm low, 2}\,, \,\,\mc{F}_{\rm low, 3}\,, \,\,\mc{F}_{\rm int,2}\,; \\
\mc{F}_{\text{\ClockLogo}}\,,\,\, \mc{F}_{\rm comp,2}\,;\,\,\mc{F}_{\measuredangle,1}\,,\,\, \mc{F}_{\sun,1}\,,\,\, \mc{F}_{\sun,2}\,;\,\,\mc{F}_{\rm comp,1}\,;\\
\mc{F}_{\measuredangle,2}\,,\,\,\mc{F}_{\measuredangle,3}\,,\,\, \mc{F}_{\rm comp,3}\,,
\end{gather*}
contains the set of all admissible frequencies with respect to $s$, i.e.\ $\cup_{0\leq a\leq M}\lp(\mc{F}_{\mathrm{ admiss}(a,M,s)}\rp)$.
\end{lemma}
\begin{proof}
We need only show that the conditions
\begin{align*}
m\omega\geq m^2\upomega_++\beta_1\Lambda\,, \qquad \Lambda\geq \varepsilon_{\rm width}^{-1}\omega^2\,, \qquad \Lambda\geq \varepsilon_{\rm width}^{-1}\omega_{\rm high}^2
\end{align*}
cannot hold simultaneously. Note that, by the two final conditions, one can show that $\Lambda\geq m^2$ (see Lemma~\ref{lemma:angular-dominated-frequency-properties}\ref{it:angular-dominated-frequency-properties-Lambda-m2}) for sufficiently large $\omega_{\rm high}$ and $\varepsilon_{\rm width}^{-1}$. We deduce
\begin{align*}
m\omega \geq m^2\upomega_+ +\beta_1\Lambda \geq \beta_1 \epsilon^{-1}_{\rm width} \omega^2 &\Rightarrow |m|\geq \beta_1 \epsilon^{-1}_{\rm width} |\omega| \text{~ and ~} |\omega|\geq \beta_1 \Lambda|m|^{-1} \geq \beta_1  |m|\\
&\Rightarrow |m|\geq \beta_1^2\epsilon_{\rm width}^{-1}|m|\,,
\end{align*}
which cannot hold if $\epsilon_{\rm width}$ is taken small enough that $\epsilon_{\rm width}< \beta_1^2$. 
\end{proof}

\begin{lemma}[Subextremal frequency ranges] Fix $s\in\mathbb{Z}$, $M>0$ and $a_0<M$. Under the same assumptions on the parameters $\tilde{a}_0,\tilde{a}_1,r_0',\varepsilon_{\rm width},\omega_{\rm high},\omega_{\rm low}$ as in Lemma~\ref{lemma:black-hole-partition}, the space of admissible frequency parameters with respect to spin $s$ and Kerr parameter $a\in[0, a_0]$, are contained in $\cup_{0\leq a\leq a_0}\lp(\mc{F}_{\mathrm{ admiss}(a,M,s)}\rp)$,  is contained in the union of all frequency ranges defined in the present section with the exception of the following four sets:
$\mc{F}_{\measuredangle,2}$, 
$\mc{F}_{\measuredangle,3}$, 
$\mc{F}_{\sun,2}$, 
$\mc{F}_{\rm comp,3}$, $\mc{F}_{\rm low,3}$, and  $\mc{F}_{\rm int,2}$.
\end{lemma}

In what follows, our strategy will be to treat the frequency parameter ranges contained in $\mc{F}_{\rm sub}$ uniformly in the full range of the specific angular momentum allowed, including the extremal case, when pertinent.

\subsection{Basic estimates from the transport equations}
\label{sec:basic-estimates-transformed}

In this section, we introduce some basic estimates we will use throughout the following sections to control the error terms arising from the application of multiplier currents due to the coupling between the radial ODE~\eqref{eq:transformed-k-separated} for $k=|s|$ and $\uppsi_{(k)}$ for $k<|s|$. These are based on the transport equations~\ref{eq:transformed-transport-separated}.

\begin{lemma} \label{lemma:basic-estimate-1}
Let $c(r^*):\mathbb{R}\to\mathbb{R}$ be a bounded, continuous and piecewise continuously differentiable function. For some  $R_-^*,R_+^*\in\mathbb{R}\cup\{\pm \infty\}$, we have
\begin{align}
\int_{R_-^*}^{R_+^*} c'\lp|\uppsi_{(k)}\rp|^2dr^* = c(R_+^*)\lp|\uppsi_{(k)}\rp|^2(R_+^*)-c(R_-^*)\lp|\uppsi_{(k)}\rp|^2(R_-^*)+\int_{R_-^*}^{R_+^*} 2\sign s\,wc\Re\lp[\overline{\uppsi_{(k)}}\uppsi_{(k+1)}\rp]dr^* \label{eq:basic-identity-1}\,,
\end{align}
and 
\begin{equation}\label{eq:basic-identity-1'}
\begin{split}
\int_{R_-^*}^{R_+^*} c'\lp|\uppsi_{(k)}'\rp|^2dr^* &= c(R_+^*)\lp|\uppsi_{(k)}'\rp|^2(R_+^*)-c(R_-^*)\lp|\uppsi_{(k)}'\rp|^2(R_-^*)+\int_{R_-^*}^{R_+^*}  2\sign s \,c\Re\lp[\overline{\uppsi_{(k)}}'\lp(w\uppsi_{(k+1)}\rp)'\rp]\\
&+\int_{R_-^*}^{R_+^*}  4am\sign s\, \frac{rw c}{r^2+a^2}\Im\lp[\overline{\uppsi_{(k)}}'\uppsi_{(k)}\rp]\,.
\end{split}
\end{equation}

Moreover, 
\begin{align}
\begin{split}
\int_{R_-^*}^{R_+^*} c'\mathbbm{1}_{\{c'>0\}}\lp|\uppsi_{(k)}\rp|^2dr^*  &\leq \int_{R_-^*}^{R_+^*} \lp(\frac{8w^2c^2}{|c'|}\mathbbm{1}_{\{c'\neq 0\}}+2|c|w\mathbbm{1}_{\{c'=0\}}\rp)\lp|\uppsi_{(k+1)}\rp|^2dr^*\\
&\qquad+\int_{R_-^*}^{R_+^*} \lp(3|c'|\mathbbm{1}_{\{c'<0\}}+2|c|w\mathbbm{1}_{\{c'=0\}}\rp)\lp|\uppsi_{(k)}\rp|^2dr^*\\
&\qquad +2c(R_+^*)\lp|\uppsi_{(k)}\rp|^2(R_+^*)-2c(R_-^*)\lp|\uppsi_{(k)}\rp|^2(R_-^*)\,. 
\end{split}\label{eq:basic-estimate-1}
\end{align}
and
\begin{align*}
\int_{R_-^*}^{R_+^*} c'\mathbbm{1}_{\{c'>0\}}\lp|\uppsi_{(k)}'\rp|^2dr^*  &\leq 12\int_{R_-^*}^{R_+^*} \frac{w^2c^2}{|c'|}\mathbbm{1}_{\{c'\neq 0\}}\lp|\uppsi_{(k+1)}'\rp|^2dr^*+\int_{R_-^*}^{R_+^*} |c'|\mathbbm{1}_{\{c'<0\}}\lp|\uppsi_{(k)}'\rp|^2dr^* \numberthis \label{eq:basic-estimate-1'}\\
&\qquad + 12\int_{R_-^*}^{R_+^*} \mathbbm{1}_{\{c'\neq 0\}} \lp\{\frac{(w')^2c^2}{|c'|}\lp|\uppsi_{(k+1)}\rp|^2+\frac{16w^2c^2r^2}{|c'|(r^2+a^2)^2}a^2m^2\lp|\uppsi_{(k)}\rp|^2\rp\}dr^*\\
&\qquad + 4\sign s\int_{R_-^*}^{R_+^*} c\mathbbm{1}_{c'=0}\lp\{\Re\lp[\overline{\uppsi_{(k)}}'\lp(w\uppsi_{(k+1)}\rp)'\rp]+  \frac{2amrw}{r^2+a^2}\Im\lp[\overline{\uppsi_{(k)}}'\uppsi_{(k)}\rp] \rp\}dr^*\\
&\qquad +2c(R_+^*)\lp|\uppsi_{(k)}\rp|^2(R_+^*)-2c(R_-^*)\lp|\uppsi_{(k)}\rp|^2(R_-^*)\,. 
\end{align*}
\end{lemma}
\begin{proof} Recall the transport equations~\eqref{eq:transformed-transport-separated}, which can be written as
\begin{align*} 
 w\uppsi_{(k+1)}=\lp(-\mr{sgn}(s)\frac{d}{dr^*}-i\omega+i\frac{am}{r^2+a^2}\rp)\uppsi_{(k)}\,.
\end{align*}
To obtain the first statement, multiply the previous identity by $c(r^*)\overline{\uppsi_{(k)}}$, take the real part and integrate in $r^*$. By an application of Cauchy--Schwarz, if $c'>0$, we have 
\begin{align*}
\int_{R_-^*}^{R^*_+} c'\lp|\uppsi_{(k)}\rp|^2dr^* \leq  2c(R_+^*)|\uppsi_{(k)}|^2(R_+^*)-2c(R_-^*)|\uppsi_{(k)}|^2(R_-^*)+\int_{R_-^*}^{R^*_+} \frac{4w^2c^2}{(c')^2}c'|\uppsi_{(k+1)}|^2\,.
\end{align*}
Hence, more generally,
\begin{align*}
&\int_{R_-^*}^{R^*_+} \lp(c'-c'\mathbbm{1}_{\{c'<0\}}\rp)\lp|\uppsi_{(k)}\rp|^2dr^*\\
&\quad=\int_{R_-^*}^{R^*_+} \lp[2(c\mathbbm{1}_{\{c'>0\}}+c\mathbbm{1}_{\{c'=0\}}+c\mathbbm{1}_{\{c'<0\}})w\Re\lp[\overline{\uppsi_{(k)}}\uppsi_{(k+1)}\rp]+ |c'|\mathbbm{1}_{\{c'<0\}}|\uppsi_{(k)}|^2\rp]dr^*\\
&\quad\qquad +\lp[c(r^*)\lp|\uppsi_{(k)}\rp|^2\rp](R_+^*)-\lp[c(r^*)\lp|\uppsi_{(k)}\rp|^2\rp](R_-^*)\\
&\quad\leq \int_{R_-^*}^{R^*_+} \lp(\frac{8w^2c^2}{|c'|}\mathbbm{1}_{\{c'\neq 0\}}+2|c|w\mathbbm{1}_{\{c'=0\}}\rp)\lp|\uppsi_{(k+1)}\rp|^2dr^*+\int_{R_-^*}^{R^*_+}\lp(3|c'|\mathbbm{1}_{\{c'<0\}}+2|c|w\mathbbm{1}_{\{c'=0\}}\rp)\lp|\uppsi_{(k)}\rp|^2dr^*\\
&\quad\qquad +\lp[c(r^*)\lp|\uppsi_{(k)}\rp|^2\rp](R_+^*)-\lp[c(r^*)\lp|\uppsi_{(k)}\rp|^2\rp](R_-^*)\,.
\end{align*}

On the other hand, we can commute the transport equations~\eqref{eq:transformed-transport-separated} with $\frac{d}{dr^*}$ to obtain
\begin{align*}
\lp(\mr sign s \frac{d}{dr^*} -i\omega+\frac{iam}{r^2+a^2}\rp)\uppsi_{(k)}'=-(w\uppsi_{(k+1)})'-\frac{2iamr}{r^2+a^2}w\uppsi_{(k)}\,.
\end{align*}
The same procedure as before yields \eqref{eq:basic-estimate-1'}.
\end{proof}

As we have discussed, the classical energy currents are not conserved for the equation for $\Psi$, due to the coupling to $\uppsi_{(k)}$ with $k<|s|$. The most troublesome errors generated, at least from the point of view of the more basic Lemma~\ref{lemma:basic-estimate-1}, are those which are quadratic in the frequency parameters, i.e.\ terms of the form
$$\omega\Im\lp[\overline{\Psi}im\uppsi_{(k)}\rp]\,, \qquad m\upomega_+\Im\lp[\overline{\Psi}im\uppsi_{(k)}\rp]\,,$$
with appropriate $r$-weights. Lemma~\ref{lemma:basic-estimate-3} below is a generalization of the trick introduced in \cite[Proposition 5.3.1]{Dafermos2017} to address the former type of error. A similar estimate can be applied to the errors quadratic in $m^2$ but will not be necessary in our analysis.

\begin{lemma} \label{lemma:basic-estimate-3} Let $c(r^*)$ be continuously differentiable function such that $c$ and $c'$ are bounded in $r^*\in\mathbb{R}$; let $0\leq \chi\leq 1$ be a bump function compactly supported in $(r_+,\infty)$. For any $\varepsilon>0$, we have
\begin{align}
\begin{split}\label{eq:basic-estimate-3T-trapping-lower}
&\int_{-\infty}^\infty wc(r)\Im\lp[i\omega \overline{\Psi}m\uppsi_{(k)}\rp]dr^* \\
&\quad\leq B\int_{-\infty}^{\infty}\lp\{w\omega^2(1-\chi)|\Psi|^2+\chi\omega^2 \lp|\uppsi_{(|s|-1)}\rp|^2+ \chi m^2\lp|\uppsi_{(k+1)}\rp|^2+(\chi+w)m^2\lp|\uppsi_{(k)}\rp|^2\rp\}dr^*\,,
\end{split}
\end{align}
if $k=0,\dots |s|-2$, and if $k=|s|-1$, 
\begin{align}
\begin{split}\label{eq:basic-estimate-3T-trapping-top}
&\int_{-\infty}^\infty wc(r)\Im\lp[i\omega \overline{\Psi}m\uppsi_{(k)}\rp]dr^* \\
&\quad\leq B\int_{-\infty}^{\infty}\lp\{\varepsilon w |\Psi'|^2+w\lp[1+\varepsilon(1-\chi)m^2\rp]|\Psi|^2+w(1+\varepsilon^{-1})m^2\lp|\uppsi_{(k)}\rp|^2\rp\}dr^*\,.
\end{split}
\end{align}
\end{lemma}

\begin{proof} 
Using the identity \eqref{eq:transformed-transport-separated} for the term $i\omega \uppsi_{(k)}$, we obtain
\begin{align*}
wc\Im\lp[i\omega \overline{\Psi}m\uppsi_{(k)}\rp] &= wcm\Im\lp[\overline{\Psi}\lp(w\uppsi_{(k+1)}+\sign s \uppsi_{(k)}'-i\frac{am}{r^2+a^2}\uppsi_{(k)}\rp)\rp]\\
&= cw^2 m \Im\lp[\overline{\Psi}\uppsi_{(k+1)}\rp]-\Im\lp[\lp(cw\frac{iam^2}{r^2+a^2}+(wc)'m\sign s\rp)\overline{\Psi}\uppsi_{(k)}\rp] \\
&\qquad-cwm \sign s \Im\lp[\overline{\Psi}'\uppsi_{(k)}\rp]+\lp(cwm\sign s \Im\lp[\overline{\Psi}\uppsi_{(k)}\rp]\rp)'\,. \numberthis\label{eq:basic-estimate-3-intermediate}
\end{align*}
After integration in $r^*$, the boundary terms  vanish due to the asymptotics of $\uppsi_{(k)}$ (see Lemma~\ref{eq:uppsi-general-asymptotics} and Definition~\ref{def:uhor-uout}). Now suppose $k=|s|-1$; we have
\begin{align*}
\begin{split} 
\Im\lp[\overline{\Psi}\uppsi_{(k+1)}\rp]&= 0\,,\\
wc\Im\lp[i\overline{\Psi}\uppsi_{(k)}\rp]&= -\frac12c\sign s  \lp(\lp|\uppsi_{(k)}\rp|^2\rp)'=-\lp(\frac12c\sign s  \lp|\uppsi_{(k)}\rp|^2\rp)'+ \frac12c'\sign s  \lp|\uppsi_{(k)}\rp|^2\,.
\end{split}
\end{align*}
For the boundary term produced here to vanish, $c$ must itself vanish as $r^*\to\pm\infty$ at a suitable rate; we use a cutoff function $\chi$ so that this is the case: since $acw(r^2+a^2)^{-1}=acw(r^2+a^2)^{-1}\lp[\chi+(1-\chi)\rp]$, we have
\begin{align*}
\begin{split}
&-\sign s\int_{-\infty}^\infty wc(r)\Im\lp[i\omega \overline{\Psi}m\uppsi_{(k)}\rp]dr^* \\
&\quad= \int_{-\infty}^\infty \lp\{wcm\Im\lp[\overline{\Psi}'\uppsi_{(k)}\rp]+(wc)'m\Im\lp[\overline{\Psi}\uppsi_{(k)}\rp]\rp\}dr^*\\
&\quad\qquad+\int_{-\infty}^\infty \lp\{\sign s \frac{a}{r^2+a^2}cw(1-\chi)m^2\Re\lp[\overline{\Psi}\uppsi_{(k)}\rp]+\frac12 \lp(c\frac{a}{r^2+a^2}\chi\rp)'m^2\lp|\uppsi_{(k)}\rp|^2\rp\}dr^*\,,
\end{split}
\end{align*}
 Then, estimate \eqref{eq:basic-estimate-3T-trapping-top} follows by Cauchy--Schwarz.

Alternatively, the identity \eqref{eq:transformed-transport-separated} can be applied to the $\Psi$ term to lower it: letting $K=|s|-1$,
\begin{align*}
wc\Im\lp[i\omega \overline{\Psi}m\uppsi_{(k)}\rp] &= c\omega\Im\lp[\lp(-\sign s \overline{\uppsi_{(K}}'+i\omega\overline{\uppsi_{(K)}} -\frac{iam}{r^2+a^2}\overline{\uppsi_{(K)}}\rp)im\uppsi_{(k)}\rp] \\
&=\lp(-\sign s c\omega \Im\lp[\overline{\uppsi_{(K)}}im\uppsi_{(k)}\rp]\rp)'+\sign s c'\omega\Im\lp[\overline{\uppsi_{(K)}}im\uppsi_{(k)}\rp]\\
&\qquad -c\omega\Im\lp[\overline{\uppsi_{(K)}}\lp(-\sign s \uppsi_{(k)}'-i\omega\uppsi_{(k)}+\frac{iam}{r^2+a^2}\uppsi_{(k)}\rp)\rp]\\
&=\lp(-\sign s c\omega \Im\lp[\overline{\uppsi_{(K)}}im\uppsi_{(k)}\rp]\rp)'+\sign s c'\omega\Im\lp[\overline{\uppsi_{(K)}}im\uppsi_{(k)}\rp]\\
&\qquad-cw\omega\Im\lp[\overline{\uppsi_{(K)}}im\uppsi_{(k+1)}\rp]\,. \numberthis\label{eq:basic-estimate-3-intermediate-2}
\end{align*}
In this identity, upon integration, the boundary term only vanishes if $c$ itself vanishes; we use a bump function $\chi$ with compact support in $(r_+,\infty)$ so that this is the case. Thus, when $k<|s|$, we consider
\begin{align*}
wc\Im\lp[i\omega \overline{\Psi}m\uppsi_{(k)}\rp] &= wc(1-\chi)\Im\lp[i\omega \overline{\Psi}m\uppsi_{(k)}\rp] + (c\chi)'\omega\Im\lp[\overline{\uppsi_{(|s|-1)}}im\uppsi_{(k)}\rp]\\
&\qquad+c\chi w\omega\Im\lp[\overline{\uppsi_{(|s|-1)}}im\uppsi_{(k+1)}\rp]+\lp(-\sign s c\chi\omega \Im\lp[\overline{\uppsi_{(K)}}im\uppsi_{(k)}\rp]\rp)'\,,
\end{align*}
to which an application of Cauchy--Schwarz yields \eqref{eq:basic-estimate-3T-trapping-lower}.
\end{proof}


\subsection{The bounded frequencies with small time frequency}
\label{sec:bounded-smallness}

In this section, we will establish Theorems~\ref{thm:ode-estimates-Psi-A} and~\ref{thm:ode-estimates-Psi-B} when the frequencies $(\omega,m,\Lambda)$ are bounded and $|\omega|$ is small. We begin, in Section~\ref{sec:bounded-smallness-properties} by stating some important properties of the frequency triple and of the transformed system potentials in these regimes. We provide a brief description of the strategy of our proof in \ref{sec:bounded-smallness-overview}, which we further expand upon at the beginning of Sections~\ref{sec:F-low-1} and \ref{sec:F-low-2}, where the proofs are carried out respectively for $\mc{F}_{\rm low, 1}$ and $\mc{F}_{\rm low,2}$. These proofs are built upon the hierarchy of estimates on the transformed system which is laid out in Sections~\ref{sec:low-hierarchy-h} and \ref{sec:low-hierarchy-y}, as well as the model virial currents described there.

\subsubsection{Properties of the frequency parameters and potential}
\label{sec:bounded-smallness-properties}

We begin by introducing the key properties of the potentials $\mc{V}_{(k)}$, $k=0,...,|s|$, in \eqref{eq:radial-ODE-potential-k-tilde} which are  essential to the construction of suitable currents in the low frequency regime $\mc{F}_{\rm low}(\omega_{\rm high},\epsilon_{\rm width})$:

\begin{lemma}[Properties of the frequency parameters in $\mc{F}_{\rm low,1}\cup \mc{F}_{\rm low,2}$]
\label{lemma:properties-frequencies-low-omega} Let $(\omega,m,\Lambda)\in\mc{F}_{\rm low,1}\cup\mc{F}_{\rm low,2}(\varepsilon_{\rm width},\omega_{\rm high},\omega_{\rm low})$. Then, for sufficiently small $\omega_{\rm low}$ depending on $M$, $\varepsilon_{\rm width}$ and $\omega_{\rm high}$, we have
\begin{enumerate}[label=(\roman*)]
\item $\Lambda+s^2\geq m^2$; \label{it:properties-frequencies-low-omega-Lambda-m2}
\item $\Lambda- 2am\omega\geq \frac{3|s|}{4}\geq 0$; \label{it:properties-frequencies-low-omega-Lambda-nondegenerate}
\item in $\mc{F}_{\rm low,1}$, $(\omega-m\upomega_+)^2\leq 4\omega_{\rm low}^2$; in $\mc{F}_{\rm low,2}$, $(\omega-m\upomega_+)^2\geq b(\tilde{a}_0)$; \label{it:properties-frequencies-low-omega-K-frequency}
\item for $|s|=0,1,2$, we have $\mathfrak{C}_s/\mathfrak{D}_s^{\mc{I}},\mathfrak{C}_s/\mathfrak{D}_s^{\mc{H}}, \mathfrak{D}_s^{\mc I}/\mathfrak{D}_s^{\mc{H}} \in\lp[\frac13,\frac53\rp]$; \label{it:properties-frequencies-low-omega-TS-constants}
\item for $s=0,1,2$ and $k<s$, we have 
$$\frac{\lp| \swei{A}{\mp s}_{k,\mc{H}^\pm}\rp|^2}{\lp| \swei{A}{\mp s}_{\mc{H}^\pm}\rp|^2}\leq B(\omega-m\upomega_+)^2\lp[(\omega-m\upomega_+)^{2(|s|-k-1)}+\frac{(r_+-M)^2}{M^2}\rp]\,,\qquad \frac{\lp| \swei{A}{\pm s}_{k,\mc{I}^\pm}\rp|^2}{\lp| \swei{A}{\pm s}_{\mc{I}^\pm}\rp|^2}\leq B\omega^{2(|s|-k)}\,.$$ \label{it:properties-frequencies-low-omega-bdry-terms}
\end{enumerate}
\end{lemma}
\begin{proof}
 From property \ref{it:admissible-freqs-triple-Lambda-lower-bound} of the admissibility conditions in Definition~\ref{def:admissible-freqs}, we have
\begin{align*}
&\Lambda+s^2-m^2\geq \max\{|m|,s^2+|s|\}-2|s||a\omega|\geq \max\{|m|,|s|\}(1-2M\omega_{\rm low})\geq 0\,,\\
&\Lambda-2am\omega\geq \max\{m^2-s^2+|m|,|s|\}-2(|s|+|m|)|a\omega| \geq \frac{3|s|}{4}+\frac14\max\{|m|,|s|\}(1-8M\omega_{\rm low})\geq \frac{3|s|}{4}\,,
\end{align*}
which proves \ref{it:properties-frequencies-low-omega-Lambda-m2} and \ref{it:properties-frequencies-low-omega-Lambda-nondegenerate}.

Property \ref{it:properties-frequencies-low-omega-K-frequency} follows from
\begin{align*}
(\omega-m\upomega_+)^2\leq (|\omega|+|m|\upomega_+)^2 \leq \lp(\omega_{\rm low}+\frac{\tilde{a}_0}{4M^2}\varepsilon_{\rm width}^{-1/2}\omega_{\rm high}\rp)^2\leq 4\omega_{\rm low}^2\,,
\end{align*}
as long as $\tilde{a}_0\leq 4M^2\omega_{\rm low}\varepsilon_{\rm width}^{1/2}\omega_{\rm high}^{-1}$.

We now turn to the statements regarding the Teukolsky constants and the boundary terms for the transformed system. By inspection of the explicit forms of the constants, in Proposition~\ref{prop:TS-angular-radial-constant} and Lemma~\ref{lemma:uppsi-general-asymptotics}, both \ref{it:properties-frequencies-low-omega-TS-constants} and \ref{it:properties-frequencies-low-omega-bdry-terms} can be derived as follows. 

 Recall \ref{it:properties-frequencies-low-omega-Lambda-m2} and \ref{it:properties-frequencies-low-omega-Lambda-nondegenerate} and the fact that, for $\omega_{\rm low}$ sufficiently small,
\begin{align*}
\Lambda-2am\omega+|s|\geq 
\begin{dcases}
m^2-|s|(|s|-1)+|m|(1-4|a\omega|)\geq m^2-|s|(|s|-1)+\frac12|m| &\text{if~} |m|> |s|\\
2|s|-2am\omega-2|s||a\omega|\geq 2|s|(1-2|a\omega|)\geq \frac{3}{2}|s|&\text{if~} |m|\leq |s|
\end{dcases}
\end{align*}
From the previous lines, it is clear that $\mc{D}_{s,1}^\mc{H}$ and $\mc{D}_{s,1}^\mc{I}$ are bounded away from 0. Moreover,
\begin{align*}
\mathfrak{D}_2^{\mc{I}} &= (\Lambda-2am\omega+2)^2(\Lambda-2am\omega+4)^2+144M^2\omega^2\\
&\qquad+12a\omega m \lp[12am\omega+2(\Lambda-2am\omega+2)(\Lambda-2am\omega+4)\rp] \\
&= (\Lambda-2am\omega+2)^2(\Lambda-2am\omega+4)^2+O(|\omega|)\geq 200\,, \numberthis \label{eq:properties-potential-low-omega-intermediate-0}\\
\mathfrak{D}_{2}^\mc{H}&=(\Lambda-2am\omega+2)^2(\Lambda-2am\omega+4)^2+9m^4\lp(\frac{a^2}{M^2}+\frac{2r_-(r_+-M)}{M^2}\rp)^2\\
&\qquad +9a^2m^2(\Lambda-2am\omega+2)^2\frac{(r_+-M)^2}{(Mr_+)^2}\lp(1+\frac{3(r_+-M)}{M}\rp)^2\\
&\qquad+2(\Lambda-2am\omega+2)^2m^2\lp[\frac{5a^2}{M^2}-\frac{6r_-(r_+-M)}{M^2}\rp]\\
&\qquad+12(\Lambda-2am\omega+2)m^2\lp[2\frac{a^2(r_+-M)}{M^2r_+}\lp(1+\frac{3(r_+-M)}{M}\rp)-\lp(\frac{a^2}{M^2}+\frac{2r_-(r_+-M)}{M^2}\rp)\rp]\\
&\qquad + 24am\omega\frac{r_+-M}{M}\lp[(\Lambda-2am\omega+2)(\Lambda-2am\omega+4)+6am\omega\frac{r_+-M}{M}\rp.\\
&\qquad\qquad\qquad\qquad\qquad\qquad\lp.-3m^2\frac{r_+-M}{M}\lp(\frac{a^2}{M^2}+\frac{2r_-(r_+-M)}{M^2}\rp)\rp]\\
&\geq (\Lambda-2am\omega+2)^2(\Lambda-2am\omega+4)^2+18(\Lambda-2am\omega+2)\frac{a^2m^2}{M^2}\lp(1+4\frac{M^2-a^2}{Mr_+}\rp)+O(|\omega|)\\
&\geq (\Lambda-2am\omega+2)^2(\Lambda-2am\omega+4)^2+O(|\omega|)\geq \frac{19}{20}(\Lambda-2am\omega+2)^2(\Lambda-2am\omega+4)^2 \geq 200\,,\numberthis \label{eq:properties-low-frequencies-intermediate-1}
\end{align*}
where, to obtain the first inequality for $\mathfrak{D}_{2}^\mc{H}$, we simply used $\Lambda-2am\omega+2\geq 3$ to combine the terms in the third and fourth lines.

If $|s|=1$,  then \ref{it:properties-frequencies-low-omega-bdry-terms} is clear from the previous considerations. Moreover, for \ref{it:properties-frequencies-low-omega-TS-constants}, we have
\begin{align*}
\lp|\frac{\mathfrak{C}_1}{\mathfrak{D}_1^{\mc{I}}}-1\rp|&= \lp|\frac{4a\omega(m-a\omega)}{(\Lambda-2am\omega+1)^2}\rp|= O(|\omega|)\,,\\
\lp|\frac{\mathfrak{C}_1}{\mathfrak{D}_1^{\mc{H}}}-1\rp|&\leq \lp|\frac{4a\omega(m-a\omega)}{(\Lambda-2am\omega+1)^2+a^2/M^2m^2}\rp|+\frac{a^2m^2/M^2}{(\Lambda-2am\omega+1)^2}\leq \frac{a^2/M^2 m^2}{(m^2+1/2|m|)^2+a^2/M^2m^2}+O(|\omega|) \\
&\leq \frac{4}{13}+O(|\omega|)\,,\\
\lp|\frac{\mathfrak{D}_1^{\mc{H}}}{\mathfrak{D}_1^{\mc{I}}}-1\rp|&\leq\frac{a^2m^2/M^2}{(\Lambda-2am\omega+1)^2}\leq \frac49\leq 2/3\,. 
\end{align*}

Now let $|s|=2$. Using the results in \eqref{eq:properties-potential-low-omega-intermediate-0}, 
\begin{align*}
\frac{\mathfrak{D}_{2,1}^{\mc{I}}}{\mathfrak{D}_2^{\mc{I}}}&\leq \frac{(\Lambda-2am\omega+2)^2}{(\Lambda-2am\omega+2)^2(\Lambda-2am\omega+4)^2+ O(|\omega|)}\leq \frac{1}{5^2+O(|\omega|)}\leq \frac{1}{20}\,,\\
\lp|\frac{\mathfrak{C}_2}{\mathfrak{D}_2^{\mc{I}}}-1\rp|&= |4a\omega|\lp|\frac{(4m-10a\omega)(\Lambda-2am\omega+2)^2+12a\omega(\Lambda-2am\omega+2)+144a^2\omega^2(a\omega-2m)}{\mathfrak{D}_2^{\mc{I}}}\rp|\\
&= O(|\omega|)\,.
\end{align*}
Turning to the boundary terms at the event horizons, we have from \eqref{eq:properties-low-frequencies-intermediate-1}
\begin{align*}
\frac{\mathfrak{D}_{2,1}^{\mc{H}}}{\mathfrak{D}_2^{\mc{H}}}&\leq\frac{\lp(\Lambda-2am\omega+2+3\frac{r_+-M}{M}\rp)^2+9a^2m^2/M^2}{\lp(\Lambda-2am\omega+2\rp)^2\lp(\Lambda-2am\omega+4\rp)^2+O(|\omega|)} \leq \frac{8}{25}\,;
\end{align*}
this concludes the proof of \ref{it:properties-frequencies-low-omega-bdry-terms} in the case $|s|=2$. For \ref{it:properties-frequencies-low-omega-TS-constants}, note $\mathfrak{C}_2=(\Lambda-2am\omega+2)^2(\Lambda-2am\omega+4)^2+O(|\omega|)$; hence we also have 
\begin{align*}
\frac{\mathfrak{C}_2}{\mathfrak{D}_2^\mc{H}}=\frac{1+O(|\omega|)}{1+\frac{\mathfrak{D}_2^\mc{H}-(\Lambda-2am\omega+2)^2(\Lambda-2am\omega+4)^2}{(\Lambda-2am\omega+2)^2(\Lambda-2am\omega+4)^2}}\,,
\end{align*}
which can be bounded once we understand the quotient in the denominator. The constant $\mathfrak{D}_{2}^\mc{H}$ satisfies \eqref{eq:properties-low-frequencies-intermediate-1}
which clearly implies
\begin{align*}
\frac{\mathfrak{C}_2}{\mathfrak{D}_2^\mc{H}}=\frac{1+O(|\omega|)}{1+\frac{\mathfrak{D}_2^\mc{H}-(\Lambda-2am\omega+2)^2(\Lambda-2am\omega+4)^2}{(\Lambda-2am\omega+2)^2(\Lambda-2am\omega+4)^2}}\leq 1+O(|\omega|)\leq \frac{5}{3}\,.
\end{align*}
On the other hand, since
\begin{align*}
&\frac{9m^4\lp(\frac{a^2}{M^2}+\frac{2r_-(r_+-M)}{M^2}\rp)^2}{(\Lambda-2am\omega+2)^2(\Lambda-2am\omega+4)^2} +\frac{9a^2m^2\frac{(r_+-M)^2}{(Mr_+)^2}\lp(1+\frac{3(r_+-M)}{M}\rp)^2+2m^2\lp(\frac{5a^2}{M^2}-\frac{6r_-(r_+-M)}{M^2}\rp)}{(\Lambda-2am\omega+4)^2}\\
&\quad\leq \frac{16}{225}\lp(\frac{a^2}{M^2}+\frac{2r_-(r_+-M)}{M^2}\rp)^2+\frac{4}{25}\lp[9a^2\frac{(r_+-M)^2}{(M^2r_+)^2}\lp(1+\frac{3(r_+-M)}{M}\rp)^2+2\lp(\frac{5a^2}{M^2}-\frac{6r_-(r_+-M)}{M^2}\rp)\rp]\\
&\quad\leq 2+\frac{2}{5}\,,
\end{align*}
where the first inequality comes from using $\Lambda-2am\omega+2\geq 3$ in the denominator and setting $|m|=2$ and the second follows from evaluating the expression at its maximum i.e.\ when $|a|=M$, and
\begin{align*}
&\frac{12m^2\lp[2\frac{a^2(r_+-M)}{M^2r_+}\lp(1+\frac{3(r_+-M)}{M}\rp)-\lp(\frac{a^2}{M^2}+\frac{2r_-(r_+-M)}{M^2}\rp)\rp]}{(\Lambda-2am\omega+2)(\Lambda-2am\omega+4)^2}\\
&\quad \leq \frac{12m^2\lp[2\frac{a^2(r_+-M)}{M^2r_+}\lp(1+\frac{3(r_+-M)}{M}\rp)-\lp(\frac{a^2}{M^2}+\frac{2r_-(r_+-M)}{M^2}\rp)\rp]}{(3)(5)^2}\mathbbm{1}_{\lp\{|a|\leq \frac15\sqrt{2(9-\sqrt{6})}M ,|m|\leq 2\rp\}}\\
&\quad\leq \frac{4}{75}\lp[2\frac{a^2(r_+-M)}{M^2r_+}\lp(1+\frac{3(r_+-M)}{M}\rp)-\lp(\frac{a^2}{M^2}+\frac{2r_-(r_+-M)}{M^2}\rp)\rp]\mathbbm{1}_{\lp\{|a|/M\leq \frac15\sqrt{2(9-\sqrt{6}})\rp\}}\leq \frac{16}{729}\leq \frac{2}{5}
\end{align*}
we have 
\begin{align*}
\frac{\mathfrak{C}_2}{\mathfrak{D}_2^\mc{H}}=\frac{1+O(|\omega|)}{1+\frac{\mathfrak{D}_2^\mc{H}-(\Lambda-2am\omega+2)^2(\Lambda-2am\omega+4)^2}{(\Lambda-2am\omega+2)^2(\Lambda-2am\omega+4)^2}}\geq \frac{1+O(|\omega|)}{1+2.8+O(|\omega|)}\geq \frac{1}{3}+O(|\omega|)\geq \frac{1}{3}\,.
\end{align*}
Similarly, as we have $\mathfrak{D}^{\mc{I}}_2=(\Lambda-2am\omega+2)^2(\Lambda-2am\omega+2)^2$ as for $\mathfrak{C}_2$, the same conclusions hold for $\mathfrak{D}^{\mc{I}}_2/\mathfrak{D}^{\mc{H}}_2$. This concludes the proof of \ref{it:properties-frequencies-low-omega-TS-constants} for $|s|=2$.
\end{proof}

\begin{lemma}[Properties of the potential for frequencies in $\mc{F}_{\rm low,1}\cup\mc{F}_{\rm low,2}$] \label{lemma:properties-potential-low-omega} Let $(\omega,m,\Lambda)\in\mc{F}_{\rm low,1}\cup\mc{F}_{\rm low,2}$ be an admissible frequency triple. For each $k=0,...,|s|$, the potentials  $\mc{V}_{(k)}$ introduced in \eqref{eq:radial-ODE-potential-k-tilde}, respectively, has the following properties. 
\begin{enumerate}[label=(\roman*)]
\item If $s\neq 0$ or if $s=0$ but $m\neq 0$ or $\Lambda\neq 0$, in a compact range of $r^*$ which can be taken to be arbitrarily large,  \label{it:properties-potential-low-omega-compact-region}
$$\frac18 w\leq \Re\mc{V}_{(k)}-\sigma(|s|-k)\lp(\frac{w'}{w}\rp)'-\omega^2\leq (\Lambda+2+3s^2)w $$
for $\sigma\in[0,4/5]$, as long as,
\begin{enumerate}[label=\alph*.]
\item if $|\omega|\leq \omega_{\rm low}$ and $a\leq \tilde{a}_0$ are small enough depending on the size of the $r^*$-region;
\item  if $m=0$ and $|\omega|\leq\omega_{\rm low}$ is small enough depending on the size of the $r^*$-region;
\item  if $m\neq 0$, $a>\tilde a_0$ , the compact range of $r^*$ is sufficiently far from $r^*=-\infty$ and $|\omega|\leq \omega_{\rm low}$ is small enough depending on the upper bound of $r^*$.
\end{enumerate}
If $s=m=\Lambda=0$, then 
$$M\Delta r^{-5} \leq \mc{V}-\omega^2 \leq 2M\Delta r^{-5} \,.$$

\item $\Re \mc{V}_{(k)}'<0$ for sufficiently large positive $r^*$. \label{it:properties-potential-low-omega-derivative-plus-infinity} Concretely, if $s\neq 0$ or if $s=0$ but $m,\Lambda\neq 0$,
\begin{align*}
\Re\mc{V}_{(k)}'=-br^{-3} +O(r^{-4})\text{~~as~}r\to\infty\,,
\end{align*}
with $7|s|/2 \leq b \leq 2\Lambda+s^2$; if $s=m=\Lambda=0$, then 
\begin{align*}
\Re\mc{V}_{(k)}'=-6 Mr^{-4} +O(r^{-5})\text{~~as~}r\to\infty\,.
\end{align*}

\item If $(\omega,m,\Lambda)\in\mc{F}_{\rm low,1}$ and $\tilde{a}_0$ is sufficiently small depending on $\varepsilon_{\rm width}$ and $\omega_{\rm high}$, we have $\Re\mc{V}_{(k)}'>0$ for sufficiently negative $r^*$.  \label{it:properties-potential-low-omega-derivative-minus-infinity} Concretely, if $m=0$,
\begin{align*}
\Re\mc{V}_{(k)}' = \frac{b}{2M^2r_+^2}\Delta(r-M)+ O(\Delta^2)\,,
\end{align*}
for $3|s|/4\leq b-1\leq \Lambda+2|s|(|s|-1)$; if $|a|\leq \tilde{a}_0$ sufficiently small, for  $3|s|/4\leq b\leq \Lambda+2|s|(|s|-1)+2$,
\begin{align*}
\Re\mc{V}_{(k)}' = \frac{b}{2M^2r_+^2}\Delta+ O(\Delta^2)\,.
\end{align*}

\item If $(\omega,m,\Lambda)\in\mc{F}_{\rm low,2}$ then $\omega^2-\mc{V}_{(k)}(r_+)=(\omega-m\upomega_+)^2\geq b(\tilde{a}_0)>0$. \label{it:properties-potential-low-omega-value-minus-infinity}
\end{enumerate}
\end{lemma}
\begin{proof}
The lemma follows by direct inspection of  \eqref{eq:radial-ODE-potential-k-tilde}. We remark that one can only have $\Lambda=0$ if $m=0$ and $s=0$ (by properties \ref{it:angular-dominated-frequency-properties-Lambda-m2} and \ref{it:angular-dominated-frequency-properties-Lambda-nondegenerate} in Lemma~\ref{lemma:properties-frequencies-low-omega}). 

For statement \ref{it:properties-potential-low-omega-compact-region}, let us first consider the case $s=0$. We have
\begin{align*}
\Re\mc{V}_{(k)} &\geq w\lp\{\Lambda+a^2w+\frac{2Mr(r^2-a^2)}{(r^2+a^2)^2}\rp\} + \frac{4Mram\omega-a^2m^2}{(r^2+a^2)^2}\\
&\geq w\Lambda +\frac{2Mr(r^2-a^2)w+4Mram\omega-a^2m^2}{(r^2+a^2)^2}\,.
\end{align*}
If $s\neq 0$, the proof follows by analyzing the frequency independent part of $\Re\mc{V}_{(k)}-\sigma (|s|-k)(w'/w)'$ for $\sigma\in[0,5/4]$.  By Lemma~\ref{lemma:properties-frequencies-low-omega}\ref{it:properties-frequencies-low-omega-Lambda-nondegenerate},
\begin{align*}
U_{(k)}&:=\frac{1}{w}\lp(\mc{V}_{(k)}-\frac{4Mram\omega-a^2m^2}{(r^2+a^2)^2}-\sigma(s-k)\lp(\frac{w'}{w}\rp)' \rp)\\
&\geq 3\frac{|s|}{4}+|s| + k (2 |s|-k-1)-2\sigma(|s|-k)\\
&\qquad+a^2w[1-2|s|-k (2 |s|-k-1)+4\sigma(|s|-k)]\\
&\qquad+\frac{2Mr(r^2-a^2)}{(r^2+a^2)^2}[1-3|s| + 2 s^2 - 3 k (2 |s| - k - 1) + 6\sigma (|s| - k)]\,.
\end{align*}
Now, we note the estimates
\begin{align*}
\frac{2Mr(r^2-a^2)}{(r^2+a^2)^2}\leq \begin{dcases}
\sqrt{1-\frac{a^2}{M^2}} &\text{if~} \lp|\frac{a}{M}\rp|\in\lp[0,\frac{1}{\sqrt{2}}\rp]\\
\frac{1}{\sqrt{2}} &\text{if~} \lp|\frac{a}{M}\rp|\in\lp[\frac{1}{\sqrt{2}},1\rp]
\end{dcases}\,, \quad
a^2 w\leq \frac{3-2\sqrt{2}}{4} \frac{a^2}{M^2}\begin{dcases}
\frac{a^2/M^2}{20} &\text{if~} \lp|\frac{a}{M}\rp|\in\lp[0,\frac{1}{\sqrt{2}}\rp]\\
\frac{1}{20} &\text{if~} \lp|\frac{a}{M}\rp|\in\lp[\frac{1}{\sqrt{2}},1\rp]
\end{dcases}\,.
\end{align*}
To obtain the desired bounds, we consider the following cases:
\begin{itemize}
\item $k\geq |s|-\frac12 +\sigma\lp(1-\sqrt{1-\frac{1}{\sigma}+\frac{s^2}{3\sigma^2}+\frac{1}{12\sigma^2}}\rp)$. In this case, both coefficients on $a^2w$ and $2Mr(r^2-a^2)(r^2+a^2)^{-2}$ are non-positive. First, suppose $|a/M|\leq 1/\sqrt{2}$; then
\begin{align*}
U_{(k)} &\geq 3\frac{|s|}{4}+|s| + 2 k (2 |s|-k-1)-2\sigma(|s|-k) +\frac{a^2/M^2}{20}[1-2|s|-2k (2 |s|-k-1)+4\sigma(|s|-k)]\\
&\qquad+\sqrt{1-a^2/M^2}[1-3|s| + 2 s^2 - 3 k (2 |s| - k - 1) + 6\sigma (|s| - k)]\\
&\geq s^2 +\frac{a^2/M^2}{20}(1-2s^2)+\sqrt{1-a^2/M^2}(1-s^2)\geq \frac{39}{40}\,,
\end{align*}
where the right hand side value is attained when $k=|s|=1$ and $|a/M|=1/\sqrt{2}$.
Now assume $|a/M|\geq 1/\sqrt{2}$; we have
\begin{align*}
U_{(k)} &\geq 3\frac{|s|}{4}+|s| + k (2 |s|-k-1)-2\sigma(|s|-k) +\frac{1}{20}[1-2|s|-2k (2 |s|-k-1)+4\sigma(|s|-k)]\\
&\qquad+\frac{1}{\sqrt{2}}[1-3|s| + 2 s^2 - 3 k (2 |s| - k - 1) + 6\sigma (|s| - k)]\\
&\geq s^2 +\frac{1}{20}(1-2s^2)+\frac{1}{\sqrt{2}}(1-s^2)\geq \frac{19}{20}\,,
\end{align*}
where the right hand side value is once again attained when $k=|s|=1$ and $|a/M|=1/\sqrt{2}$. Moreover,
\begin{align*}
U_{(k)} &\leq \Lambda+|s| + k (2 |s|-k-1)-2\sigma(|s|-k) \leq \Lambda+|s|\,.
\end{align*}

\item $k<|s|-\frac12 +\sigma \lp(1-\sqrt{1-\frac{1}{\sigma}+\frac{4s^2-1}{4\sigma^2}}\rp)$. In this case, both coefficients on $a^2w$ and $2Mr(r^2-a^2)(r^2+a^2)^{-2}$ are positive, so 
\begin{align*}
U_{(k)} &\geq \frac{3|s|}{4}+|s| + k (2 |s|-k-1)-2\sigma(|s|-k)\geq \frac{3|s|}{4}+|s|(1-2\sigma)\geq \frac{|s|}{4}\geq \frac{1}{8}+\frac{1}{40}\,;\\
U_{(k)} &\leq \Lambda+|s| + k (2 |s|-k-1)-2\sigma(|s|-k) +\frac{1}{20}[1-2|s|-2k (2 |s|-k-1)+4\sigma(|s|-k)]\\
&\qquad+[1-3|s| + 2 s^2 - 3 k (2 |s| - k - 1) + 6\sigma (|s| - k)]\\
&\leq \Lambda+2s^2+\frac{63}{50}|s|+\frac{21}{20}\,.
\end{align*}

\item  $|s|-\frac12 +\sigma \lp(1-\sqrt{1-\frac{1}{\sigma}+\frac{4s^2-1}{4\sigma^2}}\rp)\leq k < |s|-\frac12 +\sigma\lp(1-\sqrt{1-\frac{1}{\sigma}+\frac{s^2}{3\sigma^2}+\frac{1}{12\sigma^2}}\rp)$. In this case, the coefficient on $a^2w$ is negative but the one on  $2Mr(r^2-a^2)(r^2+a^2)^{-2}$ is positive, so
\begin{align*}
U_{(k)} \geq \frac{3|s|}{4}+|s| + 2k (2 |s|-k-1)-2\sigma(|s|-k)+\frac{1}{20}[1-2|s|-k (2 |s|-k-1)+4\sigma(|s|-k)]\geq \frac12\,,
\end{align*}
which is attained at the lower end of the interval. Moreover,
\begin{align*}
U_{(k)} &\leq \Lambda+|s| + k (2 |s|-k-1)-2\sigma(|s|-k) +[1-3|s| + 2 s^2 - 3 k (2 |s| - k - 1) + 6\sigma (|s| - k)]\\
&\leq \Lambda+2s^2\,.
\end{align*}
\end{itemize}

For statements \ref{it:properties-potential-low-omega-derivative-plus-infinity} and \ref{it:properties-potential-low-omega-derivative-minus-infinity}, it is enough to consider the expansions
\begin{align*}
\Re\mc{V}_{(k)}'&=-\frac{2[\Lambda+|s|-k+k(2|s|-k)]}{r^3}\\
&\qquad-\frac{6M\lp[1-2s^2-4k(2|s|-k)-4(|s|-k)-(\Lambda-2am\omega)\rp]}{r^4}+O\lp(r^{-5}\rp)\,,\quad\text{~as~}r\to\infty\,,\\
\frac{d}{dr}\Re\mc{V}_{(k)}&=\frac{2m\upomega_+(m\upomega_+-\omega)}{M}+\frac{2m(r_+-M)(m-a\omega)}{2M^2r_+^2}+(r-M)\frac{\Lambda+|s|+k(2|s|-k-1)-2m^2}{2M^2r_+^2}\\
&\qquad+(r-M)(r_+-M)\frac{\lp(1-3 |s| + 2 s^2- 3 k (2 |s|- k- 1)\rp)}{2M^3r_+^2}
+O(\Delta)\,,\quad\text{~as~}r\to r_+\,,
\end{align*}
where we note that, if $|a|<M$,
\begin{align*}
&\frac{2m(r_+-M)(m-a\omega)}{2M^2r_+^2}+(r-M)\frac{\Lambda+|s|+k(2|s|-k-1)-2m^2}{2M^2r_+^2} \\
&\quad= (r_+-M)\frac{\Lambda+|s|+k(2|s|-k-1)-2am\omega}{2M^2r_+^2} +O(r-r_+)\,,\quad\text{~as~}r\to r_+\,.
\end{align*}
Coupled with Lemma~\ref{lemma:properties-frequencies-low-omega}\ref{it:properties-frequencies-low-omega-Lambda-nondegenerate}, these identities give the results.
\end{proof}

\subsubsection{Overview of the section}
\label{sec:bounded-smallness-overview}

Recall Lemma~\ref{lemma:properties-potential-low-omega}, and let us first focus on the case $k=|s|=0$. As in \cite{Dafermos2016b}, one can exploit property 1 by constructing a suitable $h$ current of compact $r^*$ support where $\mc{V}-\omega^2$  has a good sign. Moreover, one can exploit two of the properties 3, 4 and 5 using a $y$ current which is good near $r^*=-\infty$ and one which is good near $r^*=\infty$; these should be constructed so that they absorb the errors introduced by the compact support assumption on $h$. Boundary terms generated by the $y$ currents are controlled by applications of energy currents whose possible localization errors (these appear whenever the possibility of superradiance forces us to employ localized energy currents) should be absorbed by the $h$ current in the region where it is strongest region.

Turning now to $s\neq 0$, we find that, while the strategy laid out above seems to hold from the point of view of the left hand side of the radial ODE~\eqref{eq:radial-ODE-Psi}, the coupling terms that emerge are quickly seen to be an important issue. From the point of view of the basic estimate~\eqref{eq:basic-estimate-1} from Lemma~\ref{eq:basic-estimate-1}, it seems like there can be a gain in $r$-weights as one climbs the hierarchy,  in the sense that a bulk term in $\uppsi_{(k)}$ can be controlled by boundary terms and a bulk term in $\uppsi_{(k+1)}$ with weights with stronger decay, at least at one end. However, the global, rather than pointwise, nature of such a gain means it can be insufficient for applications. More importantly, in a bounded $r$ region, unless $|a|$ is small, there is no reason for the coupling terms to be seen as ``lower order'' corrections to the left hand side of the radial ODE~\eqref{eq:radial-ODE-Psi} when $\omega$ is small. 

The upshot is that, when $s\neq 0$, we must take into account the entire transformed system and, thus, all of the radial ODEs~\ref{eq:transformed-k-separated} for $k=0,\dots,|s|$. We saw in Lemma~\ref{lemma:properties-potential-low-omega} that in fact the properties of the real part of the potential are quite similar across all the values of $k$. Our strategy, therefore, will be to apply the same type of $h$ and $y$ currents along each step of the hierarchy. For each $k$, in considering the virial currents to use, we divide the $r^*$ real line into three regions:
\begin{itemize}
\item A region where $\Re\mc{V}_{(k)}$ has good positivity properties, where we apply an $h$ current that produces little or no errors, i.e.\ where 
\begin{align*}
h''+(|s|-k)\frac{w'}{w}h'\,, \text{~~or~~}h''+(|s|-k)\frac{w'}{w}h'-\frac{(|s|-k)}{2}\lp(\frac{w'}{w}\rp)'h\,,
\end{align*}
vanish or are negligible compared to the good part of the $h$ bulk term. Whenever superradiance can occur, it is somewhere inside this region region that we will concentrate the localization errors for the energy currents, which are typically of Teukolsky--Starobinsky-type for the reasons previously discussed. Thus, it is important that $h$ is taken as large as possible in a subset of the present region.

\item A region to the right of the first one, where we take our $h$ current down to zero as a $y$ current becomes strong. The idea is to have $y$, if possible, absorb errors introduced by $h$. Concretely, when $|\omega|$ is small and $\Re\mc{V}_{(k)}'<0$ at large $r^*$ at the expected rate from Lemma~\ref{lemma:properties-potential-low-omega}, we start $y\geq 0$ when $h$ is still generating a strong, positive bulk term in a way which ensures $y'$ is a small multiple of $h$, so
\begin{align*}
h(\Re\mc{V}_{(k)}-\omega^2)-y'(\Re\mc{V}_{(k)}-\omega^2)-y\Re\mc{V}_{(k)}'\leq \frac{15}{16} h(\Re\mc{V}_{(k)}-\omega^2)-y\Re\mc{V}_{(k)}'\geq 0\,.
\end{align*}
Then, if $h$ can be taken down to zero in a controlled, integrable manner, we let $y$ grow so that absorbs the errors due to $h$ in that region, i.e.\ so that
$$-(y\Re\mc{V}_{(k)})'-\frac12\lp(h''+(|s|-k)\frac{w'}{w}h'\rp)\geq 0\,.$$ 
Finally, we make $y$ grow slowly so that $y'\to 0$ as $r^*\to \infty$. 

\item A region to the left of the first one, where we take our $h$ current down to zero as another $y$ current becomes strong so that $h$ and $y$ absorb each other's errors once again. Assume $|\omega|$ is small; we have different strategies depending on the size of $|\omega-m\upomega_+|$:
\begin{itemize}
\item if $|\omega-m\upomega_+|$ is small and $\Re\mc{V}_{(k)}'>0$ at large $-r^*$ at the expected rate from Lemma~\ref{lemma:properties-potential-low-omega} (see \ref{sec:F-low-1}), we follow the same strategy as for the $r^*\to\infty$ end when $|\omega|$ is small and $\Re\mc{V}_{(k)}'<0$ at large $r^*$, which was already described;
\item if $|\omega-m\upomega_+|$ is bounded away from zero (see Section~\ref{sec:F-low-2}), we employ a global, exponential-type $y$ current (degenerating as $r^*\to \infty$) which generates a good bulk term for all $r^*$ proportional to 
$$y'(\omega-m\upomega_+)^2\,;$$
then, in the region we are discussing, we take $h$ down in such a way that $h'<0$ and $h''$ is a small multiple of $y'(\omega-m\upomega_+)^2$, so that
$$\frac12y'(\omega-m\upomega_+)^2-\frac12\lp(h''+(|s|-k)\frac{w'}{w}h'\rp)\geq \frac14 y'(\omega-m\upomega_+)^2\,.$$
\end{itemize}
\end{itemize}

Our construction is based on that in the analogous regimes in \cite{Dafermos2016b}. The goal now is to tweak the currents $h$ and $y$ applied at each level $k$ to produce smallness in either the coupling to $k+1$  (we appeal to Lemma~\ref{lemma:h-y-identity-k}) or the coupling to $j<k$. Then, though we cannot close a good estimate for any $\uppsi_{(k)}$, $k<|s|$, by itself, we we can iterate along the hierarchy and use the absence of an imaginary potential at the top level $k=|s|$ to close a good bulk estimate for $\Psi$.

Finally, when applying energy currents, we note that the Killing currents also produce coupling errors for $s\neq 0$. As the natural weights appearing are not sufficiently strong, these errors cannot be absorbed unless either $a$ can be made small or only a small multiple of the Killing current is added. For general $a$, these currents are therefore insufficient to control the boundary terms generated by our virial estimates; we must use Teukolsky--Starobinsky-type energy currents instead, which do not produce coupling errors, to complete our proofs.

Sections  \ref{sec:low-hierarchy-h} and \ref{sec:low-hierarchy-y} contain the barebones structure of the iteration procedure for applying the $h$ and $y$ currents, respectively, at each level $k$. It also introduces model $h$ and $y$ currents that will be the starting point for our constructions of the necessary virial currents for $\mc{F}_{\rm low,1}$, in Section~\ref{sec:F-low-1}, and $\mc{F}_{\rm low,2}$, in Section~\ref{sec:F-low-2}. We direct the reader to these sections for further details on implementation.

\subsubsection{Current estimates for entire transformed system without the imaginary potential component}
\label{sec:low-hierarchy-h}

As $\Im\mc{V}_{(k)}$ does not have good decay properties, it natural to first consider currents which are blind to the imaginary part of the potential for a Schr\"{o}dinger-type ODE; as can be seen in section~\ref{sec:virial-current-templates}, the $h$ current is the sole one, among the currents we introduced, that fulfills this requirement. In a region where $h$ is ``good'', we  can control $\uppsi_{(k)}$ by the inhomogeneity $\mathfrak{G}_{(k)}$, as long as we tolerate  errors due to derivatives of $h$. 

\begin{lemma} \label{lemma:h-estimate-low}
Let $s\in\mathbb{Z}\backslash\{0\}$. For $k=0,...,|s|$, let $\uppsi_{(k)}$ be a solution of \eqref{eq:transformed-k-separated} defined by \eqref{eq:transformed-transport-separated}.  Then, for any function $h(r^*):\mathbb{R}\to\mathbb{R}_+$ such that 
\begin{itemize}[noitemsep]
\item $h\in C^1$ and is piecewise $C^2$, 
\item  $h$ compactly supported in a region of $r^*$ where, for $\sigma\in[0,5/4]$, $\mc{V}_{(k)}-\sigma(|s|-k)\lp(\frac{w'}{w}\rp)' -\omega^2\geq \frac18 w\,,$ for all $k=0,...,|s|$,
\end{itemize}
we have the estimates
\begin{align*}
&\int_{-\infty}^\infty \lp\{h\lp|\uppsi_{(k)}'\rp|^2+\frac34  h\lp(\Re\mc{V}_{(k)}-\sigma(|s|-k)\lp(\frac{w'}{w}\rp)'-\omega^2\rp)\lp|\uppsi_{(k)}\rp|^2\rp\}dr^* \numberthis \label{eq:h-estimate-low} \\
&\qquad \leq \int_{-\infty}^\infty h \Re\lp[\mathfrak{G}_{(k)}\overline{\uppsi_{(k)}}\rp]dr^*+\frac12 \int_{-\infty}^\infty \lp[h''+(|s|-k)\frac{w'}{w}h'-(s-k)(2\sigma-1)\lp(\frac{w'}{w}\rp)'h\rp]\lp|\uppsi_{(k)}\rp|^2dr^* \\
&\quad\qquad+ \lp\{\begin{array}{lr} 
1 &\text{if~} |s|\neq 1=k\\
\frac{a^2}{M^2}&\text{otherwhise} 
\end{array}\rp\}
k\sum_{j=0}^{k-1}\int_{-\infty}^\infty   B(1+m^2) h\lp(\Re\mc{V}_{(j)}-\sigma(|s|-j)\lp(\frac{w'}{w}\rp)'-\omega^2\rp) \lp|\uppsi_{(j)}\rp|^2 dr^*\,.
\end{align*}
\end{lemma}
\begin{proof}
At the level $k$, application of an $h$ current gives (see Lemma~\ref{lemma:h-y-identity-k}), for instance, for some $\sigma\in(0,1)$,
\begin{align*}
&\int_{-\infty}^\infty \lp\{h\lp|\uppsi_{(k)}'\rp|^2+h\lp[\Re\mc{V}_{(k)}-\sigma(|s|-k)\lp(\frac{w'}{w}\rp)'-\omega^2\rp]\lp|\uppsi_{(k)}\rp|^2\rp\}dr^* \\
&\quad =\int_{-\infty}^\infty \frac12\lp[h''+ (|s|-k)h'\frac{w'}{w}-(|s|-k)(2\sigma-1)h\lp(\frac{w'}{w}\rp)'\rp]\lp|\uppsi_{(k)}\rp|^2dr^*+\int_{-\infty}^\infty h \Re\lp[\mathfrak{G}_{(k)}\overline{\uppsi_{(k)}}\rp]dr^*\\
&\quad\qquad+\int_{-\infty}^\infty \sum_{j=0}^{k-1} ahw\Re\lp[\lp(imc_{s,k,j}^m+c_{s,k,j}^\mr{id}\rp)\uppsi_{(j)}\overline{\uppsi_{(k)}}\rp]dr^*\,,
\end{align*}
where the last term is absent for $k=0$. For $k=1,...,|s|$, we apply Cauchy--Schwarz
\begin{align*}
&\int_{-\infty}^\infty \sum_{j=0}^{k-1}ahw\Re\lp[\lp(imc_{s,k,j}^m+c_{s,k,j}^\mr{id}\rp)\uppsi_{(j)}\overline{\uppsi_{(k)}}\rp]dr^*\\
&\quad \leq \int_{-\infty}^\infty \frac{a^2}{M^2} k B(M,k)(1+m^2) \sum_{j=0}^{k-1} h\lp(\Re\mc{V}_{(j)}-\sigma(|s|-j)\lp(\frac{w'}{w}\rp)'-\omega^2\rp) \lp|\uppsi_{(j)}\rp|^2 dr^* \\   
&\quad\qquad+\int_{-\infty}^\infty \frac18 h \lp(\Re\mc{V}_{(k)}-\sigma(|s|-k)\lp(\frac{w'}{w}\rp)'-\omega^2\rp) \lp|\uppsi_{(k)}\rp|^2dr^*\,,
\end{align*}
where, noting that $ac_{s,1,0}^{\rm id}$ should be replaced by $c_{s,1,0}^{\rm id}$ for $|s|\neq 1$ (Proposition~\ref{prop:transformed-system}, we replace $a^2/M^2$ by 1 in the second line when $|s|\neq 1$ and $k=1$. Consequently, for each $k=0,...,|s|$, one obtains \eqref{eq:h-estimate-low}.
\end{proof}

Our goal is to construct $h$ such that the error introduced by its derivatives, i.e.\ the term proportional to $|\uppsi_{(k)}|^2$ on the right hand side of \eqref{eq:h-estimate-low}, is as small as possible. With this in mind, we introduce here several building blocks which will be relevant for the remainder of this section.

\paragraph{Model $h$ currents in a compact region of $|r^*|\gtrsim 1$}

We begin with a building block with good properties at large $r^*$, which we denote $h_+$. Let $R_3>0$ be very large and $p\in(0,1)$ be very small; fix $N\in\mathbb{N}_0$. Assume $h_+(R_3^*)>0$ and $h'_+(R_3^*)\geq 0$ are given. After $R_3^*$, we want $h_+\to 0$ such that $h_+''$ is very small, hence we let $h$ decrease for a large region, of size $(2e^{p^{-1}}-1)R_3^*$, before it reaches 0: in the range $[R_3^*,e^{p^{-1}}R_3^*]$, we take $h_+$ to be a linear combination of the constant function, $(r^*)^{-1}$ and $\log(r^*)$ such that $h_+$ and its derivatives are of size $p$ at the end of the interval and $h_+>0$; in $[e^{p^{-1}}R_3^*,2e^{p^{-1}}R_3^*]$, we let $h_+$ be a linear combination of $r^*$, the constant function, $(r^*)^{-1}$ and $\log(r^*)$ such that $h_+$ remains $C^1$ and $h_+,h_+'$ vanish at the end of the interval. Concretely,
\begin{align}
\begin{split}
h_+&=\lp[C_{N+5}^++C_{N+6}^+\frac{R_3^*}{r^*}+C_{N+7}^+p\log\lp(\frac{R_3^*}{r^*}\rp)\rp]\mathbbm{1}_{[R_3^*,e^{p^{-1}}R_3^*]}\\
&\qquad + \lp[C_{N+1}^++C_{N+2}^+\frac{r^*}{e^{p^{-1}}R_3^*}+C_{N+3}^+\log\lp(\frac{r^*}{e^{p^{-1}}R_3^*}+C_{N+4}^+\frac{e^{p^{-1}}R_3^*}{r^*}\rp)\rp]\mathbbm{1}_{[e^{p^{-1}}R_3^*,2e^{p^{-1}}R_3^*]}\,,
\end{split} \label{eq:standard-current-h+}
\end{align}
where $C_{N+1}^+$ up to $C_{N+4}^+$  are fixed in terms of the remaining constants by the constraints that $h_+$ is $C^1$ at $r=e^{p^{-1}}R_3^*$ and $h_+(2e^{p^{-1}}R_3^*)=h_+'(e^{p^{-1}}R_3^*)=0$, $C_{N+5}^+$ and $C_{N+6}^+$ are fixed in terms of the given $h_+(R_3^*),h'_+(R_3^*)$ by enforcing that $h_+$ is $C^1$ at $r^*=R_2^*$ and $C_{N+7}^+$ is chosen so that $h_+(e^{p^{-1}}R_3^*),h_+'(e^{p^{-1}}R_3^*)=O(p)$:
\begin{gather*}
C_{N+1}^+=\frac{2(1-\log 2)}{\log 8 -2}\lp[p(1-e^{-1/p})\lp(h_+(r_3^*)+R_3h_+'(R_3^*)\rp)-e^{-1/p}R_3^*h'_+(R_3^*)\rp]\,, \\
 C_{N+2}^+=\frac{\log 2}{2(1-\log 2)}C_{N+1}^+\,, \quad C_{N+3}^+=\frac{C_{N+1}^+}{1-\log 2}\,,\quad C_{N+4}=-2C_{N+1}^+\,, \\
C_{N+5}^+=(1+p)C_0^+\,,\quad C_{N+6}=ph_+(R_3^*)-(1+p)C_0^+ \,,\quad C_{N+7}^+=pC_0^+\,, \quad C_0^+=h_+(R_3^*)+R_3h_+'(R_3^*)\,.
\end{gather*}
The derivatives of $h_+$ satisfy
\begin{align*}
-h_+'&\leq \frac{pC_0^+}{r^*}\mathbbm{1}_{[R_3^*,e^{p^{-1}}R_3^*]}\,;\quad h_+''\mathbbm{1}_{[R_3^*,\infty)}\leq \frac{26pC_0^+}{(r^*)^2}\mathbbm{1}_{[R_3^*,2e^{p^{-1}}R_3^*]}\,,
\end{align*}
thus, as $w'/w=-2r^{-1}+O(r^{-2})$ as $r\to \infty$, we have
\begin{align*}
& h''_+\,,\,\,h''_++\lp(|s|-k\rp)\frac{w'}{w}h_+'\leq \frac{30p\lp(1+2|s|\rp)C_0^+}{(r^*)^2}\mathbbm{1}_{\supp (h_+')}\numberthis \label{eq:h+-estimate-error}\,,
\end{align*}
for sufficiently large $R_3^*$.

We now focus on a building block, $h_-$, which has good properties as $r\to r_+$. Note that, for $a\in[0,M)$, if we were to set $h_-(r^*)=h_+(-r^*)$, we could not show \eqref{eq:h+-estimate-error}, as $w'/w$ does not decay as $r\to r_+$; we could only prove the weaker bound
\begin{align*}
\lp(h_-''-\lp(|s|-k\rp)\frac{w'}{w}h_-'\rp) \leq \frac{Bp}{|r^*|}\mathbbm{1}_{\supp (h_-')}\,,\numberthis \label{eq:h+-estimate-error-weak}
\end{align*}
which is not integrable. Instead, we will take $h_-$ as follows. Let $R_1$ be sufficiently close to $r_+$ and $p\in(0,1)$ sufficiently small, define $R_1=r_++(R_2-r_+)e^{-1/p}$ and $R_{0}=r_++(R_1-r_+)/2$. Given $h_-(R_2)>0$ and $h_-'(R_2)\leq 0$, the building block $h_-(r^*)$ is the function
\begin{align*}
h_-&= \lp[C_{N+5}^-\frac{r-r_+}{R_2-r_+}\frac{R_2-M}{r-M}+C_{N+6}^-\frac{r-M}{R_2-M}+C_{N+7}^- \log\lp(\frac{r-M}{R_2-M}\rp)\rp]\mathbbm{1}_{[R_1,R_2]}\numberthis \label{eq:standard-current-h-}\\
&\qquad + \lp[C_{N+1}^-\frac{r-r_+}{R_1-r_+}\frac{R_1-M}{r-M}+C_{N+2}^-\frac{R_1-M}{r-M}+C_{N+3}^-\log\lp(\frac{R_1-M}{r-M}\rp)+C_{N+4}^-\frac{r-M}{R_1-M}\rp]\mathbbm{1}_{[R_{0},R_1]}\,,
\end{align*}
where $C_{N+1}^-$ up to $C_{N+4}^-$  are fixed in terms of the remaining constants by the constraints that $h_-$ is $C^1$ at $r=R_1$ and $h_-(R_{0})=h_-'(R_{0})=0$, $C_{N+5}^-$ and $C_{N+6}^-$ are fixed in terms of the given $h(R_2), h'(R_2)$ by enforcing that $h_-$ is $C^1$ at $r=R_2$, $C_{N+7}$ is chosen so that $h_-(R_1),h_-'(R_1)=O(p)$ when $a=M$. Indeed, we take
\begin{gather*}
C^-_{N+5}=(1+p)C^-_0\,,\quad C_{N+6}^-=ph_-(R_2)-(1+p)C_0^- \,,\quad C_{N+7}^-=\lp(p-\frac{(1+p)\sqrt{M^2-a^2}}{R_1-r_+}\rp)C_0^-\,,\\
C_0^-= h_-(R_2)-h_-'(R_2)\frac{(R_2-M)(R_2^2+a^2)}{\Delta(R_2)}\,.
\end{gather*}

The building block $h_-$ will be useful precisely in the limit $a\to M$. If $a$ is sufficiently close to $M$ that $\sqrt{M^2-a^2}\leq M e^{-\hat{c}/p}$ for $\hat{c}>1$, then in fact $C_{N+1}$ up to $C_{N+7}$ are the same as their $|a|=M$ counterparts up to $o(p)$ corrections: as
\begin{align*}
e^{1/p}\frac{R_1-M}{R_2-M}&=\frac{(R_2-r_+)+e^{1/p}\sqrt{M^2-a^2}}{R_2-M}
=\frac{R_2-M+(e^{1/p}-1)\sqrt{M^2-a^2}}{R_2-M}\\
&= \frac{R_2-M+O\lp(Me^{-(\hat{c}-1)/p}-e^{-\hat{c}/p}\rp)}{R_2-M} =1 +o(p)\,,\\
\frac{\sqrt{M^2-a^2}}{r-M}\mathbbm{1}_{[R_{0},R_2]}&\leq \frac{\sqrt{M^2-a^2}}{R_{0}-M}\leq \frac{\sqrt{M^2-a^2}}{R_{0}-r_+}\leq \frac{2e^{1/p}\sqrt{M^2-a^2}}{R_{2}-r_+} =O\lp(\frac{2e^{-(\hat{c}-1)/p}}{R_{2}-r_+}\rp)=o(p)\,, \numberthis\label{eq:h--intermediate}
\end{align*}
the same mechanism as for $h_+$ will enable us to conclude $h_-'(R_1),h_-(R_1)=O(p)$ with same leading order coefficient as for $|a|=M$ and, moreover,
\begin{gather} \label{eq:coefficients-h-}
\begin{gathered}
C^-_{N+1}=(1-\log 2)C_{N+3}^-+o(p)\,,\quad C_{N+2}^-=-2(1-\log 2)C_{N+3}^-+o(p) \,,\\
 C_{N+3}^-=\frac{2 p}{\log 8-2}\lp[h(R_1)\frac{e^{-1/p}}{p}+C_0^+\lp(1+e^{-1/p}-\frac{e^{-1/p}}{p}\rp)\rp]+o(p)\,, \quad C_{N+4}^-=\frac{\log 2}{2}C_{N+3}+o(p)\,.
 \end{gathered}
\end{gather}
As derivatives of $h_-$ satisfy
\begin{align*}
h_-''&\leq 6pwC_0^- \lp[1+\lp(1+\frac{1}{p}\rp)\frac{\sqrt{M^2-a^2}}{r-M}\rp]\mathbbm{1}_{[R_1,R_2]}\\
&\qquad+w\lp[\lp(3+\frac{\sqrt{M^2-a^2}}{r-M}\rp)\lp(C_{N+1}^-\frac{\sqrt{M^2-a^2}}{r-M}+C_{N+3}^-\rp)+C_{N+2}^-\frac{R_1-r_++\sqrt{M^2-a^2}}{M}\rp]\mathbbm{1}_{[R_0,R_1]}\\
h_-'&\leq \lp\{2pC_0^-\lp[1+\lp(1+\frac{1}{p}\rp)\frac{\sqrt{M^2-a^2}}{r-M}\rp]\mathbbm{1}_{[R_1,R_2]}+\lp(C_{N+1}^-\frac{\sqrt{M^2-a^2}}{r-M}+C_{N+3}^-\rp)\mathbbm{1}_{[R_0,R_1]}\rp\}\frac{\Delta}{r^2+a^2}\,,
\end{align*}
and we have
\begin{align*}
(r^2+a^2)^2\frac{w'}{w}=4Mr_+(r_+-M)+4a^2(r-r_+)+O(\Delta)\,, 
\end{align*}
as $r\to r_+$, for $R_2$ sufficiently close to $r_+$ and for  $\sqrt{M^2-a^2}\leq M e^{-\hat{c}/p}$ with $\hat{c}>1$, we have (c.f.\ \eqref{eq:h+-estimate-error})
\begin{align}
h_-''\,,\,\,h_-''+(|s|-k)\frac{w'}{w}h_-'\leq 30(3+2|s|)pw C_0^- \mathbbm{1}_{\supp(h_-')}\leq \frac{30(3+2|s|)Bp C_0^-}{(r^*)^2} \mathbbm{1}_{\supp(h_-')}\,.\label{eq:h--estimate-error} 
\end{align}

\paragraph{Model $h$ currents in a region of $|r^*|\lesssim 1$}

Finally, we discuss building blocks for $h$ which are suitable in a large region of $r$. In an application of the $h$ current to the transformed variable $\uppsi_{(k)}$, $k=0,\dots, |s|$, we note that, for $\sigma\in[0,4/5]$ (see Lemma~\ref{lemma:h-estimate-low})
\begin{align*}
h''+(|s|-k)\frac{w'}{w}h'+(1-2\sigma)(|s|-k)\lp(\frac{w'}{w}\rp)'h
\end{align*}
 represent error terms if their sum is positive for some $r\in[r_+,\infty)$. Letting $h$ be a multiple of $w^{-q}$ for some $q\in\mathbb{R}$ leads to
\begin{align*}
&h''+(|s|-k)\frac{w'}{w}h'+(1-2\sigma)(|s|-k)\lp(\frac{w'}{w}\rp)'h\\
&\quad=\lp\{\lp[-q+(1-2\sigma)(|s|-k)\rp]\lp[\lp(\frac{w'}{w}\rp)'-q\lp(\frac{w'}{w}\rp)^2\rp]-2\sigma q \lp(\frac{w'}{w}\rp)^2\rp\}h\,, \numberthis \label{eq:h-polynomial-new}
\end{align*}
which is certainly non-positive for all $r$ if $q=(1-2\sigma)(|s|-k)$; hence, there are no errors arising from such a choice of current $h$.

For $\sigma=1/2$, it will also be convenient to define $h$ to be some multiple of $r^N$, for $N\in\mathbb{N}_0$, when $r$ is very large but bounded. We have
\begin{align}
&h''+(|s|-k)\frac{w'}{w}h'=hr^{-2}\lp[N(N-2(|s|-k)-1)+\lp(\frac{w'}{w}+\frac{2}{r}\rp)(|s|-k)\rp]\leq 0\,, \label{eq:error-h-polynomial}
\end{align}
as $r\to\infty$ if, when $k\neq |s|$, $N<1+2(|s|-k)$ and, when $k=|s|$, if $N\leq 1$.

\subsubsection{Current estimates for entire transformed system with the imaginary potential component}
\label{sec:low-hierarchy-y}

This section focuses on separated current templates which are affected by $\Im\mc V_{(k)}$. The most glaring example are the Killing energy currents; motivated by identities of Lemma~\ref{lemma:Killing-identity-k}, we prove
\begin{lemma}\label{lemma:Killing-current-errors-bdd} Let $s\in\mathbb{Z}\backslash\{0\}$. For $k=0,...,|s|$, let $\uppsi_{(k)}$ be a solution of \eqref{eq:transformed-k-separated} defined by \eqref{eq:transformed-transport-separated}. For any $R_b\in(r_+,\infty)$ and $\varepsilon\in(0,1)$, we have  
\begin{align*}
&\int_{R_b^*}^\infty\omega\sum_{j=0}^{k-1}\Im[(ac_{s,k,j}^{\rm id}+imc_{s,k,j}^\Phi)\uppsi_{(j)}\overline{\uppsi_{(k)}}]dr^*\\
&\qquad +\int_{R_b^*}^\infty \sign s (|s|-k)\omega\lp[w'\Im[\uppsi_{(k+1)}\overline{\uppsi_{(k)}}]+\frac{4amrw}{r^2+a^2}|\uppsi_{(k)}|^2\rp]dr^*\\
&\quad \leq \int_{R_b^*/2}^\infty \lp[\lp(\varepsilon w\omega^2+\lp(\varepsilon+B(\varepsilon_{\rm width},\omega_{\rm high})|\omega|\varepsilon^{-1}\rp) \frac{w}{r}\rp)|\uppsi_{(k)}|^2+\omega^2\varepsilon^{-1}B(\varepsilon_{\rm width},\omega_{\rm high})\frac{w}{r}|\uppsi_{(k+1)}|^2\rp]dr^*\\
&\quad\qquad +\int_{R_b^*/2}^\infty\sum_{j=0}^{k-1}B(\varepsilon_{\rm width},\omega_{\rm high})|\omega|\varepsilon^{-1} wR_b^{1/2}\mathbbm{1}_{[R_b^*/2,R_b^*]}|\uppsi_{(j)}|^2dr^*\numberthis\label{eq:QT-current-errors-bdd}
\end{align*}
and, writing $\omega_0:=\omega-m\upomega_+$, 
\begin{align*}
&\int_{-\infty}^{-R_b^*}\omega_0\sum_{j=0}^{k-1}\Im[(ac_{s,k,j}^{\rm id}+imc_{s,k,j}^\Phi)\uppsi_{(j)}\overline{\uppsi_{(k)}}]dr^*\\
&\qquad +\int_{-\infty}^{R_b^*} \sign s (|s|-k)\omega_0\lp[w'\Im[\uppsi_{(k+1)}\overline{\uppsi_{(k)}}]+\frac{4amrw}{r^2+a^2}|\uppsi_{(k)}|^2\rp]dr^*\\
&\quad \leq \int_{R_b^*/2}^\infty \lp[\lp(\varepsilon w\omega_0^2+\lp(\varepsilon+B(\varepsilon_{\rm width},\omega_{\rm high})|\omega_0|\varepsilon^{-1}\rp) \frac{w}{r}\rp)|\uppsi_{(k)}|^2+\omega_0^2\varepsilon^{-1}B(\varepsilon_{\rm width},\omega_{\rm high})\frac{w}{r}|\uppsi_{(k+1)}|^2\rp]dr^*\\
&\quad\qquad +\int_{R_b^*/2}^\infty\sum_{j=0}^{k-1}B(\varepsilon_{\rm width},\omega_{\rm high})|\omega|wR_b^{1/2}|\omega_0|\mathbbm{1}_{[R_b^*/2,R_b^*]}|\uppsi_{(j)}|^2dr^*\,.\numberthis\label{eq:QK-current-errors-bdd}
\end{align*}
\end{lemma}
\begin{proof}
Let us note that first that a simple application of Cauchy--Schwarz yields
\begin{align*}
-\sign s (|s|-k)\omega \lp[w'\Im[\uppsi_{(k+1)}\overline{\uppsi_{(k)}}]+\frac{4amrw}{r^2+a^2}|\uppsi_{(k)}|^2\rp] \leq \frac{w}{r}\lp(\varepsilon|\uppsi_{(k)}|^2+\varepsilon^{-1}\omega^2 B(\varepsilon_{\rm width},\omega_{\rm high})|\uppsi_{(k+1)}|^2\rp)\,,
\end{align*}
where the right hand side already has the desired behavior as $r^*\to\pm \infty$.

In the coupling terms, the $r$-weights are weaker. For instance, an application of Cauchy--Schwarz yields
\begin{align*}
&\omega w \mathbbm{1}_{[R_b,\infty)} \Im[(ac_{s,k,j}^{\rm id}+imc_{s,k,j}^\Phi)\uppsi_{(j)}\overline{\uppsi_{(k)}}]\\
&\quad\leq \lp(\varepsilon|\omega| \frac{w}{r^{1/2}}|\uppsi_{(k)}|^2+ B(\varepsilon_{\rm width},\omega_{\rm high})\varepsilon^{-1}|\omega| wr^{1/2}|\uppsi_{(j)}|^2\rp)\mathbbm{1}_{[R_b,\infty)}\\
&\quad\leq\lp(\varepsilon \lp(w\omega^2+\frac{w}{r}\rp)|\uppsi_{(k)}|^2+ B(\varepsilon_{\rm width},\omega_{\rm high})\varepsilon^{-1}|\omega| \frac{\Delta}{r^{2+3/2}}|\uppsi_{(j)}|^2\rp)\mathbbm{1}_{[R_b,\infty)}\,,
\end{align*}
where the $r$-weights on the latter are not strong enough to conclude. To remedy this, we may apply the transport estimates of Lemma~\ref{lemma:basic-estimate-1}. Choosing
\begin{align*}
c=-\frac{1}{r^{1/2}}\mathbbm{1}_{[R_b,\infty)}-\frac{2r-R_b}{R_b^{3/2}}\mathbbm{1}_{[R_b/2,R_b]}\,,
\end{align*}
we find that
\begin{align*}
&\int \frac{\Delta}{r^{2+3/2}}|\uppsi_{(j)}|^2 \mathbbm{1}_{[R_b,\infty)} dr^*\leq B\int \lp(c'|\uppsi_{(j)}|^2 +\frac{1}{R_b^{3/2}}\mathbbm{1}_{[R_b/2,R_b]}\rp) dr^*\\
&\quad\leq B\int \lp[\frac{w}{r^{3/2}}|\uppsi_{(j+1)}|^2 \mathbbm{1}_{[R_b,\infty)} + w\lp(\frac{1}{R_b^{3/2}}|\uppsi_{(j+1)}|^2+\lp(1+\frac{r^2}{R_b^{3/2}}\rp)|\uppsi_{(j)}|^2\rp)\mathbbm{1}_{[R_b/2,R_b]}\rp] dr^*\\
&\quad \leq B\int \lp(\frac{w}{r}|\uppsi_{(j+1)}|^2 \mathbbm{1}_{[R_b/2,\infty)} +w R_b^{1/2}|\uppsi_{(j)}|^2\mathbbm{1}_{[R_b/2,R_b]}\rp) dr^*\,,
\end{align*}
where the last line follows from the fact that $r^2\sim R_b^2$ in $[R_b/2,R_b]$. The upshot here is that we have pushed the errors to a compact $r$ range; moreover, in the $\mc F_{\rm low}$ setting these come with a small parameter. Note that an analogous argument (in fact, an easier argument) works for the coupling errors arising from the $K$ Killing energy current.
\end{proof}

Next, we look at the virial currents. The $y$ current, since it is not blind to  $\Im\mc{V}_{(k)}$, cannot produce an estimate such as that in Lemma \ref{lemma:h-estimate-low}. By making use of the formulation of the $y$ current in Lemma \ref{lemma:h-y-identity-k}, we show that it can be used to obtain an estimate where, in a region where $y$ is ``good'', a bulk term in $\uppsi_{(k)}$ is controlled by a boundary term, by the inhomogeneity $\mathfrak{G}_{(k)}$ and, if $k<|s|$, by some multiple of the bulk term in $\uppsi_{(k+1)}$.

\begin{lemma} \label{lemma:y-estimate-low}
Let $s\in\mathbb{Z}\backslash\{0\}$. Suppose for some  $\varepsilon_{\rm width},\omega_{\rm high}>0$, we have $m^2, \Lambda\leq \varepsilon_{\rm width}^{-1}\omega_{\rm high}^2$, $|\omega|\leq \omega_{\rm high}$ and $|\Re\mc{V}_{(k)}|,|\Re\mc{V}_{(k)}'|\leq B(\varepsilon_{\rm width},\omega_{\rm high})$. For $k=0,...,|s|$, let $\uppsi_{(k)}$ be a solution of \eqref{eq:transformed-k-separated} defined by \eqref{eq:transformed-transport-separated}.  Let $\omega_0:=\omega-m\upomega_+$ and $\hat{\mc{V}}_{(k)}=\mc{V}_{(k)}-\mc{V}_{(k)}(r_+)$.

\begin{itemize}

\item For any function $y(r^*):\mathbb{R}\to\mathbb{R}$ such that 
\begin{itemize}
\item $y\in C^0$ is piecewise $C^1$ and $y,y'\geq 0$,
\item there is some $R_l^*\gg M\max\{|m|,1\}$ such that $y$ is supported in the range $r^*\in[R_l^*,\infty)$ and there $-\Re\mc{V}_{(k)}'\geq 2|s|r^{-3}$ for all $k=0,...,|s|$,
\item $r^*y'\leq (1-b_0)y$ for $r^*\in[R_r^*,\infty)$, for some $R_r^*\geq R_l^*\gg 1$ and $b_0\in(0,1)$,
\end{itemize}
we have the estimate
\begin{align*}
&\int_{-\infty}^\infty \lp\{y'|\uppsi_{(k)}'|^2+\lp[\frac34 y'\omega^2-y'\Re\mc{V}_{(k)}-\frac34 y\Re\mc{V}_{(k)}'\lp(\frac12 +B|am\omega|\rp) \rp]|\uppsi_{(k)}|^2\rp\}dr^*\\
&\qquad -\lp\{\begin{array}{lr}
1\,, &|s|\neq 1=k\\
a^2\,, &\text{otherwise}
\end{array}\rp\}
\int_{-\infty}^\infty\frac{B(b_0)(1+m^2)}{R_r^*}\lp\lVert \frac{wy}{y'}\mathbbm{1}_{[R_r^*,\infty)}\rp\rVert_\infty\frac{w^2 y^2}{ y'}\mathbbm{1}_{[R_l^*,R_r^*]}\lp|\uppsi_{(k)}\rp|^2dr^*\\
&\quad\leq  2y(\infty)\omega^2\lp|\swei{A}{s}_{k,\mc{I}^{-\sign s}}\rp|^2+\int_{-\infty}^\infty 2y\Re\lp[\mathfrak{G}_{(k)}\overline{\uppsi_{(k)}}'\rp]dr^* \numberthis \label{eq:y-estimate-low-k}\\
&\quad\qquad + \lp\{\begin{array}{lr}
1\,, &|s|\neq 1=k\\
a^2\,, &\text{otherwise}
\end{array}\rp\}
kB(1+m^2)\sum_{j=0}^{k-1} \int_{-\infty}^\infty \lp\{ - \frac{1}{r^2}y\Re\mc{V}'_{(j+1)}|\uppsi_{(j+1)}|^2\rp. \\
&\quad\qquad\qquad \lp.+\lp[-y\Re\mc{V}'_{(j)}\lp(1+\omega^2 R_r^2\rp)+ B(b_0) R_r^* y'w \mathbbm{1}_{[R_l^*,R_r^*]}\rp]|\uppsi_{(j)}|^2\rp\}dr^*\\
&\quad\qquad - (|s|-k)B\int_{-\infty}^\infty \lp(\frac{1}{r^2}+\omega^2\rp)y\Re\mc{V}'_{(k+1)}|\uppsi_{(k+1)}|^2dr^*\,.
\end{align*}

\item For any function $\hat y(r^*):\mathbb{R}\to\mathbb{R}$ such that $\hat y\in C^0$ is piecewise $C^1$, $-\hat y,\hat y'\geq 0$  and $\hat{y}(\infty)=0$, if $\omega_0\neq 0$ we have
\begin{align*}
&\int_{-\infty}^\infty \lp\{\hat y'|\uppsi_{(k)}'|^2+\lp[\frac{3}{4}\hat y'\omega_0^2\lp[\frac12 -B\frac{-\hat{y}\Delta r^{-3}}{\hat{y}'}\lp(\frac{1}{r^2}+\frac{\omega^2}{\omega_0^2}+\frac{|\omega|}{|\omega_0|}\rp)\rp]-(\hat y\Re\hat{\mc{V}}_{(k)})'\rp]|\uppsi_{(k)}|^2\rp\}dr^*\\
&\quad\leq -2\hat y(-\infty)\omega^2\lp|\swei{A}{s}_{k,\mc{H}^{+\sign s}}\rp|^2+\int_{-\infty}^\infty 2\hat y\Re\lp[\mathfrak{G}_{(k)}\overline{\uppsi_{(k)}}'\rp]dr^*\\
&\quad\qquad +k\sum_{j=0}^{k-1} \frac{B}{r^2}\hat{y}'|\uppsi_{(j)}|^2\lp(\frac{\hat y \Delta r^{-3}}{\hat y'}\rp)^2\lp(\frac{1+\omega^2}{\omega_0}+\lp|\frac{\omega}{\omega_0}\rp|+1\rp)dr^* \numberthis\label{eq:hat-y-estimate-low-k-global}\\
&\quad\qquad +(|s|-k)\int_{-\infty}^\infty \frac{B}{r^4}\hat{y}'|\uppsi_{(k+1)}|^2\lp[\lp(\frac{\hat y \Delta r^{-3}}{\hat y'}\rp)^2\lp(\frac{1}{r^4}+\lp(\frac{\omega}{\omega_0}\rp)^2\rp)+\lp|\frac{\hat y \Delta r^{-3}}{\hat y'}\rp|\rp]dr^*\,.
\end{align*}
\item Assume $m=0$. For any function $\hat y(r^*):\mathbb{R}\to\mathbb{R}$ such that 
\begin{itemize}
\item $\hat y\in C^0$ is piecewise $C^1$ and $-\hat y,\hat y'\geq 0$,
\item there is some $R_l^*\gg 1$ such that $\hat y$ is supported  in the range $r^*\in(-\infty, -R_l^*]$, where $-\Re\hat{\mc{V}}_{(k)}'\geq  \frac{3|s|}{4}(r-M)\Delta $ for all $k=0,...,|s|$,
\end{itemize}
we have the estimate
\begin{align*}
&\int_{-\infty}^\infty \lp\{\hat y'|\uppsi_{(k)}'|^2+\lp[\frac34 \hat y'\omega_0^2-\hat y'\Re\hat{\mc{V}}_{(k)}- \frac{15}{16}\hat y\Re\hat{\mc{V}}_{(k)}' \rp]|\uppsi_{(k)}|^2\rp\}dr^*\\
&\quad\leq  2\hat y(-\infty)\omega^2\lp|\swei{A}{s}_{k,\mc{H}^{+\sign s}}\rp|^2+\int_{-\infty}^\infty 2\hat y\Re\lp[\mathfrak{G}_{(k)}\overline{\uppsi_{(k)}}'\rp]dr^* \numberthis \label{eq:hat-y-estimate-low-k}\\
&\quad\qquad - \lp\{\begin{array}{lr}
1\,, &|s|\neq 1=k\\
a^2\,, &\text{otherwise}
\end{array}\rp\}
kB\sum_{j=0}^{k-1} \int_{-\infty}^\infty \lp\{ \hat y\Re\hat{\mc{V}}'_{(j)}\lp(1+\omega^2(R_r^*)^2 \rp)|\uppsi_{(j)}|^2  \rp.\\
&\quad\qquad\qquad\lp. +\hat y\Re\hat{\mc{V}}'_{(j+1)}(r-r_+)^2|\uppsi_{(j+1)}|^2 - \hat{y}\Delta|\omega|\mathbbm{1}_{(-\infty, -R_r^*]}|\uppsi_{(k)}||\uppsi_{(j)}|\rp\}dr^* \\
&\quad\qquad - (|s|-k)B\int_{-\infty}^\infty \hat{y}\omega^2\Re\hat{\mc{V}}'_{(k+1)}|\uppsi_{(k+1)}|^2dr^*\,,
\end{align*}
where, if $\hat y$ additionally satisfies $r^*\hat y'/|\hat y|\leq 1-b_0$ for $r^*\in(-\infty,-R_r^*]$, for some $R_r^*\geq R_l^*\gg 1$ and $b_0\in(0,1)$, the term involving $|\uppsi_{(k)}||\uppsi_{(j)}|$ can be controlled by
\begin{align*}
&\int_{-\infty}^\infty \hat{y}w\omega|\uppsi_{(k)}||\uppsi_{(j)}|\mathbbm{1}_{(-\infty, -R_r^*]}dr^* \\
&\quad\leq \frac18 \int_{-\infty}^\infty \hat y'\omega^2|\uppsi_{(k)}|^2dr^*\\
&\qquad\quad+\frac{B(b_0)}{R_r^*}\lp\lVert\frac{w\hat{y}}{\hat{y}'}\mathbbm{1}_{(-\infty,-R_r^*]}\rp\rVert\int_{-\infty}^\infty \lp\{\frac{w^2\hat{y}^2}{\hat y'}|\uppsi_{(k)}|^2+\frac{\hat y'w}{w(-R_r^*)}  |\uppsi_{(j)}|^2\rp\}\mathbbm{1}_{[-R_r^*,-R_l^*]}dr^*\,.
\end{align*}
\end{itemize}
\end{lemma}

\begin{proof}
Our starting point is the application of a $y$ current to \eqref{eq:transformed-k-separated} for $0\leq k \leq |s|$, see Lemma~\ref{lemma:h-y-identity-k}. If $k\neq 0$, this procedure generates coupling terms; using the boundary conditions for $\uppsi_{(k)}$, we can integrate by parts 
\begin{align*}
&\int_{-\infty}^\infty 2ayw\Re\lp[\overline{\uppsi_{(k)}}'(c_{s,k,j}^{\rm id}+imc_{s,k,i}^\Phi)\uppsi_{(j)}\rp]dr^* \\
&\quad = -\int_{-\infty}^\infty \Re\lp[\overline{\uppsi_{(k)}}\lp(2ayw(c_{s,k,j}^{\rm id}+imc_{s,k,j}^\Phi)\rp)'\uppsi_{(j)}\rp]dr^*\\
&\quad\qquad-\sign s\int_{-\infty}^\infty 2ayw\Re\lp[\overline{\uppsi_{(k)}}(c_{s,k,j}^{\rm id}+imc_{s,k,j}^\Phi)\lp(w\uppsi_{(j+1)}+i\lp(\omega-\frac{am}{r^2+a^2}\rp)\uppsi_{(j)}\rp)\rp]dr^*,.
\end{align*}
From Lemma~\ref{lemma:h-y-identity-k}, we now find that the $y$ current gives the identity
\begin{align*}
&\int_{-\infty}^\infty \lp\{y'|\uppsi_{(k)}'|^2+\lp[y'\omega^2-(y\Re\mc{V}_{(k)})'\rp]|\uppsi_{(k)}|^2\rp\}dr^*\\
&\quad= 2y(\infty)\omega^2\lp|\swei{A}{s}_{k,\mc{I}^{-\sign s}}\rp|^2-2y(-\infty)\omega_0^2\lp|\swei{A}{s}_{k,\mc{H}^{+\sign s}}\rp|^2+\int_{-\infty}^\infty 2y\Re\lp[\mathfrak{G}_{(k)}\overline{\uppsi_{(k)}}'\rp]dr^*\\
&\quad\qquad -\sum_{j=0}^{k-1} \int_{-\infty}^\infty \Re\lp[\overline{\uppsi_{(k)}}\lp(2ayw(c_{s,k,j}^{\rm id}+imc_{s,k,j}^\Phi)\rp)'\uppsi_{(j)}\rp]dr^* \numberthis\label{eq:y-low-intermediate}\\
&\quad\qquad -\sign s\sum_{j=0}^{k-1}\int_{-\infty}^\infty 2ayw\Re\lp[\overline{\uppsi_{(k)}}(c_{s,k,j}^{\rm id}+imc_{s,k,j}^\Phi)\lp(w\uppsi_{(j+1)}+i\lp(\omega-\frac{am}{r^2+a^2}\rp)\uppsi_{(j)}\rp)\rp]dr^*\\
&\quad\qquad + \int_{-\infty}^\infty 2y(|s|-k)w'\lp\{w|\uppsi_{(k+1)}|^2-\lp(\omega-\frac{am}{r^2+a^2}\rp)\Re\lp[\uppsi_{(k+1)}\overline{\uppsi_{(k)}}\rp]\rp\}dr^* \\
&\quad\qquad+ \int_{-\infty}^\infty2y(|s|-k)\frac{4amrw}{(r^2+a^2)}\lp\{ w\Im \lp[\uppsi_{(k+1)}\overline{\uppsi}_{(k)}\rp]+ \lp(\omega-\frac{am}{r^2+a^2}\rp)|\uppsi_{(k)}|^2\rp\}dr^*\,,
\end{align*}
or, alternatively, writing $\Re\hat{\mc{V}}_{(k)}=\Re\mc{V}_{(k)}-\Re\mc{V}_{(k)}(r_+)$ and $\omega_0:=(\omega-m\upomega_+)$ 
\begin{align*}
&\int_{-\infty}^\infty \lp\{ y'|\uppsi_{(k)}'|^2+\lp[ y'\omega_0^2-(y\Re\hat{\mc{V}}_{(k)})'\rp]|\uppsi_{(k)}|^2\rp\}dr^*\\
&\quad= 2y(\infty)\omega^2\lp|\swei{A}{s}_{k,\mc{I}^{-\sign s}}\rp|^2-2y(-\infty)\omega_0^2\lp|\swei{A}{s}_{k,\mc{H}^{+\sign s}}\rp|^2+\int_{-\infty}^\infty 2y\Re\lp[\mathfrak{G}_{(k)}\overline{\uppsi_{(k)}}'\rp]dr^*\\
&\quad\qquad -\sum_{j=0}^{k-1} \int_{-\infty}^\infty 2a\lp\{\Re\lp[\overline{\uppsi_{(k)}}\lp(yw(c_{s,k,j}^{\rm id}+imc_{s,k,j}^\Phi)\rp)'\uppsi_{(j)}\rp]+yw^2\Re\lp[\overline{\uppsi_{(k)}}(c_{s,k,j}^{\rm id}+imc_{s,k,j}^\Phi)\uppsi_{(j+1)}\rp]\rp\}dr^* \\
&\quad\qquad +\sign s\sum_{j=0}^{k-1}\int_{-\infty}^\infty 2ayw\Im\lp[\overline{\uppsi_{(k)}}(c_{s,k,j}^{\rm id}+imc_{s,k,j}^\Phi)\lp(\omega_0+\frac{m\upomega_+(r-r_+)(r+r_+)}{r^2+a^2}\rp)\uppsi_{(j)}\rp]dr^*\numberthis\label{eq:y-low-intermediate-2}\\
&\quad\qquad + \int_{-\infty}^\infty 2y(|s|-k)w'\lp\{w|\uppsi_{(k+1)}|^2-\lp(\omega_0+\frac{m\upomega_+(r-r_+)(r+r_+)}{r^2+a^2}\rp)\Re\lp[\uppsi_{(k+1)}\overline{\uppsi_{(k)}}\rp]\rp\}dr^* \\
&\quad\qquad- \int_{-\infty}^\infty 8y(|s|-k)\lp(\omega_0+\frac{m\upomega_+(r-r_+)(r+r_+)}{r^2+a^2}-\omega\rp)rw^2\Im \lp[\uppsi_{(k+1)}\overline{\uppsi}_{(k)}\rp]dr^*\\
&\quad\qquad+ \int_{-\infty}^\infty 8y(|s|-k)\lp[-\lp(\omega-\frac{am}{r^2+a^2}\rp)^2+\omega\lp(\omega_0+\frac{m\upomega_+(r-r_+)(r+r_+)}{r^2+a^2}\rp)\rp]rw|\uppsi_{(k)}|^2dr^*\,,
\end{align*}
where $ac_{s,k,j}^{\mr{id}}$ must be replaced by $c_{s,k,j}^{\mr{id}}$ if $|s|\neq 1=k$ and $i=0$ and where the second and third lines on the right hand side of both identities is absent for any $s$ as long as $k=0$ and the following lines are absent if $k=|s|$. In particular, all but the first line on the right hand side is absent if $s=0$, hence the estimates are trivial in that case. From now on, assume $s\neq 0$.

\medskip
\noindent \textit{Current at large $r$.}
If $k<|s|\neq 0$, terms involving $\uppsi_{(k+1)}$ appear in the last two lines of \eqref{eq:y-low-intermediate}. As long as $y$ is supported at sufficiently large $r$ that $am\ll r$ in $y$'s support, they can be controlled by
\begin{align*}
&-\frac{1}{16}\int_{-\infty}^\infty y\Re\mc{V}'_{(k)}|\uppsi_{(k)}|^2\lp[\frac12+48\frac{7|s|/2wr^{-1}}{\Re\mc{V}'_{(k)}}\lp(|am\omega|+\frac{a^2m^2}{r^2}\rp)\rp]dr^*\\
&\qquad-B(M,|s|)\int_{-\infty}^\infty \lp(\omega^2+\frac{1}{r^2}\rp)y\Re\mc{V}'_{(k+1)}|\uppsi_{(k+1)}|^2\lp[\frac{wr^{-1}}{\Re\mc{V}_{(k+1)}'}+\frac{(wr^{-1})^2}{\Re\mc{V}_{(k+1)}'\Re\mc{V}_{(k)}'}\lp(1+\frac{a^2m^2}{r^2}\rp)\rp]dr^*\\
&\quad\leq \frac{1}{16}\int_{-\infty}^\infty \lp\{-(1+B|am\omega|)y\Re\mc{V}'_{(k)}|\uppsi_{(k)}|^2-B\lp(\omega^2+\frac{1}{r^2}\rp)y\Re\mc{V}'_{(k+1)}|\uppsi_{(k+1)}|^2\rp\}\,.
\end{align*}

For the coupling errors we again use the assumption on the support of $y$ to show that the second and third lines of \eqref{eq:y-low-intermediate} are controlled by
\begin{align*}
& \{1,|a|\} B(M,|s|)(1+|m|)\int_{-\infty}^\infty \lp\{\frac{wy}{r}\lp(|\uppsi_{(k)}||\uppsi_{(j)}|+\frac{1}{r}|\uppsi_{(k)}||\uppsi_{(j+1)}|\rp)+ wy\omega|\uppsi_{(k)}||\uppsi_{(j)}|\rp\}dr^*\\
&\quad \leq \frac18 \int_{-\infty}^\infty\lp[y'\omega^2- y\Re\mc{V}'_{(k)}\rp]|\uppsi_{(k)}|^2dr^* + \{1,a^2\}B(M,|s|)(1+m^2)\int_{-\infty}^\infty \frac{1}{r^2}(-y)\Re\mc{V}'_{(j+1)}|\uppsi_{(j+1)}|^2dr^*\\
&\quad \qquad + \{1,a^2\}B(M,|s|)(1+m^2)\int_{-\infty}^\infty \lp\{ -y\Re\mc{V}'_{(j)}\lp(1+\omega^2 r^2 \mathbbm{1}_{[R^*_l,R^*_r]}\rp)+\frac{wy}{y'}wy\mathbbm{1}_{[R_r^*,\infty)}\rp\}|\uppsi_{(j)}|^2dr^*\,,
\end{align*}
where we note that the second term on the right hand side is absent if $j+1=k$: assuming $y$ is supported at sufficiently large $r$, the term can be absorbed into the first if $j+1=k$.  The alternatives in $\{1,a^2\}$ and $\{1,|a|\}$ given are for $k=1\neq |s|$ and any other pair $(s,k)$, respectively.

Let us focus on the error terms with integrals of $wy|\uppsi_{(j)}|^2$ in $[R_r^*,\infty)$. Choosing
\begin{align*}
c&=-\frac{y}{r^*}\mathbbm{1}_{[R_r^*,\infty)} -\frac{y}{R_r^*}\mathbbm{1}_{[R_l^*,R_r^*]}\,,\\
c'\mathbbm{1}_{[R_r^*,\infty)} &=\frac{y}{(r^*)^2}\lp(1-\frac{r^*y'}{y}\rp)\mathbbm{1}_{[R_r^*,\infty)} \geq b_0\frac{y}{(r^*)^2}\mathbbm{1}_{[R_r^*,\infty)}\,,
\end{align*}
we can apply \eqref{eq:basic-estimate-1} from Lemma~\ref{lemma:basic-estimate-1} to obtain
\begin{align*}
\int_{-\infty}^\infty wy\mathbbm{1}_{[R_r^*,\infty)}|\uppsi_{(j)}|^2dr^* &\leq B(b_0) \int_{-\infty}^\infty \lp[c'+\frac{y'}{R_r^*}\mathbbm{1}_{[R_l^*,R_r^*]}\rp]|\uppsi_{(j)}|^2dr^* \\
&\leq B(b_0) \int_{-\infty}^\infty \lp[-\frac{1}{r}y\Re\mc{V}_{(j+1)}'\mathbbm{1}_{[R_r^*,\infty)}\rp]|\uppsi_{(j+1)}|^2dr^*\\
&\qquad + B(b_0)\int_{-\infty}^\infty  \lp[\frac{w^2y^2}{y'R_r^*}|\uppsi_{(j+1)}|^2+R_r^* y'w|\uppsi_{(j)}|^2\rp] \mathbbm{1}_{[R_l^*,R_r^*]}dr^*\,,
\end{align*}
and repeat until $j+1=k$. In that estimate, if the compact support of $y$ is at sufficiently large $r$, then the first term can be absorbed into the left hand side of \eqref{eq:y-low-intermediate}. 

Putting all of the previous together, we obtain \eqref{eq:y-estimate-low-k}.

\medskip
\noindent \textit{Current near $r=r_+$.} If $k<|s|$, terms involving $\uppsi_{(k+1)}$ appear in the last two lines of \eqref{eq:y-low-intermediate}. Proceeding as in the previous step, we see that, as long as $R^*_l$ is sufficiently large depending on $\varepsilon_{\rm width}$ and $\omega_{\rm high}$, those two lines can be controlled by 
\begin{align*}
&-\frac{1}{16} \int_{-\infty}^\infty \hat y\Re\hat{\mc{V}}'_{(k)}|\uppsi_{(k)}|^2dr^* - (|s|-k)B(\varepsilon_{\rm width},\omega_{\rm high})\int_{-\infty}^\infty \omega^2\hat y\Re\hat{\mc{V}}'_{(k+1)}|\uppsi_{(k+1)}|^2dr^*\,.
\end{align*}

For the coupling terms, we find that each element in the sums in the second and third lines of \eqref{eq:y-low-intermediate} can be controlled by 
\begin{align*}
&\{1,|a|\} B\int_{-\infty}^\infty (-\hat{y})\lp\{w(r-M) \lp|\uppsi_{(k)}\rp|\lp(\lp|\uppsi_{(j)}\rp|+(r-r_+)\lp|\uppsi_{(j+1)}\rp|\rp)\rp\}dr^*\\
&\qquad + \{1,|a|\} B(\varepsilon_{\rm width},\omega_{\rm high})\int_{-\infty}^\infty \lp\{(-\hat{y})w\omega\lp|\uppsi_{(k)}\rp|\lp|\uppsi_{(j)}\rp|\lp(\mathbbm{1}_{[-R_r^*,-R_l^*]} +\mathbbm{1}_{(-\infty,-R_r^*]}\rp)\rp\}dr^*\\
&\leq \frac18 \int_{-\infty}^\infty\lp[\hat y'\omega^2- \hat y\Re\hat{\mc{V}}'_{(k)}\rp]|\uppsi_{(k)}|^2dr^* + \{1,a^2\}B\int_{-\infty}^\infty (r-r_+)^2(-\hat y)\Re\hat{\mc{V}}'_{(j+1)}|\uppsi_{(j+1)}|^2dr^*\\
&\quad \qquad + \{1,a^2\}B\int_{-\infty}^\infty \lp\{ -\hat y\Re\hat{\mc{V}}'_{(j)}\lp(1+\frac{\omega^2}{(r-M)^2}  \mathbbm{1}_{[-R^*_r,R^*_l]}\rp)+\frac{w(-\hat{y})}{\hat y'}w(-\hat{y})\mathbbm{1}_{(-\infty,-R_r^*]}\rp\}|\uppsi_{(j)}|^2dr^*\,,
\end{align*}
where the alternatives in $\{1,a^2\}$ and $\{1,|a|\}$ given are for $k=1\neq |s|$ and any other pair $(s,k)$, respectively.
We can now deal with the term involving $|\uppsi_{(j)}|^2$ in an analogous manner as for the the previous section of the proof, we obtaining finally \eqref{eq:hat-y-estimate-low-k}.

\medskip
\noindent \textit{Global current for $\omega_0\neq 0$.} We note the equalities
\begin{align*}
\frac{am}{r^2+a^2}=-(\omega-\omega_0)\frac{2Mr_+}{r^2+a^2}\,, \quad \omega-\frac{am}{r^2+a^2}= \omega_0\frac{2Mr_+}{r^2+a^2}+\frac{(r-r_+)(r+r_+)}{r^2+a^2}\omega\,.
\end{align*}

Suppose $\omega_0\neq 0$. We can control the terms in the last three lines of \eqref{eq:y-low-intermediate-2}  by
\begin{align*}
&\int_{-\infty}^\infty \frac{1}{16}\hat{y}'\omega_0^2|\uppsi_{(k)}|^2\lp[1+B\frac{-\hat{y}\Delta r^{-3}}{\hat{y}'}\lp(\frac{1}{r^2}+\frac{\omega^2}{\omega_0^2}+\frac{|\omega|}{|\omega_0|}\rp)\rp]dr^*\\
&\qquad +(|s|-k)\int_{-\infty}^\infty \frac{B}{r^4}\hat{y}'|\uppsi_{(k+1)}|^2\lp[\lp(\frac{\hat y \Delta r^{-3}}{\hat y'}\rp)^2\lp(\frac{1}{r^4}+\lp(\frac{\omega}{\omega_0}\rp)^2\rp)+\lp|\frac{\hat y \Delta r^{-3}}{\hat y'}\rp|\rp]dr^*\,,
\end{align*}
and the coupling terms in the second and third lines of \eqref{eq:y-low-intermediate-2} by
\begin{align*}
\int_{-\infty}^\infty \frac{1}{8}\hat{y}'\omega_0^2|\uppsi_{(k)}|^2dr^*+\int_{-\infty}^\infty \frac{B}{r^2}\hat{y}'|\uppsi_{(j)}|^2\lp(\frac{\hat y \Delta r^{-3}}{\hat y'}\rp)^2\lp(\frac{1+\omega^2}{\omega_0}+\lp|\frac{\omega}{\omega_0}\rp|+1\rp)dr^*\,.
\end{align*}

It is now easy to obtain \eqref{eq:hat-y-estimate-low-k-global}.
\end{proof}

\paragraph{Two model $y$ currents}

We begin with a $y$ current supported at large $r^*$. Let $R_3^*\gg 1$, $R_2^*\in[R_3^*/2,R_3^*-1]$, $R_4^*> R_3^*$, $p\in(0,1)$ be sufficiently small, and $h$ be a $C^1$ function such that $h(R_3^*)>0$ and $h'(R_3^*)\geq 0$. For some $\tilde{p}\in(0,1/8)$, some $\delta\in(0,1]$ and $\epsilon\in(0,1)$, define 
\begin{align*}
y(r^*)&= \tilde{p}\lp(\int_{R_2^*}^{r^*}h(x^*)dx^*\rp)\mathbbm{1}_{[R_2^*,R_3^*]}+y(R_3^*)\lp[1+\frac{64\epsilon}{R_3^*}\lp(\frac{1}{r^*\Re\mc{V}_{(k)}}-\frac{1}{R_3^*\Re\mc{V}_{(k)}|_{r^*=R_3^*}}\rp)\rp]\mathbbm{1}_{[R_3^*, R_4^*]} \\
&\qquad+y(R_4^*)\lp[1+\frac{1}{(R_4^*)^{\delta}}-\frac{1}{(r^*)^{\delta}}\rp]\mathbbm{1}_{[R_4^*, \infty)}\,, \numberthis \label{eq:low-frequencies-y}
\end{align*}
which, as long as $\lp(-r^*\Re\mc{V}_{(k)}\rp)'\geq 0$ when $r^*\geq R_3^*$ for sufficiently large $R_3^*$, satisfies
\begin{align*}
y'(r^*)&=\tilde{p}h\mathbbm{1}_{[R_2^*,R_3^*]}+\frac{64\epsilon y(R_3^*)}{R_3^*}\frac{1}{(r^*\Re\mc{V}_{(k)})^2}\lp(-r^*\Re\mc{V}_{(k)}\rp)'\mathbbm{1}_{[R_3^*, R_4^*]} +\frac{\delta y(R_4^*)}{(r^*)^{1+\delta}}\mathbbm{1}_{[R_4^*, \infty)}\geq 0\,.\numberthis \label{eq:low-frequencies-y'}
\end{align*}

Assume $\mc{V}_{k}\leq b r^{-2}$ as $r\to \infty$; then for $R_3^*$ sufficiently large we can choose $\epsilon$ sufficiently small independently of $R_3^*$ so that $64\epsilon\leq (R_3^*)^2\mc{V}_{(k)}(R_3^*)/4$. Assuming $\Re\mc{V}_{(k)}'<0$ in the support of $y$, we obtain 
\begin{align*}
-(y\Re\mc{V}_{(k)})'\mathbbm{1}_{[R_3^*,\infty)}&\geq\lp[\frac{64 \epsilon y(R_3^*)}{R_3^*(r^*)^2}-y(R_3^*)\Re\mc{V}_{(k)}'\lp(1-\frac{64\epsilon}{R_3^2\mc{V}_{(k)}(R_3^*)|}\rp)\rp]\mathbbm{1}_{[R_3^*,R_4^*]}\\
&\qquad-y\Re\mc{V}_{(k)}'\mathbbm{1}_{[R_4^*,\infty)}\lp(1-\frac{B\delta}{|R^*_4|^\delta}\rp)\\
&\geq \frac{64\epsilon y(R_3^*)/R_3^*}{(r^*)^2}\mathbbm{1}_{[R_3^*, R_4^*]}-\frac34y(R_3^*)\Re\mc{V}_{(k)}'\mathbbm{1}_{[R_3^*, R_4^*]}-\frac34y\Re\mc{V}_{(k)}'\mathbbm{1}_{[R_4^*,\infty)}\\
&\geq b\lp[R_3^*h'(R_3^*)+h(R_3^*)\rp]\frac{\tilde{p}\epsilon}{(r^*)^2}\mathbbm{1}_{[R_3^*, R_4^*]}-\frac34y\Re\mc{V}_{(k)}'\mathbbm{1}_{[R_3^*, \infty)}\,. \numberthis \label{eq:low-frequencies-yV}
\end{align*}
Moreover, if $0<\lp(-r^*\Re\mc{V}_{(k)}\rp)'\leq  2\Re\mc{V}_{(k)}$ for $r^*\in [R_3^*,R_4^*]$ and $\epsilon$ is sufficiently small, we have
\begin{align*}
\lp(1-\frac{r^*y'}{y}\rp)\mathbbm{1}_{[R_3^*,\infty)} &= 1-\frac{64\epsilon/R_3^* }{\lp(1-\frac{64\epsilon}{(R_3^*)^2\Re\mc{V}_{(k)}(R_3^*)}\rp)r^*\Re\mc{V}_{(k)}+64\epsilon/R_3^*} \frac{\lp(-r^*\Re\mc{V}_{(k)}\rp)'}{\Re\mc{V}_{(k)}}\mathbbm{1}_{[R_3^*,R_4^*]}\\
&\qquad- \frac{\delta}{(R_4^*)^\delta+1}\mathbbm{1}_{[R_4^*,\infty)}\\
&\geq 1-\frac{128\epsilon }{3/4+64\epsilon} \mathbbm{1}_{[R_3^*,R_4^*]}- \frac{\delta}{(R_4^*)^\delta+1}\mathbbm{1}_{[R_4^*,\infty)}\geq b\lp(\epsilon,\delta\rp)\,.
\end{align*}

If $h$ satisfies the estimate 
\begin{align*}
h'' +(|s|-k)\frac{w'}{w}h'\leq \frac{Bp \lp[R_3^*h'(R_3^*)+h(R_3^*)\rp]}{(r^*)^2}\mathbbm{1}_{[R_3^*,R_4^*]}\,,
\end{align*}
then as long as $p\ll \tilde{p}$,
\begin{align*}
&\frac34 h(\mc{V}_{(k)} -\omega^2) +\frac34 y'\omega^2-\frac34(y\Re\mc{V}_{(k)})'-\frac14y'\Re\mc{V}_{(k)} -\lp(h_+'' +(|s|-k)\frac{w'}{w}h_+'\rp)\mathbbm{1}_{[R_3^*,2e^{p^{-1}}R_3^*]}\\
&\quad \geq \frac{1}{2}h\lp(\mc{V}_{(k)} -\omega^2\rp)+\frac34 y'\omega^2+b\lp[R_3^*h'(R_3^*)+h(R_3^*)\rp]\frac{\tilde{p}\epsilon}{(r^*)^2}\mathbbm{1}_{[R_3^*,R_4^*]}-\frac12 y\Re\mc{V}_{(k)}'\mathbbm{1}_{[R_3^*,\infty)}\,. \numberthis \label{eq:low-frequencies-h+y}
\end{align*}

Similarly, as $r^*\to -\infty$, we construct a current $\hat{y}$ as follows. Let $r_+<R_0<R_2<R_5<\infty$ with $R_2$ very close to $r_+$ and $-R_5^*\in[-R_2^*/2, -R_3^*+1]$. Let $p\in(0,1)$ be sufficiently small, and $h$ be a $C^1$ function such that $h(R_2^*)>0$ and $h'(R_2^*)\leq 0$. For some $\tilde{p}\in(0,1/8)$, some $\delta\in(0,1]$ and $\epsilon\in(0,1)$, define 
\begin{align*}
\hat y(r^*)&= -\tilde{p}\lp(\int_{R_5^*}^{r^*}h(x^*)dx^*\rp)\mathbbm{1}_{[R_2^*,R_5^*]}+\hat{y}(R_2^*)\lp[1-\frac{64\epsilon}{(-R_2^*)}\lp(\frac{1}{r^*\Re\hat{\mc{V}}_{(k)}}-\frac{1}{R_2^*\Re\hat{\mc{V}}_{(k)}|_{r^*=R_2^*}}\rp)\rp]\mathbbm{1}_{[R_0^*, R_2^*]} \\
&\qquad+\hat{y}(R_0^*)\lp[1+\frac{1}{(-R_0^*)^{\delta}}-\frac{1}{(-r^*)^{1/2}}\rp]\mathbbm{1}_{(-\infty,R_0^*]}\,, \numberthis \label{eq:low-frequencies-hat-y}
\end{align*}
which, as long as $\lp(-r^*\Re\hat{\mc{V}}_{(k)}\rp)'>0$ when $r^*\leq R_2^*$ for sufficiently negative $R_2^*$, satisfies
\begin{align*}
\hat y'(r^*)&=\tilde{p}h\mathbbm{1}_{[R_2^*,R_5^*]}+\frac{64\epsilon \lp(-\hat{y}(R_2^*)\rp)}{(-R_2^*)}\frac{1}{(r^*\Re\hat{\mc{V}}_{(k)})^2}\lp(-r^*\Re\hat{\mc{V}}_{(k)}\rp)'\mathbbm{1}_{[R_0^*, R_2^*]} -\frac{ \hat{y}(R_4^*)}{(-r^*)^{3/2}}\mathbbm{1}_{(-\infty, R_0^*]}\geq 0\,,
\end{align*}
as long as $\lp(-r^*\Re\hat{\mc{V}}_{(k)}\rp)'>0$ for $r^*\in[R_0^*,R_2^*]$ for sufficiently negative $R_2^*$.

Suppose that, for $R_2^*$ sufficiently negative, one can choose $\epsilon$ sufficiently small but independent of $R_2^*$ such that $64\epsilon \leq (R_2^*)^2\mc{V}_{(k)}(R_2^*)/4$. Then, if $\Re\hat{\mc{V}}_{(k)}'>0$ in the support of $\hat{y}$
\begin{align*}
-(\hat y\Re\hat{\mc{V}}_{(k)})'\mathbbm{1}_{[R_3^*,\infty)}&\geq -\hat y(R_2^*)\lp[\frac{64\epsilon/(-R_2^*)}{(r^*)^2}+\Re\hat{\mc{V}}_{(k)}'\lp(1-\frac{64\epsilon}{(R_2^*)^2\hat{\mc{V}}_{(k)}(R_2^*)}\rp)\rp]\mathbbm{1}_{[R_0^*,R_2^*]}\\
&\qquad-y\Re\mc{V}_{(k)}'\mathbbm{1}_{(-\infty,R_0^*]}\lp(1-\frac{B\delta}{|R^*_0|^\delta}\rp)\\
&\geq \frac{64\epsilon |\hat{y}(R_2^*)|/R_2^*}{(r^*)^2}\mathbbm{1}_{[R_0^*, R_2^*]}-\frac34\hat y(R_3^*)\Re\hat{\mc{V}}_{(k)}'\mathbbm{1}_{[R_0^*, R_2^*]}-\frac34\hat y\Re\hat{\mc{V}}_{(k)}'\mathbbm{1}_{(-\infty,R_2^*]}\\
&\geq b\lp[R_2^*h'(R_2^*)+h(R_2^*)\rp]\frac{\tilde{p}\epsilon}{(r^*)^2}\mathbbm{1}_{[R_0^*, R_2^*]}-\frac34\hat y\Re\hat{\mc{V}}_{(k)}'\mathbbm{1}_{(-\infty,R_2^*]}\,. \numberthis \label{eq:low-frequencies-yV-hat}
\end{align*}

If $h$ satisfies the estimate 
\begin{align*}
h'' +(|s|-k)\frac{w'}{w}h'\leq \frac{Bp \lp[R_2^*h'(R_2^*)+h(R_2^*)\rp]}{(r^*)^2}\mathbbm{1}_{[R_0^*,R_2^*]}\,,
\end{align*}
then as long as $p\ll \tilde{p}$,
\begin{align*}
&\frac34 h(\mc{V}_{(k)} -\omega^2) +\frac34 \hat y'\omega^2-\frac34(\hat y\Re\hat{\mc{V}}_{(k)})'-\frac14\hat y'\Re\hat{\mc{V}}_{(k)} -\lp(h'' +(|s|-k)\frac{w'}{w}h'\rp)\mathbbm{1}_{[R_0^*,R_2^*]}\\
&\quad \geq \frac{1}{2}h\lp(\mc{V}_{(k)} -\omega^2\rp)+\frac34 y'\omega^2+b\lp[R_2^*h'(R_2^*)+h(R_2^*)\rp]\frac{\tilde{p}\epsilon}{(r^*)^2}\mathbbm{1}_{[R_0^*,R_2^*]}-\frac12 \hat y\Re\mc{V}_{(k)}'\mathbbm{1}_{(-\infty,R_2^*]}\,. \numberthis \label{eq:low-frequencies-h+y-hat}
\end{align*}

Note that the currents $y$ and $\hat{y}$ are designed to take advantage of the sign of $\Re\mc{V}_{(k)}'$ while simultaneously absorbing the error terms generated by derivatives of an $h$ current of the type $h_+$ \eqref{eq:h+-estimate-error} or, when $a$ is sufficiently close to $M$, $h_-$ \eqref{eq:h--estimate-error}.

\subsubsection{Boundary term estimates for entire transformed system}

In this section, we obtain some boundary term estimates for the lower level transformed variables. 

\begin{lemma} \label{lemma:bdry-term-via-constraint-bdd} Let $(\omega,m,\Lambda)\in\mc F_{\rm bdd}$ be an admissible frequency triple. For $k=0,\dots,|s|$, assume $\uppsi_{k}$ are solutions to the radial ODEs~\eqref{eq:transformed-k-separated} with outgoing boundary conditions. Then, we have the following boundary term estimates for $k<|s|$. For $s>0$, writing $\kappa=\frac{r_+-M}{2Mr_+}$, 
\begin{align*}
&\lp[(|s|-k-1)^2\kappa^2 +(\omega-m\upomega_+)^2\rp]|\swei{A}{s}_{k+1,\mc H^+}|^2\\
&\quad\leq  B\lp|\frac{\mathfrak{G}_{(k)}}{w}\rp|^2(-\infty)+B(\omega_{\rm width},\varepsilon_{\rm width})|\swei{A}{s}_{k,\mc H^+}|^2+B(\omega_{\rm width},\varepsilon_{\rm width})\sum_{i=0}^{k-1}|\swei{A}{s}_{i,\mc H^+}|^2\,.
\end{align*}
For $s<0$, we have instead 
\begin{align*}
&\omega^2|\swei{A}{s}_{k+1,\mc H^+}|^2\leq  B\lp|\frac{\mathfrak{G}_{(k)}}{w}\rp|^2(-\infty)+B(\omega_{\rm width},\varepsilon_{\rm width})|\swei{A}{s}_{k,\mc H^+}|^2+B(\omega_{\rm high},\varepsilon_{\rm width})\sum_{i=0}^{k-1}|\swei{A}{s}_{i,\mc H^+}|^2\,.
\end{align*}
\end{lemma}

\begin{proof}
Recall the radial ODE \eqref{eq:transformed-k-separated} for $s>0$. Dropping subscripts, using the definition of $\uppsi_{(k+1)}$ \eqref{eq:transformed-transport-separated}, and dividing through by $w$, we find that \eqref{eq:transformed-k-separated} can be recast as
\begin{align*}
&\lp(\frac{d}{dr^*}-i(\omega-m\upomega_+)\rp)\uppsi_{(k+1)}\Big|_{r=r_+}- 2(|s|-k-1)\kappa \uppsi_{(k+1)}(-\infty)\\
&\quad\qquad+\lp(|s|+k(2|s|-k-1)+\frac{r_+-M}{M}(1-3|s|+2s^2-3k(2|s|-k-1))\rp)\uppsi_{(k)}(-\infty)\\
&\quad\qquad +\lp(\Lambda-2am\omega-\frac{a}{M}\sign s(2|s|-2k-1) im\rp)\uppsi_{(k)}(-\infty)\\
&\quad= \frac{\mathfrak{G}_{(k)}}{w}\Big|_{r=r_+}+\sum_{i=0}^{k-1}\lp(a c_{s,\,k,\,i}^\Phi im+ac^{\rm id}_{s,\,k,\,i}\rp)\Big|_{r=r_+}\uppsi_{(i)}(-\infty)\,, 
\end{align*}
at $r^*=-\infty$, where $\kappa=\frac{r_+-M}{2Mr_+}$. Assuming outgoing boundary conditions, we have 
\begin{align*}
&2\lp[ (|s|-k-1)\kappa -i(\omega-m\upomega_+)\rp]\swei{A}{s}_{k+1,\mc H^+}= -\frac{\mathfrak{G}_{(k)}}{w}\Big|_{r=r_+}-\sum_{i=0}^{k-1}\lp(a c_{s,\,k,\,i}^\Phi im+ac^{\rm id}_{s,\,k,\,i}\rp)\Big|_{r=r_+}\swei{A}{s}_{i,\mc H^+}\\
&\qquad+\lp(|s|+k(2|s|-k-1)+\frac{r_+-M}{M}(1-3|s|+2s^2-3k(2|s|-k-1))\rp)\swei{A}{s}_{k,\mc H^+}\\
&\qquad +\lp(\Lambda-2am\omega-\frac{a}{M}\sign s(2|s|-2k-1) im\rp)\swei{A}{s}_{k,\mc H^+}\,.  \numberthis\label{eq:ode-bdry-constraint}
\end{align*}
Multiplying \eqref{eq:ode-bdry-constraint} by $\lp[ (|s|-k-1)\kappa +i(\omega-m\upomega_+) \rp]\overline{\swei{A}{s}_{k+1,\mc H^+}}$, we obtain
\begin{align*}
&\lp[(|s|-k-1)^2\kappa^2 +(\omega-m\upomega_+)^2\rp]|\swei{A}{s}_{k+1,\mc H^+}|^2\\
&\leq  B\lp|\frac{\mathfrak{G}_{(k)}}{w}\rp|^2(-\infty)+B(1+\Lambda+|am\omega|)|\swei{A}{s}_{k,\mc H^+}|^2+B(1+m^2)\sum_{i=0}^{k-1}|\swei{A}{s}_{i,\mc H^+}|^2\,.  
\end{align*}

Similarly, if $s<0$, we can rewrite \eqref{eq:transformed-k-separated} as 
\begin{align*}
&\lp(-\frac{d}{dr^*}-i\omega\rp)\uppsi_{(k+1)}\Big|_{r=\infty}+\lp(\Lambda+|s|+k(2|s|-k-1)\rp)\uppsi_{(k)}\Big|_{r=\infty}\\
&\quad= \frac{\mathfrak{G}_{(k)}}{w}\Big|_{r=\infty}+\sum_{i=0}^{k-1}\lp(a c_{s,\,k,\,i}^\Phi im+ac^{\rm id}_{s,\,k,\,i}\rp)\Big|_{r=\infty}\uppsi_{(i)}\Big|_{r=\infty}\,, 
\end{align*}
thus arguing as before concludes the proof.
\end{proof}

\subsubsection{Low frequencies and axisymmetry or small black hole angular momentum}
\label{sec:F-low-1}

In this section, we will focus on frequency triples in $\mc{F}_{\rm low,1}$. Concretely, we will show
\begin{proposition}[Estimates in $\mc{F}_{\rm low,1a}$] \label{prop:low1a} Fix $s\in\{0,\pm 1,\pm 2\}$ and $M>0$. Then, for all $\delta\in(0,1]$, $\omega_{\rm high}>0$, $\varepsilon^{-1}_{\rm width}>0$, for all $E,E_W>0$ such that one is sufficiently large, for all $\tilde{a}_0$ and $\omega_{\rm low}$ sufficiently small depending on the latter, for all $(a,\omega,m,\Lambda) \in \mc{F}_{\rm low,1a}(\omega_{\rm high},\varepsilon_{\rm width},\omega_{\rm low},\tilde{a}_0)$, there exist functions $y$, $\hat{y}$, $h$, $\chi_1$ and $\chi_2$ satisfying the uniform bounds
$|y| +|\hat{y}|+|h|+|\chi_1|+|\chi_2|+|\chi_3|+|\chi_4|\leq B \,,$
such that, for all smooth $\Psi$ arising from a smooth solution to the radial ODE~\eqref{eq:radial-ODE-alpha} via \eqref{eq:def-psi0-separated} and \eqref{eq:transformed-transport-separated} and itself satisfying the radial ODE~\eqref{eq:radial-ODE-Psi},  if $\Psi$ has the general boundary conditions of Lemma~\ref{lemma:uppsi-general-asymptotics}, we have the estimate
\begin{align*}
&b\lp[\omega^2\lp|\swei{A}{s}_{\mc{I}^+}\rp|^2+(\omega-m\upomega_+)^2\lp|\swei{A}{s}_{\mc{H}^+}\rp|^2\rp] + b(\delta)\int_{-\infty}^\infty \frac{\Delta}{r^2}\Big[\frac{1}{r^{1+\delta}}|\Psi'|^2+\frac{1}{r^3}(r^{-\delta}+\Lambda)|\Psi|^2\Big]dr^*\\
&\quad\leq B(E,E_W)\lp[\omega^2\lp|\swei{A}{s}_{\mc{I}^-}\rp|^2+(\omega-m\upomega_+)^2\lp|\swei{A}{s}_{\mc{H}^-}\rp|^2\rp]\\
&\quad\qquad+\int_{-\infty}^\infty \lp\{2(y+\hat{y})\Re\lp[\mathfrak{G}\overline{\Psi}'\rp]+h\Re\lp[\mathfrak{G}\overline{\Psi}\rp]-E\lp(\chi_2\omega+\chi_1(\omega-m\upomega_+)\rp)\Im\lp[\mathfrak{G}\overline{\Psi}\rp]\rp\}dr^*\\
&\quad\qquad+\int_{-\infty}^\infty E_W\lp[\chi_3\lp(Q^{W,K}\rp)'+\chi_4\lp(Q^{W,T}\rp)'\rp]\,.
\end{align*}
If, for $k=0,\dots|s|$, $\uppsi_{(k)}$ have outgoing boundary conditions as in Definition \ref{def:outgoing-bdry-uppsi}, then we must take $E_W\equiv 0$ and either add \eqref{eq:ode-estimates-outgoing-addition} to the above right hand side.
\end{proposition}

\begin{proposition}[Estimates in $\mc{F}_{\rm low,1b}$] \label{prop:low1b} Fix $s\in\{0,\pm 1,\pm 2\}$ and $M>0$. Then, for all $\delta\in(0,1]$, $\omega_{\rm high}>0$, $\varepsilon^{-1}_{\rm width}>0$ and $\tilde{a}_1>0$, for all $E,E_W>0$ such that one is sufficiently large,  for all $\omega_{\rm low}$ sufficiently small depending on the latter quantities, and for all $(a,\omega,m,\Lambda) \in \mc{F}_{\rm low,1b}(\varepsilon_{\rm width},\omega_{\rm high},\tilde{a}_1)$ with $m=0$,  there exist functions $y_{(k)}$, $\hat{y}_{(k)}$ and $h_{(k)}$,  for $0\leq k\leq |s|$ (we drop the subscript for $k=|s|$), satisfying the uniform bounds
$|y_{(k)}| +|\hat{y}_{(k)}|+|h_{(k)}|\leq B(\tilde{a}_1) \,,$
such that, for all smooth $\Psi$ arising from a smooth solution to the radial ODE~\eqref{eq:radial-ODE-alpha} via \eqref{eq:def-psi0-separated} and \eqref{eq:transformed-transport-separated} and itself satisfying the radial ODE~\eqref{eq:radial-ODE-Psi},  if $\Psi$ has the general boundary conditions of Lemma~\ref{lemma:uppsi-general-asymptotics}, we have the estimate
\begin{align*}
&\omega^2\lp|\swei{A}{s}_{\mc{I}^+}\rp|^2+(\omega-m\upomega_+)^2\lp|\swei{A}{s}_{\mc{H}^+}\rp|^2 + \int_{-\infty}^\infty \frac{\Delta}{r^2}\Big[r^{-1-\delta}|\Psi'|^2+r^{-3}(\Lambda+r^{-\delta})|\Psi|^2\Big]dr^*\\
&\quad\leq B(E_W)\lp[\omega^2\lp|\swei{A}{s}_{\mc{I}^-}\rp|^2+(\omega-m\upomega_+)^2\lp|\swei{A}{s}_{\mc{H}^-}\rp|^2\rp]\\
&\quad\qquad+\int_{-\infty}^\infty \lp\{ 2(y+\hat{y})\Re\lp[\mathfrak{G}\overline{\Psi}'\rp]+h\Re\lp[\mathfrak{G}\overline{\Psi}\rp]+E_W\lp(Q^{W,\,T}\rp)'-E\omega\Im\lp[\mathfrak{G}\overline{\Psi}\rp]\rp\}dr^*\\
&\quad\qquad+B\sum_{k=0}^{|s|-1}\int_{-\infty}^\infty \lp\{ 2(y_{(k)}+\hat{y}_{(k)})\Re\lp[\mathfrak{G}_{(k)}\overline{\uppsi_{(k)}}'\rp]+h_{(k)}\Re\lp[\mathfrak{G}_{(k)}\overline{\uppsi_{(k)}}\rp]-E\omega\Im\lp[\mathfrak{G}_{(k)}\overline{\uppsi_{(k)}}\rp]\rp\}dr^*\,.
\end{align*}
If, for $k=0,\dots|s|$, $\uppsi_{(k)}$ have outgoing boundary conditions as in Definition \ref{def:outgoing-bdry-uppsi}, then we must take $E_W\equiv 0$ and add \eqref{eq:ode-estimates-outgoing-addition} to the above right hand side.
\end{proposition}

\begin{proposition}[Estimates in $\mc{F}_{\rm low,1c}$] \label{prop:low1c} Fix $s\in\{0,\pm 1,\pm 2\}$ and $M>0$. Then, for all $\delta\in(0,1]$, $\omega_{\rm high}>0$ and $\varepsilon^{-1}_{\rm width}>0$, for all $E,E_W>0$ such that one is sufficiently large, for all $\tilde{a}_1$ sufficiently small, for all  $\omega_{\rm low}$  depending on $\omega_{\rm high}$, $\varepsilon_{\rm width}$ and $\tilde{a}_1$, and for all $(a,\omega,m,\Lambda) \in \mc{F}_{\rm low,1c}(\omega_{\rm high},\varepsilon_{\rm width},\omega_{\rm low},\tilde{a}_1)$, there exist functions $y_{(k)}$, $\hat{y}_{(k)}$ and $h_{(k)}$, for $k=0,\cdots, |s|$ (we drop the subscript for $k=|s|$), satisfying the uniform bounds
$|y_{(k)}| +|\hat{y}_{(k)}|+|h_{(k)}|\leq B \,,$
such that, for all smooth $\Psi$ arising from a smooth solution to the radial ODE~\eqref{eq:radial-ODE-alpha} via \eqref{eq:def-psi0-separated} and \eqref{eq:transformed-transport-separated} and itself satisfying the radial ODE~\eqref{eq:radial-ODE-Psi},  if $\Psi$ has the general boundary conditions of Lemma~\ref{lemma:uppsi-general-asymptotics} or the outgoing boundary conditions of Definition \ref{def:outgoing-bdry-uppsi}, we have the estimate
\begin{align*}
&b\lp[\omega^2\lp|\swei{A}{s}_{\mc{I}^+}\rp|^2+(\omega-m\upomega_+)^2\lp|\swei{A}{s}_{\mc{H}^+}\rp|^2\rp] + b(\delta)\int_{-\infty}^\infty \frac{\Delta}{r^2}\Big[r^{-1-\delta}|\Psi'|^2+\frac{(r-M)}{r^{4}}(\Lambda+r^{-\delta})|\Psi|^2\Big]dr^*\\
&\quad\leq B(E_W)\lp[\omega^2\lp|\swei{A}{s}_{\mc{I}^-}\rp|^2+(\omega-m\upomega_+)^2\lp|\swei{A}{s}_{\mc{H}^-}\rp|^2\rp]\\
&\quad\qquad+\int_{-\infty}^\infty \lp\{ 2(y+\hat{y})\Re\lp[\mathfrak{G}\overline{\Psi}'\rp]+h\Re\lp[\mathfrak{G}\overline{\Psi}\rp]+E_W\lp(Q^{W,\,T}\rp)'-E\omega\Im\lp[\mathfrak{G}\overline{\Psi}\rp]\rp\}dr^*\\
&\quad\qquad+B\sum_{k=0}^{|s|-1}\int_{-\infty}^\infty \lp\{ 2(y_{(k)}+\hat{y}_{(k)})\Re\lp[\mathfrak{G}_{(k)}\overline{\uppsi_{(k)}}'\rp]+h_{(k)}\Re\lp[\mathfrak{G}_{(k)}\overline{\uppsi_{(k)}}\rp]-E\omega\Im\lp[\mathfrak{G}_{(k)}\overline{\uppsi_{(k)}}\rp]\rp\}dr^*\,.
\end{align*}
\end{proposition}

We briefly review the strategy. Recall that in $\mc{F}_{\rm low,1}$, both $|\omega|$ and $|\omega-m\upomega_+|$ are small. By Lemma~\ref{lemma:properties-potential-low-omega}, we have $\mc{V}-\omega^2>0$ in a large compact range of $r^*$ values, which we can exploit using a compactly supported $h$ current. Moreover, $\mc{V}'$ has a definite sign as $r^*\to\pm \infty$, so we connect the $h$ current with $y$ currents supported near those ends. The $y$ currents are built so that they start off within the region where $h$ gives a good estimate, but provide strong estimates, and furthermore could absorb the $h$ errors if they are integrable, when we start decreasing $h$ to zero. Finally, we add energy currents which, if localized (when superradiance can occur), produce localization errors in a region where $h$ is strongest.

To be more concrete, let us distinguish between the three subranges of $\mc{F}_{\rm low,1}$. In $\mc{F}_{\rm low,1a}$, since $a$ is small, the coupling terms in the radial ODE for $\Psi$ \eqref{eq:radial-ODE-Psi} can actually be seen as lower order. Thus, our combination of $h$ and $y$ currents need only be applied at the level of \eqref{eq:radial-ODE-Psi} and localized  Killing or Teukolsky--Starobinsky energy currents can be used to control the boundary terms in this possibly superradiant regime (see Figure~\ref{fig:low1a}).

\begin{figure}[htbp]
\centering
\includegraphics[scale=1]{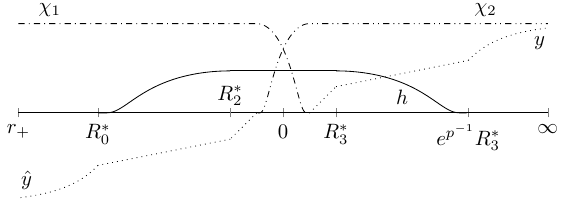}
\caption{Currents in the proof of Theorems~\ref{thm:ode-estimates-Psi-A} and \ref{thm:ode-estimates-Psi-B} for the frequency ranges $\mc{F}_{\rm low,1a}$ (relative size and scale not accurate).}
\label{fig:low1a}
\end{figure}

Turning to $\mc{F}_{\rm low,1b}$ and $\mc{F}_{\rm low,1c}$, we find that the coupling terms are not necessarily lower order. One consequence of this fact is that we find it necessary to apply our virial and Killing energy current techniques to the entire system of transformed equations together in order to close an estimate at the level of $\Psi$. The $h$ and $y$ currents as previously described are applied to each $\uppsi_{(k)}$, $k=0,\dots,|s|$, but a difficulty occurs for $k<|s|$: the subextremal ends $r^*\to \infty$ and $r^*\to \pm \infty$ are quite different and thus, while in the $h$ error term formula
$$h''+(|s|-k)\frac{w'}{w}h'$$
the $w'/w$ term contributes with decay as $r^*\to \infty$, it does not as $r^*\to -\infty$. The upshot is that, for $k<|s|$, the $h$ error terms may come with non-integrable weights as $r^*\to \infty$ and thus we may not be able to absorb them by application of a $y$ current at the level of the $k$th equation. We distinguish between $\mc{F}_{\rm low,1b}$ and $\mc{F}_{\rm low,1c}$ because:
\begin{itemize}
\item $\mc{F}_{\rm low,1c}$ is a range where $a$ is so close to extremality that we can set up an $h$ current where $h''+(|s|-k)\frac{w'}{w}h'$ is bounded by an integrable function, regardless of the fact that the background may be subextremal. Then, from the point of view of our estimates, the $r^*\to \pm \infty$  ends are symmetric and we can construct $y$ currents at each of these ends whose errors are absorbed by the $h$ current and that in turn absorb the $h$ current errors. We note that our construction of $h$ is made possible not only by the fact that $M-a$ is small but also crucially by the compact support assumption on $h$ and the fine tuning of that support in terms of the difference $M-a$. 
\item $\mc{F}_{\rm low,1b}$ is a range where $a$ is genuinely subextremal. Again a $y$ current supported near $r=r_+$ can be made to grow from zero so that its errors are absorved by the $h$ current. But now, if $k<|s|$, we can bound $h''+(|s|-k)\frac{w'}{w}h'$ uniformly in terms of the size of the support of $h$ by a function which is $O(|r^*|^{-1})$ only as $r^*\to -\infty$, hence the $y$ current cannot absorb the $h$ current error. Instead, for each $k$ we use a version of the basic estimate~\ref{eq:basic-estimate-1} from Lemma~\ref{lemma:basic-estimate-1}, together with the subextremality assumption, to convert the error into a boundary term of $\uppsi_{(k)}$ and a bulk term in $\uppsi_{(k+1)}$ which we control by our application of the $y$ current at that level. Therefore, we can still absorb the $h$ errors by application of a $y$ current, but the latter is at a higher $k$ level than the former.
\end{itemize}
For $k<|s|$, after converting the dependence of $y$ estimates on $\Im\mc{V}_{(k)}$ and $\uppsi_{(k)}'$ into a coupling to the $(k+1)$th equation in the system by Lemma~\ref{lemma:h-y-identity-k}, we select $h$ (hence both $y$ currents, which are started off by $h$; see Figure~\ref{fig:low1k}) so as to make that coupling appear with a small constant. Then, setting the inhomogeneity aside, a good bulk term in $k$ can be controlled by, if $k\neq 0$, the bulk terms in $j<k$ and, if $k<|s|$, by a small multiple of a bulk term in $k+1$. By iterating along the system, we can close an estimate for $\Psi$, i.e.\ at the level $k=|s|$.

\begin{figure}[htbp]
\centering
\includegraphics[scale=1]{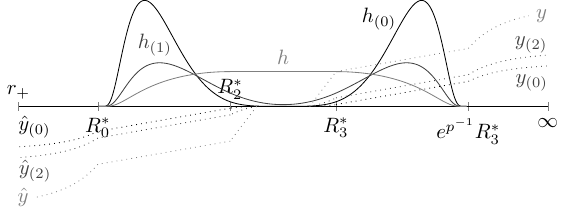}
\caption{Virial currents in the proof of Theorems~\ref{thm:ode-estimates-Psi-A} and \ref{thm:ode-estimates-Psi-B} for the frequency ranges $\mc{F}_{\rm low,1b}$ and $\mc{F}_{\rm low,1c}$ (relative size and scale not accurate).}
\label{fig:low1k}
\end{figure}

Concretely:

\begin{proof}[Proof of Proposition~\ref{prop:low1c}] Let $(\omega,m,\Lambda)\in \mc{F}_{\rm low,1c}$ be an admissible frequency triple, i.e.\ assume all frequency parameters are bounded in terms of $\varepsilon_{\rm width}$ and $\omega_{\rm high}$, $m=0$, $|\omega|\leq \omega_{\rm low}$  is sufficiently small and, for some sufficiently small $\tilde{a}_0$ we have $0\leq M-a<\tilde{a}_1$ .

Let $R_3^*$ be a sufficiently large and positive value, $R_2^*$ be sufficiently negative. Let $p\in(0,1)$ be defined in terms of $M$ and $\tilde{a}_1$ by $\sqrt{M^2-(M-\tilde{a}_1)^2}= e^{-2/p}$ and  define $R_1=r_++(R_2-r_+)e^{-1/p}$ and $R_{0}=r_++(R_1-r_+)/2$. The scaling of the $r$-values $R_1$ and $R_0$ with $p$ is critical to our proof: it ensures that, in practice, the present regime can be treated as if the Kerr black hole were exactly extremal, since
\begin{align}
\Delta(r-M)^{-2} = 1+\lp(\frac{\sqrt{M^2-a^2}}{r-M}\rp)^2 =1+O(p)\,, \label{eq:proof-low-1c-critical}
\end{align}
for $r\geq R_0$.

\medskip
\noindent\textit{The $h$ current.} In $\mc{F}_{\rm low,1c}$, we apply an $h$ current to each equation in the transformed system. Thus, for each $k=0,\dots, |s|$, we consider the function $h_{(k)}$ composed by
\begin{align*}
h_{(k)}=h_{\rm left}+h_{k,\mr{int}}+h_{\rm right}\,.
\end{align*}
Here, $h_{k,\mr{int}}$ is supported in the large region $r^*\in[R_2^*,R_3^*]$ and equal to 
\begin{align*}
h_{k,\mr{int}}= \lp(\frac{w}{\max\{w(R_3^*),w(R_2^*)\}}\rp)^{-(s-k)/2}\mathbbm{1}_{[R_2^*,R_3^*]}\,,
\end{align*}
which produces no errors in terms of $|\uppsi_{(k)}|^2$ by \eqref{eq:h-estimate-low} when $\sigma=1/4$ (see Lemma~\ref{lemma:h-estimate-low}).

At the endpoints of the interval $[R_2^*,R_3^*]$, we glue in the functions $h_{\rm left}$ and $h_{\rm right}$ in a $C^1$ fashion; these are modeled after the currents introduced in \eqref{eq:standard-current-h+} and \eqref{eq:standard-current-h-} respectively.  Indeed:
\begin{align*}
h_{\rm left}&=  \lp[C_{k,0}^-\lp(p-\frac{(1+p)\sqrt{M^2-a^2}}{R_1-r_+}\rp) \log\lp(\frac{r-M}{R_2-M}\rp)+ph_{(k)}(R_2)\frac{r-M}{R_2-M}\rp]\mathbbm{1}_{[R_1,R_2]} \\
&\qquad +C_{k,0}^-\lp[(1+p)\frac{r-r_+}{R_2-r_+}\frac{R_2-M}{r-M}-(1+p)\frac{r-M}{R_2-M}\rp]\mathbbm{1}_{[R_1,R_2]}\\
&\qquad + \lp[C_{N+1}^-\frac{r-r_+}{R_1-r_+}\frac{R_1-M}{r-M}+C_{N+2}^-\frac{R_1-M}{r-M}+C_{N+3}^-\log\lp(\frac{R_1-M}{r-M}\rp)+C_{N+4}^-\frac{r-M}{R_1-M}\rp]\mathbbm{1}_{[R_{0},R_1]}\,,\\
C_{k,0}^-&:=h_{(k)}(R_2)-(R_2-M)\frac{dh_{(k)}}{dr}(R_2)\,, 
\end{align*}
where $C_{N+1}^-$ up to $C_{N+4}^-$  are fixed in terms of the remaining constants by the constraints that $h_{\rm left}$ is $C^1$ at $r=R_1$ and $h_{\rm left}(R_{0})=h_{\rm left}'(R_{0})=0$ (see also \eqref{eq:coefficients-h-}); moreover,
\begin{align*}
h_{\rm right}&=C_{k,0}^+\lp[(1+p)-(1+p)\frac{R_3^*}{r^*}+p\log\lp(\frac{R_3^*}{r^*}\rp)\rp]\mathbbm{1}_{[R_3^*,e^{p^{-1}}R_3^*]}+ph_{(k)}(R_3^*)\frac{R_3^*}{r^*}\mathbbm{1}_{[R_3^*,e^{p^{-1}}R_3^*]}\\
&\qquad + \frac{2}{\log 8 -2}\lp[p(1-e^{-1/p})C_{k,0}^+-e^{-1/p}R_3^*h_{(k)}'(R_3^*)\rp]\\
&\qquad\qquad\times \lp[1-\frac{2e^{p^{-1}}R_3^*}{r^*}+\lp(\frac{r^*}{2e^{p^{-1}}R_3^*}+\frac{2e^{p^{-1}}R_3^*}{r^*}\rp)\log 2-\log\lp(\frac{2r^*}{e^{p^{-1}}R_3^*}\rp)\rp]\mathbbm{1}_{[e^{p^{-1}}R_3^*,2e^{p^{-1}}R_3^*]},,\\
C_{k,0}^+&=h_{(k)}(R_3^*)+R_3^*h_{(k)}'(R_3^*)\,.
\end{align*}
Note that, by \eqref{eq:h+-estimate-error} and \eqref{eq:h--estimate-error}, the errors generated by such currents,
\begin{gather*}
h_{(k)}''\mathbbm{1}_{[R_2^*,R_3^*]^c}\,,\,\,\lp(h_{(k)}''+(|s|-k)\frac{w'}{w}h_{(k)}'\rp)\mathbbm{1}_{[R_2^*,R_3^*]^c} \leq \frac{Bp }{(r^*)^2}\lp[C_{k,0}^-\mathbbm{1}_{[R_0,R_2]}+C_{k,0}^+\mathbbm{1}_{[R_3^*,2e^{p^{-1}}R_3^*]}\rp]\,,
\end{gather*}
are integrable in $r^*$ and come with a small parameter $p$. 

In conclusion, application of the $h_{(k)}$ current to the radial ODE \eqref{eq:transformed-k-separated} gives the estimate (See Lemma~\ref{lemma:h-estimate-low})
\begin{align*}
&\frac34\int_{-\infty}^\infty \lp[h_{(k)}\lp(\Re\mc{V}_{(k)}-\frac{|s|-k}{2}\lp(\frac{w'}{w}\rp)'-\omega^2\rp)- h_{k,\mr{int}}\frac{|s|-k}{4}\lp(\frac{w'}{w}\rp)'\rp]\lp|\uppsi_{(k)}\rp|^2dr^*\\
&\qquad+\int_{-\infty}^\infty \lp\{h\lp|\uppsi_{(k)}'\rp|^2 - \frac{Bp }{(r^*)^2}\lp[C_{k,0}^-\mathbbm{1}_{[R_0,R_2]}+C_{k,0}^+\mathbbm{1}_{[R_3^*,2e^{p^{-1}}R_3^*]}\rp]\lp|\uppsi_{(k)}\rp|^2-h_{(k)}\Re\lp[\mathfrak{G}_{(k)}\overline{\uppsi}_{(k)}\rp]\rp\}dr^*\\
&\quad\leq kB\sum_{j=0}^{k-1}\int_{-\infty}^\infty (h_{\rm left}+h_{\rm right})\lp(\Re\mc{V}_{(j)}-\frac{|s|-j}{2}\lp(\frac{w'}{w}\rp)'-\omega^2\rp)\lp|\uppsi_{(j)}\rp|^2dr^* \numberthis \label{eq:proof-low-1c-intermediate-1}\\
&\quad\qquad +kB\sum_{j=0}^{k-1}\int_{-\infty}^\infty h_{k,\mr{int}}\lp(\Re\mc{V}_{(j)}-\frac{|s|-j}{3}\lp(\frac{w'}{w}\rp)'-\omega^2\rp)\lp|\uppsi_{(j)}\rp|^2dr^*\,.
\end{align*}

\medskip
\noindent \textit{The $y$ and $\hat{y}$ currents.} The next step is to introduce currents that can absorb the error term introduced by $h_{\rm left}$ and $h_{\rm right}$ while also ensure adequate control over $\uppsi_{(k)}$ as $r^*\to \pm \infty$. This is done by using currents based on  the model currents \eqref{eq:low-frequencies-y} and \eqref{eq:low-frequencies-hat-y},  introduced in Section~\ref{sec:low-hierarchy-y}; in the present regime, we define and apply them for each $k=0,\dots, |s|$. Indeed, for some $\tilde{p}\in(0,3/32)$, we set
\begin{align*}
y_{(k)}&=\tilde{p}\mathbbm{1}_{[R_3^*/2,R_3^*]}\int_{\frac{R_3^*}{2}}^{r^*} h_{k,\mr{int}}(x^*)dx^* +y(R_3^*)\lp[1+\frac{64\epsilon}{R_3^*}\lp(\frac{1}{r^*\Re\mc{V}_{(k)}}-\frac{1}{R_3^*\Re\mc{V}_{(k)}|_{r^*=R_3^*}}\rp)\rp]\mathbbm{1}_{[R_3^*, 2e^{p^{-1}}R_3^*]} \\
&\qquad+y(2e^{p^{-1}}R_3^*)\lp[1+\frac{1}{(2e^{p^{-1}}R_3^*)^{\delta}}-\frac{1}{(r^*)^{\delta}}\rp]\mathbbm{1}_{[2e^{p^{-1}}R_3^*, \infty)}\geq 0\,,\\
\hat y_{(k)}&= -\tilde{p}\mathbbm{1}_{[R_2^*,R_2^*/2]}\int_{R_2^*/2}^{r^*} h_{k,\mr{int}}(x^*)dx^*+\hat y_{(k)}(R_2^*)\lp[1-\frac{64\epsilon}{(-R_2^*)}\lp(\frac{1}{r^*\hat{\mc{V}}_{(k)}}-\frac{1}{R_2^*\hat{\mc{V}}_{(k)}|_{r^*=R_2^*}}\rp)\rp]\mathbbm{1}_{[R_0^*,R_2^*]} \\
&\qquad+\hat{y}_{(k)}(R_0^*)\lp[1-\frac{1}{(-R_0^*)^{1/2}}+\frac{1}{(-r^*)^{1/2}}\rp]\mathbbm{1}_{(-\infty, R_0^*]}\leq 0\,.
\end{align*}
Noting the properties of the potential derived in Lemma~\ref{lemma:properties-potential-low-omega}, the critical observation~\eqref{eq:proof-low-1c-critical} and the considerations regarding the model currents \eqref{eq:low-frequencies-y} and \eqref{eq:low-frequencies-hat-y}, it is clear that, as long as $\tilde{a}_1$ or $p$ are sufficiently small, one can choose $\epsilon$ sufficiently small, $R_3^*$ sufficiently large and $R_2^*$ sufficiently negative so that $y',\hat{y}'>0$, and so that
\begin{align*}
\lp(1-\frac{r^*\hat{y}_{(k)}}{\hat{y}_{(k)}'}\rp)\mathbbm{1}_{(-\infty,R_2^*]}\geq b\mathbbm{1}_{(-\infty,R_2^*]}\,,\qquad\lp(1-\frac{r^*{y}_{(k)}}{{y}_{(k)}'}\rp)\mathbbm{1}_{[R_3^*,\infty)}\geq b\mathbbm{1}_{[R_3^*,\infty)}\,,
\end{align*}
for some $b>0$ independent of $R_3^*$. Moreover, if $\epsilon$  and $p$ are sufficiently small,
\begin{align*}
\lp(1-\frac{r^*\hat{y}_{(k)}}{\hat{y}_{(k)}'}\rp)\mathbbm{1}_{[R_0^*,R_2^*]} &\geq 1-\frac{64\epsilon(-r^*)(r-M)}{64\epsilon +\lp(1-\frac{64\epsilon}{(R_2^*)^2\hat{\mc{V}}(R(_2^*)}\rp)(-r^*)(-R_2^*)\hat{\mc{V}}_{(k)}}\frac{\lp(-r^*\hat{\mc{V}}\rp)'}{\hat{\mc{V}}(-r^*)(r-M)}\mathbbm{1}_{[R_0^*,R_2^*]}\\
&\geq 1-\frac{64\epsilon}{64\epsilon +\frac34(-r^*)(-R_2^*)\hat{\mc{V}}_{(k)}}\frac{4(-r^*)(r-r_+)+4(-r^*)e^{p^{-2}}}{2M^2}\mathbbm{1}_{[R_0^*,R_2^*]}\\
&\geq 1-\frac{6\times 64\epsilon}{64\epsilon +\frac34(R_2^*)^2\hat{\mc{V}}_{(k)}(R_2^*)}\mathbbm{1}_{[R_0^*,R_2^*]}\geq b\\
\lp(1-\frac{r^*\hat{y}_{(k)}}{\hat{y}_{(k)}'}\rp)\mathbbm{1}_{(-\infty,R_0^*]}&\geq 1-\frac{(-r^*)^{-1/2}}{2(1-(-R_0^*)^{-1/2}+(-r^*)^{-1/2}}\geq 1-\frac{1}{2R_0^*}\geq b\,. 
\end{align*}

Choosing $\tilde{p}\ll 1$, as long as $\tilde{a}_1$ is sufficiently small that $p\ll \tilde{p}$,  by adapting \eqref{eq:low-frequencies-h+y} and \eqref{eq:low-frequencies-h+y-hat}, we find
\begin{align*}
&\frac34 h_{k\,\mr{int}}\lp(\Re\mc{V}_{(k)}-\frac{|s|-k}{4}\lp(\frac{w'}{w}\rp)' -\omega^2\rp)-\frac{Bp}{(r^*)^2}\lp[C_0^-\mathbbm{1}_{[R_0,R_2]}+C_0^+\mathbbm{1}_{[R_3^*,2e^{p^{-1}}R_3^*]}\rp]\\
&\qquad+\lp[\frac34 y_{(k)}'\omega^2-\frac34 \lp(y_{(k)}\Re\mc{V}_{(k)}\rp)' -\frac14y_{(k)}\Re\mc{V}_{(k)}'\rp]\\
&\quad\qquad +\lp[ \frac34 \hat y'\omega^2-\frac34 \lp(\hat{y}_{(k)}\Re\hat{\mc{V}}_{(k)}\rp)' -\frac14 \hat{y}_{(k)}\Re\hat{\mc{V}}_{(k)}'\rp]\numberthis \\
&\quad\geq \frac12 h_{k\,\mr{int}}\lp(\Re\mc{V}_{(k)}-\frac{(|s|-k)}{3}\lp(\frac{w'}{w}\rp)' -\omega^2\rp) \\
&\quad\qquad+\frac34 \hat{y}_{(k)}'(\omega-m\upomega_+)^2 -\frac12 \hat{y}_{(k)}\Re\hat{\mc{V}}_{(k)}'\mathbbm{1}_{(-\infty,R_2^*]} +\frac34 y_{(k)}'\omega^2 -\frac12 y_{(k)}\Re\mc{V}_{(k)}'\mathbbm{1}_{[R_3^*,\infty)}\,.
\end{align*}

We want to add \eqref{eq:proof-low-1c-intermediate-1} and the estimates \eqref{eq:y-estimate-low-k} and \eqref{eq:hat-y-estimate-low-k} from Lemma~\ref{lemma:y-estimate-low} (the latter with the further simplification laid out below its statement). Note that, if $|\omega|\leq \omega_{\rm low}$ is sufficiently small, $\omega|R_3^*|,\omega|R_2^*|\ll 1$; on the other hand, if $R_3^*$ and $|R_2^*|$ are sufficiently large, the errors
\begin{align*}
&\int_{-\infty}^\infty \lp[y_{(k)}'w\lp(\frac{wy_{(k)}^2}{(y_{(k)}')^2}\rp)\mathbbm{1}_{[R_3^*/2,R_3^*]}+y'w\lp(\frac{w\hat{y}_{(k)}^2}{(\hat{y}_{(k)}')^2}\rp)\mathbbm{1}_{[R_2^*,R_2^*/2]}\rp]|\uppsi_{(k)}|^2\\
&\quad\ll \int_{-\infty}^\infty h_{k\,\mr{int}}\lp(\Re\mc{V}_{(k)}-\frac{|s|-k}{3}\lp(\frac{w'}{w}\rp)' -\omega^2\rp)|\uppsi_{(k)}|^2\,,
\end{align*}
are easily absorbed into the left hand side of \eqref{eq:proof-low-1c-intermediate-1}. We obtain
\begin{align*}
&\int_{-\infty}^\infty \lp[\frac34 \lp(\hat{y}_{(k)}'+y_{(k)}'\rp)\omega^2 -\frac12 \hat{y}_{(k)}\Re\hat{\mc{V}}_{(k)}'\mathbbm{1}_{(-\infty,R_2^*]}  -\frac12 y_{(k)}\Re\mc{V}_{(k)}'\mathbbm{1}_{[R_3^*,\infty)} \rp]|\uppsi_{(k)}|^2dr^*\\
&\qquad +\int_{-\infty}^\infty \lp\{\lp(y_{(k)}'+\hat y_{(k)}'+h\rp)|\uppsi_{(k)}'|^2+\frac14 h_{k,\mr{int}}\lp(\Re\mc{V}_{(k)}-\frac{|s|-k}{3}\lp(\frac{w'}{w}\rp)' -\omega^2\rp)\lp|\uppsi_{(k)}\rp|^2\rp\}dr^*\\
&\qquad +\int_{-\infty}^\infty \frac34 (h_{k,\mr{left}}+h_{k,\mr{right}})\lp(\Re\mc{V}_{(k)}-\frac{(|s|-k)}{2}\lp(\frac{w'}{w}\rp)' -\omega^2\rp)\lp|\uppsi_{(k)}\rp|^2dr^*\\
&\quad\leq  2y_{(k)}(\infty)\omega^2\lp|\swei{A}{s}_{k,\mc{I}^{+\sign s}}\rp|^2+2|\hat y_{(k)}(-\infty)|\omega^2\lp|\swei{A}{s}_{k,\mc{H}^{-\sign s}}\rp|^2\\
&\quad\qquad+\int_{-\infty}^\infty \lp\{h_{(k)}\Re\lp[\mathfrak{G}_{(k)}\overline{\uppsi_{(k)}}\rp] +2(y_{(k)}+\hat{y}_{(k)})\Re\lp[\mathfrak{G}_{(k)}\overline{\uppsi_{(k)}}'\rp]\rp\}dr^* \numberthis \label{eq:proof-low-1c-intermediate-2}\\
&\quad\qquad - kB\sum_{j=0}^{k-1} \int_{-\infty}^\infty \lp[y_{(k)}\Re\mc{V}'_{(j)}+\hat{y}_{(k)}\Re\hat{\mc{V}}'_{(j)}-\lp(R_3^*\mathbbm{1}_{\lp[\tfrac{R_3^*}{2},R_3^*\rp]} +|R_2^*|\mathbbm{1}_{\lp[R_2^*,\tfrac{R_2^*}{2}\rp]}\rp)h_{k,\mr{int}} w\rp]|\uppsi_{(j)}|^2dr^* \\
&\quad\qquad +kB\sum_{j=0}^{k-1} \int_{-\infty}^\infty \lp[h_{k}\lp(\Re\mc{V}_{(j)}-\frac{|s|-j}{2}\lp(\frac{w'}{w}\rp)' -\omega^2\rp)+h_{k,\mr{int}}\frac{|s|-j}{6}\lp(\frac{w'}{w}\rp)' \rp]|\uppsi_{(j)}|^2dr^*\\
&\quad\qquad - (k-1)B\sum_{j=0}^{k-2} \int_{-\infty}^\infty \lp[ \frac{1}{r^2}y_{(k)}\Re\mc{V}'_{(j+1)}+(r-r_+)^2\hat{y}_{(k)}\Re\hat{\mc{V}}_{(j+1)}'\rp]|\uppsi_{(j+1)}|^2dr^*\\
&\quad\qquad - (|s|-k)B\int_{-\infty}^\infty \omega^2 \lp[y_{(k)}\Re\mc{V}'_{(k+1)}+\hat{y}_{(k)}\Re\hat{\mc{V}}'_{(k+1)}\rp]|\uppsi_{(k+1)}|^2dr^*\,.
\end{align*}

By our definition of $h_{k,\mr{int}}$, $h_{j,\mr{int}}\mathbbm{1}_{[R_3^*/2,R_3^*]}=(R_3^*)^{k-j}h_{k,\mr{int}}\mathbbm{1}_{[R_3^*/2,R_3^*]}$; likewise, by our crucial assumption \eqref{eq:proof-low-1c-critical}, the same occurs at the other end, i.e.\  $h_{j,\mr{int}}\mathbbm{1}_{[R_2^*,R_2^*/2]}=(R_2^*)^{k-j}h_{k,\mr{int}}\mathbbm{1}_{[R_2^*,R_2^*/2]}$. The same gain is inherited by $y$ and $\hat{y}$. Consequently, lines 2 through 4 of \eqref{eq:proof-low-1c-intermediate-2} are controlled by the expression on the left hand side of \eqref{eq:proof-low-1c-intermediate-2} when $k$ is replaced by $j$. 
Thus, we have:
\begin{align*}
&\int_{-\infty}^\infty \lp[\frac34 \lp(\hat{y}_{(k)}'+y_{(k)}'\rp)\omega^2 -\frac12 \hat{y}_{(k)}\Re\hat{\mc{V}}_{(k)}'\mathbbm{1}_{(-\infty,R_2^*]}  -\frac12 y_{(k)}\Re\mc{V}_{(k)}'\mathbbm{1}_{[R_3^*,\infty)} \rp]|\uppsi_{(k)}|^2dr^*\\
&\qquad +\int_{-\infty}^\infty \lp\{\lp(y_{(k)}'+\hat y_{(k)}'+h\rp)|\uppsi_{(k)}'|^2+\frac14 h_{k,\mr{int}}\lp(\Re\mc{V}_{(k)}-\frac{|s|-k}{3}\lp(\frac{w'}{w}\rp)' -\omega^2\rp)\lp|\uppsi_{(k)}\rp|^2\rp\}dr^*\\
&\qquad +\int_{-\infty}^\infty \frac34 (h_{k,\mr{left}}+h_{k,\mr{right}})\lp(\Re\mc{V}_{(k)}-\frac{(|s|-k)}{2}\lp(\frac{w'}{w}\rp)' -\omega^2\rp)\lp|\uppsi_{(k)}\rp|^2dr^*\\
&\quad\leq  2y_{(k)}(\infty)\omega^2\lp|\swei{A}{s}_{k,\mc{I}^{+\sign s}}\rp|^2+2|\hat y_{(k)}(-\infty)|\omega^2\lp|\swei{A}{s}_{k,\mc{H}^{-\sign s}}\rp|^2\\
&\quad\qquad + kB\sum_{j=0}^{k-1} \int_{-\infty}^\infty \lp[h_{j,\mr{int}}\lp(\Re\mc{V}_{(j)}-\frac{(|s|-j)}{3}\lp(\frac{w'}{w}\rp)' -\omega^2\rp) - y_{(j)}\Re\mc{V}'_{(j)}\mathbbm{1}_{[R_3^*,\infty)}-\hat{y}_{(j)}\Re\hat{\mc{V}}'_{(j)}\mathbbm{1}_{(-\infty,R_2^*]}  \rp. \\
&\qquad\qquad\qquad\qquad\qquad\qquad\qquad\lp. +(h_{j}-h_{j,\mr{int}})\lp(\Re\mc{V}_{(j)}-\frac{(|s|-j)}{2}\lp(\frac{w'}{w}\rp)' -\omega^2\rp)\rp]|\uppsi_{(j)}|^2dr^* \\
&\qquad  -(|s|-k)\omega_{\rm low}\int_{-\infty}^\infty  \lp[y_{(k+1)}\Re\mc{V}'_{(k+1)}+\hat{y}_{(k+1)}\Re\hat{\mc{V}}'_{(k+1)}\rp]|\uppsi_{(k+1)}|^2dr^*\,. \numberthis \label{eq:proof-low-1c-intermediate-3}
\end{align*}

\medskip
\noindent \textit{The Killing energy currents.} We now need to control the boundary terms appearing on the right hand side of \eqref{eq:proof-low-1c-intermediate-3}. Since $m=0$ in $\mc{F}_{\rm low, 1c}$, there is no superradiance so we can apply global energy currents. From Lemmas~\ref{lemma:Killing-identity-k} and \ref{lemma:Killing-current-errors-bdd}, a multiple $E\propto \max_k\lp\{y_k(\infty),\hat y_k(-\infty)\rp\}$ of the Killing $T$-energy current generates errors which, choosing $\omega_{\rm low}$ sufficiently small depending on $R_3^*$, $p$, $\omega_{\rm high}$ and $\varepsilon_{\rm width}^{-1}$, can be controlled by the left hand side and last line of \eqref{eq:proof-low-1c-intermediate-3}. Hence, we have
\begin{align*}
&\int_{-\infty}^\infty \frac12\lp[\frac34 \lp(\hat{y}_{(k)}'+y_{(k)}'\rp)\omega^2 -\frac12 \hat{y}_{(k)}\Re\hat{\mc{V}}_{(k)}'\mathbbm{1}_{(-\infty,R_2^*]}  -\frac12 y_{(k)}\Re\mc{V}_{(k)}'\mathbbm{1}_{[R_3^*,\infty)} \rp]|\uppsi_{(k)}|^2dr^*\\
&\qquad +\int_{-\infty}^\infty \frac12\lp\{\lp(y_{(k)}'+\hat y_{(k)}'+h\rp)|\uppsi_{(k)}'|^2+\frac14 h_{k,\mr{int}}\lp(\Re\mc{V}_{(k)}-\frac{|s|-k}{3}\lp(\frac{w'}{w}\rp)' -\omega^2\rp)\lp|\uppsi_{(k)}\rp|^2\rp\}dr^*\\
&\qquad +\int_{-\infty}^\infty \frac38 (h_{k,\mr{left}}+h_{k,\mr{right}})\lp(\Re\mc{V}_{(k)}-\frac{(|s|-k)}{2}\lp(\frac{w'}{w}\rp)' -\omega^2\rp)\lp|\uppsi_{(k)}\rp|^2dr^*\\
&\qquad +y_{(k)}(\infty)\omega^2\lp|\swei{A}{s}_{k,\mc{I}^{+}}\rp|^2+|\hat y_{(k)}(-\infty)|\omega^2\lp|\swei{A}{s}_{k,\mc{H}^{+}}\rp|^2\\
&\quad\leq 5E\omega^2\lp|\swei{A}{s}_{k,\mc{I}^{+}}\rp|^2+5E\omega^2\lp|\swei{A}{s}_{k,\mc{H}^{-}}\rp|^2\\
&\quad\qquad + kB\sum_{j=0}^{k-1} \int_{-\infty}^\infty \lp[h_{j,\mr{int}}\lp(\Re\mc{V}_{(j)}-\frac{(|s|-j)}{3}\lp(\frac{w'}{w}\rp)' -\omega^2\rp) - y_{(j)}\Re\mc{V}'_{(j)}\mathbbm{1}_{[R_3^*,\infty)}-\hat{y}_{(j)}\Re\hat{\mc{V}}'_{(j)}\mathbbm{1}_{(-\infty,R_2^*]}  \rp. \\
&\qquad\qquad\qquad\qquad\qquad\qquad\qquad\lp. +(h_{j}-h_{j,\mr{int}})\lp(\Re\mc{V}_{(j)}-\frac{(|s|-j)}{2}\lp(\frac{w'}{w}\rp)' -\omega^2\rp)\rp]|\uppsi_{(j)}|^2dr^* \\
&\qquad  -(|s|-k)\omega_{\rm low}\int_{-\infty}^\infty  \lp[y_{(k+1)}\Re\mc{V}'_{(k+1)}+\hat{y}_{(k+1)}\Re\hat{\mc{V}}'_{(k+1)}\rp]|\uppsi_{(k+1)}|^2dr^*\,. \numberthis \label{eq:proof-low-1c-intermediate-3b}
\end{align*}
Iterating for $k=0,\dots,|s|$, we obtain
\begin{align*}
&\frac{1}{16}\int_{-\infty}^\infty \lp[\hat{y}_{(k)}'(\omega-m\upomega_+)^2 - \hat{y}_{(k)}\Re\hat{\mc{V}}_{(k)}'\mathbbm{1}_{(-\infty,R_2^*]} + y_{(k)}'\omega^2 - y_{(k)}\Re\mc{V}_{(k)}'\mathbbm{1}_{[R_3^*,\infty)} \rp]|\uppsi_{(k)}|^2dr^*\\
&\qquad +\frac{1}{16}\int_{-\infty}^\infty \lp\{\lp(y_{(k)}'+\hat y_{(k)}'+h\rp)|\uppsi_{(k)}'|^2+ h_{k,\mr{int}}\lp(\Re\mc{V}_{(k)}-\frac{3(|s|-k)}{4}\lp(\frac{w'}{w}\rp)' -\omega^2\rp)\lp|\uppsi_{(k)}\rp|^2\rp\}dr^*\\
&\qquad+\frac{1}{16}\int_{-\infty}^\infty  (h_{k,\mr{left}}+h_{k,\mr{right}})\lp(\Re\mc{V}_{(k)}-\frac{(|s|-k)}{2}\lp(\frac{w'}{w}\rp)' -\omega^2\rp)\lp|\uppsi_{(k)}\rp|^2dr^*\\
&\qquad +E\omega^2\sum_{j=0}^k\lp\{\lp|\swei{A}{s}_{j,\mc{I}^{+}}\rp|^2+\lp|\swei{A}{s}_{j,\mc{H}^{+}}\rp|^2\rp\}\\
&\quad\leq  BE\omega^2\sum_{j=0}^{k} \lp\{\lp|\swei{A}{s}_{j,\mc{I}^{-}}\rp|^2+\lp|\swei{A}{s}_{j,\mc{H}^{-}}\rp|^2\rp\}-B(|s|-k)\int_{-\infty}^\infty  \lp[y_{(k+1)}\Re\mc{V}'_{(k+1)}+\hat{y}_{(k+1)}\Re\hat{\mc{V}}'_{(k+1)}\rp]|\uppsi_{(k+1)}|^2dr^*
\\&\quad\qquad +B\sum_{j=0}^{k}\int_{-\infty}^\infty \lp\{h_{(j)}\Re\lp[\mathfrak{G}_{(j)}\overline{\uppsi_{(j)}}\rp] +2(y_{(j)}+\hat{y}_{(j)})\Re\lp[\mathfrak{G}_{(j)}\overline{\uppsi_{(j)}}'\rp]-E\omega\Im\lp[\mathfrak{G}_{(j)}\overline{\uppsi_{(j)}}\rp]\rp\}dr^* 
\numberthis \label{eq:proof-low-1c-intermediate-4}\,.
\end{align*}
We can now rescale the currents so that $\max_k\{y_{(k)}(\infty),|\hat{y}_{(k)}(-\infty)|\}=1$, and so we may drop $E$ from above.

\medskip
\noindent \textit{The Teukolsky--Starobinsky energy currents.} Alternatively to the last step, we can iterate \eqref{eq:proof-low-1c-intermediate-3b} in $k=0,\dots |s|$ and then introduce a multiple, $E_W$, of a $T$-type Teukolsky--Starobinsky energy current, which does not interact with the coupling, at the level $k=|s|$. Then, the boundary terms are
\begin{align*}
&\lp[2y(\infty)\lp\{1,1+B\sum_{j=0}^{|s|-1} \frac{y_{(j)}(\infty)}{y(\infty)} \frac{\lp|\swei{A}{s}_{j,\mc{I}^+}\rp|^2}{\lp|\swei{A}{s}_{\mc{I}^+}\rp|^2}\rp\}-E_W\lp\{1,\frac{\mathfrak{C}_s}{\mathfrak{D}_s^{\mc{I}}}-E\rp\}\rp]\omega^2\lp|\swei{A}{s}_{\mc{I}^+}\rp|^2\\
&\qquad+\lp[2y(\infty)\lp\{1+B\sum_{j=0}^{|s|-1} \frac{y_{(j)}(\infty)}{y(\infty)}\frac{\lp|\swei{A}{s}_{j,\mc{I}^-}\rp|^2}{\lp|\swei{A}{s}_{\mc{I}^-}\rp|^2},1\rp\}+E_W\lp\{\frac{\mathfrak{C}_s}{\mathfrak{D}_s^{\mc{I}}},1\rp\}+E\rp]\omega^2\lp|\swei{A}{s}_{\mc{I}^-}\rp|^2\\
&\qquad +\lp[2|\hat y(-\infty)|\lp\{1+B\sum_{j=0}^{|s|-1} \frac{\hat{y}_{(j)}(-\infty)}{\hat y(-\infty)} \frac{\lp|\swei{A}{s}_{j,\mc{H}^+}\rp|^2}{\lp|\swei{A}{s}_{\mc{H}^+}\rp|^2},1\rp\}-E_W\lp\{\frac{\mathfrak{C}_s}{\mathfrak{D}_s^{\mc{H}}},1\rp\}-E\rp]\omega^2\lp|\swei{A}{s}_{\mc{H}^+}\rp|^2\\
&\qquad+\lp[2|\hat{y}(-\infty)|\lp\{1,1+B\sum_{j=0}^{|s|-1} \frac{\hat{y}_{(j)}(-\infty)}{\hat y(-\infty)}\frac{\lp|\swei{A}{s}_{j,\mc{H}^-}\rp|^2}{\lp|\swei{A}{s}_{\mc{H}^-}\rp|^2}\rp\}+E_W\lp\{1,\frac{\mathfrak{C}_s}{\mathfrak{D}_s^{\mc{H}}}\rp\}+E\rp]\omega^2\lp|\swei{A}{s}_{\mc{H}^-}\rp|^2\\
&\quad \leq \omega^2\lp\{\lp[3y(\infty)-\frac13 E_W-E\rp]\lp|\swei{A}{s}_{\mc{I}^+}\rp|^2+\lp[3y(\infty)+\frac53 E_W+E\rp]\lp|\swei{A}{s}_{\mc{I}^-}\rp|^2\rp\}\\
&\qquad\quad+\omega^2\lp\{\lp[3|\hat y(-\infty)|-\frac13 E_W-E\rp]\lp|\swei{A}{s}_{\mc{H}^+}\rp|^2+\lp[3|\hat y(-\infty)|+\frac53 E_W+E\rp]\lp|\swei{A}{s}_{\mc{H}^-}\rp|^2\rp\}\\
&\leq \omega^2\lp\{-\lp|\swei{A}{s}_{\mc{I}^+}\rp|^2-\lp|\swei{A}{s}_{\mc{H}^+}\rp|^2+(5E_W+3E)\lp|\swei{A}{s}_{\mc{I}^-}\rp|^2+(5E_W+3E)\lp|\swei{A}{s}_{\mc{H}^-}\rp|^2\rp\}\max\{y(\infty),|\hat{y}(-\infty)|\}\,.
\end{align*} 
Here, the alternatives given are for $\sign s =\pm 1$, respectively. The conclusion is obtained appealing to Lemma~\ref{lemma:properties-frequencies-low-omega}\ref{it:properties-frequencies-low-omega-TS-constants}--\ref{it:properties-frequencies-low-omega-bdry-terms}, using the smallness of $\tilde{a}_1$, making  $|\omega|\leq \omega_{\rm low}$ sufficiently small depending on $R_2^*,R_3^*$ (i.e.\ on $\tilde{a}_1$) and choosing some $E_W\geq 12$.
\end{proof}

\begin{proof}[Proof of Propositions~\ref{prop:low1a} and \ref{prop:low1b}]
Let $(a,\omega,m,\Lambda)\in \mc{F}_{\rm low,1a}\cup \mc{F}_{\rm low,1b}$, i.e.\ assume all frequency parameters are bounded in terms of $\varepsilon_{\rm width}$ and $\omega_{\rm high}$, $|\omega|\leq \omega_{\rm low}$  is sufficiently small, for some small $\tilde{a}_1$ fixed by the previous proof we have $0\leq a\leq M-\tilde{a}_1$, and either $|a|\leq \tilde{a}_0$ or $m=0$.

Let $R_3^*$ be a sufficiently large and positive value and let $p\in(0,1)$ be a small value to be fixed.

\medskip
\noindent\textit{The $h$ current.} In $\mc{F}_{\rm low,1a}$ and $\mc{F}_{\rm low,1b}$, as $a$ is subextremal, we cannot expect the model current \eqref{eq:standard-current-h-} to generate small, integrable error terms when $k\leq s\neq 0$. At it turns out, in $\mc{F}_{\rm low,1a}$, it is enough to apply currents only at the level $k=|s|$, using the smallness of $a$ to obtain easy control of any coupling errors. In $\mc{F}_{\rm low,1b}$, however, we must apply multiplier currents to all $0\leq k\leq |s|$; there, we settle for smallness of the $h$ current errors, dropping the integrability requirement: for the $r^*<0$ end, we apply the model current \eqref{eq:standard-current-h+} with symmetric argument. 

As the currents we consider are the same in both cases, though those for $k<|s|$ are not necessary in $\mc{F}_{\rm low,1a}$, we construct them together. For each $k=0,\dots, |s|$, we consider the function $h_{(k)}$ composed by
\begin{align*}
h_{(k)}=h_{\rm left}+h_{k,\mr{int}}+h_{\rm right}\,,
\end{align*}
writing $h$ for $h_{(|s|)}$. Here, $h_{k,\mr{int}}$ is  given by (c.f. the case $\mc{F}_{\rm low, 1c}$)
\begin{align*}
h_{k,\mr{int}}= \lp(\frac{w}{\max\{w(R_3^*),w(-R_3^*)\}}\rp)^{-(s-k)/2}\mathbbm{1}_{[-R_3^*,R_3^*]}\,,
\end{align*}
which produces no errors in terms of $|\uppsi_{(k)}|^2$ by \eqref{eq:h-estimate-low} when $\sigma=1/4$ (see Lemma~\ref{lemma:h-estimate-low} and the remarks at the end of Section~\ref{sec:low-hierarchy-h}); $h_{\rm right}$ is defined in the same way as in the case $\mc{F}_{\rm low, 1c}$ and 
\begin{align*}
h_{\rm left}&=C_{k,0}^-\lp[(1+p)-(1+p)\frac{R_3^*}{|r^*|}+p\log\lp(\frac{R_3^*}{|r^*|}\rp)\rp]\mathbbm{1}_{[-e^{p^{-1}}R_3^*,-R_3]}+ph_{(k)}(R_3^*)\frac{R_3^*}{|r^*|}\mathbbm{1}_{[-e^{p^{-1}}R_3^*,-R_3^*]}\\
&\qquad + \frac{2}{\log 8 -2}\lp[p(1-e^{-1/p})C_{k,0}^--e^{-1/p}R_3^*h_{(k)}'(-R_3^*)\rp]\\
&\qquad\qquad\times \lp[1-\frac{2e^{p^{-1}}R_3^*}{|r^*|}+\lp(\frac{|r^*|}{2e^{p^{-1}}R_3^*}+\frac{2e^{p^{-1}}R_3^*}{|r^*|}\rp)\log 2-\log\lp(\frac{2|r^*|}{e^{p^{-1}}R_3^*}\rp)\rp]\mathbbm{1}_{[-2e^{p^{-1}}R_3^*,e^{p^{-1}}R_3^*]},,\\
C_{k,0}^-&=h_{(k)}(-R_3^*)+R_3^*h_{(k)}'(-R_3^*)\,.
\end{align*}

As we have seen for $\mc{F}_{\rm low, 1c}$ the errors generated by $h_{\rm right}$ for each $k\leq |s|$ are integrable in $r^*$ and come with a small parameter $p$; thus they can will be absorbed by a suitable choice of $y$ current. On the other hand, by \eqref{eq:h+-estimate-error-weak}, errors due to $h_{\rm left}$ come with a smallness parameter but are not integrable for $0\leq k<|s|$: 
\begin{gather*}
\lp(h_{\rm left}''+(|s|-k)\frac{w'}{w}h_{\rm left}'\rp) \leq Bp C_{k,0}^-\lp(\frac{1}{(r^*)^2}+\frac{|s|-k}{|r^*|}\rp)\mathbbm{1}_{[-2e^{p^{-1}}R_3^*,-R_3^*]}\,.
\end{gather*}
For $0\leq k<|s|$, these errors cannot be absorbed by any $y$ current. In $\mc{F}_{\rm low, 1b}$, where this is pertinent, we use the basic estimate \eqref{eq:basic-estimate-1} from Lemma~\ref{lemma:basic-estimate-1} instead, with 
\begin{align*}
c&=-\log |r^*|\mathbbm{1}_{[-2e^{p^{-1}}R_3^*,R_3^*]}+\lp(\frac{1}{|r^*|}-\frac{1}{2e^{p^{-1}}R_3^*}-\log(2e^{p^{-1}}R_3^*)\rp)\mathbbm{1}_{(-\infty,-2e^{p^{-1}}R_3^*]}\,,\\
c'&=\frac{1}{|r^*|}\mathbbm{1}_{[-2e^{p^{-1}}R_3^*,R_3^*}+\frac{1}{(r^*)^2}\mathbbm{1}_{(-\infty,-2e^{p^{-1}}R_3^*]}\,.
\end{align*} 

\begin{figure}[b]
\centering
\includegraphics[scale=1]{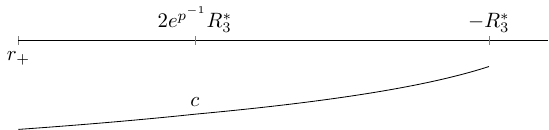}
\caption{A weight $c(r)$ for application of basic estimate \eqref{eq:basic-estimate-1} from Lemma~\ref{lemma:basic-estimate-1} in the analysis of $\mc{F}_{\rm low,1b}$.}
\label{fig:low1b_c}
\end{figure}

As long as $R_3^*$ is sufficiently large depending on $\tilde{a}_0$, $wc^2/c'$ is a decreasing function. If moreover $R_3^*$ is sufficiently large depending on $p$, $R_3\geq 10(M\tilde{a}_1p)^{-1}$, the latter is bounded above independently of the parameters:
\begin{align*}
\frac{wc^2}{c'}&\leq B w(r^*)^2\lp(\log(2e^{p^{-1}}R_3^*)\rp)^2 \leq B\exp\lp(-\frac{\sqrt{M^2-a^2}}{Mr_+}|r^*|\rp)|r^*|^2\lp(\log(2e^{p^{-1}}R_3^*)\rp)^2\mathbbm{1}_{(-\infty, -R_3^*]}\\
&\leq B\exp\lp(-\frac{\sqrt{M^2-a^2}}{Mr_+}R_3^*\rp)(R_3^*)^2\lp(\log(2e^{p^{-1}}R_3^*)\rp)^2\mathbbm{1}_{(-\infty, -R_3^*]} \leq B\,.
\end{align*}
On the other hand, as long as  $R_3\leq e^{B/p}$ for some finite $B$, we also have 
\begin{align*}
p\log(2e^{p^{-1}}R_3^*) =1+p\log(2R_3^*) \leq B\,.
\end{align*}
Thus, from Lemma \eqref{lemma:basic-estimate-1}, for $0\leq j<|s|$, we obtain
\begin{align*}
& \int_{-\infty}^{-R_3^*} p c'\lp|\uppsi_{(k+1)}\rp|^2dr^*\\
&\quad\leq Bp\log(2e^{p^{-1}}R_3^*)\lp|\swei{A}{s}_{k,\mc{H}^{-\sign s}}\rp|^2+ Bp\int_{-\infty}^{-R_3^*} w \frac{wc^2}{c'}\lp|\uppsi_{(k+1)}\rp|^2dr^*\\
&\quad\leq B \lp|\swei{A}{s}_{k,\mc{H}^{-\sign s}}\rp|^2 + \int_{-\infty}^\infty pw\mathbbm{1}_{(-\infty,-R_3^*]} \lp|\uppsi_{(k+1)}\rp|^2dr^*\,.
\end{align*}

By Lemma~\ref{lemma:h-estimate-low}, application of the $h_{(k)}$ current to the $k$th radial ODE \eqref{eq:transformed-k-separated} gives the estimate 
\begin{align*}
&\frac34\int_{-\infty}^\infty \lp[h_{(k)}\lp(\Re\mc{V}_{(k)}-\frac{(|s|-k)}{2}\lp(\frac{w'}{w}\rp)'-\omega^2\rp)- h_{k,\mr{int}}\frac{(|s|-k)}{4}\lp(\frac{w'}{w}\rp)'\rp]\lp|\uppsi_{(k)}\rp|^2dr^*\\
&\qquad+\int_{-\infty}^\infty \lp\{h_{(k)}\lp|\uppsi_{(k)}'\rp|^2 - \frac{Bp}{(r^*)^2} \lp[C_{k,0}^-\mathbbm{1}_{[-2e^{p^{-1}}R_3^*,-R_3^*]}+C_{k,0}^+\mathbbm{1}_{[R_3^*,2e^{p^{-1}}R_3^*]}\rp]\lp|\uppsi_{(k)}\rp|^2\rp\}dr^*\\
&\quad\leq B(|s|-k)C_{k,0}^- \lp|\swei{A}{s}_{k,\mc{H}^{-\sign s}}\rp|^2+ \int_{-\infty}^\infty h_{(k)}\Re\lp[\mathfrak{G}_{(k)}\overline{\uppsi}_{(k)}\rp]dr^* \\
&\quad\qquad+ (|s|-k)B C_{k,0}^-\int_{-\infty}^{-R_3^*}  pw \lp|\uppsi_{(k+1)}\rp|^2dr^*\\
&\quad\qquad +k(|s|-k+a^2)B\sum_{j=0}^{k-1}\int_{-\infty}^\infty (h_{\rm left}+h_{\rm right})\lp(\Re\mc{V}_{(j)}-\frac{(|s|-j)}{2}\lp(\frac{w'}{w}\rp)'-\omega^2\rp)\lp|\uppsi_{(j)}\rp|^2dr^* \numberthis \label{eq:proof-low-1b-intermediate-1}\\
&\quad\qquad +k(|s|-k+a^2)B\sum_{j=0}^{k-1}\int_{-\infty}^\infty h_{k,\mr{int}}\lp(\Re\mc{V}_{(j)}-\frac{(|s|-j)}{4}\lp(\frac{w'}{w}\rp)'-\omega^2\rp)\lp|\uppsi_{(j)}\rp|^2dr^*\,.
\end{align*}
We recall the reader that only the case $k=|s|$ of \eqref{eq:proof-low-1b-intermediate-1} need be taken into account to deal with the regime $\mc{F}_{\rm low, 1a}$, where $a$ is small.

\medskip
\noindent\textit{The $y$ currents.}  We further consider currents $y_{(k)}$ and $\hat{y}_{(k)}$ similarly constructed to the case of $\mc{F}_{\rm low, 1c}$, i.e.\ based on the model currents \eqref{eq:low-frequencies-y} and \eqref{eq:low-frequencies-hat-y}: for some $\epsilon>0$ and $\tilde{p}>0$,
\begin{align*}
y_{(k)}&=\tilde{p}\mathbbm{1}_{[R_3^*/2,R_3^*]}\int_{\frac{R_3^*}{2}}^{r^*} h_{k,\mr{int}}(x^*)dx^* +y(R_3^*)\lp[1+\frac{64\epsilon}{R_3^*}\lp(\frac{1}{r^*\Re\mc{V}_{(k)}}-\frac{1}{R_3^*\Re\mc{V}_{(k)}|_{r^*=R_3^*}}\rp)\rp]\mathbbm{1}_{[R_3^*, 2e^{p^{-1}}R_3^*]} \\
&\qquad+y(2e^{p^{-1}}R_3^*)\lp[1+\frac{1}{(2e^{p^{-1}}R_3^*)^{\delta}}-\frac{1}{(r^*)^{\delta}}\rp]\mathbbm{1}_{[2e^{p^{-1}}R_3^*, \infty)}\geq 0\,,\\
\hat y_{(k)}&= -\tilde{p}\mathbbm{1}_{[-R_3^*,-R_3^*/2]}\int_{-R_3^*/2}^{r^*} h_{k,\mr{int}}(x^*)dx^*\\
&\qquad+\hat y_{(k)}(-R_3^*)\lp[1-\frac{64\epsilon}{R_3^*}\lp(\frac{1}{r^*\hat{\mc{V}}_{(k)}}-\frac{1}{(-R_3^*)\hat{\mc{V}}_{(k)}|_{r^*=-R_3^*}}\rp)\rp]\mathbbm{1}_{[-2e^{p^{-1}}R_3^*,-R_3^*]} \\
&\qquad+\hat{y}_{(k)}(-2e^{p^{-1}}R_3^*)\lp[1-\frac{1}{(2e^{p^{-1}}R_3^*)^{1/2}}+\frac{1}{(-r^*)^{1/2}}\rp]\mathbbm{1}_{(-\infty, -2e^{p^{-1}}R_3^*]}\leq 0\,.
\end{align*}

As in the case of $\mc{F}_{\rm low, 1c}$, we can fix $\epsilon$ and $\tilde{p}$ such that, as long as $R_3^*$ is sufficiently large, $\hat{y}',y'>0$ and
\begin{align*}
&\frac34 h_{k,\mr{int}}\lp(\Re\mc{V}_{(k)}-\frac{(|s|-k)}{4}\lp(\frac{w'}{w}\rp)' -\omega^2\rp)-\frac{Bp}{(r^*)^2}\lp[C_{k,0}^-\mathbbm{1}_{[-2e^{p^{-1}}R_3^*,-R_3^*]}+C_{k,0}^+\mathbbm{1}_{[R_3^*,2e^{p^{-1}}R_3^*]}\rp]\\
&\qquad+\lp[\frac34 y_{(k)}'\omega^2-\frac34 \lp(y_{(k)}\Re\mc{V}_{(k)}\rp)' -\frac14y_{(k)}\Re\mc{V}_{(k)}'\rp] +\lp[ \frac34 \hat y'\omega^2-\frac34 \lp(\hat{y}_{(k)}\Re\hat{\mc{V}}_{(k)}\rp)' -\frac14 \hat{y}_{(k)}\Re\hat{\mc{V}}_{(k)}'\rp]\numberthis \\
&\quad\geq \frac12 h_{k,\mr{int}}\lp(\Re\mc{V}_{(k)}-\frac{(|s|-k)}{3}\lp(\frac{w'}{w}\rp)' -\omega^2\rp) \\
&\quad\qquad+\frac34 \hat{y}_{(k)}'(\omega-m\upomega_+)^2 -\frac12 \hat{y}_{(k)}\Re\hat{\mc{V}}_{(k)}'\mathbbm{1}_{(-\infty,R_2^*]} +\frac34 y_{(k)}'\omega^2 -\frac12 y_{(k)}\Re\mc{V}_{(k)}'\mathbbm{1}_{[R_3^*,\infty)}\,.
\end{align*}
Hence, we can invoke estimates~\eqref{eq:y-estimate-low-k} and \eqref{eq:hat-y-estimate-low-k} in the form given in Lemma~\ref{lemma:y-estimate-low}. 

Before doing so, note that, for $k<|s|$, the anti-derivative of $w^{(|s|-k)/2}$ in $r^*$ is comparable to $w^{(|s|-k)/2}/\sqrt{M^2-a^2}$, since
\begin{align}
 \exp\lp(-\frac{(|s|-k)}{2}\frac{\sqrt{M^2-a^2}}{Mr_+}|r^*|\rp)= \frac{2}{(|s|-k)}
\frac{Mr_+}{\sqrt{M^2-a^2}}\frac{d}{dr^*}\exp\lp(-\frac{(|s|-k)}{2}\frac{\sqrt{M^2-a^2}}{Mr_+}|r^*|\rp)\,. \label{eq:proof-low-1b-intermediate-2}
\end{align}
Moreover, as $\Re\mc{V}_{(k)}'\geq b w\sqrt{M^2-a^2}$ for some $b>0$, we conclude that $-\hat{y}\Re\mc{V}_{(k)}\geq b h_{k,\mr{int}} w$. Using the information already derived for $y$ in the case $\mc{F}_{\rm low, 1c}$,  we obtain
\begin{align*}
&\int_{-\infty}^\infty \lp[\frac34 \hat{y}_{(k)}'(\omega-m\upomega_+)^2 -\frac12 \hat{y}_{(k)}\Re\hat{\mc{V}}_{(k)}'\mathbbm{1}_{(-\infty,-R_3^*]} +\frac34 y_{(k)}'\omega^2 -\frac12 y_{(k)}\Re\mc{V}_{(k)}'\mathbbm{1}_{[R_3^*,\infty)} \rp]|\uppsi_{(k)}|^2dr^*\\
&\qquad +\int_{-\infty}^\infty \lp\{\lp(y_{(k)}'+\hat y_{(k)}'+h\rp)|\uppsi_{(k)}'|^2+\frac14 h_{k,\mr{int}}\lp(\Re\mc{V}_{(k)}-\frac{(|s|-k)}{3}\lp(\frac{w'}{w}\rp)' -\omega^2\rp)\lp|\uppsi_{(k)}\rp|^2\rp\}dr^*\\
&\qquad +\int_{-\infty}^\infty \frac34 \lp(h_{k,\mr{left}}+h_{k,\mr{right}}\rp)\lp(\Re\mc{V}_{(k)}-\frac{(|s|-k)}{2}\lp(\frac{w'}{w}\rp)' -\omega^2\rp)\lp|\uppsi_{(k)}\rp|^2dr^*\\
&\quad\leq  2y_{(k)}(\infty)\omega^2\lp|\swei{A}{s}_{k,\mc{I}^{-\sign s}}\rp|^2+2\lp[|\hat y_{(k)}(-\infty)|\omega^2+B(|s|-k)C_{k,0}^-\rp]\lp|\swei{A}{s}_{k,\mc{H}^{+\sign s}}\rp|^2\\
&\quad\qquad+\int_{-\infty}^\infty \lp\{h_{(k)}\Re\lp[\mathfrak{G}_{(k)}\overline{\uppsi_{(k)}}\rp] +2(y_{(k)}+\hat{y}_{(k)})\Re\lp[\mathfrak{G}_{(k)}\overline{\uppsi_{(k)}}'\rp]\rp\}dr^* \numberthis \label{eq:proof-low-1b-intermediate-3}\\
&\quad\qquad - k(|s|-k+a^2)B\sum_{j=0}^{k-1} \int_{-\infty}^\infty \lp[y_{(k)}\Re\mc{V}'_{(j)}+\hat{y}_{(k)}\Re\hat{\mc{V}}'_{(j)}-R_3^*\mathbbm{1}_{[R_3^*/2,R_3^*]}h_{k,\mr{int}} w\rp]|\uppsi_{(j)}|^2dr^* \\
&\quad\qquad +k(|s|-k+a^2)B\sum_{j=0}^{k-1}\int_{-\infty}^\infty (h_{\rm left}+h_{\rm right})\lp(\Re\mc{V}_{(j)}-\frac{(|s|-j)}{2}\lp(\frac{w'}{w}\rp)'-\omega^2\rp)\lp|\uppsi_{(j)}\rp|^2dr^* \\
&\quad\qquad +k(|s|-k+a^2)B\sum_{j=0}^{k-1}\int_{-\infty}^\infty h_{k,\mr{int}}\lp(\Re\mc{V}_{(j)}-\frac{(|s|-j)}{3}\lp(\frac{w'}{w}\rp)'-\omega^2\rp)\lp|\uppsi_{(j)}\rp|^2dr^*\\
&\quad\qquad - (k-1)(|s|-k+a^2)B\sum_{j=0}^{k-2} \int_{-\infty}^\infty \lp[ \frac{1}{r^2}y_{(k)}\Re\mc{V}'_{(j+1)}+(r-r_+)^2\hat{y}_{(k)}\Re\hat{\mc{V}}_{(j+1)}'\rp]|\uppsi_{(j+1)}|^2dr^*\\
&\quad\qquad - B(|s|-k)\int_{-\infty}^\infty \lp\{\omega^2 \lp[y_{(k)}\Re\mc{V}'_{(k+1)}+\hat{y}_{(k)}\Re\hat{\mc{V}}'_{(k+1)}\rp]-pw(|s|-k)C_{k,0}^-\mathbbm{1}_{(-\infty,-R_3^*]}\rp\}|\uppsi_{(k+1)}|^2dr^*\,.
\end{align*}
by adding as well  \eqref{eq:proof-low-1b-intermediate-1}, \eqref{eq:y-estimate-low-k} and \eqref{eq:hat-y-estimate-low-k}. 

\medskip
\noindent\textit{The case of $\mc{F}_{\rm low, 1a}$: virial currents' coupling errors and energy currents.} As mentioned, in $\mc{F}_{\rm low,1a}$, we only consider the case $k=|s|$ of \eqref{eq:proof-low-1b-intermediate-3}. We deal with the coupling errors by applications of the basic estimates \eqref{eq:basic-estimate-1} and \eqref{eq:basic-estimate-1'} from Lemma~\ref{lemma:basic-estimate-1}. First, with $c=-1/r$ if $s<0$ and $c=\frac{r-r_+}{r}$ if $s>0$, \eqref{eq:basic-estimate-1} and \eqref{eq:basic-estimate-1'} read, respectively,
\begin{align*}
\int_{-\infty}^\infty w|\uppsi_{(j)}|^2dr^*&\leq \int_{-\infty}^\infty c'|\uppsi_{(j)}|^2dr^* \leq   B\int_{-\infty}^\infty w|\uppsi_{(j+1)}|^2dr^*\leq B\int_{-\infty}^\infty w|\uppsi_{(|s|-1)}|^2dr^*\,,\\
\int_{-\infty}^\infty w|\uppsi_{(j)}'|^2dr^*
&\leq B\int_{-\infty}^\infty w\lp[|\uppsi_{(j+1)}'|^2+\frac{16a^2m^2}{r^2+a^2}|\uppsi_{(j)}|^2\rp]dr^*\\
&\leq B\int_{-\infty}^\infty w\lp[|\uppsi_{(|s|-1)}'|^2+\frac{a^2m^2}{r^2+a^2}|\uppsi_{(|s|-1)}|^2\rp]dr^*\,;
\end{align*}
then, taking $c=-1/r$ to relate $\uppsi_{(|s|-1)}$ to $\Psi$ yields
\begin{align*}
\int_{-\infty}^\infty w|\uppsi_{j}|^2dr^* &\leq   B|\swei{A}{s}_{|s|-1,\mc H^{+\sign s}}|^2+ B\int_{-\infty}^\infty \frac{w}{r^2+a^2}|\Psi|^2 dr^*\,,\numberthis\label{eq:proof-low-1ab-basic-estimate-1/r}\\
\int_{-\infty}^\infty w|\uppsi_{(j)}'|^2dr^*& \leq   B|\swei{A}{s}_{|s|-1,\mc H^{+\sign s}}|^2+ B\int_{-\infty}^\infty \frac{w}{r^2+a^2}\lp[|\Psi'|^2+\frac{a^2m^2}{r^2+a^2}|\Psi|^2\rp]dr^*\,.\numberthis\label{eq:proof-low-1ab-basic-estimate-derivative-1/r}
\end{align*}
Thus, lines 6 to 9 of \eqref{eq:proof-low-1b-intermediate-3} for $k=|s|$ can be controlled, in $\mc{F}_{\rm low,1a}$, by
\begin{align*}
& a^2|s|B(\varepsilon_{\rm width},\omega_{\rm high})\sum_{j=0}^{|s|-1} \int_{-\infty}^\infty \lp[y(\infty)-\hat{y}(-\infty)+R_3^*\rp]w|\uppsi_{(j)}|^2dr^*\\
&\quad\leq a^2|s|B(\varepsilon_{\rm width},\omega_{\rm high})\lp[y(\infty)-\hat{y}(-\infty)+R_3^*\rp]\lp\{\lp|\swei{A}{s}_{|s|-1, \mc{H}^{+\sign s}}\rp|^2 +\int_{-\infty}^\infty \frac{w}{r^2+a^2}|\Psi|^2dr^*\rp\}\\
&\quad \ll \tilde{a}_0\lp|\swei{A}{s}_{|s|-1, \mc{H}^{+\sign s}}\rp|^2 + \tilde{a}_0\int_{-\infty}^\infty \lp(hw-y\Re\mc{V}'-\hat{y}\Re\hat{\mc V}'\rp)|\Psi|^2dr^*\,,
\end{align*}
if $a\leq \tilde{a}_0$ is sufficiently small. Note the second term can already be absorbed into the left hand side of  \eqref{eq:proof-low-1b-intermediate-3} when $k=|s|$. We obtain:
\begin{align*}
&\frac12\int_{-\infty}^\infty \lp[\frac34 \hat{y}'(\omega-m\upomega_+)^2 -\frac12 \hat{y}\hat{\mc{V}}'\mathbbm{1}_{(-\infty,-R_3^*]} +\frac34 y_{(k)}'\omega^2 -\frac12 y\mc{V}'\mathbbm{1}_{[R_3^*,\infty)} + \rp]|\Psi|^2dr^*\\
&\qquad +\int_{-\infty}^\infty \lp\{\lp(y'+\hat y'+h\rp)|\Psi'|^2+ \frac18 h\lp(\mc{V}-\omega^2\rp)\lp|\Psi\rp|^2\rp\}dr^*\\
&\quad\leq  2y(\infty)\omega^2\lp(\lp|\swei{A}{s}_{\mc{I}^{+}}\rp|^2+\lp|\swei{A}{s}_{\mc{I}^{-}}\rp|^2\rp)+2|\hat y(-\infty)|(\omega-m\upomega_+)^2\lp(\lp|\swei{A}{s}_{\mc{H}^{+}}\rp|^2+\lp|\swei{A}{s}_{\mc{H}^{-}}\rp|^2\rp)\\
&\quad\qquad+ \tilde{a}_0\lp|\swei{A}{s}_{|s|-1, \mc{H}^{+\sign s}}\rp|^2+\int_{-\infty}^\infty \lp\{h\Re\lp[\mathfrak{G}\overline{\Psi}\rp] +2(y+\hat{y})\Re\lp[\mathfrak{G}\overline{\Psi}'\rp]\rp\}dr^*\,. \numberthis \label{eq:proof-low-1b-intermediate-4a}
\end{align*}

We now turn to the boundary terms, which we seek to control. In the case $m=0$, there is no superradiance, so a global application of a $T$-type energy current is sufficient; this is not so for $m\neq 0$ but $|a|\ll 1$, where we must apply a $K$-type current near $r=r_+$ and $T$ current near $r=\infty$. Thus, to keep our treatment uniform, we apply localized currents in both cases.

Take $E>0$ to be determined; let $\chi_1\leq 1$ be a smooth function which is 1 for $r^*\in(-\infty,-R_3^*+1]$ and 0 for $r^*\in[R_3^*-1,\infty)$ and $\chi_2=1-\chi_1$. After adding currents $E\chi_1Q^K$ and $E\chi_2 Q^T$, the localization errors produced are
\begin{align*}
\int_{-\infty}^\infty E[\chi_1'(\omega-m\upomega_+)+\chi_2'\omega]\Im\lp[\Psi'\overline{\Psi}\rp]dr^* &\leq \frac{BE|am|}{R_3^*}\int_{-\infty}^\infty \lp[|\Psi'|^2+|\Psi|^2\rp]\mathbbm{1}_{[-R_3^*,R_3^*]}dr^*\\
&\leq \frac{1}{32}\int_{-\infty}^\infty \lp[h|\Psi'|^2+h(\mc{V}-\omega^2)|\Psi|^2\rp]dr^*\,,
\end{align*} 
for $|am|$ sufficiently small depending on $E$. The energy currents also generate errors due to coupling to $\uppsi_{(j)}$ with $j<|s|$; if $|a|$ is sufficiently small, each of these satisfies
\begin{align*}
&\int_{-\infty}^\infty aE (\omega\chi_2+(\omega-m\upomega_+)\chi_1) w\Im\lp[\overline{\Psi}\lp(c_{s,|s|,j}^{\rm id}+imc_{s,|s|,j}^{m}\rp)\uppsi_{(j)}\rp]dr^* \\
&\quad\leq  \int_{-\infty}^\infty \frac{1}{64|s|}\lp[h\lp(\mc{V}-\omega^2\rp)+y'\omega^2+\hat{y}'(\omega-m\upomega_+)^2\rp]|\Psi|^2 \\
&\quad\qquad + a^2|s|E^2B(\varepsilon_{\rm width},\omega_{\rm high})\int_{-\infty}^\infty\lp[\lp(\omega^2+(\omega-m\upomega_+)^2\rp)\mathbbm{1}_{[-R_3^*,R_3^*]}+\lp(\frac{w}{y'}+\frac{w}{\hat y '}\rp)\mathbbm{1}_{[-R_3^*,R_3^*]^c}\rp]w|\uppsi_{(j)}|^2\\
&\quad\leq  \int_{-\infty}^\infty \frac{1}{64|s|}\lp[h\lp(\mc{V}-\omega^2\rp)+y'\omega^2+\hat{y}'\omega_0^2\rp]|\Psi|^2dr^* \\
&\quad\qquad+ a^2|s|E^2(R_3^*)B(\varepsilon_{\rm width},\omega_{\rm high})\lp\{\lp|\swei{A}{s}_{|s|-1, \mc{H}^{+\sign s}}\rp|^2 + \int_{-\infty}^\infty \frac{w}{r^2+a^2}|\Psi|^2dr^*\rp\} \,,
\end{align*}
using \eqref{eq:proof-low-1ab-basic-estimate-1/r} in the last inequality. If $|a|\leq\tilde{a}_0$ can be made sufficiently small, the above can be controlled by 
\begin{align*}
\int_{-\infty}^\infty \frac{1}{32|s|}\lp[h\lp(\mc{V}-\omega^2\rp)+y'\omega^2+\hat{y}'\omega_0^2\rp]|\Psi|^2dr^* + \tilde{a}_0E^2\lp|\swei{A}{s}_{|s|-1, \mc{H}^{+\sign s}}\rp|^2\,, \numberthis\label{eq:proof-low-1b-intermediate-5a}
\end{align*}
and so we deduce
\begin{align*}
&\omega^2\lp|\swei{A}{s}_{\mc{I}^+}\rp|^2 +(\omega-m\upomega_+)^2\lp|\swei{A}{s}_{\mc{H}^+}\rp|^2+\frac18\int_{-\infty}^\infty (y'+\hat y'+h)|\Psi'|^2dr^*\\
&\qquad +\int_{-\infty}^\infty \lp[ h(\mc{V} -\omega^2) + \hat y'\omega_0^2 - \hat y\Re\hat{\mc{V}}'\mathbbm{1}_{(-\infty,R_3^*]} + y'\omega^2 - y\Re\mc{V}'\mathbbm{1}_{[R_3^*,\infty)}\rp]|\Psi|^2dr^*\\
&\quad\leq  B(E,E_W)\omega^2\lp|\swei{A}{s}_{\mc{I}^-}\rp|^2+B(E,E_W)(\omega-m\upomega_+)^2\lp|\swei{A}{s}_{\mc{H}^-}\rp|^2 +\tilde{a}_0B(E)\lp|\swei{A}{s}_{|s|-1, \mc{H}^{+\sign s}}\rp|^2\\
&\qquad+\int_{-\infty}^\infty \lp\{h\Re\lp[\mathfrak{G}\overline{\Psi}\rp]+2(y+\hat{y})\Re\lp[\mathfrak{G}\overline{\Psi}'\rp]-E\lp[\chi_2\omega+\chi_1(\omega-m\upomega_+)\rp]\Im\lp[\mathfrak{G}\overline{\Psi}\rp]\rp\}dr^*\,. \numberthis\label{eq:proof-low-1a-end-number1}
\end{align*}

The last term in \eqref{eq:proof-low-1a-end-number1} can be treated using Lemma~\ref{lemma:bdry-term-via-constraint-bdd} or, if we make use of the precise boundary relations of Lemma~\ref{lemma:uppsi-general-asymptotics}, it can be absorbed into the previous term. If the boundary relations of Lemma~\ref{lemma:uppsi-general-asymptotics} are available then we can take $E_W>0$, and consider applications of Teukolsky--Starobinsky-type localized energy currents as follows. Let $\chi_3\leq 1$ be a smooth function such that $\chi_3=1$ for $r^*\leq R_3^*/3$ and $\chi_3=0$ for $r^*\geq 2R_3^*/3$; set $\chi_4=1-\chi_3$. After adding currents $E_W\chi_3Q^{W,K}$ and $E_W\chi_4 Q^{W,T}$, the localization errors produced are bounded by (see Lemma~\ref{lemma:wronskian-energy-currents})
\begin{align*}
&\frac{BE_W}{R_3^*}|am|\int_{\supp(\chi_3')}\sum_{k=0}^{|s|}(R_3^*)^{2(|s|-k)+1}\lp[|\uppsi_{(k)}'|^2+R_3^{-2}|\uppsi_{(k)}|^2\rp]dr^*\\
&\quad\leq \tilde a_0B(\omega_{\rm high},\varepsilon_{\rm width})E_W\int_{-\infty}^\infty \lp(\lp[h|\Psi|^2+h(\mc{V}-\omega^2)|\Psi|^2\rp]dr^*+|s|\sum_{k=0}^{|s|-1}\int w\lp[|\uppsi_{(k)}'|^2+|\uppsi_{(k)}|^2\rp]dr^*\rp)\\
&\quad\leq \tilde{a_0}^{1/2} \int_{-\infty}^\infty \lp[\lp(h+\frac{w}{r^2+a^2}\rp)|\Psi|^2+\lp(h(\mc{V}-\omega^2)+\frac{w}{r^2+a^2}\rp)|\Psi|^2\rp]dr^*  + \tilde{a_0}^{1/2}|\swei{A}{s}_{|s|-1,\mc H^{+\sign s}}|^2\,,
\end{align*}
where we have taken $\tilde{a}_0$ sufficiently small depending on $E_W$, $\varepsilon_{\rm width}$, $\omega_{\rm high}$ and $R_3^*$ and we have used \eqref{eq:proof-low-1ab-basic-estimate-1/r} and \eqref{eq:proof-low-1ab-basic-estimate-derivative-1/r}. If $\tilde{a}_0$ is  sufficiently small, the first term can be absorbed by the left hand side of \eqref{eq:proof-low-1b-intermediate-4a}.

Note that, by application of the Killing energy currents to control boundary terms by \eqref{eq:proof-low-1b-intermediate-5a}, we have actually produced more boundary terms with the wrong sign; however, they come with a small parameter $\tilde{a}_0$.  Thus, putting all boundary terms produced by the virial and energy currents, we have 
\begin{align*}
&\lp[2y(\infty)-E-E_W\lp\{1,\frac{\mathfrak{C}_s}{\mathfrak{D}_s}\rp\}\rp]\omega^2\lp|\swei{A}{s}_{\mc{I}^+}\rp|^2+\lp[2y(\infty)+E+E_W\lp\{\frac{\mathfrak{C}_s}{\mathfrak{D}_s},1\rp\}\rp]\omega^2\lp|\swei{A}{s}_{\mc{I}^-}\rp|^2\\
&\qquad+\lp[2|\hat y(\infty)|+\tilde{a}_0^{1/2}(E^2+1)\frac{\lp|\swei{A}{s}_{|s|-1,\mc{H}^+}\rp|^2}{\omega_0^2\lp|\swei{A}{s}_{\mc{H}^+}\rp|^2}\mathbbm{1}_{\{s>0\}}-E-E_W\lp\{\frac{\mathfrak{C}_s}{\mathfrak{D}_s},1\rp\}\rp]\omega_0^2\lp|\swei{A}{s}_{\mc{H}^+}\rp|^2\\
&\qquad+\lp[2|\hat y(\infty)|+\tilde{a}_0^{1/2}(E^2+1)\frac{\lp|\swei{A}{s}_{|s|-1,\mc{H}^-}\rp|^2}{\omega_0^2\lp|\swei{A}{s}_{\mc{H}^-}\rp|^2}\mathbbm{1}_{\{s<0\}}+E+E_W\lp\{1,\frac{\mathfrak{C}_s}{\mathfrak{D}_s}\rp\}\rp]\omega_0^2\lp|\swei{A}{s}_{\mc{H}^-}\rp|^2\\
&\leq \max\{y(\infty),|\hat{y}(-\infty)|\}\lp[-\omega^2\lp|\swei{A}{s}_{\mc{I}^+}\rp|^2-(\omega-m\upomega_+)^2\lp|\swei{A}{s}_{\mc{H}^+}\rp|^2\rp]\\
&\qquad +\lp[4\max\{y(\infty),|\hat{y}(-\infty)|\}+E+\frac53E_W\rp]\lp[\omega^2\lp|\swei{A}{s}_{\mc{I}^-}\rp|^2+(\omega-m\upomega_+)^2\lp|\swei{A}{s}_{\mc{H}^-}\rp|^2\rp]\,,
\end{align*}
where we write $\omega_0:=\omega-m\upomega_+$ and where the alternatives given are for $s>0$ and $s\leq 0$, respectively. To obtain the conclusion, we have set either $E\geq 4\max\{y(\infty),|\hat{y}(-\infty)|\}$ or $E_W\geq 12\max\{y(\infty),|\hat{y}(-\infty)|\}$, assumed $\tilde{a}_0$ is sufficiently small depending on $E$ and $R_3^*$ so that $2\tilde{a}_0\leq \min\{1,E^{-2}\}$ and used Lemma~\ref{lemma:properties-frequencies-low-omega}\ref{it:properties-frequencies-low-omega-TS-constants}--\ref{it:properties-frequencies-low-omega-bdry-terms}. To conclude, we fix the remaining free parameters so that $\max\{y(\infty),|\hat{y}(-\infty)|\}=1$. Then, we obtain from \eqref{eq:proof-low-1b-intermediate-4a} with $k=|s|$,
\begin{align*}
&\omega^2\lp|\swei{A}{s}_{\mc{I}^+}\rp|^2 +(\omega-m\upomega_+)^2\lp|\swei{A}{s}_{\mc{H}^+}\rp|^2+\frac18\int_{-\infty}^\infty (y'+\hat y'+h)|\Psi'|^2dr^*\\
&\qquad +\frac{1}{16}\int_{-\infty}^\infty \lp[ h(\mc{V} -\omega^2) + \hat y'\omega_0^2 - \hat y\Re\hat{\mc{V}}'\mathbbm{1}_{(-\infty,R_3^*]} + y'\omega^2 - y\Re\mc{V}'\mathbbm{1}_{[R_3^*,\infty)}\rp]|\Psi|^2dr^*\\
&\quad\leq  B(E,E_W)\omega^2\lp|\swei{A}{s}_{\mc{I}^-}\rp|^2+B(E,E_W)(\omega-m\upomega_+)^2\lp|\swei{A}{s}_{\mc{H}^-}\rp|^2 + \int_{-\infty}^\infty E_W\lp[\chi_3\lp(Q^{W,K}\rp)'+\chi_4\lp(Q^{W,T}\rp)'\rp]dr^*\\
&\qquad+\int_{-\infty}^\infty \lp\{h\Re\lp[\mathfrak{G}\overline{\Psi}\rp]+2(y+\hat{y})\Re\lp[\mathfrak{G}\overline{\Psi}'\rp]-E\lp[\chi_2\omega+\chi_1(\omega-m\upomega_+)\rp]\Im\lp[\mathfrak{G}\overline{\Psi}\rp]\rp\}dr^*\,. \numberthis\label{eq:proof-low-1a-end}
\end{align*}

This concludes the proof for $\mc{F}_{\rm low, 1a}$.

\medskip
\noindent\textit{The case of $\mc{F}_{\rm low, 1b}$: virial currents' coupling errors and energy currents.} In $\mc{F}_{\rm low,1b}$, to deal with \eqref{eq:proof-low-1b-intermediate-3}, we need the more delicate argument which we have already appealed to in $\mc{F}_{\rm low, 1c}$: we consider the inequality \eqref{eq:proof-low-1b-intermediate-3} for all $k=0,\dots |s|$. By construction, when comparing $h_{k,\mr{int}}$ to $h_{j,\mr{int}}$, there is a gain of $(R_3^*)^{k-j}$ for $r^*\geq R_3^*/2$, whereas there is no gain for $r^*\leq -R_3^*/2$ due to \eqref{eq:proof-low-1b-intermediate-2} (c.f.\ $\mc{F}_{\rm low, 1c}$, in which the gain is observed at both ends). Thus, the last six lines of \eqref{eq:proof-low-1b-intermediate-3} can be controlled by
\begin{align*}
& kB\sum_{j=0}^{k-1} \int_{-\infty}^\infty \lp[h_{j,\mr{int}}\lp(\Re\mc{V}_{(j)}-\frac{(|s|-j)}{3}\lp(\frac{w'}{w}\rp)' -\omega^2\rp)\rp.\\
& \qquad\qquad\qquad\qquad\qquad\qquad\qquad +\lp(h_{i,\mr{left}}+h_{i,\mr{right}}\rp)\lp(\Re\mc{V}_{(j)}-\frac{(|s|-j)}{2}\lp(\frac{w'}{w}\rp)' -\omega^2\rp)\\
&\qquad\qquad\qquad\qquad\qquad\qquad\qquad\lp. - y_{(j)}\Re\mc{V}'_{(j)}\mathbbm{1}_{[R_3^*,\infty)}-\hat{y}_{(j)}\Re\hat{\mc{V}}'_{(j)}\mathbbm{1}_{(-\infty,-R_3^*]}\rp]|\uppsi_{(j)}|^2dr^* \\
&\qquad  -(|s|-k)(p+\omega_{\rm low})\int_{-\infty}^\infty  \lp[y_{(k+1)}\Re\mc{V}'_{(k+1)}+\hat{y}_{(k+1)}\Re\hat{\mc{V}}'_{(k+1)}\rp]|\uppsi_{(k+1)}|^2dr^*\,.
\end{align*}
Thus, similarly to $\mc F_{\rm low,1c}$, if we choose $\omega_{\rm low}$ sufficiently small compared to the other constants involved in the problem, and iterate in $k=0,\dots,|s|$, we obtain an analogue of \eqref{eq:proof-low-1c-intermediate-4}:
\begin{align*}
&\frac{1}{16}\int_{-\infty}^\infty \lp[\hat{y}_{(k)}'(\omega-m\upomega_+)^2 - \hat{y}_{(k)}\Re\hat{\mc{V}}_{(k)}'\mathbbm{1}_{(-\infty,-R_3^*]} + y_{(k)}'\omega^2 - y_{(k)}\Re\mc{V}_{(k)}'\mathbbm{1}_{[R_3^*,\infty)} \rp]|\uppsi_{(k)}|^2dr^*\\
&\qquad +\frac{1}{16}\int_{-\infty}^\infty \lp\{\lp(y_{(k)}'+\hat y_{(k)}'+h\rp)|\uppsi_{(k)}'|^2+ h_{k,\mr{int}}\lp(\Re\mc{V}_{(k)}-\frac{3(|s|-k)}{4}\lp(\frac{w'}{w}\rp)' -\omega^2\rp)\lp|\uppsi_{(k)}\rp|^2\rp\}dr^*\\
&\qquad+\frac{1}{16}\int_{-\infty}^\infty  (h_{k,\mr{left}}+h_{k,\mr{right}})\lp(\Re\mc{V}_{(k)}-\frac{(|s|-k)}{2}\lp(\frac{w'}{w}\rp)' -\omega^2\rp)\lp|\uppsi_{(k)}\rp|^2dr^*\\
&\qquad +E\omega^2\sum_{j=0}^k\lp\{\lp|\swei{A}{s}_{j,\mc{I}^{+}}\rp|^2+\lp|\swei{A}{s}_{j,\mc{H}^{+}}\rp|^2\rp\}\\
&\quad\leq  B(|s|-k)\sum_{j=0}^k\lp|\swei{A}{s}_{j,\mc{H}^{+}}\rp|^2 + B(E)\omega^2\sum_{j=0}^{k} \lp\{\lp|\swei{A}{s}_{j,\mc{I}^{-}}\rp|^2+\lp|\swei{A}{s}_{j,\mc{H}^{-}}\rp|^2\rp\} \\
&\quad\qquad-B(|s|-k)\int_{-\infty}^\infty  \lp[y_{(k+1)}\Re\mc{V}'_{(k+1)}+\hat{y}_{(k+1)}\Re\hat{\mc{V}}'_{(k+1)}\rp]|\uppsi_{(k+1)}|^2dr^* \\
&\quad\qquad+B\sum_{j=0}^{k}\int_{-\infty}^\infty \lp\{h_{(j)}\Re\lp[\mathfrak{G}_{(j)}\overline{\uppsi_{(j)}}\rp] +2(y_{(j)}+\hat{y}_{(j)})\Re\lp[\mathfrak{G}_{(j)}\overline{\uppsi_{(j)}}'\rp]-E\omega\Im\lp[\mathfrak{G}_{(j)}\overline{\uppsi_{(j)}}\rp]\rp\}dr^*
\numberthis \label{eq:proof-low-intermeditate-4}\,,
\end{align*}
where, we have applied Killing $T$ type currents to each of the transformed radial ODEs with $j=0,\dots,k$. We can appeal to Lemma~\ref{lemma:bdry-term-via-constraint-bdd} to handle the firts term on the right hand side of \eqref{eq:proof-low-intermeditate-4}. However,  if the boundary term relations of Lemma~\ref{lemma:uppsi-general-asymptotics} hold, then we can convert it into a term at the top level:
\begin{align*}
\sum_{j=0}^k\lp|\swei{A}{s}_{j,\mc{H}^{+}}\rp|^2\leq B\omega^2\lp|\swei{A}{s}_{\mc{H}^{+}}\rp|^2\,,
\end{align*}
and so, by making $E$ slightly larger, we end up with 
\begin{align*}
&\frac{1}{16}\sum_{k=0}^{|s|}\int_{-\infty}^\infty \lp[\hat{y}_{(k)}'(\omega-m\upomega_+)^2 - \hat{y}_{(k)}\Re\hat{\mc{V}}_{(k)}'\mathbbm{1}_{(-\infty,-R_3^*]} + y_{(k)}'\omega^2 - y_{(k)}\Re\mc{V}_{(k)}'\mathbbm{1}_{[R_3^*,\infty)} \rp]|\uppsi_{(k)}|^2dr^*\\
&\qquad +\frac{1}{16}\sum_{k=0}^{|s|}\int_{-\infty}^\infty \lp\{\lp(y_{(k)}'+\hat y_{(k)}'+h\rp)|\uppsi_{(k)}'|^2+ h_{k,\mr{int}}\lp(\Re\mc{V}_{(k)}-\frac{3(|s|-k)}{4}\lp(\frac{w'}{w}\rp)' -\omega^2\rp)\lp|\uppsi_{(k)}\rp|^2\rp\}dr^*\\
&\qquad+\frac{1}{16}\sum_{k=0}^{|s|}\int_{-\infty}^\infty  (h_{k,\mr{left}}+h_{k,\mr{right}})\lp(\Re\mc{V}_{(k)}-\frac{(|s|-k)}{2}\lp(\frac{w'}{w}\rp)' -\omega^2\rp)\lp|\uppsi_{(k)}\rp|^2dr^*\\
&\qquad +E\omega^2\sum_{j=0}^k\lp\{\lp|\swei{A}{s}_{j,\mc{I}^{+}}\rp|^2+\lp|\swei{A}{s}_{j,\mc{H}^{+}}\rp|^2\rp\}\\
&\quad\leq  B(E)\omega^2\sum_{j=0}^{|s|} \lp\{\lp|\swei{A}{s}_{j,\mc{I}^{-}}\rp|^2+\lp|\swei{A}{s}_{j,\mc{H}^{-}}\rp|^2\rp\} \\
&\quad\qquad+B\sum_{j=0}^{|s|}\int_{-\infty}^\infty \lp\{h_{(j)}\Re\lp[\mathfrak{G}_{(j)}\overline{\uppsi_{(j)}}\rp] +2(y_{(j)}+\hat{y}_{(j)})\Re\lp[\mathfrak{G}_{(j)}\overline{\uppsi_{(j)}}'\rp]-E\omega\Im\lp[\mathfrak{G}_{(j)}\overline{\uppsi_{(j)}}\rp]\rp\}dr^*
\numberthis \label{eq:proof-low-intermeditate-4b}\,.
\end{align*}

Alternatively, if we iterate in $k=0,\dots,|s|$ without applying the Killing energy currents, we obtain
\begin{align*}
&\frac{1}{16}\int_{-\infty}^\infty \lp[\hat{y}_{(k)}'(\omega-m\upomega_+)^2 - \hat{y}_{(k)}\Re\hat{\mc{V}}_{(k)}'\mathbbm{1}_{(-\infty,-R_3^*]} + y_{(k)}'\omega^2 - y_{(k)}\Re\mc{V}_{(k)}'\mathbbm{1}_{[R_3^*,\infty)} \rp]|\uppsi_{(k)}|^2dr^*\\
&\qquad +\frac{1}{16}\int_{-\infty}^\infty \lp\{\lp(y_{(k)}'+\hat y_{(k)}'+h\rp)|\uppsi_{(k)}'|^2+ h_{k,\mr{int}}\lp(\Re\mc{V}_{(k)}-\frac{3(|s|-k)}{4}\lp(\frac{w'}{w}\rp)' -\omega^2\rp)\lp|\uppsi_{(k)}\rp|^2\rp\}dr^*\\
&\qquad+\frac{1}{16}\int_{-\infty}^\infty  (h_{k,\mr{left}}+h_{k,\mr{right}})\lp(\Re\mc{V}_{(k)}-\frac{(|s|-k)}{2}\lp(\frac{w'}{w}\rp)' -\omega^2\rp)\lp|\uppsi_{(k)}\rp|^2dr^*\\
&\quad\leq  B\omega^2\sum_{j=0}^{k} \lp\{y_{(j)}(\infty)\lp|\swei{A}{s}_{j,\mc{I}^{-\sign s}}\rp|^2+|\hat y_{(j)}(-\infty)|\lp|\swei{A}{s}_{j,\mc{H}^{+\sign s}}\rp|^2\rp\}  \\
&\quad\qquad+B\sum_{j=0}^{k}\int_{-\infty}^\infty \lp\{h_{(j)}\Re\lp[\mathfrak{G}_{(j)}\overline{\uppsi_{(j)}}\rp] +2(y_{(j)}+\hat{y}_{(j)})\Re\lp[\mathfrak{G}_{(j)}\overline{\uppsi_{(j)}}'\rp]-E\omega\Im\lp[\mathfrak{G}_{(j)}\overline{\uppsi_{(j)}}\rp]\rp\}dr^*\\
&\quad\qquad-B(|s|-k)\int_{-\infty}^\infty  \lp[y_{(k+1)}\Re\mc{V}'_{(k+1)}+\hat{y}_{(k+1)}\Re\hat{\mc{V}}'_{(k+1)}\rp]|\uppsi_{(k+1)}|^2dr^*\,, \numberthis\label{eq:proof-low-1b-intermediate-4b}
\end{align*}
and, using the boundary term relations of Lemma~\ref{lemma:apha-general-asymptotics}, we can control the boundary terms by a Teukolsky--Starobinky energy current. For $k=|s|$, the boundary terms at $r^*=\infty$ become
\begin{align*}
&\lp[2y(\infty)\lp\{1,1+B\sum_{j=0}^{|s|-1} \frac{y_{(j)}(\infty)}{y(\infty)} \frac{\lp|\swei{A}{s}_{j,\mc{I}^+}\rp|^2}{\lp|\swei{A}{s}_{\mc{I}^+}\rp|^2}\rp\}-E_W\lp\{1,\frac{\mathfrak{C}_s}{\mathfrak{D}_s^{\mc{I}}}\rp\}\rp]\omega^2\lp|\swei{A}{s}_{\mc{I}^+}\rp|^2\\
&\qquad+\lp[2y(\infty)\lp\{1+B\sum_{j=0}^{|s|-1} \frac{y_{(j)}(\infty)}{y(\infty)}\frac{\lp|\swei{A}{s}_{j,\mc{I}^-}\rp|^2}{\lp|\swei{A}{s}_{\mc{I}^-}\rp|^2},1\rp\}+E_W\lp\{\frac{\mathfrak{C}_s}{\mathfrak{D}_s^{\mc{I}}},1\rp\}\rp]\omega^2\lp|\swei{A}{s}_{\mc{I}^-}\rp|^2\\
&\quad \leq \omega^2\lp\{\lp[3y(\infty)-\frac13 E_W\rp]\lp|\swei{A}{s}_{\mc{I}^+}\rp|^2+\lp[3y(\infty)+\frac53 E_W\rp]\lp|\swei{A}{s}_{\mc{I}^-}\rp|^2\rp\}\,,
\end{align*} 
where the alternatives given are for $\sign s =\pm 1$, respectively. To obtain the final bound,  we have used Lemma~\ref{lemma:properties-frequencies-low-omega}\ref{it:properties-frequencies-low-omega-TS-constants}--\ref{it:properties-frequencies-low-omega-bdry-terms}. As $r\to r_+$, the boundary terms are
\begin{align*}
&\lp[2|\hat y(-\infty)|\lp(1+B\sum_{j=0}^{|s|} \lp\{\frac{\hat y_{(j)}(-\infty)}{\hat y(-\infty)}\omega^2 +1+E\rp)\frac{\lp|\swei{A}{s}_{j,\mc{H}^+}\rp|^2}{\omega^2\lp|\swei{A}{s}_{\mc{H}^+}\rp|^2},1\rp\}-E_W\lp\{\frac{\mathfrak{C}_s}{\mathfrak{D}_s^{\mc{I}}},1\rp\}\rp]\omega^2\lp|\swei{A}{s}_{\mc{H}^+}\rp|^2\\
&\qquad+\lp[2|\hat y(-\infty)|\lp\{1+B\sum_{j=0}^{|s|} \lp(\frac{\hat y_{(j)}(-\infty)}{\hat y(-\infty)}\omega^2+2\rp)\frac{\lp|\swei{A}{s}_{j,\mc{H}^-}\rp|^2}{\omega^2\lp|\swei{A}{s}_{\mc{H}^-}\rp|^2},1 \rp\}+E_W\lp\{1,\frac{\mathfrak{C}_s}{\mathfrak{D}_s^{\mc{I}}}\rp\}\rp]\omega^2\lp|\swei{A}{s}_{\mc{H}^-}\rp|^2\\
&\quad \leq \omega^2\lp\{\lp[3B(\tilde{a}_1)|\hat{y}(-\infty)|-\frac13 E_W\rp]\lp|\swei{A}{s}_{\mc{H}^+}\rp|^2+\lp[3B(\tilde{a}_1)|\hat{y}(-\infty)|+\frac53 E_W\rp]\lp|\swei{A}{s}_{\mc{H}^-}\rp|^2\rp\}\,,
\end{align*} 
since $|\hat{y}_{(j)}(-\infty)|$ are comparable for all $j=0,...,|s|$ and using Lemma~\ref{lemma:properties-frequencies-low-omega}\ref{it:properties-frequencies-low-omega-TS-constants}--\ref{it:properties-frequencies-low-omega-bdry-terms}. Finally, setting $E_W\geq  12B(\tilde{a}_1)\max\{y(\infty),|\hat{y}(-\infty)|\}$ for some sufficiently large $B$ and normalizing the currents such that $B(\tilde{a}_1)\max\{y(\infty),|\hat{y}(-\infty)|\}=1$ for the same $B$, we obtain from \eqref{eq:proof-low-1b-intermediate-4b},
\begin{align*}
&\omega^2\lp|\swei{A}{s}_{\mc{I}^+}\rp|^2 +\omega^2\lp|\swei{A}{s}_{\mc{H}^+}\rp|^2+\frac18\int_{-\infty}^\infty (y'+\hat y'+h)|\Psi'|^2dr^*\\
&\qquad +\frac{1}{16}\int_{-\infty}^\infty \lp[ h(\mc{V} -\omega^2) + \hat y'\omega^2 - \hat y\Re\hat{\mc{V}}'\mathbbm{1}_{(-\infty,R_3^*]} + y'\omega^2 - y\Re\mc{V}'\mathbbm{1}_{[R_3^*,\infty)}\rp]|\Psi|^2dr^*\\
&\quad\leq  5E_W\omega^2\lp[\lp|\swei{A}{s}_{\mc{I}^-}\rp|^2+\lp|\swei{A}{s}_{\mc{H}^-}\rp|^2\rp] \\
&\quad\qquad+\int_{-\infty}^\infty \lp\{h\Re\lp[\mathfrak{G}\overline{\Psi}\rp]+2(y+\hat{y})\Re\lp[\mathfrak{G}\overline{\Psi}'\rp]-Eb(\tilde{a}_1)\omega\Im\lp[\mathfrak{G}\overline{\Psi}\rp]+E_W(Q^W)'\rp\}dr^*\\
&\quad\qquad+B\sum_{j=0}^{|s|-1}\int_{-\infty}^\infty \lp\{h_{(j)}\Re\lp[\mathfrak{G}_{(j)}\overline{\uppsi_{(j)}}\rp] +2(y_{(j)}+\hat{y}_{(j)})\Re\lp[\mathfrak{G}_{(j)}\overline{\uppsi_{(j)}}'\rp]\rp\}dr^*\,,
\end{align*}
which concludes the proof for $\mc{F}_{\rm low, 1b}$.
\end{proof}

\subsubsection{Low frequencies outside of axisymmetry for nonzero black hole angular momentum}
\label{sec:F-low-2}

In this section, we prove Theorems~\ref{thm:ode-estimates-Psi-A} and \ref{thm:ode-estimates-Psi-B} in the frequency range $\mc{F}_{\rm low,2}$.  Concretely, we will show 
\begin{proposition}[Estimates in $\mc{F}_{\rm low,2}$] \label{prop:low2} Fix $s\in\{0,\pm 1,\pm 2\}$ and $M>0$. Then, for all $\delta\in(0,1]$, $\omega_{\rm high}>0$, $\varepsilon^{-1}_{\rm width}>0$, $\tilde{a}_0>0$, for all $E,E_W>0$ such that one is sufficiently large, for all $\omega_{\rm low}$ sufficiently small depending on the latter, and for all $(a,\omega,m,\Lambda) \in \mc{F}_{\rm low,2}(\omega_{\rm high},\varepsilon_{\rm width},\omega_{\rm low},\tilde{a}_0)$, there exist functions $y$, $\hat{y}$, $h$, $\chi_1$ and $\chi_2$ satisfying the uniform bounds
$|y| +|\hat{y}|+|h|+|\chi_1|+|\chi_2|+|\chi_4|\leq B(\tilde{a}_0) \,,$
such that, for all smooth $\Psi$ arising from a smooth solution to the radial ODE~\eqref{eq:radial-ODE-alpha} via \eqref{eq:def-psi0-separated} and \eqref{eq:transformed-transport-separated} and itself satisfying the radial ODE~\eqref{eq:radial-ODE-Psi},  if $\Psi$ has the general boundary conditions as in Lemma~\ref{lemma:uppsi-general-asymptotics}, we have the estimate
\begin{align*}
&b(\tilde{a}_0)\lp[\omega^2\lp|\swei{A}{s}_{\mc{I}^+}\rp|^2+(\omega-m\upomega_+)^2\lp|\swei{A}{s}_{\mc{H}^+}\rp|^2\rp] + b(\delta,\tilde{a}_0)\int_{-\infty}^\infty \frac{\Delta}{r^2}\Big[r^{-1-\delta}|\Psi'|^2+r^{-3}(r^{-\delta}+\Lambda)|\Psi|^2\Big]dr^*\\
&\quad\leq B(\tilde{a}_0,E,E_W)\lp[\omega^2\lp|\swei{A}{s}_{\mc{I}^-}\rp|^2+(\omega-m\upomega_+)^2\lp|\swei{A}{s}_{\mc{H}^-}\rp|^2\rp]\\
&\quad\qquad+\int_{-\infty}^\infty \lp\{-E\lp[\chi_1(\omega-m\upomega_+)+\chi_2\omega\rp]\Im\lp[\mathfrak{G}\overline{\Psi}\rp]+E_W\lp[\chi_1\lp(Q^{W,\,K}\rp)'+\chi_4\lp(Q^{W,\,T}\rp)'\rp]\rp\}dr^*\\
&\quad\qquad +B(\tilde{a}_0)\sum_{k=0}^{|s|}\int_{-\infty}^\infty \lp\{2(y+\hat{y})\Re\lp[\mathfrak{G}_{(k)}\overline{\uppsi_{(k)}}'\rp]+h\Re\lp[\mathfrak{G}_{(k)}\overline{\uppsi_{(k)}}\rp]-E\chi_2\omega\Im\lp[\mathfrak{G}_{(k)}\overline{\uppsi}_{(k)}\rp]\rp\}dr^*\,.
\end{align*}
If, for $k=0,\dots|s|$, $\uppsi_{(k)}$ have outgoing boundary conditions as in Definition \ref{def:outgoing-bdry-uppsi}, then we must take $E_W=0$ and either add \eqref{eq:ode-estimates-outgoing-addition}.
\end{proposition}

We start by briefly recalling the strategy, which is best illustrated by Figure~\ref{fig:low2}.

\begin{figure}[htbp]
\centering
\begin{subfigure}{\textwidth}
\centering
\includegraphics[scale=1]{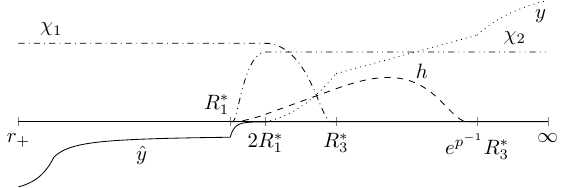}
\caption{}
\end{subfigure}
\begin{subfigure}{\textwidth}
\centering
\includegraphics[scale=1]{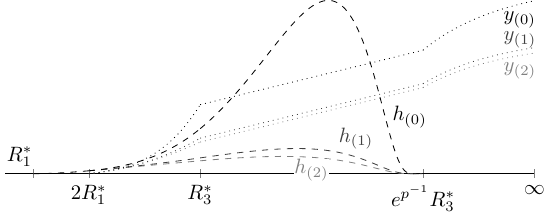}
\caption{}
\end{subfigure}
\caption{Currents in the proof of Theorems~\ref{thm:ode-estimates-Psi-A} and \ref{thm:ode-estimates-Psi-B} for the frequency ranges $\mc{F}_{\rm low,2}$ (relative size not accurate): (a) and (b) depicts currents applied to the transformed equation~\eqref{eq:transformed-k-separated} with $k=|s|$ and $k<|s|$, respectively.}
\label{fig:low2}
\end{figure}

In this case, $|\omega-m\upomega_+|>0$ and this bound can be used to construct a $y$ current yielding a good bulk globally but which degenerates at $r^*=+\infty$. By Lemma~\ref{lemma:properties-potential-low-omega}, for small enough $\omega_{\rm low}$, the potentials $\Re\mc{V}_{(k)}$ have good positivity properties in a compact range of $r^*$ located near $r^*=\infty$, which we can exploit using a compactly supported $h$ current that takes off in a region where the errors can be absorbed by the global $y$ current. For each $k$, the $h$ current goes down in the same controlled fashion already considered in the previous section; we connect it to a similarly constructed $y$ current localized near $r^*=+\infty$ whose goal is to exploit the sign of $\Re\mc{V}_{(k)}'$. As before, we convert the dependence of $y$ estimates on $\Im\mc{V}_{(k)}$ and $\uppsi_{(k)}'$ into a coupling to the $(k+1)$th equation in the transformed system (see Lemma~\ref{lemma:h-y-identity-k}).

As superradiance can occur, to control the boundary terms introduced by the $y$ currents, we must employ localized energy currents which, given the strength of the coupling error terms arising from Killing energy currents, must be of Teukolsky--Starobinsky type at least at $r=r_+$. Unfortunately, though the Teukolsky--Starobinsky-type energy currents do not generate coupling errors, their localization errors are somewhat more severe, as they involve strongly $r$-weighted terms involving the $k<|s|$ transformed variables $\uppsi_{(k)}$. To absorb the localization errors, we must therefore, for each $k$, take $h$ very large in a region where the $k$th potential enjoy good positivity properties. Thus, whereas in $\mc{F}_{\rm low,1b}$ and $\mc{F}_{\rm low,1c}$, we could choose $h$ so that the coupling of the $k$th equation to the $(k+1)$th equation was small, here that smallness will be absent; instead, it will be the coupling to $j<k$ which is small: aside from the inhomogeneity, a good bulk term in $k$ can be controlled by, if $k<|s|$, the bulk terms in $k+1$ and, if $k\neq 0$, by small multiples of the bulk term in $j<k$. By iterating along the system, we can nevertheless close an estimate for $\Psi$, i.e.\ at the level $k=|s|$.

Concretely:

\begin{proof}[Proof of Proposition~\ref{prop:low2} for $\mc{F}_{\rm low,2}$] Let $(\omega,m,\Lambda)$  be an admissible frequency triple such that $(\omega,m,\Lambda)\in\mc{F}_{\rm low,2}$ and $a>\tilde{a}_0$. Let $\omega_0:=|\omega-m\upomega_+|>0$. In our construction of virial currents, we will use large positive constants $C$ and $R_1^*,R_2^*,R_3^*$ as well as small positive parameters $p$, $\hat{p}$, $\tilde{p}$ and $\epsilon$.

\medskip
\noindent\textit{The $\hat{y}$ current.} To take advantage of the fact that $\omega_0$ is bounded away from zero, we define a current, $\hat{y}$, such that the term $\hat{y}'\omega_0^2|\Psi|^2$ dominates. One could consider the exponential current in \cite[Proposition 8.7.3]{Dafermos2016b}. Indeed, we will choose a current which is exponentially decaying as $r^*\to \infty$ in a large compact $r^*$ region, but in contrast with the $s=0$ case, we make $\hat{y}'$ have milder decay as $r^*\to\pm \infty$.  Indeed, for some $R^*\geq C^2>1$ to be fixed, we set $\hat y=\hat{y}_0$, where $\hat{y}_0$ is given in \eqref{eq:standard-current-y0-hat} setting $\delta=1/2$, i.e.
\begin{align*}
\hat y(r^*)&:=-\exp\lp[-C(r-r_+)\rp]\mathbbm{1}_{[-R^*,R^*]}+\hat y(-R^*)\lp(2-\frac{(R^*)^{1/2}}{(-r^*)^{1/2}}\rp)\mathbbm{1}_{(-\infty,-R^*)}+\frac{\hat y(R^*)(R^*)^{1/2}}{(r^*)^{1/2}}\mathbbm{1}_{(R^*,\infty)}\,,
\end{align*}
with the properties outlined in \eqref{eq:standard-current-y0-hat-properties}.

With our choice of $\hat{y}$, we indeed have $\hat{y}'\omega_0^2|\Psi|^2$ dominating the bulk term on the left hand side of a $\hat{y}$ estimate. To see this, let $\chi$ be a smooth cutoff function which is $0$ on $(-\infty,-R^*]$ and 1 on $[-R^*+1,\infty)$. Let $\hat{\mc{V}}_{(k)}=\Re\mc{V}_{(k)}-\Re\mc{V}_{(k)}(r_+)$. As $\hat{\mc{V}}_{(k)}=O(r-r_+)$, $\hat{\mc{V}}_{(k)}'=O(\Delta)$ as $r\to r_+$ and $\hat{\mc{V}}_{(k)}=O(r^{-2})$ as $r\to \infty$, we have
\begin{align*}
&\lp|\int_{-\infty}^\infty (\hat y\hat{\mc{V}}_{(k)})'|\uppsi_{(k)}|^2dr^*\rp| \\
&\quad\leq \int_{-\infty}^\infty \hat{y}'\lp(\lp|\hat{\mc{V}}_{(k)}(1-\chi)\rp|+\lp|\hat{y}\hat{[\mc{V}}_{(k)}(1-\chi)]'/\hat{y}'\rp|\rp)|\uppsi_{(k)}|^2dr^*+\int_{-\infty}^\infty 2|\hat{y}\chi\hat{\mc{V}}_{(k)}||\uppsi_{(k)}'||\uppsi_{(k)}|dr^* \\
&\quad\leq  \int_{-\infty}^\infty \lp\{\frac18\hat{y}'|\uppsi_{(k)}'|+\frac18\hat{y}'\omega_0^2|\uppsi_{(k)}|^2B\lp[\lp(\frac{\hat y \chi\hat{\mc{V}}_{(k)}}{\hat y'\omega_0}\rp)^2+\lp|\frac{\hat{y}\hat{\mc{V}}_{(k)}(1-\chi)}{\hat{y}'\omega_0^2}\rp|+\lp|\frac{\hat{y}\hat{\mc{V}}_{(k)}'\chi}{\hat{y}'\omega_0^2}\rp|+\lp|\frac{\hat{y}\hat{\mc{V}}_{(k)}\chi'}{\hat{y}'\omega_0^2}\rp|\rp]\rp\}dr^*\\
&\quad\leq \frac18 \int_{-\infty}^\infty \lp\{\hat{y}'|\uppsi_{(k)}'|+\hat{y}'\omega_0^2|\uppsi_{(k)}|^2\frac{B(M,\omega_0,|s|)}{C}\rp\}dr^*\leq \frac18 \int_{-\infty}^\infty \lp\{\hat{y}'|\uppsi_{(k)}'|+\hat{y}'\omega_0^2|\uppsi_{(k)}|^2\rp\}dr^*\,,
\end{align*}
for large enough $C$ depending on $\tilde{a}_0$.  Moreover, note
\begin{align*}
\frac{1}{r^{1/2}}\frac{|\hat{y}|\Delta r^{-3}}{\hat y'}\leq \frac{B}{C}\,,
\end{align*}
so estimate \eqref{eq:hat-y-estimate-low-k-global} from Lemma~\ref{lemma:y-estimate-low} becomes
\begin{align*}
&\int_{-\infty}^\infty \lp\{\frac12 \hat y'|\uppsi_{(k)}'|^2+\frac{5}{8}\hat y'\omega_0^2|\uppsi_{(k)}|^2\rp\}dr^*\\
&\quad \leq  2|\hat y(-\infty)|\omega_0^2\lp(\lp|\swei{A}{s}_{k,\mc{H}^{+}}\rp|^2+\lp|\swei{A}{s}_{k,\mc{H}^{-}}\rp|^2\rp)+\int_{-\infty}^\infty 2\hat y\Re\lp[\mathfrak{G}_{(k)}\overline{\uppsi_{(k)}}'\rp]dr^*\\
&\quad \qquad +k\sum_{j=0}^{k-1} \frac{B}{C}\hat{y}'\omega_0^2|\uppsi_{(j)}|^2dr^* +(|s|-k)\int_{-\infty}^\infty \frac{B}{C r^2}\hat{y}'|\uppsi_{(k+1)}|^2dr^*\,,
\end{align*}
for $\omega_{\rm low}$ sufficiently small and for large enough $C$ depending on $\tilde{a}_0$. As the estimate holds for all $k\in\{0,...,|s|\}$, we can apply it to each $j$ in the sum on the right hand side. For instance, by applying it to $j=k-1$, we obtain
\begin{align*}
&\int_{-\infty}^\infty \lp\{\frac12 \hat y'|\uppsi_{(k)}'|^2+\frac{5}{8}\hat y'\omega_0^2|\uppsi_{(k)}|^2\rp\}dr^*\\
&\quad \leq  2|\hat y(-\infty)|\omega_0^2\lp(\lp|\swei{A}{s}_{k,\mc{H}^{+}}\rp|^2+\lp|\swei{A}{s}_{k,\mc{H}^{-}}\rp|^2\rp)+\int_{-\infty}^\infty 2\hat y\Re\lp[\mathfrak{G}_{(k)}\overline{\uppsi_{(k)}}'\rp]dr^*\\
&\quad\qquad+ \frac{B}{C}\lp\{|\hat y(-\infty)|\omega_0^2\lp|\swei{A}{s}_{k-1,\mc{H}^{+\sign s}}\rp|^2+\int_{-\infty}^\infty 2\hat y\Re\lp[\mathfrak{G}_{(k-1)}\overline{\uppsi_{(k-1)}}'\rp]dr^*\rp\}\\
&\quad \qquad +  k\sum_{j=0}^{k-2}\int_{-\infty}^\infty \frac{B}{C}\hat{y}'\omega_0^2|\uppsi_{(j)}|^2dr^* +(|s|-k)\int_{-\infty}^\infty \frac{B\hat{y}'}{C r^2}|\uppsi_{(k+1)}|^2dr^* +(|s|-k+1)\int_{-\infty}^\infty \frac{B\hat{y}'}{C^2}|\uppsi_{(k)}|^2dr^*\,,
\end{align*}
where the last term on the right hand side can clearly be absorbed into the left hand side as long as $C=C(\tilde{a}_0)$ is sufficiently large. By iterating, we arrive at
\begin{align*}
&\int_{-\infty}^\infty \lp\{\frac12 \hat y'|\uppsi_{(k)}'|^2+\frac{1}{2}\hat y'\omega_0^2|\uppsi_{(k)}|^2\rp\}dr^*\numberthis \label{eq:low2-cumulative-1}\\
&\quad \leq  (|s|-k)\int_{-\infty}^\infty \frac{B\hat{y}'}{C r^2}|\uppsi_{(k+1)}|^2dr^*+2|\hat y(-\infty)|\omega_0^2\sum_{j=0}^k\lp(\lp|\swei{A}{s}_{j,\mc{H}^{+}}\rp|^2+\lp|\swei{A}{s}_{j,\mc{H}^{-}}\rp|^2\rp)\\
&\quad\qquad \sum_{j=0}^{k-1}\int_{-\infty}^\infty 2\hat y\Re\lp[\mathfrak{G}_{(j)}\overline{\uppsi_{(j)}}'\rp]dr^*\,. 
\end{align*}

\medskip
\noindent\textit{The energy currents near $r=r_+$.} Our $\hat{y}$ current produces boundary terms at $r=r_+$ which have the wrong sign. To correct them, we introduce a $K$-type energy current at the level of $\Psi$ and, given the fact that the frequency regime $\mc{F}_{\rm low, 2}$ contains superradiant frequencies, it must be localized away from $r=\infty$.

As in all other frequency regimes, we add Killing energy currents; however, as $\omega_0$ is not small, we will require the $Q^K$ current to be multiplied by a small parameter, hence it will not be strong enough to control the boundary terms generated by $\hat y$. To control the boundary terms, we can appeal to Lemma~\ref{lemma:bdry-term-via-constraint-bdd} or, if the precise boundary term relations of Lemma~\ref{lemma:uppsi-general-asymptotics} hold, we can introduce a Teukolsky--Starobinsky $K$-type energy current as well, which does not produce coupling error terms; for this reason, we now restrict to $|s|=0,1,2$.

Let $\chi_1\leq 1$ be a smooth function with $\chi_1=1$ in $(-\infty,R_2^*]$ and 0 in $[R_3^*,\infty)$ for some $R_3=2R_2^*\gg 1$. Adding $E_W\chi_1Q^{W,\,K}$, we generate bulk error terms given by Lemma~\ref{lemma:wronskian-energy-currents}; moreover  $E_1\chi_1Q^{K}$ will also generate localization errors. Assuming $E_1\leq B E_W$ and $R_3$ sufficiently large that $|am|/R_3^*\ll 1$, such errors are controlled by
\begin{align}
\frac{BE_W\omega_0}{R_3^*} \int_{R_2^*}^{R_3^*} \sum_{k=0}^{|s|}\lp(R_3^*|\uppsi_{(k)}'|^2+\frac{1}{R_3^*}|\uppsi_{(k)}|^2\rp)R_3^{2(|s|-k)}dr^*\,. \label{eq:low2-wronskian-errors}
\end{align}
As the Killing current $Q^K$ is not conserved for $s\neq 0$ and $a\neq 0$, it generates coupling error terms:
\begin{align*}
&E_1 \sum_{k=0}^{|s|-1}\int_{-\infty}^\infty aw\chi_1(\omega-m\upomega_+)\Im\lp[\overline{\Psi}(c_{s,s,k}^{\mr{id}}+imc_{s,s,k}^{\Phi})\uppsi_{(k)}\rp]dr^* \\
&\quad\leq \int_{-\infty}^\infty \lp\{\frac{1}{4}\hat{y}'\omega_0^2|\Psi|^2+|s|\sum_{k=0}^{|s|-1}\frac{1}{8|s|}\hat{y}'\omega_0^2|\uppsi_{(k)}|^2\frac{B E_1^2}{C|\hat{y}(R^*)|^2}|\uppsi_{(k)}|^2\rp\}dr^* \,,
\end{align*}
where we have used the properties of $\hat{y}$ in \eqref{eq:standard-current-y0-hat-properties} and the assumption that $C(\tilde{a}_0)$ is sufficiently large that $(C\omega_0^2)^{-1}\leq 1$. We take $E_1$ to sufficiently small depending on $R^*$ and $C$ so that it can be controlled by the left hand side of \eqref{eq:low2-cumulative-1}.

Finally, let us discuss the boundary terms: from Lemma~\ref{lemma:properties-frequencies-low-omega}(iv), we have
\begin{align*}
|\hat y(-\infty)|\omega_0^2\lp(\lp|\swei{A}{s}_{k,\mc{H}^{+}}\rp|^2+\lp|\swei{A}{s}_{k,\mc{H}^{-}}\rp|^2\rp) &\leq B|\hat y(-\infty)|\omega_0^2\frac{\mathfrak{D}_{s,k}^{\mc H}}{\mathfrak{D}_{s}^{\mc H}}\lp(\lp|\swei{A}{s}_{\mc{H}^{+}}\rp|^2+\lp|\swei{A}{s}_{\mc{H}^{-}}\rp|^2\rp)\\
&\leq B|\hat y(-\infty)|\omega_0^2\lp(\lp|\swei{A}{s}_{\mc{H}^{+}}\rp|^2+\lp|\swei{A}{s}_{\mc{H}^{-}}\rp|^2\rp)\,.
\end{align*}
Hence,  from \eqref{eq:low2-cumulative-1}, we arrive at
\begin{align*}
&\sum_{k=0}^{|s|}\int_{-\infty}^\infty \lp\{\frac12 \hat y'|\uppsi_{(k)}'|^2+\frac{1}{2}\hat y'\omega_0^2|\uppsi_{(k)}|^2\rp\}dr^*+ \sum_{k=0}^{|s|}|\hat y(-\infty)|\omega_0^2\lp|\swei{A}{s}_{k,\mc{H}^{+}}\rp|^2\\
&\quad \leq  \frac{BE_W\omega_0}{R_3^*} \int_{R_2^*}^{R_3^*} \sum_{k=0}^{|s|}\lp(R_3^*|\uppsi_{(k)}'|^2+\frac{1}{R_3^*}|\uppsi_{(k)}|^2\rp)R_3^{2(|s|-k)}dr^*\\
&\quad\qquad +B\sum_{k=0}^{|s|}|\hat y(-\infty)|\omega_0^2\lp|\swei{A}{s}_{k,\mc{H}^{-}}\rp|^2+B\int_{-\infty}^\infty \lp\{\sum_{k=0}^{|s|}\hat y\Re\lp[\mathfrak{G}_{(k)}\overline{\uppsi_{(k)}}'\rp]+E_W(Q^{W,K})'\rp\}dr^* \,,\numberthis \label{eq:low2-cumulative-2a}
\end{align*}
where we have chosen $E_W$ to be at least proportional to  $|\hat y(-\infty)|$ and used Lemma~\ref{lemma:properties-frequencies-low-omega}(iv) again.

\medskip
\noindent \textit{The $h$ current and control of $K$-type energy localization errors.} It is immediate from \eqref{eq:low2-cumulative-2a} that we require an additional current or currents to absorb the first terms on the right hand side, which are localization errors of the $K$-type energy current at levels $k=0,...,|s|$. To construct such currents, we exploit the positiveness of $\Re\mc{V}_{(k)}$ in a bounded region of large $r$: for each $k\in\{0,...,|s|\}$, we consider an $h_{(k)}$ supported on $[R_1^*,2e^{p^{-1}}R_3^*]$, for some $R_3^*>R_1^*=R^*$ sufficiently large and $p\in(0,1)$ small; in fact, we will have $h_{(k)}=h_{k,\rm{left}}+h_{k,\rm{right}}$.

Given the Teukolsky--Starobinsky energy localization errors, \eqref{eq:low2-wronskian-errors}, to absorb, we would like $h_{k}$ to be a polynomial in $r$ in the region where such errors occur. We are, however, constrained by the error which is further introduced by derivatives of $h_{(k)}$. From the calculation \eqref{eq:error-h-polynomial} in Section~\ref{sec:bounded-smallness}, we find that $\tilde{h}_k=r^N$ for  $N=2[|s|-k-1/(4|s|+4)]+1$, generates no errors. Thus, let $\chi$ be some cutoff function taking values in $[0,1]$ and such that $\chi=0$ for $r^*\leq R_1^*$ and $\chi=1$ for $r^*\geq 2R_1^*$, we set
\begin{align*}
h_{k,\rm{left}}= \hat{p}\frac{|\hat{y}(R^*)|}{(2R_1)^{N-1}}\omega_0^2\chi \tilde{h}_{k}\mathbbm{1}_{[R_1^*,R_3^*]}\,,
\end{align*}
for some small $\hat{p}$ to be determined. Using \eqref{eq:error-h-polynomial}, we obtain
\begin{align*}
&h_{k,\rm{left}}''+(|s|-k)\frac{w'}{w}h_{k,\rm{left}}' \\
&\quad= \hat{p}\frac{|\hat{y}(R^*)|}{(2R_1)^{N-1}}\omega_0^2\lp[\lp(\tilde{h}_{k}''+(|s|-k)\frac{w'}{w}\tilde{h}_{k}'\rp)\chi+\lp(\chi''+(|s|-k)\frac{w'}{w}\chi'\rp)\tilde{h}_{k}+2 \chi'\tilde{h}_{k}'\rp]\mathbbm{1}_{[R_1^*,R_3^*]}\\
&\quad\leq  \hat{p}\hat{y}'\omega_0^2\lp(\frac{2(r^*)^{3/2}}{(R_1^*)^{1/2}r}\rp)\lp(\frac{r}{2R_1}\rp)^{N-1}\lp[\lp(\chi''+(|s|-k)\frac{w'}{w}\chi'\rp)r^2+2 r\chi'N\rp]\mathbbm{1}_{[R_1^*,2R_1^*]}\\
&\quad\leq B\hat{p}\hat{y}'\omega_0^2\mathbbm{1}_{[R_1^*,2R_1^*]}\,,
\end{align*}
which can be absorbed by the $\hat{y}$ current if $\hat{p}$ is sufficiently small.
Finally, $h_{k, \rm{right}}$ is modeled after the standard current $h_+$ introduced in Section~\ref{sec:low-hierarchy-h}: $h_{k, \rm{right}}=h_+\mathbbm{1}_{[R_3^*,\infty)}$, hence
\begin{align*}
h_{k, \rm{right}}''+(|s|-k)\frac{w'}{w}h_{k, \rm{right}}'=\frac{Bp}{(r^*)^2}\hat{p}|\hat{y}(R^*)|R_1\lp(\frac{R_3}{R_1}\rp)^N\omega_0^2\mathbbm{1}_{[R_3^*,2e^{p^{-1}}R_3^*]}\,,
\end{align*} 
which is integrable in $r^*$.

Using the new currents, we now find that \eqref{eq:low2-wronskian-errors} is controlled by
\begin{align*}
&\frac{B E_W \omega_0}{R_3^{1-\frac{(|s|-k)}{2(|s|+1)}}}\int_{-\infty}^{\infty}\lp[|\uppsi_{(k)}'|^2+\frac{1}{R_3^2}|\uppsi_{(k)}|^2\rp]R_3^{N}\mathbbm{1}_{[R_2^*,R_3^*]}dr^* \numberthis \label{eq:low2-wronskian-errors-control}\\
&\quad \leq \frac{BE_W(R_1)^{2|s|}}{R_3^{1/2}\hat{p}\omega_0\lp|\hat{y}(R^*)\rp|}\int_{-\infty}^{\infty}\lp\{h_{(k)}|\uppsi_{(k)}'|^2+h_{(k)}\lp[\Re\mc{V}_{(k)}-\frac{|s|-k}{2}\lp(\frac{w'}{w}\rp)'-\omega^2\rp]|\uppsi_{(k)}|^2\rp\}\mathbbm{1}_{[R_2^*,R_3^*]}dr^*\,,
\end{align*} 
where the prefactor is small as long $R_3$ is sufficiently large compared to $R_1$. We stress that even in the $s=0$ case, a current such as $h:=h_{(|s|)}$ is still necessary: the standard energy current introduces errors exactly of the form above (when $k=|s|$). However, in that case, as our $h$ is simply linear in the region $[R_2^*,R_3^*]$. In the case where $E_W\equiv 0$, nevertheless  the largeness of $R_3$ allows us to obtain an estimate
\begin{align*}
&\sum_{k=0}^{|s|}(R_3^*)^{1/2}\int_{R_2^*}^{R_3^*} w(|\uppsi_{(k)}'|^2+|\uppsi_{(k)}|^2)dr^* \leq B\sum_{k=0}^{|s|} \int_{-\infty}^{\infty} \lp((h_{(k)}+\hat y)|\uppsi_{(k)}'|^2+(\hat y'\omega_0^2+h_{(k)}(\mc V-\omega^2))|\uppsi_{(k)}|^2\rp)dr^*\\
&\quad \leq \sum_{k=0}^{|s|} |\hat{y}(-\infty)|^2\omega_0^2|\swei{A}{s}_{k,\mc H^+}|^2
+B\sum_{k=0}^{|s|} \int_{R_3^*}^\infty \lp(h_{(k)}''+(|s|-k)\frac{w'}{w}h_{(k)}'\rp)|\uppsi_{(k)}|^2dr^*\,,
\end{align*}
where there is a gain in the bulk term on the left hand side when compared to the boundary term on the right hand side.

\medskip
\noindent \textit{The $y$ current.} After the application of the $h_{(k)}$ current, we still need to worry about the error terms introduced by derivatives of $h_{(k)}$ in $[R_1^*,2R_1^*]$ and $[R_3^*,2e^{p^{-1}}R_3^*]$: as previously computed,
\begin{align*}
&\frac12\lp(h''_{(k)}+\lp(|s|-k\rp)\frac{w'}{w}h_{(k)}'\rp) \leq B\hat{p}\hat{y}'\omega_0^2\mathbbm{1}_{[R_1^*,2R_1^*]}+\frac{Bp}{(r^*)^2}\hat{p}|\hat{y}(R^*)|R_1\lp(\frac{R_3}{R_1}\rp)^N\omega_0^2\mathbbm{1}_{[R_3^*,2e^{p^{-1}}R_3^*]}\,.
\end{align*}
While the first term, supported in $[R_1^*,2R_1^*]$,  has been specifically chosen so that it can be absorbed by the $\hat{y}$ current, the latter cannot be absorbed by either $h_{(k)}$ or $\hat y$. 

As we have done in $\mc{F}_{\rm low, 1}$, we control the errors by gluing in a $y_{(k)}$ current  near $r^*=\infty$ that takes advantage of the fact that $\Re\mc{V}'<0$. To this end, for some $\delta\in(0,1]$ and $\tilde{p}\in(0,1)$, we let $y_{(k)}$ be modeled on  \eqref{eq:low-frequencies-y} introduced earlier, i.e.
\begin{align*}
y_{(k)}(r^*)&= \tilde{p}\lp(\int_{R_2^*}^{r^*}h_{(k)}(x^*)dx^*\rp)\mathbbm{1}_{[R_2^*,R_3^*]}\\
&\qquad +y_{(k)}(R_3^*)\lp[1+\frac{64\epsilon}{R_3^*}\lp(\frac{1}{r^*\Re\mc{V}_{(k)}}-\frac{1}{R_3^*\Re\mc{V}_{(k)}|_{r^*=R_3^*}}\rp)\rp]\mathbbm{1}_{[R_3^*, 2e^{p^{-1}}R_3^*]} \\
&\qquad+y_{(k)}(2e^{p^{-1}}R_3^*)\lp[1+\frac{1}{(2e^{p^{-1}}R_3^*)^{\delta}}-\frac{1}{(r^*)^{\delta}}\rp]\mathbbm{1}_{[2e^{p^{-1}}R_3^*, \infty)}\,.
\end{align*}
Then, for $\epsilon$ sufficiently small and for $p$ sufficiently small depending on $\tilde{p}\epsilon$, $y_{(k)}'>0$, 
\begin{align*}
1-\frac{r^* y_{(k)}'}{y_{(k)}}\geq b(\tilde{p},p,\epsilon) \frac{y}{(r^*)^2}\,,
\end{align*}
and \eqref{eq:low-frequencies-h+y} holds, i.e.
\begin{align*}
&\frac34 h_{(k)}\lp(\mc{V}_{(k)} +\frac{|s|-k}{2}\lp(\frac{w'}{w}\rp)'-\omega^2\rp) +\frac34 y'\omega^2-\frac34(y\Re\mc{V}_{(k)})' -\frac14y'\Re\mc{V}_{(k)}-\frac12\lp(h_{k}'' +(|s|-k)\frac{w'}{w}h_{k}'\rp)\\
&\quad \geq  \frac{1}{2}h_{(k)}\lp(\mc{V}_{(k)}+\frac{3(|s|-k)}{4}\lp(\frac{w'}{w}\rp)' -\omega^2\rp)+\frac34 y_{(k)}'\omega^2-\frac12 y_{(k)}\Re\mc{V}_{(k)}'\mathbbm{1}_{[R_3^*,\infty)}  -B\hat{p}\hat{y}'\omega_0^2\mathbbm{1}_{[R_1^*,2R_1^*]}\,,
\end{align*}
and hence
\begin{align*}
&\int_{-\infty}^\infty \lp\{(y_{(k)}'+h_{(k)})|\uppsi_{(k)}'|^2+\frac{1}{2}h_{(k)}\lp(\mc{V}_{(k)}-\frac{3(|s|-k)}{4}\lp(\frac{w'}{w}\rp)' -\omega^2\rp)|\uppsi_{(k)}|^2\rp\}dr^*\\
&\qquad +\frac{1}{2}\int_{-\infty}^\infty\lp[y_{(k)}'\omega^2-y_{(k)}\Re\mc{V}_{(k)}'\mathbbm{1}_{[R_3^*,\infty)}\rp]|\uppsi_{(k)}|^2dr^*\\
&\quad \leq  2y_{(k)}(\infty)\omega^2\lp(\lp|\swei{A}{s}_{k,\mc{I}^{+}}\rp|^2+\lp|\swei{A}{s}_{k,\mc{I}^{-}}\rp|^2\rp)+ \int_{-\infty}^\infty B\hat{p}\hat{y}'\omega_0^2|\uppsi_{(k)}|^2\mathbbm{1}_{[R_1^*,2R_1^*]}dr^*\\
&\quad\qquad -B(\varepsilon_{\rm width},\omega_{\rm high})(|s|-k)\int_{-\infty}^\infty R_3^2\lp(\omega^2+\frac{1}{r^2}\rp)y_{(k+1)}\Re\mc{V}_{(k+1)}'|\uppsi_{(k+1)}|^2dr^*\\
&\quad\qquad+\int_{-\infty}^\infty \{2y_{(k)}\Re\lp[\mathfrak{G}_{(k)}\overline{\uppsi_{(k)}}'\rp]+h_{(k)}\Re\lp[\mathfrak{G}_{(k)}\overline{\uppsi_{(k)}}\rp]\}dr^*\\ 
&\quad\qquad +k\sum_{j=0}^{k-1} B(\varepsilon_{\rm width},\omega_{\rm high}) \lp[R_3^* h_{(k)}\lp(\mc{V}_{(j)}-\frac{3(|s|-k)}{4}\lp(\frac{w'}{w}\rp)'-\omega^2\rp)-y_{(k)}\Re\mc{V}_{(j)}'\rp]|\uppsi_{(j)}|^2dr^*\,, \numberthis\label{eq:low2-cumulative-3a}
\end{align*}
where we have used estimate~\eqref{eq:h-estimate-low} from Lemma~\ref{lemma:h-estimate-low} and estimate \eqref{eq:y-estimate-low-k} from Lemma~\ref{lemma:y-estimate-low}. 

\medskip 
\noindent \textit{The energy current near $r=\infty$.} Though it corrects the $h''$ error, adding a $y$ current produces boundary terms at $r=\infty$ which manifestly have the wrong sign. We add a localized $Q^T$ current, i.e.\ $E\chi_2Q^T$, where $\chi_2\leq 1$ is a smooth function that is 1 in $[R_2^*,\infty)$ and 0 in $(-\infty,R_1^*]$ and where $E\geq 3\max\{|\hat y(-\infty)|,\max _k y_{(k)}(\infty)\}$. For any $s$, we obtain a localization error
\begin{align*}
\int_{-\infty}^\infty E\chi_2'\omega|\uppsi_{(k)}'||\uppsi_{(k)}|dr^*&\leq \frac{BE\omega_{\rm low}}{R_2^*-R_1^*}\int_{-\infty}^\infty|\Psi'||\Psi|\mathbbm{1}_{[R_1^*,R_2^*]}dr^*\\
&\leq \int_{-\infty}^\infty \lp(\lp(\frac{BE}{(R_2^*-R_1^*)R^*|\hat{y}(R^*)|}\frac{\omega_{\rm low}}{\omega_0\hat p}\rp)^2\hat y'|\uppsi_{(k)}'|+B\hat p\hat{y}'\omega_0^2|\uppsi_{(k)}|^2\rp)dr^*\\
&\leq \int_{-\infty}^\infty \lp(\frac{1}{128}\hat y'|\uppsi_{(k)}'|+B\hat p\hat{y}'\omega_0^2|\uppsi_{(k)}|^2\rp)dr^*
\end{align*}
by making $\omega_{\rm low}$ sufficiently small depending on $E$, $R^*$, $R_2^*$ and $\hat{p}$. For $s\neq 0$, we moreover have coupling errors
\begin{align*}
&\int_{-\infty}^\infty \sum_{j=0}^{k-1} E\chi_2aw\omega\Im\lp[\overline{\uppsi_{(k)}}\lp(c_{s,k,j}^{\rm id}+imc_{s,k,j}^\Phi\rp)\uppsi_{(j)}\rp]dr^* \\
&\quad\leq \int_{-\infty}^\infty \lp\{B\hat p \hat y'\omega_0^2|\uppsi_{(k)}|+ |s|B(\varepsilon_{\rm width},\omega_{\rm high})\lp(\frac{E}{R^*|\hat{y}(R^*)|}\frac{\omega_{\rm low}}{\omega_0^2\hat p}\rp)^2\sum_{j=0}^{k-1}\hat{y}'\omega_0^2\mathbbm{1}_{[R_2^*,\infty)}|\uppsi_{(j)}|^2\rp\}dr^*\\
&\quad\leq \sum_{j=0}^k\int_{-\infty}^\infty B\hat p \hat y'\omega_0^2|\uppsi_{(j)}|dr^*\,,
\end{align*}
where we have appealed to the smallness of $\omega_{\rm low}$ compared to $E$, $R^*$ and $p$. Thus, from \eqref{eq:low2-cumulative-3a} we arrive at 
\begin{align*}
&\int_{-\infty}^\infty \lp\{(y_{(k)}'+h_{(k)})|\uppsi_{(k)}'|^2+\frac{1}{2}h_{(k)}\lp(\mc{V}_{(k)}-\frac{3(|s|-k)}{4}\lp(\frac{w'}{w}\rp)' -\omega^2\rp)|\uppsi_{(k)}|^2\rp\}dr^*\\
&\qquad +\frac{1}{2}\int_{-\infty}^\infty\lp[y_{(k)}'\omega^2-y_{(k)}\Re\mc{V}_{(k)}'\mathbbm{1}_{[R_3^*,\infty)}\rp]|\uppsi_{(k)}|^2dr^* + y_{(k)}(\infty)\omega^2\lp|\swei{A}{s}_{k,\mc{I}^{+}}\rp|^2\\
&\quad \leq  B E\omega^2\lp|\swei{A}{s}_{k,\mc{I}^{-}}\rp|^2+ \int_{-\infty}^\infty \lp[\frac{1}{128}\hat y'|\uppsi_{(k)}|+B\sum_{j=0}^k\hat{p}\hat{y}'\omega_0^2|\uppsi_{(j)}|^2\rp]\mathbbm{1}_{[R_1^*,2R_1^*]}dr^*\\
&\quad\qquad -B(\varepsilon_{\rm width},\omega_{\rm high})(|s|-k)\int_{-\infty}^\infty R_3^2\lp(\omega^2+\frac{1}{r^2}\rp)y_{(k+1)}\Re\mc{V}_{(k+1)}'|\uppsi_{(k+1)}|^2dr^*\\
&\quad\qquad+\int_{-\infty}^\infty \{2y_{(k)}\Re\lp[\mathfrak{G}_{(k)}\overline{\uppsi_{(k)}}'\rp]+h_{(k)}\Re\lp[\mathfrak{G}_{(k)}\overline{\uppsi_{(k)}}\rp]-E\omega\Im\lp[\mathfrak{G}_{(k)}\overline{\uppsi_{(k)}}\rp]\}dr^*\\ 
&\quad\qquad +k\sum_{j=0}^{k-1} B(\varepsilon_{\rm width},\omega_{\rm high}) \lp[R_3^* h_{(k)}\lp(\mc{V}_{(j)}-\frac{3(|s|-k)}{4}\lp(\frac{w'}{w}\rp)'-\omega^2\rp)-y_{(k)}\Re\mc{V}_{(j)}'\rp]|\uppsi_{(j)}|^2dr^*\,.
\end{align*}

Now, notice that $\omega_{\rm low}$ is small enough that $R_3\omega\ll 1$, rending the term in $\uppsi_{(k+1)}$ involving $\omega^2R_3^2$ neutral in powers of $R_3$.  Since  $h_{(k)},y_{(k)}$ differ from $h_{(j)},y_{(j)}$ by a factor of $R_3^{-2(k-j)}$, each of the terms involving $\uppsi_{(j)}$ in the last line is easily controlled by the left hand side of the same estimate with $k=j$. Thus, iterating the estimate, we obtain
\begin{align*}
&\int_{-\infty}^\infty \lp\{(y_{(k)}'+h_{(k)})|\uppsi_{(k)}'|^2+\frac{1}{2}h_{(k)}\lp(\mc{V}_{(k)}-\frac{3(|s|-k)}{4}\lp(\frac{w'}{w}\rp)' -\omega^2\rp)|\uppsi_{(k)}|^2\rp\}dr^*\\
&\qquad +\frac{1}{2}\int_{-\infty}^\infty\lp[y_{(k)}'\omega^2-y_{(k)}\Re\mc{V}_{(k)}'\mathbbm{1}_{[R_3^*,\infty)}\rp]|\uppsi_{(k)}|^2dr^*+E\omega^2\lp|\swei{A}{s}_{k,\mc{I}^{+}}\rp|^2\\
&\quad \leq  BE\omega^2\sum_{j=0}^k\lp|\swei{A}{s}_{j,\mc{I}^{-}}\rp|^2+ \sum_{j=0}^k\int_{-\infty}^\infty \lp[\frac{1}{128}\hat y' |\uppsi_{(j)}'|^2+B\hat{p}\hat{y}'\omega_0^2|\uppsi_{(j)}|^2\rp]\mathbbm{1}_{[R_1^*,2R_1^*]}dr^*\\
&\quad\qquad -B(\varepsilon_{\rm width},\omega_{\rm high})(|s|-k)\int_{-\infty}^\infty R_3^2\lp(\omega^2+\frac{1}{r^2}\rp)y_{(k+1)}\Re\mc{V}_{(k+1)}'|\uppsi_{(k+1)}|^2dr^*\\
&\quad\qquad+B\sum_{j=0}^k\int_{-\infty}^\infty \{2y_{(j)}\Re\lp[\mathfrak{G}_{(j)}\overline{\uppsi_{(k)}}'\rp]+h_{(j)}\Re\lp[\mathfrak{G}_{(j)}\overline{\uppsi_{(j)}}\rp]-E\omega\Im\lp[\mathfrak{G}_{(j)}\overline{\uppsi_{(j)}}\rp]\}dr^*\,. \numberthis\label{eq:low2-cumulative-3}
\end{align*}
Combining with \eqref{eq:low2-cumulative-2a} and \eqref{eq:low2-wronskian-errors-control}, if $\hat p$ is sufficiently small, we obtain
\begin{align*}
&\sum_{k=0}^{|s|}\int_{-\infty}^\infty \lp\{(\hat y +y_{(k)}'+h_{(k)})|\uppsi_{(k)}'|^2+h_{(k)}\lp(\mc{V}_{(k)}-\frac{3(|s|-k)}{4}\lp(\frac{w'}{w}\rp)' -\omega^2\rp)|\uppsi_{(k)}|^2\rp\}dr^*\\
&\qquad +\sum_{k=0}^{|s|}\int_{-\infty}^\infty\lp[\hat y'\omega_0^2 +y_{(k)}'\omega^2-y_{(k)}\Re\mc{V}_{(k)}'\mathbbm{1}_{[R_3^*,\infty)}\rp]|\uppsi_{(k)}|^2dr^*+\sum_{k=0}^{|s|}\omega^2\lp|\swei{A}{s}_{k,\mc{I}^{+}}\rp|^2\\
&\quad \leq  B  E\omega^2\sum_{j=0}^{|s|}\lp(\lp|\swei{A}{s}_{j,\mc{I}^{-}}\rp|^2+\lp|\swei{A}{s}_{j,\mc{I}^{-}}\rp|^2\rp)\\
&\quad\qquad+B\sum_{j=0}^{|s|}\int_{-\infty}^\infty \{2\hat y\Re\lp[\mathfrak{G}_{(j)}\overline{\uppsi_{(k)}}'\rp]+2y_{(j)}\Re\lp[\mathfrak{G}_{(j)}\overline{\uppsi_{(k)}}'\rp]+h_{(j)}\Re\lp[\mathfrak{G}_{(j)}\overline{\uppsi_{(j)}}\rp]-E\omega\Im\lp[\mathfrak{G}_{(j)}\overline{\uppsi_{(j)}}\rp]\}dr^*\,.
\end{align*}

Alternatively, if the precise boundary term relations of Lemma~\ref{lemma:uppsi-general-asymptotics} hold, we can control the boundary term at infinity by adding a localized Teukolsky--Starobinsky type $T$ energy current, $E_{W,4}\chi_4Q^{W,T}$, where $\chi_4=1$ for $r^*\geq R_2^*$ and $\chi_4=0$ for $r^*\leq R_2^*/2$, and where $E_{W,4}$ is sufficiently large. We address the localization error terms using the same strategy as in $\mc{F}_{\rm low,1a}$: with $R_2^*=2R_3^*$  sufficiently large depending on $\varepsilon_{\rm width}$ and $\omega_{\rm high}$ (as previously required) and $\omega_{\rm low}$ sufficiently small depending on $R_3^*$, then
\begin{align*}
&B\frac{E_W\omega_{\rm low}}{R_3^*}\int_{\supp(\chi_4')}\sum_{k=0}^{|s|}(R_3^*)^{2(|s|-k)+1}\lp(1+\frac{|am|}{R_3^*}\rp)\lp[|\uppsi_{(k)}|^2+(R_3^*)^{-2}|\uppsi_{(k)}|^2\rp]dr^*\\
&\quad\leq B\omega_{\rm low}^{1/2}\int_{-\infty}^\infty \lp[h|\Psi'|^2+h(\mc{V}-\omega^2)|\Psi|^2\rp]dr^*+B|s|\omega_{\rm low}^{1/2}|s|\sum_{k=0}^{|s|-1}\int_{-\infty}^\infty w\lp[|\uppsi_{(k)}'|^2+|\uppsi_{(k)}|^2\rp]dr^*
\end{align*}
can be absorbed into the left hand side of \eqref{eq:low2-cumulative-3}. Combining with \eqref{eq:low2-cumulative-2a} and \eqref{eq:low2-wronskian-errors-control} as before, we have the estimate
\begin{align*}
&\omega^2\lp|\swei{A}{s}_{\mc{I}^+}\rp|^2+\omega_0^2\lp|\swei{A}{s}_{\mc{H}^+}\rp|^2+\frac 12 \int_{-\infty}^\infty \lp[y'+h+\hat y'\rp]|\Psi'|^2dr^* \\
&\quad\qquad +\int_{-\infty}^\infty \frac{1}{8}\lp[\frac18\hat y'\omega_0^2+y'\omega^2+h\lp(\mc{V} -\omega^2\rp)-y\mc{V}'\mathbbm{1}_{[R_3^*,\infty)}\rp]|\Psi|^2dr^* \numberthis \label{eq:low2-cumulative-4}\\
&\quad \leq  B(E, E_W,E_{W,4})\omega^2\lp|\swei{A}{s}_{\mc{I}^-}\rp|^2+ B(E,E_W)\omega_0^2\lp|\swei{A}{s}_{\mc{H}^-}\rp|^2+\int_{-\infty}^\infty \lp[E_W\chi_1(Q^{W,K})'+E_{W,4}\chi_4(Q^{W,T})'\rp]dr^*\\
&\quad\qquad +\int_{-\infty}^\infty \lp\{2(\hat y+y)\Re\lp[\mathfrak{G}\overline{\Psi}'\rp]+h\Re\lp[\mathfrak{G}\overline{\Psi}\rp]-\lp[E_1\chi_1(\omega-m\upomega_+)+E\chi_2\omega\rp]\Im\lp[\mathfrak{G}\overline{\Psi}\rp]\rp\}dr^*\\
&\qquad\quad + B|s|\sum_{j=0}^{|s|-1}\int_{-\infty}^\infty \lp((\hat{y}_{(j)}+y_{(j)})\Re\lp[\mathfrak{G}_{(j)}\overline{\uppsi_{(j)}}'\rp]+h_{(j)}\Re\lp[\mathfrak{G}_{(j)}\overline{\uppsi_{(j)}}'\rp]\rp)dr^*\,.
\end{align*}
This concludes the proof.
\end{proof}


\subsection{The bounded frequencies with nonzero time frequency}
\label{sec:intermediate}

In this section, we will establish Theorems~\ref{thm:ode-estimates-Psi-A} and~\ref{thm:ode-estimates-Psi-B} in the cases where the frequencies $(\omega,m,\Lambda)$ are assumed to be bounded and the time frequency, $\omega$, is bounded away from zero. Concretely, we will show:

\begin{proposition}\label{prop:int1}
Fix $s\in\{0,\pm 1,\pm 2\}$ and $M>0$. Then, for all $\delta\in(0,1]$, $\omega_{\rm high}>0$, $\varepsilon^{-1}_{\rm width}>0$, $\omega_{\rm low}>0$ and $a_0\in[0,M)$, for all $E,E_W>0$ such that one is sufficiently large depending on the latter quantities, for all $R^*$ sufficiently large depending on the previous, and for all $(a,\omega,m,\Lambda) \in \mc{F}_{\rm int,1}(\varepsilon_{\rm width},\omega_{\rm high},a_0)$,  there exist functions $y$, $\hat{y}$, $\chi_1$ and $\chi_2$,  for $0\leq k\leq |s|$ (we drop the subscript for $k=|s|$), satisfying the uniform bounds
$|y| +|\hat{y}|+|\chi_1|+|\chi_2|\leq 1 \,,$
such that, for all smooth $\Psi$ arising from a smooth solution to the radial ODE~\eqref{eq:radial-ODE-alpha} via \eqref{eq:def-psi0-separated} and \eqref{eq:transformed-transport-separated} and itself satisfying the radial ODE~\eqref{eq:radial-ODE-Psi}, we have the following estimates:
\begin{enumerate}[font=\normalfont\bfseries, label=\Alph*.]
\item[A.]  if $\Psi$ has the general boundary conditions of Lemma~\ref{lemma:uppsi-general-asymptotics}, we have the estimate
\begin{align*}
&b(a_0)\omega^2\lp|\swei{A}{s}_{\mc{I}^+}\rp|^2+b(a_0)(\omega-m\upomega_+)^2\lp|\swei{A}{s}_{\mc{H}^+}\rp|^2 + b(\delta,a_0)\int_{-\infty}^\infty \frac{\Delta}{r^2}\Big[r^{-1-\delta}|\Psi'|^2+r^{-3}(\Lambda+r^{-\delta})|\Psi|^2\Big]dr^*\\
&\quad\leq B(E,E_W)\lp[\omega^2\lp|\swei{A}{s}_{\mc{I}^-}\rp|^2+(\omega-m\upomega_+)^2\lp|\swei{A}{s}_{\mc{H}^-}\rp|^2\rp]\\
&\quad\qquad+B(E,E_W,R^*,\omega_{\rm high}, \varepsilon_{\rm width})\int_{-R^*}^{R^*}\sum_{k=0}^{|s|}\lp(|\uppsi_{(k)}'|^2+|\uppsi_{(k)}|^2\rp) dr^*\\
&\quad\qquad+\int_{-\infty}^\infty \lp\{ 2(y+\hat{y})\Re\lp[\mathfrak{G}\overline{\Psi}'\rp]+E_W\lp(Q^{W,\,T}\rp)'-E\lp(\chi_1(\omega-m\upomega_+)+\chi_2\omega\rp)\Im\lp[\mathfrak{G}\overline{\Psi}\rp]\rp\}dr^*\\
&\quad\qquad+B\sum_{k=0}^{|s|-1}\int_{-\infty}^\infty \lp\{2(y+\hat{y})\Re\lp[\mathfrak{G}_{(k)}\overline{\uppsi_{(k)}}'\rp]-E\lp(\chi_1(\omega-m\upomega_+)+\chi_2\omega\rp)\Im\lp[\mathfrak{G}_{(k)}\overline{\uppsi_{(k)}}\rp]\rp\}dr^*\,,
\end{align*}
where, if $a_0$ is sufficiently small depending on $\delta$, $\omega_{\rm high}$, $\varepsilon^{-1}_{\rm width}$, $\omega_{\rm low}$, $E$ and $E_W$, we can drop the second and fourth lines in the estimate;
\item[B.] if $\Psi$ arises from a solution of the homogeneous Teukolsky radial ODE~\eqref{eq:radial-ODE-alpha} and has the general boundary conditions in Lemma~\ref{lemma:uppsi-general-asymptotics}, we have the estimates:
\begin{equation*} 
\begin{split}
& b(a_0,\omega_{\rm high}, \varepsilon_{\rm width},\omega_{\rm low})\lp[\omega^2
\lp\{\begin{array}{lr}
\frac{\mathfrak{C}_s}{\mathfrak{D}_s^{\mc I}}\,, &s\leq 0\\
1\,, &s>0
\end{array}\rp\}
\lp|\swei{A}{s}_{\mc{I}^+}\rp|^2+(\omega-m\upomega_+)^2
\lp\{\begin{array}{lr}
1\,, &s\leq 0\\
\frac{\mathfrak{C}_s}{\mathfrak{D}_s^{\mc H}}\,, &s>0
\end{array}\rp\}
\lp|\swei{A}{s}_{\mc{I}^+}\rp|^2\rp]
\\
&\quad\leq\omega^2
\lp\{\begin{array}{lr}
1\,, &s\leq 0\\
\frac{\mathfrak{C}_s}{\mathfrak{D}_s^{\mc I}}\,, &s>0
\end{array}\rp\}
\lp|\swei{A}{s}_{\mc{I}^-}\rp|^2+(\omega-m\upomega_+)^2
\lp\{\begin{array}{lr}
\frac{\mathfrak{C}_s}{\mathfrak{D}_s^{\mc H}}\,, &s\leq 0\\
1\,, &s>0
\end{array}\rp\}
\lp|\swei{A}{s}_{\mc{I}^-}\rp|^2.
\end{split}
\end{equation*}
\end{enumerate}
\end{proposition}

\subsubsection{The case of inhomogeneous radial ODEs}

As usual, we begin with some lemmas concerning the frequency parameters and the potentials in \eqref{eq:radial-ODE-potential-k-tilde} in the relevant frequency range.

\begin{lemma}[Properties of the frequency parameters in $\mc F_{\rm int,1}$]  \label{lemma:properties-frequencies-int} Let $(\omega,m,\Lambda)\in\mc F_{\rm int,1}$ be an admissible frequency triple. Then, we have
\begin{align*}
m^2, |\Lambda| \geq B(\omega_{\rm width},\omega_{\rm high})\,.
\end{align*}
However, if $a_0$ is assumed to be sufficiently small, then for $a\in[0,a_0]$ we have 
\begin{align*}
\Lambda\geq \min\{m^2,\frac12|s|\}
\end{align*}
is non-negative.
\end{lemma}
\begin{proof}Follows from the admissibility conditions in Definition~\ref{def:admissible-freqs}.
\end{proof}

\begin{lemma}[Properties of the potential for frequencies in $\mc F_{\rm int,1}$]  \label{lemma:properties-potential-int-omega} Let $(\omega,m\Lambda)\in\mc F_{\rm int}$ be an admissible frequency triple. For each $k=0,\dots,|s|$, the potential $\mc V_{(k)}$ introduced in \eqref{eq:radial-ODE-potential-k-tilde} has the following properties:
\begin{align*}
|\Re\mc V_{(k)}|\leq B(\omega_{\rm high},\varepsilon_{\rm width})\frac{w}{r}\,,
\end{align*}
for any $r\in(r_+,\infty)$ and, for $r$ sufficiently close to $r=r_+$, there exists a constant $c=c(\omega_{\rm high},\varepsilon_{\rm width})$ such that
\begin{align*}
&\Re\hat{\mc V}_{(k)}' -c|\omega-m\upomega_+|(r-r_+)\\
&\quad\geq -B(\omega_{\rm high},\varepsilon_{\rm width})(r-r_+)\,. \numberthis\label{eq:intermediate-derivative-hatV}
\end{align*}
Furthermore, if $a\leq a_0$ with $a_0$ sufficiently small, then for sufficiently large $r$, we have
\begin{align*}
\mc V\geq \frac{\Lambda+s^2}{r^2}-\frac{a_0 B(\omega_{\rm high},\varepsilon_{\rm width})}{r^3}\,,\quad \mc V'\geq -\frac{2(\Lambda+s^2)}{r^3}-\frac{a_0 B(\omega_{\rm high},\varepsilon_{\rm width})}{r^3}\,, \numberthis\label{eq:intermediate-V-smalla}
\end{align*}
and, for $r$ sufficiently close to $r=r_+$, if $s\leq 2$ and $a_0$ is sufficiently small,
\begin{align*}
\frac{d^2}{dr^2}\hat{\mc V} geq  \frac{3\Lambda+4-s^2}{16M^4} -B(\omega_{\rm high},\varepsilon_{\rm high})(r-r_+))\,, 
\end{align*}
\end{lemma}
\begin{proof}The behavior at $r\to\infty$ follows easily by inspection of \eqref{eq:radial-ODE-potential-k-tilde}. For the $r\to r_+$ end, we can compute 
\begin{align*}
(r_+^2+a^2)^3\frac{d}{dr}\Re\hat{\mc V}_{(k)}(r_+)&= 2(r_+-r_-)\lp(Mr_+\lp(\Lambda+|s|+k(2|s|-k-1)\rp)-a^2m^2\rp)\\
&\qquad-4M am (\omega-m\upomega_+)(3r_+^2-a^2)\,.
\end{align*}
Using the bound on $\Lambda$ from Lemma~\ref{lemma:properties-frequencies-int}, we estimate the first term  by $B(\omega_{\rm high},\varepsilon_{\rm width})(r_+-r_-)$; using the bound on $|m|$, we estimate the second term similarly.   A similar procedure works for the second derivative.
\end{proof}

In proving Proposition~\ref{prop:int1}A, our strategy is to apply currents at large $|r^*|$ which generate a good bulk there at the cost of potentially introducing errors in a bounded $|r^*|$ region. In light of the previous Section~\ref{sec:bounded-smallness}, it is not hard to see that $y$ virial currents, together with energy currents to control the boundary terms those introduce, are well-suited for the task. The bounded $|r^*|$ errors will be controlled in our upcoming \cite{SRTdC2022} by making use of Theorem~\ref{thm:mode-stability-intro-AB}.

\begin{proof}[Proof of Proposition~\ref{prop:int1}A] Let $(\omega,m,\Lambda)\in\mc F_{\rm int,1}$ be an admissible frequency triple with respect to $s\in\mathbb{Z}$ and $a\in[0,a_0)$. We will introduce $y$ and $\hat{y}$ currents similar to \eqref{eq:standard-current-y0} and \eqref{eq:standard-current-y0-hat}: for some $R^*>0$ and $\delta
\in(0,1]$, we let
\begin{align*}
y:=1-\frac{(R^*)^{\delta}}{(r^*)^\delta}\mathbbm{1}_{[R^*,\infty)}\,,\qquad \hat y := -\lp(1-\frac{(-R^*)^{1/2}}{(-r^*)^{1/2}}\rp)\mathbbm{1}_{(-\infty, -R^*]}\,.
\end{align*}
In what follows, we will explain why these currents generate a good bulk in the regions $|r^*|\to \infty$. To control the boundary terms we will add localized versions of the Killing  $T$ and $K$ energy currents; we could additionally add localized Teukolsky--Starobinsky energy currents without extra effort. For brevity, we will use the notation $\omega_0:=|\omega-m\upomega_+|$.

\medskip
\noindent \textit{The large $r^*$ currents.} Let us first consider the $y$ current; it will be useful to recall Lemma~\ref{lemma:h-y-identity-k}. Choosing $R^*$ sufficiently large depending on $\omega_{\rm high}$, $\varepsilon_{\rm width}$ and $\omega_{\rm low}$, we have 
\begin{align*}
\int_{-\infty}^\infty\lp\{y'|\uppsi_{(k)}'|^2+\lp[y'\omega^2-(y\Re\mc V_{(k)})'\rp]|\uppsi_{(k)}|^2\rp\}dr^* \geq \int_{-\infty}^\infty\lp\{y'|\uppsi_{(k)}'|^2+\frac12y'\omega^2|\uppsi_{(k)}|^2\rp\}dr^*
\end{align*}
since $\omega^2\geq \omega_{\rm low}^2>0$ and $(y\Re \mc V_{(k)})'$ is lower order as $r\to \infty$.  Making $R^*$ even larger depending also on $\delta$, the error terms 
\begin{align*}
&\int_{-\infty}^\infty 2y(|s|-k)w'\lp\{w|\uppsi_{(k+1)}|^2-\lp(\omega-\frac{am}{r^2+a^2}\rp)\Re\lp[\uppsi_{(k+1)}\overline{\uppsi_{(k)}}\rp]\rp\}dr^* \\
&\qquad+ \int_{-\infty}^\infty2y(|s|-k)\frac{4amrw}{(r^2+a^2)}\lp\{ w\Im \lp[\uppsi_{(k+1)}\overline{\uppsi}_{(k)}\rp]+ \lp(\omega-\frac{am}{r^2+a^2}\rp)|\uppsi_{(k)}|^2\rp\}dr^*\\
&\quad \leq \int_{-\infty}^\infty \frac{B(\omega_{\rm high},\varepsilon_{\rm width},\omega_{\rm low},\delta)}{R^*}y'\omega^2\lp(|\uppsi_{(k)}|^2+(|s|-k)|\uppsi_{(k+1)}|^2\rp)dr^*\,,
\end{align*}
are sufficiently small; so for the coupling terms we simply apply Cauchy--Schwarz:
\begin{align*}
\sum_{j=0}^{k-1}2yw\Re\lp[\lp(c_{s,k,j}^{\rm id}+imc_{s,k,j}^{\rm \Phi}\rp)\uppsi_{(k)}'\uppsi_{(j)}\rp] \leq \frac12 y|\uppsi_{(k)}'|^2+ B(\omega_{\rm high},\varepsilon_{\rm width},\omega_{\rm low},\delta)\sum_{j=0}^{k-1} y'\omega^2|\uppsi_{j}|^2\,.
\end{align*}
Combining the previous, if $R^*$ is sufficiently large we have 
\begin{align*}
&\frac12\int_{-\infty}^\infty\lp\{y'|\uppsi_{(k)}'|^2+y'\omega^2|\uppsi_{(k)}|^2\rp\}dr^*\\
&\quad\leq  2\omega^2\lp(|A_{k,\mc I^+}|^2+|A_{k,\mc I^-}|^2\rp) +\int_{-\infty}^\infty 2y\Re[\mathfrak{G}_{(k)}\overline{\uppsi_{(k)}}']\\
&\quad\qquad+B(\omega_{\rm high},\varepsilon_{\rm width},\omega_{\rm low})\int_{-\infty}^\infty y'\omega^2\lp[\frac{(|s|-k)}{R^*}|\uppsi_{(k+1)}|^2+\sum_{j=0}^k |\uppsi_{(j)}|^2\rp]dr^* \numberthis\label{eq:int-intermediate-1}
 \,.
\end{align*}

To control the boundary terms in the estimate above, we introduce a localized Killing $T$ energy current: for $\chi_2$ a smooth cutoff which is equal to $1$ for $r^*\in[R^*,\infty)$ and vanishes for $r^*\leq R^*-1$, we add $E\chi_2Q^T$ with $E\geq 3$. For the errors involving $\uppsi_{(k+1)}$, we again obtain
\begin{align*}
&(|s|-k)\omega \lp[w'\Im[\uppsi_{(k+1)}\overline{\uppsi_{(k)}}]+\frac{4amrw}{r^2+a^2}|\uppsi_{(k)}|^2\rp] \\
&\quad\leq \frac{B(\omega_{\rm high},\varepsilon_{\rm width},\omega_{\rm low},\delta)}{R^*}y'\omega^2\lp(|\uppsi_{(k)}|^2  + (|s|-k)|\uppsi_{(k+1)}|^2\rp),
\end{align*}
where, choosing $R^*$ appropriately, these error terms come with a smallness parameter; and for the coupling errors
\begin{align*}
\sum_{j=0}^{k-1}w\omega\Re\lp[\lp(c_{s,k,j}^{\rm id}+imc_{s,k,j}^{\rm \Phi}\rp)\uppsi_{(k)}\uppsi_{(j)}\rp] \leq \frac14 y'\omega^2|\uppsi_{(k)}'|^2+ B(\omega_{\rm high},\varepsilon_{\rm width},\omega_{\rm low},\delta)\sum_{j=0}^{k-1} y'\omega^2|\uppsi_{j}|^2\,.
\end{align*}
Combining the previous with \eqref{eq:int-intermediate-1}, if  $R^*$ is sufficiently large we have 
\begin{align*}
&\frac14\int_{-\infty}^\infty\lp\{y'|\uppsi_{(k)}'|^2+y'\omega^2|\uppsi_{(k)}|^2\rp\}dr^*+\omega^2|A_{k,\mc I^+}|^2\\
&\quad\leq  (2+E)\omega^2|A_{k,\mc I^-}|^2+B(\omega_{\rm high},\varepsilon_{\rm width},R^*)\int_{-R^*}^{R^*}\lp(|\uppsi_{(k)}'|^2+|\uppsi_{(k)}'|^2\rp)dr^* +\int_{-\infty}^\infty 2y\Re[\mathfrak{G}_{(k)}\overline{\uppsi_{(k)}}']\\
&\quad\qquad+B(\omega_{\rm high},\varepsilon_{\rm width},\omega_{\rm low},\delta)\int_{-\infty}^\infty y'\omega^2\lp[\frac{(|s|-k)}{R^*}|\uppsi_{(k+1)}|^2+\sum_{j=0}^k |\uppsi_{(j)}|^2\rp]dr^*\\
&\quad\leq  B\omega^2|A_{k,\mc I^-}|^2 +B\sum_{j=0}^k\int_{-\infty}^\infty y\Re[\mathfrak{G}_{(j)}\overline{\uppsi_{(j)}}']+B(|s|-k)\int_{-\infty}^\infty y'\omega^2|\uppsi_{(k+1)}|^2dr^*\\
&\quad\qquad +B(\omega_{\rm high},\varepsilon_{\rm width},R^*)\sum_{j=0}^k\int_{-R^*}^{R^*}\lp(|\uppsi_{(j)}'|^2+|\uppsi_{(j)}'|^2\rp)dr^*\,,
\end{align*}
after iterating in $k=0,\dots, |s|$.

\medskip
\noindent \textit{The very negative $r^*$ currents.} We first show that
\begin{align*}
\int_{-\infty}^\infty \lp[\hat y'|\uppsi_{(k)}'|^2+(\hat y'\omega_0^2 -\hat y|r^*|^{-4})|\uppsi_{(k)}|^2\rp]dr^*&\leq 2 \int_{-\infty}^\infty\lp[\hat y'|\uppsi_{(k)}'|+\lp(\hat y'\omega_0^2-\hat y\hat{\mc V}_{(k)}' -\hat y'\hat{\mc V}_{(k)} \rp)|\uppsi_{(k)}|^2\rp]dr^* \\
&\qquad +B(\omega_{\rm high},\varepsilon_{\rm width},R^*)\int_{-R^*}^{R^*}|\uppsi_{(k)}|^2dr^* \numberthis\label{eq:int-intermediate-2}
\end{align*}
for sufficiently large $R^*$ depending on $\omega_{\rm high}$, $\varepsilon_{\rm width}$, $\omega_{\rm low}$ and $a_0$. Since the last term on the left hand side is lower order as $r^*\to -\infty$, the only term which might have the wrong sign is $(-\hat y)\hat{\mc V}_{(k)}'$. However, recall \eqref{eq:intermediate-derivative-hatV} from Lemma~\ref{lemma:properties-potential-int-omega}: assuming $a\leq a_0$ and that $R^*$ is sufficiently large depending on $\omega_{\rm high}$, $\varepsilon_{\rm width}$ and $a_0$, we have
\begin{align*}
(-\hat y)c(\omega_{\rm high},\varepsilon_{\rm width})\omega_0(r-r_+)\leq \frac12 \hat y' \omega_0^2+B(\omega_{\rm high},\varepsilon_{\rm width})\varepsilon^{-1}\hat y^2|r^*|^{3/2}(r-r_+)^2 \leq \frac12 \hat y' \omega_0^2+(-\hat y)|r^*|^{-4}
\end{align*}
 and, by a Hardy inequality,
\begin{align*}
&\int_{-\infty}^\infty B(\omega_{\rm high},\varepsilon_{\rm width}) (-\hat y)(r-r_+)|\uppsi_{(k)}|^2dr^* \\
&\quad\leq\int_{-\infty}^\infty B(\omega_{\rm high},\varepsilon_{\rm width}) (-\hat y)|r^*|^{-4}|\uppsi_{(k)}|^2dr^* \\
&\quad \leq \int_{-\infty}^{-R^*}\frac{B(a_0,\omega_{\rm high},\varepsilon_{\rm width})}{\sqrt{R^*}}\hat y '|\uppsi_{(k)}'|^2dr^* +B(\omega_{\rm high},\varepsilon_{\rm width},R^*)\int_{-R^*}^{R^*}|\uppsi_{(k)}|^2dr^*\\
&\quad \leq \int_{-\infty}^{-R^*}\frac12\hat y '|\uppsi_{(k)}'|^2dr^* +B(\omega_{\rm high},\varepsilon_{\rm width},R^*)\int_{-R^*}^{R^*}|\uppsi_{(k)}|^2dr^*\,.
\end{align*} Thus, we conclude
\begin{align*}
\int_{-\infty}^\infty\hat y\hat{\mc V}_{(k)}'|\uppsi_{(k)}|^2dr^*\leq \int_{-\infty}^{-R^*}\frac{1}{2}\hat y'\lp(|\uppsi_{(k)}'|^2+\omega_0^2|\uppsi_{(k)}|^2\rp)dr^* +B(\omega_{\rm high},\varepsilon_{\rm width},R^*)\int_{-R^*}^{R^*}|\uppsi_{(k)}|^2dr^*\,,
\end{align*}
concluding the proof of \eqref{eq:int-intermediate-2}. For the terms in \eqref{eq:y-identity-k} involving $\uppsi_{(k+1)}$, we have
\begin{align*}
&\int_{-\infty}^\infty 2\hat y(|s|-k)w'\lp\{w|\uppsi_{(k+1)}|^2-\lp(\omega-\frac{am}{r^2+a^2}\rp)\Re\lp[\uppsi_{(k+1)}\overline{\uppsi_{(k)}}\rp]\rp\}dr^* \\
&\qquad+ \int_{-\infty}^\infty2\hat y(|s|-k)\frac{4amrw}{(r^2+a^2)}\lp\{ w\Im \lp[\uppsi_{(k+1)}\overline{\uppsi}_{(k)}\rp]+ \lp(\omega-\frac{am}{r^2+a^2}\rp)|\uppsi_{(k)}|^2\rp\}dr^*\\
&\quad \leq \int_{-\infty}^\infty \lp(\frac{B(\omega_{\rm high},\varepsilon_{\rm width},a_0)}{R^*}(-\hat y)|r^*|^{-4}|\uppsi_{(k+1)}|^2+\frac12 (\hat y'\omega_0^2+(-\hat y)|r^*|^{-4})|\uppsi_{(k)}|^2\rp)dr^*\,,
\end{align*}
and for the coupling terms 
\begin{align*}
\sum_{j=0}^{k-1}2a \hat yw\Re\lp[\lp(c_{s,k,j}^{\rm id}+imc_{s,k,j}^{\rm \Phi}\rp)\uppsi_{(k)}'\uppsi_{(j)}\rp] \leq \frac12 \hat y'|\uppsi_{(k)}'|^2+a^2\sum_{j=0}^{k-1}\hat y |r^*|^{-4} |\uppsi_{j}|^2\,,
\end{align*}
using the largeness of $R^*$. We proceed similarly with the errors due to localized, by a cutoff $\chi_2(r^*)=\chi(-r^*)$, Killing $K$ energy current, which we add a multiple $E\geq 3$ of to control the  boundary terms introduced by the $\hat y$ current. Putting all this together, we obtain
\begin{align*}
&\frac14 \int_{-\infty}^\infty \lp[\hat y'|\uppsi_{(k)}'|^2+(\hat y'\omega_0^2 -\hat y|r^*|^{-4})|\uppsi_{(k)}|^2\rp]dr^* + \omega_0^2|A_{k,\mc H^+}|^2\\
&\quad \leq 5\omega_0^2|A_{k,\mc H^-}|^2  +B(\omega_{\rm high},\varepsilon_{\rm width},R^*)\int_{-R^*}^{R^*}\lp(|\uppsi_{(k)}'|^2+|\uppsi_{(k)}'|^2\rp)dr^*  +\int_{-\infty}^\infty 2y\Re[\mathfrak{G}_{(k)}\overline{\uppsi_{(k)}}']\\
&\quad\qquad+\int_{-\infty}^\infty (-\hat y)|r^*|^{-4}|\lp[\frac{B(\omega_{\rm high},\varepsilon_{\rm width},a_0)(|s|-k)}{R^*}|\uppsi_{(k+1)}|^2+\sum_{j=0}^k |\uppsi_{(j)}|^2\rp]dr^*\\
&\quad \leq B(a_0)E\sum_{j=0}^k\omega_0^2|A_{j,\mc H^-}|^2  +B(a_0,\omega_{\rm high},\varepsilon_{\rm width},R^*)\sum_{j=0}^k\int_{-R^*}^{R^*}\lp(|\uppsi_{(j)}'|^2+|\uppsi_{(j)}'|^2\rp)dr^*  \\
&\quad\qquad +B(a_0)\sum_{j=0}^k\int_{-\infty}^\infty 2\hat y\Re[\mathfrak{G}_{(k)}\overline{\uppsi_{(k)}}']\,.
\end{align*}
In the last line we have iterated in $k=0,\dots, |s|$. 

\medskip
\noindent \textit{The case of very small $a_0$.} We define a $y$ current similarly to the previous:
\begin{align*}
y(r^*)&:=\exp\lp(-C\int_{r(r^*)}^\infty \frac{dr}{\Delta} \rp)\mathbbm{1}_{(-\infty,R^*]}+y(R^*)\lp(2-\frac{(R^*)^{\delta}}{(r^*)^{\delta}}\rp)\mathbbm{1}_{(R^*,\infty)}\,,\\
y'(r^*)&:=C\frac{y}{r^2+a^2}\mathbbm{1}_{(-\infty,R^*)}+\delta\frac{y(R^*)(R^*)^{\delta}}{(r^*)^{1+\delta}}>0\,, 
\end{align*}
Choosing $R^*$ sufficiently large depending on $\omega_{\rm high}$, $\omega_{\rm low}$, $\varepsilon_{\rm width}$, we have 
\begin{align*}
\lp[y'\omega^2-(y\mc V)'\rp]\mathbbm{1}_{[R^*+1,\infty)} \geq \lp[\frac12y'\omega^2-y\frac{\Lambda}{r^3}\rp]\mathbbm{1}_{[R^*+1,\infty)}.
\end{align*}
Moreover, for $\chi$ a cutoff which is identically 1 for $r^*\leq R^*-1$ and equal to zero for $r^*\geq R^*$, we have 
\begin{align*}
\int_{-\infty}^{\infty}(y\mc V)'|\Psi|\mathbbm{1}_{(-\infty,R^*-1]}dr^* &\leq  \int_{-\infty}^{\infty}(y\chi\mc V)'|\Psi|^2\leq \lp|\int_{-\infty}^{\infty} 2y\mc V\Re[\Psi'\overline{\Psi}]\mathbbm{1}_{(-\infty,R^*]}dr^*\rp|\\
&\leq \int_{R^*}^{\infty} \lp[\frac12 y'|\Psi'|^2+ y'\omega^2\lp(\frac{V(r^2+a^2)}{C\omega}\rp)^2|\Psi|^2\rp]dr^*\\
&\leq \int_{R^*}^{\infty} \lp[\frac12 y'|\Psi'|^2+\frac12 y'\omega^2|\Psi|^2\rp]dr^*
\end{align*}
choosing $C$ to be sufficiently large depending on $\omega_{\rm high}$, $\omega_{\rm low}$, $\varepsilon_{\rm width}$, and then fixing it.

To strengthen the weights near $r=r_+$, we add a $\hat y$ current similar to the previous case:
\begin{align*}
\hat y := -\lp(1-\frac{(-R^*)^{1/2}}{(-r^*)^{1/2}}\rp)\mathbbm{1}_{(-\infty, -R^*]}\,,
\end{align*}
noting that \eqref{eq:int-intermediate-2} still applies. For the coupling terms, we have
\begin{align*}
2a(y+\hat y)w\Re\lp[(c_{s,k}^{\rm id}+imc_{s,k}^{\Phi})\uppsi_{(k)}\overline{\Psi}'\rp] \leq \frac14(\hat y'+y')|\Psi'|^2+a^2B(\omega_{\rm high},\varepsilon_{\rm width})w|\uppsi_{(k)}|^2\,,
\end{align*}
if $R^*$ is sufficiently large depending on $\delta$ that $\delta R^*\geq b$. We then apply \eqref{eq:basic-estimate-1} with $c=(r-r_+)/r$ if $s>0$ and $c=-1/r$ if $s<0$ to obtain
\begin{align*}
\int_{-\infty}^\infty w|\uppsi_{(k)}|^2dr^* \leq \int_{-\infty}^\infty w|\Psi|^2dr^* \,.
\end{align*}
To summarize, we see that if $a_0$ is sufficiently small and $R^*$ is sufficiently large, both depending on  $\omega_{\rm high}$, $\omega_{\rm low}$ and $\varepsilon_{\rm width}$, we have 
\begin{align*}
&\int_{-\infty}^\infty \frac14\lp[y'|\Psi'|^2+(y'\omega^2+\hat y'\omega_0^2-\hat{y}|r^*|^{-4})|\Psi|^2\rp]dr^* \\
&\quad\leq 2y(\infty)\omega^2\lp(|A_{\mc I^+}|^2+|A_{\mc I^-}|^2\rp)+2|\hat y(-\infty)|\omega^2_0\lp(|A_{\mc H^+}|^2+|A_{\mc H^-}|^2\rp)+\int_{-\infty}^\infty 2(y+\hat y) \Re[\mathfrak G \overline{\Psi'}]dr^*\\
&\quad\qquad +a^2B(\omega_{\rm high},\varepsilon_{\rm width})\int_{-\infty}^{\infty}w|\Psi|^2dr^*\,.
\end{align*}
By adding a Killing $K$ and $T$ current localized, respectively, by cutoffs $\chi_1$ and $\chi_2=1-\chi_1$ and multiplied by $E\geq 3$, we have
\begin{align*}
&\int_{-\infty}^\infty \frac14\lp[y'|\Psi'|^2+(y'\omega^2+\hat y'\omega_0^2-\hat{y}|r^*|^{-4})|\Psi|^2\rp]dr^*  +\omega^2|A_{\mc I^+}|^2+\omega^2_0|A_{\mc H^+}|^2\\
&\quad\leq (E+2)\omega^2|A_{\mc I^-}|^2+(E+2)\omega_0^2|A_{\mc H^-}|^2+\int_{-\infty}^\infty \lp(2y \Re[\mathfrak G \overline{\Psi'}]-E(\chi_1\omega_0^2+\chi_2\omega) \Im[\mathfrak G\overline{\Psi}]\rp)dr^*\\
&\quad\qquad +a^2B(\omega_{\rm high},\varepsilon_{\rm width})(E+1)\int_{-\infty}^{\infty}w |\Psi|^2dr^* +\frac{am}{2Mr_+}\int_{-\infty}^\infty E\chi_1'\Im[\Psi' \overline{\Psi}]dr^*\,, \numberthis\label{eq:int-smalla-intermediate-1}
\end{align*}
where we have addressed the coupling terms of the Killing current in a similar fashion to those of the $y$ current. Using the smallness of $a_0$, we can absorb the last line in \eqref{eq:int-smalla-intermediate-1} into the left hand side, thus concluding the proof.
\end{proof}

\subsubsection{The case of homogeneous radial ODEs}

We begin by introducing some notation. Recall the functions $u^{[s],\, (a\omega)}_{ml,\,\mc{I}^\pm}$ and $u^{[s],\, (a\omega)}_{ml,\,\mc{H}^\pm}$ defined in Definition~\ref{def:uhor-uout}. For simplicity, we will now drop the sub and superscripts. 

\begin{lemma} Fix $M>0$ and $s\in\frac12\mathbb{Z}$. For any $|a|\leq M$ and any $(\omega,m,\Lambda)$ admissible with respect to $a$ and $s$ satisfying $\omega\neq 0,m\upomega_+$, the quantities
\begin{equation} \label{eq:five-Wronskians}
\begin{gathered}
\swei{\mathfrak{W}}{s}:= W\lp[\swei{u}{s}_{\mc{I}^+},\swei{u}{s}_{\mc{H}^+}\rp]\,,\\
\swei{\mathfrak{W}}{s}_{\mc H}:= W\lp[\swei{u}{s}_{\mc{H}^+},\swei{u}{s}_{\mc{H}^-}\rp]\,,\qquad
\swei{\mathfrak{W}}{s}_{\mc I}:= W\lp[\swei{u}{s}_{\mc{I}^+},\swei{u}{s}_{\mc{I}^-}\rp]\,,\\
\swei{\mathfrak{W}}{s}_{1}:= W\lp[\swei{u}{s}_{\mc{H}^+},\swei{u}{s}_{\mc{I}^-}\rp]\,,\qquad
\swei{\mathfrak{W}}{s}_{2}:= W\lp[\swei{u}{s}_{\mc{I}^+},\swei{u}{s}_{\mc{H}^-}\rp]\,,
\end{gathered}
\end{equation} 
 where $W[y_1,y_2]=y_1'y_2-y_1y_2'$ represents the Wronskian of functions $y_1(r^*)$ and $y_2(r^*)$ and where we have suppressed some of the sub and superscripts in the functions $u^{[s],\, (a\omega)}_{ml,\,\mc{I}^\pm}$ and $u^{[s],\, (a\omega)}_{ml,\,\mc{H}^\pm}$ defined in Definition~\ref{def:uhor-uout}, depend only on the spin $s$, the Kerr parameters $(a,M)$ and the frequency triple $(\omega,m,\Lambda)$.
\end{lemma}
\begin{proof} The Wronskian of any combination of two of the functions $\swei{u}{s}_{\mc{I}^\pm}$ and $\swei{u}{s}_{\mc{H}^\pm}$ does not depend on $r^*$ by the fact that these functions solve a radial ODE, \eqref{eq:radial-ODE-u}, which involves no first order derivatives. 

By computations as $r\to \infty$ or $r\to r_+$, depending on what is simpler, we can even explicitly obtain
\begin{equation}\label{eq:Wronskian-H-I}
\swei{\mathfrak{W}}{s}_{\mc I}= 2i\omega\,, \qquad 
 \swei{\mathfrak{W}}{s}_{\mc H}= -\lp\{\begin{array}{lr}
 \lp[4Mr_+i(\omega-m\upomega_+)+s(r_+-r_-)\rp](r_+-r_-)^{s}\,, &|a|<M\\
 4M^2i(\omega-m\upomega_+)\,, &|a|=M
 \end{array}\rp\}\,,
\end{equation}
which are manifestly independent of $r^*$.
\end{proof}

We now state a slight refinement of the real mode stability statement for subextremal $|a|<M$ Kerr black holes in Theorem~\ref{thm:mode-stability-intro-AB}, which will appear in \cite{TdC-thesis}: 

\begin{proposition}[Subextremal real mode stability, refined version] \label{prop:subextremal-mode-stability-refined} Fix $M>0$, $s\in\mathbb{Z}_{\geq 0}$ and and $a_0\in[0,M)$. For any $|a|\leq a_0$ and any $(\omega,m,\Lambda)$ admissible with respect to $a$ and $s$ such that $(\omega,m,\Lambda)\in\mc{A}$ satisfying
\begin{align*}
C_{\mc{A}}:=\sup_{\mc{A}}\lp(|\omega|+|\omega|^{-1}+|m|+|\Lambda|\rp)<\infty\,,
\end{align*}
we have the bounds
\begin{align*}
\lp|\swei{\mathfrak{W}}{\pm s}\rp|^{-2}&\leq B(a,M,|s|,C_{\mc{A}})\,;\qquad
(\omega-m\upomega_+)^{-2}\lp|\swei{\mathfrak{W}}{+s}\rp|^{-2}&\leq B(a,M,|s|,C_{\mc{A}})\text{~~if~} s\geq 1\,.
\end{align*}
\end{proposition}

Proposition~\ref{prop:int1}B follows immediately:

\begin{proof}[Proof of Proposition~\ref{prop:int1}B]
Let $s\in\mathbb{Z}$ be any integer. Recall the identity
\begin{equation*}
\begin{split}
\swei{u}{s}&= \swei{\tilde a}{s}_{\mc{H}^+}\cdot \swei{u}{s}_{\mc{H}^+}+\swei{\tilde{a}}{s}_{\mc{H}^-}\cdot \swei{u}{s}_{\mc{H}^-}= \swei{\tilde{a}}{s}_{\mc{I}^+}\cdot \swei{u}{s}_{\mc{I}^+}+\swei{\tilde{a}}{s}_{\mc{I}^-}\cdot \swei{u}{s}_{\mc{I}^-}
\end{split}\,,
\end{equation*}
for some coefficients $\swei{\tilde{a}}{s}_{\mc{I}^\pm}$ and $\swei{\tilde{a}}{s}_{\mc{H}^\pm}$ from Lemma~\ref{lemma:apha-general-asymptotics}; by taking its Wronskian with all four of $\swei{u}{s}_{\mc{H}^\pm}$ and $\swei{u}{s}_{\mc{I}^\pm}$, if $\mathfrak{W}\neq 0$ we obtain 
\begin{align*}
\lp(\begin{array}{l}
\swei{\tilde a}{s}_{\mc{I}^+}\\
\swei{\tilde a}{s}_{\mc{H}^+}
\end{array}\rp)
= \frac{1}{\mathfrak{W}} \bm{Q}
\lp(\begin{array}{l}
\swei{\tilde a}{s}_{\mc{I}^-}\\
\swei{\tilde a}{s}_{\mc{H}^-}
\end{array}\rp)\,, \qquad \bm{Q}:=\lp(\begin{array}{ll}
\swei{\mathfrak{W}}{s}_{1} & -\swei{\mathfrak{W}}{s}_{\mc H}\\
\swei{\mathfrak{W}}{s}_{\mc I} & -\swei{\mathfrak{W}}{s}_{2}
\end{array}\rp) \,.
\end{align*}
By further imposing a conservation law taking the form
\begin{align}
\omega^2\mathfrak{C}_s^{(11)}\lp|\swei{\tilde a}{s}_{\mc{I}^+}\rp|^2+\omega(\omega-m\upomega_+)\mathfrak{C}_s^{(12)}\lp|\swei{\tilde a}{s}_{\mc{H}^+}\rp|^2= \omega^2\mathfrak{C}_s^{(13)}\lp|\swei{\tilde a}{s}_{\mc{I}^-}\rp|^2 +\omega(\omega-m\upomega_+)\mathfrak{C}_s^{(14)}\lp|\swei{\tilde a}{s}_{\mc{H}^-}\rp|^2\,, \label{eq:general-conservation}
\end{align}
we deduce
\begin{gather*}
\begin{dcases}
0=\omega^2\mathfrak{C}_s^{(11)}\bm{Q}_{11}\overline{\bm{Q}}_{12}+\omega(\omega-m\upomega_+)\mathfrak{C}_s^{(12)}\bm{Q}_{21}\overline{\bm{Q}}_{22}\\
\omega^2\mathfrak{C}_s^{(11)}\lp|\bm{Q}_{11}\rp|^2+\omega(\omega-m\upomega_+)\mathfrak{C}_s^{(12)}\lp|\bm{Q}_{21}\rp|^2= \omega^2\mathfrak{C}_s^{(13)}\lp|\swei{\mathfrak{W}}{s}\rp|^2\\
\omega^2\mathfrak{C}_s^{(11)}\lp|\bm{Q}_{12}\rp|^2+\omega(\omega-m\upomega_+)\mathfrak{C}_s^{(12)}\lp|\bm{Q}_{22}\rp|^2= \omega(\omega-m\upomega_+)\mathfrak{C}_s^{(14)}\lp|\swei{\mathfrak{W}}{s}\rp|^2\\
\end{dcases}\\
\Leftrightarrow
\begin{dcases}
\swei{\mathfrak{W}}{s}_1= -\frac{\omega-m\upomega_+}{\omega}\frac{\mathfrak{C}_s^{(12)}}{\mathfrak{C}_s^{(11)}}\lp(\swei{\mathfrak{W}}{s}_{\mc I}\rp)\lp(\overline{\swei{\mathfrak{W}}{s}_{\mc H}}\rp)^{-1} \overline{\swei{\mathfrak{W}}{s}_2}\\
\lp|\swei{\mathfrak{W}}{s}_1\rp|^2 =\frac{\mathfrak{C}_s^{(13)}}{\mathfrak{C}_s^{(11)}}\lp|\swei{\mathfrak{W}}{s}\rp|^2 -\frac{\omega-m\upomega_+}{\omega}\frac{\mathfrak{C}_s^{(12)}}{\mathfrak{C}_s^{(11)}}\lp|\swei{\mathfrak{W}}{s}_{\mc I}\rp|^2\\
\lp|\swei{\mathfrak{W}}{s}_2\rp|^2 =\frac{\mathfrak{C}_s^{(14)}}{\mathfrak{C}_s^{(12)}}\lp|\swei{\mathfrak{W}}{s}\rp|^2 -\frac{\omega}{\omega-m\upomega_+}\frac{\mathfrak{C}_s^{(11)}}{\mathfrak{C}_s^{(12)}}\lp|\swei{\mathfrak{W}}{s}_{\mc H}\rp|^2
\end{dcases}\,,
\end{gather*}
and, finally,
\begin{align*}
&\omega^2\mathfrak{C}_{s}^{(11)}\lp|\swei{\tilde a}{s}_{\mc{I}^+}\rp|^2 +(\omega-m\upomega_+)^2\mathfrak{C}_{s}^{(12)}\lp|\swei{\tilde a}{s}_{\mc{H}^+}\rp|^2\\
&\quad \leq \lp[1+\frac{m\upomega_+(\omega-m\upomega_+)}{\omega^2}\frac{\lp|\swei{\mathfrak{W}}{s}_{\mc I}\rp|^2}{\lp|\swei{\mathfrak{W}}{s}\rp|^2}\frac{\mathfrak{C}_{s}^{(12)}}{\mathfrak{C}_{s}^{(13)}}\rp]\omega^2\mathfrak{C}_{s}^{(13)}\lp|\swei{\tilde a}{s}_{\mc{I}^-}\rp|^2\\
&\quad\qquad +\lp[1-\frac{m\upomega_+\omega}{(\omega-m\upomega_+)^2}\frac{\lp|\swei{\mathfrak{W}}{s}_{\mc H}\rp|^2}{\lp|\swei{\mathfrak{W}}{s}\rp|^2}\frac{\mathfrak{C}_{s}^{(11)}}{\mathfrak{C}_{s}^{(14)}}\rp](\omega-m\upomega_+)^2\mathfrak{C}_{s}^{(14)}\lp|\swei{\tilde a}{s}_{\mc{H}^-}\rp|^2\,.\numberthis\label{eq:scattering-int-intermediate-1}
\end{align*}

The coefficients $\swei{\tilde{a}}{s}_{\mc{I}^\pm}$ and $\swei{\tilde{a}}{s}_{\mc{H}^\pm}$ for the homogeneous Teukolsky radial ODE~\eqref{eq:radial-ODE-u} do indeed satisfy a conservation law of the type \eqref{eq:general-conservation}, given in Lemma~\ref{lemma:wronskian-energy-currents}. The factors $\mathfrak{C}_s^{(11)}$ through $\mathfrak{C}_s^{(14)}$ can be read off directly from the formulas given there;  the relation between coefficients $\swei{\tilde{a}}{s}_{\mc{I}^\pm}$, $\swei{\tilde{a}}{s}_{\mc{H}^\pm}$ and coefficients $\swei{A}{s}_{\mc{I}^\pm}$ and $\swei{A}{s}_{\mc{H}^\pm}$, which follows from an asymptotic analysis, is given in Lemma~\ref{lemma:uppsi-general-asymptotics}. We also note the explicit computations \eqref{eq:Wronskian-H-I}. After using all the simplifications we have described, one easily bounds \eqref{eq:scattering-int-intermediate-1} using Proposition~\ref{prop:subextremal-mode-stability-refined} and the definition of $\mc{A}$.

To exemplify, let us consider the case of strictly positive spin $s$, which is arguably the most complicated situation. We have, for $|a|<M$,
\begin{align*}
&\omega^2\lp|\swei{A}{s}_{\mc{I}^+}\rp|^2 +(\omega-m\upomega_+)^2\frac{\mathfrak{C}_s}{\mathfrak{D}_s}\lp|\swei{A}{s}_{\mc{H}^+}\rp|^2\\ 
&\quad\leq \lp[1+\frac{(Mr_+)^{-1}m\upomega_+(\omega-m\upomega_+)(2\omega)^{2s}\cdot (\omega-m\upomega_+)^{-2}\lp|\swei{\mathfrak{W}}{s}\rp|^{-2}}{\displaystyle\prod_{j=1}^{|s|-1}\lp\{[4Mr_+(\omega-m\upomega_+)]^2+(s-j)^2(r_+-r_-)^2\rp\}}\rp]\omega^2\frac{\mathfrak{C}_s}{\mathfrak{D}_s^{\mc I}}\lp|\swei{A}{s}_{\mc{I}^-}\rp|^2\\
&\quad\qquad +\lp[1-\frac{m\upomega_+\omega(2\omega)^{2s}\cdot (\omega-m\upomega_+)^{-2}\lp|\swei{\mathfrak{W}}{s}\rp|^{-2}}{\displaystyle\prod_{j=1}^{|s|-1}\lp\{[4Mr_+(\omega-m\upomega_+)]^2+(s-j)^2(r_+-r_-)^2\rp\} }\rp](\omega-m\upomega_+)^2\lp|\swei{A}{s}_{\mc{H}^-}\rp|^2\,,
\end{align*}
which can be bounded from above using the refined bound for $|s|\geq 1$ in Proposition~\ref{prop:subextremal-mode-stability-refined}.  
\end{proof}


\subsection{The unbounded frequencies}
\label{sec:unbounded}

In this section, we will establish Theorems~\ref{thm:ode-estimates-Psi-A} and~\ref{thm:ode-estimates-Psi-B}  for frequency triples $(\omega,m,\Lambda)$ where at least one entry is unbounded. We begin, in Section~\ref{sec:unbounded-properties-potential} by stating some important properties of the potential for the transformed radial ODE~\eqref{eq:radial-ODE-Psi} in these regimes. We provide a brief description of the strategy of our proof in \ref{sec:unbounded-overview}, which we further expand upon at the beginning of the remaining sections.

\subsubsection{Properties of the potential}
\label{sec:unbounded-properties-potential}

We begin by discussing the critical point structure of the potential for the radial ODE for $\Psi$. For large frequency parameters, the behavior of the potential will be driven by $\mc{V}_0$.

\begin{lemma}[Critical points of $\mc{V}_0$]\label{lemma:critical-points-V0}
Fix $s\in\mathbb{Z}$, $M>0$ and $a_0\in[0, M)$. Assume $0\leq a\leq M$. Then, for all admissible frequency triples $(\omega,m,\Lambda)$ such that $\Lambda\geq \frac23 m^2$, we have:
\begin{enumerate}[label=(\roman*)]
\item As a function of $r\in(r_+,\infty)$, the potential $\mc{V}_0$, is either (a) strictly decreasing, (b) has a unique critical value $r^0_{\rm max}$ which is a global maximum or (c) has exactly two critical values $r^0_{\rm min}<r^0_{\rm max}$ which are a local minimum and maximum respectively. 
\item If $r_{\rm max}^0$ exists and, for some small $c>0$, $\Lambda\geq c^{-1}$ and $\Lambda\geq c|am\omega|$, then $r^0_{\rm max}\leq B(c)$ if $m\omega<0$ and $r^0_{\rm max}\leq B$ independently of $c$ if $m\omega>0$. 
\item When  $a\in[0,a_0]$, there is a sufficiently small $c=c(a_0,M)$ such that, if $\Lambda\geq c^{-1}$, $\Lambda\geq c|am\omega|$ and $m^2\leq c\Lambda$, we have $r_{\rm max}^0-(1+\sqrt{2})M>b(c)$. 
\item If it exists, $r^0_{\rm min}$ satisfies $\mc{V}_0(r_{\rm min}^0) >\omega^2$.
\end{enumerate}
\end{lemma}
\begin{proof} This is a simple adaptation of \cite{Dafermos2016b}. We begin by identifying the critical point structure of $\mc{V}_0$. Define
$$\sigma=\frac{am\omega}{\Lambda+s^2}\,.$$
It is easy to check that
\begin{align}
(r^2+a^2)^3\frac{d\mc{V}_0}{dr} &=-2(\Lambda+s^2)\lp[
r^3 -3Mr^2(1-2\sigma)+a^2r\lp(1-\frac{2m^2}{\Lambda+s^2}\rp)+Ma^2(1-2\sigma)\rp]\,, \label{eq:dV0/dr} \\
\frac{d}{dr}\lp[(r^2+a^2)^3\frac{d\mc{V}_0}{dr}\rp]&=-12(\Lambda+s^2) \lp[\frac{1}{2}r^2-M(1-2\sigma) r +\frac{1}{6}a^2\lp(1-\frac{2m^2}{\Lambda+s^2}\rp)\rp]\,, \label{eq:ddV0/ddr}
\end{align}
and the latter has roots
\begin{align*}
r_{1,\,2}=M(1-2\sigma)\pm \sqrt{M^2(1-2\sigma)-\frac{a^2}{3}\lp(1-\frac{2m^2}{\Lambda+s^2}\rp)}\,.
\end{align*}

If $m\omega> 0 \Rightarrow \sigma > 0$, we have $\Re r_2<M\leq r_+$. Now suppose that $m\omega \leq 0$ and rewrite $r_2$ as
\begin{align*}
r_{2}=M(1+2|\sigma|)\lp(1- \sqrt{1-x}\rp)\,,\quad x=\frac{a^2}{3M^2(1+2|\sigma|)^2}\lp(1-\frac{2m^2}{\Lambda+s^2}\rp)\,.
\end{align*}
Assuming $\Lambda\geq \frac23m^2$,
\begin{align*}
0\leq |x| =\frac{a^2}{3M^2(1+2|\sigma|)^2}\lp|1-\frac{2m^2}{\Lambda+s^2}\rp|\leq \frac{a^2}{3M^2}\leq \frac{1}{3}\lp|1-\frac{2m^2}{\Lambda+s^2}\rp|\leq \frac23\,,
\end{align*}
and, since $0\leq |x| \leq  2/3$, we have the inequality $\sqrt{1-x}\geq 1-2|x|/3 \Rightarrow 1-\sqrt{1-x}\leq 2|x|/3$, we obtain 
\begin{align*}
\Re r_2 < \frac{2M(1+2|\sigma|)a^2}{9M^2(1+2|\sigma|)^2}\lp|1-\frac{2m^2}{\Lambda+s^2}\rp| \leq \frac{2M}{9}\lp|1-\frac{2m^2}{\Lambda+s^2}\rp|\leq \frac{4}{9}M <M\leq r_+\,.
\end{align*}
We can now conclude that there is only one root for $\frac{d}{dr}\lp[(r^2+a^2)^3\frac{d}{dr}\mc{V}_0\rp]$, at  $r_1\in[r_+,\infty)$. 

By Rolle's theorem, there are at most two zeros of $\frac{d}{dr}\mc{V}_0$. As the leading order term of (\ref{eq:dV0/dr}) is negative, since $\Lambda+s^2\geq 0$, we conclude that $\frac{d}{dr}\mc{V}_0$ either vanishes nowhere, vanishes at a unique point $r^0_{\rm max}$, or vanishes at points $r^0_{\min}<r^0_{\rm max}$. Hence, in $[r_+,\infty)$, $\mc{V}_0$ is either (a) strictly decreasing, (b) has a unique critical value $r^0_{\rm max}$ which is a global maximum or (c) has exactly two critical values $r^0_{\rm min}<r^0_{\rm max}$ which are a local minimum and maximum respectively. This concludes the proof of statement 1.

We now move on to proving some general points about the minimum and maximum of $\mc{V}_0$, if they exist. For statement 3, note that
\begin{align}
\omega^2-\mc{V}_0(r_+)=(\omega-m\upomega_+)^2\geq 0\,, \label{eq:V0-r+}
\end{align}
so $\mc{V}_0(r_{\rm min}^0)> \omega^2$ if $r_{\rm min}^0\in(r_+,\infty)$ exists. For statement 2, we observe that, as $r\to\infty$
\begin{align*}
(r^2+a^2)^3\frac{d\mc{V}_0}{dr}(r)=0 \Leftrightarrow r = 3M\lp(1-\frac{2am\omega}{\Lambda+s^2}\rp) + O\lp(\frac{1}{\Lambda r^2}\rp)\,,
\end{align*}
so the right hand side, and, consequently, when $\Lambda\geq c^{-1}\gg 1$, $r^0_{\max}$ is bounded from above if $m\omega\geq 0$ (in this case, independently of the frequency parameters) and if $\Lambda+s^2\geq c |am\omega|$ (in this case, depending on $c$). If additionally $m^2 \leq c\Lambda$ and $a<M$,
\begin{align*}
(r^2+a^2)^3\frac{d\mc{V}_0}{dr}\Big|_{r=(1+\sqrt{2})M} &= O(c) -\lp[r^3-3Mr^2+a^2(r+M)\rp]\big|_{r=(1+\sqrt{2})M}\\
&= O(c)+2(2+\sqrt{2})M(M^2-a^2)>(2+\sqrt{2})M(M^2-a_0^2)\,,
\end{align*}
for $c$ sufficiently small depending on $a_0$ and $M$, hence $r_{\rm max}^0$ must lie above $(1+\sqrt{2})M$.
\end{proof}

\begin{remark} In part (c) of Lemma~\ref{lemma:critical-points-V0}, $(1+\sqrt{2})M$ is the lowest $r$-value at which the maximum of the potential can be attained if $m=0$ and $\Lambda+s^2>0$. Indeed, when $m=0$, \eqref{eq:dV0/dr} and \eqref{eq:ddV0/ddr} reduce to
\begin{align*}
(r^2+a^2)^3\frac{d\mc{V}_0}{dr} &=-2(\Lambda+s^2)\lp[r^3-3Mr^2+a^2(r+M)\rp]\,,\\
\frac{d}{dr}\lp((r^2+a^2)^3\frac{d\mc{V}_0}{dr}\rp)&=-2(\Lambda+s^2)\lp[3r(r-M) +a^2\rp]\,,
\end{align*}
thus the potential has a maximum located at $r_{\rm max}^0\in[(1+\sqrt{2})M,3M]$, the lower bound being attained when $a=M$ and the upper bound for $a=0$. As $1+\sqrt{2}>2$, the maximum is always attained outside the ergorregion.
\end{remark}

We now turn to a more concrete characterization of the trapped frequencies. We begin with two lemmas that show there are no such frequencies when $\Lambda\gg \omega^2$ and for most (near extremal Kerr) or all (for subextremal Kerr) superradiant frequencies.

\begin{lemma} \label{lemma:derivative-V0-r+}
Fix $s\in\mathbb{Z}$, $M>0$ and $a_0\in [0,M)$.  Let $(\omega,m,\Lambda)$ be an admissible frequency triple satisfying one of the following three requirements
\begin{align*}
(i)&~ 0<m\omega\leq m\upomega_+ +\beta_1\Lambda\,,\,\, 0\leq a\leq a_0\,,\,\,\Lambda\geq \max\lp\{\frac{a_0^2}{M^2},\frac23\rp\}m^2\,,\\
(ii)&~0<m\omega\leq m\upomega_+ -\beta_2\Lambda\,,\,\, a_0< a\leq M\,, \,\,\Lambda\geq m^2\,,\\ 
(iii)&~m\omega\leq 0\,, \,\,\Lambda\geq 0\,,
\end{align*}
for some $\beta_1=\beta_1(a_0)>0$ such that $a_0\beta_1<(M^2-a_0^2)/(6M^2)$ and for any $\beta_2>0$. Then, in (i--ii) we have
\begin{align*}
\frac{d}{dr} \mc{V}_0(r_+)\geq b\Lambda \geq 0\,,
\end{align*}
where $b$ is a possibly small constant which, in addition to $M$ and $s$, depends on $a_0$ and $\beta_1$ in case (i) and on $a_0$ and $\beta_3$ in case (ii). In case (iii), we have 
\begin{align*}
\frac{d}{dr} \mc{V}_0(r)\geq b(r-M)\Lambda \geq 0\,,
\end{align*}
for $r$ sufficiently close to $r_+$, independently of the frequency parameters.
\end{lemma}
\begin{proof}  At $r=r_+$, the derivatives of $\mc{V}_0$ at $r_+$ satisfy
\begin{align*}
(r_+^2+a^2)^3\frac{d\mc{V}_0}{dr}&=4r_+ m\upomega_+(m\upomega_+-\omega)+4Mr_+(r-M)\lp(\Lambda-2am\omega\rp)\\
&=4r_+m\upomega_+(m\upomega_+-\omega)+8Mr_+m(r_+-M)(m-a\omega)+4Mr_+(r-M)\lp(\Lambda-2m^2\rp)\\
&= -4Mam(\omega-m\upomega_+)(3r_+^2-a^2)+4(r_+-M)[(\Lambda+s^2)Mr_+-a^2m^2]\,.
\end{align*}
Clearly, if $m\omega\leq 0$ and $a\in[0,a_0]$, then
\begin{align*}
(r_+^2+a^2)^3\frac{d\mc{V}_0}{dr}(r)&\geq \frac{(r-M)}{2M^2r_+^2}\Lambda\,.
\end{align*}
If $m\omega\leq 0$ and $a\in(a_0,M]$, then either $m^2\leq \frac14 \Lambda$, in which case  $\Lambda-2m^2\geq \frac12\Lambda$, or $m^2\geq \frac14\Lambda$, in which case $64M^2m^2\upomega_+^2\leq \Lambda$; in both instances, $\frac{d\mc{V}_0}{dr}$ is positive and comparable to $\Lambda(r-M)$ in the limit $r\to r_+$.

so the conclusion holds in case (iii).

Let us first take $a\in[0,a_0]$. Using the bound $\Lambda\geq \max\lp\{\frac23,\frac{a_0^2}{M^2}\rp\} m^2$, we find
\begin{align*}
(r_+^2+a^2)^3\frac{d}{dr}\mc{V}_0(r_+)&\geq -4Mam(\omega-m\upomega_+)(3r_+^2-a^2)+4M(r_+-M)\sqrt{M^2-a^2_0}\Lambda\\
&\geq -4Mam(\omega-m\upomega_+)(3r_+^2-a^2)+4(r_+-M)\sqrt{M^2-a^2_0}\Lambda \\
&\geq 4M\Lambda\lp[-2ar_+\beta_1(M+2\sqrt{M^2-a^2})+\sqrt{M^2-a^2}\sqrt{M^2-a^2_0}\rp]\\
&\geq 4\Lambda\lp[-6M^2a_0\beta_1+(M^2-a_0^2)\rp]>0\,,
\end{align*}
in case (i), by the assumption on $\beta_1$. 

Now suppose $a\in[a_0,M]$. Using the bound $\Lambda\geq m^2$, we have
\begin{align*}
(r_+^2+a^2)^3\frac{d}{dr}\mc{V}_0(r_+)&\geq -4Mam(\omega-m\upomega_+)(3r_+^2-a^2)+4M(r_+-M)^2\Lambda\\
&\geq 8M^3a_0\beta_2\Lambda\geq b(a_0,\beta_2)\Lambda\,.
\end{align*}
in case (ii) This concludes the proof.
\end{proof}

By Lemmas~\ref{lemma:critical-points-V0} and \ref{lemma:derivative-V0-r+},  we conclude that $\mc{V}_0$ has a unique critical point, located at $r^0_{\rm max}$, which is a maximum. The following lemma gives some quantitative information about the value of the potential at this maximum, and generalizes \cite[Lemma 6.4.2]{Dafermos2016b}.

\begin{lemma} \label{lemma:V0-r0max}
Fix $M>0$, $a_0\in[0,M)$ and $\beta_1,\beta_2>0$ such that Lemma~\ref{lemma:derivative-V0-r+} holds. For all sufficiently small $\varepsilon_{\rm width}$, all $a\in[0,M]$ and all admissible frequency parameters $(\omega,m,\Lambda)$ satisfying one of the following
\begin{align*}
(i)&~ 0<m\omega\leq m\upomega_+ +\beta_1\Lambda\,,\,\, 0\leq a\leq a_0\,,\,\,\Lambda\geq \max\lp\{\frac{a_0}{M},\frac23\rp\}m^2\,,\\
(ii)&~0<m\omega\leq m\upomega_+ -\beta_2\Lambda\,,\,\, a_0< a\leq M\,, \,\,\Lambda\geq m^2\,,\\
(iii)&~ \Lambda\geq \varepsilon_{\rm with}^{-1}\omega^2 \,,\,\,\Lambda\geq \max\lp\{\frac{a_0}{M},\frac23\rp\}m^2\,,
\end{align*}
we have the estimate
$$ \mc{V}_0(r^0_{\rm max})-\omega^2\geq b \Lambda\,,$$
where $b=b(a_0,\beta_1)$ in case (i), $b=b(a_0,\beta_2)$ in case (ii) and $b$ depends only on $M$ and $s$ in case (iii).
\end{lemma}
\begin{proof}  Let $\epsilon>0$ be a fixed sufficiently small constant. As long as $\varepsilon_{\rm width}$ is sufficiently small that  $\varepsilon_{\rm width}<\varepsilon$, the frequency ranges (i), (ii) and (iii) are are contained in one of the five ranges
\begin{align*}
(a)&~\omega^2\leq \epsilon \Lambda \,,\,\, \Lambda\geq \max\lp\{\frac23,\frac{a_0^2}{M^2}\rp\}m^2\,,\\
(b)&~m^2\upomega_+-\epsilon|m|\sqrt{\Lambda}\leq m\omega \leq m^2\upomega_++\beta_1\Lambda\,,\,\, \omega^2>\epsilon\Lambda\,,\,\, 0\leq a \leq a_0\,, \\
(c)&~m^2\upomega_+-\epsilon|m|\sqrt{\Lambda}\leq m\omega \leq m^2\upomega_+-\beta_2\Lambda\,,\,\, \omega^2>\epsilon\Lambda\,,\,\, a>a_0\,, \\
(d)&~0<m\omega<m^2\upomega_+-\epsilon|m|\sqrt{\Lambda}\,,\,\, \omega^2>\epsilon\Lambda\,,\,\,0\leq a\leq a_0\,,\,\, \Lambda\geq \max\lp\{\frac23,\frac{a_0^2}{M^2}\rp\}m^2\,, \\
(e)&~0<m\omega<m^2\upomega_+-\epsilon|m|\sqrt{\Lambda}\,,\,\, \omega^2>\epsilon\Lambda\,,\,\,a>a_0\,,\,\, \Lambda\geq m^2\,. 
\end{align*}
In each of these five ranges, we will find some $r$-value for which the potential $\mc{V}_0$ is quantitatively above $\omega^2$, implying that so is the maximum. 

In range (a),  if we have the bound $\Lambda\geq \frac23 m^2$, then $2am\omega \geq -M\sqrt{6\epsilon}\Lambda$, $-a^2m^2 \geq -2M^2 \Lambda$, hence
\begin{align*}
\mc{V}_0(r)-\omega^2\geq - \epsilon \Lambda+ \frac{\Lambda+s^2}{r^2}+O\lp(\frac{\tilde\Lambda}{r^3}\rp) 
\end{align*}
as $r\to \infty$. Thus, for $\tilde{r}$ sufficiently large and $\epsilon$ sufficiently small depending on $\tilde{r}$, we have 
\begin{align*}
\mc{V}_0(\tilde{r})-\omega^2 \geq b(\epsilon)\Lambda\,.
\end{align*}

In ranges (b) and (c), $|m\upomega_+-\omega|\leq \max\{\varepsilon\sqrt{\Lambda},|\omega|\}$, hence
\begin{align*}
\omega^2-\mc{V}_0(r_+)=(\omega-m\upomega_+)^2\leq \epsilon^2\Lambda\,.
\end{align*}
Combining with \cref{lemma:derivative-V0-r+}, we have, for sufficiently small $\delta>0$ and smaller $\epsilon$,
\begin{align*}
\mc{V}_0(r_++\delta)-\omega^2 =-\epsilon^2\Lambda+\frac{d\mc{V}_0}{dr}(r_+)\delta+O(\delta^2)\geq b\Lambda \,,
\end{align*}
for $b=b(a_0,\beta_1)$ in case (i) and $b=b(a_0,\beta_2)$ in case (ii).

Finally, in the last two frequency ranges, (d) and (e), define
\begin{align*}
r_0:=\frac{m \upomega_+}{\omega}r_+ \in \lp(\frac{\upomega_+}{\upomega_+-\epsilon}r_+,\frac{\upomega_{+}}{\sqrt{\epsilon}}r_+\rp)\,,
\end{align*}
where the bounds are computed using the definition of the frequency range. Letting $\Delta_{r_0}=r_0^2-2Mr_0+a^2$, we compute
\begin{align*}
\mc{V}_0(r_0)-\omega^2 &= \frac{-4M^2r_0^2\omega^2+4Mr_0am -a^2m^2-\omega^2((r_0^2+a^2)^2-4Mr_0^2)+\Delta_{r_0}\Lambda}{(r_0^2+a^2)^2} \\
&=\frac{\Delta_{r_0}}{(r_0^2+a^2)^2}\lp[\Lambda+s^2-\frac{a^2}{4M^2}\lp(1+\frac{2M}{r_0}+\frac{a^2}{r_0^2}\rp)m^2\rp]\\
&\geq \frac{\Delta_{r_0}}{(r_0^2+a^2)^2}\lp[\Lambda-\frac{a^2}{M^2}(1+M^2\epsilon^2-M\epsilon)m^2\rp]\,.
\end{align*}
To conclude, we use $\Lambda\geq a_0^2m^2/M^2$ if $a\in[0,a_0]$ and $\Lambda\geq m^2$ if $a\in(a_0,M]$:
$$\mc{V}_0(r_0)-\omega^2 \geq \frac{\Delta_{r_0}}{(r_0^2+a^2)^2} M\epsilon(1-M\epsilon)\Lambda\geq b(\epsilon)\Lambda\,;$$
this concludes the proof.
\end{proof}

\begin{remark}
The previous lemma shows that, in particular, in the superradiant frequency regime, the maximum of $\mc{V}_0$ is always quantitatively above the energy level of trapped geodesics in the subextremal case $0\leq a \leq a_0$. This property degenerates as $a\to M$: for $a_0<a\leq M$, it only holds if the frequencies are quantitatively away from the superradiant threshold $\omega=m\upomega_+$. 
\end{remark}

We can further restrict the set of nonsuperradiant trapped frequencies:

\begin{lemma}\label{lemma:V0-trapping}
Fix $M>0$ and suppose $a\in[0,M]$. For any $\beta_3>0$, all sufficiently small $\epsilon_{\rm width}$ depending on $\beta_1$, all sufficiently large $\omega_{\rm high}$ depending on $\epsilon_{\rm width}$, and all $(\omega,m,\Lambda)$ admissible satisfying $\Lambda\geq\frac23 m^2$ and 
\begin{align*}
\varepsilon_{\rm width}\omega^2\leq \Lambda\leq \varepsilon_{\rm width}^{-1}\omega^2\,, \quad |\omega|\geq \omega_{\rm high}\,, \quad m\omega\notin (0, m^2\upomega_++\beta_3\Lambda]\,,
\end{align*}
there is a $c=c(\varepsilon_{\rm width})$ such that $\omega^2-\mc V_0(r_+)\geq c\Lambda$. Thus, for some $r'$ depending on the frequency triple but satisfying $|r'-r_+|\geq b(\varepsilon_{\rm width})$, one has
\begin{align}
\mc{V}_0-\omega^2\leq -\frac14 c\Lambda \,, \quad\forall\,r\in[r_+,r_0']\,. \label{eq:trapping-V0-bound-r'}
\end{align}
Let $r_0'\in(r_+,\infty]$ be the supremum over such $r'$. If $r_0'<\infty$, it satisfies $r_0'\leq B(\varepsilon_{\rm width})$.
\end{lemma}

\begin{proof}
As $\sigma:=2am\omega/\Lambda\leq 3M\epsilon_{\rm width}^{1/2}\leq 1$ for sufficiently small $\varepsilon_{\rm width}$, we can obtain the uniform bounds
\begin{align} 
\lp|\frac{d\mc{V}_0}{dr}\rp|+r\lp|\frac{d^2\mc{V}_0}{dr^2}\rp| \leq B(\epsilon_{\rm width})\frac{\Lambda}{r^3}\,, \qquad \lp|\frac{d}{dr}\lp((r^2+a^2)^3\frac{d\mc{V}_0}{dr}\rp)\rp|\leq B(\epsilon_{\rm width})\Lambda r^2\,. \label{eq:trapping-V0-bounds}
\end{align}
Moreover, we have
\begin{align*}
\omega^2-\mc{V}(r_+)=(\omega-m\upomega_+)^2\geq \min\lp\{\lp(\sqrt{\epsilon_{\rm width}}+\frac{1}{2M}\sqrt{\frac32}\rp)^2, \frac23 \beta_3^2\rp\}\Lambda \geq \varepsilon_{\rm width}\Lambda\,,
\end{align*}
hence we can fix $c(\varepsilon_{\rm width})$ to be simply  $\varepsilon_{\rm width}$.

Let $r_0\in(r_+,\infty]$ be the largest value such that
$$\mc{V}_0(r)\leq \mc{V}_0(r_+)+\frac{c}{2}\Lambda\,, \qquad \forall r\in[r_+,r_0]\,.$$
Clearly $r_0-r_+\geq b(\varepsilon_{\rm width})$. If $r_0=\infty$, then the potential is always below $\mc{V}_0(r_+)+\frac{c}{2}\Lambda$, hence any $r_0'>r_+$ satisfies \eqref{eq:trapping-V0-bound-r'} for $b\leq c/2$; thus we can take $r_0'=\infty$. 

We now turn to the case $r_0<\infty$. As $\frac{d\mc{V}_0}{dr}(r_0)\geq 0$, by Lemma~\ref{lemma:critical-points-V0}, $\mc{V}_0$ has a maximum at $r_{\rm max}^0\geq r_0$ and can also have a minimum at $r_{\rm min}^0$ satisfying $r_{\rm min}^0<r_0\leq r_{\rm max}^0<B(\varepsilon_{\rm width})$. We can consider two cases
\begin{itemize}
\item If $\mc{V}_0(r_{\rm max}^0)\leq \mc{V}_0(r_+)+\frac{3c}{4}\Lambda$, the inequality holds in $[r_+,\infty)$, so the bound \eqref{eq:trapping-V0-bound-r'} holds for any $r'\geq r_0$. Thus, $r_0'=\infty$.
\item If $\mc{V}_0(r_{\rm max}^0)>\mc{V}_0(r_+)+\frac{3c}{4}\Lambda$, from the bound for the first derivative in \eqref{eq:trapping-V0-bounds}, one can show that $r_{\rm max}^0-r_0\geq b(\epsilon_{\rm width})$. Since $\frac{d\mc{V}_0}{dr}(r_0)\geq 0$, by continuity, there is a $r'\in[r_0,r_{\rm max}^0]$ such that
\begin{align*}
\mc{V}_0(r) \leq \mc{V}_0(r_+)+\frac{3\Lambda}{4}\leq \omega^2 -\frac{1}{4}c\Lambda\,, \quad \forall r\in[r_+,r']\,.
\end{align*}
Thus, as any such $r'$  lies below $r_{\rm max}^0\leq B(\varepsilon_{\rm width})$, so does $r_0'$.
\end{itemize}
This concludes the proof.
\end{proof}

\subsubsection{Overview of the section}
\label{sec:unbounded-overview}

Recall that transformed radial ODE for $\Psi$ \eqref{eq:radial-ODE-Psi}, which we can write as 
\begin{align}
\Psi''+\lp(\omega^2-\frac{\Delta\Lambda+4Mram\omega-a^2m^2}{(r^2+a^2)^2}-\mc{V}_1\rp)\Psi=\mathfrak{G}+aw\sum_{k=0}^{|s|-1}\lp(im c_{s,\,|s|,\,k}^{\Phi}+ c_{s,\,|s|,\,k}^{\mr{id}}\rp)\uppsi_{(k)}\,,\label{eq:radial-ODE-Psi-again}
\end{align}
where $\mc{V}_1(r)$ is the frequency independent part of the potential \eqref{eq:radial-ODE-Psi-potentials}, which can be neglected if $\omega$ or $\Lambda$ are very large. Naively, as the left hand side of this equation depends on the frequency parameters quadratically whereas the coupling errors on the right hand side are, at most, linear in the frequency, we expect that the latter do not play a role in the dynamics. Thus, one can hope to establish estimates for $\Psi$ using the methods of \cite{Dafermos2016b} and, in contrast with the case in  Section~\ref{sec:bounded-smallness}, without needing a very detailed understanding of the behavior of $\uppsi_{(k)}$, for $k<|s|$.

Our naive argument, when put in place, immediately encounters two important obstacles. First and foremost, our reasoning completely fails if there is an $r$ value, or range of $r$ values, for which there is some cancellation between $\omega^2$ and the frequency dependent part of the potential causing the left hand side of \eqref{eq:radial-ODE-Psi-again} to be less than quadratic in the frequency parameters.   A second issue, of a more technical nature, is that in employing virial or Killing energy currents to analyze \eqref{eq:radial-ODE-Psi-again}, applications of Cauchy--Schwarz on the coupling terms may leave quadratic or higher frequency factors on the lower level $\uppsi_{(k)}$, $k<|s|$, which \textit{a priori} would seem to compete with or surpass those on $\Psi$.

Let us first explain how one can reduce the frequency weights on $\uppsi_{(k)}$ to be at most at the level of those on $\Psi$. While we would like to treat \eqref{eq:radial-ODE-Psi-again} as a wave equation for spin-weighted functions, it is not one, hence it does not share the wave equation's conserved quantities. For instance, application of a global Killing $T$ energy produces errors due to coupling which look like
\begin{align*}
-aw\sum_{k=0}^{|s|-1}\omega\Im\lp[\lp(im c_{s,\,|s|,\,k}^{\Phi}+ c_{s,\,|s|,\,k}^{\mr{id}}\rp)\uppsi_{(k)}\overline{\Psi}\rp]\,.
\end{align*}
When trapping occurs, as we do not control $|\Psi|^2$ with quadratic frequency weights, the weights falling on $\uppsi_{(k)}$ by Cauchy--Schwarz greatly surpass those on $\Psi$. We reduce our frequency weights by borrowing an important observation appearing already in \cite{Dafermos2017}: the transport equations relating the transformed variables can be used to eliminate the dependence of the coupling errors on $\omega$ (see Lemma~\ref{lemma:basic-estimate-3}). With this reduction, it is easy to see that the coupling errors to control arising from either the Killing energy or virial current can be no worse than $a^2(m^2+1)|\uppsi_{(k)}|^2$, with suitable $r$-weights, for $k<|s|$.

\begin{itemize}[noitemsep]
\item If $\omega^2-\mc{V}_0\gtrsim \omega^2$ globally, then we can exploit that with a global $y$-type virial current to obtain a good bulk estimate. The virial current generates coupling errors of strength $a^2(m^2+1)|\uppsi_{(k)}|^2$ in the frequency parameters and a boundary term with the wrong sign. 
\begin{itemize}[noitemsep]
\item \textit{Virial current coupling errors.} By optimizing the ratio $y/y'$, we may gain a smallness parameter on the coupling errors due to the $y$ current, hence they are controlled by our previous bulk estimate with Lemma~\ref{lemma:basic-estimate-1}.
\item \textit{Boundary terms.} Using Lemma~\ref{lemma:basic-estimate-1} may further contribute to the boundary terms which we need to control; we must therefore be careful with their contributions being bounded, i.e.\ with having sufficient control over $\mathfrak{D}_{s,k}^{\mc{I}}$ or $\mathfrak{D}_{s,k}^{\mc{H}}$ (depending on where the boundary terms lie). The global sign in $\omega^2-\mc{V}_0\gtrsim \omega^2$ means, in particular, that superradiance cannot take place, hence a global energy current can be applied to control the boundary term. If it is the Teukolsky--Starobinsky energy current, and we have enough control over $\mathfrak{C}_s$ and $\mathfrak{D}_{s,k}^{\mc{I}}$ or $\mathfrak{D}_{s,k}^{\mc{H}}$ (depending on where the boundary terms lie), we are done. If the energy current of our choice is Killing, it also generates coupling errors of strength $a^2(m^2+1)|\uppsi_{(k)}|^2$ in the frequency parameters, but no smallness can be associated to them if this energy current is to absorb the boundary terms produced by $y$ and Lemma~\ref{lemma:basic-estimate-1}.
\end{itemize}
This is the case in $\mc{F}_{\text{\ClockLogo}}$ and $\mc{F}_{\rm comp,2}$, analyzed in Section~\ref{sec:F-time}.

\item If $\omega^2-\mc{V}_0$ has no global sign, we instead seek to combine local information on this sign with global or almost global information on the critical point structure of $\mc{V}_0$. In this case, $\mc{V}_0$ and thus the full $\mc{V}$ have a maximum and, in some circumstances, a minimum located at a lower $r$ value, but no other critical points (see Lemma~\ref{lemma:critical-points-V0}). An $f$ current which is positive above the maximum and negative below it can be used to take advantage of the sign of $\mc{V}'$ (almost) globally.
\begin{itemize}[noitemsep]
\item \textit{No trapping.} Outside the trapping range, there is no minimum, hence the $f$ current creates a generally good bulk term. Yet there are two catches in using the $f$ current: first, near the maximum of the potential, the bulk term is degenerate; second, a boundary term with the wrong sign is created. To deal with the second, we must add energy currents and, as superradiance may occur, these may have to introduce localization errors. To deal with the first, note that near the maximum of the potential  $\omega^2-\mc{V}$ has a sign (this is why we say there is no trapping) that can be exploited with a compactly supported and very large $h$ current; it is in the region where $h$ is strong that we bury the localization errors arising from applications of energy currents. 
\item \textit{Trapping, no superradiance.} In the trapping regime, the same downsides to the $f$ current are present. Moreover, we will see that the potential may have a minimum, in which case the bulk term due to $f$ has a bad sign between $r=r_+$ and the minimum. By definition of trapping, there is no clear sign to $\omega^2-\mc{V}$ at the maximum of the potential, so we can only expect to control $|\Psi|^2$  with linear frequency weights by borrowing from the $|\Psi'|^2$ bulk that does not degenerate at the maximum of the potential. At the minimum, however, $\omega^2-\mc{V}$ has a sign  that can be exploited using a $y$ current localized between $r=r_+$ and the minimum. The only drawback left to address are the boundary terms to which $y$ adds: for those, we consider a global energy current, as superradiance does not occur, avoiding localization errors.
\end{itemize}
There is also the option of adding $y$ currents near $r=\infty$ or $r=r_+$ in both cases in order to improve the weights in the bulk terms; this adds to the boundary terms  which we will have to control. 

The procedure we outlined applies to the $\mc{F}_{\measuredangle,1-2}$ and $\mc{F}_{\sun,1-2}$ and to the $\mc{F}_{\rm comp,1}$ frequency regimes, respectively. However, the the presence of the coupling terms means there are additional difficulties when $s\neq 0$: using our reduction, we see that, if we are to apply the Killing energy currents to control the boundary terms produced, then all coupling errors are controlled by large multiples of $a^2(m^2+1)|\uppsi_{(k)}|^2$ bulk terms with suitable $r$-weights, for $k<|s|$.
\end{itemize}

The only question that remains is how to deal with the $a^2(m^2+1)|\uppsi_{(k)}|^2$ bulk terms that the methods described above produce. If $|a|$ were small, these could easily be treated with Lemma~\ref{lemma:basic-estimate-1} (c.f.\ \cite{Dafermos2017}). For more general $a$ the answer comes from the remaining transformed system, i.e.\ the ODEs other than \eqref{eq:radial-ODE-Psi-again}, for $k<|s|$. By a very simple estimate on these ODEs, we are able to show that $(m^2+1)|\uppsi_{(k)}|^2$ bulk terms can be controlled by $|\Psi|^2$ bulk terms for frequency triples with certain properties (see already Lemma~\ref{lemma:basic-estimate-3}). This novel estimate is the key new ingredient in the analysis of the unbounded frequency regimes for $s\neq 0$. We use it whenever available, most notably for in the trapped frequency range.

\subsubsection{An estimate on the transformed system}

Note that, while Lemma~\ref{lemma:basic-estimate-1} provides an unconditional way of climbing the hierarchy of the transformed system, i.e.\ converting estimates in $\uppsi_{(k)}$ into  estimates on $\uppsi_{(k+1)}$, the procedure loses derivatives. The following lemma provides some examples of more refined estimates which however  have a more limited range of applicability:

\begin{lemma} \label{lemma:basic-estimate-2} Let $(\omega,m,\Lambda)$ be an admissible frequency triple such that $\Lambda\geq \Lambda_{\rm high}>0$. 

First, assume that, for some $\varepsilon_1\in(0,1)$ sufficiently small, there is a small constant $b$ such that 
\begin{align}
\Lambda-2am\omega-\varepsilon_1\lp(\omega^2+\frac{a^2}{Mr_+}m^2\rp)-1\geq b(\varepsilon_1)\Lambda \text{~~and~~} \tilde{\varepsilon}_1:=\frac{1}{b(\varepsilon_1,\varepsilon_2)\varepsilon_1 \Lambda_{\rm high}}\ll 1\,; \label{eq:basic-estimate-2-condition-1}
\end{align}
then, we have
\begin{align}
\begin{split}
\int_{-\infty}^\infty b(\varepsilon_1)w\Lambda\lp|\uppsi_{(k)}\rp|^2dr^*
&\leq  \int_{-\infty}^\infty Bw\varepsilon_1^{-1} \lp|\uppsi_{(k+1)}\rp|^2dr^* +B\int_{-\infty}^\infty k\sum_{i=0}^{k}\lp|\Re\lp[\mathfrak{G}_i\overline{\uppsi_{(i)}}\rp]\rp|dr^*\\
&\leq  \int_{-\infty}^\infty B\tilde{\varepsilon}_1 b(\varepsilon_1)\Lambda w \lp|\uppsi_{(k+1)}\rp|^2dr^* +B\int_{-\infty}^\infty k\sum_{i=0}^{k}\lp|\Re\lp[\mathfrak{G}_i\overline{\uppsi_{(i)}}\rp]\rp|dr^*\,. 
\end{split}\label{eq:basic-estimate-2-a}\\
\begin{split}
\int_{-\infty}^\infty b(\varepsilon_1)w\Lambda\lp|\uppsi_{(k)}'\rp|^2dr^*
&\leq  \int_{-\infty}^\infty B\tilde{\varepsilon}_1 b(\varepsilon_1)\Lambda w \lp(\lp|\uppsi_{(k+1)}'\rp|^2+\lp|\uppsi_{(k+1)}\rp|^2\rp)dr^* \\
&\qquad+B\int_{-\infty}^\infty k\sum_{i=0}^{k}\lp(\lp|\Re\lp[\mathfrak{G}_i\overline{\uppsi_{(i)}}\rp]\rp|+\lp|\Re\lp[\mathfrak{G}_i'\overline{\uppsi_{(i)}}'\rp]\rp|\rp)dr^*\,. 
\end{split}\label{eq:basic-estimate-2-a-prime}
\end{align}

Second, assume for some $\varepsilon_1, \varepsilon_2\in(0,1)$ sufficiently small, $\omega^2\leq B \varepsilon_1^{-1} \varepsilon_2\Lambda_{\rm high}$, 
\begin{align}
\Lambda-2am\omega-\varepsilon_1\frac{a^2}{Mr_+}m^2-1-\varepsilon_2\Lambda_{\rm high}\geq b(\varepsilon_1,\varepsilon_2)\Lambda \text{~~and~~} \tilde{\varepsilon}_1:=\frac{1}{b(\varepsilon_1)\varepsilon_1 \Lambda_{\rm high}}\ll 1\,;
\label{eq:basic-estimate-2-condition}
\end{align}
then, we have
\begin{align}
\begin{split}
&\int_{-\infty}^\infty b(\varepsilon_1,\varepsilon_2)w\Lambda\lp|\uppsi_{(k)}\rp|^2dr^*\\
&\quad\leq \int_{-\infty}^\infty Bw\lp(\frac{\omega^2}{\varepsilon_2\Lambda_{\rm high}} \omega^2+\frac{2\varepsilon_1^{-1}}{r}\rp)\lp|\uppsi_{(k+1)}\rp|^2dr^*+B\int_{-\infty}^\infty \sum_{i=0}^{k}  \lp|\Re\lp[\mathfrak{G}_i\overline{\uppsi_{(i)}}\rp]\rp|dr^*\\
&\quad\leq \int_{-\infty}^\infty Bw\lp(\frac{\omega^2}{\varepsilon_2\Lambda_{\rm high}}+\frac{2\tilde\varepsilon_1 b(\varepsilon_1,\varepsilon_2)\Lambda}{r}\rp)\lp|\uppsi_{(k+1)}\rp|^2dr^*+B\int_{-\infty}^\infty \sum_{i=0}^{k}  \lp|\Re\lp[\mathfrak{G}_i\overline{\uppsi_{(i)}}\rp]\rp|dr^*\,.
\end{split}\label{eq:basic-estimate-2-b}\\
\begin{split}
&\int_{-\infty}^\infty b(\varepsilon_1,\varepsilon_2)w\Lambda\lp|\uppsi_{(k)}'\rp|^2dr^*\\
&\quad\leq \int_{-\infty}^\infty Bw\lp(\frac{\omega^2}{\varepsilon_2\Lambda_{\rm high}}+\frac{2\tilde\varepsilon_1 b(\varepsilon_1,\varepsilon_2)\Lambda}{r}\rp)\lp(\lp|\uppsi_{(k+1)}\rp|^2+\lp|\uppsi_{(k+1)}'\rp|^2\rp)dr^*\\
&\quad\qquad+B\int_{-\infty}^\infty \sum_{i=0}^{k}  \lp(\lp|\Re\lp[\mathfrak{G}_i\overline{\uppsi_{(i)}}\rp]\rp|+\lp|\Re\lp[\mathfrak{G}_i'\overline{\uppsi_{(i)}}'\rp]\rp|\rp)dr^*\,.
\end{split}\label{eq:basic-estimate-2-b-prime}
\end{align}
\end{lemma}
\begin{proof}
Consider the radial ODE \eqref{eq:transformed-k-separated}. Dropping subscripts and using the definition of $\uppsi_{(k+1)}$ \eqref{eq:transformed-transport-separated}, we find that \eqref{eq:transformed-k-separated} can be rewritten as
\begin{align*}
&\lp(\sign s \frac{d}{dr^*}-i\omega+\frac{iam}{r^2+a^2}\rp)\lp(w\uppsi_{(k+1)}\rp)-\sign s (|s|-k)w'\uppsi_{(k+1)}+w(\Lambda-2am\omega)\uppsi_{(k)}\\
&\quad\qquad-\frac{2arw}{r^2+a^2}\sign s(2|s|-2k-1) im \uppsi_{(k)}+U\uppsi_{(k)}\\
&\quad= \mathfrak{G}_{(k)}+aw\sum_{i=0}^{k-1}\lp(im c_{s,\,k,\,i}^\Phi+c^{\rm id}_{s,\,k,\,i}\rp)\uppsi_{(i)}\,,  \numberthis\label{eq:ode-constraint}
\end{align*}
where $ac^{\rm id}_{s,\,k,\,i}$ should be replaced by $c^{\rm id}_{s,\,k,\,i}$ if $|s|\neq 1$ and $(k,i)=(1,0)$, and where
\begin{equation}
\begin{split}
U_{(k)}(r)&:=w\left\{|s|+k(2|s|-k-1)+a^2 \frac{\Delta}{(r^2+a^2)^2}[1-2s-2k(2|s|-k-1)]\rp.\\
&\qquad\lp.+\frac{2Mr(r^2-a^2)}{(r^2+a^2)^2}[1-3|s|+2s^2-3k(2|s|-k-1)]\right\}\,,
\end{split} \label{eq:basic-estimate-2-intermediate-1}
\end{equation}
which, we recall from the proof of Lemma~\ref{lemma:properties-potential-low-omega}\ref{it:properties-potential-low-omega-compact-region}, can be bounded from below by $\frac18 w$. Multiplying \eqref{eq:basic-estimate-2-intermediate-1} by $\overline{\uppsi_{(k)}}$ and taking real parts, we obtain the identity (noting again the change mentioned above for $(k,i)=(1,0)$ and $|s|\neq 1$)
\begin{align*}
w(\Lambda-2am\omega)\lp|\uppsi_{(k)}\rp|^2&= -w^2\lp|\uppsi_{(k+1)}\rp|^2-2w\lp(\omega-\frac{am}{r^2+a^2}\rp)\Im\lp[\uppsi_{(k+1)}\overline{\uppsi_{(k)}}\rp]-U_{(k)}\lp|\uppsi_{(k)}\rp|^2\\
&\qquad+\sign s (|s|-k)w'\Re\lp[\uppsi_{(k+1)}\overline{\uppsi_{(k)}}\rp]+\frac12 \lp(-\sign s w\Re\lp[\overline{\uppsi_{(k)}}\uppsi_{(k+1)}\rp]\rp)'\numberthis \label{eq:basic-estimate-2-intermediate-2}\\
&\qquad -\Re\lp[\mathfrak{G}_k\overline{\uppsi_{(k)}}\rp]-aw\sum_{i=0}^{k-1}\Re\lp[\overline{\uppsi_{(k)}}\lp(im c_{s,\,k,\,i}^\Phi+c^{\rm id}_{s,\,k,\,i}\rp)\uppsi_{(i)}\rp]\,. 
\end{align*}

After integration in $r^*$, given the asymptotics of $\uppsi_{(k)}$ laid out in Lemma~\ref{lemma:uppsi-general-asymptotics} (see also Definition~\ref{def:uhor-uout}), we find that the boundary terms vanish. Applying Cauchy--Schwarz on all terms and bounding some of the expressions from above, we find that for any $\varepsilon_1>0$,
\begin{align*}
&\int_{-\infty}^\infty w(\Lambda-2am\omega)\lp|\uppsi_{(k)}\rp|^2dr^*\\
&\quad\leq  \int_{-\infty}^\infty w\lp[\varepsilon_1 \lp(\omega^2+\frac{a^2}{2Mr_+}m^2\rp)+1\rp]\lp|\uppsi_{(k)}\rp|^2dr^*+\int_{-\infty}^\infty\lp[\frac{2Mr_+\varepsilon_1^{-1}}{(r^2+a^2)^2}+\varepsilon_1^{-1}+\frac{s^2}{4}\lp(\frac{w'}{w}\rp)^2\rp]w \lp|\uppsi_{(k+1)}\rp|^2dr^*\\
&\quad\qquad-\int_{-\infty}^\infty\Re\lp[\mathfrak{G}_k\overline{\uppsi_{(k)}}\rp]dr^*+\int_{-\infty}^\infty \lp\{\varepsilon_1\frac{a^2}{2Mr_+}m^2w\lp|\uppsi_{(k)}\rp|^2 + kC \sum_{i=0}^{k-1}w\varepsilon_1^{-1}\lp|\uppsi_{(i)}\rp|^2\rp\}dr^*\\
&\quad\leq  \int_{-\infty}^\infty w\lp\{\varepsilon_1 \lp(\omega^2+\frac{a^2}{Mr_+}m^2\rp)+1\rp\}\lp|\uppsi_{(k)}\rp|^2dr^*+\int_{-\infty}^\infty\lp[\lp(\varepsilon_1^{-1}+s^2\rp)\frac{1}{r^2+a^2}+\varepsilon_1^{-1}\rp] w \lp|\uppsi_{(k+1)}\rp|^2dr^*\\
&\quad\qquad-\int_{-\infty}^\infty\Re\lp[\mathfrak{G}_k\overline{\uppsi_{(k)}}\rp]dr^*+\int_{-\infty}^\infty  kC \sum_{i=0}^{k-1}w\varepsilon_1^{-1}\lp|\uppsi_{(i)}\rp|^2dr^*\,,
\end{align*}
where $C:=\max_{0\leq i,k\leq |s|} 4M^2\lp(\lp\lVert c_{s,\,k,\,i}^\Phi\rp\rVert_\infty^2+\lp\lVert c_{s,\,k,\,i}^{\rm id}\rp\rVert_\infty^2\rp)$. Alternatively, if we let the $\omega^2$ weight fall on $\uppsi_{(k+1)}$, for some $\varepsilon_1,\varepsilon_2>0$,
\begin{align*}
&\int_{-\infty}^\infty w(\Lambda-2am\omega)\lp|\uppsi_{(k)}\rp|^2dr^*\\
&\quad\leq  \int_{-\infty}^\infty w\lp(\varepsilon_2\Lambda_{\rm high}+\varepsilon_1 \frac{a^2}{Mr_+}m^2+1\rp)\lp|\uppsi_{(k)}\rp|^2dr^*+\int_{-\infty}^\infty\lp[\lp(\varepsilon_1^{-1}+s^2\rp)\frac{1}{r^2+a^2}+\frac{\omega^2}{\varepsilon_2\Lambda_{\rm high}}\rp]w \lp|\uppsi_{(k+1)}\rp|^2dr^*\\
&\quad\qquad-\int_{-\infty}^\infty\Re\lp[\mathfrak{G}_k\overline{\uppsi_{(k)}}\rp]dr^*+\int_{-\infty}^\infty   kC \sum_{i=0}^{k-1}w\varepsilon_1^{-1}\lp|\uppsi_{(i)}\rp|^2dr^*\,.
\end{align*}

Now let $|s|\geq 2$. Assume $\Lambda\geq \Lambda_{\rm high}$ is sufficiently large that
$$\quad\tilde{\varepsilon}_1:=\lp[\varepsilon_1b(\varepsilon_1)\Lambda_{\rm high}\rp]^{-1}\ll \min\{1,|s|C\}\,.$$
Then, since the inequality
\begin{align*}
\int_{-\infty}^\infty b(\varepsilon_1)w\Lambda\lp|\uppsi_{(k)}\rp|^2
&\leq  \int_{-\infty}^\infty B\varepsilon_1^{-1} w \lp|\uppsi_{(k+1)}\rp|^2dr^* -\int_{-\infty}^\infty\Re\lp[\mathfrak{G}_k\overline{\uppsi_{(k)}}\rp]dr^*+\int_{-\infty}^\infty   kC \sum_{i=0}^{k-1}w \varepsilon^{-1}_1\lp|\uppsi_{(i)}\rp|^2dr^*\,,
\end{align*}
holds for all $k=0,...,|s|-1$, we can iterate it along the hierarchy to eliminate its dependence on $|\uppsi_{(i)}|^2$ for $i<k$. To illustrate the procedure, we start at $k=0$, which already has no such terms, and show how to obtain close an estimate for $|\uppsi_{(1)}|^2$:
\begin{align*}
\int_{-\infty}^\infty b(\varepsilon_1)w\Lambda\lp|\uppsi_{(0)}\rp|^2
&\leq \int_{-\infty}^\infty B\varepsilon_1^{-1}w\lp|\uppsi_{(1)}\rp|^2dr^*-\int_{-\infty}^\infty\Re\lp[\mathfrak{G}_0\overline{\uppsi_{(0)}}\rp]dr^*\,,\\
\int_{-\infty}^\infty b(\varepsilon_1)w\Lambda\lp|\uppsi_{(1)}\rp|^2
&\leq  \int_{-\infty}^\infty B\varepsilon_1^{-1} w \lp|\uppsi_{(2)}\rp|^2dr^* -\int_{-\infty}^\infty\Re\lp[\mathfrak{G}_1\overline{\uppsi_{(1)}}\rp]dr^*+\int_{-\infty}^\infty   C  \tilde{\varepsilon}_1  w b(\varepsilon_1)\Lambda \lp|\uppsi_{(0)}\rp|^2dr^*\\
&\leq \int_{-\infty}^\infty B \varepsilon_1^{-1} w \lp|\uppsi_{(2)}\rp|^2dr^* +\int_{-\infty}^\infty   B  \tilde{\varepsilon}_1^2 b(\varepsilon_1)\Lambda  \lp|\uppsi_{(1)}\rp|^2dr^*\\
&\qquad-\int_{-\infty}^\infty\lp\{\Re\lp[\mathfrak{G}_1\overline{\uppsi_{(1)}}\rp]-kC \tilde{\varepsilon}_1\Re\lp[\mathfrak{G}_0\overline{\uppsi_{(0)}}\rp]\rp\}dr^*\,,
\end{align*}
where, in the last inequality, the term in $|\uppsi_{(1)}|^2$ on the right hand side can clearly be absorbed into the left hand side. This concludes the proof of \eqref{eq:basic-estimate-2-a}. Estimate~\ref{eq:basic-estimate-2-b} follows similarly.

Finally, note that we can commute \eqref{eq:ode-constraint} with $\p_{r^*}$ to obtain
\begin{align*}
&\lp(\sign s \frac{d}{dr^*}-i\omega+\frac{iam}{r^2+a^2}\rp)\lp(w\uppsi_{(k+1)}'\rp)-\sign s (|s|-k-1)w'\uppsi_{(k+1)}'+w(\Lambda-2am\omega)\uppsi_{(k)}'\\
&\quad\qquad-\frac{2arw}{r^2+a^2}\sign s(2|s|-2k-1) im \uppsi_{(k)}'+U\uppsi_{(k)}'-\sign s(|s|-k-1) \lp(w'\uppsi_{(k)}\rp)'\\
&\quad\qquad +\lp\{-i\lp[w\lp(\omega-\frac{am}{r^2+a^2}\rp)\rp]'+w'(\Lambda-2am\omega)-\lp(\frac{2arw}{r^2+a^2}\rp)'\sign s(2|s|-2k-1) im +U_{(k)}'\rp\}\uppsi_{(k)}\\
&\quad= \mathfrak{G}_{(k)}'+\sum_{i=0}^{k-1}\lp\{w\lp(ac_{s,\,k,\,i}^\Phi im+ac^{\rm id}_{s,\,k,\,i}\rp)\uppsi_{(i)}'+\lp(wac_{s,\,k,\,i}^\Phi im+wac^{\rm id}_{s,\,k,\,i}\rp)'\uppsi_{(i)}\rp\}\,.
\end{align*}
Thus, repeating the same procedure as above and using the fact that we have already shown the estimates \eqref{eq:basic-estimate-2-a} and \eqref{eq:basic-estimate-2-b}, we can obtain \eqref{eq:basic-estimate-2-a-prime} and \eqref{eq:basic-estimate-2-b-prime}.
\end{proof}

\subsubsection{Time dominated regime and large, comparable, non-superradiant and non-trapped regime}
\label{sec:F-time}

This section concerns the frequency ranges $\mc{F}_{\text{\ClockLogo}}$ and $\mc{F}_{\rm comp,2}$. We will prove

\begin{proposition}[Estimates in $\mc{F}_{\text{\ClockLogo}}$] \label{prop:time-dominated} Fix $s\in\{0,\pm 1,\pm 2\}$ and  $M>0$. For all $\delta\in(0,1]$, for all $E,E_W>0$ such that one of these is sufficiently large, for all $\omega_{\rm high}$ sufficiently large, for all $\varepsilon_{\rm width}$ sufficiently small, and for all $(a,\omega,m,\Lambda)\in\mc{F}_{\text{\ClockLogo}}(\varepsilon_{\rm width},\omega_{\rm high})$, there exist functions $y$ and $\hat{y}$ satisfying the bounds $|y|+|\hat{y}|\leq B$ such that for all smooth $\Psi$ arising from a smooth solution to the radial ODE~\eqref{eq:radial-ODE-alpha} via \eqref{eq:def-psi0-separated} and \eqref{eq:transformed-transport-separated} and itself satisfying the radial ODE~\eqref{eq:radial-ODE-Psi}, 
\begin{enumerate}[font=\normalfont\bfseries, label=\Alph*.]
\item if $\Psi$ has outgoing boundary conditions as in Definition~\ref{def:outgoing-bdry-uppsi}, we have the estimate
\begin{align*}
&\lp[\omega^2\lp|\swei{A}{s}_{\mc{I}^+}\rp|^2+(\omega-m\upomega_+)^2\lp|\swei{A}{s}_{\mc{H}^+}\rp|^2\rp] + b(\delta)\int_{-\infty}^\infty \frac{\Delta}{r^{3+\delta}}\Big[|\Psi'|^2+\omega^2|\Psi|^2\Big]dr^*\\
&\quad\leq \int_{-\infty}^\infty \lp\{2(y+\hat{y})\Re\lp[\mathfrak{G}\overline{\Psi}'\rp]-E\omega\Im\lp[\mathfrak{G}\overline{\Psi}\rp]+E_W\lp(Q^{W,T}\rp)'\rp\}dr^*\,;
\end{align*}
\item if $\Psi$ arises from a solution of the homogeneous Teukolsky radial ODE~\eqref{eq:radial-ODE-alpha} and has the general boundary conditions in Lemma~\ref{lemma:uppsi-general-asymptotics}, we have the estimates:
\begin{equation*} 
\begin{split}
&\omega^2
\lp\{\begin{array}{lr}
\frac{\mathfrak{C}_s}{\mathfrak{D}_s^{\mc I}}\,, &s\leq 0\\
1\,, &s>0
\end{array}\rp\}
\lp|\swei{A}{s}_{\mc{I}^+}\rp|^2+(\omega-m\upomega_+)^2
\lp\{\begin{array}{lr}
1\,, &s\leq 0\\
\frac{\mathfrak{C}_s}{\mathfrak{D}_s^{\mc H}}\,, &s>0
\end{array}\rp\}
\lp|\swei{A}{s}_{\mc{I}^+}\rp|^2
\\
&\quad\leq 4\lp[\omega^2
\lp\{\begin{array}{lr}
1\,, &s\leq 0\\
\frac{\mathfrak{C}_s}{\mathfrak{D}_s^{\mc I}}\,, &s>0
\end{array}\rp\}
\lp|\swei{A}{s}_{\mc{I}^-}\rp|^2+(\omega-m\upomega_+)^2
\lp\{\begin{array}{lr}
\frac{\mathfrak{C}_s}{\mathfrak{D}_s^{\mc H}}\,, &s\leq 0\\
1\,, &s>0
\end{array}\rp\}
\lp|\swei{A}{s}_{\mc{I}^-}\rp|^2\rp]\,,
\end{split}
\end{equation*}
 and 
\begin{align*}
&\lp[\omega^2\lp|\swei{A}{s}_{\mc{I}^+}\rp|^2+(\omega-m\upomega_+)^2\lp|\swei{A}{s}_{\mc{H}^+}\rp|^2\rp] + b(\delta)\int_{-\infty}^\infty \frac{\Delta}{r^{3+\delta}}\Big[|\Psi'|^2+\omega^2|\Psi|^2\Big]dr^*\\
&\quad\leq B(E)\lp[\omega^2\lp|\swei{A}{s}_{\mc{I}^-}\rp|^2+(\omega-m\upomega_+)^2\lp|\swei{A}{s}_{\mc{H}^-}\rp|^2\rp]+\int_{-\infty}^\infty \lp\{2(y+\hat{y})\Re\lp[\mathfrak{G}\overline{\Psi}'\rp]-E\omega\Im\lp[\mathfrak{G}\overline{\Psi}\rp]\rp\}dr^*\,,
\end{align*}
if $a=0$, $s=0$ or, if $s\neq 0$, as long as $s$, $a$, and $(\omega,m,\Lambda)$ further satisfy the requirement that
\begin{align*}
a^2\lp|\frac{\swei{A}{s}_{k,\mc{H}^-}}{\swei{A}{s}_{\mc{H}^-}}\rp|\text{~and~}\frac{\mathfrak{C}_s}{\mathfrak{D}_s^{\mc H}} \text{~~if~}s<0\,, \text{~~or~~}a^2\lp|\frac{\swei{A}{s}_{k,\mc{I}^-}}{\swei{A}{s}_{\mc{I}^-}}\rp|\text{~and~}\frac{\mathfrak{C}_s}{\mathfrak{D}_s^{\mc I}} \text{~~if~}s>0
\end{align*} 
are bounded independently of the frequency parameters for all $0\leq k\leq |s|$.
\end{enumerate}
\end{proposition}

\begin{proposition}[Estimates in $\mc{F}_{\rm comp,2}$] \label{prop:comp2} Fix $s\in\{0,\pm 1,\pm 2\}$ and $M>0$. Then, for all $\beta_3>0$, for all $\varepsilon^{-1}_{\rm width}$ sufficiently large depending on $\beta_3$, for all $E,E_{\rm width}>0$ such that one of these is sufficiently large depending on $\varepsilon_{\rm width}$, for all $\omega_{\rm high}$ sufficiently large depending on $\varepsilon_{\rm width}$, and for all  $(a,\omega,m,\Lambda) \in  \mc{F}_{\rm comp,2}(\omega_{\rm high},\varepsilon_{\rm width},r_0',\beta_3)$ where $r_0'$ is fixed by Lemma~\ref{lemma:V0-trapping}, there exist functions $y$ and $\hat{y}$ satisfying the uniform bounds
$|y| +|\hat{y}|\leq B \,,$
such that for all smooth $\Psi$ arising from a smooth solution to the radial ODE~\eqref{eq:radial-ODE-alpha} via \eqref{eq:def-psi0-separated} and \eqref{eq:transformed-transport-separated} and itself satisfying the radial ODE~\eqref{eq:radial-ODE-Psi}, 
\begin{enumerate}[font=\normalfont\bfseries, label=\Alph*.]
\item  if $\Psi$ has outgoing boundary conditions of Definition~\ref{def:outgoing-bdry-uppsi}, we have the estimate
\begin{align*}
&\lp[\omega^2\lp|\swei{A}{s}_{\mc{I}^+}\rp|^2+(\omega-m\upomega_+)^2\lp|\swei{A}{s}_{\mc{H}^+}\rp|^2\rp] + b(\delta, \varepsilon_{\rm width})\int_{-\infty}^\infty \frac{\Delta}{r^{3+\delta}}\Big[|\Psi'|^2+\omega^2|\Psi|^2\Big]dr^*\\
&\quad\leq \int_{-\infty}^\infty \lp\{\sum_{k=0}^{|s|}2(y+\hat{y})\Re\lp[\mathfrak{G}_{(k)}\overline{\uppsi_{(k)}}'\rp]-E\omega\Im\lp[\mathfrak{G}\overline{\Psi}\rp]+E_W\lp(Q^{W}\rp)'\rp\}dr^*\,,
\end{align*}
where we may drop the terms on the right hand side with $k\neq |s|$ if $|a|\leq {a}_0$ is sufficiently small depending on $\varepsilon_{\rm width}$;
\item  if $\Psi$ arises from a solution of the homogeneous Teukolsky radial ODE~\eqref{eq:radial-ODE-alpha} and has the general boundary conditions in Lemma~\ref{lemma:uppsi-general-asymptotics}, we have the estimate:
\begin{equation*} 
\begin{split}
&\omega^2
\lp\{\begin{array}{lr}
\frac{\mathfrak{C}_s}{\mathfrak{D}_s^{\mc I}}\,, &s\leq 0\\
1\,, &s>0
\end{array}\rp\}
\lp|\swei{A}{s}_{\mc{I}^+}\rp|^2+(\omega-m\upomega_+)^2
\lp\{\begin{array}{lr}
1\,, &s\leq 0\\
\frac{\mathfrak{C}_s}{\mathfrak{D}_s^{\mc H}}\,, &s>0
\end{array}\rp\}
\lp|\swei{A}{s}_{\mc{I}^+}\rp|^2
\\
&\quad\leq \varepsilon_{\rm width}^{-2}\lp[\omega^2
\lp\{\begin{array}{lr}
1\,, &s\leq 0\\
\frac{\mathfrak{C}_s}{\mathfrak{D}_s^{\mc I}}\,, &s>0
\end{array}\rp\}
\lp|\swei{A}{s}_{\mc{I}^-}\rp|^2+(\omega-m\upomega_+)^2
\lp\{\begin{array}{lr}
\frac{\mathfrak{C}_s}{\mathfrak{D}_s^{\mc H}}\,, &s\leq 0\\
1\,, &s>0
\end{array}\rp\}
\lp|\swei{A}{s}_{\mc{I}^-}\rp|^2\rp]\,.
\end{split}
\end{equation*}
\end{enumerate}
\end{proposition}

Let us begin by discussing the specific properties of the frequency parameters in the relevant frequency ranges and its consequences for the behavior of the potential.

\begin{lemma}[Properties of the frequency parameters in $\mc{F}_{\text{\ClockLogo}}$] \label{lemma:time-dominated-frequency-properties} Fix $s\in\mathbb{Z}$ and $M>0$. Let $(\omega,m,\Lambda)\in \mc{F}_{\text{\ClockLogo}}(\varepsilon_{\rm width},\omega_{\rm high})$ be an admissible frequency triple with respect to $a\in[0,M]$. Then, for sufficiently small $\varepsilon_{\rm width}$ and sufficiently large $\omega_{\rm high}$, we have
\begin{enumerate}[label=(\roman*)]
\item $m^2\leq 2\varepsilon_{\rm width}\omega^2$; \label{it:time-dominated-frequency-properties-m2}
\item $|\Lambda|\leq 2\varepsilon_{\rm width}\omega^2$; \label{it:time-dominated-frequency-properties-Lambda-upper-lower}
\item $2\omega^2\geq \omega(\omega-m\upomega_+)\geq  \frac12 \omega^2$; \label{it:time-dominated-frequency-properties-no-superrad}
\item for $|s|=1,2$ and $a=0$ we have \label{it:time-dominated-frequency-properties-Cs-Ds}
\begin{align*}
\frac{\mathfrak{C}_s}{\mathfrak{D}_s^{\mc{I}}}= \frac{\mathfrak{C}_s}{\mathfrak{D}_s^{\mc{H}}}= \frac{\mathfrak{D}_s^{H}}{\mathfrak{D}_s^{\mc{I}}} =1\,.
\end{align*}
\end{enumerate}
\end{lemma}
\begin{proof}
By property \ref{it:admissible-freqs-triple-Lambda-lower-bound} in Definition~\ref{def:admissible-freqs}, we obtain statements \ref{it:time-dominated-frequency-properties-m2} and \ref{it:time-dominated-frequency-properties-Lambda-upper-lower}:
\begin{align*}
m^2&\leq \Lambda+2|s||a\omega|+s^2 \leq 2\varepsilon_{\rm width}\omega^2\,,\\
|\Lambda|&\leq \max\{-\Lambda,\Lambda\} \leq \max\{|s|(2M|\omega|-1),\varepsilon_{\rm width}\omega^2\}\leq 2\varepsilon_{\rm width}\omega^2\,,
\end{align*}
for sufficiently small $\varepsilon_{\rm width}$ and sufficiently large $\omega_{\rm high}$. From these bounds, we conclude
\begin{align*}
\omega(\omega-m\upomega_+)&\geq \omega^2\lp(1-\sqrt{2\varepsilon_{\rm width}}\frac{1}{2M}\rp)\geq \frac12 \omega^2\,,\\
\omega(\omega-m\upomega_+)&\leq \omega^2\lp(1+\sqrt{2\varepsilon_{\rm width}}\frac{1}{2M}\rp)\leq 2\omega^2\,,
\end{align*}
which proves \ref{it:time-dominated-frequency-properties-no-superrad}.
\end{proof}

\begin{lemma}[Properties of the frequency parameters in $\mc{F}_{\rm comp,2}$] \label{lemma:comparable-frequency-properties} Fix $s\in\mathbb{Z}$, $M>0$ and $\beta_3>0$. Let $r_0'$ be defined by Lemma~\ref{lemma:V0-trapping} and $\beta_3>0$. Let $(\omega,m,\Lambda)\in \mc{F}_{\rm comp,2}(\varepsilon_{\rm width}, \omega_{\rm high},r_0',\beta_3)$ be an admissible frequency triple with respect to $a\in[0,M]$. Then, for sufficiently small $\varepsilon_{\rm width}$ and $\omega_{\rm high}^{-2}$ with respect to $\beta_3$,
\begin{enumerate}[label=(\roman*)]
\item $\Lambda\geq \max\{\frac{a_0^2}{M^2},\frac23\}m^2$; \label{it:comparable-frequency-properties-Lambda-m2}
\item $\Lambda\geq \varepsilon_{\rm width}^{1/2}|m\omega|$; \label{it:comparable-frequency-properties-sigma}
\item in $\mc{F}_{\rm comp,\,2a}(\varepsilon_{\rm width}, \omega_{\rm high},r_0',\beta_3,\beta_4)$, if $\omega_{\rm high}$ is sufficiently large depending also on $\beta_4$, the conditions for Lemma~\ref{lemma:basic-estimate-2} hold, namely, for $\varepsilon_1=\varepsilon_{\rm width}\beta_4/3$, we have 
$$\Lambda-2am\omega-\varepsilon_1(m^2+\omega^2)\geq \frac18\Lambda\,, \qquad \tilde{\varepsilon}_1 =\frac{6}{\beta_4\varepsilon_{\rm width}^2\omega_{\rm high}^2}\ll 1\,,$$
\label{it:comparable-frequency-properties-Lambda-nondegenerate}
as long as $\omega_{\rm high}^2$ is sufficiently large depending on $\varepsilon_{\rm width}$;
\item $2\varepsilon_{\rm width}^{1/2}/M\omega^2\leq \omega(\omega-m\upomega_+)\leq  \varepsilon_{\rm width}^{-1} \omega^2$; \label{it:comparable-frequency-properties-no-superrad}
\item for $|s|=1,2$, in $\mc{F}_{\rm comp,\,2a}(\varepsilon_{\rm width}, \omega_{\rm high},r_0',\beta_3,\beta_4)$, if $\omega_{\rm high}$ is sufficiently large (depending on $\varepsilon_{\rm width}$ and $\beta_4$), we have \label{it:comparable-frequency-properties-Cs-Ds}
\begin{align*}
\frac{\mathfrak{C}_s}{\mathfrak{D}_s^{\mc{I}}}\,,\,\, \frac{\mathfrak{C}_s}{\mathfrak{D}_s^{\mc{H}}}\,, \,\, \frac{\mathfrak{D}_s^{H}}{\mathfrak{D}_s^{\mc{I}}} \in\lp[\frac13,\frac53\rp]\,.
\end{align*}
\end{enumerate}
\end{lemma}
\begin{proof}
We begin by showing that \ref{it:comparable-frequency-properties-Lambda-m2} holds more generally when $\Lambda\geq \varepsilon_{\rm width}\omega^2$ and $|\omega|\geq \omega_{\rm high}$. If $m^2\leq \varepsilon_{\rm width}\omega^2$, of course $\Lambda\geq m^2$. On the other hand, if $m^2>\varepsilon_{\rm width}\omega^2$ then in particular $|m|\gg |s|$, so it will be convenient to recall condition \ref{it:admissible-freqs-triple-Lambda-lower-bound} of Definition~\ref{def:admissible-freqs}:
\begin{align*}
\Lambda\geq l(l+1)-s^2-2|s||a\omega| \,.
\end{align*}
If $\omega_{\rm high}$ is sufficiently large and either $0\leq |a|< \tilde{a}_0\ll M$ sufficiently small or $m^2\geq \varepsilon_{\rm high}^{-1}\omega^2$, for instance, clearly $\Lambda\geq m^2$. Thus, we only have to check the case $|a|\geq \tilde{a}_0$ and $\varepsilon_{\rm width}\omega^2< m^2<\varepsilon_{\rm width}^{-1}\omega^2$. Then,
\begin{align*}
\Lambda\geq l(l+1)-s^2-2|s||a\omega| \geq (1-\epsilon)m^2+(\epsilon \omega_{\rm high}\varepsilon_{\rm width}+\sqrt{\varepsilon_{\rm width}}-2|s|M-s^2\omega_{\rm high}^{-1})|\omega|  \geq (1-\epsilon)m^2\,,
\end{align*}
 for any $\epsilon\leq 4M|s|\omega_{\rm high}^{-1}\varepsilon_{\rm width}^{-1}$ as long as $\omega_{\rm high}\gg \varepsilon_{\rm width}^{-1}$ is sufficiently large. In particular, for sufficiently small $\varepsilon_{\rm width}$ and $\omega_{\rm high}^{-1}$, we can take $\epsilon=\min\lp\{\frac13,\frac{M-a_0}{M}\rp\}$. The proof of \ref{it:comparable-frequency-properties-sigma} now follows easily.

We turn to \ref{it:comparable-frequency-properties-Lambda-nondegenerate}: if $\Lambda-2am\omega\geq \beta_4\Lambda$, then
\begin{align*}
\Lambda-2am\omega-\varepsilon_1(m^2+\omega^2)\geq \Lambda\lp[\beta_4-\varepsilon_1\lp(2+\varepsilon_{\rm width}^{-1}\rp)\rp]\geq \frac12 \beta_4\Lambda
\end{align*}
as long as $\varepsilon_1\leq \varepsilon_{\rm width}\beta_4/3$.

Let us now focus on \ref{it:comparable-frequency-properties-no-superrad}. Either $m\omega\leq 0$, in which case
\begin{align*}
1\geq \frac{\omega}{\omega-m\upomega_+}=\frac{|\omega|}{|\omega|+|m|\upomega_+} \geq \frac{1}{1+\sqrt{\frac32}\varepsilon_{\rm width}^{-1/2}\upomega_+}\geq \frac{2}{M}\varepsilon_{\rm width}^{1/2}\,,
\end{align*}
or $m(\omega-m\upomega_+)\geq \beta_3\Lambda$, in which case
\begin{align*}
1< \frac{\omega}{\omega-m\upomega_+}=\frac{|m\omega|}{|m(\omega-m\upomega_+)|} \leq \frac{|m\omega|}{\beta_3\Lambda}\leq \sqrt{\frac32}\frac{\varepsilon_{\rm width}^{-1/2}}{\beta_3}\leq \varepsilon_{\rm width}^{-1}\,,
\end{align*}
using the fact that we can choose $\varepsilon_{\rm width}<\frac23 \beta_3^2$ for sufficiently small $\varepsilon_{\rm width}$.

Finally, for \ref{it:comparable-frequency-properties-Cs-Ds}, we refer the reader to the proof of the analogous statement in Lemma~\ref{lemma:angular-dominated-frequency-properties}\ref{it:angular-dominated-frequency-properties-Cs-Ds}: by \ref{it:comparable-frequency-properties-Lambda-nondegenerate}, the conclusion follows easily.
\end{proof}

Based on the previous bounds for the frequency parameters, we obtain

\begin{lemma}[Behavior of the potential in $\mc{F}_{\text{\ClockLogo}}$ and $\mc{F}_{\rm comp,2}$] \label{lemma:time-dominated-potential-properties}
Let $(\omega,m,\Lambda)\in \mc{F}_{\text{\ClockLogo}}(\varepsilon_{\rm width},\omega_{\rm high})$ be an admissible frequency triple. Then, for sufficiently small $\varepsilon_{\rm width}$ and sufficiently large $\omega_{\rm high}$, we have
\begin{alignat}{3}
&\mc{V}\leq \frac{\varepsilon_{\rm width}^{1/2}\omega^2}{r^2}\,, &&\qquad \mc{V}'\leq \frac{\varepsilon_{\rm width}^{1/2}\omega^2\Delta}{r^3(r^2+a^2)}\,, &&\qquad \text{~in~} \mc{F}_{\text{\ClockLogo}}\,, \label{eq:time-dominated-potential-bound}\\
&\mc{V}\leq b(\varepsilon_{\rm width})\frac{\omega^2}{r^2}\,,\qquad
&&|\mc{V}'|\leq b(\varepsilon_{\rm width})\frac{\omega^2\Delta}{r^5}\,, &&\qquad \text{~in~} \mc{F}_{\rm comp,2}\,.\label{eq:comparable-potential-bound}
\end{alignat}
\end{lemma}
\begin{proof}
In $\mc{F}_{\text{\ClockLogo}}$, by inspection of \eqref{eq:radial-ODE-Psi-potentials} and \eqref{eq:dV0/dr}, we get from Lemma~\ref{lemma:time-dominated-frequency-properties},
\begin{align*}
\mc{V}_0&=\frac{\Delta\Lambda+4Mram\omega-a^2m^2}{(r^2+a^2)^2}\leq \frac{\Delta +4M |a|r(2\varepsilon_{\rm width})^{-1/2}+2a^2\varepsilon_{\rm width}}{(r^2+a^2)^2}(2\varepsilon_{\rm width})\omega^2 \leq \frac{2\varepsilon_{\rm width}\omega^2}{r^2}\,,\\
\mc{V}_0'&= -2\frac{\Delta}{(r^2+a^2)^4} \lp[\Lambda
(r^3 -3Mr^2+a^2r+Ma^2)-(6Mr^2+Ma^2)am\omega+2ra^2m^2\rp]\\
&\leq \frac{\Delta\lp[(r^3 +3Mr^2+a^2r+Ma^2)+(6Mr^2+Ma^2)(2\varepsilon_{\rm width})^{-1/2}+2ra^2\rp]}{(r^2+a^2)^{4}}4\varepsilon_{\rm width}\omega^2\leq \frac{2\varepsilon_{\rm width}^{1/2}\omega^2\Delta}{r^3(r^2+a^2)} \,,
\end{align*}
for  $\varepsilon_{\rm width}$ sufficiently small and $\omega_{\rm high}$ sufficiently large depending on $M$. 

In $\mc{F}_{\rm comp,2}$, by Lemma~\ref{lemma:V0-trapping}, we have
\begin{align*}
\mc{V}_0\leq b(\varepsilon_{\rm wdith})\frac{\omega^2}{r^2}\,,
\end{align*}
and, making use of the properties in Lemma~\ref{lemma:comparable-frequency-properties} and the fact that $\varepsilon_{\rm width}\leq \Lambda\leq \varepsilon_{\rm width}\omega^2$, we obtain
\begin{align*}
\mc{V}_0'\leq \frac{2\Delta}{(r^2+a^2)^4} \lp[\Lambda
|r^3 -3Mr^2+a^2r+Ma^2|+|6Mr^2+Ma^2|a|m\omega|+2ra^2m^2\rp]\leq b(\varepsilon_{\rm width})\frac{\omega^2\Delta}{r^5}\,.
\end{align*}

In both cases, if  $\omega_{\rm high}$ sufficiently large, the contribution from $\mc{V}_1$ can be neglected, and one obtains \eqref{eq:time-dominated-potential-bound} and \eqref{eq:comparable-potential-bound}, respectively.
\end{proof}

We briefly describe the strategy to prove Propositions~\ref{prop:time-dominated} and \ref{prop:comp2}, which is best illustrated by Figure~\ref{fig:time-dominated}. 

\begin{figure}[htbp]
\centering
\includegraphics[scale=1]{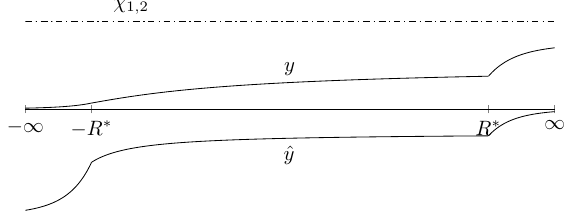}
\caption{Currents in the proof of Theorems~\ref{thm:ode-estimates-Psi-A} and \ref{thm:ode-estimates-Psi-B} for the frequency ranges $\mc{F}_{\text{\ClockLogo}}$ and $\mc{F}_{\rm comp,2}$ (relative size not accurate).}
\label{fig:time-dominated}
\end{figure}

Lemma \ref{lemma:time-dominated-potential-properties} has shown that, in both the time dominated and the large, comparable non-trapped cases, $\omega^2-\mc{V}$ is globally positive and large, prompting us to take advantage of it using a $y$ current. One could become concerned about the term $y\mc{V}'$: for a suitable choice of $y$, one can have $y'(\omega^2-\mc{V})\gg |y\mc{V}'|$ as $r^*\to \pm \infty$, but in a compact region of $r^*$, the latter can affect the positivity of the $y$ current bulk. In the time dominated case, to ensure positivity, we appeal to the smallness of $\Lambda$, hence $|\mc{V}'|$, compared to $\omega^2$; in the regime where $\Lambda$ is comparable to $\omega^2$, we must instead create smallness in the ratio $|y/y'|$ restricted to a compact region. 

One of the main differences with respect to the $s=0$ case (c.f.\ \cite[Propositions 8.4.1 and 8.6.1]{Dafermos2016b}) is the presence of coupling terms when $s\neq 0$. In the time dominated regime, such a presence can easily be overlooked: the frequency dependence of terms coupling terms is such that, as long as the $r$-weights decay sufficiently fast,
they can always be absorbed by noting $\omega^2\gg m^2$ is very large. In the comparable regime, this is no longer the case; however, creating smallness in the ratio $|y/y'|$ will again be sufficient.

Regarding boundary terms, we note that in both the time dominated case and the large, comparable, non-superradiant non-trapped regime, there is no superradiance (by assumption for $\mc{F}_{\rm comp,2}$ and by statement Lemma~\ref{lemma:time-dominated-frequency-properties}\ref{it:time-dominated-frequency-properties-no-superrad} for $\mc{F}_{\text{\ClockLogo}}$), so one can apply global energy currents. In some circumstances, we will be able to add a large multiple of a global Killing $T$ current, capable of dominating the $y$ boundary term: as the coupling errors created by it have strength $m\omega$ in the frequency variables, they can easily be controlled in the time dominated case, when $m\omega\ll \omega^2$, and whenever $\Lambda-2am\omega\geq b\Lambda$ (using the improved estimates in Lemma~\ref{lemma:basic-estimate-2}).  

However, when $m\omega\sim \omega^2\sim \Lambda$ and $\Lambda-2am\omega = o(\Lambda)$, the Killing $T$ current's coupling errors are not obviously small. A natural idea would therefore be to appeal to a Teukolsky--Starobinsky $T$ energy  current, which does not produce coupling errors. Yet, as $\Lambda-2am\omega = o(\Lambda)$, not all of the boundary terms of the Teukolsky--Starobinsky energy current are equally strong as the boundary terms produced by a $y$ current. For outgoing boundary conditions, there is a way out of this conundrum: we first choose $y$ such that it only produces the boundary term that can be absorbed by the Teukolsky--Starobinsky energy current (we impose one of $y(\pm \infty)=0$ depending on the sign of the spin $s$), add the Teukolsky--Starobinsky energy current to control this term and finally control the remaining boundary term by adding a small multiple of the $T$ current. Alternatively, by designing the $y$ currents in this way and applying them to the radial ODEs for $\uppsi_k$, for each $k=0,\dots,|s|$, rather than just the $k=|s|$ ODE, we then control bulk terms of  $\uppsi_k$ which are similar to those produced by the Killing $T$ current's coupling errors. This strategy again requires outgoing boundary conditions; more general boundary conditions are in general not tractable with these methods.

We note that in the time dominated case,  unless $|\Lambda-2am\omega| \geq b |a\omega|$ for some $b>0$, our results only deal with the case of outgoing boundary conditions. Indeed, absent Lemma~\ref{lemma:basic-estimate-2}, we must employ instead Lemma~\ref{lemma:basic-estimate-1}, which produces boundary terms at the levels $k\leq |s|$. By a judicious choice of $r$-weight to apply in Lemma~\ref{lemma:basic-estimate-1}, we can ensure that only \textit{past}  boundary terms are created; however these may not be controlled by the $\Psi$ boundary terms in the entire $\mc{F}_{\text{\ClockLogo}}$ range.

\begin{proof}[Proof of Propositions~\ref{prop:time-dominated} and \ref{prop:comp2}] Let $(\omega,m,\Lambda)$ be any admissible frequency triple with respect to $s$ and $a\in[0,M]$. In both the frequency ranges we intend to study, we will employ $y$ and $Q^T$ currents. We begin by giving explicit currents $y$ and $\hat{y}$ that will be our starting point. For some $R^*\geq 1$ and $C\geq 1$ to be fixed, set
\begin{align*}
y(r^*)&:=\exp\lp(-C\int_{r(r^*)}^\infty \frac{dr}{\Delta} \rp)\mathbbm{1}_{[-R^*,R^*]}+\frac{y(-R^*)(R^*)^{1/2}}{(-r^*)^{1/2}}\mathbbm{1}_{(-\infty,-R^*)}+y(R^*)\lp(2-\frac{(R^*)^{\delta}}{(r^*)^{\delta}}\rp)\mathbbm{1}_{(R^*,\infty)}\,,\\
y'(r^*)&=\frac{y}{2(-r^*)}\mathbbm{1}_{(-\infty,-R^*)}+C\frac{y}{r^2+a^2}\mathbbm{1}_{[-R^*,R^*)}+\delta\frac{y(R^*)(R^*)^{\delta}}{(r^*)^{1+\delta}}>0\,,\\
\hat y(r^*)&:=-\exp\lp(-C(r-r_+)\rp)\mathbbm{1}_{[-R^*,R^*]}+\hat y(-R^*)\lp(2-\frac{(R^*)^{1/2}}{(-r^*)^{1/2}}\rp)\mathbbm{1}_{(-\infty,-R^*)}+\frac{\hat y(R^*)(R^*)^{\delta}}{(r^*)^{\delta}}\mathbbm{1}_{(R^*,\infty)}\,,\\
\hat y'(r^*)&=\frac{\hat y(-R^*)(R^*)^{1/2}}{2(-r^*)^{3/2}}\mathbbm{1}_{(-\infty,-R^*)}+C\frac{\hat y}{r^2+a^2}\mathbbm{1}_{[-R^*,R^*)}+\delta\frac{\hat{y}}{r^*}\mathbbm{1}_{(R^*,\infty)}>0\,,
\end{align*}
so that $y(\infty)=2y(R^*)$, $y(-\infty)=0$, $\hat y(\infty)=0$ and $y(-\infty)=2\hat y(-R^*)$.

When we apply $y$, $\hat{y}$ or $EQ^T$ currents to the equation for $\Psi$, error terms arise due to coupling with $\uppsi_{(k)}$. For $y$, note that 
\begin{align*}
\lp|\frac{wy^2}{y'}\rp|+\lp|\frac{w\hat y^2}{\hat{y}'}\rp|&\leq 4\max\{y(R^*),\hat{y}(r^*)\}\lp(\frac{r^*}{r}\rp)^{1+\delta}\lp(\frac{R}{R^*}\rp)^\delta\frac{1}{\delta R}\mathbbm{1}_{(R^*,\infty)} +\lp(\lp|\frac{wy^2}{y'}\rp|+\lp|\frac{w\hat y^2}{\hat{y}'}\rp|\rp)\mathbbm{1}_{(-\infty,R^*]}\\
&\leq B(M)\lp(\frac{\mathbbm{1}_{[-R^*,R^*]^c}}{\sqrt{R^*}}+\frac{\mathbbm{1}_{[-R^*,R^*]}}{C}\rp)\leq B(M) \max\{R^{-1/2},C^{-1}\}\,,\numberthis\label{eq:time-dominated-y-properties}
\end{align*}
as long as $R$ is large enough that $\sqrt{R}\geq \delta^{-1}$. Thus, the coupling errors generated by the $y$ current are
\begin{align*}
&\sum_{k}^{|s|-1}\int_{-\infty}^\infty 2ayw \Re\lp[\lp(c_{s,|s|,k}^{\rm id}+imc_{s,|s|,k}^{\Phi}\rp)\uppsi_{(k)}\overline{\Psi}'\rp]dr^*\\
&\quad\leq \int_{-\infty}^\infty \lp\{\frac{1}{2}y'|\Psi'|^2+2|s|a^2\max_{k}\lp(\lVert c_{s,|s|,k}^{\rm id}\rVert_\infty^2+\lVert c_{s,|s|,k}^{\Phi}\rVert_\infty^2\rp)\sum_{k=0}^{|s|-1}(1+m^2)w\lp(\frac{wy^2}{y'}\rp)\lp|\uppsi_{(k)}\rp|^2\rp\}dr^*\\
&\quad\leq \int_{-\infty}^\infty\lp\{ \frac{1}{2}y'|\Psi'|^2+\frac{a^2B}{\min\{C,\sqrt{R}\}}(1+m^2) w\lp|\uppsi_{(k)}\rp|^2\rp\}dr^*\,, 
\end{align*}
and, similarly,
\begin{align*}
&\sum_{k}^{|s|-1}\int_{-\infty}^\infty 2a\hat{y}w \Re\lp[\lp(c_{s,|s|,k}^{\rm id}+imc_{s,|s|,k}^{\Phi}\rp)\uppsi_{(k)}\overline{\Psi}'\rp]dr^*\leq \int_{-\infty}^\infty\lp\{ \frac{\hat{y}'}{2}|\Psi'|^2+\frac{a^2|s|B(1+m^2)}{\min\{C,\sqrt{R}\}}w\lp|\uppsi_{(k)}\rp|^2\rp\}dr^*\,, 
\end{align*}
which we have treated using Cauchy--Schwarz. Similarly, for the coupling terms due to adding $EQ^T$, we have, for some $\varepsilon>0$,
\begin{align*}
&\sum_{k}^{|s|-1}\int_{-\infty}^\infty E aw \Im\lp[\lp(c_{s,|s|,k}^{\rm id}+imc_{s,|s|,k}^{\Phi}\rp)\omega\uppsi_{(k)}\overline{\Psi}\rp]dr^*\\
&\quad\leq \int_{-\infty}^\infty E |s|B(M,|s|)\lp[A^2w(|\omega|+\varepsilon^{-1} m^2)\lp|\uppsi_{(k)}\rp|^2+w(|\omega|+\epsilon\omega^2)|\Psi|^2\rp]dr^*\,. 
\end{align*} 

If we repeatedly apply estimate \eqref{eq:basic-estimate-1}, from Lemma~\ref{lemma:basic-estimate-1}, with
\begin{gather*}
c=c_+:=\frac{r-r_+}{r}\Rightarrow \frac{w^2c^2}{c'}=\frac{(r-r_+)^2}{(r^2+a^2)r_+}w\quad \text{if~} s>0\,,\qquad
c=c_-:=-\frac{1}{r}\Rightarrow \frac{w^2c^2}{c'}\leq\frac{w}{r^2+a^2}\quad \text{if~} s<0\,,
\end{gather*}
then, by Lemma~\ref{lemma:uppsi-general-asymptotics}, the boundary terms at $\mc{I}^+$ and $\mc{H}^+$ vanish, so we obtain
\begin{align*}
&\int_{-\infty}^\infty \lp\{2a\{y,\hat{y}\}w \Re\lp[\lp(c_{s,k,k}^{\rm id}+imc_{s,k,k}^{\Phi}\rp)\uppsi_{(k)}\overline{\Psi}'\rp]-E aw \Im\lp[\lp(c_{s,k,k}^{\rm id}+imc_{s,k,k}^{\Phi}\rp)\omega\uppsi_{(k)}\overline{\Psi}\rp]\rp\}dr^*\\
&\quad\leq \int_{-\infty}^\infty \lp\{\frac{\{y',\hat y'\}}{4}|\Psi'|^2+w\lp[a^2B\max\{C^{-1},R^{-1/2},E\varepsilon^{-1},E\}(m^2+|\omega|)+E (\epsilon\omega^2+|\omega|)\rp]\lp|\Psi\rp|^2\rp\}dr^* \numberthis \label{eq:time-dominated-basic-estimate-1}\\
&\quad\qquad +a^2|s|B(M,|s|)\max\{C^{-1},R^{-1/2},E\varepsilon^{-1},E,(E\varepsilon^{-1})^{|s|},E^{|s|}\}\sum_{k=0}^{|s|-1}(|\omega|+m^2)\lp|A_{k,\,\{\mc{I}^-,\mc{H}^-\}}^{[s]}\rp|^2\,,
\end{align*}
where the alternatives are obtained by choosing the $y$ current and having $s\geq 0$ or having the $\hat{y}$ current and $s<0$, respectively.

On the other hand, we can apply Lemma~\ref{lemma:basic-estimate-2} if the conditions for  \eqref{eq:basic-estimate-2-a} hold. Then, we obtain
\begin{align*}
&\int_{-\infty}^\infty \lp\{2a\{y,\hat{y}\}w \Re\lp[\lp(c_{s,k,k}^{\rm id}+imc_{s,k,k}^{\Phi}\rp)\uppsi_{(k)}\overline{\Psi}'\rp]-E aw \Im\lp[\lp(c_{s,k,k}^{\rm id}+imc_{s,k,k}^{\Phi}\rp)\omega\uppsi_{(k)}\overline{\Psi}\rp]\rp\}dr^*\\
&\quad\leq \int_{-\infty}^\infty \lp\{\frac14\{y',\hat y'\}|\Psi'|^2+w\lp[a^2|s|B\max\{C^{-1},R^{-1/2},E\varepsilon^{-1},E\}\tilde{\epsilon}_1\Lambda+E (\epsilon\omega^2+|\omega|)\rp]\lp|\Psi\rp|^2\rp\}dr^*\numberthis \label{eq:time-dominated-basic-estimate-2}\\
&\quad\qquad +|s|\sum_{k=0}^{|s|-1}\int_{-\infty}^\infty \lp|\Re\lp[\mathfrak{G}_{(k)}\uppsi_{(k)}\rp]\rp|dr^*\,. 
\end{align*}

\medskip

\noindent \textit{Time dominated case.} Let us first consider the case of time dominated frequencies,  $\mc{F}_{\text{\ClockLogo}}(\varepsilon_{\rm width},\omega_{\rm high})$. Let $C=1$. By Lemma~\ref{lemma:time-dominated-potential-properties},
\begin{align*}
y'\omega^2-(y\mc{V})'&\geq y'\omega^2\lp[1-\frac{\varepsilon_{\rm width}^{1/2}}{r^2}-\frac{\varepsilon_{\rm width}^{1/2}\Delta}{r^3(r^2+a^2)}\lp(\frac12 |r^*|\mathbbm{1}_{(-\infty,-R^*)}+(r^2+a^2)\mathbbm{1}_{[-R^*,R^*]}+6\frac{|r^*|^{3/2}}{(R^*)^{1/2}}\rp)\rp]\\
&\geq \frac12 y'\omega^2\,,
\end{align*}
for sufficiently small $\varepsilon_{\rm width}$, and, likewise,
\begin{align*}
\hat{y}'\omega^2-(\hat{y}\mc{V})'&\geq \frac12 \hat{y}'\omega^2\,.
\end{align*}

Suppose $s>0$. Rescale $y$ so that $y(\infty)=1$. Note that Lemma~\ref{lemma:basic-estimate-2} cannot be applied, as $\Lambda-2am\omega$ may not have the necessary positivity. However, for some $\varepsilon\in(0,1)$, there are sufficiently large $\omega_{\rm high}$ and $\varepsilon_{\rm width}^{-1}$ so that
\begin{align*}
 \max\{R^{-1/2},E\varepsilon^{-1}\}(1+m^2) \leq E\varepsilon \omega^2\lp[\max\{R^{-1/2}E^{-1},\varepsilon^{-1}\}(\omega_{\rm high}^{-2}+\varepsilon_{\rm width})\rp] \leq B E\varepsilon \omega^2\,,
\end{align*}
hence, by our previous considerations, in particular \eqref{eq:time-dominated-basic-estimate-1}, $Q^y+EQ^T+E_WQ^{W,T}$ currents give the estimate
\begin{align*}
&\int_{-\infty}^\infty\frac14\lp( y'|\Psi|^2+y'\omega^2|\Psi|^2\rp)dr^*\\
&\quad\leq \int_{-\infty}^\infty B\lp(E\varepsilon^{-1}(m^2+|\omega|)+E\varepsilon\omega^2\rp)w\lp|\Psi\rp|^2dr^* +\int_{-\infty}^\infty \lp\{2y\Re\lp[\mathfrak{G}\overline{\Psi}'\rp]-E\omega\Im\lp[\mathfrak{G}\overline{\Psi}\rp]\rp\}dr^*\\
&\quad\qquad-\omega(\omega-m\upomega_+)\lp(E+E_W\frac{\mathfrak{C}_s}{\mathfrak{D}_s^{\mc H}}\rp)\lp|\swei{A}{s}_{\mc{H}^+}\rp|^2+\omega^2(2-E-E_W)\lp|\swei{A}{s}_{\mc{I}^+}\rp|^2 + \omega(\omega-m\upomega_+)(E+E_W)\lp|\swei{A}{s}_{\mc{H}^-}\rp|^2\\
&\quad\qquad +\lp[\frac{a^2(|\omega|+m^2)}{\omega^2}|s|B(M,|s|)(E\varepsilon^{-1})^{|s|}\sum_{k=0}^{|s|-1}\frac{\lp|A_{k,\,\mc{I}^-}^{[s]}\rp|^2}{\lp|\swei{A}{s}_{\mc{I}^-}\rp|^2}+2+E+E_W\frac{\mathfrak{C}_s}{\mathfrak{D}_s^{\mc I}}\rp]\omega^2\lp|\swei{A}{s}_{\mc{I}^-}\rp|^2
\,,
\end{align*}
where we have used \eqref{eq:time-dominated-basic-estimate-1} with $\varepsilon\in(0,1)$. The first term on the right hand side can be absorbed into the left hand side if we make $\epsilon $, $\varepsilon_{\rm width}$ and $\omega_{\rm high}^{-1}$ small enough depending on $E$; we fix $\epsilon$ in this way. Choosing $E>0$ and $\max\{E_W,E\}\geq 4$, we obtain
\begin{align*}
&\omega^2\lp|\swei{A}{s}_{\mc{I}^+}\rp|^2+\frac12E(\omega-m\upomega_+)^2\lp|\swei{A}{s}_{\mc{H}^+}\rp|^2+\int_{-\infty}^\infty\frac14\lp( y'|\Psi|^2+y'\omega^2|\Psi|^2\rp)dr^*\\
&\quad\leq \lp(2+E+E_W\frac{\mathfrak{C}_s}{\mathfrak{D}_s^{\mc I}}\rp)\omega^2\lp[1+|s|B(M,|s|)\sum_{k=0}^{|s|-1}\frac{a^2(|\omega|+m^2)}{\omega^2}\frac{\lp|A_{k,\,\mc{I}^-}^{[s]}\rp|^2}{\lp|\swei{A}{s}_{\mc{I}^-}\rp|^2}\rp]\lp|\swei{A}{s}_{\mc{I}^-}\rp|^2 \numberthis \label{eq:time-dominated-intermediate-1}\\
&\quad\qquad+2(E+E_W)(\omega-m\upomega_+)^2\lp|\swei{A}{s}_{\mc{H}^-}\rp|^2+\int_{-\infty}^\infty \lp\{2y\Re\lp[\mathfrak{G}\overline{\Psi}'\rp]-E\omega\Im\lp[\mathfrak{G}\overline{\Psi}\rp]+E_W(Q^{W,T})'\rp\}dr^*\,.
\end{align*}
If $s<0$, the analogous argument with $\hat{y}'$ gives
\begin{align*}
&E\omega^2\lp|\swei{A}{s}_{\mc{I}^+}\rp|^2+(\omega-m\upomega_+)^2\lp|\swei{A}{s}_{\mc{H}^+}\rp|^2+\int_{-\infty}^\infty\frac14\lp( \hat{y}'|\Psi|^2+\hat{y}'\omega^2|\Psi|^2\rp)dr^*\\
&\quad\leq 2\lp(2+E+E_W\frac{\mathfrak{C}_s}{\mathfrak{D}_s^{\mc I}}\rp)(\omega-m\upomega_+)^2\lp[1+|s|B(M,|s|)\sum_{k=0}^{|s|-1}\frac{a^2(|\omega|+m^2)}{\omega^2}\frac{\lp|A_{k,\,\mc{H}^-}^{[s]}\rp|^2}{\lp|\swei{A}{s}_{\mc{H}^-}\rp|^2}\rp]\lp|\swei{A}{s}_{\mc{H}^-}\rp|^2\numberthis \label{eq:time-dominated-intermediate-2}\\
&\quad\qquad+(E+E_W)\omega^2\lp|\swei{A}{s}_{\mc{I}^-}\rp|^2+\int_{-\infty}^\infty \lp\{2\hat{y}\Re\lp[\mathfrak{G}\overline{\Psi}'\rp]-E\omega\Im\lp[\mathfrak{G}\overline{\Psi}\rp]+E_W\lp(Q^{W,T}\rp)'\rp\}dr^*\,.
\end{align*}

Note that, since $|\omega|\geq \omega_{\rm high}>1$ and, by Lemma~\ref{lemma:time-dominated-frequency-properties}\ref{it:time-dominated-frequency-properties-m2}, $m^2\ll \omega^2$, as long as
\begin{align*}
\frac{\lp|A_{k,\,\mc{H}^-}^{[s]}\rp|^2}{\lp|\swei{A}{s}_{\mc{H}^-}\rp|^2} \text{~and~}\frac{\mathfrak{C}_s}{\mathfrak{D}_s^{\mc H}} \text{~~if~}s<0\,,\qquad \frac{\lp|A_{k,\,\mc{I}^-}^{[s]}\rp|^2}{\lp|\swei{A}{s}_{\mc{I}^-}\rp|^2} \text{~and~}\frac{\mathfrak{C}_s}{\mathfrak{D}_s^{\mc I}} \text{~~if~}s>0\,,
\end{align*}
are bounded, the right hand side of the estimates \eqref{eq:time-dominated-intermediate-1} and \eqref{eq:time-dominated-intermediate-2} can be replaced by
\begin{align*}
&B(E,E_W)\lp\{\omega^2\lp|\swei{A}{s}_{\mc{I}^-}\rp|^2+B(E,E_W)(\omega-m\upomega_+)^2\lp|\swei{A}{s}_{\mc{H}^-}\rp|^2\rp\}\\
&\qquad+\int_{-\infty}^\infty \lp\{2\hat{y}\Re\lp[\mathfrak{G}\overline{\Psi}'\rp]-E\omega\Im\lp[\mathfrak{G}\overline{\Psi}\rp]+E_W\lp(Q^{W,T}\rp)'\rp\}dr^*\,.
\end{align*}

To conclude, we note that when $\Psi$ arises from a solution to the homogeneous radial ODE~\eqref{eq:radial-ODE-alpha}, then one can easily apply the $T$-type Teukolsky--Starobinsky energy current and conclude, by Lemma~\ref{lemma:time-dominated-potential-properties}\ref{it:time-dominated-frequency-properties-no-superrad}
\begin{align*} 
&\omega^2\frac{\mathfrak{C}_s}{\mathfrak{D}_s^{\mc I}}
\lp|\swei{A}{s}_{\mc{I}^+}\rp|^2+\frac12(\omega-m\upomega_+)^2
\lp|\swei{A}{s}_{\mc{I}^+}\rp|^2\\
&\quad \leq \omega^2\frac{\mathfrak{C}_s}{\mathfrak{D}_s^{\mc I}}
\lp|\swei{A}{s}_{\mc{I}^+}\rp|^2+\omega(\omega-m\upomega_+)
\lp|\swei{A}{s}_{\mc{I}^+}\rp|^2= \omega^2
\lp|\swei{A}{s}_{\mc{I}^-}\rp|^2+\omega(\omega-m\upomega_+)
\frac{\mathfrak{C}_s}{\mathfrak{D}_s^{\mc H}}
\lp|\swei{A}{s}_{\mc{I}^-}\rp|^2\\
&\quad\leq \omega^2
\lp|\swei{A}{s}_{\mc{I}^-}\rp|^2+2(\omega-m\upomega_+)^2
\frac{\mathfrak{C}_s}{\mathfrak{D}_s^{\mc H}}
\lp|\swei{A}{s}_{\mc{I}^-}\rp|^2\,,
\end{align*}
for $s\leq 0$ and similarly for $s>0$. However, in general, one does not have control over the constants $\mathfrak{C}_s/\mathfrak{D}_s^{\mc H}$ and $\mathfrak{C}_s/\mathfrak{D}_s^{\mc I}$ in this range.

\medskip
\noindent \textit{Large comparable frequencies case.} Now let $r_0'$ be as in Lemma~\ref{lemma:V0-trapping} and $\beta_1$ as in Lemma~\ref{lemma:derivative-V0-r+}. Consider the case where $(\omega,m,\Lambda)\in\mc{F}_{\rm comp}(\varepsilon_{\rm width},\omega_{\rm high})$. Let $C,R^*$ be large; by Lemma~\ref{lemma:time-dominated-potential-properties}, we have
\begin{align*}
y'(\omega^2-\mc{V})-y\mc{V}'
&\geq y'(\omega^2-\mc{V})\lp(1-\frac{B(\varepsilon_{\rm width})\Delta}{C r^3}\mathbbm{1}_{[-R^*,R^*]}-\frac{B(\varepsilon_{\rm width})}{\delta (R^*)^2}\mathbbm{1}_{[R^*,\infty)}-\frac{B(\varepsilon_{\rm width})}{R^*}\mathbbm{1}_{(-\infty,R^*]}\rp)\\
&\geq \frac12 y'(\omega^2-\mc{V})\,,\\
\hat y'(\omega^2-\mc{V})-\hat y\mc{V}'&\geq \frac12 \hat y'(\omega^2-\mc{V})\,,
\end{align*}
as long as $C$ and $R^*$ are sufficiently large depending on $\varepsilon_{\rm width}$. As in the time dominated case, we present the proof for $s\geq 0$, where we use the current $y$ and pick up boundary terms at $\mc{I}^-$ when using the estimate from \eqref{eq:time-dominated-basic-estimate-1}; the $s<0$ case is completely analogous, but uses $\hat{y}$ and picks up boundary terms at $\mc{H}^-$ in the aforementioned estimate.

In the $s\geq 0$ case, we consider $Q^y+EQ^T+E_WQ^{W,T}$ currents for some $E,E_W>0$. In the case $(\omega,m,\Lambda)\in\mc{F}_{\rm comp,2a}$, where Lemma~\ref{lemma:basic-estimate-2} is available by Lemma~\ref{lemma:comparable-frequency-properties}\ref{it:comparable-frequency-properties-Lambda-nondegenerate}, we can make use of the improved estimate \eqref{eq:time-dominated-basic-estimate-2} to bound the coupling terms arising due to $Q^y+EQ^T$ currents:
\begin{align*}
&\int_{-\infty}^\infty\frac12\lp( \frac12y'|\Psi|^2+y'(\omega^2-\mc{V})|\Psi|^2\rp)dr^*\\
&\quad\leq \int_{-\infty}^\infty \lp[B(M,|s|)\max\{C^{-1},R^{-1/2},E\varepsilon^{-1},E\}\tilde{\varepsilon}_1\Lambda+E (\epsilon\omega^2+|\omega|)\rp]w\lp|\Psi\rp|^2dr^*\\
&\quad\qquad +|s|\sum_{k=0}^{|s|-1}\int_{-\infty}^\infty \lp|\Re\lp[\mathfrak{G}_{(k)}\uppsi_{(k)}\rp]\rp|dr^*+\int_{-\infty}^\infty \lp\{2y\Re\lp[\mathfrak{G}\overline{\Psi}'\rp]-E\omega\Im\lp[\mathfrak{G}\overline{\Psi}\rp]+E_W\lp(Q^{W,T}\rp)'\rp\}dr^*\\
&\quad\qquad-\omega(\omega-m\upomega_+)\lp(E+E_W\frac{\mathfrak{C}_s}{\mathfrak{D}_s^{\mc H}}\rp)\lp|\swei{A}{s}_{\mc{H}^+}\rp|^2+\omega^2(2-E-E_W)\lp|\swei{A}{s}_{\mc{I}^+}\rp|^2\\
&\quad\qquad+\lp(E+E_W\rp)\omega(\omega-m\upomega_+)\lp|\swei{A}{s}_{\mc{H}^-}\rp|^2+\lp(2+E+E_W\frac{\mathfrak{C}_s}{\mathfrak{D}_s^{\mc I}}\rp)\omega^2\lp|\swei{A}{s}_{\mc{I}^-}\rp|^2 \,.
\end{align*}
The first term on the right hand side can be absorbed into the left hand side by choosing  $\epsilon$ sufficiently small depending on $E$ and by choosing $\omega_{\rm high}$ sufficiently large depending on $\varepsilon_{\rm width}$ and $E$, so that $\tilde\varepsilon_1$ is sufficiently small (see Lemma~\ref{lemma:comparable-frequency-properties}\ref{it:comparable-frequency-properties-Lambda-nondegenerate}). Then, setting $\max\{E,E_W\}\geq 12\varepsilon_{\rm width}^{-1/2}/M$, we obtain
\begin{align*}
&\omega^2\lp|\swei{A}{s}_{\mc{I}^+}\rp|^2+(\omega-m\upomega_+)^2\lp|\swei{A}{s}_{\mc{H}^+}\rp|^2+\int_{-\infty}^\infty\frac14\lp( y'|\Psi|^2+y'(\omega^2-\mc{V})|\Psi|^2\rp)dr^*\\
&\quad\leq |s|\sum_{k=0}^{|s|-1}\int_{-\infty}^\infty \lp|\Re\lp[\mathfrak{G}_{(k)}\uppsi_{(k)}\rp]\rp|dr^*+\int_{-\infty}^\infty \lp\{2y\Re\lp[\mathfrak{G}\overline{\Psi}'\rp]-E\omega\Im\lp[\mathfrak{G}\overline{\Psi}\rp]+E_W\lp(Q^{W,T}\rp)'\rp\}dr^*\\
&\quad\qquad+(E+E_W)\varepsilon_{\rm width}^{-1}(\omega-m\upomega_+)^2\lp|\swei{A}{s}_{\mc{H}^-}\rp|^2+\lp(2+E+\frac53E_W\rp)\omega^2\lp|\swei{A}{s}_{\mc{I}^-}\rp|^2 \,.
\end{align*}

Now let us turn to the case $(\omega,m,\Lambda)\in\mc{F}_{\rm comp,2b}$, or more generally to $(\omega,m,\Lambda)\in\mc{F}_{\rm comp,2}$ as a whole. If $|a|\leq a_0$ is sufficiently small depending on $\varepsilon_{\rm width}$, then this smallness is enough to show that the errors in \eqref{eq:time-dominated-basic-estimate-1} can be absorbed. Otherwise, for general $|a|\leq M$, we can apply Cauchy--Schwarz,
\begin{align*}
&\int_{-\infty}^\infty \lp\{2ay w \Re\lp[\lp(c_{s,s,k}^{\rm id}+imc_{s,s,k}^{\Phi}\rp)\uppsi_{(k)}\overline{\Psi}'\rp]-E aw \Im\lp[\lp(c_{s,s,k}^{\rm id}+imc_{s,k,k}^{\Phi}\rp)\omega\uppsi_{(k)}\overline{\Psi}\rp]\rp\}dr^*\\
&\quad\leq \int_{-\infty}^\infty \lp(\frac18 y'|\Psi'|^2+\frac18 y'(\omega^2-\mc V)|\Psi|^2\rp)dr^*\\
&\quad\qquad+ \int_{-\infty}^\infty B\lp[\lp(\frac{wy}{y'}\rp)^2c(\varepsilon_{\rm width})+E^2\frac{w^2\omega^2(m^2+1)}{(y')^2(\omega^2-\mc V)(\omega^2-\Re\mc V_{(k)})}\rp]y'(\omega^2-\Re\mc V_{(k)})|\uppsi_{(k)}|^2dr^*\,,
\end{align*}
noting that $\Re\mc V_{(k)}$ clearly enjoys the same properties as $\mc V$, see Lemma~\ref{lemma:time-dominated-potential-properties}. While estimate \eqref{eq:time-dominated-basic-estimate-2} is unavailable, we can instead consider a hierarchy of currents as in Section~\ref{sec:bounded-smallness}: applying the same current $y$ to the equations \eqref{eq:transformed-k-separated} with $k<|s|$, we deduce that 
\begin{align*}
&\int_{-\infty}^\infty \frac{1}{4}\lp(y'|\uppsi_{(k)}'|^2+y'(\omega^2-\Re\mc V_{(k)})|\uppsi_{(k)}|^2\rp)dr^*\\
&\quad \leq B(\varepsilon_{\rm width})\int_{-\infty}^\infty \frac{w}{r}(|s|-k)|\uppsi_{(k+1)}|^2dr^*+ B\sum_{j=0}^{k-1}\int_{-\infty}^\infty \lp(\frac{wy}{y'}\rp)^2c(\varepsilon_{\rm width}) y'(\omega^2-\Re\mc V_{(j)})|\uppsi_{(j)}|^2 dr^* \\
&\quad\qquad  + \int_{-\infty}^\infty 2y\Re[\mathfrak{G}_{(k)}\overline{\uppsi_{(k)}}']dr^*+ 2\omega^2|\swei{A}{s}_{k,\mc I^-}|^2 \\
&\quad\leq \frac{B(\varepsilon_{\rm width})}{\omega^2_{\rm high}}\int_{-\infty}^\infty y'(\omega^2-\Re\mc V_{(k+1)})(|s|-k)|\uppsi_{(k+1)}|^2dr^*\\
&\quad\qquad + B(\varepsilon_{\rm width})\sum_{j=0}^k\lp(\omega^2|\swei{A}{s}_{j,\mc I^-}|^2 +\int_{-\infty}^\infty y\Re[\mathfrak{G}_{(j)}\overline{\uppsi_{(j)}}']dr^*\rp)
\end{align*}
for sufficiently large $R^*$ and $C$. Notice that we have appealed to Lemma~\ref{lemma:h-y-identity-k} and applied Cauchy--Schwarz to absorb terms involving $\uppsi_k$ to the left hand side. Now, fix some $\max\{E,E_W\}\geq 12\varepsilon_{\rm width}^{-1/2}/M$; by then choosing $\omega_{\rm high}$ sufficiently high depending on $R^*$, $C$ and $\varepsilon_{\rm width}$, we deduce that 
\begin{align*}
&\int_{-\infty}^\infty B\lp[\lp(\frac{wy}{y'}\rp)^2+E^2\frac{w^2\omega^2(m^2+1)}{(y')^2(\omega^2-\mc V)^2}\rp]y'(\omega^2-\Re\mc V_{(k)})|\uppsi_{(k)}|^2dr^*\\
&\quad  \leq \frac{B(\varepsilon_{\rm width},R^*,C)}{\omega^2_{\rm high}}\int_{-\infty}^\infty y'(\omega^2-\mc V)|\Psi|^2dr^*+ B(\varepsilon_{\rm width})\sum_{j=0}^k\lp(\omega^2|\swei{A}{s}_{j,\mc I^-}|^2 +\int_{-\infty}^\infty y\Re[\mathfrak{G}_{(j)}\overline{\uppsi_{(j)}}']dr^*\rp)  \\
&\quad  \leq \int_{-\infty}^\infty \frac18 y'(\omega^2-\mc V)|\Psi|^2dr^*+ B(\varepsilon_{\rm width})\sum_{j=0}^k\lp(\omega^2|\swei{A}{s}_{j,\mc I^-}|^2 +\int_{-\infty}^\infty y\Re[\mathfrak{G}_{(j)}\overline{\uppsi_{(j)}}']dr^*\rp)\,.
\end{align*}
and hence we obtain
\begin{align*}
&\omega^2\lp|\swei{A}{s}_{\mc{I}^+}\rp|^2+(\omega-m\upomega_+)^2\lp|\swei{A}{s}_{\mc{H}^+}\rp|^2+\int_{-\infty}^\infty\frac14\lp( y'|\Psi|^2+y'(\omega^2-\mc{V})|\Psi|^2\rp)dr^*\\
&\quad\leq \int_{-\infty}^\infty \lp\{\sum_{k=0}^{|s|-1}B(\varepsilon_{\rm width})y\Re\lp[\mathfrak{G}_{(k)}\overline{\uppsi_{(k)}}'\rp]+2y\Re\lp[\mathfrak{G}\overline{\Psi}'\rp]-E\omega\Im\lp[\mathfrak{G}\overline{\Psi}\rp]+E_W\lp(Q^{W,T}\rp)'\rp\}dr^*\\
&\quad\qquad+(E+E_W)\varepsilon_{\rm width}^{-1}(\omega-m\upomega_+)^2\lp|\swei{A}{s}_{\mc{H}^-}\rp|^2+\omega^2\lp(2+E+E_W\frac{\mathfrak C_s}{\mathfrak D_s^{\mc I}}+B(\varepsilon_{\rm width})\sum_{j=0}^{|s|-1}\lp|\frac{\swei{A}{s}_{k,\mc I^+}}{\swei{A}{s}_{\mc I^+}}\rp|\rp)\lp|\swei{A}{s}_{\mc{I}^-}\rp|^2 \,.
\end{align*}
Note that, in $(\omega,m,\Lambda)\in\mc{F}_{\rm comp,2b}$, the coefficients
\begin{align*}
\frac{\lp|A_{k,\,\mc{I}^-}^{[s]}\rp|^2}{\lp|\swei{A}{s}_{\mc{I}^-}\rp|^2}
\end{align*} 
may not be bounded, hence the restriction in Proposition~\ref{prop:comp2} to outgoing boundary conditions. 

To conclude, we note once again that when $\Psi$ arises from a solution to the inhomogeneous radial ODE~\eqref{eq:radial-ODE-alpha}, then one can easily apply the $T$-type Teukolsky--Starobinsky energy current and conclude, by Lemma~\ref{lemma:comparable-frequency-properties}\ref{it:comparable-frequency-properties-no-superrad}
\begin{align*} 
&\omega^2\frac{\mathfrak{C}_s}{\mathfrak{D}_s^{\mc I}}
\lp|\swei{A}{s}_{\mc{I}^+}\rp|^2+2\varepsilon_{\rm width}^{1/2}/M(\omega-m\upomega_+)^2
\lp|\swei{A}{s}_{\mc{I}^+}\rp|^2\\
&\quad \leq \omega^2\frac{\mathfrak{C}_s}{\mathfrak{D}_s^{\mc I}}
\lp|\swei{A}{s}_{\mc{I}^+}\rp|^2+\omega(\omega-m\upomega_+)
\lp|\swei{A}{s}_{\mc{I}^+}\rp|^2= \omega^2
\lp|\swei{A}{s}_{\mc{I}^-}\rp|^2+\omega(\omega-m\upomega_+)
\frac{\mathfrak{C}_s}{\mathfrak{D}_s^{\mc H}}
\lp|\swei{A}{s}_{\mc{I}^-}\rp|^2\\
&\quad\leq \omega^2
\lp|\swei{A}{s}_{\mc{I}^-}\rp|^2+\varepsilon_{\rm width}^{-1}(\omega-m\upomega_+)^2
\frac{\mathfrak{C}_s}{\mathfrak{D}_s^{\mc H}}
\lp|\swei{A}{s}_{\mc{I}^-}\rp|^2\,,
\end{align*}
for $s\leq 0$ and similarly for $s>0$, in the entire range $\mc{F}_{\rm comp,2}$. However, in general in this frequency range, one does not have control over the constants $\mathfrak{C}_s/\mathfrak{D}_s^{\mc H}$ and $\mathfrak{C}_s/\mathfrak{D}_s^{\mc I}$.
\end{proof}

\subsubsection{Angular dominated and large, non-trapped superradiant regimes}
\label{section:F-angular-superradiant}

In this section, we consider the frequency ranges $\mc{F}_{\sun}$ and $\mc{F}_{\measuredangle}\backslash\mc{F}_{\measuredangle, 3}$. We will show
\begin{proposition}[Estimates in $\mc{F}_{\measuredangle,1}\cup\mc{F}_{\measuredangle,2}$] \label{prop:angular-dominated} Fix $s\in\{0,\pm 1,\pm 2\}$, $M>0$ and $a_0\in[0,M)$. Then, for all $\delta\in(0,1]$,  for all $\beta_1=\beta_1(a_0)>0$ and $\beta_2>0$ such that Lemma~\ref{lemma:derivative-V0-r+} hold, for all $E,E_W$ such that one of these is sufficiently large, for all $\varepsilon^{-1}_{\rm width}$ sufficiently large depending on $a_0$, $\beta_1$, $\beta_2$, for all  $\omega_{\rm high}$ sufficiently large depending on $a_0$, $\beta_1$, $\beta_2$, $E$ and $E_W$, and for all  $(a,\omega,m,\Lambda) \in \mc{F}_{\measuredangle,1}(\omega_{\rm high},\varepsilon_{\rm width},a_0,\beta_1)\cup\mc{F}_{\measuredangle,1}(\omega_{\rm high},\varepsilon_{\rm width},a_0,\beta_2)$,  there exist functions $f$, $y$, $\hat{y}$, $h$ and $\chi_1,\chi_2$ satisfying the uniform bounds
$|f| +|y|+|\hat{y}|+|h|+|\chi_1|+|\chi_2|\leq B(\varepsilon_{\rm width},\omega_{\rm high}) \,,$
such that, for all smooth $\Psi$ arising from a smooth solution to the radial ODE~\eqref{eq:radial-ODE-alpha} via \eqref{eq:def-psi0-separated} and \eqref{eq:transformed-transport-separated} and itself satisfying the radial ODE~\eqref{eq:radial-ODE-Psi}, if $\Psi$ has the general boundary conditions of Lemma~\ref{lemma:uppsi-general-asymptotics} or the outgoing boundary conditions of Definition \ref{def:outgoing-bdry-uppsi}, we have the estimates
\begin{align*}
&\omega^2\lp|\swei{A}{s}_{\mc{I}^+}\rp|^2+(\omega-m\upomega_+)^2\lp|\swei{A}{s}_{\mc{H}^+}\rp|^2 + b(\delta)\int_{-\infty}^\infty \frac{\Delta}{r^{3+\delta}}\Big[|\Psi'|^2+\lp(\frac{\Lambda}{r^2}+\omega^2\rp)|\Psi|^2\Big]dr^*\\
&\quad\leq B(E,E_W)\lp[\omega^2\lp|\swei{A}{s}_{\mc{I}^-}\rp|^2+(\omega-m\upomega_+)^2\lp|\swei{A}{s}_{\mc{H}^-}\rp|^2\rp]-E\int_{-\infty}^\infty\lp[\omega\chi_2+(\omega-m\upomega_+)\chi_1\rp]\Im\lp[\mathfrak{G}\overline{\Psi}\rp]dr^*\\
&\qquad\quad+\int_{-\infty}^\infty \lp\{E_W\lp[\chi_1\lp(Q^{W,K}\rp)'+\chi_2\lp(Q^{W,T}\rp)'\rp]+ 2(y+\hat{y}+f)\Re\lp[\mathfrak{G}\overline{\Psi}'\rp]+(h+f')\Re\lp[\mathfrak{G}\overline{\Psi}\rp]\rp\}dr^*\\
&\quad\qquad +|s|\sum_{k=0}^{|s|-1}\int_{-\infty}^\infty\lp|\Re\lp[\mathfrak{G}_{(k)}\uppsi_{(k)}\rp]\rp|\,,
\end{align*}
and
\begin{align*}
\lp|\frac{\mathfrak{D}_s^\mc{H}}{\mathfrak{D}_s^\mc{I}}\rp|+\lp|\frac{\mathfrak{D}_s^\mc{H}}{\mathfrak{D}_s^\mc{I}}\rp|^{-1}+\lp|\frac{\mathfrak{C}_s}{\mathfrak{D}_s^\mc{I}}\rp|+\lp|\frac{\mathfrak{C}_s}{\mathfrak{D}_s^\mc{I}}\rp|^{-1}+\lp|\frac{\mathfrak{C}_s}{\mathfrak{D}_s^\mc{H}}\rp|+\lp|\frac{\mathfrak{C}_s}{\mathfrak{D}_s^\mc{H}}\rp|^{-1}\leq B\,.
\end{align*}
\end{proposition}

\begin{proposition}[Estimates in $\mc{F}_{\sun,1}\cup\mc{F}_{\sun,2}$] \label{prop:superradiant} Fix $s\in\{0,\pm 1,\pm 2\}$, $M>0$ and $a_0\in[0,M)$. Then, for all $\delta\in(0,1]$,  for all $\beta_1=\beta_1(a_0)>0$ and $\beta_2>0$ such that Lemma~\ref{lemma:derivative-V0-r+} hold, for all $E,E_W>0$ such that one of these is sufficiently large, for all $\varepsilon^{-1}_{\rm width}$ sufficiently large depending on $a_0$, $\beta_1$, $\beta_2$, for all $\omega_{\rm high}$ sufficiently large depending on $\varepsilon_{\rm width}$, $a_0$, $\beta_1$, $\beta_2$, $E$ and $E_W$, and for all  $(a,\omega,m,\Lambda) \in \mc{F}_{\sun,1}(\omega_{\rm high},\varepsilon_{\rm width},a_0,\beta_1)\cup\mc{F}_{\sun,1}(\omega_{\rm high},\varepsilon_{\rm width},a_0,\beta_2)$,  there exist functions $f$, $y$, $\hat{y}$, $h$ and $\chi_1,\chi_2$ satisfying the uniform bounds
$|f| +|y|+|\hat{y}|+|h|+|\chi_1|+|\chi_2|\leq B(\varepsilon_{\rm width},\omega_{\rm high}) \,,$
such that, for all smooth $\Psi$ arising from a smooth solution to the radial ODE~\eqref{eq:radial-ODE-alpha} via \eqref{eq:def-psi0-separated} and \eqref{eq:transformed-transport-separated} and itself satisfying the radial ODE~\eqref{eq:radial-ODE-Psi}, if $\Psi$ has the general boundary conditions of Lemma~\ref{lemma:uppsi-general-asymptotics} or the outgoing boundary conditions of Definition \ref{def:outgoing-bdry-uppsi}, we have the estimates
\begin{align*}
&\omega^2\lp|\swei{A}{s}_{\mc{I}^+}\rp|^2+(\omega-m\upomega_+)^2\lp|\swei{A}{s}_{\mc{H}^+}\rp|^2 + b(\delta)\int_{-\infty}^\infty \frac{\Delta}{r^{3+\delta}}\Big[|\Psi'|^2+\lp(\frac{\Lambda}{r^2}+\omega^2\rp)|\Psi|^2\Big]dr^*\\
&\quad\leq B(E,E_W)\lp[\omega^2\lp|\swei{A}{s}_{\mc{I}^-}\rp|^2+(\omega-m\upomega_+)^2\lp|\swei{A}{s}_{\mc{H}^-}\rp|^2\rp]-E\int_{-\infty}^\infty\lp[\omega\chi_2+(\omega-m\upomega_+)\chi_1\rp]\Im\lp[\mathfrak{G}\overline{\Psi}\rp]dr^*\\
&\qquad\quad+\int_{-\infty}^\infty \lp\{E_W\lp[\chi_1\lp(Q^{W,K}\rp)'+\chi_2\lp(Q^{W,T}\rp)'\rp]+ 2(y+\hat{y}+f)\Re\lp[\mathfrak{G}\overline{\Psi}'\rp]+(h+f')\Re\lp[\mathfrak{G}\overline{\Psi}\rp]\rp\}dr^*\\
&\quad\qquad +|s|\sum_{k=0}^{|s|-1}\int_{-\infty}^\infty\lp|\Re\lp[\mathfrak G_{(k)}\uppsi_{(k)}\rp]\rp|\,,
\end{align*}
and 
\begin{align*}
\lp|\frac{\mathfrak{D}_s^\mc{H}}{\mathfrak{D}_s^\mc{I}}\rp|+\lp|\frac{\mathfrak{D}_s^\mc{H}}{\mathfrak{D}_s^\mc{I}}\rp|^{-1}+\lp|\frac{\mathfrak{C}_s}{\mathfrak{D}_s^\mc{I}}\rp|+\lp|\frac{\mathfrak{C}_s}{\mathfrak{D}_s^\mc{I}}\rp|^{-1}+\lp|\frac{\mathfrak{C}_s}{\mathfrak{D}_s^\mc{H}}\rp|+\lp|\frac{\mathfrak{C}_s}{\mathfrak{D}_s^\mc{H}}\rp|^{-1}\leq B\,.
\end{align*}
\end{proposition}

Let us first discuss the specific properties of the frequency parameters in this range, then proceed to do a careful analysis of the potential. The section concludes with the proof of Theorems~\ref{thm:ode-estimates-Psi-A} and \ref{thm:ode-estimates-Psi-B} for frequency triples in $\mc{F}_{\sun}$ or $\mc{F}_{\measuredangle}\backslash\mc{F}_{\measuredangle,3}$.

\begin{lemma}[Properties of the frequency parameters in $\mc{F}_{\sun}$ and $\mc{F}_{\measuredangle}$] \label{lemma:angular-dominated-frequency-properties} Fix $s\in\mathbb{Z}$, $M>0$ and $a_0\in[0,M)$. Assume $a_0\beta_1<\sqrt{M^2-a_0^2}/(6M)$. Let $(\omega,m,\Lambda)\in \mc{F}_{\text{\normalfont\sun}}\cup\mc{F}_{\rm \measuredangle}$ be an admissible frequency triple with respect to $a\in[0,M]$ such that $a\in[0,a_0]$ if $(\omega,m,\Lambda)\in \mc{F}_{\text{\normalfont\sun,1}}$ and $a\in(a_0,M]$ if $(\omega,m,\Lambda)\in \mc{F}_{\text{\normalfont\sun,2}}$. Then, for sufficiently small $\varepsilon_{\rm width}$ and $\beta_1$,
\begin{enumerate}[label=(\roman*)]
\item $\Lambda\geq m^2$ in $\mc{F}_{\measuredangle}$ and $\mc{F}_{\text{\normalfont\sun,2}}$, but $\Lambda\geq \max\lp\{\frac23,\frac{a_0^2}{M^2}\rp\}$ in $\mc{F}_{\text{\normalfont\sun,1}}$; \label{it:angular-dominated-frequency-properties-Lambda-m2}
\item the conditions for Lemma~\ref{lemma:basic-estimate-2} are satisfied, namely  \label{it:angular-dominated-frequency-properties-Lambda-nondegenerate}
\begin{itemize}
\item in $\mc{F}_{\measuredangle}$, we have $\Lambda\geq \Lambda_{\rm high}:=\varepsilon_{\rm width}^{-1}\omega_{\rm high}^2\gg 1$  and, for $\varepsilon_1=\varepsilon_2=1/8$,
\begin{align*}
\Lambda-2am\omega-\varepsilon_1\frac{a^2}{2Mr_+}m^2-\varepsilon_2\Lambda_{\rm high}\geq\frac12 \Lambda\,, \quad \tilde{\varepsilon}_1=\frac{16}{\varepsilon_1\Lambda_{\rm high}}\ll 1\,;
\end{align*}
\item in $\mc{F}_{\text{\normalfont \sun, 1}}$, we have $\Lambda\geq \Lambda_{\rm high}:=\varepsilon_{\rm width}\omega_{\rm high}^2\gg 1$  and, for $\varepsilon_1=\frac{\sqrt{M^2-a_0}}{3M}\varepsilon_{\rm width}$,
\begin{align*}
\Lambda-2am\omega-\varepsilon_1\lp(\omega^2+m^2\rp)\geq \frac{M^2-a_0}{12M^2} \Lambda\,, \quad \tilde{\varepsilon}_1=\frac{432 M^4}{(M^2-a_0^2)^4\varepsilon_{\rm width}^2\omega_{\rm high}^2}\ll 1\,,
\end{align*}
as long as $\omega_{\rm high}$ is sufficiently large depending on $\varepsilon_{\rm width}$ and $a_0$;

\item if $a_0>0$, in $\mc{F}_{\text{\normalfont \sun, 2}}$, we have $\Lambda\geq \Lambda_{\rm high}:=\varepsilon_{\rm width}\omega_{\rm high}^2\gg 1$  and, for $\varepsilon_1=a_0\beta_2/4\varepsilon_{\rm width}$,
\begin{align*}
\Lambda-2am\omega-\varepsilon_1\lp(\omega^2+m^2\rp)\geq a_0\beta_2 \Lambda\,, \quad \tilde{\varepsilon}_1=\frac{4}{a_0^2\beta_2^2\varepsilon_{\rm width}^2\omega_{\rm high}^2}\ll 1\,,
\end{align*}
as long as $\omega_{\rm high}$ is sufficiently large depending on $\varepsilon_{\rm width}$, $\beta_2$ and $a_0$.
\end{itemize}

\item for $|s|=1,2$, we have, for $\omega_{\rm high}$ sufficiently large (depending on $\varepsilon_{\rm width}$, $a_0$ and $\beta_2$ in the case of $\mc{F}_{\sun}$), \label{it:angular-dominated-frequency-properties-Cs-Ds}
\begin{align*}
\frac{\mathfrak{C}_s}{\mathfrak{D}_s^{\mc{I}}}\,,\,\, \frac{\mathfrak{C}_s}{\mathfrak{D}_s^{\mc{H}}}\,, \,\, \frac{\mathfrak{D}_s^{H}}{\mathfrak{D}_s^{\mc{I}}} \in\lp[\frac13,\frac53\rp]\,.
\end{align*}
\end{enumerate}
\end{lemma}
\begin{proof}
We start with \ref{it:angular-dominated-frequency-properties-Lambda-m2}. For $\mc{F}_{\text{\normalfont sun}}$, the bound $\Lambda\geq \max\lp\{\frac23,\frac{a_0^2}{M^2}\rp\}$ is shown already in the the proof of Lemma~\ref{lemma:comparable-frequency-properties}\ref{it:comparable-frequency-properties-Lambda-m2}, as well as the improvement that $\Lambda\geq m^2$ if $m^2\leq \varepsilon_{\rm width}\omega^2$ or if $m^2\geq \varepsilon_{\rm high}^{-1}\omega^2$. In the subrange $\mc{F}_{\text{\normalfont sun,2}}\cap \{\varepsilon_{\rm width}\omega^2<m^2<\varepsilon_{\rm high}^{-1}\omega^2\}$, since
\begin{align*}
a_0\sqrt{\varepsilon_{\rm width}}<q:=\frac{a\omega}{m}=a\upomega_+-\frac{a\beta_2\Lambda}{m^2}\leq \frac{a^2}{2Mr_+}\leq \frac{1}{2}\,,
\end{align*}
 by item \ref{it:admissible-freqs-triple-Lambda-q-bound} of the admissibility conditions on $\Lambda$ laid out in Definition~\ref{def:admissible-freqs}, we conclude that
\begin{align*}
\Lambda\geq 2am\omega+m^2(q-1)^2=m^2(q^2+1) \geq m^2\,.
\end{align*}

We now prove \ref{it:angular-dominated-frequency-properties-Lambda-m2} for $\mc{F}_{\measuredangle}$, i.e.\ when $\Lambda\geq \max\{\varepsilon_{\rm width}^{-1}\omega^2, \varepsilon_{\rm width}^{-1}\omega_{\rm high}^2\}$. Clearly, the conclusion follows if $m^2\leq \varepsilon_{\rm width}^{-1}\omega^2_{\rm high}$ or if $m^2\leq \varepsilon_{\rm width}^{-1}\omega^2$. The remaining cases can be delt with using condition \ref{it:admissible-freqs-triple-Lambda-lower-bound} from Definition~\ref{def:admissible-freqs}: as long as $\varepsilon_{\rm width}^{-1}$ and $\omega_{\rm high}$ sufficiently large, if $m^2>\varepsilon_{\rm width}^{-1}\omega^2$ and $|\omega|\geq \omega_{\rm high}$,
\begin{align*}
\Lambda-m^2\geq \max\{|m|,|s|\}-s^2-2|s||a\omega|\geq (\varepsilon_{\rm width}^{-1/2}-2|s|M)|\omega|-s^2\geq 0\,;
\end{align*}
if $m^2>\varepsilon_{\rm width}^{-1}\omega_{\rm high}^2$ and $|\omega|< \omega_{\rm high}$,
\begin{align*}
\Lambda-m^2\geq \max\{|m|,|s|\}-s^2-2|s||a\omega|\geq (\varepsilon_{\rm width}^{-1/2}-2|s|M)\omega_{\rm high}-s^2\geq 0\,.
\end{align*}  

For \ref{it:angular-dominated-frequency-properties-Lambda-nondegenerate}, first consider $\mc{F}_\measuredangle$, in which
\begin{align*}
\Lambda-2am\omega -\varepsilon_1m^2-\varepsilon_2\Lambda_{\rm high} &\geq \Lambda\lp(1-2M\varepsilon_{\rm width}^{1/2}\Lambda-\varepsilon_1-\varepsilon_2\rp)\Lambda\geq \frac12 \Lambda\,,
\end{align*}
for $\varepsilon_1,\varepsilon_2\in(0,1/4]$ and $\varepsilon_{\rm width}\leq 1$. In $\mc{F}_{\text{\normalfont \sun,1}}$,
\begin{align*}
\Lambda-2am\omega -\varepsilon_1(\omega^2+m^2)&\geq \lp(1-\frac{a_0^2}{Mr_+}\frac{|m|}{\sqrt{\Lambda}}-2a_0\beta_1-\varepsilon_1\lp(\varepsilon_{\rm width}^{-1}+\frac32\rp)\rp)\Lambda \\
&\geq \lp(1-\frac{M}{r_+}-\frac13\frac{M^2-a_0^2}{M^2}-\varepsilon_1(\varepsilon_{\rm width}^{-1}+2)\rp)\Lambda \geq \frac{M^2-a_0^2}{12M^2}\Lambda\,,
\end{align*}
where we have used $a_0\beta_1<(M^2-a_0^2)/(6M^2)$, as long as $\varepsilon_1\leq \frac{M^2-a_0^2}{36M^2}\varepsilon_{\rm width}$. Finally, in $\mc{F}_{\text{\normalfont \sun,2}}$, if $a_0>0$
\begin{align*}
\Lambda-2am\omega-\varepsilon_1(\omega^2+m^2)\geq \Lambda\lp[1-\frac{a^2}{Mr_+}+2a\beta_2-\varepsilon_1(\varepsilon_{\rm width}^{-1}+1)\rp]\geq a_0\beta_2\Lambda\,,
\end{align*}
as long as $\varepsilon_1\leq a_0\beta_2\varepsilon_{\rm width}/4$.

Finally, we show \ref{it:angular-dominated-frequency-properties-Cs-Ds} regarding the properties of the constants $\mathfrak{C}_s$, $\mathfrak{D}_s^\mc{I}$ and $\mathfrak{D}_s^\mc{H}$, which depend only on $|s|$, $a$ and the frequency triple $(\omega,m,\Lambda)$. We begin with $|s|=1$:
\begin{align*}
\lp|\frac{\mathfrak{C}_1}{\mathfrak{D}_1^\mc{I}}-1\rp| &=\lp|\frac{4a\omega(m-a\omega)}{(\Lambda-2am\omega+1)^2}\rp|\leq 4\lp(M^2\frac{\omega^2}{\Lambda}+\sqrt{\frac32}M\frac{|\omega|}{\sqrt{\Lambda}}\rp)\frac{\Lambda}{(\Lambda-2am\omega+1)^2},,\\
\lp|\frac{\mathfrak{D}_1^{\mc{H}}}{\mathfrak{D}_1^\mc{I}}-1\rp| &=\lp|\frac{a^2m^2/M^2}{(\Lambda-2am\omega+1)^2}\rp|\leq \frac 32 \frac{\Lambda}{(\Lambda-2am\omega+1)^2}\,,\\
\lp|\frac{\mathfrak{C}_1}{\mathfrak{D}_1^\mc{H}}-1\rp| &=\lp|\frac{4a\omega(m-a\omega)-a^2m^2/M^2}{(\Lambda-2am\omega+1)^2+a^2m^2/M^2}\rp|\leq \lp(\frac32 +4M^2\frac{\omega^2}{\Lambda}+4\sqrt{\frac32}M\frac{|\omega|}{\sqrt{\Lambda}}\rp) \frac{\Lambda}{(\Lambda-2am\omega+1)^2}\,,
\end{align*}
and the results follow from using the relation between $\omega$ and $\Lambda$ in each of the frequency ranges,  \ref{it:angular-dominated-frequency-properties-Lambda-nondegenerate} and the largeness of $\omega_{\rm high}$ compared to any other parameters involved in the definition of the ranges. A similar argument can be used to obtain the result in the case $|s|=2$.
\end{proof}

By Lemmas~\ref{lemma:critical-points-V0} and \ref{lemma:derivative-V0-r+},  we conclude that $\mc{V}_0$ has a unique critical point, located at $r^0_{\rm max}$, which is a maximum. With the uniform bounds for $\sigma:=am\omega/(\Lambda+s^2)$  in these regimes,
\begin{align*}
|\sigma|&\leq M \frac{|\omega|}{\sqrt{\Lambda}}\leq M\epsilon_{\rm width}^{1/2}\leq 1 \text{~ in ~} \mc{F}_{\measuredangle}\,, \\
\sigma&\leq \frac{a(m^2\upomega_++\beta_2\Lambda)}{\Lambda}\leq \frac{a^2}{2Mr_+}\frac{m^2}{\Lambda}+a\beta_1\leq \frac34 +\frac16\leq 1\text{~ in ~} \mc{F}_{\sun,1}\,,\\
\sigma&\leq \frac{a(m^2\upomega_++\beta_2\Lambda)}{\Lambda}\leq \frac{a^2}{2Mr_+}\frac{m^2}{\Lambda}\leq \frac12 \text{~ in ~} \mc{F}_{\sun,2}\,,
\end{align*}
we conclude, by statement 3 in Lemma~\ref{lemma:critical-points-V0} and the previous bounds, that $r_{\rm max}^0\leq B$ independently of the frequency parameters. Moreover, the bounds on $\sigma$ provide the uniform bounds
\begin{align}
\lp|\frac{d\mc{V}_0}{dr}\rp|+r\lp|\frac{d^2\mc{V}_0}{dr^2}\rp| \leq \frac{B\Lambda}{r^3}\,, \qquad \lp|\frac{d}{dr}\lp((r^2+a^2)^3\frac{d\mc{V}_0}{dr}\rp)\rp|\leq B\Lambda r^2\,, \label{eq:ang-dominated-V0-bounds}
\end{align}
which show $r^0_{\rm max}-r_+\geq b$ independently of the frequency parameters. 

Moreover, Lemma~\ref{lemma:V0-r0max} demonstrates the the maximum of $\mc{V}_0$ is quantitatively above the trapping energy $\omega^2$. Let $c$ (independent of the frequency parameters in $\mc{F}_{\measuredangle}$, depending on $\beta_1$ and $a_0$ in $\mc{F}_{\sun,1}$, and depending on $\beta_2$ and $a_0$ in $\mc{F}_{\sun,2}$) be the such that
\begin{align*}
\mc{V}_0(r^0_{\rm max})-\omega^2\geq c\Lambda\,
\end{align*}
We conclude from bounds~\eqref{eq:ang-dominated-V0-bounds} that there is some $\delta_1>0$ with the same dependences as $c$ such that 
\begin{align*}
\mc{V}_0(r)-\omega^2\geq \frac{c}{2}\Lambda\,,\quad \forall\, r\in[r_{\rm max}^0-\delta_1,r_{\rm max}^0+\delta_1]
\end{align*}
In the following lemma, which is a generalization of \cite[Lemma 8.3.1]{Dafermos2016b}, we show that a similar statement holds for the full potential $\mc{V}:=\mc{V}_0+\mc{V}_1$, as $\mc{V}_0$ dominates over $\mc{V}_1$, so that the latter cannot bring about a new critical point or drastically change the location of or the value of the maximum.

\begin{lemma}[Behavior of the potential in  $\mc{F}_{\sun}$ and $\mc{F}_{\measuredangle}$] \label{lemma:V-angular-dominated} 
Fix $s\in\mathbb{Z}$, $M>0$, $a_0\in[0,M)$. Let $\beta_1>0$ be small enough that Lemma~\ref{lemma:derivative-V0-r+} holds.  For all sufficiently small $\epsilon_{\rm width}$ depending on $\beta_1$ and $\beta_2$, for all sufficiently large $\omega_{\rm high}$ depending on $\epsilon_{\rm width}$,  for all  $(\omega,m,\Lambda)\in\mc{F}_{\mathrm{\measuredangle}}(\epsilon_{\rm width},\omega_{\rm high},a_0,\beta_1,\beta_2)\cup\mc{F}_{\text{\normalfont \sun}}(\epsilon_{\rm width},\omega_{\rm high},a_0,\beta_1,\beta_2)$  satisfying one of the following three
\begin{align*}
(i)&~0<m\omega<m^2\upomega_++\beta_1\Lambda\,, \,\, 0\leq a\leq a_0\,,\\
(ii)&~0<m\omega<m^2\upomega_+-\beta_2\Lambda\,, \,\, a_0<a\leq M\,,\\
(iii)&~m\omega\leq 0\,,
\end{align*}
the potential $\mc{V}$ has a unique critical point $r_{\rm max}$ that satisfies $|r_{\rm max}-r_{\rm max}^0|\leq B\Lambda^{-1}$ and
\begin{align*}
\mc{V}(r)-\omega^2 \geq b\Lambda \,, &\quad\forall\, r\in(r_{\rm max}-\tilde\delta, r_{\rm max}+\tilde\delta)\,,
\end{align*}
for some $\tilde{\delta}$ and $b$ which depend on $a_0$ and $\beta_1$ in case (i), depend on $a_0$ and $\beta_2$ in case (ii) and are independent of the frequency parameters in case (iii). Moreover, in cases (i) and (ii), there is a constant $b=b(a_0,\beta_1)$ in case (i) and $b=b(a_0,\beta_2)$ in case (ii) such that
\begin{align*}
-(r-r_{\rm max})\frac{d\mc{V}}{dr} \geq b\Lambda \frac{(r-r_{\rm max})^2}{r^4}\,, &\quad \forall\, r\in[r_+,\infty)\,;
\end{align*}
in case (iii), independently of the frequency parameters, we have
\begin{align*}
-(r-r_{\rm max})\frac{d\mc{V}}{dr} \geq b(r-M)\Lambda \frac{(r-r_{\rm max})^2}{r^4}\,, &\quad \forall\, r\in[r_+,\infty)\,.
\end{align*}

\end{lemma}
\begin{proof} We begin by refining our estimates on $\mc{V}_0$: we want to show that there exists an $r_1'\in(r_+,r_{\rm max}^0)$ such that
\begin{align}
\begin{split} \label{eq:ang-dominated-r1'}
\frac{d\mc{V}_0}{dr}(r) \geq \frac{\hat{c}}{2}\frac{(r_+^2+a^2)^3}{(r_1^2+a^2)^3} \Lambda\,, \quad \forall\, r\in[r_+,r_1']\,, \\
\frac{d}{dr}\lp((r^2+a^2)^3\frac{d\mc{V}_0}{dr}\rp)\leq -\tilde{c}\Lambda r^2\,, \quad \forall\, r\in[r_1',\infty)\,,
\end{split}
\end{align}
for some $\hat{c}>0$ possibly depending on $a_0,\beta_1,\beta_2$ as in the statement and $\tilde{c}$ independent of the frequency parameters. Recall that, in \cref{lemma:critical-points-V0}, it was seen that $\frac{d}{dr}\lp((r^2+a^2)^3\frac{d\mc{V}_0}{dr}\rp)$ is either non-positive on $[r_+,\infty)$ or there exists a unique point $r_1\in[r_+,r_{\rm max})$ at which it switches from positive to negative. We have two cases
\begin{enumerate}[label=(\alph*)]
\item \textit{$r_1$ exists.} In this case, $(r^2+a^2)^3\frac{d\mc{V}_0}{dr}$ is increasing up to $r_1$, hence
\begin{align*}
\frac{d\mc{V}_0}{dr}(r) \geq \frac{(r_+^2+a^2)^3}{(r_1^2+a^2)^3} \frac{d\mc{V}_0}{dr}(r_+) \geq \frac{\hat{c}}{2}\frac{(r_+^2+a^3)^3}{(r_1^2+a^2)^3} \Lambda\,, \quad \forall\, r\in[r_+,r_1]\,,
\end{align*}
where the final bound comes from \cref{lemma:derivative-V0-r+} so $\hat{c}>0$ inherits the dependence on other parameters from it. Thus, by the bound on $\frac{d}{dr}\lp((r^2+a^2)^3\frac{d\mc{V}_0}{dr}\rp)$ from before and the fact that this is negative in $(r_1,\infty)$ by definition of $r_1$, there must be some $r_1'\in(r_1,r_{\rm max}^0)$ such that \eqref{eq:ang-dominated-r1'} holds.
\item \textit{$r_1$ does not exist.} In this case, the second line of \eqref{eq:ang-dominated-r1'} holds for all $r$ by Lemma~\ref{lemma:critical-points-V0} and the bounds \eqref{eq:ang-dominated-V0-bounds}. Moreover, $(r^2+a^2)^3\frac{d\mc{V}_0}{dr}$ is decreasing until it attains the value 0 at $r^0_{\rm max}$, so there must be $r_1'$ such that the first line of \eqref{eq:ang-dominated-r1'} holds.
\end{enumerate}

Now note that the frequency-independent potential $\mc{V}_1$ satisfies
\begin{align*}
\lp|\mc{V}_1\rp| \leq Br^{-3}\,, \quad \lp|\frac{d\mc{V}_1}{dr}\rp|\leq Br^{-4}\,,\quad \lp|\frac{d}{dr}\lp((r^2+a^2)^3\frac{d\mc{V}_1}{dr}\rp)\rp|\leq Br\,,
\end{align*}
hence if $\omega_{\rm high}$ is sufficiently large, $\mc{V}_0$ dominates. Thus, the full potential cannot have any critical points in $[r_+,r_1']$ and will have a unique maximum, $r_{\rm max}$, in $[r_1',\infty)$ which satisfies $|r_{\rm max}-r_{\rm max}^0|\leq B\Lambda^{-1}$. 
\end{proof}

We briefly recall the strategy to prove Proposition~\ref{prop:angular-dominated}, which is best exemplified by Figure~\ref{fig:angular-dominated}. 

\begin{figure}[htbp]
\centering
\includegraphics[scale=1]{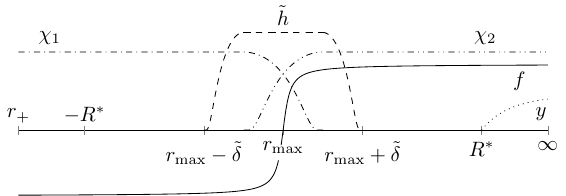}
\caption{Currents in the proof of Theorems~\ref{thm:ode-estimates-Psi-A} and \ref{thm:ode-estimates-Psi-B} for the frequency ranges $\mc{F}_{\sun}$ and $\mc{F}_{\measuredangle}$ (relative size not accurate).}
\label{fig:angular-dominated}
\end{figure}

We generate a positive bulk by means of an $f$ current that changes sign at the unique maximum of the potential and the degeneracy at this maximum is eliminated using an $h$ current that takes advantage of the positivity of $\mc{V}-\omega^2$ in a neighborhood of the maximum of $\mc{V}$ (see Lemma~\ref{lemma:V-angular-dominated}), c.f.\ \cite[Propositions 8.3.1 and 8.5.1]{Dafermos2016b} for $s=0$. We also introduce a $y$ current at large $r$, as well as another one localized near $r=r_+$, to improve the $r$-weights in the bulk terms in the estimate.

As superradiance can occur in the regimes analyzed in this section, we control the boundary terms by localized energy estimates which can be either of the Killing or Teukolsky--Starobinsky kind. The Killing currents can provide this control precisely because the conditions of applicability of  Lemma~\ref{lemma:basic-estimate-2} are met (statement (ii) of Lemma~\ref{lemma:angular-dominated-frequency-properties}), allowing us to easily bound the seemingly strong (in the frequency parameters) coupling errors introduced by these currents. The coupling errors associated with the virial currents $f$, $h$ and $y$ are also easily delt with by Lemma~\ref{lemma:basic-estimate-2}.

\begin{proof}[Proof of Propositions~\ref{prop:angular-dominated}]  
Suppose $(\omega,m,\Lambda)\in \mc{F}_{\text{\normalfont \sun}}(\varepsilon_{\rm width},\omega_{\rm high},a_0)\cup\mc{F}_{\measuredangle}(\varepsilon_{\rm width},\omega_{\rm high})$ is an admissible frequency triple with respect to $a\in[0,M]$ such that $a\in[0,a_0]$ if $(\omega,m,\Lambda)\in \mc{F}_{\text{\normalfont\sun,1}}$ and $a\in(a_0,M]$ if $(\omega,m,\Lambda)\in \mc{F}_{\text{\normalfont\sun,2}}$.

We begin by constructing explicit functions $f$ and $y$: take 
$$f(r)=f_0(r)=\arctan(r-r_{\rm max})\,,$$
as defined in \eqref{eq:standard-current-f0} and, for some $R^*>0$ and $\delta\in(0,1]$, $$y=y_\delta=\lp(1-\frac{(R^*)^{\delta}}{(r^*)^{\delta}}\rp)\mathbbm{1}_{[R^*,\infty)}\,,$$ 
as defined in \eqref{eq:standard-current-y-delta}; recall the properties of these functions as laid out in Section~\ref{sec:virial-current-templates}. Then, using $\mc{V}=\Lambda r^{-2} +O(r^{-3})$ as $r\to \infty$, we have
\begin{align*}
-\lp(y\mc{V}\rp)' = -y\mc{V}'\lp(1-\frac{\delta(R^*)^{\delta}}{(r^*)^{1+\delta}}\frac{\mc{V}}{-\mc{V}'}\rp)\mathbbm{1}_{[R^*,\infty)}\geq -y\mc{V}'\lp[1-\delta\lp(\frac{r}{2r^*}+o(1)\rp)\rp]\mathbbm{1}_{[R^*,\infty)}\geq -(1-\delta)y\mc{V}'\geq 0\,,
\end{align*}
as long as $\delta\in(0,1]$ and $R^*$ is sufficiently large; we fix $R^*$ in this way. We may also find it convenient to define
\begin{align*}
\hat{y}=\lp(1-\frac{(R^*)^{\delta}}{(-r^*)^{1/2}}\rp)\mathbbm{1}_{(-\infty,-R^*]}\,,
\text{~~satisfying~~}
\frac{w\hat{y}^2}{\hat{y}'}\leq B\,.
\end{align*}
When $m\omega<0$,  writing $\hat{\mc{V}}:=\mc{V}-\mc{V}(r_+)$ and noting, as $r\to r_+$,
\begin{align*}
\hat{\mc{V}}= \frac{2}{M}m\upomega_+(m\upomega_+-\omega)(r-r_+)+\Delta\frac{\Lambda-2am\omega}{4M^2r_+^2}+O\lp((r-r_+)^2\rp)\,,
\end{align*}
we deduce
\begin{align*}
-\lp(\hat y\hat{\mc{V}}\rp)' &= -\hat y\hat{\mc{V}}'\lp(1-\frac{(R^*)^{1/2}}{2(-r^*)^{3/2}}\frac{\hat{\mc{V}}}{\hat{\mc{V}}'}\rp)\mathbbm{1}_{(-\infty, -R^*]}\\
&\geq -\hat y\hat{\mc{V}}'\lp[1-\frac12\lp(1+o(1)\rp)\frac{(r-M)\hat{\mc{V}}}{2M^2\hat{\mc{V}}'}\rp]\mathbbm{1}_{[R^*,\infty)}\geq -\frac13 \hat y\hat{\mc{V}}'\geq 0\,.
\end{align*}

Recall, by Lemma~\ref{lemma:V-angular-dominated},
\begin{align*}
 -f\mc{V}'\geq b\frac{(r-M)\Delta(r-r_{\rm max})^2}{r^8}\Lambda \text{~~if~}m\omega\leq 0\,; \qquad -f\mc{V}'\geq b(a_0,\beta_1,\beta_2)\frac{\Delta(r-r_{\rm max})^2}{r^8}\Lambda \text{~~if~}m\omega> 0\,,
\end{align*}
whereas $|f'''|\leq B(r-M)^2w r^{-4}$. For $m\omega\leq 0$, it must be the case that $(\omega,m,\Lambda)\in\mc{F}_{\measuredangle}$, so by taking $\omega_{\rm high}$ and $\varepsilon_{\rm width}^{-1}$ sufficiently large, the error due to $f'''$ can be absorbed away from $r_{\rm max}$. If $m\omega>0$, then as long as $\epsilon_{\rm width}\omega_{\rm high}^2$ is sufficiently large depending on $a_0$ and $\beta_1$ or $\beta_2$, the same conclusion holds. Thus, for $\tilde\delta$ as defined in Lemma~\ref{lemma:V-angular-dominated}, we have
\begin{align*}
 -f\mc{V}'-\frac12f''' \geq -\frac34 f\mc{V}'-B\mathbbm{1}_{\{|r-r_{\rm max}|\leq \tilde{\delta}/2\}}\,,
\end{align*}
if $\omega_{\rm high}$ is sufficiently large possibly depending on the other parameters. By application of these currents, we obtain
\begin{align*}
&\int_{-\infty}^\infty \lp\{(2f'+y')|\Psi'|^2+\lp[-\frac34 f\mc{V}'-B\mathbbm{1}_{\{|r-r_{\rm max}|\leq \tilde{\delta}/2\}} -(1-\delta)y\mc{V}'+y'\omega^2\rp]|\Psi|^2\rp\}dr^* \\
&\quad\leq |s|\sum_{k=0}^{|s|-1}\int_{-\infty}^\infty\lp\{2aw(y+f)\Re\lp[\overline{\Psi}'\lp(c_{s,k,k}^{\rm id}+imc_{s,k,k}^{\rm \Phi}\rp)\uppsi_{(k)}\rp]+awf'\Re\lp[\overline{\Psi}\lp(c_{s,k,k}^{\rm id}+imc_{s,k,k}^{\rm \Phi}\rp)\uppsi_{(k)}\rp]\rp\}dr^*\\
&\quad\qquad+\int_{-\infty}^{\infty} \lp\{2(y+f)\Re\lp[\mathfrak{G}\overline{\Psi}'\rp]+f'\Re\lp[\mathfrak{G}\overline{\Psi}\rp]\rp\}dr^*\numberthis \label{eq:angular-dominated-intermediate-1}\\
&\quad\qquad+4\omega^2\lp[\lp|\swei{A}{s}_{\mc{I}^+}\rp|^2-\lp|\swei{A}{s}_{\mc{I}^-}\rp|^2\rp]+4(\omega-m\upomega_+)^2\lp[\lp|\swei{A}{s}_{\mc{H}^+}\rp|^2-\lp|\swei{A}{s}_{\mc{H}^-}\rp|^2\rp]\,. 
\end{align*}

It is easy to see that, on the left hand side of \eqref{eq:angular-dominated-intermediate-1}, the weight on $\Lambda|\Psi|^2$ vanishes as one approaches the maximum of this potential (moreover, as we were not careful to build $f$ so that $f'''<0$ near $r=r_{\rm max}$, a frequency independent error located there has appeared). Such a degeneration is expected for trapped frequencies; however in the cases of $\mc{F}_{\text{\normalfont \sun}}$ and $\mc{F}_{\measuredangle}\backslash\mc{F}_{\measuredangle,3}$, we can eliminate it by considering an additional $h$ current. With $\tilde\delta$ defined in Lemma~\ref{lemma:V-angular-dominated}, we let $h=A\tilde{h}$ for some $A>0$  and $0\leq \tilde{h}\leq 1$ is a smooth function supported for $r\in[r_{\rm max}-\delta,r_{\rm max}+\delta]$ such that $\tilde{h}=1$ in $r\in[r_{\rm max}-\tilde\delta/2,r_{\rm max}+\tilde\delta/2]$. For the two square brackets  in the expression
\begin{align*}
-f\mc{V}'-\frac12f''' +h(\mc{V}-\omega^2)-\frac12h'' &\geq -\frac12 f\mc{V}'+ \frac12 h(\mc{V}-\omega^2)\\
&\qquad+\lp[-\frac14f\mc{V}'-\frac12 A\tilde{h}''\rp]+\lp[\frac12 A(\mc{V}-\omega^2)-B\rp]\mathbbm{1}_{\{|r-r_{\rm max}|\leq \tilde{\delta}/2\}}\,,
\end{align*}
to become positive, we need $A$ to satisfy
\begin{align*}
B\varepsilon_{\rm width}\omega_{\rm high}^{-2}\leq A\leq b\tilde{\delta}^2\varepsilon_{\rm width}^{-1}\omega_{\rm high}^2 \text{~~if~}m\omega\leq 0\,,\\
B(a_0,\beta_1,\beta_2)\varepsilon_{\rm width}^{-1}\omega_{\rm high}^{-2}\leq A\leq b(a_0,\beta_1,\beta_2)\tilde{\delta}^2(a_0,\beta_1,\beta_2)\varepsilon_{\rm width} \omega_{\rm high}^2 \text{~~if~}m\omega>0\,.
\end{align*}

On the other hand, we also need to correct the signs of the boundary terms in \eqref{eq:angular-dominated-intermediate-1}. Since the frequencies we consider may be superradiant, we will need to use localized energy currents; we choose Killing energy currents $E\chi_1Q^T$ and $E\chi_2Q^T$, where we let $\chi_1\leq 1$ be a smooth function which is 1 in $[r_+,r_{\rm max}-\tilde\delta/2]$ and 0 for $r\geq r_{\rm max}+\tilde\delta/2$ and $\chi_2\leq 1$ be a smooth function which is 1 in $[r_{\rm max}+\tilde\delta/2,\infty)$ and 0 for $r\in[r_+,r_{\rm max}-\tilde\delta/2]$. We generate localization errors
\begin{align*}
&\int_{-\infty}^\infty  [E\chi_1'(\omega-m\upomega_+)+E\chi_2'\omega]\Im\lp[\Psi'\overline{\Psi}\rp]dr^* \\
&\quad\leq BE\tilde\delta^{-1}\int_{-\infty}^\infty \lp[|\Psi'|^2+\lp(\omega^2+(\omega-m\upomega_+)^2\rp)|\Psi|^2\rp]\mathbbm{1}_{|r-r_{\rm max}|\leq \tilde\delta/2]}dr^*\\
&\quad\ll \int_{-\infty}^\infty \lp\{\frac12 (f'+y')|\Psi'|^2+\frac14\lp[- f\mc{V}' +y'\omega^2+h\lp(\mc{V}-\omega^2\rp)\rp]|\Psi|^2\rp\}dr^*\,,
\end{align*}
as long as $\omega_{\rm high}$ is large enough depending on $E$ (and, if $m\omega>0$, depending on $\varepsilon_{\rm width}$, $a_0$ and $\beta_1,\beta_2$) that $E\ll A\tilde{\delta}$. Similarly, if we add localized Teukolsky--Starobinsky-type energy currents, $E_W(\chi_1Q^{W,K}+\chi_2Q^{W,T})$, the errors produced are
\begin{align*}
&BE_W\tilde\delta^{-1}\sum_{k=0}^{|s|}\int_{-\infty}^\infty \lp[|\uppsi_{(k)}'|^2+\lp(\omega^2+(\omega-m\upomega_+)^2\rp)|\uppsi_{(k)}|^2\rp]\mathbbm{1}_{|r-r_{\rm max}|\leq \tilde\delta/2]}dr^*\\
&\quad\leq BE_W\tilde\delta^{-1}\int_{-\infty}^\infty \lp[|\Psi'|^2+\lp(\omega^2+(\omega-m\upomega_+)^2+1\rp)|\Psi|^2\rp]\mathbbm{1}_{|r-r_{\rm max}|\leq \tilde\delta/2]}dr^* \\
&\quad\qquad+BE_W\tilde{\delta}^{-1}|s|\sum_{k=0}^{|s|-1}\int_{-\infty}^\infty (\omega^2+m^2)|\uppsi_{(k)}|^2\mathbbm{1}_{|r-r_{\rm max}|\leq \tilde\delta/2]}dr^*\,,
\end{align*}
where the second line can be treated as for the Killing energy currents.

Putting together the $f$, $h$, $y$ and the  Killing energy currents, \eqref{eq:angular-dominated-intermediate-1} becomes
\begin{align*}
&\int_{-\infty}^\infty \lp\{\frac12(f'+y')|\Psi'|^2+\frac14\lp[-f\mc{V}' + h(\mc{V}-\omega^2)+y'\omega^2+\hat{y}'(\omega-m\upomega_+)\mathbbm{1}_{m\omega\leq 0}\rp]|\Psi|^2\rp\}dr^* \\
&\quad\leq \sum_{k=0}^{|s|-1}\int_{-\infty}^\infty aw\Re\lp[\lp(2(y+f+\hat{y}\mathbbm{1}_{m\omega\leq 0})\overline{\Psi}'+(h+f')\overline{\Psi}\rp)\lp(c_{s,k,k}^{\rm id}+imc_{s,k,k}^{\rm \Phi}\rp)\uppsi_{(k)}\rp]dr^*\\
&\quad\qquad -\sum_{k=0}^{|s|-1}\int_{-\infty}^\infty Eaw\lp[\omega\chi_1+(\omega-m\upomega_+)\chi_2\rp]\Im\lp[\overline{\Psi}\lp(c_{s,k,k}^{\rm id}+imc_{s,k,k}^{\rm \Phi}\rp)\uppsi_{(k)}\rp]dr^* \numberthis \label{eq:angular-dominated-intermediate-2}\\
&\quad\qquad+\int_{-\infty}^{\infty} \lp\{2(y+f+\hat{y}\mathbbm{1}_{m\omega\leq 0})\Re\lp[\mathfrak{G}\overline{\Psi}'\rp]+(h+f')\Re\lp[\mathfrak{G}\overline{\Psi}\rp]-E\lp[\chi_1(\omega-m\upomega_+)+\chi_2\omega\rp]\Im\lp[\mathfrak{G}\overline{\Psi}\rp]\rp\}dr^*\\
&\quad\qquad+\omega^2\lp[(4-E)\lp|\swei{A}{s}_{\mc{I}^+}\rp|^2+(4+E)\lp|\swei{A}{s}_{\mc{I}^-}\rp|^2\rp)+(\omega-m\upomega_+)^2\lp[(4-E)\lp|\swei{A}{s}_{\mc{H}^+}\rp|^2+(4+E)\lp|\swei{A}{s}_{\mc{H}^-}\rp|^2\rp]\,.
\end{align*}
We may also add the Teukolsky--Starobinsky-type energy currents with any $E_W$, at the cost of adding, to \eqref{eq:angular-dominated-intermediate-2},
\begin{align*}
&E_W\omega^2\lp\{\begin{array}{lr}
1&\text{if~} s\geq 0\\
\frac{\mathfrak{C}_s}{\mathfrak{D}_s^{\mc{I}}}&\text{if~} s<0
\end{array}\rp\}\lp|\swei{A}{s}_{\mc{I}^-}\rp|^2+(\omega-m\upomega_+)^2E_W\lp\{\begin{array}{lr}
1&\text{if~} s< 0\\
\frac{\mathfrak{C}_s}{\mathfrak{D}_s^{\mc{H}}}&\text{if~} s\geq 0
\end{array}\rp\}\lp|\swei{A}{s}_{\mc{H}^-}\rp|^2\\
&\qquad +BE_W\tilde{\delta}^{-1}|s|\sum_{k=0}^{|s|-1}\int_{-\infty}^\infty (\omega^2+m^2)|\uppsi_{(k)}|^2\mathbbm{1}_{|r-r_{\rm max}|\leq \tilde\delta/2]}dr^* +\int_{-\infty}^\infty E_W\lp[\chi_1\lp(Q^{W,K}\rp)'+\chi_2\lp(Q^{W,T}\rp)'\rp]dr^*
\end{align*}
Then, choosing either $E\geq 5$ or $E_W\geq 15$, we have the appropriate sign on the boundary terms (see Lemma~\ref{lemma:angular-dominated-frequency-properties}\ref{it:angular-dominated-frequency-properties-Cs-Ds}).

We now focus on the terms that appear for $|s|\neq 0$. By a simple application of Cauchy-Schwarz,
\begin{align*}
& \int_{-\infty}^\infty\lp\{2(y+f+\hat{y})aw\Re\lp[\overline{\Psi}'\lp(c_{s,k,k}^{\rm id}+imc_{s,k,k}^{\rm \Phi}\rp)\uppsi_{(k)}\rp]+(h+f')aw\Re\lp[\overline{\Psi}\lp(c_{s,k,k}^{\rm id}+imc_{s,k,k}^{\rm \Phi}\rp)\uppsi_{(k)}\rp]\rp\}dr^* \\
&\quad\qquad -\int_{-\infty}^\infty Eaw\lp[\omega\chi_1+(\omega-m\upomega_+)\chi_2\rp]\Im\lp[\overline{\Psi}\lp(c_{s,k,k}^{\rm id}+imc_{s,k,k}^{\rm \Phi}\rp)\uppsi_{(k)}\rp]dr^*\\
&\quad\qquad+BE_W\tilde{\delta}^{-1}|s|\sum_{k=0}^{|s|-1}\int_{-\infty}^\infty (m^2+\omega^2)|\uppsi_{(k)}|^2\mathbbm{1}_{|r-r_{\rm max}|\leq \tilde\delta/2]}dr^*\\
&\quad\leq \int_{-\infty}^\infty  \lp\{\frac{1}{4|s|}(y'+f')|\Psi'|^2+\frac{1}{8|s|}w\lp[\frac{1+\varepsilon m^2}{r}+\varepsilon\omega^2\rp]|\Psi|^2\rp\}dr^*+\int_{-\infty}^\infty BE_W\tilde{\delta}^{-1}w(m^2+\omega^2)|\uppsi_{(k)}|^2dr^*\\
&\quad\qquad+ \int_{-\infty}^\infty 4|s|a^2B\lp(\frac{wy^2}{y'}+\frac{w\hat y^2}{\hat y'}+\frac{wf^2}{f'}+f'^2r+h^2r+E^2\varepsilon^{-1}(1+\chi_2^2r)\rp)w(1+m^2)\lp|\uppsi_{(k)}\rp|^2dr^* \\
&\quad\leq \int_{-\infty}^\infty  \frac{1}{4|s|}(y'+f')|\Psi'|^2+\frac{1}{8|s|}\lp[\varepsilon w \omega^2+\lp(\frac32\varepsilon+\frac{1}{\Lambda}\rp)\frac{\Delta \Lambda}{r^5}\rp]|\Psi|^2dr^*\\
&\quad\qquad+\int_{-\infty}^\infty |s|Bw\lp[\varepsilon^{-1}(1+E^2)+E_W\tilde{\delta}^{-1}(1+\omega^2/\Lambda)\rp]\Lambda\lp|\uppsi_{(k)}\rp|^2dr^*\,,
\end{align*}
where we have used the fact that $wy^2/y'$, $w\hat y^2/\hat y'$, $wf^2/f'$, $h^2r$, $f'^2r$ and $\chi_2^2 r$ are bounded in the integration range. Clearly, as long as $\Lambda$ is sufficiently large (in $\mc{F}_{\sun}$, this means requiring $\omega_{\rm high}$ sufficiently large depending on $\varepsilon_{\rm width}$; in $\mc{F}_{\measuredangle}$ requiring $\omega_{\rm high}$ and $\varepsilon_{\rm high}^{-1}$ sufficiently large), if $\varepsilon$ is sufficiently small
\begin{align*}
&\int_{-\infty}^\infty\lp\{\frac{1}{4|s|}(y'+f')|\Psi'|^2+\frac{1}{8|s|}\lp[\varepsilon w \omega^2+\lp(\frac32\varepsilon+\frac{1}{\Lambda}\rp)\frac{\Delta \Lambda}{r^5}\rp]|\Psi|^2\rp\}dr^* \\
&\quad\leq \int_{-\infty}^\infty \lp\{\frac14(f'+y')|\Psi'|^2+\frac18\lp[-f\mc{V}' + h(\mc{V}-\omega^2)+y'\omega^2\rp]|\Psi|^2\rp\}dr^* \,;
\end{align*}
we fix $\varepsilon$ so that the inequality holds. 

The only terms left to estimate are integrals of $B(1+E^2+E_W\tilde{\delta}^{-1}\varepsilon_{\rm width}^{-1})w\Lambda\lp|\uppsi_{(k)}\rp|^2$ for $k<|s|$; to control then, we appeal to Lemma~\ref{lemma:angular-dominated-frequency-properties}\ref{it:angular-dominated-frequency-properties-Lambda-nondegenerate} concerning properties of the frequency triples and to Lemma~\ref{lemma:basic-estimate-2}. 

For $\mc{F}_{\sun}$, iterating \eqref{eq:basic-estimate-2-a} gives
\begin{align*}
\int_{-\infty}^\infty w\Lambda\lp|\uppsi_{(k)}\rp|^2dr^* 
&\leq B\int_{-\infty}^\infty \lp(w\mathbbm{1}_{(-\infty,R^*]}+\frac{y'}{\sqrt{R}}\rp)\omega_{\rm high}^{-1}\Lambda\lp|\Psi\rp|^2dr^* +\int_{-\infty}^\infty k\sum_{i=0}^{k}\lp|\Re\lp[\mathfrak{G}_{(i)}\uppsi_{(i)}\rp]\rp|dr^*\\
&\leq \frac{\omega_{\rm high}^{-1}}{2|s|}\int_{-\infty}^\infty \frac{1}{16} \lp(y'\omega^2- f\mc{V}'+h\lp(\mc{V}-\omega^2\rp)\rp)\lp|\Psi\rp|^2dr^* +\int_{-\infty}^\infty |s|\sum_{i=0}^{|s|-1}\lp|\Re\lp[\mathfrak{G}_{(i)}\uppsi_{(i)}\rp]\rp|dr^*\,, 
\end{align*}
as long as $\omega_{\rm high}$ is sufficiently large depending on $\varepsilon_{\rm width}^{-1}$, $\tilde{a}_0$ and, if $a\in(a_0,M]$, on $\beta_2$ (so that $\tilde{\varepsilon}_1$ is sufficiently small). Choosing $\omega_{\rm high}$ sufficiently high depending on $E$ and $E_W$, errors arising from bulk terms proportional to $B(1+E^2+E_W\tilde{\delta}^{-1})w\Lambda\lp|\uppsi_{(k)}\rp|^2$, for $k<|s|$, can clearly be absorbed into the left hand side of \eqref{eq:angular-dominated-intermediate-2}, concluding the proof in $\mc{F}_{\sun}$.

For $\mc{F}_{\measuredangle,1}$ and $\mc{F}_{\measuredangle,2}$, iterating \eqref{eq:basic-estimate-2-b} gives
\begin{align*}
\int_{-\infty}^\infty w\Lambda\lp|\uppsi_{(k)}\rp|^2dr^* &\leq \int_{-\infty}^\infty w\lp(\frac{\omega^2}{\omega_{\rm high}^2}+\frac{\varepsilon_{\rm width}^{-1}\omega_{\rm high}^{-2}\Lambda}{r}\rp)\lp|\Psi\rp|^2dr^* +\int_{-\infty}^\infty k\sum_{i=0}^{k}\lp|\Re\lp[\mathfrak{G}_{(i)}\uppsi_{(i)}\rp]\rp|dr^* 
\end{align*}
which, if $m\omega>0$, is controlled by
\begin{align*}
\frac{\omega_{\rm high}^{-1}}{2|s|}\int_{-\infty}^\infty \frac{1}{16} \lp(y'\omega^2- f\mc{V}'+h\lp(\mc{V}-\omega^2\rp)\rp)\lp|\Psi\rp|^2dr^* +\int_{-\infty}^\infty |s|\sum_{i=0}^{|s|-1}\lp|\Re\lp[\mathfrak{G}_{(i)}\uppsi_{(i)}\rp]\rp|dr^*\,, 
\end{align*}
and, if $m\omega\leq 0$, is certainly controlled by
\begin{align*}
\frac{\omega_{\rm high}^{-1}}{2|s|}\int_{-\infty}^\infty \frac{1}{16} \lp(y'\omega^2+\hat y'(\omega-m\upomega_+)^2- f\mc{V}'+h\lp(\mc{V}-\omega^2\rp)\rp)\lp|\Psi\rp|^2dr^* +\int_{-\infty}^\infty |s|\sum_{i=0}^{|s|-1}\lp|\Re\lp[\mathfrak{G}_{(i)}\uppsi_{(i)}\rp]\rp|dr^*\,, 
\end{align*}
Hence, if $\omega_{\rm high}$ is sufficiently large with respect to $E$ and with respect to $E_W\tilde{\delta}^{-1}$ (and $\varepsilon_{\rm high}^{-1}$ in $\mc{F}_{\sun}$), the errors due to $B(1+E^2+E_W\tilde\delta^{-1}(1+\omega^2/\Lambda))\Lambda|\uppsi_{(k)}|^2$  can be absorbed into the left hand side of \eqref{eq:angular-dominated-intermediate-2}.
\end{proof}

\subsubsection{Nonsuperradiant trapping regime}

This section concerns the frequency range $\mc{F}_{\rm comp,1}$. We will show

\begin{proposition}[Estimates in $\mc{F}_{\rm comp,1}$] \label{prop:comp1} Fix $s\in\{0,\pm 1, \pm 2\}$ and $M>0$.  Then, for all $\delta\in(0,1]$, for all $\beta_3$, for all $\varepsilon_{\rm width}>0$ sufficiently small depending on $\beta_3$, for all $E,E_W>0$ such that $\max\{E,E_W\}$ is sufficiently large depending on $\varepsilon_{\rm width}$, for all $\omega_{\rm high}$ sufficiently large depending on $\varepsilon_{\rm width}$ and $E$, and for all $(a,\omega,m,\Lambda) \in \mc{F}_{\rm comp,1}(\omega_{\rm high},\varepsilon_{\rm width},\beta_3,r_0')$ where $r_0'=r_0'(\varepsilon_{\rm width})$ is set by Lemma~\ref{lemma:V0-trapping}, there exists a value $r_{\rm trap}\in(r_+,\infty)$ satisfying $b(\varepsilon_{\rm width})\leq r_{\rm trap}-r_+\leq B(\varepsilon_{\rm width})$ and functions $y$, $\tilde{y}$ and $f$ satisfying the uniform bounds
$|y| +|\tilde{y}|+|f|\leq B(\varepsilon_{\rm width}) \,,$
such that, for all smooth $\Psi$ arising from a smooth solution to the radial ODE~\eqref{eq:radial-ODE-alpha} via \eqref{eq:def-psi0-separated} and \eqref{eq:transformed-transport-separated} and itself satisfying the radial ODE~\eqref{eq:radial-ODE-Psi}, if $\Psi$ has the general boundary conditions of Lemma~\ref{lemma:uppsi-general-asymptotics} or the outgoing boundary conditions of Definition \ref{def:outgoing-bdry-uppsi}, we have the estimates
\begin{align*}
&\omega^2\lp|\swei{A}{s}_{\mc{I}^+}\rp|^2+(\omega-m\upomega_+)^2\lp|\swei{A}{s}_{\mc{H}^+}\rp|^2 + b(\delta,\varepsilon_{\rm width})\int_{-\infty}^\infty \frac{\Delta}{r^{3+\delta}}\Big[|\Psi'|^2+\lp(1+\omega^2(1+r^{-1}r_{\rm trap})\rp)|\Psi|^2\Big]dr^*\\
&\quad\leq B(\varepsilon_{\rm width},E,E_W)\lp[\omega^2\lp|\swei{A}{s}_{\mc{I}^-}\rp|^2+(\omega-m\upomega_+)^2\lp|\swei{A}{s}_{\mc{H}^-}\rp|^2\rp]\\
&\quad\qquad+\int_{-\infty}^\infty \lp\{2(y+\hat{y}+f)\Re\lp[\mathfrak{G}\overline{\Psi}'\rp]+f\Re\lp[\mathfrak{G}\overline{\Psi}\rp]-E\omega\Im\lp[\mathfrak{G}\overline{\Psi}\rp] +E_W\lp(Q^{W,T}\rp)'\rp\}dr^*\\
&\quad\qquad+\int_{-\infty}^\infty |s|B(\varepsilon_{\rm width},E)\sum_{k=0}^{|s|-1}\lp|\Re\lp[\mathfrak{G}_{(k)}\uppsi_{(k)}\rp]\rp| dr^*\,,
\end{align*}
and
\begin{align*}
\lp|\frac{\mathfrak{D}_s^\mc{H}}{\mathfrak{D}_s^\mc{I}}\rp|+\lp|\frac{\mathfrak{D}_s^\mc{H}}{\mathfrak{D}_s^\mc{I}}\rp|^{-1}+\lp|\frac{\mathfrak{C}_s}{\mathfrak{D}_s^\mc{I}}\rp|+\lp|\frac{\mathfrak{C}_s}{\mathfrak{D}_s^\mc{I}}\rp|^{-1}+\lp|\frac{\mathfrak{C}_s}{\mathfrak{D}_s^\mc{H}}\rp|+\lp|\frac{\mathfrak{C}_s}{\mathfrak{D}_s^\mc{H}}\rp|^{-1}\leq B\,.
\end{align*}
\end{proposition}

We first discuss the specific properties of the frequency parameters in the trapped nonsuperradiant regime and its consequences for the behavior of the potential:

\begin{lemma}[Properties of the frequency parameters in $\mc{F}_{\rm comp,1}$] \label{lemma:trapping-frequency-properties} Fix $s\in\mathbb{Z}$ and $M>0$. Let $a\in[0,M]$ and $(\omega,m,\Lambda)\in \mc{F}_{\rm comp,1}(\varepsilon_{\rm width}, \omega_{\rm high},r_0',\beta_3)$. Then, for sufficiently small $\varepsilon_{\rm width}$ depending on $\beta_3$ and for sufficiently large $\omega_{\rm high}$ depending on $\varepsilon_{\rm width}$, fixing $r_0'$ by Lemma~\ref{lemma:V0-trapping}, we have
\begin{enumerate}[label=(\roman*)]
\item $\varepsilon_{\rm width}\leq \omega/(\omega-m\upomega_+) \leq \varepsilon_{\rm width}^{-1}$; \label{it:trapping-frequency-properties-K-frequency}
\item $\Lambda\geq \frac23 m^2$; \label{it:trapping-frequency-properties-Lambda-m2}
\item the conditions for Lemma~\ref{lemma:basic-estimate-3} are satisfied; in fact, there is a $\beta_0=\beta_0(\varepsilon_{\rm width})$ such that $\Lambda-2am\omega\geq \beta_0\Lambda$; hence, for $\varepsilon_1=\beta_0\varepsilon_{\rm width}/6$, we have
$$\Lambda-2am\omega-\varepsilon_1(m^2+\omega^2)\geq \frac12 \beta_0\Lambda\,,\quad \tilde{\varepsilon}_1=\frac{12}{\varepsilon_{\rm width}^2\omega_{\rm high}^2\beta_0^2}\ll 1\,,$$
as long as $\omega_{\rm high}$ is sufficiently large depending on $\varepsilon_{\rm width}$; \label{it:trapping-frequency-properties-Lambda-nondegenerate}
\item for $|s|=1,2$, we have, for $\omega_{\rm high}$ sufficiently large (depending on $\varepsilon_{\rm width}$, $a_0$ and $\beta_2$ in the case of $\mc{F}_{\sun}$), \label{it:trapping-frequency-properties-Cs-Ds}
\begin{align*}
\frac{\mathfrak{C}_s}{\mathfrak{D}_s^{\mc{I}}}\,,\,\, \frac{\mathfrak{C}_s}{\mathfrak{D}_s^{\mc{H}}}\,, \,\, \frac{\mathfrak{D}_s^{H}}{\mathfrak{D}_s^{\mc{I}}} \in\lp[\frac13,\frac53\rp]
\end{align*}
\end{enumerate}
\end{lemma}
\begin{proof}
For \ref{it:trapping-frequency-properties-Lambda-m2} and \ref{it:trapping-frequency-properties-K-frequency}, we direct the reader to the proofs of the analogous results in Lemma~\ref{lemma:comparable-frequency-properties}.

Regarding \ref{it:trapping-frequency-properties-Lambda-nondegenerate}, note that the result follows easily if $m^2<(2M)^{-2}\varepsilon_{\rm width}\Lambda$ or if $a\leq \tilde{a}_0:=\varepsilon_{\rm width}^{1/2}/2$ (see also the proof of Lemma~\ref{lemma:comparable-frequency-properties}\ref{it:comparable-frequency-properties-Lambda-nondegenerate}). Hence, suppose $(2M)^{-2}\varepsilon_{\rm width}\Lambda\leq m^2\leq \Lambda$ and $a\geq \tilde{a}_0$. If $\Lambda-am\omega=o(\Lambda)$, then
\begin{align*}
\omega^2-\mc{V}_0&=\lp(\omega-\frac{am}{r^2+a^2}\rp)^2-\frac{\Delta(\Lambda-2am\omega)}{(r^2+a^2)^2}\geq (\omega-m\upomega_+)^2-\frac{\Delta(\Lambda-2am\omega)}{(r^2+a^2)^2} \\
&\geq \varepsilon_{\rm width}\Lambda -\frac{\Delta(\Lambda-2am\omega)}{(r^2+a^2)^2}\geq \frac12 \varepsilon_{\rm width}\Lambda\,,\quad \forall r\in[r_+,\infty)\,,
\end{align*}
using \ref{it:trapping-frequency-properties-K-frequency}, which contradicts the assumption that $r_0'<\infty$ for $\mc{F}_{\rm comp,1}$. Thus, there must exist a $\beta_0>0$ such that $\Lambda-2am\omega\geq \beta_0\Lambda$; this constant depends on $r_0'$ and hence on $\varepsilon_{\rm width}$. That the conditions of Lemma~\ref{lemma:basic-estimate-3} are satisfied then follows from \ref{it:trapping-frequency-properties-Lambda-m2} and the remaining restriction on the frequencies in $\mc{F}_{\rm comp,1}$.

Finally, for \ref{it:trapping-frequency-properties-Cs-Ds}, the proof is obtained by a similar strategy as in Lemma~\ref{lemma:angular-dominated-frequency-properties}\ref{it:angular-dominated-frequency-properties-Cs-Ds}, making use of \ref{it:trapping-frequency-properties-Lambda-nondegenerate}.
\end{proof}

\begin{lemma}[Behavior of the potential in $\mc{F}_{\rm comp,1}$]\label{lemma:V-trapping}
Fix $s\in\mathbb{Z}$, $M>0$ and $r_0'$ by Lemma~\ref{lemma:V0-trapping}. For all $\beta_3>0$, all  sufficiently small $\epsilon_{\rm width}$ depending on $\beta_3$ and all sufficiently large $\omega_{\rm high}$ depending on $\epsilon_{\rm width}$, fixing $r_0'$ by Lemma~\ref{lemma:V0-trapping}, if $(\omega,m,\Lambda)\in\mc{F}_{\rm comp,1}(\varepsilon_{\rm width},\omega_{\rm high},\beta_3,r_0')$, then there exists an $r_3\in(r_+,\infty)$ satisfying $|r_3-r_+|\geq B(\epsilon_{\rm width})$ and $r_3\leq B(\epsilon_{\rm width})$ such that
\begin{align*}
\mc{V}-\omega^2\leq -b(\epsilon_{\rm width})\Lambda \,, \quad\forall\,r\in[r_+,r_3]\,.
\end{align*}
and the potential $\mc{V}$ has a unique maximum $r_{\rm max}\in[r_3,\infty)$ which satisfies $|r_{\rm max}-r_+|\leq b(\epsilon_{\rm width})$ and $r_{\rm max}\leq B(\varepsilon_{\rm width})$.
\end{lemma}
\begin{proof}
Recall the bounds on $\mc{V}_0$ from \eqref{eq:trapping-V0-bounds}
\begin{align*} 
\lp|\frac{d\mc{V}_0}{dr}\rp|+r\lp|\frac{d^2\mc{V}_0}{dr^2}\rp| \leq B(\epsilon_{\rm width})\frac{\Lambda}{r^3}\,, \qquad \lp|\frac{d}{dr}\lp((r^2+a^2)^3\frac{d\mc{V}_0}{dr}\rp)\rp|\leq B(\epsilon_{\rm width})\Lambda r^2\,, 
\end{align*}
and the construction of $r_0'$ in Lemma~\ref{lemma:V0-trapping}. Clearly, if $r_0'<\infty$, then it has a maximum at $r^0_{\rm max}\geq r_0'$. Thus, $r_{\rm max}^0$ is bounded away from $r_+$ and, moreover, by Lemma~\ref{lemma:critical-points-V0}, bounded above, with both bounds depending on $\varepsilon_{\rm width}$.

To extend these conclusions to the full potential $\mc{V}$, we appeal to the bounds
\begin{align*}
\lp|\mc{V}_1\rp| \leq Br^{-3}\,, \quad \lp|\frac{d\mc{V}_1}{dr}\rp|\leq Br^{-4}\,,\quad \lp|\frac{d}{dr}\lp((r^2+a^2)^3\frac{d\mc{V}_1}{dr}\rp)\rp|\leq Br\,,
\end{align*}
hence if $\omega_{\rm high}$ is very large, $\Lambda$ is as well, so $\mc{V}_0$ dominates. Hence, we will still be able to obtain an $r_3$ with the same key properties as those of $r_0'$. Since $r_0',r_3<\infty$ by assumption, $\mc{V}$ has a unique maximum, $r_{\rm max}$, in $[r_3,\infty)$ which satisfies $|r_{\rm max}-r_{\rm max}^0|\leq B(\epsilon_{\rm width})\Lambda^{-1}$. 
\end{proof}

We are finally ready to make use of the above information to prove Proposition~\ref{prop:comp1}. We briefly recall the strategy, which is best exemplified by Figure~\ref{fig:nonsuperrad-trapping}.

\begin{figure}[htbp]
\centering
\includegraphics[scale=1]{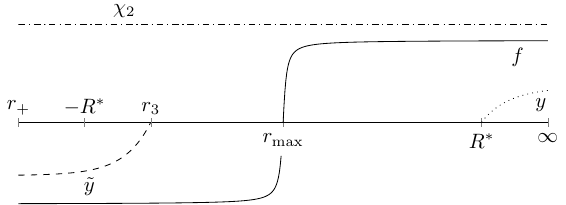}
\caption{Currents in the proof of Theorems~\ref{thm:ode-estimates-Psi-A} and \ref{thm:ode-estimates-Psi-B} for the frequency ranges $\mc{F}_{\rm comp,1}$ (relative size not accurate).}
\label{fig:nonsuperrad-trapping}
\end{figure}

We generate a positive bulk in $[r_3,\infty)$ by constructing an $f$ that changes sign with $\mc{V}'$, so that $-f\mc{V'}\geq 0$  in $[r_3,\infty)$; for $[r_+,r_3]$, we add a $y$ current that takes advantage of the positivity of $\omega^2-\mc{V}$ there (see Lemma~\ref{lemma:V-trapping}), c.f.\ \cite[Proposition 8.6.1]{Dafermos2016b} for $s=0$. The bulk term is thus degenerate at the maximum, located at $r=r_{\rm max}$, of the potential; unlike what occurs in the previous frequency ranges, \textbf{the degeneration cannot be eliminated}\footnote{Strictly speaking, in this regime, close to the maximum, while we cannot hope to have control of $\omega^2|\Psi|^2$, we can actually obtain control over $|\omega|[|\omega|(r-r_{\rm max})+1]|\Psi|^2$ as well as the weaker $[\omega^2(r-r_{\rm max})+1]|\Psi|^2$ in our Theorems~\ref{thm:ode-estimates-Psi-A} and ~\ref{thm:ode-estimates-Psi-fixed-m}.}.

For $s\neq 0$, due to the presence of errors due to coupling to $\uppsi_{(k)}$, $k=0,..,|s|-1$, we must add a $y$ current near $r^*=\pm\infty$ to absorb eventual errors of size $\omega^2$ in frequency and slow decay compared to $f\mc{V}'$. Moreover, a simple application of Cauchy-Schwarz yields errors of the form $m^2|\uppsi_{(k)}|^2$, with an $r$-weight determined by the choice of currents. To absorb them, an option would be appeal to the basic estimate in Lemma~\ref{lemma:basic-estimate-1}, in which case we would require two conditions:
\begin{itemize}
\item that the weight $c(r)$ in our application of Lemma~\ref{lemma:basic-estimate-1} can be taken to have a simple zero at $r=r_{\rm max}$ (note that, if $c'$ is to have the same sign globally, this implies $c$ must approach a nonzero value as $|r^*|\to \infty$, so that we pick up boundary terms of $\uppsi_{(k)}$), so that we can control $c'(r)m^2|\uppsi_{(k)}|^2$ by $(1-r_{\rm max}/r)^2wm^2|\uppsi_{(k)}|^2$;
\item to construct $f$ so that $f/f'$ is small in an arbitrarily large compact range of $r^*$, in order for the error to be small enough to absorb by the left hand side.
\end{itemize}
Instead, we note that by Lemma \ref{lemma:trapping-frequency-properties}\ref{it:trapping-frequency-properties-Lambda-nondegenerate}, we can actually make use of the improved estimates in Lemma~\ref{lemma:basic-estimate-2} and directly control $m^2|\uppsi_{(k)}|^2$ by $|\Psi|^2$ with appropriate $r$-weights.

Crucially, as had been the case for $s=0$ and $a\in[0,a_0]$ in \cite[Proposition 8.6.1]{Dafermos2016b}, the disjointness of trapping and superradiance implies that a global energy current is enough to deal with any boundary terms arising from $f$ and $y$ currents. One option is to use a global Teukolsky--Starobinsky energy current. On the other hand, we may also apply standard $T$-energy current: to deal with errors due to coupling, the final two estimates in Lemma~\ref{lemma:basic-estimate-3} are essential, as they allow us to shift all the frequency weights onto the lower level $\uppsi_{(k)}$ and make use of Lemma~\ref{lemma:basic-estimate-2} to eliminate them as we climb upwards on the hierarchy. 

With the procedure outlined to address the coupling errors, the final errors expressed in terms of $\Psi$ have no frequency weights; thus, they are easily absorbed by $|\Psi'|^2 + (1-r_{\rm max}/r)^2\Lambda|\Psi|^2$ with appropriate $r$-weights.

\begin{proof}
Suppose $(\omega,m,\Lambda)\in \mc{F}_{\rm comp,1}$, i.e.\ $\omega$ is large, of the same size as $\Lambda$ and all frequency triples are quantitatively (depending on $\beta_1$) nonsuperradiant. We set $r_{\rm trap}:=r_{\rm max}$ to be the maximum of the potential, whose existence and uniqueness was established in Lemma~\ref{lemma:V-trapping}.

We start with the $f$ current: let 
$$f=f_0=\arctan(r^*-r_{\rm max}^*)\,,$$
the standard current defined by \eqref{eq:standard-current-f0} in Section~\ref{sec:virial-current-templates}, and recall the properties there laid out. We have
\begin{align*}
-f\mc{V}' -\frac12 f'''\geq  \lp(-f\mc{V}'-\frac12 f'''\rp)\mathbbm{1}_{[r_+,r_3]}+b(\varepsilon_{\rm width})\Lambda\frac{((r-r_{\rm max})^2\Delta}{r^7}\mathbbm{1}_{[r_3,\infty)}\,.
\end{align*}

Additionally, for some $\delta\in(0,1]$ and $R^*\geq \delta^{-2}$, consider 
$$y=y_{\delta}=\lp(1-\frac{(R^*)^{\delta}}{(r^*)^{\delta}}\rp)\mathbbm{1}_{[R^*,\infty)}\,,$$
as defined in \eqref{eq:standard-current-y-delta}, and a current $\tilde{y}$  given by
\begin{alignat*}{3}
\tilde{y}&=\frac12\lp(1-e^{C(r_3-r)}\rp)\mathbbm{1}_{[r_+,r_3]}\,,\qquad&&\tilde{y}'&&= \frac12 Ce^{C(r_3-r)}\frac{\Delta}{r^2+a^2}\mathbbm{1}_{[r_+,r_3]}\,,
\end{alignat*}
where $C$ is large enough so that $\tilde{y}<0$, $\tilde{y}'>0$. Fix $2R^*\geq \delta^{-1}\varepsilon_{\rm width}^{-1/3}$ so that it is sufficiently large to have
\begin{gather*}
y'\omega^2-(y\mc{V})'\geq y'\omega^2\lp[1-\frac{\delta^{-1}\varepsilon_{\rm width}^{-1}}{r^3}\lp(1-\frac{(R^*)^{\delta}}{(r^*)^{\delta}}\rp)\rp]\geq \frac12 y'\omega^2\geq b  y'\Lambda\,,
\end{gather*}
Also note that
\begin{gather*}
\frac{|w\tilde{y}^2|}{\tilde{y}'}\leq \frac{B(M)}{C}\,.
\end{gather*}

Finally, using bounds \eqref{eq:trapping-V0-bounds} on the first derivative of the potential as well as Lemma~\ref{lemma:V-trapping},
\begin{align*}
&\lp[\tilde{y}'(\omega^2-\mc{V})-\tilde{y}\mc{V}'-f\mc{V}'-\frac12f'''\rp]\mathbbm{1}_{[r_+,r_3]}\\
&\quad=\frac12 \tilde{y}'(\omega^2-\mc{V})+B\lp(C e^{C(r_3-r)}b(\varepsilon_{\rm width}-B(\varepsilon_{\rm width})(2-e^{C(r_3-r)})-\frac{B}{\varepsilon_{\rm width}\omega_{\rm high}^2}\rp)\Lambda\frac{\Delta}{r^2}\mathbbm{1}_{[r_+,r_3]}\\
&\quad\geq \frac12 \tilde{y}'(\omega^2-\mc{V})\,,
\end{align*}
if $C=C(\varepsilon_{\rm width})$ is sufficiently large and if $\omega_{\rm high}$ is sufficiently large depending on $\varepsilon_{\rm width}$.

By applying these currents, together with a global $T$ energy current multiplied by some $E>1$ to be fixed, we obtain the identity
\begin{align*}
&\int_{-\infty}^\infty \lp[(2f'+y'+\tilde{y}')|\Psi'|^2+\lp(\frac12 \tilde{y}'(\omega^2-\mc{V})-f\mc{V}'\mathbbm{1}_{[r_3,\infty)}+\frac12 y'\omega^2\rp)|\Psi|^2\rp]dr^*\\
&\quad \leq \int_{-\infty}^\infty \lp[\lp(2f'+\tilde{y}'+y'\rp)|\Psi'|^2+\lp(\tilde{y}'(\omega^2-\mc{V})-(\tilde{y}+f)\mc{V}'+y'\omega^2-(y\mc{V})'\rp)|\Psi|^2\rp]dr^*\\
&\quad\leq \lp[4(\omega-m\upomega_+)-E\omega\rp](\omega-m\upomega_+)\lp|\swei{A}{s}_{\mc{H}^+}\rp|^2+\lp[4(\omega-m\upomega_+)+E\omega\rp](\omega-m\upomega_+)\lp|\swei{A}{s}_{\mc{H}^-}\rp|^2 \numberthis \label{eq:trapping-intermediate-1}\\
&\quad\qquad+(4-E)\omega^2\lp|\swei{A}{s}_{\mc{I}^+}\rp|^2 +(4+E)\omega^2\lp|\swei{A}{s}_{\mc{I}^-}\rp|^2 \\
&\quad\qquad +\int_{-\infty}^\infty \frac12 f'''|\Psi|^2\mathbbm{1}_{\{|r-r_{\rm max}|\leq 2\tilde\delta\}}dr^*+\int_{-\infty}^{\infty} \lp\{2(y+\tilde{y}+f)\Re\lp[\mathfrak{G}\overline{\Psi}'\rp]+f'\Re\lp[\mathfrak{G}\overline{\Psi}\rp]-E\omega\Im\lp[\mathfrak{G}\overline{\Psi}\rp]\rp\}dr^* \\
&\quad\qquad+\sum_{k=0}^{|s|-1} \int_{-\infty}^\infty 2 a w(\tilde{y}+y+f)\Re\lp[\lp(c_{s,k,k}^{\rm id}+imc_{s,k,k}^{\Phi}\rp)\uppsi_{(k)}\overline{\Psi}'\rp]dr^*\\
&\quad\qquad+\sum_{k=0}^{|s|-1} \int_{-\infty}^\infty \lp\{a wf'\Re\lp[\lp(c_{s,k,k}^{\rm id}+imc_{s,k,k}^{\Phi}\rp)\uppsi_{(k)}\overline{\Psi}\rp]-Ea\omega w\Im\lp[(c_{s,s,k}^{id}+imc_{s,s,k}^\Phi)\uppsi_{(k)}\overline{\Psi}\rp]\rp\}dr^*\,.
\end{align*}
We note that, while the weights on $|\Psi'|^2$ and $|\Psi|^2$ on the left hand side are nonnegative by our previous considerations, the latter degenerates at $r=r_{\rm max}$. Unlike in the angular dominated and superradiant regimes analyzed in Section~\ref{section:F-angular-superradiant}, we cannot fully eliminate the degeneration, as $\mc{V}(r_{\rm max})\sim \omega^2$. 

Let us first deal with the boundary terms. Suppose we add a Teukolsky--Starobinsky $T$-type energy current to \eqref{eq:trapping-intermediate-1}. It is clear that, if $\varepsilon_{\rm width}<1$, setting $\max\{E,E_W\}\geq 30\varepsilon_{\rm width}^{-1/2}/M$, we obtain
\begin{align*}
&\lp(4-E-E_W\lp\{1,\frac{\mathfrak{C}_s}{\mathfrak{D}_s^{\mc I}}\rp\}\rp)\omega^2\lp|\swei{A}{s}_{\mc{I}^+}\rp|^2 +\lp(4+E+E_W\lp\{\frac{\mathfrak{C}_s}{\mathfrak{D}_s^{\mc I}},1\rp\}\rp)\omega^2\lp|\swei{A}{s}_{\mc{I}^-}\rp|^2 \\
&\quad\leq -\omega^2\lp|\swei{A}{s}_{\mc{I}^+}\rp|^2 +9\max\{E,E_W\}\omega^2\lp|\swei{A}{s}_{\mc{I}^-}\rp|^2\,,
\end{align*}
where the options given are for $s>0$ and $s\leq 0$, respectively. Under the same convention, and, using Lemma~\ref{lemma:trapping-frequency-properties}\ref{it:trapping-frequency-properties-K-frequency}, as long as $\varepsilon_{\rm width}<1$,
\begin{align*}
&\lp[4-\lp(E+E_W\lp\{\frac{\mathfrak{C}_s}{\mathfrak{D}_s^{\mc I}},1\rp\}\rp)\frac{\omega}{\omega-m\upomega_+}\rp](\omega-m\upomega_+)^2\lp|\swei{A}{s}_{\mc{H}^+}\rp|^2\\
&\qquad+\lp[4+\lp(E+E_W\lp\{1,\frac{\mathfrak{C}_s}{\mathfrak{D}_s^{\mc I}}\rp\}\rp)\frac{\omega}{\omega-m\upomega_+}\rp](\omega-m\upomega_+)^2\lp|\swei{A}{s}_{\mc{H}^-}\rp|^2\\
&\quad\leq-(\omega-m\upomega_+)^2\lp|\swei{A}{s}_{\mc{H}^+}\rp|^2+ 9\max\{E,E_W\}\varepsilon_{\rm width}^{-1}(\omega-m\upomega_+)^2\lp|\swei{A}{s}_{\mc{H}^-}\rp|^2\,.
\end{align*}
The inequalities above have the signs we expect in the final statement.

Finally, we address the coupling terms in the last two lines of \eqref{eq:trapping-intermediate-1} as follows. We begin with an application of \eqref{eq:basic-estimate-3T-trapping-top} and \eqref{eq:basic-estimate-3T-trapping-lower}, from Lemma~\ref{lemma:basic-estimate-3}: for some $\varepsilon\in(0,1)$ sufficiently small and for $0\leq \chi \leq 1$ a smooth bump function supported in a neighborhood of $r_{\rm max}$, of size $\tilde{\delta}=\tilde\delta(\varepsilon_{\rm width})$, with $\chi(r_{\rm max})=1$,
\begin{align*}
&|s|\int_{-\infty}^\infty\sum_{k=0}^{|s|-1}Ea\omega w\Im\lp[(c_{s,s,k}^{id}+imc_{s,s,k}^\Phi)\uppsi_{(k)}\overline{\Psi}\rp]dr^*\\
&\quad \leq B \int_{-\infty}^\infty\lp\{\varepsilon w|\Psi'|^2+w\lp[1+\varepsilon(1-\chi)(\Lambda+\omega^2)\rp]|\Psi|^2+|s|E^2\sum_{k=0}^{|s|-1}w(\varepsilon^{-1}\Lambda+\omega^2)|\uppsi_{(k)}|^2\rp\}dr^*\\
&\quad \leq \int_{-\infty}^\infty \lp\{\frac14 (y'+\tilde{y}'+f')|\Psi'|^2+\lp[\frac{1}{16}\tilde{y}'(\omega^2-\mc{V})-\frac18f\mc{V}'+\frac{1}{16} y'\omega^2+B w\mathbbm{1}_{\{|r-r_{\rm max}|\leq 2 \tilde\delta\}}\rp]|\Psi|^2\rp\}dr^* \\
&\quad\qquad + |s|\sum_{k=0}^{|s|-1}\int_{-\infty}^\infty BE^2w\Lambda|\uppsi_{(k)}|^2dr^*\,;
\end{align*}
we fix $\varepsilon$ so that the last inequality above holds. By applications of Cauchy-Schwarz, and using the properties of the currents laid out above and in \eqref{eq:standard-current-f0-properties} and \eqref{eq:standard-current-y-delta-properties}, we also obtain
\begin{align*}
&2 a w(y+\tilde{y}+f)\Re\lp[\lp(c_{s,k,k}^{\rm id}+imc_{s,k,k}^{\Phi}\rp)\uppsi_{(k)}\overline{\Psi}'\rp]\\
&\quad\leq 2|s|M^2\max_{k\leq |s|-1}\lp(\lp\lVert c_{s,k,k}^{\Phi} \rp\rVert_\infty^2+\lp\lVert c_{s,k,k}^{\rm id} \rp\rVert_\infty^2\rp)(m^2+1)w^2\lp(\frac{f^2}{f'}+\frac{\tilde{y}^2}{\tilde{y}'}+\frac{y^2}{y'}\rp)\lp|\uppsi_{(k)}\rp|^2+\frac{1}{2|s|}(f'+\tilde{y}'+y')|\Psi'|^2\\
&\quad\leq \frac{1}{2|s|}(f'+\tilde{y}'+y')|\Psi'|^2+|s|B(M,s)w\Lambda\lp|\uppsi_{(k)}\rp|^2\,,\\
&2 a wf'\Re\lp[\lp(c_{s,k,k}^{\rm id}+imc_{s,k,k}^{\Phi}\rp)\uppsi_{(k)}\overline{\Psi}\rp]\\
&\quad\leq \frac12 M\max_{k\leq |s|-1}\lp(\lp\lVert c_{s,k,k}^{\Phi} \rp\rVert_\infty+\lp\lVert c_{s,k,k}^{\rm id} \rp\rVert_\infty\rp)\lp(f'w|\Psi|^2+|s|(m^2+1)f'w\lp|\uppsi_{(k)}\rp|^2\rp)\\
&\quad\leq B(M,s) \lp(w|\Psi|^2\mathbbm{1}_{\{|r-r_{\rm max}|\geq 2 \tilde\delta\}}+w|\Psi|^2\mathbbm{1}_{\{|r-r_{\rm max}|\leq 2 \tilde\delta\}}+|s|(m^2+1)w\lp|\uppsi_{(k)}\rp|^2\rp)\\
&\quad\leq \lp[\frac{1}{8}\tilde{y}'(\omega^2-\mc{V})-\frac14f\mc{V}'+\frac18 y'\omega^2\rp]|\Psi|^2+Bw|\Psi|^2\mathbbm{1}_{\{|r-r_{\rm max}|\leq 2 \tilde\delta\}}+|s|Bw\Lambda\lp|\uppsi_{(k)}\rp|^2\,,
\end{align*}
for $C$ sufficiently large and $\delta\in(0,1]$. 

It remains to estimate the terms with $k<|s|$. To do so, we repeatedly apply \eqref{eq:basic-estimate-2-a} from Lemma~\ref{lemma:basic-estimate-2}, since its conditions are met (see Lemma~\ref{lemma:trapping-frequency-properties}\ref{it:trapping-frequency-properties-Lambda-nondegenerate}). Using the notation there, we see that for $s\neq 0$ the errors due to $f$, $y$, $\tilde{y}$ and $T$ currents are controlled by
\begin{align*}
&|s|\sum_{k=0}^{|s|-1} \int_{-\infty}^\infty B\frac{(1+E^2)}{\beta_0}\beta_0 w\Lambda\lp|\uppsi_{(k)}\rp|^2dr^* \\
&\quad\leq \frac{B(1+E^2)}{\varepsilon_{\rm width}\beta_0^2\omega_{\rm high}^2}\int_{-\infty}^\infty \mathbbm{1}_{\{|r-r_{\rm max}|\geq 2 \tilde\delta\}}\lp(w\mathbbm{1}_{(-\infty,R^*]}+\frac{y'}{\sqrt{R}}\rp)\omega^2|\Psi|^2dr^*\\
&\quad\qquad+\frac{B(1+E^2)}{\varepsilon_{\rm width}\beta_0^2}\int_{-\infty}^\infty |\Psi|^2\mathbbm{1}_{\{|r-r_{\rm max}|\leq 2 \tilde\delta\}}dr^*+\frac{B|s|(1+E^2)}{\beta_0}\sum_{k=0}^{|s|-1}\int_{-\infty}^\infty\lp|\Re\lp[\mathfrak{G}_{k}\uppsi_{k}\rp]\rp|dr^* \\
&\quad\leq \int_{-\infty}^\infty \lp\{\lp[\frac{1}{8}\tilde{y}'(\omega^2-\mc{V})-\frac14f\mc{V}'+\frac18y'\omega^2\rp]|\Psi|^2+\frac{B(1+E^2)}{\varepsilon_{\rm width}\beta_0^2}|\Psi|^2\mathbbm{1}_{\{|r-r_{\rm max}|\leq 2 \tilde\delta\}}\rp\}dr^*\\
&\quad\qquad+\frac{B|s|(1+E^2)}{\beta_0}\sum_{k=0}^{|s|-1}\int_{-\infty}^\infty\lp|\Re\lp[\mathfrak{G}_{k}\uppsi_{k}\rp]\rp|dr^*\,,
\end{align*}
for sufficiently large $\omega_{\rm high}$, depending on $\varepsilon_{\rm width}$ and $E$, and sufficiently large $C$; $C$ can now be fixed to satisfy all the previous constraints. 

From \eqref{eq:trapping-intermediate-1}, we obtain
\begin{align*}
&\omega^2\lp|\swei{A}{s}_{\mc{I}^+}\rp|^2+(\omega-m\upomega_+)^2\lp|\swei{A}{s}_{\mc{H}^+}\rp|^2+\int_{-\infty}^\infty \lp[\frac{y'+f'}{2}|\Psi'|^2+\lp(\frac18 \tilde{y}'(\omega^2-\mc{V})-\frac38f\mc{V}'\mathbbm{1}_{[r_3,\infty)}+\frac18 y'\omega^2\rp)|\Psi|^2\rp]dr^*\\
&\quad\leq 5E\omega^2\lp|\swei{A}{s}_{\mc{I}^-}\rp|^2+5E\varepsilon_{\rm width}^{-1}(\omega-m\upomega_+)^2\lp|\swei{A}{s}_{\mc{H}^-}\rp|^2+\int_{-\infty}^\infty \frac{B(1+E^2)}{\varepsilon_{\rm width}\beta_0^2}w|\Psi|^2\mathbbm{1}_{\{|r-r_{\rm max}|\leq 2\tilde\delta\}}dr^* \numberthis \label{eq:trapping-intermediate-2}\\
&\quad\qquad +\int_{-\infty}^{\infty} \lp\{2(y+\hat{y}+f)\Re\lp[\mathfrak{G}\overline{\Psi}'\rp]+f'\Re\lp[\mathfrak{G}\overline{\Psi}\rp]-E\omega\Im\lp[\mathfrak{G}\overline{\Psi}\rp]\rp\}dr^* \\
&\quad\qquad+B(\varepsilon_{\rm width},E)|s|\sum_{k=0}^{|s|-1}\int_{-\infty}^\infty\lp|\Re\lp[\mathfrak{G}_{k}\uppsi_{k}\rp]\rp|dr^*\,.
\end{align*}
It remains to absorb the $|\Psi|^2\mathbbm{1}_{\{|r-r_{\rm max}|\leq 2\tilde\delta\}}$ term into the left hand side \eqref{eq:trapping-intermediate-2}.  To bound it, let $\chi$ be a bump function which is $1$ in $\{r:|r-r_{\rm max}|\leq \tilde\delta\}$ and 0 in $\{r:|r-r_{\rm max}|\geq 2\tilde\delta\}$; then 
\begin{align*}
\int_{-\infty}^\infty \chi |\Psi|^2 dr^*
&\leq \int_{-\infty}^\infty (r-r_{\rm max}-2\tilde\delta) \lp(\chi' |\Psi|^2+2\chi\Re\lp[\Psi'\overline{\Psi}\rp]\rp)dr^*\\
&\leq \int_{-\infty}^\infty\lp[4(r-r_{\rm max}-2\tilde\delta)^2|\Psi'|^2+2|r-r_{\rm max}-2\tilde\delta||\chi'||\Psi|^2\rp]dr^* \\
&\leq B\int_{\supp \chi'}|\Psi|^2dr^*+B\int_{-\infty}^\infty\tilde\delta^2\chi|\Psi'|^2dr^* \,.
\end{align*}
Thus, the remaining bulk error,
\begin{align*}
&\int_{-\infty}^\infty \frac{B(1+E^2)}{\varepsilon_{\rm width}\beta_0^2}w|\Psi|^2\mathbbm{1}_{\{|r-r_{\rm max}|\leq 2\tilde\delta\}}dr^* \\
&\quad\leq \int_{-\infty}^\infty \frac{B(1+E^2)}{\varepsilon_{\rm width}\beta_0^2\omega_{\rm high}^2}w\omega^2|\Psi|^2\mathbbm{1}_{\{|r-r_{\rm max}|\geq 2\tilde\delta\}}dr^* +B\int_{-\infty}^\infty\frac{B(1+E^2)\tilde\delta^2}{\varepsilon_{\rm width}\beta_0^2}|\Psi'|^2\,,
\end{align*}
can be absorbed into the left hand side of \eqref{eq:trapping-intermediate-2} as long as $\tilde\delta\sim \omega_{\rm high}^{-1}$ is sufficiently small depending on $\varepsilon_{\rm width}$ and $E$.
\end{proof}


\subsection{Proof of Theorems~\ref{thm:frequency-estimates-big}A and~\ref{thm:frequency-estimates-big}B: combining the frequency estimates of Sections~\ref{sec:bounded-smallness}, \ref{sec:intermediate} and \ref{sec:unbounded}}
\label{sec:combining-ode-estimates}

In this section, we combine the results of the previous sections to establish Theorem~\ref{thm:ode-estimates-Psi-A} and Theorem~\ref{thm:ode-estimates-Psi-fixed-m}.

\begin{proof}[Proof of Theorems~\ref{thm:ode-estimates-Psi-A} and \ref{thm:ode-estimates-Psi-B}]
Recall that $M>0$, $s\in\{0,\pm 1, \pm 2\}$ and $a_0\in[0,M)$ are fixed, and that $\beta_1=\beta_1(a_0)$ is chosen as the largest value for which Lemma~\ref{lemma:derivative-V0-r+} holds. Fix $\beta_2=\beta_2(a_0)>0$ and let $\beta_3=\beta_1(a_0)$. It is easy to see that every $(\omega,m,\Lambda)$ considered in Theorems~\ref{thm:ode-estimates-Psi-A} and \ref{thm:ode-estimates-Psi-B} lies in at least one of the frequency ranges 
\begin{gather*}
\mc{F}_{\text{\ClockLogo}}(\omega_{\rm high},\varepsilon_{\rm width})\,,\,\,\mc{F}_{\measuredangle,1}(\omega_{\rm high},\varepsilon_{\rm width})\,,\,\,\mc{F}_{\measuredangle,2}(\omega_{\rm high},\varepsilon_{\rm width})\,,\,\,\mc{F}_{\sun,1}(\omega_{\rm high},\varepsilon_{\rm width})\,,\,\,\\
\mc{F}_{\sun,1}(\omega_{\rm high},\varepsilon_{\rm width})\,,\,\,\mc{F}_{\rm comp,1}(\omega_{\rm high},\varepsilon_{\rm width},r_0')\,,\,\,\mc{F}_{\rm comp,2}(\omega_{\rm high},\varepsilon_{\rm width},r_0')\,,\,\,\\
\mc{F}_{\rm int,1}(\omega_{\rm high},\varepsilon_{\rm width},\omega_{\rm low})\,,\\
\mc{F}_{\rm low,1a}(\omega_{\rm high},\varepsilon_{\rm width},\omega_{\rm low},\tilde{a}_0)\,,\,\,\mc{F}_{\rm low,1b}(\omega_{\rm high},\varepsilon_{\rm width},\omega_{\rm low},\tilde{a}_0,\tilde{a}_1)\,,\,\,\\
\mc{F}_{\rm low,1c}(\omega_{\rm high},\varepsilon_{\rm width},\omega_{\rm low},\tilde{a}_0,\tilde{a}_1)\,,\,\,\mc{F}_{\rm low,2}(\omega_{\rm high},\varepsilon_{\rm width},\omega_{\rm low},\tilde{a}_0)\,.
\end{gather*}
In order to combine the estimates in each of this regimes, we must therefore fix $E$, $E_W$, $\varepsilon_{\rm width}$, $\omega_{\rm high}$, $r_0'$, $\omega_{\rm low}$, $\tilde{a}_0$ and $\tilde{a}_1$.

First, let us choose $\varepsilon_{\rm width}$ sufficiently small, depending on $a_0$, $\beta_1$ and $\beta_2$, consistent with the requirements of Propositions~\ref{prop:comp2}, \ref{prop:angular-dominated} and \ref{prop:superradiant}, and with the requirement that $\varepsilon_{\rm width}^{-1}\geq s^2+2M|s|$.

Now, if the all the precise boundary relations of Lemma~\ref{lemma:uppsi-general-asymptotics} hold, we can choose $E_W$ to be sufficiently large, depending on the already fixed $\varepsilon_{\rm width}$, in accordance with the statements of Propositions~\ref{prop:low1a}, \ref{prop:low1b}, \ref{prop:low1c}, \ref{prop:low2}, \ref{prop:int1},  \ref{prop:time-dominated}, \ref{prop:comp2}, \ref{prop:angular-dominated}, \ref{prop:superradiant}, and \ref{prop:comp1}. Note that this is certainly the case if one is considering \textit{homogeneous} solutions to the radial ODEs. On the other hand, for  $k=0,\dots,|s|$, $\uppsi_k$ are outgoing in the sense of Definition~\ref{def:outgoing-bdry-teukolsky} but do not necessarily verify all the relations in Lemma~\ref{lemma:uppsi-general-asymptotics}, we choose $E$ to be sufficiently large, depending on the already fixed $\varepsilon_{\rm width}$, as per the statements of the aforementioned propositions. In this case, by Propositions~\ref{prop:low1a}, \ref{prop:low1b} and \ref{prop:low2}, we must pick up the error term \eqref{eq:ode-estimates-outgoing-addition}.

By imposing that $\omega_{\rm high}\geq 1$ is sufficiently large in terms of all the remaining parameters, we can now determine $r_0'$ by Lemma~\ref{lemma:V0-trapping} and Propositions~\ref{prop:comp2}, \ref{prop:comp1}, \ref{prop:angular-dominated}, \ref{prop:superradiant} can be applied; we fix $\omega_{\rm high}$ so that this is the case. Moreover, from Definition~\ref{def:admissible-freqs}\ref{it:admissible-freqs-triple-Lambda-lower-bound}, we can ensure
\begin{align*}
m^2\leq \Lambda+2|s||a\omega|+s^2\leq 2\varepsilon_{\rm width}^{-1}\omega_{\rm high}^2\,,
\end{align*}
whenever $|\omega|\leq \omega_{\rm high}$ and $\Lambda\leq \varepsilon_{\rm width}^{-1}\omega_{\rm high}^2$.

The only parameters left to fix are those used to separate the bounded frequency ranges with smallness. We first choose $\tilde{a}_0$, $\tilde{a}_1$ and $\omega_{\rm low}$ to be sufficiently small that Propositions~\ref{prop:low1a} and  \ref{prop:low1c} can be applied; we fix $\tilde{a}_0$ and $\tilde{a}_1$ in this way. Then, we fix $\omega_{\rm low}$ depending on $\tilde{a}_0$ and $\tilde{a}_1$ so that Propositions~\ref{prop:low1b} and \ref{prop:low2} hold. 

Finally, with $\omega_{\rm low}$, $\omega_{\rm high}$ and $\varepsilon_{\rm width}$ fixed, we choose $R^*$ to be sufficiently large so that Proposition~\ref{prop:int1} applies.
\end{proof}

\begin{proof}[Proof of Theorem~\ref{thm:ode-estimates-Psi-fixed-m}] We need only modify our proof of Theorem~\ref{thm:ode-estimates-Psi-A} in the following manner: we choose $\omega_{\rm high}$ sufficiently large, now depending on $|m|$, so that, in addition to all other constraints, the hypothesis of Lemma~\ref{lemma:critical-points-V0}(iii) hold.
\end{proof}


\bibliographystyle{../../halpha-abbrv-rita}
{\small \bibliography{unpub,../../library}}

\begin{thebibliography}{AMPW17}
\expandafter\ifx\csname url\endcsname\relax
  \def\url#1{\texttt{#1}}\fi
\expandafter\ifx\csname doi\endcsname\relax
  \def\doi#1{\burlalt{doi:#1}{http://dx.doi.org/#1}}\fi
\expandafter\ifx\csname urlprefix\endcsname\relax\def\urlprefix{URL }\fi
\expandafter\ifx\csname href\endcsname\relax
  \def\href#1#2{#2}\fi
\expandafter\ifx\csname burlalt\endcsname\relax
  \def\burlalt#1#2{\href{#2}{#1}}\fi

\bibitem[AAG20a]{Angelopoulos2019a}
Y.~Angelopoulos, S.~Aretakis, and D.~Gajic.
\newblock {A Non-degenerate Scattering Theory for the Wave Equation on Extremal
  Reissner–Nordstr{\"{o}}m}.
\newblock {\em Commun. Math. Phys.}, 380(1):323--408, 2020.

\bibitem[AAG20b]{Angelopoulos2018a}
Y.~Angelopoulos, S.~Aretakis, and D.~Gajic.
\newblock {Late-time asymptotics for the wave equation on extremal
  Reissner–Nordstr{\"{o}}m backgrounds}.
\newblock {\em Adv. Math. (N. Y).}, 375(107363):1--135, 2020.

\bibitem[AAG20c]{Angelopoulos2019}
Y.~Angelopoulos, S.~Aretakis, and D.~Gajic.
\newblock {Nonlinear Scalar Perturbations of Extremal
  Reissner–Nordstr{\"{o}}m Spacetimes}.
\newblock {\em Ann. PDE}, 6(12):1--124, 2020.

\bibitem[AB15a]{Andersson2015}
L.~Andersson and P.~Blue.
\newblock {Hidden symmetries and decay for the wave equation on the Kerr
  spacetime}.
\newblock {\em Ann. Math.}, 182(3):787--853, 2015.

\bibitem[AB15b]{Andersson2015a}
L.~Andersson and P.~Blue.
\newblock {Uniform energy bound and asymptotics for the Maxwell field on a
  slowly rotating Kerr black hole exterior}.
\newblock {\em J. Hyperbolic Differ. Equations}, 12(04):689--743, 2015.

\bibitem[ABBM19]{Andersson2019}
L.~Andersson, T.~B{\"{a}}ckdahl, P.~Blue, and S.~Ma.
\newblock {Stability for linearized gravity on the Kerr spacetime}.
\newblock {\itshape Preprint}, 2019,
  \href{http://arxiv.org/abs/1903.03859}{{\ttfamily arXiv:1903.03859}}.

\bibitem[AG00]{Andersson2000}
N.~Andersson and K.~Glampedakis.
\newblock {Superradiance Resonance Cavity Outside Rapidly Rotating Black
  Holes}.
\newblock {\em Phys. Rev. Lett.}, 84(20):4537--4540, 2000.

\bibitem[AIK10]{Alexakis2010}
S.~Alexakis, A.~D. Ionescu, and S.~Klainerman.
\newblock {Uniqueness of smooth stationary black holes in vacuum: Small
  perturbations of the Kerr spaces}.
\newblock {\em Commun. Math. Phys.}, 299(1):89--127, 2010.

\bibitem[Ali09]{Alinhac2009}
S.~Alinhac.
\newblock {Energy Multipliers for Perturbations of the Schwarzschild Metric}.
\newblock {\em Commun. Math. Phys.}, 288(1):199--224, 2009.

\bibitem[AMPW17]{Andersson2017}
L.~Andersson, S.~Ma, C.~Paganini, and B.~F. Whiting.
\newblock {Mode stability on the real axis}.
\newblock {\em J. Math. Phys.}, 58(7):072501, 2017.

\bibitem[Are11a]{Aretakis2011}
S.~Aretakis.
\newblock {Stability and Instability of Extreme Reissner-Nordstr{\"{o}}m Black
  Hole Spacetimes for Linear Scalar Perturbations I}.
\newblock {\em Commun. Math. Phys.}, 307(1):17--63, 2011.

\bibitem[Are11b]{Aretakis2011b}
S.~Aretakis.
\newblock {Stability and Instability of Extreme Reissner-Nordstr{\"{o}}m Black
  Hole Spacetimes for Linear Scalar Perturbations II}.
\newblock {\em Ann. Henri Poincar{\'{e}}}, 12(8):1491--1538, 2011.

\bibitem[Are12]{Aretakis2012a}
S.~Aretakis.
\newblock {Decay of axisymmetric solutions of the wave equation on extreme Kerr
  backgrounds}.
\newblock {\em J. Funct. Anal.}, 263(9):2770--2831, 2012.

\bibitem[Are15]{Aretakis2012}
S.~Aretakis.
\newblock {Horizon instability of extremal black holes}.
\newblock {\em Adv. Theor. Math. Phys.}, 19(3):507--530, 2015.

\bibitem[Bie09]{Bieri2009}
L.~Bieri.
\newblock {Solutions of the Einstein Vacuum Equations}.
\newblock In {\em Extensions of the Stability Theorem of the Minkowski Space in
  General Relativity}. AMS/IP Studies in Advanced Mathematics, 2009.

\bibitem[Blu08]{Blue2008}
P.~Blue.
\newblock {Decay of the Maxwell field on the Schwarzschild manifold}.
\newblock {\em J. Hyperbolic Differ. Equations}, 05(04):807--856, 2008.

\bibitem[BP73]{Bardeen1973}
J.~M. Bardeen and W.~H. Press.
\newblock {Radiation fields in the Schwarzschild background}.
\newblock {\em J. Math. Phys.}, 14(1):7--19, 1973.

\bibitem[Car68]{Carter1968}
B.~Carter.
\newblock {Hamilton-Jacobi and Schr{\"{o}}dinger Separable Solutions of
  Einstein's Equations}.
\newblock {\em Commun. Math. Phys.}, 10(4):280--310, 1968.

\bibitem[CBG69]{Choquet-Bruhat1969}
Y.~Choquet-Bruhat and R.~Geroch.
\newblock {Global aspects of the Cauchy problem in General Relativity}.
\newblock {\em Commun. Math. Phys.}, 14(4):329--335, 1969.

\bibitem[CCH12]{Chrusciel2012}
P.~T. Chru{\'{s}}ciel, J.~L. Costa, and M.~Heusler.
\newblock {Stationary Black Holes: Uniqueness and Beyond}.
\newblock {\em Living Rev. Relativ.}, 15(1):7, 2012.

\bibitem[CD76]{Chandrasekhar1976}
S.~Chandrasekhar and S.~Detweiler.
\newblock {On the Equations Governing the Gravitational Perturbations of the
  Kerr Black Hole}.
\newblock {\em Proc. R. Soc. A Math. Phys. Eng. Sci.}, 350(1661):165--174,
  1976.

\bibitem[CGPW20]{Cvetic2020}
M.~Cveti{\v{c}}, G.~W. Gibbons, C.~N. Pope, and B.~F. Whiting.
\newblock {Positive Energy Functional for Massless Scalars in Rotating Black
  Hole Backgrounds of Maximal Ungauged Supergravity}.
\newblock {\em Phys. Rev. Lett.}, 124(231102), 2020.

\bibitem[CGZ16]{Casals2016}
M.~Casals, S.~E. Gralla, and P.~Zimmerman.
\newblock {Horizon instability of extremal Kerr black holes: Nonaxisymmetric
  modes and enhanced growth rate}.
\newblock {\em Phys. Rev. D}, 94(6):1--8, 2016.

\bibitem[Cha75]{Chandrasekhar1975a}
S.~Chandrasekhar.
\newblock {On the Equations Governing the Perturbations of the Schwarzschild
  Black Hole}.
\newblock {\em Proc. R. Soc. A Math. Phys. Eng. Sci.}, 343(1634):289--298,
  1975.

\bibitem[Cha76]{Chandrasekhar1976a}
S.~Chandrasekhar.
\newblock {On a Transformation of Teukolsky's Equation and the Electromagnetic
  Perturbations of the Kerr Black Hole}.
\newblock {\em Proc. R. Soc. A Math. Phys. Eng. Sci.}, 348(1652):39--55, 1976.

\bibitem[Cha83]{Chandrasekhar}
S.~Chandrasekhar.
\newblock {\em {The Mathematical Theory of Black Holes}}.
\newblock Oxford University Press, Clarendon Press, New York, 1983.

\bibitem[Civ14a]{Civin2014a}
D.~Civin.
\newblock {Quantitative mode stability for the wave equation on the Kerr-Newman
  spacetime}.
\newblock {\itshape Preprint}, 2014,
  \href{http://arxiv.org/abs/1405.3620}{{\ttfamily arXiv:1405.3620}}.

\bibitem[Civ14b]{Civin2014}
D.~Civin.
\newblock {\em {Stability of charged rotating black holes for linear scalar
  perturbations}}.
\newblock PhD thesis, University of Cambrigde, 2014.

\bibitem[CK93]{Christodoulou1993}
D.~Christodoulou and S.~Klainerman.
\newblock {\em {The global nonlinear stability of the Minkowski space}}.
\newblock Princeton University Press, Princeton, 1993.

\bibitem[DH06]{Dafermos2006}
M.~Dafermos and G.~Holzegel.
\newblock {Dynamic instability of solitons in 4 + 1-dimensional gravity with
  negative cosmological constant}.
\newblock 2006.

\bibitem[DHR13]{Dafermos2013}
M.~Dafermos, G.~Holzegel, and I.~Rodnianski.
\newblock {A scattering theory construction of dynamical vacuum black holes}.
\newblock {\em Differ. Geom. (to Appear.}, 2013.

\bibitem[DHR19a]{Dafermos2017}
M.~Dafermos, G.~Holzegel, and I.~Rodnianski.
\newblock {Boundedness and Decay for the Teukolsky Equation on Kerr Spacetimes
  I: The Case $|a|\ll M$}.
\newblock {\em Ann. PDE}, 5(2):1--118, 2019.

\bibitem[DHR19b]{Dafermos2016a}
M.~Dafermos, G.~Holzegel, and I.~Rodnianski.
\newblock {The linear stability of the Schwarzschild solution to gravitational
  perturbations}.
\newblock {\em Acta Math.}, 222(1):1--214, 2019.

\bibitem[DL17]{Dafermos2017a}
M.~Dafermos and J.~Luk.
\newblock {The interior of dynamical vacuum black holes I: The $C^0$-stability
  of the Kerr Cauchy horizon}.
\newblock {\itshape Preprint}, 2017,
  \href{http://arxiv.org/abs/1710.01722}{{\ttfamily arXiv:1710.01722}}.

\bibitem[DR10]{Dafermos2010}
M.~Dafermos and I.~Rodnianski.
\newblock {Decay for solutions of the wave equation on Kerr exterior spacetimes
  I-II: The cases $|a| \ll M$ or axisymmetry}.
\newblock {\itshape Preprint}, 2010,
  \href{http://arxiv.org/abs/1010.5132}{{\ttfamily arXiv:1010.5132}}.

\bibitem[DR11]{Dafermos2011}
M.~Dafermos and I.~Rodnianski.
\newblock {A proof of the uniform boundedness of solutions to the wave equation
  on slowly rotating Kerr backgrounds}.
\newblock {\em Invent. Math.}, 185(3):467--559, 2011.

\bibitem[DR13]{Dafermos2008}
M.~Dafermos and I.~Rodnianski.
\newblock {Lectures on black holes and linear waves}.
\newblock In {\em Evol. equations. Clay Math. Proceedings, vol. 17}, pages
  97--205. American Mathematical Society, Providence, Rhode Island, 2013.

\bibitem[DRSR16]{Dafermos2016b}
M.~Dafermos, I.~Rodnianski, and Y.~Shlapentokh-Rothman.
\newblock {Decay for solutions of the wave equation on Kerr exterior spacetimes
  III: The full subextremal case $|a| < M$}.
\newblock {\em Ann. Math.}, 183(3):787--913, 2016.

\bibitem[DRSR18]{Dafermos2014}
M.~Dafermos, I.~Rodnianski, and Y.~Shlapentokh-Rothman.
\newblock {A scattering theory for the wave equation on Kerr black hole
  exteriors}.
\newblock {\em Ann. Sci. l'{\'{E}}cole Norm. sup{\'{e}}rieure}, 51(2):371--486,
  2018.

\bibitem[Dya15]{Dyatlov2015}
S.~Dyatlov.
\newblock {Asymptotics of Linear Waves and Resonances with Applications to
  Black Holes}.
\newblock {\em Commun. Math. Phys.}, 335(3):1445--1485, 2015.

\bibitem[Erd56]{Erdelyi1956}
A.~Erd{\'{e}}lyi.
\newblock {\em {Asymptotic Expansions}}.
\newblock Dover Publications, New York, 1956.

\bibitem[FB52]{Foures-Bruhat1952}
Y.~Four{\`{e}}s-Bruhat.
\newblock {Th{\'{e}}or{\`{e}}me d'existence pour certains syst{\`{e}}mes
  d'{\'{e}}quations aux d{\'{e}}riv{\'{e}}es partielles non lin{\'{e}}aires}.
\newblock {\em Acta Math.}, 88:141--225, 1952.

\bibitem[FI72]{Fackerell1972}
E.~D. Fackerell and J.~R. Ipser.
\newblock {Weak electromagnetic fields around a rotating black hole}.
\newblock {\em Phys. Rev. D}, 5(10):2455--2458, 1972.

\bibitem[Fri86]{Friedrich1986}
H.~Friedrich.
\newblock {On the existence of n-geodesically complete or future complete
  solutions of Einstein's field equations with smooth asymptotic structure}.
\newblock {\em Commun. Math. Phys.}, 107(4):587--609, 1986.

\bibitem[FS09]{Finster2009}
F.~Finster and J.~Smoller.
\newblock {Decay of solutions of the Teukolsky equation for higher spin in the
  Schwarzschild geometry}.
\newblock {\em Adv. Theor. Math. Phys.}, 13(1):71--110, 2009.

\bibitem[FS17]{Finster2017}
F.~Finster and J.~Smoller.
\newblock {Linear stability of the non-extreme Kerr black hole}.
\newblock {\em Adv. Theor. Math. Phys.}, 21(8):1991--2085, 2017.

\bibitem[GGH17]{Georgescu2017}
V.~Georgescu, C.~G{\'{e}}rard, and D.~H{\"{a}}fner.
\newblock {Asymptotic completeness for superradiant Klein–Gordon equations
  and applications to the De Sitte–Kerr metric}.
\newblock {\em J. Eur. Math. Soc.}, 19(8):2371--2444, 2017.

\bibitem[Gio20a]{Giorgi2020}
E.~Giorgi.
\newblock {The Linear Stability of Reissner–Nordstr{\"{o}}m Spacetime for
  Small Charge}.
\newblock {\em Ann. PDE}, 6(8), 2020.

\bibitem[Gio20b]{Giorgi2019a}
E.~Giorgi.
\newblock {The Linear Stability of Reissner–Nordstr{\"{o}}m Spacetime: The
  Full Subextremal Range $|Q|<M$}.
\newblock {\em Commun. Math. Phys.}, 380(3):1313--1360, 2020.

\bibitem[GJK17]{Glampedakis2017}
K.~Glampedakis, A.~D. Johnson, and D.~Kennefick.
\newblock {Darboux transformation in black hole perturbation theory}.
\newblock {\em Phys. Rev. D}, 96(2):024036, 2017.

\bibitem[HHV21]{Hafner2019}
D.~H{\"{a}}fner, P.~Hintz, and A.~Vasy.
\newblock {Linear stability of slowly rotating Kerr black holes}.
\newblock {\em Invent. Math.}, 223(3):1227--1406, 2021.

\bibitem[Hin21]{Hintz2018}
P.~Hintz.
\newblock {Normally hyperbolic trapping on asymptotically stationary
  spacetimes}.
\newblock {\em Probab. Math. Phys.}, 2(1):71--126, 2021.

\bibitem[HKW17]{Hung2017}
P.-K. Hung, J.~Keller, and M.-T. Wang.
\newblock {Linear Stability of Schwarzschild Spacetime: Decay of Metric
  Coefficients}.
\newblock {\itshape Preprint}, 2017,
  \href{http://arxiv.org/abs/1702.02843}{{\ttfamily arXiv:1702.02843}}.

\bibitem[HN04]{Hafner2004}
D.~H{\"{a}}fner and J.-P. Nicolas.
\newblock {Scattering of massless Dirac fields by a Kerr black hole}.
\newblock {\em Rev. Math. Phys.}, 16(01):29--123, 2004.

\bibitem[Hol16]{Holzegel2016}
G.~Holzegel.
\newblock {Conservation laws and flux bounds for gravitational perturbations of
  the Schwarzschild metric}.
\newblock {\em Class. Quantum Gravity}, 33(20):205004, 2016.

\bibitem[Hug00]{Hughes2000}
S.~A. Hughes.
\newblock {Computing radiation from Kerr black holes: Generalization of the
  Sasaki-Nakamura equation}.
\newblock {\em Phys. Rev. D}, 62(4):8, 2000.

\bibitem[Hun18a]{Huneau2018b}
C.~Huneau.
\newblock {Stability of Minkowski Space-Time with a Translation Space-Like
  Killing Field}.
\newblock {\em Ann. PDE}, 4(1):12, 2018.

\bibitem[Hun18b]{Hung2018a}
P.-K. Hung.
\newblock {The linear stability of the Schwarzschild spacetime in the harmonic
  gauge: odd part}.
\newblock {\itshape Preprint}, 2018,
  \href{http://arxiv.org/abs/1803.03881}{{\ttfamily arXiv:1803.03881}}.

\bibitem[Hun19]{Hung2019}
P.-K. Hung.
\newblock {The linear stability of the Schwarzschild spacetime in the harmonic
  gauge: even part}.
\newblock {\itshape Preprint}, 2019,
  \href{http://arxiv.org/abs/1909.06733}{{\ttfamily arXiv:1909.06733}}.

\bibitem[HV18]{Hintz2016}
P.~Hintz and A.~Vasy.
\newblock {The global non-linear stability of the Kerr–de Sitter family of
  black holes}.
\newblock {\em Acta Math.}, 220(1):1--206, 2018.

\bibitem[HV20]{Hintz2020}
P.~Hintz and A.~Vasy.
\newblock {Stability of Minkowski space and polyhomogeneity of the metric}.
\newblock {\em Ann. PDE}, 6(1):2, 2020.

\bibitem[IK15]{Ionescu2015}
A.~D. Ionescu and S.~Klainerman.
\newblock {On the Global Stability of the Wave-map Equation in Kerr Spaces with
  Small Angular Momentum}.
\newblock {\em Ann. PDE}, 1(1), 2015.

\bibitem[Inc56]{Ince1956}
E.~L. Ince.
\newblock {\em {Ordinary differential equations}}.
\newblock Dover Publications, 1956.

\bibitem[Joh19]{Johnson2018}
T.~W. Johnson.
\newblock {The Linear Stability of the Schwarzschild Solution to Gravitational
  Perturbations in the Generalised Wave Gauge}.
\newblock {\em Ann. PDE}, 5(2):13, 2019.

\bibitem[Kei18]{Keir2018}
J.~Keir.
\newblock {The weak null condition and global existence using the p-weighted
  energy method}.
\newblock {\itshape Preprint}, 2018,
  \href{http://arxiv.org/abs/1808.09982}{{\ttfamily arXiv:1808.09982}}.

\bibitem[Ker63]{Kerr1963}
R.~P. Kerr.
\newblock {Gravitational Field of a Spinning Mass as an Example of
  Algebraically Special Metrics}.
\newblock {\em Phys. Rev. Lett.}, 11(5):237--238, 1963.

\bibitem[KMW89]{Kalnins1989}
E.~G. Kalnins, W.~Miller, and G.~C. Williams.
\newblock {Teukolsky--Starobinsky identities for arbitrary spin}.
\newblock {\em J. Math. Phys.}, 30(12):2925--2929, 1989.

\bibitem[KS20]{Klainerman2017}
S.~Klainerman and J.~Szeftel.
\newblock {\em {Global Nonlinear Stability of Schwarzschild Spacetime under
  Polarized Perturbations}}.
\newblock Annals of Mathematics Studies, 2020.

\bibitem[LR10]{Lindblad2010}
H.~Lindblad and I.~Rodnianski.
\newblock {The global stability of Minkowski space-time in harmonic gauge}.
\newblock {\em Ann. Math.}, 171(3):1401--1477, 2010.

\bibitem[LR12]{Lucietti2012}
J.~Lucietti and H.~S. Reall.
\newblock {Gravitational instability of an extreme Kerr black hole}.
\newblock {\em Phys. Rev. D}, 86(10):1--7, 2012.

\bibitem[LT18]{Lindblad2018}
H.~Lindblad and M.~Tohaneanu.
\newblock {Global existence for quasilinear wave equations close to
  Schwarzschild}.
\newblock {\em Commun. Partial Differ. Equations}, 43(6):893--944, 2018.

\bibitem[LT20]{Lindblad2020}
H.~Lindblad and M.~Tohaneanu.
\newblock {A Local Energy Estimate for Wave Equations on Metrics Asymptotically
  Close to Kerr}.
\newblock {\em Ann. Henri Poincar{\'{e}}}, 21(11):3659--3726, 2020.

\bibitem[Luk12]{Luk2012}
J.~Luk.
\newblock {On the local existence for the characteristic initial value problem
  in general relativity}.
\newblock {\em Int. Math. Res. Not.}, 2012(20):4625--4678, 2012.

\bibitem[Luk13]{Luk2013}
J.~Luk.
\newblock {The null condition and global existence for nonlinear wave equations
  on slowly rotating Kerr spacetimes}.
\newblock {\em J. Eur. Math. Soc.}, 15(5):1629--1700, 2013.

\bibitem[Ma20a]{Ma2017}
S.~Ma.
\newblock {Uniform Energy Bound and Morawetz Estimate for Extreme Components of
  Spin Fields in the Exterior of a Slowly Rotating Kerr Black Hole I: Maxwell
  Field}.
\newblock {\em Ann. Henri Poincar{\'{e}}}, 21(3):815--863, 2020.

\bibitem[Ma20b]{Ma2017a}
S.~Ma.
\newblock {Uniform Energy Bound and Morawetz Estimate for Extreme Components of
  Spin Fields in the Exterior of a Slowly Rotating Kerr Black Hole II:
  Linearized Gravity}.
\newblock {\em Commun. Math. Phys.}, 377(3):2489--2551, 2020.

\bibitem[Mas20]{Masaood2020}
H.~Masaood.
\newblock {A Scattering Theory for Linearised Gravity on the Exterior of the
  Schwarzschild Black Hole I: The Teukolsky Equations}.
\newblock {\itshape Preprint}, 2020,
  \href{http://arxiv.org/abs/2007.13658v1}{{\ttfamily arXiv:2007.13658v1}}.

\bibitem[MN99]{Mason1999}
L.~Mason and J.-P. Nicolas.
\newblock {Global Results for the Rarita--Schwinger Equations and Einstein
  Vacuum Equations}.
\newblock {\em Proc. London Math. Soc.}, 79(3):694--720, 1999.

\bibitem[Mos17a]{Moschidis2017a}
G.~Moschidis.
\newblock {A proof of the instability of AdS for the Einstein--null dust system
  with an inner mirror}.
\newblock {\itshape Preprint}, 2017,
  \href{http://arxiv.org/abs/1704.08681}{{\ttfamily arXiv:1704.08681}}.

\bibitem[Mos17b]{Moschidis2017b}
G.~Moschidis.
\newblock {Superradiant instabilities for short-range non-negative potentials
  on Kerr spacetimes and applications}.
\newblock {\em J. Funct. Anal.}, 273(8):2719--2813, 2017.

\bibitem[Mos18]{Moschidis2018}
G.~Moschidis.
\newblock {A proof of the instability of AdS for the Einstein--massless Vlasov
  system}.
\newblock {\itshape Preprint}, 2018,
  \href{http://arxiv.org/abs/1812.04268}{{\ttfamily arXiv:1812.04268}}.

\bibitem[MRT13]{Murata2013}
K.~Murata, H.~S. Reall, and N.~Tanahashi.
\newblock {What happens at the horizon(s) of an extreme black hole?}
\newblock {\em Class. Quantum Gravity}, 30(23):235007, 2013.

\bibitem[MS54]{Meixner1954}
J.~Meixner and F.~W. Sch{\"{a}}fke.
\newblock {\em {Mathieusche Funktionen und Sph{\"{a}}roidfunktionen}}.
\newblock Springer Berlin Heidelberg, Berlin, Heidelberg, 1954.

\bibitem[NP62]{Newman1962}
E.~T. Newman and R.~Penrose.
\newblock {An approach to gravitational radiation by a method of spin
  coefficients}.
\newblock {\em J. Math. Phys.}, 3(3):566--578, 1962.

\bibitem[Olv73]{Olver1973}
F.~W.~J. Olver.
\newblock {\em {Introduction to asymptotics and special functions}}.
\newblock Academic Press, 1973.

\bibitem[OS20]{Oliver2020}
J.~Oliver and J.~Sterbenz.
\newblock {A vector field method for radiating black hole spacetimes}.
\newblock {\em Anal. \& PDE}, 13(1):29--92, 2020.

\bibitem[Pas19a]{Pasqualotto2017}
F.~Pasqualotto.
\newblock {Nonlinear Stability for the Maxwell–Born–Infeld System on a
  Schwarzschild Background}.
\newblock {\em Ann. PDE}, 5(2):19, 2019.

\bibitem[Pas19b]{Pasqualotto2016}
F.~Pasqualotto.
\newblock {The spin $\pm 1$ Teukolsky equations and the Maxwell system on
  Schwarzschild}.
\newblock {\em Ann. Henri Poincar{\'{e}}}, 20(4):1263--1323, 2019.

\bibitem[Ren90]{Rendall1990}
A.~D. Rendall.
\newblock {Reduction of the characteristic initial value problem to the Cauchy
  problem and its applications to the Einstein equations}.
\newblock {\em Proc. R. Soc. A Math. Phys. Eng. Sci.}, 427(1872):221--239,
  1990.

\bibitem[RW57]{Regge1957}
T.~Regge and J.~A. Wheeler.
\newblock {Stability of a Schwarzschild Singularity}.
\newblock {\em Phys. Rev.}, 108(4):1063--1069, 1957.

\bibitem[Sbi16]{Sbierski2016}
J.~Sbierski.
\newblock {On the Existence of a Maximal Cauchy Development for the Einstein
  Equations: a Dezornification}.
\newblock {\em Ann. Henri Poincar{\'{e}}}, 17(2):301--329, 2016.

\bibitem[SC74]{Starobinsky1974}
A.~A. Starobinsky and S.~M. Churilov.
\newblock {Amplification of electromagnetic and gravitational waves scattered
  by a rotating ``black hole''}.
\newblock {\em J. Exp. Theor. Phys.}, 38(1):1--5, 1974.

\bibitem[Sch16]{Schwarzschild1916}
K.~Schwarzschild.
\newblock {{\"U}ber das Gravitationsfeld eines Massenpunktes nach der
  Einsteinschen Theorie}.
\newblock {\em Sitzungsberichte der K{\"{o}}niglich Preu{\ss}ischen Akad. der
  Wissenschaften}, 3:189--196, 1916.

\bibitem[SN82]{Sasaki1982}
M.~Sasaki and T.~Nakamura.
\newblock {Gravitational Radiation from a Kerr Black Hole. I. Formulation and a
  Method for Numerical Analysis}.
\newblock {\em Prog. Theor. Phys.}, 67(6):1788--1809, 1982.

\bibitem[SR14]{Shlapentokh-Rothman2014}
Y.~Shlapentokh-Rothman.
\newblock {Exponentially Growing Finite Energy Solutions for the Klein-Gordon
  Equation on Sub-Extremal Kerr Spacetimes}.
\newblock {\em Commun. Math. Phys.}, 329(3):859--891, 2014.

\bibitem[SR15]{Shlapentokh-Rothman2015}
Y.~Shlapentokh-Rothman.
\newblock {Quantitative Mode Stability for the Wave Equation on the Kerr
  Spacetime}.
\newblock {\em Ann. Henri Poincar{\'{e}}}, 16:289--345, 2015.

\bibitem[SRTdC22]{SRTdC2022}
Y.~Shlapentokh-Rothman and R.~{Teixeira da Costa}.
\newblock {Boundedness and decay for the Teukolsky equation on Kerr in the full
  subextremal range $|a|<M$: physical space analysis}.
\newblock {\itshape Preprint}, 2022.

\bibitem[ST14]{Sterbenz2014}
J.~Sterbenz and D.~Tataru.
\newblock {Local Energy Decay for Maxwell Fields Part I: Spherically Symmetric
  Black-Hole Backgrounds}.
\newblock {\em Int. Math. Res. Not.}, 2014.

\bibitem[Sta73]{Starobinsky1973}
A.~A. Starobinsky.
\newblock {Amplification of waves reflected from a rotating ``black hole''.}
\newblock {\em J. Exp. Theor. Phys.}, 37(1):28--32, 1973.

\bibitem[Tay19]{Taylor2019}
M.~Taylor.
\newblock {The nonlinear stability of the Schwarzschild family of black holes
  (joint work with M. Dafermos, G. Holzegel, I. Rodnianski)}.
\newblock In C.~Cederbaum, M.~Dafermos, J.~Isenberg, and H.~Ringstr{\"{o}}m,
  editors, {\em Oberwolfach Reports -- Math. Gen. Relativ.}, volume~15, pages
  2230--2233. 2019.

\bibitem[TdC20]{TeixeiradaCosta2019}
R.~{Teixeira da Costa}.
\newblock {Mode Stability for the Teukolsky Equation on Extremal and
  Subextremal Kerr Spacetimes}.
\newblock {\em Commun. Math. Phys.}, 378(1):705--781, 2020.

\bibitem[TdC21]{TdC-thesis}
R.~{Teixeira da Costa}.
\newblock {\em {Frequency space analysis in General Relativity}}.
\newblock PhD thesis, University of Cambridge, 2021.

\bibitem[Teu73]{Teukolsky1973}
S.~A. Teukolsky.
\newblock {Perturbations of a Rotating Black Hole. I. Fundamental Equations for
  Gravitational, Electromagnetic, and Neutrino-Field Perturbations}.
\newblock {\em Astrophys. J.}, 185:635--647, 1973.

\bibitem[TP74]{Teukolsky1974}
S.~A. Teukolsky and W.~H. Press.
\newblock {Perturbations of a rotating black hole. III - Interaction of the
  hole with gravitational and electromagnetic radiation}.
\newblock {\em Astrophys. J.}, 193:443--461, 1974.

\bibitem[TT11]{Tataru2011}
D.~Tataru and M.~Tohaneanu.
\newblock {A local energy estimate on Kerr black hole backgrounds}.
\newblock {\em Int. Math. Res. Not.}, 2011(2):248--292, 2011.

\bibitem[Wal73]{Wald1973}
R.~M. Wald.
\newblock {On perturbations of a Kerr black hole}.
\newblock {\em J. Math. Phys.}, 14(10):1453--1461, 1973.

\bibitem[Wan13]{Wang2013}
F.~Wang.
\newblock {Radiation field for Einstein vacuum equations with spacial dimension
  $n\geq 4$}.
\newblock {\itshape Preprint}, 2013,
  \href{http://arxiv.org/abs/1304.0407}{{\ttfamily arXiv:1304.0407}}.

\bibitem[Whi89]{Whiting1989}
B.~F. Whiting.
\newblock {Mode stability of the Kerr black hole}.
\newblock {\em J. Math. Phys.}, 30(6):1301--1305, 1989.

\bibitem[WP70]{Walker1970}
M.~Walker and R.~Penrose.
\newblock {On quadratic first integrals of the geodesic equations for type {22}
  spacetimes}.
\newblock {\em Commun. Math. Phys.}, 18(4):265--274, 1970.

\bibitem[WZ11]{Wunsch2011}
J.~Wunsch and M.~Zworski.
\newblock {Resolvent Estimates for Normally Hyperbolic Trapped Sets}.
\newblock {\em Ann. Henri Poincar{\'{e}}}, 12(7):1349, 2011.

\bibitem[Zer70]{Zerilli1970a}
F.~J. Zerilli.
\newblock {Effective Potential for Even-Parity Regge-Wheeler Gravitational
  Perturbation Equations}.
\newblock {\em Phys. Rev. Lett.}, 24(13):737--738, 1970.

\end{thebibliography}

\end{document}